\documentclass[a4paper,11pt]{article}

\usepackage{jheppub} 
\usepackage[T1]{fontenc} 
\usepackage{physics}
\usepackage{amsthm}
\usepackage{amssymb}
\usepackage{bbm}
\usepackage{mathrsfs}
\usepackage{amscd}
\usepackage{tikz}
\usepackage{tikz-cd}
\usepackage{enumerate}
\usepackage{textcomp}
\newcommand{\UHP}{\mathbb{H}}
\newcommand{\Hilbert}{\mathcal{H}}
\newcommand{\UD}{\mathbb{D}}
\newcommand{\deformspace}{\mathfrak{D}}
\newcommand{\cmpx}{\mathbb{C}}

\newcommand{\diffspace}{\mathcal{A}}

\newcommand{\schottky}{\mathfrak{S}}
\newcommand{\teich}{\mathcal{T}}
\newcommand{\moduli}{\mathcal{M}}
\newcommand{\symmoduli}{\mathfrak{M}}
\newcommand{\modular}{\operatorname{Mod}}
\newcommand{\puremod}{\operatorname{Mod}_0}
\newcommand{\symm}[1]{\operatorname{Symm}\left(#1\right)}
\newcommand{\kerphi}{\mathcal{N}}
\newcommand{\linebundle}{\mathscr{L}}
\newcommand{\Lie}{\mathcal{L}}

\newcommand{\PSLR}{\operatorname{PSL}(2,\mathbb{R})}
\newcommand{\PSLC}{\operatorname{PSL}(2,\mathbb{C})}

\newcommand{\equidef}{\overset{\text{\tiny \it def.}}{\equiv}}
\newcommand{\fund}{\mathcal{F}}
\newcommand{\confspace}[2]{\mathscr{F}_{#1}\left(#2\right)}
\newcommand{\Automorphism}{\operatorname{Aut}}
\newcommand{\Inn}{\operatorname{Inn}}
\newcommand{\chern}[2]{\mathsf{c}_1\left(#1, #2\right)}
\newcommand{\Gpotential}{\mathscr{S}}
\newcommand{\Lponetial}{\mathsf{h}}
\newcommand{\sigtype}{\boldsymbol{\mathsf{s}}}
\newcommand{\Sorder}[1]{\text{\scriptsize{$\mathcal{O}$}}\left(#1\right)}
\newcommand{\SchottkyFund}{\mathcal{D}}

\newcommand{\ff}{\mathrm{f}}
\newcommand{\sing}{\operatorname{Sing}}
\newcommand{\deck}{\operatorname{deck}}
\newcommand{\isom}{\operatorname{Isom}}
\newcommand{\hol}{\operatorname{hol}}
\newcommand{\dev}{\operatorname{dev}}
\newcommand{\orderrange}{\hat{\mathbb{N}}^{^{>1}}}

\newcommand{\CPl}{\mathbb{CP}^1}
\newcommand{\s}{\mathsf{s}}
\newcommand{\ord}{\operatorname{ord}}

\newcommand{\brdiv}{\mathscr{D}}
\newcommand{\singrigon}{\overset{{}_{\curlywedge}}{\Omega}}
\newcommand{\singfund}{\overset{{}_{\curlywedge}}{\mathcal{D}}}
\newcommand{\gal}{\operatorname{Gal}}
\newtheorem{theorem}{Theorem}[section]
\newtheorem{maintheorem}{Theorem}
\newtheorem*{theorem*}{Theorem}
\newtheorem{corollary}{Corollary}[section]
\newtheorem{lemma}{Lemma}[section]
\newtheorem*{lemma*}{Lemma}
\newtheorem{proposition}{Proposition}[section]
\theoremstyle{definition}
\newtheorem{definition}{Definition}[section]
\newtheorem{remark}{Remark}[section]
\usepackage{longtable}
\usepackage{lipsum} 
\title{\boldmath Classical Liouville Action and Uniformization of Orbifold Riemann Surfaces}

\author[1]{Behrad Taghavi,}
\author[1]{Ali Naseh,}
\author[2]{Kuroush Allameh}
\affiliation[1]{School of Particles and Accelerators, Institute for Research in Fundamental Sciences (IPM),\\ P.O. Box 19395-5531, Tehran, Iran}
\affiliation[2]{Physics Department, UC Santa Cruz, 1156 High Street Santa Cruz, CA 95064, USA}

\emailAdd{btaghavi@ipm.ir}
\emailAdd{naseh@ipm.ir}
\emailAdd{kallameh@ucsc.edu}
\abstract{We study the classical Liouville field theory on Riemann surfaces of genus $g>1$ in the presence of vertex operators associated with branch points of orders $m_i>1$. In order to do so, we will consider the generalized Schottky space $\schottky_{g,n}(\boldsymbol{m})$ obtained as a holomorphic fibration over the Schottky space $\schottky_g$ of the (compactified) underlying Riemann surface. The fibers of $\schottky_{g,n}(\boldsymbol{m}) \to \schottky_g$ correspond to configuration spaces of $n$ orbifold points of orders $\boldsymbol{m} = (m_1,\dots,m_n)$. Drawing on the previous work of Park, Takhtajan, and Teo \cite{park2015potentials} as well as Takhtajan and Zograf \cite{ZT_2018}, we define Hermitian metrics $\Lponetial_i$ for tautological line bundles $\linebundle_i$ over $\schottky_{g,n}(\boldsymbol{m})$. These metrics are expressed in terms of the first  coefficient of the expansion of covering map $J$ near each singular point on the Schottky domain. Additionally, we define the regularized classical Liouville action $S_{\boldsymbol{m}}$ using Schottky global coordinates on Riemann orbisurfaces with genus $g>1$. We demonstrate that $\exp[S_{\boldsymbol{m}}/\pi]$ serves as a Hermitian metric in the holomorphic $\mathbb{Q}$-line bundle $\linebundle = \bigotimes_{i=1}^{n} \linebundle_i^{\otimes (1-1/m_i^2)}$ over $\schottky_{g,n}(\boldsymbol{m})$. Furthermore, we explicitly compute the first and second variations of the smooth real-valued function $\Gpotential_{\boldsymbol{m}} = S_{\boldsymbol{m}} - \pi \sum_{i=1}^n (m_i - \tfrac{1}{m_i}) \log \Lponetial_{i}$ on the Schottky deformation space $\schottky_{g,n}(\boldsymbol{m})$. We establish two key results: (i) $\Gpotential_{\boldsymbol{m}}$ generates a combination of accessory and auxiliary parameters, and (ii) $-\Gpotential_{\boldsymbol{m}}$ acts as a K\"{a}hler potential for a specific combination of Weil--Petersson and Takhtajan--Zograf metrics that appear in the local index theorem for orbifold Riemann surfaces \cite{ZT_2018}. The obtained results can then be interpreted in terms of the complex geometry of the Hodge line bundle equipped with Quillen's metric over the moduli space $\symmoduli_{g,n}(\boldsymbol{m})$ of Riemann orbisurfaces and the tree-level approximation of conformal Ward identities associated with quantum Liouville theory.}

\dedicated{\centering Dedicated to Professor Hessamaddin Arfaei on the occasion of his 75th birthday}

\begin{document} \maketitle
\flushbottom

\section{Introduction}
Conformal field theory in two dimensions has found a wide range of applications in both physics and mathematics. Perhaps, one of the most interesting applications of CFTs in mathematical physics is to the geometry of surfaces: This is most clear in Liouville CFT, introduced by Polyakov \cite{Polyakov:1981rd}, that can be viewed as a quantum theory of geometry in two dimensions \cite{Takhtajan:1993vt,Teschner:2003at}. This theory admits two dimensional surfaces of constant negative curvature (possibly with sources)  as its classical solutions. It is then natural to consider these classical solutions as critical points of a certain functional defined on the space of all smooth conformal metrics on a given Riemann surface. In the context of string theory, this functional is known as the \emph{Liouville action functional} while its critical value is usually called the \emph{classical Liouville action}.

From a mathematical perspective, the connection between Liouville theory and complex geometry of moduli spaces of Riemann surfaces was first established by Takhtajan and Zograf \cite{ZT_1985,Zograf1988ONLE,1988SbMat..60..297Z}. One novelty of their work was the use of (Fuchsian or Schottky) projective structures on Riemann surfaces to construct the Liouville action. Zograf and Takhtajan proved that the classical Liouville action is a K\"{a}hler potential for the Weill--Petersson (WP) metric on moduli spaces of punctured Riemann spheres \cite{Zograf1988ONLE}, as well as on Schottky spaces of compact Riemann surfaces \cite{1988SbMat..60..297Z}. In the case of punctured Riemann spheres, the classical action is a generating function for the famous \emph{accessory parameters} of Klein and Poincar\'{e}. For compact Riemann surfaces, the classical Liouville action is an anti-derivative of a 1-form on the Schottky space given by the difference of Fuchsian and Schottky projective connections. In turn, this 1-form is an anti-derivative of the Weil--Petersson symplectic 2-form on the Schottky space. See \cite{takhtajan1987uniformization,Takhtajan:1991yk} for reviews of these results.

Later on, Takhtajan and Zograf  introduced a new K\"{a}hler metric \cite{takhtajan1988selberg,takhtajan1991local}, called Takhtajan--Zograf (TZ) metric \cite{obitsu1999non,weng2001omega,wolpert2007cusps}, on the moduli space $\symmoduli_{g,n}$ of punctured Riemann surfaces in the process of deriving a local index theorem (in Quillen's form) for families of Cauchy-Riemann operators (for its precise definition, see Sect.~\ref{KahlerMetrics}). In 2015, Park, Takhtajan, and Teo \cite{park2015potentials} found a K\"{a}hler potential $\Lponetial_i$ for the i-th TZ metric in terms of the first coefficient of the Fourier expansion of a covering map $J$ near the i-th puncture. The authors of \cite{park2015potentials} also showed that these K\"{a}hler potentials are essential for defining classical Liouville actions that are invariant under certain subgroups of the Teichm\"uller modular group: An appropriate definition of classical Liouville action on a punctured Riemann surface needs a regularization procedure that introduces a ``modular anomaly'' (see \cite{Zograf_1990}). K\"{a}hler potentials $\Lponetial_i$ are then essential in cancellation of these anomalous contributions.

More recently, a generalization of local index theorem \cite[Theorem~1]{takhtajan1991local} to the case of orbifold Riemann surfaces \cite{ZT_2018}, has lead Zograf and Takhtajan to introduce yet another K\"{a}hler metric on the moduli space of Riemann orbisurfaces. In order to avoid confusion, this new K\"{a}hler metric will be called the \emph{elliptic} TZ metric while the one introduced in \cite{takhtajan1988selberg,takhtajan1991local} will be called the \emph{cuspidal} TZ metric.\footnote{It was demonstrated by the authors of \cite{ZT_2018} that the elliptic TZ metric converges to the cuspidal TZ metric in the limit that the opening angle of corresponding elliptic fixed point approaches zero.} Using the results of \cite{park2015potentials}, Zograf and Takhtajan \cite{ZT_2018} also found a K\"{a}hler potential for the i-th elliptic TZ metric in the case of genus zero orbifold Riemann surfaces.

Motivated by the results of \cite{park2015potentials, ZT_2018}, this manuscript explores the classical limit of Liouville field theory (LFT) on orbifold Riemann surfaces with genus $g>1$ using the Schottky global coordinates.\footnote{While our main results are derived for the case of Riemann orbisurfaces with genus $g>1$, we still study genus zero Riemann orbisurfaces to draw some important lessons.} Our main result can be viewed as an extension of \cite[Theorems~1,2]{park2015potentials} to the case of orbifold Riemann surfaces (both compact and with punctures): While the authors of \cite{park2015potentials} considered the classical Liouville action on (generalized) Schottky space of punctured Riemann surfaces, we have to take into account the contributions of orbifold points to the Liouville action as well. Despite the fact that some aspects of this generalization might be familiar to mathematicians and experts, we have still chosen to include them here in order to make this manuscript self-contained and more accessible. In particular, while the method of proofs in this work closely resemble those in \cite{park2015potentials,ZT_2018,Zograf_1990} (and the references therein), the details of calculations for the case of orbifolds, to the best of our knowledge, have not appeared explicitly anywhere in the literature.

From a mathematical perspective, our result provides evidence that the connection between classical Liouville action and Quillen's metric in the Hodge line bundle (see \cite{Zograf_1990}) extends to the orbifold setting.\footnote{Some other mathematical implications of our result will be highlighted in a forthcoming shorter version of this manuscript.} From the point of view of physics, the results of this paper have multiple applications, many of which stem from the connection between partition function of Liouville theory on a Riemann surface with conical singularities and correlation functions of Liouville vertex operators corresponding to conical defects (see, e.g. \cite{Takhtajan:1994vt,Takhtajan:1995fd,Takhtajan:2005md}).

Studying CFTs like LFT on Riemann orbisurfaces with genus $g$ has additional reasons from the physics perspective. One notable rationale comes from the fact that many established constructions that relate geometry and entanglement are based on bipartite entanglement. For instance, the emergence of eternal black holes from quantum entanglement between two copies of the CFT in a thermofield double (TFD) state \cite{maldacena2003eternal}, or Ryu-Takayanagi  \cite{ryu2006holographic} minimal area surface in AdS anchored on the boundary of a sub-region that determines spatial entanglement between that sub-region and its complement in dual theory. Likewise, the MERA ansatz which reveals an additional dimension for AdS spacetime in the direction of increasing or decreasing entanglement \cite{swingle2012constructing}. Additionally, linearized Einstein equations which can be derived from the entanglement of the underlying quantum degrees of freedom \cite{lashkari2014gravitational}. However, it is possible for the degrees of freedom to be entangled in a multipartite manner, much like how many-point correlation functions cannot be deduced from lower correlations. This is evident in  tripartite entangled states such as GHZ (Greenberger-Horne-Zeilinger) and W states. These states (and some of their deformations) are similar to the TFD state and can be described by integrating over half of a higher genus surface (with possible singularities). Another reason stems from the R\'{e}nyi entropies, $S_n$. For a reduced density matrix $\rho$ of a spatial region A,
\begin{equation*}
S_n = \frac{1}{1-n} \log \Tr (\rho^n)=\frac{1}{1-n}\left(\log Z_n-n\log Z_1\right),
\end{equation*}
where $Z_n$ is the partition function on an orbifold Riemann surface with non-zero genus and $\mathbb{Z}_n$ symmetry. This orbifold surface can be constructed by gluing together $n$-copies of the original system across the entangling region $A$ with some disjoint regions.\footnote{See e.g.\cite{Liu:2016zmr}.} The Schottky uniformization suggests that distinct phases should be considered in studying the $Z_n$ partition function and in order to prove the RT formula, the extension of those distinct phases into a quotient of $\mathbb{H}_3$ should be explored. Actually the dominant contribution, determined by the least action principle based on the values in each phase, is important in determining the wave function and proving the RT formula.\footnote{For attempts in this direction see \cite{Faulkner:2013yia,Maxfield:2016mwh}.} Accordingly, this implies that studying the CFT on orbifold Riemann surfaces is especially instrumental in determining whether the assumption of replica symmetry holds true in the dual gravitational system.

As another reason, it is worth noting that the correlation functions of twist operators in a CFT have a connection to partition functions on orbisurfaces with different genera. Specifically, a CFT that arises from the low energy limit of a 2-dimensional sigma model with a target space of $M^n/S_n$ (symmetric orbifold of $n$ copies of $M$). This relationship was highlighted in a paper by Lunin and Mathur \cite{lunin2001correlation}. Interestingly, the sigma model with the aforementioned target space can describe the low energy behavior of the $D1-D5$ system  \cite{de1999six,dijkgraaf1999instanton,seiberg1999d1}, which generates a near horizon structure of $\text{AdS}_3\times S^3\times M_4$, as presented in \cite{larsen1999u}. Thus, possessing knowledge about the Liouville action on orbisurfaces can be highly advantageous for gaining a better insight into not only the string theoretical constructions of $\text{AdS}_3/\text{CFT}_2$, but also for addressing important topics related to black hole physics, such as microstate counting \cite{strominger1996microscopic}.

Given that we intend to allocate a part to extending our findings from Riemann orbisurfaces to conical Riemann surfaces, it is essential to outline here some motivations behind this choice. One motivation ties into the importance of investigating quantum gravity in three dimensions. Previous research has shown \cite{maloney2010quantum,keller2015poincare} \footnote{Pure gravity on global AdS$_{3}$ can be rewritten as two copies of geometric quantization of a specific coadjoint orbit of the Virasoro group \cite{witten1988coadjoint, maloney2010quantum}. For the path integral quantization of the same coadjoint orbit and accordingly exploring the dual boundary theory, see \cite{cotler2019theory}.} that if one only considers smooth saddle-points when calculating the gravitational path integral in 3-dimensional gravity, the resulting regularized partition function is plagued by two issues. Firstly, the range of twists at a constant spin is continuous rather than discrete. Secondly, when dealing with high spins and energies near the edge of the spectrum, the density of states becomes negative. The first difficulty can potentially be resolved by considering recent findings that suggest the dual theory of 2-dimensional AdS gravity is an ensemble of one-dimensional quantum mechanics \cite{saad2019jt,stanford2019jt}. To address the issue of non-unitarity, it has been proposed that extra contributions should be added to the path integral over metrics, namely Seifert manifolds, which are off-shell configurations \cite{maxfield2021path}.\footnote{In another proposal, the 3-dimensional theory is modified by adding some special massive particles which it implies that one should consider 3-dimensional conical manifolds beside the smooth saddles, see \cite{benjamin2020pure}.} It is particularly interesting to note that through Kaluza-Klein reduction, the solutions of the derived 2-dimensional theory are conical Riemann surfaces. 

The study of conical Riemann surfaces also plays a crucial role in addressing a significant concern within the realm of 2-dimensional CFTs. Ideally, one aims to resolve the constraints of conformal invariance and unitarity to ascertain the permissible values for conformal dimensions $\Delta_i$ and Operator Product Expansion (OPE) coefficients. This pursuit would lead to a comprehensive classification of 2-dimensional CFTs. But since such exhaustive classification does not exist (up to now), one can at least explore the universal aspects of those data, i.e., those hold true in any conformal field theory. When all scalar operators in three-point functions possess high dimensions (i.e., they are "heavy"), a universal formula for the averaged value of OPE coefficients emerges for any unitary and compact 2-dimensional CFT with a central charge $c$ greater than 1, as detailed in references \cite{chang2016bootstrap,Collier:2019weq} (see also \cite{Abajian:2023bqv}). Notably, when $12\Delta_i/c <1$, this formula finds \cite{Collier:2019weq} a connection to the DOZZ (Dorn-Otto-Zamolodchikov-Zamolodchikov) formula for the structure constants of vertex operators in Liouville theory. Actually, the classical correlation function of Liouville vertex operators on a Riemann surface with genus $g$ can be linked to the on-shell value of the Liouville action functional on the same Riemann surface, albeit with the insertion of conical points at the positions of those operators -effectively transforming it into a conical Riemann surface. As a result, delving into LFT on conical Riemann surfaces offers a dual benefit. It not only allows us to investigate the universal features of OPEs within 2-dimensional CFTs but also sheds light on certain facets of 3-dimensional gravity in the presence of heavy particles, a realm characterized by 3-dimensional geometry with conical defects.\footnote{The operator with $12\Delta_i/c > 1$ is dual to a black hole state. See \cite{Abajian:2023bqv}.} 

The aforementioned observation presents a different aspect when examined within the context of the bulk dual. Within semiclassical gravity, the wormhole amplitudes can be understood as averaged solutions to the mentioned CFT's bootstrap constraints in the semiclassical limit  \cite{chandra2022semiclassical,sasieta2023wormholes}.\footnote{This interpretation differs from the random matrix interpretation for 2D gravity, where averaging occurs across a family of UV-complete quantum theories.} To be more precise, the Euclidean wormhole solutions provide connected contributions to the average of products of CFT's correlation functions.\footnote{An alternative interpretation in terms of coarse graining in a single CFT is provided by \cite{chandra2022coarse}.} Moreover, by initiating from a two-sphere boundary wormhole with $(n+1)$ massive particles going through the wormhole and then analytically continuing the mass of $(m+1)$ of them to the black hole regime, the two-sphere boundaries are effectively joined at their $(m+1)$ pairs of insertion points. This results in the creation of a genus-$m$ handlebody with $2(n-m)$ conical singularities \cite{chandra2022semiclassical}.\footnote{The parameters that define operator dimensions change into moduli of the (conical) Riemann surface as the operators are made heavy.} Consequently, the LFT on the single conical boundary of the handlebody not only is connected to the  analytical continuation of the dimension of defect operators (mass of massive particles) on two-boundary wormholes but also can shed light on the statistical distribution of CFT's data on some regime of scaling dimensions. 

Furthermore, it is established that there exist numerous distinct families of black hole microstates, each comprising an infinite number of members \cite{wheeler1964relativity,balasubramanian2022microscopic}. These families are also closely related to geometries featuring Einstein-Rosen bridges of potentially immense volume. Intriguingly, it has been demonstrated that a substantial reduction in the dimension of the Hilbert space can happen by adding the contribution of wormhole saddle points in the gravitational path integral. These wormholes yield minute yet universal contributions to the quantum overlap of candidate black hole microstates and shed light on the really orthogonal ones \cite{balasubramanian2022microscopic}.\footnote{See also \cite{chakravarty2021overcounting}. } Accordingly, this concept can also help to resolve the problem \cite{susskind2020black} of growth of holographic complexity at exponentially large times. Actually, some of the microstates are created by massive particles with masses below the black hole threshold, which reside behind the horizon without altering the mass. Therefore, by analytically continuing external operators to the black hole regime in two-boundary wormholes, the LFT on conical Riemann surfaces can also offer insights into the minute overlaps between different microstates and, ultimately, provide a deeper understanding of the Bekenstein-Hawking entropy and holographic complexity.

Even more intriguingly, when one integrates out the mentioned (on-shell) wormholes in 3-dimensional gravity coupled to sufficiently massive point particles, it results in the emergence of random bulk 3-point interactions among these point particles. These interactions exhibit the same statistical properties as the boundary OPE coefficients \cite{chandra2022semiclassical}. As a consequence, it becomes apparent that the LFT on conical Riemann surfaces can also be utilized to investigate these random couplings within the bulk Effective Field Theory and to offer a controlled and semiclassical way to realize the mechanism originally proposed by Coleman, Giddings and Strominger \cite{ coleman1988black,giddings1988loss}.

There are also other reasons to write this paper, which we will mention in the Discussion section when more possible applications for our results are explored. 
\subsection*{Related Works}
At this stage, we would like to make some comments about the relation between our work and that of other authors who have studied neighbouring questions:
\begin{itemize}
	\item In addition to the study of classical Liouville action on generalized Schottky space of punctured Riemann surfaces, reference \cite{park2015potentials} also studied Liouville theory on punctured Riemann surfaces with quasi-Fuchsian global coordinates. Moreover, the authors of \cite{park2015potentials} have proved the holographic correspondence for the case of punctured Riemann surfaces with both Schottky and quasi-Fuchsian global coordinates.\footnote{The holographic correspondence for compact Riemann surfaces has been proved a long time ago (see \cite{Krasnov:2000zq,Takhtajan:2002cc, Krasnov:2009vy}) and asserts that the renormalized volume of a hyperbolic 3-manifold, which is a purely three-dimensional quantity in its definition, is equivalent with the classical Liouville action on its conformal boundary --- a purely two-dimensional quantity.} While Park and Teo \cite{Teo_2018} have already extended the results of \cite{park2015potentials} to the case of orbifold Riemann surfaces with quasi-Fuchsian global coordinates,\footnote{In particular, reference \cite{Teo_2018} also includes the proof of holographic correspondence for the case of quasi-Fuchsian orbifolds. Moreover, from a physics perspective, such orbifolds have been studied by Chandra, Collier, Hartman, and Maloney \cite{chandra2022semiclassical}.} a rigorous study of Liouville action and holographic correspondence for the case of Riemann orbisurfaces with Schottky global coordinates has not appeared anywhere in the literature. The present manuscript aims to partially fill this gap: We study the classical Liouville action on Schottky deformation space of Riemann orbisurfaces but leave a rigorous proof of the holographic correspondence to a future work \cite{ANT23}.\footnote{While a rigorous proof of holographic correspondence is still outstanding for the case of handlebody orbifolds, many references have studied Einstein-Hilbert action on AdS$_{3}$ with conical singularities in connection to correlation functions of Liouville vertex operators (see, e.g. \cite{Krasnov:2002rn,Abajian:2023bqv}).}
	
	\item Motivated by the study of quantum Hall states on singular surfaces, reference \cite{Can:2017ycp} has studied the modular invariant Liouville action on Riemann sphere with conical singularities. For the special case of Riemann orbisurfaces with genus $g=0$, our results are in agreement with that of reference \cite{Can:2017ycp}.
	
	\item As we will further discuss in section \ref{Discussion}, the main results of this paper (see Theorems~\ref{mainthrm1} and \ref{mainthrm2}) provide strong evidence for a close connection between classical Liouville action and (appropriately defined) determinant of Laplacian on Riemann orbisurfaces with genus $g>1$; for Riemann sphere with conical singularities, this connection has also been studied by Kalvin (see, e.g. \cite{kalvin2023determinants,kalvin2023triangulations}).\footnote{More generally, see \cite{kalvin2021polyakov} for the derivation of a Polyakov-type anomaly formula in this case.} In this sense, our results are closely related to the studies of Laplace-Beltrami operator on Riemann orbisurfaces and its spectral properties \cite{Montplet2016RiemannRochII,Teo_2021}.\footnote{From a physics perspective, the zeta-regularized determinant of Laplacian on Riemann orbisurfaces has been recently studied in \cite{lin2023revisiting}.}
\end{itemize}
\subsection*{Structure of the Paper}
The structure of this work is as follows: In section~\ref{Action}, we will briefly review the relationship between correlation functions of heavy Liouville vertex operators corresponding to branch points and the uniformization theory of orbifold Riemann surfaces. Section~\ref{geometry} will cover various topics related to the deformation theory of Ahlfors and Bers. This will include a discussion of some known facts about the geometry of Teichm\"uller, Schottky, and moduli spaces of Riemann orbisurfaces as well as some variational formulas which we will need throughout this manuscript. Section~\ref{sec:action} contains a detailed study of regularized Liouville action and its geometric properties: Subsection~\ref{subsec:liouvilleactionnogenus} studies the regularized classical action on Riemann orbisurfaces of genus $g=0$ while subsection~\ref{gg1} focuses on the classical Liouville action defined for Riemann orbisurfaces of genus $g > 1$. In Section \ref{sec:Potentials}, we will first study K\"ahler potentials $\Lponetial_{i}$ for cuspidal and elliptic Takhtajan-Zograf metrics on both $\moduli_{0,n}$ and $\schottky_{g,n}(\boldsymbol{m})$. In particular, we will demonstrate that for certain special line bundles $\linebundle_i$ equipped with Hermitian metrics $\Lponetial_{i}^{m_i}$, the first Chern forms are related to the K\"ahler form of TZ metrics associated with elliptic and parabolic generators. Moreover, we will show that the first Chern form of the $\mathbb{Q}$-line bundle $\linebundle = \bigotimes_{i=1}^{n} \linebundle_i^{(1-1/m_i^2)}$ with Hermitian metric $\exp[S_{\boldsymbol{m}}/\pi]$ is given by $\tfrac{1}{\pi^2} \omega_{\text{WP}}$; here, $S_{\boldsymbol{m}}$ denotes the appropriately regularized classical Liouville action. From these results, it is easy to see that a specific combination $\omega_{\text{WP}}- \tfrac{4 \pi^2}{3} \, \omega_{\text{TZ}}^{\text{cusp}}- \tfrac{\pi}{2} \, \sum_{j=1}^{n_e} (m_i -\tfrac{1}{m_i}) \hspace{.5mm} \omega^{\text{ell}}_{\text{TZ},j}$ of Weil-Petersson and Takhtajan-Zograf metrics has a global K\"{a}hler potential on $\schottky_{g,n}(\boldsymbol{m})$ given by $\Gpotential_{\boldsymbol{m}} = S_{\boldsymbol{m}} - \pi \sum_{i=1}^n (m_i - \tfrac{1}{m_i}) \log \Lponetial_{i}$. Theorems~\ref{mainthrm1} and \ref{mainthrm2} constitute the main findings of this paper and are related to the first and second variations of $\Gpotential_{\boldsymbol{m}}$. In Section \ref{Discussion}, we will provide a brief overview of some implications of our findings and discuss potential pathways for future research. Appendix~\ref{Apx:orbifoldbackground} offers some mathematical background regarding orbifold Riemann surfaces, while Appendix~\ref{Apx:geometricorbifolds} delves into geometric structures on such orbisurfaces. Moreover, Appendix~\ref{Apx:Asymptotics} outlines the derivation of various asymptotic behaviors that are used throughout the main body of this manuscript. Finally, for the convenience of the reader, a list of symbols used throughout this text is also presented in Appendix~\ref{symbol}.

\section{Classical Liouville Action and Uniformization Theory}\label{Action}
In this section, we will discuss the semi-classical limit of quantum Liouville field theory on hyperbolic Riemann orbisurfaces and its connection with the uniformization theory. Let $X$ be a compact Riemann surface of genus $g > 1$. In the so-called geometric approach to quantum Liouville theory, developed by Takhtajan in \cite{Takhtajan:1993vt,Takhtajan:1994zi,Takhtajan:1994vt,Takhtajan:1995fd} based on the original proposal by Polyakov \cite{Polyakov82},\footnote{See references  \cite{Matone:1993tj,Matone:1993nf,Menotti:2004uq,Bertoldi:2004yk,Hadasz:2005gk} for details regarding the relation between geometric and standard approaches to Liouville CFT.} the un-normalized correlation functions of Liouville vertex operators with ``charges'' $\alpha_i$ on the Riemann surface $X$ are defined by (we set $\hbar=1$)
\begin{equation}\label{LiouvilleCorrelationFunc}
\expval{V_{\alpha_{1}}(x_1) \dotsm V_{\alpha_{n}}(x_n)} = \int\limits_{\mathscr{CM}_{\boldsymbol{\alpha}}(X)} \hspace*{-15pt}\mathcal{D}\psi \,\, e^{-\frac{1}{2 \pi} \Gpotential_{\boldsymbol{\alpha}}[\psi]},
\end{equation}
where $\mathscr{CM}_{\boldsymbol{\alpha}}(X)$ is the space of all smooth conformal metrics $e^{\psi(u,\bar{u})} |du|^2$ on $X\backslash\{x_1,\dots,x_n\}$ which have conical singularities of angles $2\pi(1- \alpha_i)$ at the insertion points, and $\Gpotential_{\boldsymbol{\alpha}}[\psi(u,\bar{u})]$ is the Liouville action functional which we intend to study. 

When $\alpha_i$s take the special values $1-\tfrac{1}{m_i}$ for integers $m_i \geq 2$, vertex operators $V_{\alpha_i}$ correspond to orbifold points (or branch points) of orders $m_i$; moreover, when $m_i$ is the symbol $\infty$, $V_{\alpha_i}$ corresponds to a cusp. Then, the correlation function $\expval{V_{m_1}(x_1) \dotsm V_{m_n}(x_n)}$ of these twist-like vertex operators can be interpreted as the partition function of LFT on a (possibly punctured) Riemann orbisurface $O$ of genus $g > 1$. In this viewpoint, $X$ plays the role of compactified \emph{underlying Riemann surface} of $O$, denoted by $\hat{X}_O$, and the collection of insertion points $\sing(O):= \{x_i\}_{_{i=1,\dots,n}} \subset X$ will be called the \emph{singular set} of $O$. Then, the Riemann orbisurface $O$ can be characterized as the triple $(X,\sing(O),\textcolor{black}{\nu})$ where $\nu: \sing(O) \to \orderrange :=\big(\mathbb{N}\backslash\{1\}\big) \cup \{\infty\}$ is the so-called \emph{branching function} that assigns to each singular point $x_i$ its corresponding \emph{branching order} $m_i \in \orderrange$ for $i=1,\dots,n$. Now, let $X_O := \hat{X}_O \backslash \{x_i \, | \, m_i = \infty\}$ be the underlying Riemann surface of $O$. A Riemann orbisurface $O$ can be equivalently characterized as a pair $(X_O,\brdiv)$ where the so-called \emph{branch divisor}
\begin{equation}
\brdiv:=\sum_{\{x_i \, |\, m_i < \infty\}} \left(1-\frac{1}{\nu(x_i)}\right) x_i,
\end{equation}
is a $\mathbb{Q}$-divisor on $X_O$ (see Section \ref{sect:orbifolddef} for more details).\footnote{By a \emph{$\mathbb{Q}$-divisor} on a Riemann surface $X$, we simply mean a formal linear combination of points on $X$ with rational coefficients.} When $O$ has cusps, i.e. branch points of order $m_i = \infty$, we will denote its compactification by $\hat{O} := (\hat{X}_O,\hat{\brdiv})$ where \begin{equation}
	\hat{\brdiv}:= \brdiv + \sum_{\{x_i \, |\, m_i = \infty\}} x_i.
\end{equation}
An orbifold Riemann surface $O$ with $g>1$ handles, $n_e \geq 3$ conical points of orders $2 \leq m_1 \leq \dotsm \leq m_{n_e} < \infty$, and $n_p \geq 0$ punctures is said to have the signature $\left(g;m_1,\dots, m_{n_e};n_p\right)$.

Next, let $O$ be an orbifold Riemann surface with signature $(g;m_1,\dots, m_{n_e};n_p)$ and let $\big\{(U_{a} , u_{a})\big\}_{a \in A}$ be a complex-analytic atlas on $X_O$ with charts $U_a$, local coordinates $u_a : U_a \to \cmpx$, and transition functions $g_{ab} : u_b (U_a \cap U_b) \to u_a (U_a \cap U_b)$. Denote by $\mathscr{CM}(O)$ the space of singular conformal metrics on $X_O$ representing $\brdiv$. If $X_O^{\text{reg}}:= \hat{X}_O \backslash \sing(O)$ denotes the so-called \emph{regular locus} of $O$, every such metric $ds^2 \in \mathscr{CM}(O)$ is given by a collection $\left\{e^{\psi_{a}} |du_a|^2\right\}_{a \in A}$, where the functions $\psi_a \in \mathcal{C}^{\infty}\big(U_a \cap X_O^{\text{reg}},\mathbb{R}\big)$ satisfy
\begin{equation}\label{phitransfrom}
\psi_{a} \circ g_{ab} + \log |g'_{ab}|^2 = \psi_{b} \quad \text{on} \quad U_a \cap U_b \cap X_O^{\text{reg}},
\end{equation}
and near each singular point $x_i \in U_a$, $e^{\psi_a}$ has the form
\begin{equation}\label{metricasymp}
e^{\psi_a} \simeq \left\{
\begin{split}
& \frac{4 \, m_i^{-2} |u_a-x_i|^{\frac{2}{m_i}-2}}{\left(1-|u_a-x_i|^{\frac{2}{m_i}}\right)^2} \qquad &\text{for} \qquad &m_i < \infty,\\ \\
& \frac{1}{|u_a-x_i|^2 \big(\log |u_a-x_i|\big)^2} & \text{for}\qquad & m_i = \infty,
\end{split}
\right.
\end{equation}
as $u_a \to x_i$. According to the classical results of Picard \cite{picard1893equation,picard1905integration} and the more recent work of McOwen \cite{McOwen-1988} and Troyanov \cite{troyanov1991prescribing}, when the Euler characteristic $\chi(O) := \chi(X_O) - \deg(\brdiv) = 2-2g -n_p -\sum_{i=1}^{n_e} (1-1/m_i)$ is negative --- i.e. when $O$ is hyperbolic --- there exists a unique singular conformal metric, called the \emph{hyperbolic metric}, on $X_O$ which represents the branch divisor $\brdiv$ and has constant Gaussian curvature $-1$ everywhere on $X_O^{\text{reg}}$. If we denote this unique hyperbolic metric by $ds_{\text{hyp}}^2 := \left\{e^{\varphi_a} |du_a|^2\right\}_{a \in A}$ and assuming that each open subset $U_a \subset X_O$ include at most one singular point $x_i$ of order $m_i$, the corresponding function $\varphi_a$ on $U_a$ satisfies the so-called \emph{Liouville equation} (see, e.g. \cite{Matone:1993nf})\footnote{If the cone point $x_i$ is fixed to be at infinity, instead of the last term on the right hand side of \eqref{Eq:Liouville}, we have $\pi (1+\frac{1}{m_i})\delta(u_a-x_i)$.}
\begin{equation}\label{Eq:Liouville}
\partial_{u_a} \partial_{\bar{u}_a} \varphi_a = \frac{1}{2} e^{\varphi_a} - \pi \left(1-\frac{1}{m_i}\right) \delta(u_a - x_i),
\end{equation}
which is equivalent to $ds_{\text{hyp}}^2$ having constant curvature $-1$ on $X_{O}^{\text{reg}}$ and satisfying asymptotics \eqref{metricasymp} near each singular point. 

The problem is now to define (suitably regularized) Liouville action functional on the Riemann orbisurface $O$ --- a functional $\Gpotential_{\boldsymbol{m}}: \mathscr{CM}(O) \to \mathbb{R}$ such that its Euler--Lagrange equation is the Liouville equation. However, it is well-known (see the discussion in \cite{Takhtajan:2002cc}) that a general mathematical definition of the Liouville action functional on a Riemann orbisurface $O$ of genus $g>1$ is a non-trivial problem. This is due to the fact that the classical Liouville field $\varphi$ is not a globally defined function on $O$ but rather a logarithm of the conformal factor of the metric. More concretely, due to transformation law \eqref{phitransfrom}, local kinetic terms $|\partial_{u_a} \varphi_a|^2 du_a \wedge d\bar{u}_a$ do not glue properly on $U_a \cap U_b \cap X_O^{\text{reg}}$ and thus can not be integrated over $X_O^{\text{reg}}$. This means that ``naive'' Dirichlet' type functional is not well defined and can not serve as an action for the Liouville theory (it also diverges at the singular points). In other words, the Liouville action functional cannot be defined in terms of a Riemann orbisurface $O$ alone and additionally depends on the choice of a \emph{global coordinate} on $X_O^{\text{reg}}$ --- a representation $X_O^{\text{reg}} \cong \Omega_K/K$, where $K$ is a Kleinian group with an invariant component $\Omega \subset \cmpx$ and $\Omega_K :=\Omega\backslash\{\text{fixed points}\}$. As we will see in the next subsection, for our purposes, it will be sufficient to consider the case when $K$ is either a Fuchsian group $\Gamma$ or a Schottky group $\Sigma$.

Finally, let us mention that once the action functional is defined, we can define the partition function $\expval{O}$, or rather its \emph{free energy} $\mathbb{F}_{O} := - \log \expval{O}$, using the perturbative expansion\footnote{In order to write this perturbative expansion, we have to temporarily restore $\hbar$ in \eqref{LiouvilleCorrelationFunc}.}
\begin{equation}
\mathbb{F}_{O} = -\log\expval{O} = \frac{1}{2 \pi \hbar} \Gpotential_{\boldsymbol{m}}[\varphi] + \frac{1}{2} \log\det(\Delta_0+\frac{1}{2}) + \order{\hbar},
\end{equation}
of the r.h.s of \eqref{LiouvilleCorrelationFunc} around the classical solution $\varphi$, where $\Gpotential_{\boldsymbol{m}}[\varphi]: \deformspace(K)/\sim \, \to \mathbb{R}$ (see Section \ref{geometry}) 
is the regularized classical Liouville action and $\Delta_0$ is the \emph{Laplace operator} of the hyperbolic metric acting on functions on $O$. In the bulk of this paper, we will only concern the classical contribution (i.e. contributions of order $\hbar^{-1}$) to the free energy $\mathcal{F}_{O}$, which is given by the classical Liouville action $\Gpotential_{\boldsymbol{m}}$.
\subsection{Fuchsian and Schottky Uniformizations}
Consider a hyperbolic orbifold Riemann surface $O=(X_O,\brdiv)$ with signature $(g;m_1,\dots, m_{n_e}$
$;n_p)$ and fix a base point $x_{\ast} \in X_O^{\text{reg}}$. Let us choose a standard system of generators for the orbifold fundamental group 
\begin{equation}
\begin{split}	
&\pi_1(O,x_{\ast}) =
\bigg{\langle} \mathsf{A}_1, \mathsf{B}_1,\dots, \mathsf{A}_g, \mathsf{B}_g, \mathsf{C}_1, \dots, \mathsf{C}_{n_e}, \mathsf{P}_1, \dots, \mathsf{P}_{n_p} \,\bigg|\,\\
&\hspace{5cm} \mathsf{C}_1^{m_1} = \dotsm = \mathsf{C}_{n_e}^{m_{n_e}} = \prod_{i=1}^{g}[\mathsf{A}_i,\mathsf{B}_i] \prod_{j=1}^{n_e} \mathsf{C}_j \prod_{k=1}^{n_p} \mathsf{P}_k = \operatorname{id}\bigg{\rangle},
\end{split}
\end{equation}
where $\mathsf{A}_i$s, $\mathsf{B}_i$s, $\mathsf{C}_j$s, and $\mathsf{P}_k$s are homotopy classes of loops based at $x_{\ast}$ and $[\mathsf{A}_i,\mathsf{B}_i]:=\mathsf{A}_i \mathsf{B}_i \mathsf{A}_i^{-1} \mathsf{B}_i^{-1}$; see Definition \ref{def:orbifoldfundgroup} and the discussion following it for more details. The Riemann orbisurface $O$ with a distinguished system of generators for its fundamental group $\pi_1(O,x_{\ast})$, up to inner automorphisms of $\pi_1(O,x_{\ast})$, will be called a \emph{marked Riemann orbisurface} (see Figure \ref{fig:fundamentalgroup}).
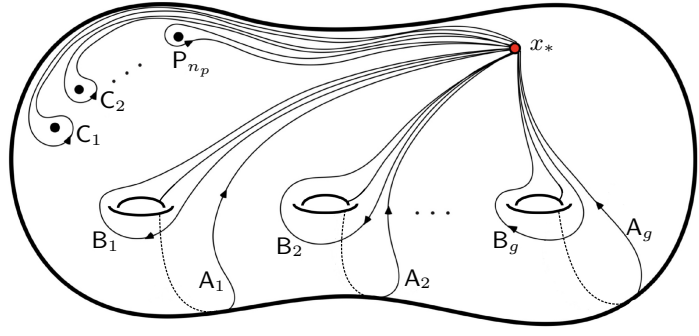
\begin{figure}
	\centering
	\includegraphics[width=0.6\linewidth]{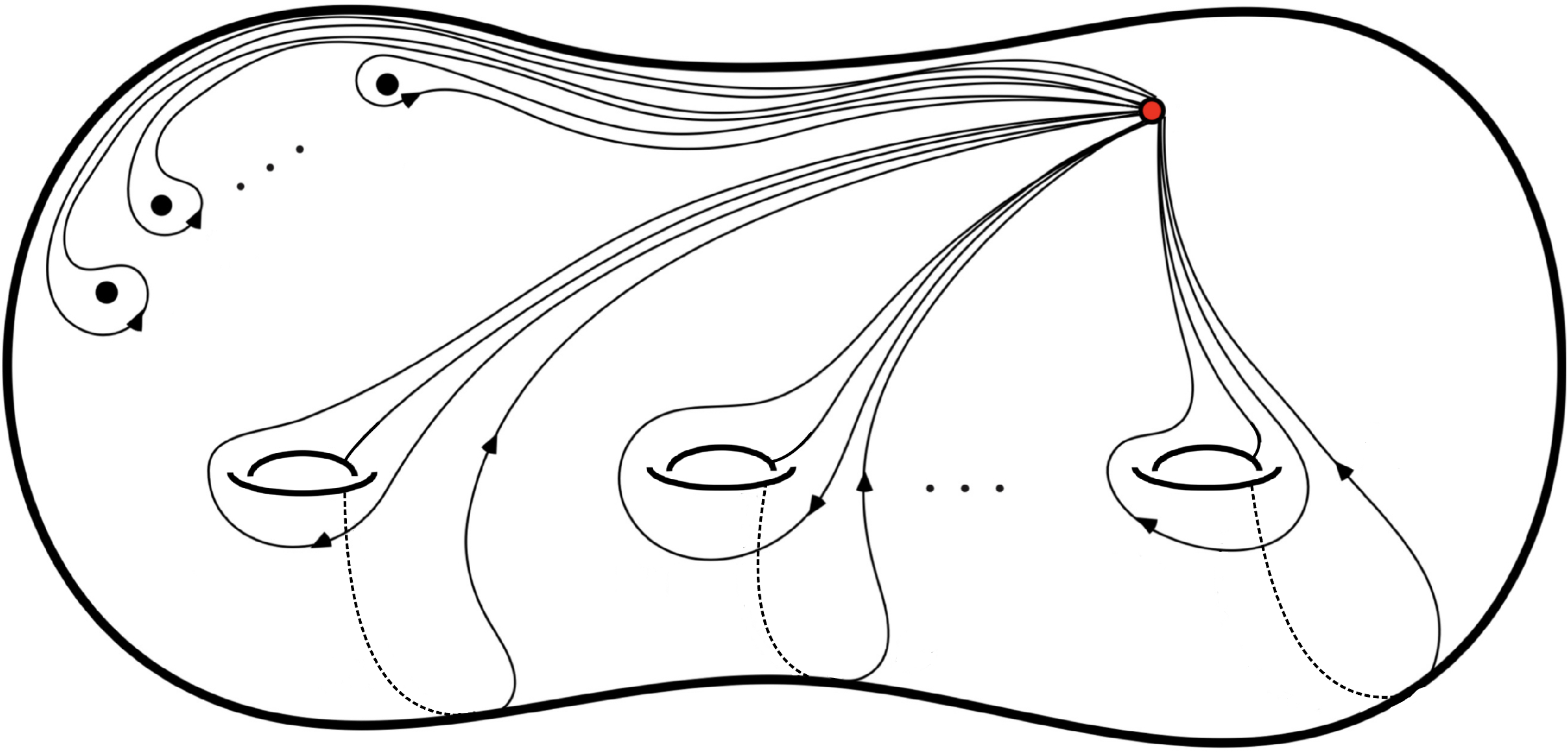}
	\put(-236,70){\rotatebox{0}{\fontsize{9}{9}\selectfont{} $\mathsf{C}_1$}}
	\put(-227,83){\rotatebox{0}{\fontsize{9}{9}\selectfont{} $\mathsf{C}_2$}}
	\put(-200,98){\rotatebox{0}{\fontsize{9}{9}\selectfont{} $\mathsf{P}_{n_p}$}}
	\put(-230,30){\rotatebox{0}{\fontsize{9}{9}\selectfont{} $\mathsf{B}_1$}}
	\put(-161,28){\rotatebox{0}{\fontsize{9}{9}\selectfont{} $\mathsf{B}_2$}}
	\put(-80,30){\rotatebox{0}{\fontsize{9}{9}\selectfont{} $\mathsf{B}_g$}}
	\put(-30,35){\rotatebox{0}{\fontsize{9}{9}\selectfont{} $\mathsf{A}_g$}}
	\put(-113,16){\rotatebox{0}{\fontsize{9}{9}\selectfont{} $\mathsf{A}_2$}}
	\put(-190,15){\rotatebox{0}{\fontsize{9}{9}\selectfont{} $\mathsf{A}_1$}}
	\put(-66,105){\rotatebox{0}{\fontsize{9}{9}\selectfont{} $x_{\ast}$}}
	\caption{\emph{Marked Riemann orbisurface}. A marked Riemann orbisurface with signature $(g;m_1,\dots,m_{n_e};n_p)$ is an orbifold Riemann surface $O$ together with a distinguished set of standard generators $\{\mathsf{A}_1, \dots, \mathsf{A}_g ; \mathsf{B}_1, \dots, \mathsf{B}_g ; \mathsf{C}_1, \dots, \mathsf{C}_{n_e} ; \mathsf{P}_1, \dots, \mathsf{P}_{n_p}\}$ for its fundamental group $\pi_1(O,x_{\ast})$.}\label{fig:fundamentalgroup}
\end{figure}

As a result of the Theorems~\ref{thm:universalcovering} and \ref{theorem:geometricorbifolds}, all hyperbolic Riemann orbisurfaces are \emph{developable} (or \emph{good} in Thurston's language) and hence can be realized as a global quotient $[\UHP/\Gamma]$ where $\UHP:=\left\{z \in \cmpx \, \big| \, \Im z >0 \right\}$ is the upper half-plane and $\Gamma \subset \PSLR$ is a Fuchsian group of the first kind\footnote {A Fuchsian group is said to be of \emph{first kind} if its limit set is the closed real line $\mathbb{R} \cup \{\infty\}$. Otherwise, a Fuchsian group is said to be of the second kind.} with signature $\left(g;m_1,\dots, m_{n_e};n_p\right)$; this is a direct consequence of the usual uniformization theorem for ordinary Riemann surfaces. The holomorphic orbifold covering map $\pi_{\Gamma}: \UHP \to O$ provides a Riemann orbisurface $O$ with the \emph{Fuchsian global coordinate}, and the hyperbolic metric $ds_{\text{hyp}}^2 \in \mathscr{CM}(O)$ is a push-forward of the Poincar\'{e} metric $(\Im z)^{-2} |dz|^2$ on $\UHP$ by the covering map $\pi_{\Gamma}$. From this point of view, $\Gamma$ can be thought of as the fundamental group of the Riemann orbisurface $O \cong [\UHP/\Gamma]$, and the group isomorphism $\Gamma \simeq \pi_1(O)$ can be viewed as being induced by the holonomy representation $\hol: \pi_1(O) \to \PSLR$ of the \emph{orbifold hyperbolic structure}.\footnote{Note that, in the language of orbifold $(G,\mathbb{X})$-structures introduced in Appx.\ref{Apx:geometricorbifolds}, orbifold hyperbolic structures are, in fact, $\big(\PSLR,\UHP\big)$-structures.} 

Such a Fuchsian group $\Gamma$ has a standard presentation, corresponding to the standard generators of $\pi_1(O)$ discussed above, which includes $2g$ hyperbolic generators $\alpha_1, \beta_1, \dots, \alpha_g, \beta_g$, $n_e$ elliptic generators $\tau_1, \dots, \tau_{n_e}$ of orders $m_1, \dots, m_{n_e}$, and $n_p$ parabolic generators $\kappa_1, \dots, \kappa_{n_p}$.\footnote{A non-identity element $\gamma \in \Gamma$ is called \emph{hyperbolic}, \emph{parabolic}, or \emph{elliptic} if $\gamma$ is conjugated in $\PSLR$ to a \emph{dilation}, \emph{horizontal translation}, or \emph{rotation} respectively. This is analogous to $|\tr(\gamma)|$ being greater than, equal, or less than 2, respectively.} Obviously, the generators of $\Gamma$ also satisfy
\begin{equation}
\prod_{i=1}^{g} \left[\alpha_i, \beta_i\right]  \prod_{j=1}^{n_e} \tau_j \prod_{k=1}^{n_p} \kappa_k = \mathbbm{1} \quad \text{and} \quad \tau_j^{m_j} =  \mathbbm{1} \quad (j=1,\dots,n_e),
\end{equation}
where $[\alpha_i, \beta_i]:=\alpha_i \beta_i \alpha_i^{-1} \beta_i^{-1}$ and $\mathbbm{1}$ is the identity element of $\PSLR$. The Fuchsian group $\Gamma$, together with a distinguished system of generators
\begin{equation}
\big\{\alpha_1, \dots, \alpha_g ; \beta_1, \dots, \beta_g ; \tau_1, \dots, \tau_{n_e} ; \kappa_1, \dots, \kappa_{n_p}\big\},
\end{equation}
is called the \emph{marked Fuchsian group} corresponding to the Riemann orbisurface $O \cong [\UHP/\Gamma]$. 

The elliptic elements of $\Gamma$ will have fixed points in $\UHP$ and are denoted by $z^{{}_e}_{1},\dots,z^{{}_e}_{n_e}$ while the fixed points of the parabolic elements lie in $\partial\UHP = \mathbb{R} \cup \{\infty\}$ and will be denoted by $z^{{}_p}_{1},\dots,z^{{}_p}_{n_p}$. The images of these elliptic and parabolic fixed points under the projection $\UHP \to O \cong [\UHP/\Gamma]$ will be the conical points $x^{{}_c}_1,\dots,x^{{}_c}_{n_e}$ and punctures $x^{{}_p}_{1},\dots,x^{{}_p}_{n_p}$ of $O$, respectively. For our future convenience, let us also introduce $\sing_m(O):=\nu^{-1}(m)$ for all $m \in \orderrange$. Note that $\bigsqcup_{m \in \orderrange} \sing_m(O)$ gives a canonical stratification of the singular set $\sing(O) \equiv \operatorname{Supp}(\hat{\brdiv})$, and each $\sing_m(O)$ for $m \neq \infty$ represents the stratum of conical points with stabilizer group $\mathbb{Z}_m$. In addition, we will denote by $\sing_{\infty}(O):=\nu^{-1}(\infty)$ and $\sing_{\curlywedge}(O):= \bigsqcup_{m \neq \infty} \sing_m(O)$ the subset of cusps and conical points in $\sing(O)$ respectively. Finally, following \cite{Bers1975DeformationsAM}, we will define the \emph{signature type} of $O$ as the unordered set $\sigtype := \{\s_m\}_{m \in \orderrange}$ where $\s_m:=\big|\sing_m(O)\big|$ denotes the cardinality of the stratum of singular points of order $m$. In particular, we have $\s_{\infty} = \big|\sing_{\infty}\big|=n_p$ and $\sum_{m\in \orderrange} \s_m =\big|\sing(O)\big| = n_e + n_p $.

\begin{remark}
	Sometimes, when we need to refer to singular points (or fixed points) collectively, we will denote them by $x_1,\dots,x_n$ (respectively $z_1,\dots,z_n$) where $n=n_e + n_p$ is the total number of singular points (or fixed points) and the indices are ordered such that the corresponding orders of isotropy increase $2 \leq m_1 \leq m_2 \leq \dots \leq m_n \leq \infty$. In this situation, the vector of orders $(m_1,\dots,m_n)$ will be denoted by $\boldsymbol{m}$. Note that with this convention, the first $n_e$ singular points $x_1,\dots,x_{n_e}$ will always correspond to conical points $x^{{}_c}_1,\dots,x^{{}_c}_{n_e}$ while the remaining $n_p \geq 0$ singular points $x_{_{n_e+1}},\dots,x_{n}$  will correspond to the punctures $x^{{}_p}_{1},\dots,x^{{}_p}_{n_p}$ of $O$.
\end{remark}

Similar to Fuchsian groups, Schottky groups can also be used to construct Riemann orbisurfaces. We begin with a few definitions: A \emph{Kleinian group} $K$ is a discrete subgroup of the M\"{o}bius group $\PSLC$ that acts properly discontinuously on a subset $\Omega \subset \hat{\cmpx}$ called the \emph{region of discontinuity} of $K$. The complement $\Lambda = \hat{\cmpx} \backslash \Omega$ is called the \emph{limit set} of $K$. In this work, we are particularly interested in Kleinian groups that are free, finitely generated, and strictly loxodromic; such Kleinian groups are called \emph{Schottky groups} and will be denoted by $\Sigma$. It is well-known that for a Schottky group $\Sigma$ of rank $g$, the limit set $\Lambda$ is a Cantor set\footnote{For more details on the geometry of limit sets see Ref.~\cite{Seade_2015}.} and the region of discontinuity $\Omega = \hat{\cmpx}\backslash\Lambda$ is a dense connected subset of $\hat{\cmpx}$ such that the Schottky group $\Sigma$ acts on $\Omega$ freely, and the quotient space $\Omega/\Sigma$ is a closed Riemann surface $X$ of genus $g$; this is called a \emph{Schottky uniformization} of $X$ and, as a consequence of the retrosection theorem \cite{koebe1910uniformisierung} (see also \cite{bers1976automorphic,ford2004automorphic}), every closed Riemann surface has such a uniformization. 

Now, let us consider uniformization of the compactified underlying Riemann surface $\hat{X}_O$ by a Schottky group $\Sigma$. If $\Omega$ denotes the region of discontinuity of $\Sigma$, we can subtract from it the pre-images of cusps by the covering map $\Omega \to \hat{X}_O$ to get another planar region $\Omega_0$. The space $\Omega_0$ will uniformize the underlying Riemann surface $X_O$ and we will denote the corresponding covering map $\Omega_0 \to X_O \cong \Omega_0/\Sigma$ by $\pi_0$. Next, we can lift the branch divisor $\brdiv$ by the covering map $\pi_{0}: \Omega_0 \to X_O$ to get another branch divisor 
\begin{equation}
	\widetilde{\brdiv} := \sum_{w_i \in \pi_{0}^{-1}(\sing_{\curlywedge}(O))} \left(1-\frac{1}{\nu\left(\pi_{\Sigma}(w_i)\right)}\right) w_i,
\end{equation}
which lives on the planar region $\Omega_0$. Then, the pair $(\Omega_0,\widetilde{\brdiv})$ will define a planar Riemann orbisurface $\singrigon$ such that $\pi_{\Sigma}:\singrigon \to O \cong \singrigon/\Sigma$ is an orbifold covering map (see \cite{Wong-1971} for more details). In addition, note that the restriction of $\pi_0$ to $\Omega^{\text{reg}} := \Omega_0 \backslash \operatorname{Supp}\widetilde{\brdiv}$ provides $X_O^{\text{reg}}$ with the \emph{Schottky global coordinate} such that the space $\mathscr{CM}(O)$ is identified with the affine subspace of $\mathcal{C}^{\infty}(\Omega^{\text{reg}}, \mathbb{R})$ consisting of functions $\psi$ satisfying the condition
\begin{equation}\label{Schottkycovariant}
	\psi \circ \sigma + \log|\sigma'|^2 = \psi \quad \text{for all} \quad \sigma \in \Sigma,
\end{equation}
and representing $\widetilde{\brdiv}$.

Let us now define a \emph{marked Schottky group} as a Schottky group $\Sigma$ of rank $g$ together with a choice of distinguished relation-free system of generators $L_1, \dots, L_g$ for it. In fact, a choice of marking for the Fuchsian group $\Gamma$ uniquely determines a marked Schottky group $(\Sigma;L_1,\dots,L_g)$. If $N$ is the smallest normal subgroup of $\Gamma$ containing $\{\alpha_1, \dots, \alpha_g ,\tau_1, \dots, \tau_{n_e},\kappa_1, \dots, \kappa_{n_p}\}$ then $\Sigma$ is isomorphic to the quotient group $\Gamma/N$. There is also a notion of equivalence between two marked Schottky groups: $(\Sigma; L_1, \dots, L_g)$ is said to be equivalent to $(\Sigma'; L'_1, \dots, L'_g)$ if and only if there exists a M\"{o}bius transformation $\varsigma \in \PSLC$ such that $L'_i = \varsigma  L_i \varsigma^{-1}$ for all $i=1,\dots,g$. Then, the \emph{Schottky space} $\schottky_{g}$ can be defined as the space of equivalence classes of marked Schottky groups of genus $g$. Now, we can introduce the generalized Schottky space $\schottky_{g,n}(\boldsymbol{m})$ of Riemann orbisurfaces, both with and without punctures. It is regarded as a holomorphic fibration $\jmath : \schottky_{g,n}(\boldsymbol{m}) \to \schottky_{g}$, where the fibers represent configuration spaces of $n$ labeled points.\footnote{For details, see Section \ref{subsec:Schottkyspace}.} Denote by $L_1, \dots, L_g$ the system of generators in $\Sigma$ corresponding to the cosets $\beta_{1}N,\dots,\beta_{g}N$ in $\Gamma$. Normalizing the marked Schottky group $(\Sigma; L_1, \dots, L_g)$,\footnote{By \emph{normalizing}, we mean using the equivalence notion of marked Schottky groups to set the attracting fixed points of generators $L_1$ and $L_2$ as well as the repelling fixed-point of generator $L_1$ equal to 0, 1, and $\infty$ respectively; see Section \ref{subsec:Schottkyspace} for more details.} we thereby associate with each marked Schottky group (equivalently, with each $O \cong \singrigon/\Sigma$) a point in the generalized Schottky space $\schottky_{g,n}(\boldsymbol{m})$.

The Schottky uniformization of an orbisurface $O$ is connected with the Fuchsian uniformization of it by the following commutative diagram
\begin{equation}\label{globalcoords}
\begin{tikzcd}[sep=large]
\UHP^{} \ar[r, "J"] \ar[rd, "\pi_{\Gamma}" '] 
&\singrigon \ar[d,"\pi_{\Sigma}"] \\
& O
\end{tikzcd}
\end{equation}
where each of the mappings is a holomorphic orbifold covering (complex-analytic covering). The normal subgroup $N$ of $\Gamma$ corresponds to the group of deck transformations of the covering $J$. A deck transformation is a homeomorphism $\text{deck}: \mathbb{H}\rightarrow \mathbb{H}$, such that the diagram of the maps
\begin{equation*}
\begin{tikzcd}[sep=large]
\UHP \ar[rr, "\text{deck}"] \ar[rd, "J" '] &
&\mathbb{H} \ar[ld,"J"] \\
& \singrigon &
\end{tikzcd}
\end{equation*}	
commutes. The set of deck transformations forms a group which is called the automorphism group of covering map, $\text{Aut}(J)$ (See Section \ref{covering}). Accordingly, the mapping $J$ can be regarded as a (meromorphic) function on $\UHP$, which is automorphic with respect to $N$ --- i.e. $J \circ N = J$. Moreover, $J \circ \beta_i = L_i \circ J$ for all $i=1,\dots,g$.
\subsection{Projective Connections and Energy-Momentum Tensor}  
Let $O=(X_O,\brdiv)$ be a hyperbolic Riemann orbisurface with signature $\left(g;m_1,\dots, m_{n_e};n_p\right)$, and let $\big\{(U_a, u_a)\big\}$ be a complex-analytic atlas on the underlying Riemann surface $X_O$ with local coordinates $u_a: U_a \to \cmpx$ and transition functions $u_a = g_{ab} \circ u_b$ on overlaps $U_a \cap U_b$. A (meromorphic) \emph{projective connection} on $O$ is a collection $R=\{r_a\}_{a \in A}$ of holomorphic functions $r_a$ defined on each $U_a \cap X_O^{\text{reg}}$ that satisfy
\begin{equation}\label{projconn}
r_b = r_a \circ g_{ab} \, (g'_{ab})^2 + \text{Sch}\left(g_{ab};u_b\right),
\end{equation}
on every intersection $U_a \cap U_b \cap X_O^{\text{reg}}$ and are compatible with $\brdiv$ --- i.e. if $x_i \in U_a \cap \sing(O)$ and $u_a(x_i) = 0$, we have
\begin{equation}\label{projconnasym}
r_a(u_a) = \frac{1-1/m_i^2}{2 u_a^2} + \order{|u_a|^{-1}} \quad \text{as} \quad u_a \to 0.
\end{equation}
In the above definition of projective connections, $\text{Sch}\left(f;z\right)$ denotes the \emph{Schwarzian derivative}
\begin{equation}\label{Sch}
\text{Sch}\left(f;z\right):= \frac{f'''}{f'} - \frac{3}{2} \left(\frac{f''}{f'}\right)^2,
\end{equation}
of a holomorphic function $f(z)$ and can be intuitively viewed as measuring the failure of $f(z)$ to be the restriction of a M\"{o}bius transformation. Such meromorphic projective connections are in one-to-one correspondence with $\CPl$-structures on $O$ (see Section \ref{sec:projconn} for more details).
\begin{remark}
	We note that the coefficient $h_i := 1-1/m_i^2$ of the leading singular term in the asymptotic behavior \eqref{projconnasym} of projective connections near each singular point $x_i \in \sing(O)$ does \emph{not} depend on the choice of chart or complex coordinate $u$.
\end{remark}
The above remark means that the difference between two projective connections is a (meromorphic) \emph{quadratic differential} on $O$ with only simple poles --- i.e. a collection $Q=\{q_a\}_{a \in A}$ of holomorphic functions on each open subset $U_a \cap X_O^{\text{reg}}$ with the transformation law
\begin{equation}
q_b  = q_a \circ g_{ab} (g'_{ab})^2,
\end{equation}
and the asymptotic behavior $q_a(u_a) = \order{|u_a|^{-1}}$ near each singular point $x_i \in U_a \cap \, \sing(O)$ with $u_a(x_i)=0$. Conversely, we can add a meromorphic quadratic differential to a given projective connection to obtain a new projective connection. Since we know that each Riemann orbisurface has at least one $\CPl$-structure, i.e. the one given by Poincar\'{e}-Koebe uniformization, we have the following (see \cite{biswas2006orbifold}): 
\begin{proposition}[Biswas]\label{Biswas}
	The space of all $\CPl$-structures on $O$, denoted by $\mathcal{P}(O)$, is an affine space for the vector space of all meromorphic quadratic differentials on $O$ with at most simple poles at singularities.
\end{proposition}

These meromorphic projective connections on $O$ have the following physical interpretation: For the hyperbolic metric $ds_{\text{hyp}}^{2} =\big\{e^{\varphi_a(u_a,\bar{u}_a)}|du_{a}|^{2}\big\}_{a \in A}$ on $O$, let us define the following functions on each open subset $U_{a}$:
\begin{equation} \label{STE}
T_{a} =  \partial^{2}_{u_{a}}\varphi_{a} -\tfrac{1}{2}(\partial_{u_{a}}\varphi_{a})^{2}
\quad\text{and}\quad 
\bar{T}_{a} =\partial^{2}_{\bar{u}_{a}}\varphi_{a} -\tfrac{1}{2}(\partial_{\bar{u}_{a}}\varphi_{a})^{2}.
\end{equation}
The collections $T_{\varphi} =\{T_{a}\}_{a\in A}$ and $\bar{T}_{\varphi} =\{\bar{T}_{a}\}_{a\in A}$ are the $(2,0)$ and $(0,2)$ components of the \emph{classical energy-momentum tensor} on $O$ and are associated with the \emph{quasi-conformal transformations} of the hyperbolic metric (see e.g. \cite[Appx.~B]{Takhtajan:2005md} for more details). In addition, the functions $T_a\big(\varphi_{a}(u_a,\bar{u}_a)\big)$ satisfy the \emph{conservation law}
\begin{equation}
\partial_{\bar{u}_a} T_a\big(\varphi_{a}(u_a,\bar{u}_a)\big) = 0,
\end{equation}
on each open subset $U_a \cap X_{O}^{\text{reg}} \subset X_O^{\text{reg}}$ and, as a result of \eqref{metricasymp}, have the asymptotic behavior
\begin{equation}
T_a(u_a) = \frac{h_i}{2 u_a^2} + \order{|u_a|^{-1}} \quad \text{as} \quad u_a \to 0,
\end{equation}
near each singular point $x_i \in U_a \cap \, \sing(O)$ with $u_a(x_i)=0$ and $h_i=1-1/m_i^2$. The property that functions $T_a(u_a)$ are meromorphic expresses the fact that the energy-momentum tensor for the classical Liouville theory is traceless and the coefficients $h_i/2$ appearing in the above asymptotics have the interpretation of \emph{conformal weights} of Liouville vertex operators corresponding to each singular point \cite{Knizhnik:1987xp}. Finally, it follows from \eqref{phitransfrom} that on every overlap $U_{a}\cap U_{b} \cap X_O^{\text{reg}}$
\begin{equation} \label{proj-conn}
T_{b}=T_{a}\circ g_{ab}\,(g'_{ab})^{2} + \text{Sch}\left(g_{ab};u_b\right),
\end{equation}
which means that, by definition, $T_{\varphi}(u)$ is a meromorphic projective connection on $O$. Since the hyperbolic metric $ds_{\text{hyp}}^{2}$ on $O$ is a push-forward of the Poincar\'{e} metric on $\UHP$ by the covering map $\pi_{\Gamma}: \UHP \rightarrow O$, a simple computation gives $T_{\varphi}(u) =\big\{\text{Sch}\left(\pi_{\Gamma}^{-1};u_a\right)\big\}_{a\in A}$. The multi-valued analytic function  $\pi_{\Gamma}^{-1}: O \to \UHP$ is a locally univalent linear polymorphic function on $O$ (this means that its branches are connected by fractional linear transformations in $\Gamma$) and, using the property $\text{Sch}\left(\varsigma(z);z\right) = 0$ for all $\varsigma \in \PSLC$, as well as the Caley identity
\begin{equation}\label{Caleyid}
\text{Sch}\left(f \circ g;z\right) = \text{Sch}\left(f ;g\right) (g')^2 + \text{Sch}\left(g;z\right),
\end{equation}
it is easy to verify directly that $\text{Sch}\left(\pi_{\Gamma}^{-1};u_a\right)$  are well-defined functions on each subset $U_a \cap X_O^{\text{reg}}$, which satisfy \eqref{proj-conn}. Slightly abusing notations, we will write $T_{\varphi}(u)=\operatorname{Sch}(\pi_{\Gamma}^{-1})$ and call it the \emph{Fuchsian projective connection} on $O$. Similarly, the Schottky global coordinate given by the orbifold covering map $\pi_{\Sigma}: \singrigon \to O \cong \singrigon/\Sigma$, produces the so-called \emph{Schottky projective connection} $\operatorname{Sch}(\pi_{\Sigma}^{-1})$ on $O$. 
\begin{remark}
	While the Fuchsian projective connection is canonically determined by the Riemann orbisurface $O$ and does not depend on the choice of marking for $\Gamma$, the Schottky projective connection is defined only for marked Riemann orbisurfaces and is uniquely determined by the normal subgroup $N \subset \Gamma$ introduced in the previous subsection.
\end{remark}
It follows from the commutative diagram \eqref{globalcoords} and the Caley identity \eqref{Caleyid} for the Schwarzian derivative that
the Fuchsian and Schottky projective connections are related by
\begin{equation}
\text{Sch}\left(\pi_{\Gamma}^{-1};u_a\right) = \text{Sch}\left(J^{-1};w\right) \circ \pi_{\Sigma}^{-1} \,\,(\partial_{u_{a}}\pi_{\Sigma}^{-1})^2 + \text{Sch}\left(\pi_{\Sigma}^{-1};u_a\right) \quad \text{for all} \quad a \in A,
\end{equation}
where $w$ is the global coordinate on $\Omega$. Therefore, the collection
\begin{equation}
\Big\{\text{Sch}\left(J^{-1};w\right) \circ \pi_{\Sigma}^{-1} \, \big(\partial_{u_{a}}\pi_{\Sigma}^{-1}\big)^2 \Big\}_{a \in A}
\end{equation}
is a meromorphic quadratic differential on $O$ and $T_{\varphi}(w) := \text{Sch}\left(J^{-1};w\right)$ is a meromorphic automorphic form of weight four for the Schottky group $\Sigma$ --- i.e. $T_{\varphi}\big(\sigma(w)\big)\, (\sigma')^2 = T_{\varphi}(w)$ for all $\sigma \in \Sigma$.

With the above explanations in mind, we are now ready to give a summary of our main results: Let $\teich_{g,n}(\boldsymbol{m})$ be the Teichm\"{u}ller space of marked Riemann orbisurfaces of genus $g > 1$ defined as the space of all equivalence classes of marked Riemann orbisurfaces with signature $(g;m_1,\dots,m_{n_e};n_p)$.\footnote{For details, see Section \ref{subsec:teichmuller}.} The affine spaces $\mathcal{P}(O)$ for varying Riemann orbisurfaces $O$ glue together to an affine bundle $\mathscr{P}_{g,n}(\boldsymbol{m}) \to \teich_{g,n}(\boldsymbol{m})$, modeled over holomorphic cotangent bundle of $\teich_{g,n}(\boldsymbol{m})$. The Fuchsian projective connection $\operatorname{Sch}(\pi_{\Gamma}^{-1})$ gives a canonical section of the affine bundle $\mathscr{P}_{g,n}(\boldsymbol{m}) \to \teich_{g,n}(\boldsymbol{m})$, while the Schottky projective connection $\operatorname{Sch}(\pi_{\Sigma}^{-1})$ gives a canonical section of the affine bundle $\mathscr{P}_{g,n}(\boldsymbol{m}) \to \schottky_{g,n}(\boldsymbol{m})$. Their difference $\mathscr{Q}:= \operatorname{Sch}(\pi_{\Gamma}^{-1}) - \operatorname{Sch}(\pi_{\Sigma}^{-1})$ can be viewed as a (1,0)-form on the Schottky space $\schottky_{g,n}(\boldsymbol{m})$ and has the following interesting properties (see Theorems~\ref{mainthrm1} and \ref{mainthrm2}): Let us denote by $\omega_{\text{WP}}$ the symplectic form of the WP metric, and by  $\omega_{\text{TZ},i}^{\text{ell}}$, $\omega_{\text{TZ},j}^{\text{cusp}}$, the symplectic forms of $i^{th}$-elliptic, $j^{th}$-cuspidal TZ metrics on generalized Schottky space $\schottky_{g,n}(\boldsymbol{m})$. Additionally, let $\partial$ and $\bar{\partial}$ denote the (1,0) and (0,1) components of the exterior differential $\dd{}$ on the Schottky space $\schottky_{g,n}(\boldsymbol{m})$ --- i.e. $\dd{} = \partial + \bar{\partial}$. Then, we have:
\begin{enumerate}
	\item $\mathscr{Q}$ is $\partial$-exact --- i.e. there exists a smooth function $\Gpotential_{\boldsymbol{m}}: \schottky_{g,n}(\boldsymbol{m}) \to \mathbb{R}$ such that
	\begin{equation}\label{firstvarr}
	\mathscr{Q} = \frac{1}{2} \partial \Gpotential_{\boldsymbol{m}}.
	\end{equation}
	\item $\mathscr{Q}$ is a $\bar{\partial}$-antiderivative (hence, a $\dd{}$-antiderivative) of the following combination of WP and TZ symplectic forms on $\schottky_{g,n}(\boldsymbol{m})$:
	\begin{equation}
	\bar{\partial}\mathscr{Q} = - \sqrt{-1}\left(\omega_{\text{WP}}-\frac{4\pi^2}{3} \omega^{\text{cusp}}_{\text{TZ}} - \frac{\pi}{2} \sum_{i=1}^{n_e} m_i h_i
	\hspace{1mm}\omega^{\text{ell}}_{\text{TZ},i}\right).
	\end{equation}
	Here, $\omega^{\text{cusp}}_{\text{TZ}}= \sum_{j=1}^{n_p}\omega^{\text{cusp}}_{\text{TZ},j}$. 
	\item 
	It follows immediately from the above two statements that the function $-\Gpotential_{\boldsymbol{m}}$ is a K\"{a}hler potential for the special combination of WP and TZ  metrics on $\schottky_{g,n}(\boldsymbol{m})$:
	\begin{equation}\label{secondvarr}
	- \bar{\partial} \partial \Gpotential_{\boldsymbol{m}} = 2\sqrt{-1}\left(\omega_{\text{WP}}-\frac{4\pi^2}{3} \omega^{\text{cusp}}_{\text{TZ}} - \frac{\pi}{2} \sum_{i=1}^{n_e} m_i h_i
	\hspace{1mm}\omega^{\text{ell}}_{\text{TZ},i}\right).
	\end{equation}
\end{enumerate}

Before ending this section, and in order to avoid confusion in the remainder of this manuscript, we feel the need to talk about our notation for coordinate functions on different spaces: In this section, we have used $u_a$ to denote the coordinate function on each open subset $U_a \subset X_{O}$. However, in what follows, we will always use $w$ to denote the coordinate function on $X_{O}^{\text{reg}}$ when $O$ has genus $g=0$.\footnote{Note that, in this situation, $X_O \cong \hat{\cmpx}$ needs to be covered with at least two coordinate charts while $X_{O}^{\text{reg}}\subset \cmpx$ can be covered with only one chart.} When the orbifold Riemann surface $O$ has genus $g>1$, $\{u_a\}_{a \in A}$ denotes the set of coordinate functions on $X_{O}$ while $w$ is used to denote the coordinate function on $\Omega^{\text{reg}} \subset \cmpx$. This notation is meant to be suggestive of the fact that the difference between Schottky and Fuchsian uniformizations of $O \cong \singrigon/\Sigma$ is effectively equivalent to Fuchsian uniformization of the planar orbifold $\singrigon$. Finally, throughout this manuscript, $z$ has always been used to denote the coordinate function on the upper half-plane $\UHP$. 

\section{Geometry of Tiechm\"{u}ller, Moduli, and Schottky Spaces}\label{geometry}
In this section, we will recall some well-known facts about the deformation theory of Ahlfors and Bers. More details can be found in \cite{Ahlfors_quasiconformal_06, Bers_1960,nag1988complex}.
\subsection{Teichm\"{u}ller Space $\teich(\Gamma)$}\label{subsec:teichmuller}
 Let $\Gamma$ be a finitely generated Fuchsian group of the first kind that uniformizes the hyperbolic orbifold Riemann surface $O$ with the signature $(g;m_1, \dots, m_{n_e};n_p)$. In this situation, the Teichm\"uller space of Riemann orbisurfaces can be equivalently described as the Teichm\"uller space $\teich(\Gamma)$ of Fuchsian groups with signature $(g;m_1, \dots, m_{n_e};n_p)$ --- i.e. the space of all equivalence classes of marked Fuchsian groups with signature $(g;m_1, \dots, m_{n_e};n_p)$. 
 
 A \emph{Beltrami differential} for $\Gamma$ is defined as $\mu:= \mu(z) \, \partial_{z} \dd{\bar{z}}$ where $\mu(z)$ is a complex-valued bounded measurable function on $\UHP$ with the property that 
\begin{equation}\label{FuchsianBeltrami}
\mu(\gamma z) \frac{\overline{\gamma'(z)}}{\gamma'(z)} = \mu(z) \qquad \text{for all} \qquad \gamma \in \Gamma \quad \text{and} \quad z \in \UHP.
\end{equation}
We will denote by $\diffspace^{-1,1}(\UHP,\Gamma)$, the complex Banach space of Beltrami differentials for $\Gamma$. Now, let $\deformspace(\Gamma)$ denote the open unit ball in $\diffspace^{-1,1}(\UHP,\Gamma)$, in the sense of the $L^{\infty}$-norm:
\begin{equation}
\deformspace(\Gamma) \equiv \left\{\mu \in \diffspace^{-1,1}(\UHP,\Gamma) \, \Big| \, \norm{\mu}_{\infty} := \sup_{z \in \UHP} |\mu(z)| < 1 \right\}.
\end{equation}
For each $\mu \in \deformspace(\Gamma)$, the \emph{Beltrami equation}
\begin{equation}\label{Beltramieq}
\partial_{\bar{z}} f^{\mu}(z) = \mu(z) \, \partial_z f^{\mu}(z), \qquad z \in \UHP,
\end{equation}
is solvable in the class of quasi-conformal homeomorphisms of $\UHP$ onto itself, and any two solutions are connected by a linear fractional transformation in $\PSLR$. Let $f^{\mu}$ be a solution of Beltrami equation \eqref{Beltramieq} that fixes points $0, 1, \infty$ and define $\Gamma^{\mu} := f^{\mu} \circ \Gamma \circ (f^{\mu})^{-1}$ where $\Gamma^{\mu}$ is a Fuchsian group with the same signature as $\Gamma$. Thus, each element $\mu \in \deformspace(\Gamma)$ gives a faithful representation $\varrho_{\mu}$ of $\Gamma$ in $\PSLR$ with $2g$ hyperbolic generators $\alpha_i^{\mu} := f^{\mu} \circ \alpha_i \circ (f^{\mu})^{-1}$ and $\beta_i^{\mu} := f^{\mu} \circ \beta_i \circ (f^{\mu})^{-1}$, $n_p$ parabolic elements $\kappa_i^{\mu} := f^{\mu} \circ \kappa_i \circ (f^{\mu})^{-1}$, as well as $n_e$ elliptic elements $\tau_i^{\mu} := f^{\mu} \circ \tau_i \circ (f^{\mu})^{-1}$ of orders $m_1, \dots, m_{n_e}$ respectively, satisfying the single relation
\begin{equation}
\alpha_1^{\mu} \beta_1^{\mu} (\alpha_1^{\mu})^{-1} (\beta_1^{\mu})^{-1} \dotsm \alpha_g^{\mu} \beta_g^{\mu} (\alpha_g^{\mu})^{-1}  (\beta_g^{\mu})^{-1} \tau_1^{\mu} \dotsm \tau_{n_e}^{\mu} \kappa_1^{\mu} \dotsm \kappa_{n_p}^{\mu} = \mathbbm{1}.
\end{equation}
But one needs to define the equivalence classes of representations $\varrho_{\mu}$ since two representations $\varrho_{\mu_1}$ and $\varrho_{\mu_2}$ are called equivalent if they differ by an inner automorphism of $\PSLR$ --- i.e. if $\varrho_{\mu_2} = \varsigma \varrho_{\mu_1} \varsigma^{-1}$ for a M\"obius transformation $\varsigma \in \PSLR$. Accordingly, the \emph{Teichm\"{u}ller space} $\teich(\Gamma)$ is defined to be the set of all equivalence classes of representations $\varrho_{\mu}: \Gamma \to \PSLR$, $\mu \in \deformspace(\Gamma)$. In other words,
\begin{equation}\label{teichmullerspcae}
\teich(\Gamma) \cong \deformspace(\Gamma)/\sim
\end{equation}
where $\mu_1 \sim  \mu_2$ if and only if $f^{\mu_1} \circ \gamma \circ (f^{\mu_1})^{-1} = f^{\mu_2} \circ \gamma \circ (f^{\mu_2})^{-1}$ for all $\gamma \in \Gamma$ (equivalently, if $f^{\mu_1} \big|_{\mathbb{R}} = f^{\mu_2} \big|_{\mathbb{R}}$). The base point of  $\teich(\Gamma)$ is defined by $\mu =0 $ and it corresponds to the group $\Gamma$. Last but not least, the projection $\Phi: \deformspace(\Gamma) \to \teich(\Gamma)$ induces a natural complex-analytic manifold structure on $\teich(\Gamma)$ which will be described in subsection~\ref{subsub:complexstructure}. 
\begin{remark}\label{remark:nobase}
	Let $\Gamma_1$ and $\Gamma_2$ be two cofinite Fuchsian groups with the same signature, and let $f: \UHP \to \UHP$ be a quasi-conformal mapping such that $\Gamma_2 = f \circ \Gamma_1 \circ f^{-1}$. Then $f$ induces a mapping $f^{\ast}: \teich(\Gamma_1) \to \teich(\Gamma_2)$ according to the formula $\varrho_{\mu_1} \mapsto \varrho_{\mu_2}$ where $\mu_1 \in \deformspace(\Gamma_1)$ and
	\begin{equation}
	\mu_2 = \left(\frac{\mu_1-(\partial_{\bar{z}} f/\partial_z f)}{1-\mu_1 \overline{(\partial_{\bar{z}} f/\partial_z f)}} \frac{\partial_z f}{\partial_z \bar{f}}\right) \circ f^{-1} \in \deformspace(\Gamma_2).
	\end{equation}
	This mapping is a complex-analytic isomorphism in the natural complex structure on $\teich(\Gamma_1)$ and $\teich(\Gamma_2)$, which makes a specific choice of base point inessential (see, e.g. \cite[Remark~3]{Zograf1988ONLE}).
\end{remark}
\begin{remark}
	Teichm\"{u}ller space $\teich(\Gamma)$ can be interpreted as the Teichm\"{u}ller space of marked Riemann orbisurfaces with signature $(g;m_1,\dots,m_{n_e};n_p)$ by assigning to each point $\Phi(\mu) \in \teich(\Gamma)$ a marked Riemann orbisurface  $O^{\mu} \cong [\UHP/\Gamma^{\mu}]$, with the orbisurface $O \cong [\UHP/\Gamma]$ playing the role of a base point. It follows from Remark~\ref{remark:nobase} that the choice of a base point is inessential and, for this reason, we will sometimes use the notation $\teich_{g,n}(\boldsymbol{m})$ to denote the Teichm\"uller space of marked Riemann orbisurfaces with signature $(g;m_1,\dots,m_{n_e};n_p)$.\footnote{It follows from the Bers-Greenberg theorem \cite{bers1971isomorphisms}, that the complex-analytic structure of $\teich(\Gamma)$ does not depend on the vector of orders --- i.e. there exists a complex-analytic isomorphism between $\teich_{g,n}(\boldsymbol{m})$ and the Teichm\"uller space $\teich_{g,n}$ of punctured Riemann surface $X_O^{\text{reg}}$. However, we will keep using the notation $\teich_{g,n}(\boldsymbol{m})$ for the Teichm\"uller space of Riemann orbisurfaces in order to emphasize that the natural K\"ahler structure and the action of orbifold mapping class group on this space does depend on $\boldsymbol{m} = (m_1,\dots,m_n)$ through dependence on the signature type of $O$.}
\end{remark}
\subsubsection{Complex Structure on $\teich(\Gamma)$}\label{subsub:complexstructure}
The complex structure on $\teich(\Gamma)$ is \emph{uniquely} characterized by the fact that the mapping $\Phi: \deformspace(\Gamma) \to \teich(\Gamma)$ is holomorphic. For a more explicit description of this canonical complex-analytic structure, we consider the space $\Hilbert^{2,0}(\UHP,\Gamma)$ of \emph{holomorphic quadratic differentials} (equivalently, \emph{holomorphic cusp forms} of weight $4$) for $\Gamma$.\footnote{By holomorphic cusp forms, we mean holomorphic $\Gamma$-automorphic forms on $\UHP$ with zero constant coefficient in their Fourier expansions near the cusps of $\Gamma$.} An arbitrary element $q \in \Hilbert^{2,0}(\UHP,\Gamma)$ has the form $q = q(z) \dd{z^2}$ where $q(z)$ is a bounded holomorphic function on $\UHP$ that transforms according to the rule $q(\gamma z) \gamma'(z)^2 = q(z)$ for all $\gamma \in \Gamma$.\footnote{As we will see in Subsec.~\ref{M0n}, any element $q \in \Hilbert^{2,0}(\UHP,\Gamma)$ corresponds to a \emph{meromorphic} quadratic differential $Q \in \Hilbert^{2,0}(O)$ --- i.e. a meromorphic (2,0)-tensor on $X_O$ with simple poles at singular points.} The dimension of the space of square-integrable meromorphic $k$-differentials on $O$, or cusp forms of weight $2k$ for $\Gamma$, is given by \emph{Riemann--Roch} formula for orbifolds:
\begin{equation}
\dim_{\cmpx} \Hilbert^{k,0}(O) = \left\{
\begin{split}
& (2k-1)(g-1)+\sum_{i=1}^{n_e} \left\lfloor k \left(1-\frac{1}{m_i}\right) \right\rfloor+(k-1)n_p, \qquad & k>1,\\
& g, & k=1,\\
& 1, & k=0,\\
& 0, & k<0,
\end{split}
\right.
\end{equation}
where $\lfloor \cdot \rfloor$ denotes the floor function (see Theorem 2.24 of \cite{Shimura_1971}). In particular, the dimension of the Hilbert space of cusp forms of weight 4 for $\Gamma$ is given by
\begin{equation}\label{dimquaddiff}
\dim_{\cmpx} \Hilbert^{2,0}(\UHP,\Gamma) = 3g-3+n_e+n_p = 3g-3 + n.
\end{equation}
The Kodaira-Serre paring
\begin{equation}\label{pairing}
\left(\mu,q\right) := \iint_{\fund(\Gamma)} \mu(z) q(z) \dd[2]{z},
\end{equation}
is well-defined on the product of $\Hilbert^{2,0}(\UHP,\Gamma)$ and $\diffspace^{-1,1}(\UHP,\Gamma)$. In the above equation, $\dd[2]{z} \equiv \frac{\sqrt{-1}}{2} \dd{z} \wedge \dd{\bar{z}} = \dd{(\Re z)} \wedge \dd{(\Im z)}$ and $\mathcal{F}(\Gamma)\subset \UHP$ denotes a fundamental domain for the Fuchsian group $\Gamma$. The subspace $\kerphi(\UHP,\Gamma) \subset \diffspace^{-1,1}(\UHP,\Gamma)$ on which this pairing is degenerate coincides with the kernel of the differential $\dd{\Phi}$ at $\mu = 0 \in \deformspace(\Gamma)$. Moreover, the space $\diffspace^{-1,1}(\UHP,\Gamma)/\kerphi(\UHP,\Gamma)$ and $\Hilbert^{2,0}(\UHP,\Gamma)$ are \emph{dual} with respect to the pairing~\eqref{pairing}. To realize $\diffspace^{-1,1}(\UHP,\Gamma)/\kerphi(\UHP,\Gamma)$ as a subspace of $\diffspace^{-1,1}(\UHP,\Gamma)$ we define the complex anti-linear mapping $\varLambda: \diffspace^{-1,1}(\UHP,\Gamma) \to \Hilbert^{2,0}(\UHP,\Gamma)$ with the help of the Bergman integral
\begin{equation}
\varLambda(\mu)(z) = \frac{12}{\pi} \iint_{\UHP} \frac{\overline{\mu(\xi)}}{(\bar{\xi}-z)^4} \dd[2]{\xi}, \qquad \mu \in \diffspace^{-1,1}(\UHP,\Gamma);
\end{equation}
its kernel coincides with $\kerphi(\UHP,\Gamma)$. The mapping $\varLambda^{\ast}: \Hilbert^{2,0}(\UHP,\Gamma) \to \diffspace^{-1,1}(\UHP,\Gamma)$ given by
\begin{equation}
\varLambda^{\ast}(q)(z) = (\Im z)^2 \, \overline{q(z)}, \qquad q \in \diffspace^{2,0}(\UHP,\Gamma),
\end{equation} 
and satisfies the condition $\varLambda \varLambda^{\ast} = \operatorname{id}$ on $\Hilbert^{2,0}(\UHP,\Gamma)$. In other words, $\varLambda^{\ast}$ splits the exact sequence
\begin{equation}
0 \to \kerphi(\UHP,\Gamma) \hookrightarrow \diffspace^{-1,1}(\UHP,\Gamma) \overset{\varLambda}{\to} \Hilbert^{2,0}(\UHP,\Gamma) \to 0.
\end{equation}
This enables us to realize $\diffspace^{-1,1}(\UHP,\Gamma)/\kerphi(\UHP,\Gamma)$ as the subspace $\Hilbert^{-1,1}(\UHP,\Gamma) = \Lambda^{\ast}\left(\Hilbert^{2,0}(\UHP,\Gamma)\right)$ of $\diffspace^{-1,1}(\UHP,\Gamma)$ with complex dimension $3g-3+n$: the space of so-called \emph{harmonic} Beltrami differentials. 

The fact that $\operatorname{Ker}\dd{\Phi} = \kerphi(\UHP,\Gamma)$ at $0 \in \deformspace(\Gamma)$ implies that $\Phi: \deformspace(\Gamma) \to \teich(\Gamma)$ maps a sufficiently small neighborhood of the point $0 \in \Hilbert^{-1,1}(\UHP,\Gamma) \cap \deformspace(\Gamma)$ injectively into $\teich(\Gamma)$ and can be regarded as a coordinate chart in a neighborhood of $\Phi(0) \in \teich(\Gamma)$. More explicitly, this coordinate chart can described as follows: Let $\mu_1, \dots, \mu_{3g-3+n}$ denote a basis in $\Hilbert^{-1,1}(\UHP, \Gamma)$ and let $\mu = t_1 \mu_1 +  \dotsi  + t_{3g-3+n} \, \mu_{3g-3+n}$ be any harmonic Beltrami differential with $\norm{\mu}_{\infty} <1$. Then, the correspondence $(t_1,\dots,t_{3g-3+n}) \mapsto \Phi(\mu)$ defines the so-called \emph{Bers coordinates} in a neighborhood of the origin $\Phi(0) \in \teich(\Gamma)$. The isomorphism $\teich(\Gamma) \cong \teich(\Gamma^{\mu})$ (see Remark~\ref{remark:nobase}) makes it possible to introduce similar coordinates in a neighborhood of an arbitrary point $\Phi(\mu) \in \teich(\Gamma)$. As a result, the \emph{holomorphic tangent space} to $\teich(\Gamma)$ at the point $\Phi(\mu)$ can be identified with $\Hilbert^{-1,1}(\UHP,\Gamma^{\mu})$ --- the complex vector space of harmonic Beltrami differentials for $\Gamma^{\mu}$. The pairing \eqref{pairing} lets us regard $\Hilbert^{2,0}(\UHP,\Gamma^{\mu})$, i.e. the vector space of holomorphic cusp forms of weight 4 for $\Gamma^{\mu}$, as the \emph{holomorphic cotangent space} to $\teich(\Gamma)$ at the point $\Phi(\mu)$. This collection of charts gives the natural complex structure mentioned in the beginning of this subsection, on the Tiechm\"{u}ller space $\teich(\Gamma)$. 

Finally, we point out that one can always associate $3g-3+n$ vector fields $\pdv{t_i}$ with the Bers' coordinates $(t_1,\dots,t_{3g-3+n})$ in a neighborhood of $\Phi(0) \in \teich(\Gamma)$. At any other point $\Phi(\mu)$ in this neighborhood, we have $\pdv{t_i}\big|_{\Phi(\mu)} = \mu_i^{^{_{\Phi(\mu)}}}$ where the harmonic Beltrami differentials $\mu_i^{\Phi(\mu)} \in \Hilbert^{-1,1}(\UHP,\Gamma^{\mu})$ are given by the formula
\begin{equation}
	\mu_i^{^{_{\Phi(\mu)}}} = \operatorname{Proj}_{_{\Hilbert^{-1,1}}} \left[\left(\frac{\mu_i}{1-|\mu|^2} \frac{\partial_z f^{\mu}}{\overline{\partial_z f^{\mu}}}\right) \circ (f^{\mu})^{-1}\right].
\end{equation}
Here, the mapping $\operatorname{Proj}_{_{\Hilbert^{-1,1}}} $ denotes a projection onto the subspace $\Hilbert^{-1,1}(\UHP,\Gamma^{\mu})$ of harmonic Beltrami differentials. Moreover, let $q_1, \dots, q_{3g-3+n}$ be the basis in $\Hilbert^{2,0}(\UHP,\Gamma)$, dual to the basis $\mu_1, \dots, \mu_{3g-3+n}$ for $\Hilbert^{-1,1}(\UHP,\Gamma)$. Then, at an arbitrary point $\Phi(\mu)$ in a neighborhood of the origin, holomorphic 1-forms $\dd{t_i}$ are represented by the holomorphic quadratic differentials $q_{i}^{^{_{\Phi(\mu)}}}$ --- i.e. $\dd{t_i}\big|_{\Phi(\mu)} = q_{i}^{^{_{\Phi(\mu)}}}$ where the basis $q_1^{^{_{\Phi(\mu)}}}, \dots, q_{3g-3+n}^{^{_{\Phi(\mu)}}} \in \Hilbert^{2,0}(\UHP,\Gamma^{\mu})$ has the property
\begin{equation}
	\operatorname{Proj}_{_{\Hilbert^{2,0}}} \left[q_i^{^{_{\Phi(\mu)}}} \circ f^{\mu}  \, (\partial_z f^{\mu})^2\right] = q_i.
\end{equation}
In the above equation, $\operatorname{Proj}_{_{\Hilbert^{2,0}}}$ denotes a projection onto the subspace $\Hilbert^{2,0}(\UHP,\Gamma)$. 
\subsubsection{Variational Formulas}\label{subsub:variationalformulas}
In order to further explore the complex-analytic structure of the Teichm\"{u}ller space, variational formulas of the hyperbolic metric $\rho(z) |\dd{z}|^2$ on $\UHP$ play a significant role. Let $\phi^{\varepsilon} \in \diffspace^{k,\ell}(\UHP,\Gamma^{\varepsilon \mu})$ be a smooth family of automorphic forms of weight $(2k,2\ell)$ where $\mu \in \Hilbert^{-1,1}(\UHP,\Gamma)$ denotes a harmonic Beltrami differential and $\varepsilon \in \cmpx$ is a sufficiently small parameter. We denote by $\left(f^{\varepsilon \mu}\right)^{\ast}(\phi^{\varepsilon}) $ the pullback of the automorphic form $\phi^{\varepsilon}$ with the unique diffeomorphism $f^{\varepsilon \mu}: \UHP \to \UHP$ that satisfies the Beltrami equation $\partial_{\bar{z}} f^{\varepsilon \mu}= (\varepsilon \mu)\, \partial_z f^{\varepsilon \mu}$ and fixes the points $0,1,\infty$. We have
\begin{equation}\label{pullback}
\left(f^{\varepsilon \mu}\right)^{\ast}(\phi^{\varepsilon}) = \phi^{\varepsilon} \circ f^{\varepsilon \mu} \left(\pdv{f^{\varepsilon \mu}}{z}\right)^{k} \left(\overline{\pdv{f^{\varepsilon \mu}}{z}}\right)^{\ell} \in \diffspace^{k,\ell}(\UHP,\Gamma).
\end{equation}
In particular, for the density $\rho(z) =(\Im z)^{-2}$ of the Poincar\'{e} metric, considered as a family of $(1,1)$-tensors, one has
\begin{equation}
(f^{\varepsilon\mu})^{\ast}(\rho) = \frac{|\partial_z f^{\varepsilon\mu}|^2}{\left(\Im f^{\varepsilon\mu}\right)^2}.
\end{equation}

Let the Lie derivatives of the family $\phi^{\varepsilon}$ in holomorphic and anti-holomorphic tangential directions, $\mu$ and $\bar{\mu}$, be defined as
\begin{equation}
\left\{
\begin{split}
&\Lie_{\mu} \phi \equidef  \left. \pdv{\varepsilon}\right|_{\varepsilon=0}(f^{\varepsilon \mu})^{\ast}(\phi^{\varepsilon}) \in \diffspace^{k+1,\ell}(\UHP,\Gamma),\\
& \Lie_{\bar{\mu}} \phi \equidef  \left. \pdv{\bar{\varepsilon}}\right|_{\varepsilon=0}(f^{\varepsilon \mu})^{\ast}(\phi^{\varepsilon}) \in \diffspace^{k,\ell+1}(\UHP,\Gamma).
\end{split}
\right.
\end{equation}
The first variational formula for $\rho(z)$ is given by the following lemma due to Ahlfors \cite{Ahlfors1961SomeRO}: 
\begin{lemma}[Ahlfors]\label{Lemma:Ahlfors}
	For any $\mu \in \Hilbert^{-1,1}(\UHP,\Gamma)$, the Lie derivatives of the density $\rho(z)$ of the Poincar\'{e} metric in both holomorphic and anti-holomorphic tangential directions vanish:
	\begin{equation*}
	\Lie_{\mu} \rho = \Lie_{\bar{\mu}} \rho  = 0.
	\end{equation*}
\end{lemma}\hspace*{-24pt}
For the second variation of $\rho$, the following formula was obtained by Wolpert (see \cite[Theorem~3.3 ]{Wolpert1986ChernFA}):
\begin{equation}\label{WolpertSecVar}
\Lie_{\mu\overline{\mu'}} \rho \equidef \left. \pdv{}{\varepsilon_1}{\bar{\varepsilon}_2}\right|_{\varepsilon_1 = \varepsilon_2 =0} (f^{\varepsilon_1 \mu + \varepsilon_2 \mu'})^{\ast}(\rho) = \frac{1}{2} \rho \left(\Delta_0 + \frac{1}{2}\right)^{-1}(\mu \overline{\mu'}) \equiv \frac{1}{2} \rho \cdot  f_{\mu \overline{\mu'}},
\end{equation}
where $\mu, \mu' \in \Hilbert^{-1,1}(\UHP,\Gamma)$. The $\Gamma$-automorphic function $f_{\mu \overline{\mu'}}$ is uniquely determined by
\begin{equation}\label{fmunu}
\left(\Delta_0  + \frac{1}{2}\right) f_{\mu \overline{\mu'}} = \mu \overline{\mu'} \quad \text{and} \quad \iint_{\fund(\Gamma)} |f_{\mu \overline{\mu'}}|^2 \rho(z) \dd[2]{z} < \infty,
\end{equation}
where $\Delta_0 :=  - \rho(z)^{-1} \pdv{}{z}{\bar{z}}$ is the \emph{Laplace operator} of the hyperbolic metric acting on $\Hilbert^{0,0}(\UHP,\Gamma)$.\footnote{See section 2 of \cite{ZT_localindextheorem_1991} for more detailed exposition.}
\subsubsection{K\"{a}hler Metrics on $\teich(\Gamma)$}\label{KahlerMetrics}
\begin{itemize}
	\item \emph{Weil-Petersson metric}. Together with the complex anti-linear isomorphism $q(z) \mapsto \mu(z) = \rho(z)^{-1} \, \overline{q(z)}$, the pairing~\eqref{pairing} defines the \emph{Petersson inner product} on $T_{\Phi(0)} \teich(\Gamma) \cong \Hilbert^{-1,1}(\UHP,\Gamma)$:
	\begin{equation}\label{Peterssoninnerprod}
	\langle\mu_1,\mu_2\rangle_{\text{WP}} = \iint_{\fund(\Gamma)} \mu_1(z) \, \overline{\mu_2(z)} \, \rho(z) \dd[2]{z}, \qquad \mu_1,\mu_2 \in \Hilbert^{-1,1}(\UHP,\Gamma).
	\end{equation}
	The Petersson inner product on the tangent spaces determines the \emph{Weil-Petersson} K\"{a}hler metric on $\teich(\Gamma)$. Its K\"{a}hler $(1,1)$-form is a symplectic form $\omega_{\text{WP}}$ on $\teich(\Gamma)$
	\begin{equation}
	\omega_{\text{WP}}(\mu_1,\bar{\mu}_2) = \frac{\sqrt{-1}}{2}\iint_{\fund(\Gamma)} \left(\mu_1(z) \, \overline{\mu_2(z)} - \overline{\mu_1(z)} \, \mu_2(z) \right)\, \rho(z) \dd[2]{z},
	\end{equation}
	where $\mu_1,\mu_2 \in T_{\Phi(0)}\teich(\Gamma)$. It is worth mentioning that the Weil-Petersson metric is both invariant under the Teichm\"{u}ller modular group $\modular(\Gamma)$ and real-analytic.
	\item \emph{Cuspidal Takhtajan-Zograf metric}. In \cite{ZT_Selberg_1988,ZT_localindextheorem_1991}, a new K\"{a}hler metric on $\teich(\Gamma)$ was introduced by Takhtajan and Zograf for the cases that the Fuchsian group $\Gamma$ has $n_p>0$ parabolic elements. Let us indicate the fixed points of the parabolic generators $\kappa_1,\dots,\kappa_{n_p}$ by $z_{n_e+1},\dots,z_{n} \in \mathbb{R} \cup \{\infty\}$.\footnote{Note that $n_e+n_p=n$.} For each $i=1,\dots,n_p$ denote by $\langle \kappa_i \rangle$ the cyclic subgroup of $\Gamma$ generated by $\kappa_i$, and let $\varsigma_i \in \PSLR$ be such that $\varsigma_i(\infty) = z_{n_e+i}$ and $\varsigma_i^{-1} \kappa_i \varsigma_i = \mqty(1 & \pm 1 \\ 0 & 1)$. Let $E_i(z,s)$ be the \emph{Eisenstein-Mass series} associated with the cusp $z_{n_e+i}$, which is defined as (see section 3 in \cite{Teo_2021})
	\begin{equation}\label{EisentienMassSeries}
	E_i(z,s) = \sum_{\Gamma\backslash \gamma \in \langle \kappa_i \rangle} \Im\left(\varsigma_i^{-1} \gamma z\right)^s .
	\end{equation}
	The series is absolutely convergent for $\Re s >1$, is positive for $s=2$ and satisfies the equation
	\begin{equation}
	\Delta_0 E_i(z,s) = \frac{1}{4} s(1-s) E_i(z,s).
	\end{equation}
	The inner product
	\begin{equation}
	\langle \mu_1, \mu_2 \rangle^{\text{cusp}}_{\text{TZ},i} = \iint_{\fund(\Gamma)} \mu_1(z) \, \overline{\mu_2(z)} E_i(z,2) \rho(z) \dd[2]{z}, \qquad i=1,\dots,n_p, 
	\end{equation}
	in $\Hilbert^{-1,1}(\UHP,\Gamma)$, and the corresponding inner products in all $\Hilbert^{-1,1}(\UHP,\Gamma^{\mu})$, determines another Hermitian metric on $\teich(\Gamma)$ which is K\"{a}hler for each $i=1,\dots,n_p$. The metric $\langle \cdot , \cdot \rangle^{\text{cusp}}_{\text{TZ}} = \langle \cdot , \cdot \rangle^{\text{cusp}}_{\text{TZ},1}+ \cdots + \langle \cdot , \cdot \rangle^{\text{cusp}}_{\text{TZ},n_p}$, called the \emph{cuspidal Takhtajan-Zograf (TZ)} metric, is invariant with respect to the Teichm\"{u}ller modular group $\modular(\Gamma)$ (see subsection \ref{subsec:moduli}). Let 
	\begin{equation}
	\omega_{\text{TZ},i}^{\text{cusp}} = \frac{\sqrt{-1}}{2} \sum_{j,k=1}^{3g-3+n}  \langle \mu_j, \mu_k \rangle^{\text{cusp}}_{\text{TZ},i} \hspace{1mm}\dd{t_j} \wedge \dd{\bar{t}_k}
	\end{equation}
	be the symplectic form of the $i$-th cuspidal TZ metric and also define $\omega_{\text{TZ}}^{\text{cusp}} = \omega^{\text{cusp}}_{\text{TZ},1} + \cdots + \omega^{\text{cusp}}_{\text{TZ},n_p}$. According to \cite[Lemma~2]{ZT_localindextheorem_1991}
	\begin{equation}\label{constcuspmetric}
	\lim_{\Im z \to \infty} \Im(\varsigma_i z)  f_{\mu_j \bar{\mu}_k}(\varsigma_i z) = \frac{4}{3} \langle\mu_j,\mu_k\rangle^{\text{cusp}}_{\text{TZ},i}, \qquad i=1,\dots,n_p,
	\end{equation}
	where $\mu_j, \mu_k \in \Hilbert^{-1,1}(\UHP,\Gamma)$ and $f_{\mu \overline{\mu'}}$ was defined in Eq.\eqref{fmunu}, the cuspidal TZ metric pertains to the second variation of the hyperbolic metric on $\UHP$.
	\item \emph{Elliptic Takhtajan-Zograf metric}. As discussed in \cite{Montplet2016RiemannRochII,ZT_2018}, when the Fuchsian group has $n_e>1$ elliptic generators $\tau_1,\dots,\tau_{n_e}$, the local index theorem of Takhtajan and Zograf \cite{ZT_localindextheorem_1991} for punctured Riemann surfaces can be generalized to include elliptic fixed points. In this case, the role of Eisenstein-Mass series $E_i(z,s)$ associated with the cusp $z_{n_e+i}$ is played by the automorphic Green's function $G_0(z,z_{j};s)$ associated with the elliptic fixed point $z_{j}$, $j=1,\dots,n_e$. More explicitly, for the elliptic generator $\tau_j$ of $\Gamma$ define
	\begin{equation}
	\langle \mu_1,\mu_2 \rangle^{\text{ell}}_{\text{TZ},j} = \iint_{\fund(\Gamma)} \mu_1(z) \, \overline{\mu_2(z)} G(z_{j},z) \rho(z) \dd[2]{z}, \qquad j=1,\dots,n_e, 
	\end{equation}
	where $z_{j}$ is the fixed-point of $\tau_j$, and $G(z,z') \equiv G_0(z,z';2)$ is the integral kernel of the resolvent $\left(\Delta_0 + \frac{1}{2}\right)^{-1}$. It was shown in \emph{Theorem 3} of \cite{ZT_2018} that the metrics $\langle \cdot , \cdot \rangle^{\text{ell}}_{\text{TZ},j}$ are also K\"{a}hler. In addition, if we denote by $\langle \cdot , \cdot \rangle^{\text{ell}}_{\{\s_m\}}$ the sum over all elliptic TZ metrics $\langle \cdot , \cdot \rangle^{\text{ell}}_{\text{TZ},j}$ associated with the elliptic generators $\tau_j$ that have the same order of isotropy $m$, we expect $\langle \cdot , \cdot \rangle^{\text{ell}}_{\{\s_m\}}$ to be invariant under the action of Teichm\"{u}ller modular group $\modular(\Gamma)$. Moreover, we will denote by $\omega^{\text{ell}}_{\text{TZ},j}$, the symplectic $(1,1)$-form
	\begin{equation}
	\omega^{\text{ell}}_{\text{TZ},j} = \frac{\sqrt{-1}}{2} \sum_{l,k=1}^{3g-3+n}  \langle \mu_l, \mu_k \rangle^{\text{ell}}_{\text{TZ},j} \hspace{1mm}\dd{t_l} \wedge \dd{\bar{t}_k}.
	\end{equation}
	Finally, the elliptic TZ metric is also intrinsically related to the second variation of the hyperbolic metric on $\UHP$: The following result was proven by Takhtajan and Zograf in \cite[Lemma 1 part (iii)]{ZT_2018}
	\begin{equation}
	\lim_{w \to w_{j}} f_{\mu_l \bar{\mu}_k}\circ J^{-1}(w)= \left\langle \pdv{w_l} , \pdv{w_k} \right\rangle_{\text{TZ},j}^{\text{ell}}, \qquad j=1,\dots,n_e,
	\end{equation}
	where $\textcolor{black}{\mu_{l}, \mu_{k}} \in \Hilbert^{-1,1}(\UHP,\Gamma)$ and $f_{\mu \overline{\mu'}}$ was defined in Eq.\eqref{fmunu}.	
\end{itemize}
\subsection{Moduli Spaces $\moduli_{g,n}$ and $\symmoduli_{g,n}(\boldsymbol{m})$}\label{subsec:moduli}
In the previous subsection, we have defined the Teichm\"{u}ller space $\teich(\Gamma)$ as the space of all equivalence classes of representations $\varrho_{\mu}: \Gamma \to \PSLR$ and we have seen that $\teich(\Gamma)$ can be realized as a bounded complex domain in $\cmpx^{3g-3+n}$ via the so-called Bers embedding. Let $\Automorphism_{\ast}(\Gamma)$ denote the group of proper automorphisms of $\Gamma$, which carry parabolic elements into parabolic elements and elliptic elements of order $m$ into elliptic elements with the same order. The group $\Automorphism_{\ast}(\Gamma)$ acts on $\teich(\Gamma)$ via
\begin{equation}
\imath(\varrho_{\mu}) = \varrho_{\mu} \circ \imath, \qquad \imath \in \Automorphism_{\ast}(\Gamma).
\end{equation}
That this is well-defined, i.e. that $\imath(\varrho_{\mu})$ is equivalent to another representation $\varrho_{\mu_{\imath}}$ for some $\mu_{\imath} \in \deformspace(\Gamma)$, follows from the fact that any automorphism $\imath \in \Automorphism_{\ast}(\Gamma)$ induces a quasi-conformal homeomorphism of $\UHP$. The group $\Inn(\Gamma)$ of inner automorphisms of $\Gamma$ obviously acts on $\teich(\Gamma)$ as the identity. Let us remind that the factor group $\modular(\Gamma) := \Automorphism_{\ast}(\Gamma)/\Inn(\Gamma)$ is called the Teichm\"{u}ller modular group and acts discretely on $\teich(\Gamma)$ by complex-analytic automorphisms that only change the marking of $\Gamma$. Denote by $\puremod (\Gamma)$ the subgroup of $\modular(\Gamma)$ consisting of \emph{pure mapping classes} --- i.e. those fixing the cusps and orbfiold points on $O$ pointwise.\footnote{The Teichm\"{u}ller modular group $\modular(\Gamma)$ acting on $\teich(\Gamma)$ can be identified with the orbifold mapping class group $\operatorname{MCG}(O)$ acting on $\teich_{g,n}(\boldsymbol{m})$. Here, $\operatorname{MCG}(O)$ is defined as $\operatorname{Homeo}^{+}(O)/\operatorname{Homeo}_{0}(O)$ where $\operatorname{Homeo}^{+}(O)$ is the group of orientation preserving homeomorphisms of $O$ (in the category of orbifolds), and $\operatorname{Homeo}_{0}(O)$ is its identity component.} The full Teichm\"uller modular group $\modular(\Gamma)$ is related to $\puremod(\Gamma)$ by the short exact sequence 
\begin{equation}\label{pmcgexactsequence}
	1 \to \puremod(\Gamma) \to \modular(\Gamma) \to \symm{\sigtype} \to 1,
\end{equation}
where $\symm{\sigtype}:=  \symm{\s_2} \times \symm{\s_3} \times \dotsm \times \symm{\s_{\infty}}$ denotes a subgroup of  $\symm{n}$ consisting of all permutations that leave the signature type $\sigtype = \{\s_m\}_{m \in \orderrange}$ invariant \cite{engber1975teichmuller}. Then, the quotient space $\teich(\Gamma)/\puremod(\Gamma)$ is isomorphic to the moduli space $\moduli_{g,n}$ of smooth algebraic curves of genus $g$ with $n=n_e + n_p$ labeled points. According to \eqref{pmcgexactsequence}, the Tiechm\"uller modular group acts on $\moduli_{g,n}$ via $\symm{\sigtype}$ and the quotient $\moduli_{g,n} / \symm{\sigtype}$ is isomorphic to $\symmoduli_{g,n}(\boldsymbol{m})$ --- the true moduli space of orbifold Riemann surfaces with signature $(g;m_1,\dots,m_{n_e};n_p)$.\footnote{When all singular points have the same order of isotropy, the situation will be similar to what has been previously studied by Zograf \cite{Zograf_1990}.} Finally, we remark that both $\teich_{g,n}(\boldsymbol{m})$ and $\symmoduli_{g,n}(\boldsymbol{m})$ depend not on the signature of $O$, but rather on its signature type (see \cite{Bers1975DeformationsAM} for more details). 
\subsubsection{$\symmoduli_{0,n}(\boldsymbol{m})$ }\label{M0n}
In the remainder of this subsection, we will focus on the $g=0$ case for the sake of simplicity and return to $g>1$ Riemann orbisurfaces in the next subsection. A \emph{normalized} orbifold Riemann surfaces with signature $(0;m_1,\dots,m_{n_e};n_p)$ is given by a pair $O=(\cmpx\backslash\{w_{n_e+1},\dots,w_n\},\brdiv)$ with $\brdiv = \sum_{i=1}^{n_e} (1-\tfrac{1}{m_i}) \, w_i$ such that $w_{n-3}$, $w_{n-2}$, and $w_n$ are at $0$, $1$, $\infty$ respectively.\footnote{We are assuming that the stratum of conical points $\sing_{m_{\max}}(O)$ with largest order of isotropy $m_{\max} \in \orderrange$, has cardinality of at least three. The analysis in cases where this assumption doesn't hold requires a change of notation, but the fundamental lessons remain the same.} Accordingly, the  moduli space  $\moduli_{0,n} = \confspace{n}{\hat{\cmpx}} \big/ \PSLC$ is given by the following domain in $\cmpx^{n-3}$
\begin{equation}
\moduli_{0,n} = \left\{(w_1,\dots, w_{n-3}) \in \cmpx^{n-3} \Big| w_i \neq 0,1 \quad \text{and} \quad w_i \neq w_k \quad \text{for} \quad i \neq k \right\},
\end{equation}
where $\confspace{n}{\hat{\cmpx}}$ is the configuration space of $n=n_e + n_p$ labeled distinct points in $\hat{\cmpx}$. We will show that $\moduli_{0,n}$\footnote{What we have called $\moduli_{0,n}$ in this paper, is isomorphic to what has been denoted by $W_n$ in \cite{Zograf1988ONLE,Zograf_1990} and by $\mathcal{Z}_n$ in \cite{Takhtajan:2001uj}.} is covered (in the complex-analytic sense) by the Teichm\"{u}ller space of orbifold Riemann surfaces with signature $(0;m_1,\dots,m_{n_e};n_p)$. This will enable us to express the vector fields $\pdv{w_i}$ on $\moduli_{0,n}$ in terms of sections of the holomorphic tangent bundle of Teichm\"{u}ller space.

Using Theorem~\ref{thm:universalcovering}, we have $O \cong [\UHP\slash\Gamma]$ where $\Gamma$ is normalized such that the fixed points of $\kappa_{n-2},\kappa_{n-1},\kappa_{n}$ are at $z_{n-2} = 0$, $z_{n-1} = 1$, and $z_n =\infty$ respectively. Denote by $\UHP^{\ast}$ the union of $\UHP$ and all parabolic points of $\Gamma$. There is a unique (universal) orbifold covering map $J: \UHP \to O$ with $\deck(J) \cong \Gamma$,\footnote{One might correctly want to identify the covering map $J$ in this subsection with the covering map $\pi_{\Gamma}$ in the commutative diagram \eqref{globalcoords}. Only for the reasons that will become clear later, we decided to call the covering map in this subsection with $J$.} which extends to a holomorphic isomorphism $[\UHP^{\ast}\slash\Gamma] \overset{\cong}{\longrightarrow} \hat{O}=(\hat{\cmpx},\hat{\brdiv})$ that fixes the points $0, 1, \infty$\footnote{In the literature, the $J$ is called \emph{Klien's Hauptmodul}. Actually, it stands as the sole $\Gamma$-automorphic function on $\UHP$ that exhibits a simple pole at $\infty$ and fixes $0$ and $1$.} and has the property that $w_i = J(z_i)$ for $i=1,\dots,n-3$.\footnote{The $\mathbb{Q}$-divisor $\hat{\brdiv}$ refers to $\brdiv+\sum_{j=n_e+1}^{n}w_j$.} The function $J$ is univalent in any \emph{fundamental domain} $\mathcal{F}$ for $\Gamma$ and has the following expansions near cusps and conical singularities (for more details, see Appendix~\ref{Apx:Asymptotics})
\begin{equation}\label{Jexpansion}
J(z) = \left\{
\begin{split}	& w_{i} + \sum_{k=1}^{\infty} J^{(i)}_k  \left(\frac{z-z_{i}}{z-\bar{z}_{i}}\right)^{k m_i} \hspace{3.7cm} (i=1,\dots,n_e),\hspace{.5cm} z \to z_{i},\\
& w_{i} + \sum_{k=1}^{\infty} J^{(i)}_k \exp(-\frac{2 \pi \sqrt{-1}k}{|\delta_{i}|(z-z_{i})}) \hspace{1.55cm} (i=n_e+1,\dots,n-1),\hspace{.2cm} z \to z_{i}, \\
&\sum_{k=-1}^{\infty} J^{(n)}_k \exp(\frac{2 \pi \sqrt{-1}k z}{|\delta_n|}) \hspace{4.5cm}  z \to z_n= \infty.
\end{split}
\right.
\end{equation}
The first coefficients of the above expansions determine the following smooth positive functions on $\moduli_{0,n}$:
\begin{equation}\label{hi}
\Lponetial_{i} =\left\{
\begin{split}
&  \left| J^{(i)}_{1} \right|^{\frac{2}{m_i}} \hspace{1.3cm} i=1,\dots,n_e,\\ 
&\left| J^{(i)}_{1} \right|^2 \hspace{1.55cm}  i=n_e+1,\dots,n-1,\\
& \left| J^{(n)}_{-1} \right|^2\hspace{1.5cm}  i=n.
\end{split}
\right.
\end{equation}
Similar to the case of $\moduli_{g,n}$ discussed at the beginning of this subsection, the symmetric group $\symm{\sigtype}$ acts on $\moduli_{0,n}$ to give $\symmoduli_{0,n}(\boldsymbol{m}) = \moduli_{0,n}/\symm{\sigtype}$, the moduli space of orbifold Riemann surfaces with signature $(0;m_1,\dots,m_{n_e};n_p)$. In order to describe the action of $\symm{\sigtype}$ on $\moduli_{0,n}$ in more detail, we will make the simplifying assumption that the signature of orbifold Riemann surface $O$ is given by $(0;\underbrace{m,\dots,m}_{\s},\underbrace{m',\dots,m'}_{\s'})$ with $\s \equiv \s_m$ and $\s' \equiv \s_{m'} >3$. Let us first focus on ordered $\s'$-tuples $(w_1,w_2,\dots,w_{\s'})$ with each $w_i\in \sing_{m'}$. If none of the points $0,1$ and $\infty$ are fixed in this set of points, then $\symm{\s'}$ will simply be the group of permutations of $\s'$ objects. This group is generated by the set of transpositions $\{\sigma_{i,i+1}\}_{i=1}^{\s'-1}$, whose action only involves interchanging $w_i$ and $w_{i+1}$ for $i\neq \s'$, and $w_1$ and $w_{\s'}$ for $i=\s'$. The situation will be a bit more complicated when the points $0,1$ and $\infty$ are fixed among these points, namely, we have a sector of the form $\{(w_1,w_2,\dots,w_{s'-3},0,1,\infty)\in\cmpx^{\s'}\}$. Then the group $\symm{\s'}$ will be generated by the transpositions $\{\sigma_{i,i+1}\}_{i=1}^{\s'-1}$ and the action of transpositions will be followed by a $\PSLC$ transformation to ensure that the last three coordinates in $\moduli_{0,n}$ remain $0,1$ and $\infty$. If $i< \s'-3$, the transpositions will not affect the points $0,1$ and $\infty$, thus no further action of $\PSLC$ will be needed. If $i=\s'-3$, then the set of branch points will change to $\{(w_1,w_2,\dots,0,w_{\s'-3},1,\infty)\}$, and we need a transformation that will take $w_{\s'-3}\rightarrow 0,1\rightarrow 1$ and $\infty\rightarrow\infty$. This transformation is $\gamma_{\s'-3,\s'-2}=(w-w_{\s'-3})/(1-w_{\s'-3})$. Thus in the end we will arrive at $\{(\frac{w_1-w_{\s'-3}}{1-w_{\s'-3}},\dots,\frac{w_{\s'-4}-w_{\s'-3}}{1-w_{\s'-3}},\frac{w_{\s'-3}}{w_{\s'-3}-1},0,1,\infty)\}$. Repeating the same procedure for $i=\s'-2$ and $i=\s'-1$ will yield the transformations, $\gamma_{\s'-2,\s'-1}=1-w$ and $\gamma_{\s'-1,\s'}=w/(w-1)$. Putting all these together, the collective action of $\symm{\s'}$ on $\moduli_{0,n}$, can be expressed as $\sigma_{i,i+1}(w_1,w_2,\dots,w_{\s'-3},0,1,\infty)=(\tilde{w}_{1},\dots,\tilde{w}_{\s'-3},0,1,\infty)$ such that
\begin{equation}\label{transpositions}
\tilde{w}_k = \left\{
\begin{split}
&w_k  \qquad (k\neq i,i+1), \qquad &i \leq \s'-4, \\
& w_{k+1} \qquad (k=i),   &i\leq \s'-4, \\
& w_{k-1} \qquad (k=i+1),  & i\leq \s'-4,\\
& \frac{w-w_{\s'-3}}{1-w_{\s'-3}} \qquad (k\leq \s'-4),  & i=\s'-3,\\
& \frac{w_{\s'-3}}{w_{\s'-3}-1} \qquad (k=\s'-3),  & i=\s'-3,\\
& 1-w_{k} \qquad (k\leq \s'-3),  & i=\s'-2,\\
& \frac{w_k}{w_k-1} \qquad (k\leq \s'-3),  & i=\s'-1.
\end{split}
\right.
\end{equation}
As we said, the effect of $\sigma'_{i,i+1}\in\symm{\s'}$ on $\moduli_{0,n}$ is followed by a $\gamma_{i,i+1}\in\PSLC$ on the orbifold Riemann surface $O$ with the coordinates $w$. This transformation is actually an isomorphism that takes $O$ to another orbifold Riemann surface with the coordinates $\tilde{w}=\gamma_{i,i+1}(w)$. For $i\leq \s'-4$, $\gamma_{i,i+1}$ is simply the identity map. For the other cases we have  $\gamma_{i,i+1}:w\rightarrow\tilde{w}=(w-w_{\s'-3})/(1-w_{\s'-3}),\:i=\s'-3$, $\gamma_{i,i+1}:w\rightarrow\tilde{w}=1-w,\:i=\s'-2$ and $\gamma_{i,i+1}:w\rightarrow\tilde{w}=w/(w-1)$ for $i=\s'-1$. If one wishes to consider the variation of the objects defined on $X$ under the action of $\symm{\s'}$, one should study the effect of these $\gamma_{i,i+1}$'s on them. Now keeping $0,1$ and $\infty$ fixed in the $\sing_{m'}$ stratum, we consider 
ordered $\s$-tuples $(w_1,w_2,\dots,w_{\s})$ with $w_j\in \sing_{m}$ for all $j=1,\dots,\s$, $\s:=|\sing_{m}|$, and $m<m'$. Here, we no longer make any assumptions about $s$. Again, the group $\symm{s}$ is generated by the transpositions $\{\sigma_{j,j+1}\}_{j=1}^{\s-1}$ and their action simply involves interchanging $w_j$ and $w_{j+1}$. Thus, the $\gamma_{j,j+1}$ defined for this set will all be identity, meaning that $\symm{\s}$ will not have any non-trivial effect on $O$. Moving forward, we can define the direct product of the two symmetric groups $\symm{\s}\times\symm{\s'}$ as ordered pairs $(\sigma,\sigma'),\sigma \in\symm{\s},\sigma' \in\symm{\s'}$. The group operations will then be defined naturally on a pairwise basis, and this group acts on $\moduli_{0,n}$,
\begin{equation*}
\begin{aligned}
\{(w_1,\dots,w_{\s})\in\cmpx^{\s}\}\times&\{(w_1,w_2,\dots,w_{\s'-3},0,1,\infty)\in\cmpx^{\s'}\}\\
&=\{(w_1,\dots,w_{\s},w_{\s+1},w_{\s+2},\dots,w_{\s+\s'-3},0,1,\infty)\}.
\end{aligned}
\end{equation*} 
It is now clear that we have chosen to fix $0$, $1$, and $\infty$ in the stratum $\sing_{m'}$ in order to comply with our convention that branch points are ordered with increasing order of isotropy. In particular, the generators of the direct product group are pairs of transpositions $\{(\sigma_{j,j+1},\sigma'_{i,i+1})\}_{j=1,i=1}^{j=s-1,i=s'-1}$. The action of these pairs on $\moduli_{0,n}$ is defined by their separate action on their corresponding subspace:
\begin{equation}\label{directproductofsym}
\begin{aligned}
(\sigma_{j,j+1},\sigma'_{i,i+1})&\big(w_1,\dots,w_{\s},w_{\s+1},w_{\s+2},\dots,w_{\s+\s'-3},0,1,\infty\big)\\
&=\Big(\sigma_{j,j+1}(w_1,\dots,w_{\s});\sigma'_{i,i+1}(w_{\s+1},w_{\s+2},\dots,w_{\s+\s'-3},0,1,\infty)\Big).
\end{aligned}
\end{equation}
Similar to the cases with the single stratum, these transpositions should be followed by a transformation $\gamma_{j,j+1;i,i+1}\in\PSLC$ on the orbifold Riemann surface $O$. For any $j$ and $i\leq \s'-4$, the corresponding transformation is simply the identity map, and for the cases with $i=\s'-3,\s'-2,\s'-1$ and arbitrary $j$ we have $\gamma_{j,j+1;i,i+1}$'s that is identical to $\gamma_{i,i+1}$'s defined for the stratum $\sing_{m'}$, with $w_{\s'-3}$ replaced by $w_{\s+\s'-3}$. This means that (\ref{directproductofsym}) can be expressed more explicitly by: \hspace{-.1cm}$	(\sigma_{j,j+1},\sigma'_{i,i+1})\big(w_1,\dots,w_{\s},w_{\s+1},w_{\s+2},\dots,w_{\s+\s'-3},
0,1,\infty\big)=\big(\tilde{w}_1,\dots,\tilde{w}_{\s},\tilde{w}_{\s+1},\tilde{w}_{\s+2},\dots,\tilde{w}_{\s+\s'-3},0,1,\infty\big)$ such that
\begin{equation}\label{transpositionsdirectprod}
\tilde{w}_k = \left\{
\begin{split}
&w_k  \qquad (k\neq j,j+1,\s+i,\s+i+1), \qquad &i \leq \s'-4, \forall j,\\
& w_{k+1} \qquad (k=j\textrm{ or }k=\s+i),   &i\leq \s'-4,\forall j, \\
& w_{k-1} \qquad (k=j+1\textrm{ or }k=\s+i+1),  & i\leq \s'-4,\forall j,\\
& \frac{w_k-w_{\s+\s'-3}}{1-w_{\s+\s'-3}} \qquad (k\leq \s+\s'-4),  & i=\s'-3,\forall j,\\
& \frac{w_{\s+\s'-3}}{w_{\s+\s'-3}-1} \qquad (k=\s+\s'-3),  & i=\s'-3,\forall j,\\
& 1-w_{k} \qquad (k\leq \s+\s'-3),  & i=\s'-2,\forall j,\\
& \frac{w_k}{w_k-1} \qquad (k\leq \s+\s'-3),  & i=\s'-1,\forall j.
\end{split}
\right.
\end{equation}
It is clear that one can continue developing larger direct products with more strata by doing simple changes to (\ref{transpositionsdirectprod}). It is natural that, similar to the case of the single stratum, one needs to look at the related M\"{o}bius transformation when dealing with the variation of objects with respect to the action of $\symm{\sigtype}$. For the symmetric group $\symm{\s'}$ one can define the 1-cocycle $\{f_{\sigma'} \}_{\sigma' \in\symm{\s'}}$. It is defined for the generators by
\begin{equation}\label{cocyclescones}
f_{\sigma'_{i,i+1}} = \left\{
\begin{split}
&1   \qquad &i=1,2,\dots, \s'-4,\s'-2, \\
& (w_{\s'-3}-1)^{h'(\s'-2)} & i=\s'-3,\\
& \prod_{k=1}^{\s'-3}(w_k-1)^{2h'} & i=\s'-1,\\
\end{split}
\right.
\end{equation}
where $h'/2$ is the conformal weight corresponding to ${m'}$. Note that if $m'\to \infty$, namely the case of punctures, (\ref{cocyclescones}) will be slightly different:
\begin{equation}\label{cocyclespunctures}
f_{\sigma'_{i,i+1}} = \left\{
\begin{split}
&1  \quad &i=1,2,\dots \s'-4,\s'-2, \\
& (w_{\s'-3}-1)^{(\s'-2)}  & i=\s'-3,\\
& \prod_{k=1}^{\s'-3}(w_k-1)^{2} & i=\s'-1.\\
\end{split}
\right.
\end{equation}
The action of the 1-cocycle $f_{\sigma'}$ can be extended to a general element in $\symm{\s'}$ by the product rule $f_{\sigma'_1\sigma'_2}=(f_{\sigma'_1}\circ\sigma'_2)f_{\sigma'_2}$, $\sigma'_1,\sigma'_2\in\symm{s'}$. For the group $\symm{s}$, the effect of the 1-cocycle on the generators is trivial, namely $f_{\sigma_{j,j+1}}=1$. Finally, for the direct product group $\symm{\s}\times\symm{\s'}$ the 1-cocycle will be given for the generators by
\begin{equation}\label{cocyclesdirectprod}
f_{(\sigma_{j,j+1},\sigma'_{i,i+1})} = \left\{
\begin{split}
&1  \quad &i=1,2,\dots \s'-4,\s'-2,\forall j \\
& (w_{\s+\s'-3}-1)^{h\s+h'(\s'-2)}  & i=\s'-3,\forall j\\
& \prod_{k=1}^{\s}(w_k-1)^{2h}\prod_{k=\s+1}^{\s'-3}(w_k-1)^{2h'} & i=\s'-1,\forall j,\\
\end{split}
\right.
\end{equation}
where $h/2,h'/2$ are the conformal weights corresponding to ${m,m'}$. Thus, we can see that the non-triviality of the 1-cocycle $f$ is only due to the fixed points $0,1$ and $\infty$, and if none of them are involved in the permutation, the 1-cocycle will be trivial. 
\begin{remark}
	As we will see in Lemma \ref{lemma:SmHermiatiannmetricoverlambdanogenus}, a simple way to determine the 1-cocycles $f$ is to calculate the variation of the Liouville action due to the transformation of $\moduli_{0,n}$ under the symmetric group. In other words, these 1-cocycles are basically the modular anomaly caused by the non-covariance of the action under the effect of the modular group.
\end{remark}
As in \cite{Zograf_1990,park2015potentials}, let $\{f_{\eta}\}_{\eta \in \symm{\sigtype}}$ be the 1-cocycle for $\symm{\sigtype}$ on $\moduli_{0,n}$. These 1-cocycles can be used to define a holomorphic $\mathbb{Q}$-line bundle over the moduli space $\symmoduli_{0,n}(\boldsymbol{m})$. To do so, one constructs the trivial bundle $\moduli_{0,n} \times \cmpx$ and defines the action of $\symm{\sigtype}$ on this bundle by 
\begin{equation}
(\boldsymbol{w},\tilde{z}) \mapsto (\eta \cdot \boldsymbol{w}, f_{\eta}(\boldsymbol{w}) \tilde{z}), \qquad \boldsymbol{w} \in \moduli_{0,n}, \, \tilde{z} \in \cmpx, \, \eta \in \symm{\sigtype}.
\end{equation}
Then, the desired holomorphic $\mathbb{Q}$-line bundle $\lambda_{0,\boldsymbol{m}}=(\moduli_{0,n}\times\cmpx)/\sim$ over moduli space $\symmoduli_{0,n}(\boldsymbol{m})=\moduli_{0,n}/\symm{\sigtype}$ is defined by the identification $(\boldsymbol{w},\tilde{z}) \sim (\eta \cdot \boldsymbol{w}, f_{\eta}(\boldsymbol{w}) \tilde{z})$ for all $\eta\in\symm{\sigtype}$.
\begin{lemma}\label{lemma:Hermitianmetriconlambdawnogenus}
	Let $O$ be a closed (i.e., $n_p=0$) orbifold Riemann surface with signature $(0;m_1,\dots,m_{n})$ and fix the last three conical points to be at $0$, $1$, and $\infty$ (as always, we assume that these three conical points belong to the same stratum). Define a positive function 	
	\begin{equation*}
	\mathsf{H} = \Lponetial_{1}^{m_1h_1}\dotsm\Lponetial_{n-1}^{m_{n-1} h_{n-1}} \Lponetial_{n}^{- m_n h_n},
	\end{equation*} 
	on $\moduli_{0,n}$. Then, the $\mathsf{H}$ determines a Hermitian metric in the holomorphic $\mathbb{Q}$-line bundle $\lambda_{0,\boldsymbol{m}}$ over $\symmoduli_{0,n}(\boldsymbol{m})$, where $m_i$'s are the branching indices and $h_i/2$'s are their corresponding conformal weights.
\end{lemma}
\begin{proof}
	We prove this lemma for the case where we have only two strata of $s$ branch points of order $m$ and $s'$ branch points of order $m'$. Furthermore, we assume that $s'> 3$, and the last three conical points in the stratum $\sing_{m'}$ ($m'>m$) are chosen to be at $0$, $1$, $\infty$. In the end, we shall explain how the rest of the cases can be dealt with similarly. We have  $\moduli_{0,n}=\{(w_1,\dots,w_{s},w_{s+1},w_{s+2},\dots,w_{s+s'-3},0,1)\}$. For simplicity of notation, we can, from now on, denote $s+s'=n$ in the proof. Each of the $\Lponetial_{k}$'s can be viewed as a function on $\cmpx$ with the appropriate asymptotics (see \eqref{hi}, part 3. in Lemma \ref{lemma:asymptotics} and Remark \ref{extraa}). Accordingly,	
	\begin{equation}
	\begin{split}
	\log \mathsf{H}=&hm\sum_{k=1}^{s}\left(-2 \log m + 2 \log 2 - \lim_{w \to w_{k}}\left(\varphi(w) + \left(1-\frac{1}{m}\right) \log|w-w_{k}|^2\right)\right)\\
	&+h'm'\sum_{k=s+1}^{n-1}\left(-2 \log m' + 2 \log 2 - \lim_{w \to w_{k}}\left(\varphi(w) + \left(1-\frac{1}{m'}\right) \log|w-w_{k}|^2\right)\right)\\
	&-h'm'\left(2 \log m' - 2 \log 2 + \lim_{w \to \infty}\left(\varphi(w) + \left(1+\frac{1}{m'}\right) \log|w|^2\right)\right).
	\end{split}
	\end{equation} 
	Now we need to calculate the variation of $\log \mathsf{H}$ under the effect of $\symm{\sigtype}=\symm{s}\times\symm{s'}$. For this, it suffices to look at the effect of the generators $\{(\sigma_{j,j+1},\sigma'_{i,i+1})\}_{j=1,i=1}^{j=s-1,i=s'-1}$. Considering the background we provided above on the structure of the generators, variation is translated through the effect of the transformation $\gamma_{j,j+1;i,i+1}$:
	\begin{equation}
	\begin{aligned}
	\Delta\log\mathsf{H}
	=\log\mathsf{H}[\gamma_{j,j+1;i,i+1}]-\log \mathsf{H}.
	\end{aligned}
	\end{equation} 
	By looking at (\ref{transpositionsdirectprod}), we see that the index $j$ does not have any non-trivial effect, and we only have to worry about different values of $i$. For $i<s'-3$ the isomorphism $\gamma_{j,j+1;i,i+1}$ is the identity so we for this cases we have $\Delta\log\mathsf{H}=0$. The non-trivial cases are $i=s'-3,i=s'-2$ and $i=s'-1$. In the following, we study them separately:
	\begin{itemize}
		\item \emph{$i=n-3$ case:}
		
		The needed M\"{o}bius transformation is given by  $\gamma_{j,j+1;i,i+1}=(w-w_{n-3})/(1-w_{n-3})$ in this case. For simplicity, we denote this by just $\gamma$. Thus, we write: 
		\begin{equation}\label{variationH2}
		\begin{aligned}
		\Delta&\log\mathsf{H}=\log\mathsf{H}[\gamma]-\log \mathsf{H}=\\
		& - hm\sum_{k=1}^{s}\left( \lim_{\gamma(w) \to \gamma(w_{k})}\left(\tilde{\varphi}(\gamma(w)) + \left(1-\frac{1}{m}\right) \log|\gamma(w)-\gamma(w_k)|^2\right)\right)\\
		& + hm\sum_{k=1}^{s}\left( \lim_{w \to w_{k}}\left(\varphi(w) + \left(1-\frac{1}{m}\right) \log|w-w_{k}|^2\right)\right)\\
		& - h'm'\sum_{k=s+1}^{n-1}\left( \lim_{\gamma(w) \to \gamma(w_{k})}\left(\tilde{\varphi}(\gamma(w)) + \left(1-\frac{1}{m'}\right) \log|\gamma(w)-\gamma(w_k)|^2\right)\right)\\
		&+h'm'\sum_{k=s+1}^{n-1}\left( \lim_{w \to w_{k}}\left(\varphi(w) + \left(1-\frac{1}{m'}\right) \log|w-w_{k}|^2\right)\right)\\
		&-h'm'\lim_{\gamma(w) \to \gamma(\infty)}\left(\tilde{\varphi}(\gamma(w)) + \left(1+\frac{1}{m'}\right) \log|\gamma(w)|^2\right)\\
		& + h'm'\lim_{w \to \infty}\left(\varphi(w) + \left(1+\frac{1}{m'}\right) \log|w|^2\right) ,
		\end{aligned}
		\end{equation} 
		where $\tilde{\varphi}$ is the transformed counterpart of $\varphi$ through the isomorphism $\gamma$. From the invariance of the hyperbolic metric, these two are related by 
		\begin{equation}\label{liouvillefieldtransform}
		\varphi(w)=\tilde{\varphi}(\gamma(w))+\log\left|\frac{\partial\gamma(w)}{\partial w}\right|^2=\tilde{\varphi}(\gamma(w))-2\log \left|1-w_{n-3}\right|.
		\end{equation}
		We also have
		\begin{equation}\label{logtransforms}
		\log|\gamma(w)-\gamma(w_k)| = \left\{
		\begin{split}
		&\log|\gamma(w)-\gamma(w_k)|=\log\left|\frac{w-w_k}{1-w_{n-3}}\right|   \:\:\:\:\:\: &k\leq n-4, \\
		& \log|\gamma(w)-\gamma(0)|=\log\left|\frac{w}{1-w_{n-3}}\right|    \:\:\:\:\:\:& k=n-3,\\
		&\log|\gamma(w)-\gamma(w_{n-3})|= \log\left|\frac{w-w_{n-3}}{1-w_{n-3}}\right|    \:\:\:\:\:\:& k=n-2,\\
		& \log|\gamma(w)-\gamma(1)|=\log\left|\frac{w-1}{1-w_{n-3}}\right|    \:\:\:\:\:\:& k=n-1.\\
		\end{split}
		\right.
		\end{equation}
		Note that in finding \eqref{logtransforms}, we first included the transposition that exchanges $w_{n-3}$ and $w_{n-2}=0$ and then included the effect of $\gamma$. Putting (\ref{liouvillefieldtransform}) and (\ref{logtransforms}) together we look more closely at (\ref{variationH2}). The terms with $k\leq n-4$ in the sum are quite straightforward:
		\begin{equation}\nonumber
		\begin{aligned}
		I_1=& - hm\sum_{k=1}^{s}\left( \lim_{\gamma(w) \to \gamma(w_{k})}\left(\tilde{\varphi}(\gamma(w)) + \left(1-\frac{1}{m}\right) \log|\gamma(w)-\gamma(w_k)|^2\right)\right)\\
		& + hm\sum_{k=1}^{s}\left( \lim_{w \to w_{k}}\left(\varphi(w) + \left(1-\frac{1}{m}\right) \log|w-w_{k}|^2\right)\right)\\
		& - h'm'\sum_{k=s+1}^{n-4}\left( \lim_{\gamma(w) \to \gamma(w_{k})}\left(\tilde{\varphi}(\gamma(w)) + \left(1-\frac{1}{m'}\right) \log|\gamma(w)-\gamma(w_k)|^2\right)\right)\\
		&+h'm'\sum_{k=s+1}^{n-4}\left( \lim_{w \to w_{k}}\left(\varphi(w) + \left(1-\frac{1}{m'}\right) \log|w-w_{k}|^2\right)\right)\\
		& =- hm\sum_{k=1}^{s}\left( \lim_{w \to w_k}\left(2\log \left|1-w_{n-3}\right| + \left(1-\frac{1}{m}\right) \log\left|\frac{w-w_k}{1-w_{n-3}}\right|^2\right)\right)\\
		& + hm\sum_{k=1}^{s}\left( \lim_{w \to w_{k}}\left( \left(1-\frac{1}{m}\right) \log|w-w_{k}|^2\right)\right)\\
		& - h'm'\sum_{k=s+1}^{n-4}\left( \lim_{w\to w_k}\left(2\log \left|1-w_{n-3}\right|+ \left(1-\frac{1}{m'}\right) \log\left|\frac{w-w_k}{1-w_{n-3}}\right|^2\right)\right)\\
		&+h'm'\sum_{k=s+1}^{n-4}\left( \lim_{w \to w_{k}}\left( \left(1-\frac{1}{m'}\right) \log|w-w_{k}|^2\right)\right)\\
		& =-2\left(hs+h'(s'-4)\right)\log\left|1-w_{n-3}\right|.
		\end{aligned}
		\end{equation}
		For $k=n-3$, we have:
		\begin{equation}\nonumber
		\begin{aligned}
		I_2=& - h'm' \lim_{\gamma(w) \to \gamma(0)}\left(\tilde{\varphi}(\gamma(w)) + \left(1-\frac{1}{m'}\right) \log|\gamma(w)-\gamma(0)|^2\right)\\
		&+h'm'\lim_{w \to w_{n-3}}\left(\varphi(w) + \left(1-\frac{1}{m'}\right) \log|w-w_{n-3}|^2\right)\\
		& = - h'm' \lim_{w\to 0}\left(\varphi(w)+2\log \left|1-w_{n-3}\right|+ \left(1-\frac{1}{m'}\right) \log\left|\frac{w}{1-w_{n-3}}\right|^2\right)\\
		&+h'm'\lim_{w \to w_{n-3}}\left(\varphi(w)+ \left(1-\frac{1}{m'}\right) \log|w-w_{n-3}|^2\right).
		\end{aligned}
		\end{equation}
		And for $k=n-2$:
		\begin{equation}\nonumber
		\begin{aligned}
		I_3=& - h'm' \lim_{\gamma(w) \to \gamma(w_{n-3})}\left(\tilde{\varphi}(\gamma(w)) + \left(1-\frac{1}{m'}\right) \log|\gamma(w)-\gamma(w_{n-3})|^2\right)\\
		&+h'm'\lim_{w \to 0}\left(\varphi(w) + \left(1-\frac{1}{m'}\right) \log|w|^2\right)\\
		& = - h'm' \lim_{w\to w_{n-3}}\left(\varphi(w)+2\log \left|1-w_{n-3}\right|+ \left(1-\frac{1}{m'}\right) \log\left|\frac{w-w_{n-3}}{1-w_{n-3}}\right|^2\right)\\
		&+h'm'\lim_{w \to 0}\left(\varphi(w)+ \left(1-\frac{1}{m'}\right) \log|w|^2\right).
		\end{aligned}
		\end{equation}
		Thus, we have:
		\begin{equation}\nonumber
		I_2+I_3=-2h'\log\left|1-w_{n-3}\right|.
		\end{equation}
		Also for $k=n-1$:
		\begin{equation}\nonumber
		\begin{aligned}
		I_4=& - h'm' \lim_{\gamma(w) \to \gamma(1)}\left(\tilde{\varphi}(\gamma(w)) + \left(1-\frac{1}{m'}\right) \log|\gamma(w)-\gamma(1)|^2\right)\\
		&+h'm'\lim_{w \to 1}\left(\varphi(w) + \left(1-\frac{1}{m'}\right) \log|w-1|^2\right)\\
		& = - h'm' \lim_{w\to 1}\left(2\log \left|1-w_{n-3}\right|+ \left(1-\frac{1}{m'}\right) \log\left|\frac{w-1}{1-w_{n-3}}\right|^2\right)\\
		&+h'm'\lim_{w \to 1}\left( \left(1-\frac{1}{m'}\right) \log|w-1|^2\right)\\
		&=-2h'\log\left|1-w_{n-3}\right|.
		\end{aligned}
		\end{equation}
		And finally, for the contribution of infinity:
		\begin{equation}\nonumber
		\begin{aligned}
		I_5=& - h'm' \lim_{\gamma(w) \to \gamma(\infty)}\left(\tilde{\varphi}(\gamma(w)) + \left(1+\frac{1}{m'}\right) \log|\gamma(w)|^2\right)\\	&+h'm'\lim_{w \to \infty}\left(\varphi(w) + \left(1+\frac{1}{m'}\right) \log|w|^2\right)\\
		& = - h'm' \lim_{w\to \infty}\left(2\log \left|1-w_{n-3}\right|+ \left(1+\frac{1}{m'}\right) \log\left|\frac{w-w_{n-3}}{1-w_{n-3}}\right|^2\right)\\
		&+h'm'\lim_{w \to \infty}\left( \left(1+\frac{1}{m'}\right) \log|w|^2\right)\\
		&=2h'\log\left|1-w_{n-3}\right|.
		\end{aligned}
		\end{equation}
		Thus, we have for all $j$:
		\begin{equation}
		\Delta\log\mathsf{H}=\sum_{i=1}^{5}I_i=-2\left(hs+h'(s'-2)\right)\log\left|1-w_{n-3}\right|=-2\log|f_{(\sigma_{j,j+1},\sigma'_{n-3,n-2})}|,
		\end{equation}
		where we used (\ref{cocyclesdirectprod}) in the last equality.
		\item \emph{$i=n-2$ case:}
		
		Now, we look at the case with $i=n-2$ and again denote the morphism by $\gamma$. We have $\gamma=1-w$ in this case. Furthermore 
		\begin{equation}\label{logtransforms1}\nonumber
		\begin{aligned}
		&\varphi(w)=\tilde{\varphi}(\gamma(w))+\log\left|\frac{\partial\gamma(w)}{\partial w}\right|^2=\tilde{\varphi}(\gamma(w)),\\
		&\log|\gamma(w)-\gamma(w_k)| = \left\{
		\begin{split}
		&\log|\gamma(w)-\gamma(w_k)|=\log\left|w-w_k\right|   \:\:\:\:\:\: &k\leq n-3, \\
		&\log|\gamma(w)-\gamma(1)|= \log\left|1-w \right|    \:\:\:\:\:\:& k=n-2,\\
		& \log|\gamma(w)-\gamma(0)|=\log\left|w\right|    \:\:\:\:\:\:& k=n-1.\\
		\end{split}
		\right.
		\end{aligned}
		\end{equation}
		Accordingly, it is quite clear that the variation is zero in this case, which is in agreement with the lemma.
		\item \emph{$i=n-1$ case:}
		
		The last and perhaps the most subtle case is the case with $i=n-1$. Here we have $\gamma=w/(w-1)$, and
		\begin{equation}\label{logtransforms2}\nonumber
		\begin{aligned}
		&\varphi(w)=\tilde{\varphi}(\gamma(w))+\log\left|\frac{\partial\gamma(w)}{\partial w}\right|^2=\tilde{\varphi}(\gamma(w))-2\log\left|1-w\right|^2,\\
		&\log|\gamma(w)-\gamma(w_k)| = \left\{
		\begin{split}
		&\log|\gamma(w)-\gamma(w_k)|=\log\left|\frac{w-w_k}{(1-w)(1-w_k)}\right|   \:\:\:\:\:\: &k\leq n-2, \\
		& \log|\gamma(w)-\gamma(\infty)|=\log\left|\frac{1}{1-w}\right|    \:\:\:\:\:\:& k=n-1.\\
		\end{split}
		\right.
		\end{aligned}
		\end{equation} 
		Using these we again decompose \eqref{variationH2} and write for $k\leq n-2$:
		\begin{equation}\nonumber
		\begin{aligned}
		\tilde{I}_1=& - hm\sum_{k=1}^{s}\left( \lim_{\gamma(w) \to \gamma(w_{k})}\left(\tilde{\varphi}(\gamma(w)) + \left(1-\frac{1}{m}\right) \log|\gamma(w)-\gamma(w_k)|^2\right)\right)\\
		& + hm\sum_{k=1}^{s}\left( \lim_{w \to w_{k}}\left(\varphi(w) + \left(1-\frac{1}{m}\right) \log|w-w_{k}|^2\right)\right)\\
		& - h'm'\sum_{k=s+1}^{n-2}\left( \lim_{\gamma(w) \to \gamma(w_{k})}\left(\tilde{\varphi}(\gamma(w)) + \left(1-\frac{1}{m'}\right) \log|\gamma(w)-\gamma(w_k)|^2\right)\right)\\
		&+h'm'\sum_{k=s+1}^{n-2}\left( \lim_{w \to w_{k}}\left(\varphi(w) + \left(1-\frac{1}{m'}\right) \log|w-w_{k}|^2\right)\right)\\
		& =- hm\sum_{k=1}^{s}\left( \lim_{w \to w_k}\left(2\log \left|1-w\right|^2 + \left(1-\frac{1}{m}\right) \log\left|\frac{w-w_k}{(1-w)(1-w_k)}\right|^2\right)\right)\\
		& + hm\sum_{k=1}^{s}\left( \lim_{w \to w_{k}}\left( \left(1-\frac{1}{m}\right) \log|w-w_{k}|^2\right)\right)\\
		& - h'm'\sum_{k=s+1}^{n-2}\left( \lim_{w\to w_k}\left(2\log\left|1-w\right|^2+ \left(1-\frac{1}{m'}\right) \log\left|\frac{w-w_k}{(1-w)(1-w_k)}\right|^2\right)\right)\\
		&+h'm'\sum_{k=s+1}^{n-2}\left( \lim_{w \to w_{k}}\left( \left(1-\frac{1}{m'}\right) \log|w-w_{k}|^2\right)\right)\\
		& =-2h\sum_{k=1}^{s}\log\left|1-w_k\right|^2-2h'\sum_{k=s+1}^{n-2}\log\left|1-w_k\right|^2\\
		& =-2h\sum_{k=1}^{s}\log\left|1-w_k\right|^2-2h'\sum_{k=s+1}^{n-3}\log\left|1-w_k\right|^2,
		\end{aligned}
		\end{equation}
		where in the last equality, we used the fact that $w_{n-2}=0$.
		For $k=n-1$ we have
		\begin{equation}\nonumber
		\begin{aligned}
		\tilde{I}_2=& - h'm' \lim_{\gamma(w) \to \gamma(\infty)}\left(\tilde{\varphi}(\gamma(w)) + \left(1-\frac{1}{m'}\right) \log|\gamma(w)-\gamma(\infty)|^2\right)\\
		&+h'm'\lim_{w \to 1}\left(\varphi(w) + \left(1-\frac{1}{m'}\right) \log|1-w|^2\right)\\
		& = - h'm' \lim_{w\to \infty}\left(\varphi(w)+2\log \left|1-w\right|^2+ \left(1-\frac{1}{m'}\right) \log\left|\frac{1}{1-w}\right|^2\right)\\
		&+h'm'\lim_{w \to 1}\left(\varphi(w)+ \left(1-\frac{1}{m'}\right) \log|1-w|^2\right),
		\end{aligned}
		\end{equation}
		and for the point at the infinity, we have
		\begin{equation}\nonumber
		\begin{aligned}
		\tilde{I}_3=& - h'm' \lim_{\gamma(w) \to \gamma(1)}\left(\tilde{\varphi}(\gamma(w)) + \left(1+\frac{1}{m'}\right) \log|\gamma(w)|^2\right)\\
		&+h'm'\lim_{w \to \infty}\left(\varphi(w) + \left(1+\frac{1}{m'}\right) \log|w|^2\right)\\
		& = - h'm' \lim_{w\to 1}\left(\varphi(w)+2\log \left|1-w\right|^2+ \left(1+\frac{1}{m'}\right) \log\left|\frac{w}{1-w}\right|^2\right)\\
		&+h'm'\lim_{w \to \infty}\left(\varphi(w)+ \left(1+\frac{1}{m'}\right) \log|w|^2\right).
		\end{aligned}
		\end{equation}
		Thus we see that $\tilde{I}_2+\tilde{I}_3=0$ and consequently again for all $j$:
		\begin{equation}
		\Delta\log\mathsf{H}=\sum_{i=1}^{3}\tilde{I}_i=-2h\sum_{k=1}^{s}\log\left|1-w_k\right|^2-2h'\sum_{k=s+1}^{n-3}\log\left|1-w_k\right|^2=-2\log|f_{(\sigma_{j,j+1},\sigma'_{n-1,n})}|,
		\end{equation}
		where we again used (\ref{cocyclesdirectprod}) in the last equality.
	\end{itemize}
	By putting together all of the cases above, we see that
	\begin{equation}\label{Hmetric}
	\Delta\log\mathsf{H}=\log \mathsf{H}\circ(\sigma_{j,j+1},\sigma'_{i,i+1})-\log \mathsf{H}=-2\log|f_{(\sigma_{j,j+1},\sigma'_{i,i+1})}|,\:\:\:\:\:\:\forall i,j.
	\end{equation}
	Had we chosen any other way to fix $0,1$ and $\infty$ by finding the appropriate generators and 1-cocycles, the calculations analogous to the one above would yield the same result. In fact, the inclusion of $h_km_k$'s in the definition of $\mathsf{H}$ ensures that no matter how we choose to fix these points, this lemma holds. Finally, the case with more kinds of branch points can be derived inductively from the calculation above. Thus \eqref{Hmetric} holds in general, and it means that under the action of the elements of $\symm{\sigtype}$, $\mathsf{H}$ transforms according to the rule $(\mathsf{H\circ\eta})|f_\eta|^2=\mathsf{H}$. This means that  $\mathsf{H}$ is a Hermitian metric in the holomorphic $\mathbb{Q}$-line bundle $\lambda_{0,\boldsymbol{m}}$ over $\symmoduli_{0,n}(\boldsymbol{m})$. 
\end{proof}
\begin{remark}
	When $O$ is a punctured orbifold Riemann surface with signature\hspace{-.5mm} $(0;m_1,\dots,m_{n_e}$
	$;n_p>3)$, the Hermitian metric $\mathsf{H}$ on $\mathbb{Q}$-line bundle $\lambda_{0,\boldsymbol{m}}$ is defined as (see \ref{hi})
	\begin{equation*}
	\mathsf{H}= \Lponetial_{1}^{m_1h_1}\dotsm\Lponetial_{n_e}^{m_{n_e} h_{n_e}}\Lponetial_{n_e + 1}\dotsm \Lponetial_{n-1} \Lponetial_{n}^{-1}.
	\end{equation*} 
	It is worth noting that the proof of Lemma~\ref{lemma:Hermitianmetriconlambdawnogenus} for the case of punctured Riemann surfaces can be found in \cite{park2015potentials} and is just a matter of redoing the calculations above with the appropriate 1-cocycles.
\end{remark}
The Poincar\'{e} metric on $\UHP$ can be push-forwarded by the covering map $J: \UHP \to O$ to obtain the hyperbolic metric $e^{\varphi(w)} |\dd{w}|^2$ on orbifold Riemann surface $O$ as follows,
\begin{equation}\label{Liouvillefield}
e^{\varphi(w)} = \frac{\left|J^{-1}(w)'\right|^2}{\left(\Im J^{-1}(w)\right)^2}.
\end{equation}
The condition that the curvature is constant and equal to $-1$, except at singularities, means that the function $\varphi(w)$ satisfies the Liouville's equation
\begin{equation}\label{Liouvilleequation}
\partial_w \partial_{\bar{w}} \varphi = \frac{1}{2} e^{\varphi},
\end{equation}
on $X_O^{\text{reg}}$. Moreover, as discussed in lemma~\ref{lemma:asymptotics}, one can derive the following asymptotic behavior for $\varphi(w)$ near cusps and conical singularities: 
\begin{equation*}
\varphi(w) = \left \{
\begin{split}
&-2 (1-\frac{1}{m_i}) \log|w-w_{i}| + \log \frac{4|J^{(i)}_1|^{-\frac{2}{m_i}}}{m_i^2} + \Sorder{1} \hspace{1.4cm} w \to w_{i},\\ 
&-2 \log |w-w_{j}| - 2 \log\left|\log\left|\frac{w-w_{j}}{J^{(j)}_{1}}\right|\right| + \Sorder{1}\hspace{2cm}  w \to w_{j},\\
& -2 \log|w| - 2 \log\log \left|\frac{w}{J^{(n)}_{-1}}\right| + \order{|w|^{-1}}, \hspace{2.8cm} w \to \infty ,
\end{split}
\right.
\end{equation*}
with $i=1,\dots,n_e$ and $j=n_e+1,\dots,n-1$. Also, we have\footnote{The Liouville's theorem in complex analysis is used.} 
\begin{equation}\label{EnergyMomentumExpansion}
\begin{split}
\hspace{1mm} T_{\varphi}(w)= \partial_w^2 \varphi - \frac{1}{2} (\partial_w \varphi)^2  =\sum_{i=1}^{n-1}\left(\frac{h_i}{2(w-w_i)^2} + \frac{c_i}{w-w_i}\right),
\end{split}
\end{equation}
with $h_i=1$ for $i=n_e+1,\dots,n-1$ and
\begin{equation}\label{EnergyMomentumExpansioninfty}
T_{\varphi}(w) = \frac{1}{2 w^2} + \frac{c_n}{w^3}+ \order{|w|^{-4}} \quad \text{as} \quad w \to \infty,
\end{equation}
where
\begin{equation}\label{accessory}
\left\{
\begin{split}
&c_{i} \equiv -\frac{h_i  J^{(i)}_2}{\left(J^{(i)}_1\right)^2}, \hspace{2.2cm} i=1,\dots,n_e,\\
& c_{j} \equiv - \frac{J^{(j)}_2}{\left(J^{(j)}_1\right)^2},  \hspace{2.1cm}j=n_e+1,\dots,n-1,\\
& c_n \equiv J_{0}^{(j)},  \hspace{3.1cm} j=n,
\end{split}
\right.
\end{equation}
are the so-called \emph{accessory parameters} of the Fuchsian differential equation. The accessory parameters $c_k = c_k(w_1,\dots,w_{n-3})$ for $k=1,\dots,n$, can be regarded as real-analytic functions on $\moduli_{0,n}$.\footnote{This is a consequence of the fact that the Liouville filed $\varphi$ is a real-analytic function on $\moduli_{0,n}$.} In addition, the expansion of \eqref{EnergyMomentumExpansion} around $w\rightarrow \infty$ gives
\begin{equation*}
T_\varphi(w)=\frac{1}{w}\sum_{k=1}^{n-1}c_k+\frac{1}{w^2}\sum_{k=1}^{n-1}\left(\frac{h_k}{2}+c_k w_k\right)+\frac{1}{w^3}\sum_{k=1}^{n-1}w_k(h_k+c_k w_k)+\order{|w|^{-4}} \quad \text{as} \quad w \to \infty.
\end{equation*}
By equating the above expansion with \eqref{EnergyMomentumExpansioninfty}, we get three conditions on the accessory parameters:
\begin{equation}\label{accessoryconstraints}
\sum_{k=1}^{n-1} c_k = 0, \qquad \sum_{k=1}^{n-1} (h_k + 2c_k w_k) = 1, \qquad \sum_{k=1}^{n-1} w_k (h_k +c_k w_k) = c_n,
\end{equation}
which enables us to express $c_{n-2}$, $c_{n-1}$, and $c_{n}$ explicitly in terms of $w_1,\dots,w_{n-1}$ and the remaining $n-3$ accessory parameters.

Consider the Riemann orbisurface $O\cong [\UHP\slash\Gamma]$ as a base point in the Teichm\"{u}ller space $\teich_{0,n}(\boldsymbol{m})$. Moreover, consider the solution of the Beltrami equation (\ref{Beltramieq})
such that the fixed points of $\kappa_{n-2},\kappa_{n-1},\kappa_{n}$ are at $z_{n-2} = 0$, $z_{n-1} = 1$, and $z_n =\infty$. Then the generators $\kappa_{n-2}^{\mu},\kappa_{n-1}^{\mu}, \kappa_n^{\mu}$ of $ \Gamma^{\mu} = f^{\mu} \circ \Gamma \circ (f^{\mu})^{-1}$ will also have fixed points $0$, $1$, and $\infty$ respectively. Accordingly, the Riemann orbisurface $O^{\mu}\cong [\UHP\slash\Gamma^{\mu}]$ can be uniquely and complex-analytically embedded in $\hat{\cmpx}$ in such a way that the punctures on $O^{\mu}$ corresponding to the elements $\kappa_{n-2}^{\mu}$,  $\kappa_{n-1}^{\mu}$, and $\kappa_n^{\mu}$ are mapped into $0$, $1$, and $\infty$. Denote by $J_{\mu}$ the normalized covering map $J_{\mu}: \UHP \to O^{\mu}$ corresponding to this embedding and let $w_i^{\mu} = \left(J_{\mu} \circ f^{\mu}\right)(z_i)$. Then, the map $\Psi: \teich_{0,n}(\boldsymbol{m}) \to \cmpx^{n-3}$, defined by
\begin{equation}\label{psicovering}
\Psi \circ \Phi(\mu) = (w_1^{\mu},\dots,w_{n-3}^{\mu}) \in \cmpx^{n-3},
\end{equation}
is well defined, and its image in $\cmpx^{n-3}$ coincides with $\moduli_{0,n}$. According to the above considerations, we have the following closed commutative diagram
\begin{equation}\label{diagramFmu}
\begin{CD}
\UHP	@> f^{\mu}>> \UHP\\
@V J VV		@VV J_{\mu} V\\
O 		@>F^{\mu}>> O^{\mu}
\end{CD}
\end{equation}
where $F^{\mu}$ is the mapping of $O$ onto $O^{\mu}$. From the above commutative diagram, we can deduce that the $F^{\mu}$ is a quasi-conformal homeomorphism of $\cmpx$ onto itself with
\begin{equation}\label{Beq}
\partial_{\bar{w}} F^{\mu} = M \partial_w F^{\mu},
\end{equation}
where 
\begin{equation}\label{Mdef}
M \equidef  \left(\mu \circ J^{-1}\right) \frac{\overline{(J^{-1})'}}{(J^{-1})'}.
\end{equation}

Consider the Beltrami differential $\varepsilon \mu$, where $\mu \in \diffspace^{-1,1}(\UHP,\Gamma)$ and $\varepsilon$ is a sufficiently small complex number. The function, $f^{\varepsilon \mu}(z)$ is a real-analytic function of $\varepsilon$ for each particular $z \in \mathbb{H}$. Let
\begin{equation}
\dot{f}_{+}^{\mu} \equidef  \left.\left(\pdv{\varepsilon} f^{\varepsilon \mu}\right) \right|_{\varepsilon=0}, \qquad \dot{f}_{-}^{\mu} \equidef  \left.\left(\pdv{\bar{\varepsilon}} f^{\varepsilon \mu}\right) \right|_{\varepsilon=0};
\end{equation}
then (see \cite[Section V.C]{Ahlfors_quasiconformal_06})
\begin{equation}\label{fdotmupm}
\left\{
\begin{split}
& \dot{f}_{+}^{\mu} = -\frac{1}{\pi} \iint_{\UHP} \mu(z^{\prime}) \, R(z^{\prime},z) \dd[2]{z^{\prime}},\\
& \dot{f}_{-}^{\mu} = -\frac{1}{\pi} \iint_{\UHP} \overline{\mu(z^{\prime})} \, R(\overline{z^{\prime}},z) \dd[2]{z^{\prime}},
\end{split}
\right.
\end{equation}
where 
\begin{equation}\label{Rab}
R(z^{\prime},z) \equidef \frac{1}{z^{\prime}-z} + \frac{z-1}{z^{\prime}} - \frac{z}{z^{\prime}-1} = \frac{z(z-1)}{(z^{\prime}-z) z^{\prime} (z^{\prime}-1)}.
\end{equation}
In turn, the function $F^{\mu}(w)$ is holomorphic with respect to $\varepsilon$ for each particular $w\in\cmpx$, and 
\begin{equation}\label{Fdot}
\dot{F}^{\mu}(w) = - \frac{1}{\pi} \iint_{\cmpx} M(w^{\prime}) \, R(w^{\prime},w) \dd[2]{w^{\prime}},
\end{equation}
where $\dot{F}^{\mu}$ is understood to be given by $(\partial F^{\varepsilon \mu}/ \partial\varepsilon)|_{\varepsilon=0}$ and $M$ was defined by Eq.\eqref{Mdef}.\footnote{According to \eqref{Beq}, $\partial_{\bar{w}}\dot{F}= M$. The Green-function equation and solution for this equation are given by $\partial_{\bar{w}}R(w^{\prime},w) = -\pi \delta(w^{\prime},w)$ and \eqref{Fdot}, respectively. Therefore, the kernel $R(w^{\prime},w)$, roughly speaking, inverts the action of $\bar{\partial}$-operator on Beltrami differentials on $\hat{\cmpx}$. The precise statement (see \cite[Lemma~5]{Takhtajan:2001uj} and \cite[Section V.C]{Ahlfors_quasiconformal_06}) is essentially a version of the Pompeiu formula.} Proof of these assertions can be found in \cite{Ahlfors_quasiconformal_06,Ahlfors1960RIEMANNSMT,Bers_1960}. Moreover, let
\begin{equation}\label{Ridef}
R_i(w) = - \frac{1}{\pi} R(w,w_i) = - \frac{w_i(w_i-1)}{\pi (w-w_i) w (w-1)}, \qquad \qquad i=1,\dots,n-3;
\end{equation}
where $R_i$s are linearly independent and generate the space $\Hilbert^{2,0}(O)$. Denote by $\{Q_i\}_{i=1,\dots,n-3}$, the basis in $\Hilbert^{2,0}(O)$ biorthogonal to $\{R_i\}_{i=1,\dots,n-3}$ in the sense of inner product \eqref{Qinnerprod} --- i.e. $\langle R_i, Q_j \rangle = \delta_{ij}$ where $\delta_{ij}$ is the Kronecker delta. The  desired basis in $\Hilbert^{-1,1}(\UHP,\Gamma) \cong T_{\Phi(0)}\teich(\Gamma)$ has the form
\begin{equation}\label{Beltramibasis}
\mu_i(z) = \rho(z)^{-1} \overline{q_i(z)},
\end{equation} 
where $q_i(z) = Q_i \circ J(z) \, J'(z)^2$ for $i=1,\dots,n-3$ form a basis of the complex vector space  $\Hilbert^{2,0}(\UHP,\Gamma) \cong T_{\Phi(0)}^{\ast} \teich(\Gamma)$. These $q_i$s are biorthogonal to $r_i = R_i \circ J \, J'^2 \in \Hilbert^{2,0}(\UHP,\Gamma)$ with respect to the Petersson inner product.\footnote{In the space $\Hilbert^{2,0}(\mathbb{H},\Gamma)$, it is defined as $\langle q_1,q_2\rangle = \iint_{\fund(\Gamma)} q_1(z) \, \overline{q_2(z)} \, \rho^{-1}(z) \dd[2]{z}.$} The basis $q_i^{^{_{\Phi(\mu)}}} \in \Hilbert^{2,0}(\UHP,\Gamma^{\mu})$ and $\mu_i^{^{_{\Phi(\mu)}}} \in \Hilbert^{-1,1}(\UHP,\Gamma^{\mu})$ for $\Phi(\mu) \in \teich(\Gamma)$ can also be defined in a similar way. Then, the following lemma connects the motion of punctures and conical singularities on $\hat{\cmpx}$ with the geometry of  Teichm\"{u}ller space:
\begin{lemma}\label{Lemma:covering}
	For The mapping $\Psi: \teich_{0,n}(\boldsymbol{m}) \to \moduli_{0,n}$ is a complex-analytic covering, and we have
	\begin{equation}\label{dpsi}
	\dd{\Psi_{\Phi(\mu)}}(\mu_i^{^{_{\Phi(\mu)}}}) = \pdv{w^{\mu}_i} \quad \text{and} \quad \Psi_{\Phi(\mu)}^{\ast}(\dd{w^{\mu}_i}) = \mathrm{r}_i^{^{_{\Phi(\mu)}}}, \qquad i=1,\dots,n-3.
	\end{equation}
\end{lemma}
\begin{proof}
	Let us start with proving the statement $\dd{\Psi_{\Phi(\mu)}}(\mu_i^{^{_{\Phi(\mu)}}}) = \pdv{w^{\mu}_i}$ which is actually repeating the proof of Lemma 3 in \cite{Zograf1988ONLE}; 
	Remark~\ref{remark:nobase} implies that it is sufficient to verify this statement at the point $\Phi(0)$. In a neighborhood of $\Phi(0) \in \teich(\Gamma)$ the Bers' coordinates $(t_1,\dots,t_{n-3}) \in \cmpx^{n-3}$ are determined in the basis $\mu_1,\dots,\mu_{n-3} \in \Hilbert^{-1,1}(\UHP,\Gamma)$ from the expansion
	\begin{equation}
	\mu = \sum_{i=1}^{n-3} t_i \mu_i,
	\end{equation}
	where $\mu \in \Hilbert^{-1,1}(\UHP,\Gamma) \cap \deformspace(\Gamma)$. In these coordinates, the mapping $\Psi$ is given by
	\begin{equation}
	(t_1,\dots,t_{n-3}) \overset{\Psi}{\longmapsto} \left(F^{\mu}(w_1), \dots, F^{\mu}(w_{n-3})\right),
	\end{equation}
	where $(w_1,\dots,w_{n-3}) = \Psi\circ\Phi(0) \in \moduli_{0,n}$. Since $F^{\mu}$ depends complex-analytically on $t_1,\dots,t_{n-3}$, the mapping $\Psi$ is also complex-analytic. We now compute its differential $\dd{\Psi}$ at the point $\Phi(0) \in \teich(\Gamma)$: It follows from the definitions of $\mu_i$ and $M_i$, equations \eqref{Beltramibasis} and \eqref{Mdef}, as well as Eq.\eqref{Liouvillefield} that 
	\begin{equation}\label{MiQi}
	M_i = \mu_i \circ J^{-1} \frac{\overline{(J^{-1})'}}{(J^{-1})'} = \left(\Im J^{-1}\right)^2 \overline{q_i \circ J^{-1}} \frac{\overline{(J^{-1})'}}{(J^{-1})'} = e^{- \varphi} \overline{Q_i},
	\end{equation}
	where $Q_i = q_i \circ J^{-1} (J^{-1})'^2$. Then, using the above equation, the definition of $R_i$s in \eqref{Ridef}, and Eq.\eqref{Fdot}, we get that 
	\begin{equation}
	\begin{split}
	\dot{F}^{\mu_i}(w_j) & = \iint_{\cmpx} M_i(w) \, R_j(w) \dd[2]{w}  \overset{\eqref{MiQi}}{=}  \iint_{\cmpx} e^{- \varphi} \overline{Q_i} \, R_j(w) \dd[2]{w}  \overset{\eqref{Qinnerprod}}{=}  \langle R_j, Q_i \rangle = \delta_{ij}.
	\end{split}
	\end{equation}
	This precisely means that in these coordinates the differential $\dd{\Psi}\big|_{\Phi(0)}$ is given by the identity matrix --- i.e. $\dd{\Psi}_{\Phi(0)}(\mu_i) = \pdv{w_i}$. Additionally, the statement $\Psi_{\Phi(0)}^{\ast}(\dd{w_i}) = \mathrm{r}_i$ follows from $\dd{\Psi}_{\Phi(0)}(\mu_i) = \pdv{w_i}$ by observing that (1,0)-forms $\dd{w}_i$ are dual to vector fields $\pdv{w_i}$ and the quadratic differentials $r_i$ are dual to harmonic Beltrami differentials $\mu_i$ with respect to the Kodaira-Serre paring~\eqref{pairing} --- i.e. $(\mu_i,r_j) = \delta_{ij}$. Accordingly, it follows from Eq.~\eqref{dpsi} that the mapping $\Psi$ is a local diffeomorphism. We now show that $\Psi$ is a covering: For $\eta \in \Automorphism_{\ast}(\Gamma)$, denote by $z_i^{\eta}$ $(i=1,\dots,n_e)$ the fixed points of elliptic element $\eta(\tau_i) \in \Gamma$ of order $m_i$ and by $z_{n_e+j}^{\eta}$ the fixed points of the parabolic elements $\eta(\kappa_{j}) \in \Gamma$ for $j=1,\dots,n_p$. Let $\mathsf{s}_{\eta}(w_1,\dots,w_{n}) \equidef \left(J(z_1^{\eta}), \dots, J(z_{n}^{\eta})\right)$; the correspondence $\eta \mapsto \mathsf{s}_{\eta}$ determines an epimorphism (i.e. a surjective group homomorphism) of the group $\modular(\Gamma)$ onto the product symmetric group $\symm{\sigtype}$.\footnote{As a concrete example, consider the case of $n$-punctured sphere. In that case, the mapping class group $\modular_{0,n}$ is given by the braid group $B_{n}$ quotiented by its center  \cite{farb2011primer}.} The kernel of this epimorphism is $\operatorname{Aut}(\Psi)$ \cite{Zograf1988ONLE}. The mapping $\Psi$ is invariant under $\operatorname{Aut}(\Psi)$ and it is clear that $\moduli_{0,n} = \teich_{0,n}(\boldsymbol{m})/\operatorname{Aut}(\Psi)$. Hence, $\Psi$ is a covering with automorphism group $\operatorname{Aut}(\Psi)$.
\end{proof}
\begin{remark}
On $\teich_{0,n}(\boldsymbol{m})$, each elliptic and cuspidal TZ metric\footnote{The  $\langle \cdot, \cdot\rangle^{\text{ell}}_{\text{TZ},i}$ for $i=1,\dots,n_e$ and $\langle \cdot, \cdot\rangle^{\text{cusp}}_{\text{TZ},n_e+j}$ for $j=1,\dots,n_p$.} remains unchanged under the automorphism group of the covering  $\Psi:\teich_{0,n}(\boldsymbol{m}) \to \moduli_{0,n}$ --- i.e. the pure mapping classes $\puremod$. Furthermore, each metric establishes a K\"{a}hler metric on  $\moduli_{0,n}$.
\end{remark}
\begin{remark}\label{EnergyMomentumEi}
	We can also rewrite the expression for Energy-Momentum Tensor \eqref{EnergyMomentumExpansion} using $R_i(w)$, \begin{equation}\label{EMT2}
	\text{Sch}\left(J^{-1};w\right) = \sum_{i=1}^{n} h_i \varE_i(w) - \pi \sum_{i=1}^{n-3} c_i R_i(w),
	\end{equation}
	where 
	\begin{equation}\label{varE}
	\left\{
	\begin{split}
	& \varE_i(w) = \frac{1}{2(w-w_i)^2} - \frac{1}{2w(w-1)},  \quad &i=1,\dots,n-1,\\
	& \varE_i(w) = \frac{1}{2w(w-1)}, \quad & i=n,
	\end{split}
	\right.
	\end{equation}
	and $h_i = 1$ for $i=n_e+1,\dots,n$. Accordingly, one can also define the $e_i(z) = \varE_i \circ J \, J'^2$ on $\UHP$. These functions are actually the automorphic forms of weight four for $\Gamma$ and they have non-vanishing constant terms at the singularities. Let us obtain the equation \eqref{EMT2} for the simplest case with one branch point of order $m$ at $w_1$ and three punctures at $w_2=0,w_3=1,$ and $w_4\rightarrow\infty$, respectively. Solving \eqref{accessoryconstraints} in favor of $c_3$ and $c_2$ gives
	\begin{equation*}
	c_3 = -(c_1+c_2),\hspace{.5cm}c_2 = -\frac{1+h_1+2c_1(w_1-w_3)}{2(w_2-w_3)}.
	\end{equation*}
	By substituting the above relations in \eqref{EnergyMomentumExpansion} and noting that $w_2=0,w_3=1$, one obtains
	\begin{equation*}
	\begin{aligned}
	&T_{\varphi} = h_1\left(\frac{1}{2(w-w_1)^2}-\frac{1}{2w(w-1)}\right)+\frac{1}{2w^2}+\frac{1}{2(w-1)^2}-\frac{1}{2w(w-1)}
	\\
	&\hspace{1cm}+c_1\left(\frac{1}{w-w_1}+\frac{w_1-1}{w}-\frac{w_1}{w-1}\right),
	\end{aligned}
	\end{equation*}
	which by using \eqref{varE} and \eqref{Ridef} is equal to \eqref{EMT2}. The desired general case can be obtained in the same way.
\end{remark}
\begin{corollary}\label{corollary:FdotiAsymp}
	The following statements are true:
	\begin{enumerate}[(i)]
		\item $\dot{F}^i(0) = \dot{F}^i(1) =0$\hspace{.7cm}and\hspace{.7cm} 
		$\dot{F}^i(w) =  \Sorder{|w|^2} \quad \text{as} \quad w \to \infty$.
		\item $\partial_{\bar{w}} \dot{F}^i = M_i = e^{-\varphi} \overline{Q_i}$.
		\item The functions $\dot{F}^{i}(w)$ have the following asymptotics
		\begin{equation*}
		\hspace{-.5cm}\dot{F}^i(w) = \left\{
		\begin{split}
		& \delta_{ij} + (w-w_j) \, \partial_w \dot{F}^i(w_j) + \Sorder{\frac{|w-w_j|}{\log|w-w_j|}} \hspace{.1cm} (j=1,\dots,n_e) \hspace{.1cm}\text{as}\hspace{.2cm}w \to w_j,\\
		&\delta_{ij} + (w-w_j) \, \partial_w \dot{F}^i(w_j) + \Sorder{\frac{|w-w_j|}{\log|w-w_j|}} \hspace{.1cm} (j=n_e+1,\dots,n-1)  \hspace{.1cm}\text{as}\hspace{.2cm} w \to w_j,\\
		&w \, \partial_w \dot{F}^i(\infty) + \Sorder{\frac{|w|}{\log|w|}} \hspace{2cm} \hspace{.1cm}\text{as}\hspace{.2cm} w \to \infty.
		\end{split}
		\right.
		\end{equation*}
	\end{enumerate}
\end{corollary}
\begin{proof}
	The part \emph{(i)} can be proved by using the equation \eqref{Fdot} and the expression for $R(w,w_i)$ in \eqref{Ridef}. The part \emph{(ii)} can also be proved by using \eqref{Fdot} and  \eqref{MiQi}. The general idea of proof the part \emph{(iii)} can also be found in  \cite[Remark~3]{park2015potentials}.
\end{proof}

In section~\ref{sec:action}, we will need the derivatives of $\varphi$ with respect to the variables $w_1,\dots,w_{n-3}$ and the following lemma will enable us to give a geometric description of them (see \cite[Lemma~4]{Zograf1988ONLE}): 
\begin{lemma}\label{varp}
	The Liouville filed $\varphi$ is a continuously differentiable function on $\moduli_{0,n}$, and 
	\begin{equation}\label{widerivative}
	\partial_{w_i} \varphi + \dot{F}^{i} \partial_w \varphi + \partial_w \dot{F}^{i} = 0, \quad \text{for} \quad i=1,\dots,n-3,
	\end{equation}
\end{lemma}
\begin{proof}
	It can be proved exactly in the same way as Lemma 4 in \cite{Zograf1988ONLE}. Let $\Gamma$ be a Fuchsian group of the first kind that uniformizes the orbifold Riemann surface $O$ and let $f^{\mu}(z)$ be the unique solution of the Beltrami equation \eqref{Beltramieq} with $\Hilbert^{-1,1}(\UHP,\Gamma)\ni \mu = t_1 \, \mu_1 + \dotsm + t_{n-3} \, \mu_{n-3}$ that fixes the points $z_{n-2}=0$, $z_{n-1}=1$, and $z_{n}=\infty$. The function $F^{\mu}(w) = (J_{\mu} \circ f^{\mu} \circ J^{-1})(w)$ is differentiable with respect to $w$ on $O$ and depends analytically on the Bers' coordinates $t_1,\dots,t_{n-3}$. It follows from the commutative diagram \eqref{diagramFmu} that $J_{\mu}$ and $J'_{\mu}$ are continuously differentiable with respect to the Bers' coordinates $t_1,\dots,t_{n-3}$ and that suitable branches of the functions $J_{\mu}^{-1}$ and $(J_{\mu}^{-1})'$ have this property locally outside the set of singular points. The continuous differentiability of $\varphi$ on  $\moduli_{0,n}$ now follows from Eq.\eqref{Liouvillefield} and Lemma~\ref{Lemma:covering}. As for equation~\eqref{widerivative}, it is a reformulation of lemma \ref{Lemma:Ahlfors} due to Ahlfors \cite{Ahlfors1961SomeRO} on the vanishing of the first variation of the area element in the Poincar\'{e} metric on $O$ under quasi-conformal mappings corresponding to harmonic Beltrami differentials: for any $\mu \in \Hilbert^{-1,1}(\UHP,\Gamma)$
	\begin{equation}\label{firstvariation}
	\Lie_{\mu}\rho = \pdv{\varepsilon}\Bigg|_{\varepsilon=0} (f^{\varepsilon\mu})^{\ast}(\rho)= \pdv{\varepsilon} \frac{|\partial_z f^{\varepsilon\mu}|^2}{\left(\Im f^{\varepsilon\mu}\right)^2}\Bigg|_{\varepsilon=0} = 0.
	\end{equation}
	
	Let $e^{\varphi^{\mu}(w)} \, |\dd{w}|^2$ be the hyperbolic metric on the Riemann orbisurface $O^{\mu}$ where $\varphi^{\mu}(w) = \varphi\left(w; F^{\mu}(w_1),\dots,F^{\mu}(w_{n-3})\right)$. From \eqref{diagramFmu} one has
	\begin{equation*}\label{pullbackfromXmu}
	(f^{\mu})^{\ast}(\rho) = (f^{\mu})^{\ast}\underbrace{(J_{\mu})^{\ast}(e^{\varphi^{\mu}})}_{\rho}.
	\end{equation*}
	Then, using Eq.\eqref{pullback} as well as $F^{\mu} \circ J = J_{\mu} \circ f^{\mu}$, one gets
	\begin{equation}
	\frac{|\partial_z f^{\mu}|^2}{\left(\Im f^{\mu}\right)^2} = \exp(\varphi^{\mu} \circ J_{\mu} \circ f^{\mu}) \left|\partial_z(J_{\mu} \circ f^{\mu})\right|^2 = \exp(\varphi^{\mu} \circ F^{\mu} \circ J ) \left|\partial_w F^{\mu} \circ J\right|^2 |J'|^2.
	\end{equation}
	Finally, it follows from \eqref{firstvariation} that
	\begin{equation}
	\pdv{\varepsilon}\Bigg|_{\varepsilon=0} e^{\varphi^{\varepsilon\mu} \circ F^{\varepsilon\mu} \circ J } \left|\partial_w F^{\varepsilon\mu} \circ J\right|^2 |J'|^2 = \pdv{\varphi^{\varepsilon \mu}}{\varepsilon}\Bigg|_{\varepsilon=0} + \partial_w \varphi \dot{F}^{\mu}+ \partial_w \dot{F}^{\mu}=0,
	\end{equation}
	which by setting $\mu =\mu_i$ and recalling Lemma \ref{Lemma:covering} (i.e. $d\Psi_{\Phi(0)}(\mu_i) = \partial/\partial_{w_i}$) gives us our desired result \eqref{widerivative}.
\end{proof}
\begin{corollary}\label{corollary:widerivativeasymptotics}
	According to Lemma \ref{lemma:asymptotics}, the $\partial_{w_i} \varphi$ has the following asymptotic expansion near the singular points
	\begin{equation}\nonumber
	\partial_{w_i} \varphi(w) = \left\{
	\begin{split}
	& - \delta_{ij}\partial_w \varphi(w) - \big((w-w_j) \partial_w \varphi(w)+1\big) \partial_w \dot{F}^i(w) + \Sorder{1} \quad  & \text{as} \quad & w \to w_{j\neq n},\\
	& - \big(w \partial_w \varphi(w)+1\big) \partial_w \dot{F}^i(w) + \Sorder{1} \quad &\text{as} \quad &  w \to w_n=\infty.
	\end{split}
	\right.
	\end{equation}
\end{corollary}
\begin{proof}
	It follows from \eqref{widerivative}
	and from the asymptotics of $\dot{F}^{i}$ (see Corollary~\ref{corollary:FdotiAsymp}).
\end{proof}
\begin{remark}
	One can also use the results of Ahlfors and Wolpert mentioned in subsection~\ref{subsub:variationalformulas}, i.e. lemma~\ref{Lemma:Ahlfors} and Eq.\eqref{WolpertSecVar}, to calculate the variations of  $\exp(\varphi^{\varepsilon\mu}(w)) $ on the orbifold Riemann surfaces $O^{\varepsilon \mu} = F^{\varepsilon \mu}(O)$. To do so, let us use the commutative diagram \eqref{diagramFmu} once again to write
	\begin{equation}
	(F^{\mu})^{\ast} (e^{\varphi^{\mu}}) = (J^{-1})^{\ast}(f^{\mu})^{\ast}(\rho).
	\end{equation}
	Then, it is easy to show that the variations of hyperbolic metrics on Riemann orbisurfaces $O^{\varepsilon \mu}$ are given by the same formulas in subsection~\ref{subsub:variationalformulas} but with $\rho$ replaced by $e^{\varphi}$ and $f_{\mu \overline{\tilde\mu}}$ replaced by $(J^{-1})^{\ast}(f_{\mu \overline{\tilde\mu}}) = f_{\mu \overline{\tilde\mu}} \circ J^{-1}$. Moreover, for any $\mathtt{a} \in \mathbb{R}$ we have
	\begin{equation}\label{pullbackpower}
	(F^{\mu})^{\ast} (e^{\mathtt{a} \varphi^{\mu}}) = \left((F^{\mu})^{\ast} (e^{\varphi^{\mu}})\right)^{\mathtt{a}}.
	\end{equation}
	Therefore, we get the following formulas for the first and second variations of $\exp(\varphi^{\varepsilon\mu}(w))$
	\begin{equation}\label{expphivariations}
	\left\{
	\begin{split}
	& \pdv{\varepsilon}\Bigg|_{\varepsilon=0} (F^{\varepsilon\mu})^{\ast} (e^{\mathtt{a} \varphi^{\varepsilon\mu}}) = 0, \\ \\
	& \pdv[2]{}{\varepsilon_1}{\bar{\varepsilon_2}} \Bigg|_{\varepsilon_1= \varepsilon_2 = 0} (F^{\varepsilon_1\mu+ \varepsilon_2 \tilde\mu})^{\ast} (e^{\mathtt{a} \varphi^{\varepsilon\mu}}) = \frac{\mathtt{a}}{2} e^{\mathtt{a}\varphi} f_{\mu \overline{\tilde\mu}} \circ J^{-1}.
	\end{split}
	\right.
	\end{equation}
\end{remark}

\subsection{Schottky Space $\schottky_{g,n}(\boldsymbol{m})$}\label{subsec:Schottkyspace}
Let us start this subsection by recalling with more detail how a compact Riemann surface $X$ of genus $g\geq2$ is uniformized by a Schottky group. We begin with a few well-known definitions. Schottky groups are an important class of \emph{Kleinian groups}: discrete subgroups of M\"{o}bius group $\PSLC$ that act properly discontinuous on some domain (called region of discontinuity) of the Riemann sphere $\hat{\cmpx}$. A \emph{Schottky group} $\Sigma$ is strictly loxodromic Kleinian group which is also free and finitely generated \cite{Maskit_1967}. If we denote its limit set by $\Lambda$ (which is a Cantor set)\footnote{For more details on geometry of limit sets see Ref.~\cite{Seade_2015}.} then the \emph{region of discontinuity} $\Omega = \hat{\cmpx}\backslash\Lambda$ would be connected. Moreover, a Schottky group $\Sigma$ of rank $g$ will be called \emph{marked} if we choose a relation-free system of generators $L_1, \dots, L_g  \in \PSLC$. There is also a notion of equivalence between two marked Schottky groups: $(\Sigma; L_1, \dots, L_g)$ is equivalent to $(\tilde{\Sigma}; \tilde{L}_1, \dots, \tilde{L}_g)$ if there exists a M\"{o}bius transformation $\varsigma \in \PSLC$ such that $\tilde{L}_i = \varsigma  L_i \varsigma^{-1}$ for all $i=1,\dots,g$. The set of equivalence classes of marked Schottky groups of genus $g$ is called the \emph{Schottky space} of genus $g$ and is denoted by $\schottky_{g}$. Similar to Fuchsian groups, Schottky groups can be used to construct surfaces since the action of $\Sigma$ on $\Omega$ produces a \emph{compact} Riemann surface $\Omega\slash\Sigma$. An important result is that for every marked Schottky group $(\Sigma;L_1,\dots,L_g)$ there is a \emph{fundamental domain} $\SchottkyFund$\footnote{However, this fundamental domain $\SchottkyFund$ is \emph{not uniquely determined} by the choice of marking for the Schottky group $\Sigma$.\label{funddomnotunique}} for $\Sigma$ in $\Omega$. This domain is a (connected) region in $\hat{\cmpx}$ and it is bounded by $2g$ disjoint Jordan curves $C_1, \dots, C_g, C'_1, \dots, C'_g$ with $C'_i = -L_i(C_i)$, $i=1,\dots, g$. The orientations of $C_i$ and $C'_i$ are opposite and related to components of  $\partial \SchottkyFund$. The standard form for representation of each $L_i$ is
\begin{equation}
\bigg(L_i (w) - \tilde{a}_i\bigg)(w-\tilde{b}_i) = \lambda_i\hspace{1mm}\bigg(L_i (w) - \tilde{b}_i\bigg)(w-\tilde{a}_i), \qquad w \in \hat{\cmpx},
\end{equation}
where $\tilde{a}_i$ and $\tilde{b}_i$ are the respective \emph{attracting} and \emph{repelling} fixed points of the loxodromic element $L_i$ and $0 < |\lambda_i| < 1$ is the corresponding multiplier. Given this normal form, one can explicitly construct the fundamental domain $\SchottkyFund$ in the following way:\footnote{For more details, see \cite[Appendix~C]{Gaberdiel:2010jf}.} Let us define the M\"{o}bius transformations
\begin{equation}
\varsigma_{\tilde{a}_i,\tilde{b}_i}(w) = \frac{\tilde{b}_i w + \tilde{a}_i}{w+1},
\end{equation}
satisfying $\varsigma_{\tilde{a}_i,\tilde{b}_i}(0) =  \tilde{a}_i$ and $\varsigma_{\tilde{a}_i,\tilde{b}_i}(\infty) =  \tilde{b}_i$, so that the generators $L_i$ of the marked Schottky group $(\Sigma; L_1, \dots, L_g)$ can be written as
\begin{equation}
L_i = \varsigma_{\tilde{a}_i,\tilde{b}_i} \hspace{1mm}\varsigma_{\lambda_i}\hspace{1mm} \varsigma_{\tilde{a}_i,\tilde{b}_i}^{-1} \quad \text{for} \quad i=1,\dots,g.
\end{equation}
In the above equation, the M\"{o}bius transformation $\varsigma_{\lambda_i}$ is defined by $\varsigma_{\lambda_i}(w) = \lambda_i w$. Then, a fundamental domain for $(\Sigma; L_1, \dots, L_g)$ is given by
\begin{equation}\label{schottkyfunddomaindef}
\SchottkyFund \equidef  \hat{\cmpx} \Big\backslash \bigcup_{i=1}^{g} (\mathrm{D}_i \cup \mathrm{D}_{-i}),
\end{equation}
where 
\begin{equation}
\left\{
\begin{split}
& \mathrm{D}_i = \left\{ w \in \cmpx \, \Bigg| \, \frac{|w-\tilde{a}_i|}{|w-\tilde{b}_i|} < |\mathrm{R}_i|\right\} = \varsigma_{\tilde{a}_i,\tilde{b}_i} \hspace{1mm}\varsigma_{\mathrm{R}_i}(\mathrm{D}),\\ \\
& \mathrm{D}_{-i} = \left\{ w \in \cmpx \, \Bigg| \, \frac{|w-\tilde{b}_i|}{|w-\tilde{a}_i|} < |\mathrm{R}_{-i}|\right\} = \varsigma_{\tilde{a}_i,\tilde{b}_i} \hspace{1mm}\varsigma_{inv} \hspace{1mm}\varsigma_{\mathrm{R}_{-i}}(\mathrm{D}),
\end{split}
\right.
\end{equation}
$\varsigma_{inv}$ is defined by $\varsigma_{inv}(w) = -1/w$, and $\mathrm{D}$ is the unit disk,
\begin{equation}
\mathrm{D}= \left\{ w \in \cmpx \, \Big| \, |w| < 1 \right\}.
\end{equation}
Here $\mathrm{R}_i$ and $\mathrm{R}_{-i}$ represent the radii of disks $\mathrm{D}_i$ and $\mathrm{D}_{-i}$ respectively and satisfy\footnote{Equation~\eqref{radiiprod} makes clear the fact that,  as mentioned in footnote~\ref{funddomnotunique}, the fundamental domain $\SchottkyFund$  cannot be uniquely determined by a choice of marking for the Schottky group $\Sigma$.\label{nonuniqueschottkyfund}}
\begin{equation}\label{radiiprod}
\mathrm{R}_i \mathrm{R}_{-i} = \lambda_i \quad \text{for} \quad i=1,\dots,g.
\end{equation}
The boundary $\partial \SchottkyFund = \bigcup_i C_i \cup C'_i$ has components
\begin{equation}
C_i = \varsigma_{\tilde{a}_i,\tilde{b}_i} \hspace{1mm}\varsigma_{\mathrm{R}_i}(C) \quad \text{and} \quad C'_i = \varsigma_{\tilde{a}_i,\tilde{b}_i} \hspace{1mm}\varsigma_{inv} \hspace{1mm}\varsigma_{\mathrm{R}_{-i}}(C),
\end{equation}
where $C = \partial \mathrm{D}$ is the unit circle. In the rest of this article, we will consistently presume that a marked Schottky group is \emph{normalized}. This means that $\tilde{a}_1$ equals 0, $\tilde{b}_1$ is infinity, and $\tilde{a}_2$ is 1. In particular, this means that $ \infty \notin \SchottkyFund$.

Under the canonical holomorphic map $\Omega \to \Omega\slash\Sigma$, the boundary curves of a standard fundamental domain described above are mapped onto smooth non-intersecting simple closed curves $\mathsf{a}_1, \dots, \mathsf{a}_g$ on the Riemann surface. This motivates the following terminology introduced by Bers \cite{Bers1961UniformizationBB}: A complete set of \emph{retrosections} on a Riemann surface of genus $g$ is a choice of $g$ smooth simple non-intersecting, homologically independent, closed curves $\mathsf{a}_1, \dots, \mathsf{a}_g$. We therefore see that a marked Schottky group, together with the choice of a standard fundamental domain, determines a Riemann surface with a complete set of retrosections. A much more profound statement is that \emph{every} compact Riemann surface can be obtained in this way, which is the content of the classical Koebe's retrosection theorem \cite{Kra_1972}:
\begin{theorem*}[Koebe]
	For every compact Riemann surface $X$ of genus $g$ with a complete set of retrosections $(\mathsf{a}_1, \dots, \mathsf{a}_g)$, there exists a marked Schottky group of genus $g$, $(\Sigma; L_1, \dots, L_g)$, and a fundamental domain $\SchottkyFund \subset \Omega$ for $\Sigma$ with $2g$ boundary curves $C_1,\dots,C_g,C'_1,\dots,C'_g$ such that $X = \Omega \slash \Sigma$ and the map $\Omega \to X$ sends both $C_i$ and $C'_i$ to $\mathsf{a}_i$. Moreover, the marked Schottky group is unique up to equivalence $(\Sigma; L_1, \dots, L_g) \sim (\tilde{\Sigma}; \tilde{L}_1, \dots, \tilde{L}_g)$ defined before as well as $L_i \mapsto L_i^{-1}$.
\end{theorem*}
\begin{remark}\label{Bcycle}
	The above theorem implies that given a compact Riemann surface $X$ uniformized by the marked Schottky group  $(\Sigma; L_1, \dots, L_g)$, we can take the homology classes of $C_1,\dots,C_g$ as the generators $[\mathsf{a}_1], \dots, [\mathsf{a}_g]$ in the symplectic basis of the first homology group $H_1(X,\mathbb{Z})$. However, determining a canonical basis for $H_1(X,\mathbb{Z})$, i.e. a symplectic basis $\left\{[\mathsf{a}_1], \dots, [\mathsf{a}_g], [\mathsf{b}_1], \dots, [\mathsf{b}_g] \right\}$ with intersection pairings given by $\#\left([\mathsf{a}_i], [\mathsf{a}_j]\right) = 0 = \#\left([\mathsf{b}_i], [\mathsf{b}_j]\right)$ and $\#\left([\mathsf{a}_i], [\mathsf{b}_j]\right) = \delta_{ij}$, depends also on the choice of  $\mathsf{b}$-cycles on the Riemann surface $X$. Therefore, we can choose the elements $[\mathsf{b}_1],\dots, [\mathsf{b}_g]$ in the canonical basis of $H_1(X,\mathbb{Z})$ such that the projections of their representative curves onto the marked Schottky group $\Sigma$ are precisely the marked generators $L_1,\dots, L_g$.
\end{remark}
\begin{remark}
The association between the set of normalized marked Schottky groups and the Schottky space $\schottky_g$, found within $\cmpx^{3g-3}$, is evidently bijective.\footnote{The space $\schottky_{g}$ is a finite covering of the moduli space $\moduli_g$ of compact Riemann surfaces.}
\end{remark}
For our purposes in the following sections, it will be crucial to give another (equivalent) definition for the Schottky space $\schottky_{g}$. Let $\diffspace^{-1,1}(\Omega,\Sigma)$ be the complex Banach space of Beltrami differentials for $\Sigma$ and, in analogy with the Teichm\"{u}ller case, let us define the deformation space $\deformspace(\Sigma)$ to be the open ball of radius $1$ (in the sense of $L^{\infty}$-norm) in $\diffspace^{-1,1}(\Omega,\Sigma)$:
\begin{equation}
\deformspace(\Sigma) = \left\{\mu \in \diffspace^{-1,1}(\Omega,\Sigma) \,  \big| \, \norm{\mu}_{\infty} <1 \right\}.
\end{equation}
A homeomorphism $F$ of a plane domain $\Omega$ onto another plane domain $\tilde{\Omega}$ is said to be \emph{quasi-conformal} if it satisfies the Beltrami equation at each point in $\Omega$. For each $\mu \in \deformspace(\Sigma)$, let $F^{\mu}$ be the unique normalized (i.e. $F^{\mu}(0)=0$ and  $F^{\mu}(1)=1$) solution of the corresponding Beltrami equation on $\cmpx$ that gives a quasi-conformal homeomorphism of $\cmpx$ onto itself. Then, the restriction of $F^{\mu}$ to the region of discontinuity $\Omega \subset \cmpx$ gives the desired quasi-conformal homeomorphism $F^{\mu}: \Omega \to \Omega^{\mu}$ and each element $\mu \in \deformspace(\Sigma)$ gives a faithful representation $\varrho_{\mu}$ of $\Sigma$ in $\PSLC$ according to the formula $\sigma \mapsto F^{\mu} \circ \sigma \circ (F^{\mu})^{-1}$, $\sigma \in \Sigma$. As mentioned before, two representations $\varrho_{\mu_1}$ and $\varrho_{\mu_2}$ are equivalent if they differ by an \emph{inner automorphism} of $\PSLC$, i.e., if $\varrho_{\mu_2} = \varsigma \varrho_{\mu_1} \varsigma^{-1}$, $\varsigma \in \PSLC$. Accordingly, the Schottky space $\schottky_{g}$ is defined to be the set of equivalence classes of representations $\left[\varrho_{\mu}\right] : \Sigma \to \PSLC$, $\mu \in \deformspace(\Sigma)$. In other words,
\begin{equation}\label{schottkyspcae}
\schottky_{g} \cong \deformspace(\Sigma)/\sim,
\end{equation}
where $\mu_1 \sim \mu_2$ if and only if $F^{\mu_1} \circ \sigma \circ (F^{\mu_1})^{-1} = F^{\mu_2} \circ \sigma \circ (F^{\mu_2})^{-1}$ for all $\sigma \in \Sigma$ (or equivalently, $F^{\mu_1} \big|_{\Lambda} = F^{\mu_2}\big|_{\Lambda}$). At $\mu=0$ one recovers the group $\Sigma$ which corresponds to the base point of $\schottky_{g}$.

The above alternative definition of Schottky space $\schottky_{g}$ gives us the opportunity to define also the generalized Schottky space $\schottky_{g,n}(\boldsymbol{m})$ for Riemann orbisurfaces with signature $(g;m_1,\dots,m_{n_e};n_p)$. Let us consider the configuration spaces $\confspace{n}{\Omega^{\mu} \slash \Sigma^{\mu}}= \confspace{n}{\SchottkyFund^{\mu}}$ with $\Sigma^{\mu} = F^{\mu} \circ \Sigma \circ (F^{\mu})^{-1},\hspace{1mm} \Omega^{\mu} = F^{\mu}(\Omega)$ and 
the deformation space of a marked Schottky group $(\Sigma;L_1,\dots,L_g)$ together with a point $(w_1,\dots,w_{n_e},w_{n_e+1},\dots,w_{n}) \in \confspace{n}{\SchottkyFund}$,
\begin{multline}\label{Dsigmamu}
\deformspace(\Sigma;L_1,\dots,L_g;w_1,\dots,w_{n_e};w_{n_e+1},\dots,w_{n})\\
=\left\{(\mu ; w_1^{\mu},\dots,w_{n_e}^{\mu};w_{n_e+1}^{\mu},\dots,w_{n}^{\mu}) \in \diffspace^{-1,1}(\Omega,\Sigma) \times \confspace{n}{\SchottkyFund^{\mu}} \, \Big| \, \norm{\mu}_{\infty} < 1 \right\},
\end{multline}
where $w^{\mu}_i = F^{\mu}(w_i)$. Just as in the case of $\schottky_{g}$, each element $\mu \in \deformspace(\Sigma;L_1,\dots,L_g;w_1,\dots,w_{n})$ gives a faithful representation $\varrho_{\mu}$ of  $(\Sigma;L_1,\dots,L_g;w_1,\dots,w_{n})$ in $\PSLC \times \confspace{n}{\SchottkyFund^{\mu}}$ according to the formula $L_i \mapsto F^{\mu} \circ L_i \circ (F^{\mu})^{-1}$ for all marked generators $L_i \in \Sigma$, $i=1,\dots,g$ and $w_j \mapsto F^{\mu}(w_j)$ for $j=1,\dots,n$. Two representations are then equivalent, $\varrho_{\mu_1} \sim \varrho_{\mu_2}$, if  $F^{\mu_1} \circ L_i \circ (F^{\mu_1})^{-1} = F^{\mu_2} \circ L_i \circ (F^{\mu_2})^{-1}$ for all $i=1,\dots,g$ as well as $w_j^{\mu_1} = w_j^{\mu_2}$ for all $j=1,\dots,n$ and the Schottky space $\schottky_{g,n}(\boldsymbol{m})$ is defined to be the set of equivalence classes of representations $\left[\varrho_{\mu}\right]$
\begin{equation}\label{genschottkyspcae}
\schottky_{g,n}(\boldsymbol{m}) \cong \deformspace(\Sigma;L_1,\dots,L_g;w_1,\dots,w_{n_e};w_{n_e+1},\dots,w_{n})/\sim.
\end{equation}
Let us remind that the Schottky uniformization of an orbisurface $O$ is connected with the Fuchsian uniformization of it by the commutative diagram \eqref{globalcoords}. Accordingly, each marked Fuchsian group $\Gamma$ with signature $(g;m_1,\dots,m_{n_e},n_p)$  corresponds to the unique marked normalized Schottky group $\Sigma \simeq \Gamma/N$ with the domain of discontinuity $\singrigon$ such that $\UHP\slash\Gamma \cong  \singrigon\slash \Sigma $ and this correspondence determines the following map
\begin{equation}
\pi: \teich_{g,n}(\boldsymbol{m}) \to \schottky_{g,n}(\boldsymbol{m})
\end{equation}
by putting $w_i = J(z_i)$ for $i=1,\dots,n$.\footnote{This map has the same role as the covering map $\Psi$ in lemma~\ref{Lemma:covering} and is a complex covering map.} We can use this map to understand the tangent and cotangent space to the Schottky space $\schottky_{g,n}(\boldsymbol{m})$: Elements of $\Hilbert^{2,0}(\UHP,\Gamma)$ descend to meromorphic quadratic differentials for $\Sigma$ --- i.e. automorphic forms of weight 4 for $\Sigma$ with simple poles at singular points of $\singrigon$. The space of meromorphic quadratic differentials for $\Sigma$ will be denoted by $\Hilbert^{2,0}(\singrigon,\Sigma)$ and each $Q \in \Hilbert^{2,0}(\singrigon,\Sigma)$ has the form
\begin{equation}
	Q(w) = (q \circ J^{-1})(w) \left(J^{-1}(w)'\right)^2, \qquad q \in \Hilbert^{2,0}(\UHP,\Gamma). 
\end{equation}
The vector space $\Hilbert^{2,0}(\singrigon,\Sigma)$ coincides with the holomorphic cotangent space $T_{\pi \circ \Phi(0)}^{\ast}\schottky_{g,n}(\boldsymbol{m})$ to $\schottky_{g,n}(\boldsymbol{m})$ at the origin. This implies that the holomorphic tangent space $T_{\pi \circ \Phi(0)}\schottky_{g,n}(\boldsymbol{m})$ is identified with the complex vector space $\Hilbert^{-1,1}(\singrigon,\Sigma)$ of harmonic Beltrami differentials.\footnote{Harmonic with respect to the hyperbolic metric on $\singrigon$.} Each $M \in \Hilbert^{-1,1}(\singrigon,\Sigma)$ has the form
\begin{equation}
M(w) = e^{-\varphi(w)} \overline{Q(w)}, \qquad Q \in \Hilbert^{2,0}(\singrigon,\Sigma).
\end{equation}
From Eq.\eqref{Dsigmamu}, which  implies that there exists a fibration $\jmath: \schottky_{g,n}(\boldsymbol{m}) \to \schottky_g$ whose fibers over the points $\pi\circ \Phi(\mu) \in \schottky_{g}$ are $\confspace{n}{\Omega^{\mu} \slash \Sigma^{\mu}}$, it follows that $T_{\pi\circ\Phi(0)}^{\ast}\schottky_{g,n}(\boldsymbol{m})$ has a subspace $\jmath^{\ast}(T_{\pi\circ\Phi(0)}^{\ast}\schottky_{g}) \cong \Hilbert^{2,0}(\Omega,\Sigma)$. The standard basis in this subspace of $\Hilbert^{2,0}(\Omega,\Sigma)$ is given by the following holomorphic automorphic forms of weight four for Schottky group, \begin{equation}
P_1(w), \dots,P_{3g-3}(w)\in \Hilbert^{2,0}(\Omega,\Sigma).
\end{equation}
Moreover, the $P_i$s actually coincide with the following cotangent vectors 
\begin{equation}\label{ls}
\dd{\lambda_1},\dots,\dd{\lambda_g},\dd{a_3}, \dots, \dd{a_g}, \dd{b_2}, \dots, \dd{b_g} \in  \jmath^{\ast}(T_{\pi\circ\Phi(0)}^{\ast}\schottky_{g}).
\end{equation}
The subspace that is isomorphic to $T_{\pi\circ\Phi(0)}^{\ast}\confspace{n}{\SchottkyFund}$ corresponds to the complement of 
$\jmath^{\ast}(T_{\pi\circ\Phi(0)}^{\ast}$
$\schottky_{g})$ in $T_{\pi\circ\Phi(0)}^{\ast}\schottky_{g,n}(\boldsymbol{m})$. This subspace, is, in fact, the cotangent space to the configuration space at $(w_1,\dots,w_{n})$. In other words, we have
\begin{equation}
T_{\pi\circ\Phi(0)}^{\ast}\schottky_{g,n}(\boldsymbol{m}) \cong \jmath^{\ast}(T_{\pi\circ\Phi(0)}^{\ast}\schottky_{g}) \oplus T_{\pi\circ\Phi(0)}^{\ast}\confspace{n}{\SchottkyFund}.
\end{equation}
It follows from Eq.\eqref{Fdot} that a standard basis for $T_{\pi\circ\Phi(0)}^{\ast}\confspace{n}{\SchottkyFund}$ is given by the following meromorphic automorphic forms of weight four,
\begin{equation}\label{ls2}
P_{3g-3+i}(w) = -\frac{1}{\pi} \sum_{\sigma \in \Sigma} R(\sigma w, w_i) \sigma'(w)^2, \qquad w \in \Omega^{\text{reg}},
\end{equation}
which represent $\dd{w_i}$ for $i=1,\dots,n$. According to the following pairing, which actually is an analog of pairing \eqref{pairing}
\begin{equation}\label{Schottkypairing}
\left( Q, M \right) = \iint_{\SchottkyFund} Q(w) M(w) \dd[2]{w},
\end{equation} 
we can obtain the dual basis for $P_1(w), \dots P_{3g-3+n}(w)$ --- i.e. the basis $M_1(w), \dots M_{3g-3+n}(w)$ in $\Hilbert^{-1,1}(\singrigon,\Sigma)$ which  coincide with the tangent vectors $\pdv{w_1}, \dots, \pdv{w_{n}} \in T_{\pi\circ\Phi(0)}\schottky_{g,n}$. Similarly, the corresponding bases in the tangent and cotangent spaces to $\schottky_{g,n}$ at an arbitrary point can also be defined. This implies that $\text{Sch}\left(J^{-1};w\right) = \partial_w^2 \varphi(w) - \frac{1}{2} \left(\partial_w \varphi(w)\right)^2$ can be decomposed as following \footnote{See the  asymptotic behavior of $\varphi(w)$ and its derivatives as $w \to w_i$ in lemma~\ref{lemma:asymptotics}.}
\begin{equation}\label{SchottkyEnergyMomentum}
\text{Sch}\left(J^{-1};w\right)= \sum_{i=1}^{n} h_i \mathscr{E}_i(w) - \pi \sum_{i=1}^{3g-3+n} c_i P_i(w),
\end{equation}
where
\begin{equation}\label{scrEdef}
\mathscr{E}_i(w)  = \frac{1}{2} \sum_{\sigma \in \Sigma} \left(\frac{1}{(\sigma w - w_i)^2} - \frac{1}{\sigma w (\sigma w -1)}\right) \sigma'(w)^2, \quad \text{for} \quad i=1,\dots,n,
\end{equation}
are meromorphic automorphic forms of weight four for Schottky group with the second order poles at $\Sigma \cdot w_i$ and $c_1,\dots,c_{3g-3+n}$ are  accessory parameters. Moreover, for the variations of hyperbolic metric we have the same formulas introduced in subsection~\ref{subsub:variationalformulas} and \ref{M0n}. Furthermore, the formula \eqref{widerivative} is also valid in this case. Finally, on $\teich_{g,n}(\boldsymbol{m})$, each cuspidal $\langle \cdot, \cdot \rangle_{\text{TZ},i}^{\text{cusp}}$ and elliptic $\langle \cdot, \cdot \rangle_{\text{TZ},i}^{\text{ell}}$ metric  remains invariant under the automorphism group of the covering $\pi: \teich_{g,n}(\boldsymbol{m}) \to \schottky_{g,n}(\boldsymbol{m})$. Accordingly, each metric descends to a K\"{a}hler metric on $\schottky_{g,n}(\boldsymbol{m})$. All these also imply that in analogy with the commutative diagram \eqref{diagramFmu}, we have
\begin{equation}\label{diagramFmuschottky}
\begin{CD}
\UHP	@> f^{\mu}>> \UHP\\
@V J VV		@VV J_{\mu} V\\
\singrigon 		@>F^{\mu}>> \singrigon\hspace{.4mm}^{\mu}
\end{CD}
\end{equation}
where
\begin{equation}
\partial_{\bar{w}} F^{\mu} = M(w)\hspace{1mm}\partial_w F^{\mu} \qquad \text{for} \quad w \in \singrigon,
\end{equation}
and $F^{\varepsilon\mu}$ is complex-analytic in $\varepsilon$ and $\dot{F}^{\mu}$ is given by similar equation with Eq.\eqref{Fdot} where $w\in \singrigon$.

Before closing this section let us also mention that the mapping $J$ in diagram \eqref{globalcoords} has the following expansions near cusps and branch points of $\Gamma$ (see Appendix~\ref{Apx:Asymptotics})
\begin{equation}\label{Jexpansion2}
J(z) = \left\{
\begin{split}
& w_{i} + \sum_{k=1}^{\infty} J^{(i)}_k  \left(\frac{z-z_{i}}{z-\bar{z}_{i}}\right)^{k m_i} \hspace{3.35cm} (i=1,\dots,n_e), \hspace{.2cm} z \to z_{i}, \\ \\
&w_{i} + \sum_{k=1}^{\infty} J^{(i)}_k \exp(-\frac{2 \pi\hspace{1mm}\sqrt{-1}k}{|\delta_{i}|(z-z_{i})}) \hspace{1cm} (i=n_e+1,\dots,n-1),\hspace{.2cm} z \to z_{i},
\end{split}
\right.
\end{equation}
where $w_i=J(z_i)$ for $i=1,\dots,n-1$.\footnote{Note that since $\Sigma$ is normalized, $\infty \notin \Omega$.}. If $e^{\varphi(w)} |\dd{w}|^2$ denotes the push-forward of the hyperbolic metric on $\UHP$ by the map $J$, then the density of hyperbolic metric on $\singrigon$ (i.e. $\rho(w) = e^{\varphi(w)}$) is once again given by Eq.\eqref{Liouvillefield}, where $\varphi(w)$ is smooth on $\Omega^{\text{reg}}$. The function $\varphi(w)$ satisfies\footnote{This equation follows from the invariance of hyperbolic metric on $\Omega^{\text{reg}}$ under the action of $\Sigma$ --- i.e. $e^{\varphi(w)}\dd{w} \dd{\bar{w}} = e^{\varphi(\sigma w)}\dd{\sigma(w)} \dd{\overline{\sigma(w)}}$ for all $\sigma \in \Sigma$.\label{footnote:sigmaonphiaction}}
\begin{equation}\label{Lonphi}
\varphi(\sigma w) = \varphi(w) - \log|\sigma'(w)|^2 \quad \text{for} \quad w \in \Omega^{\text{reg}}, \, \forall \sigma \in \Sigma.
\end{equation}
According to Lemma \ref{lemma:asymptotics}, it also has the following asymptotic form
\begin{equation}\label{asymptotics}
\varphi(w) = \left \{
\begin{split}
&-2 (1-\frac{1}{m_i}) \log|w-w_{i}| + \log \frac{4|J^{(i)}_1|^{-\frac{2}{m_i}}}{m_i^2} + \Sorder{1} \hspace{1.3cm} w \to w_{i},\\ 
&-2 \log |w-w_{j}| - 2 \log\left|\log\left|\frac{w-w_{j}}{J^{(i)}_{1}}\right|\right| + \Sorder{1}\hspace{2cm}  w \to w_{j},\\
& -2 \log|w| - 2 \log\log \left|\frac{w}{J^{(n)}_{-1}}\right| + \order{|w|^{-1}}, \hspace{2.8cm} w \to \infty ,
\end{split}
\right.
\end{equation}
for $i=1,\dots,n_e$ and $j=n_e+1,\dots,n-1$.
\begin{remark}\label{hasymptotics}
	It follows from the above asymptotics that (see statement~3 of lemma~\ref{lemma:asymptotics} for more details)
	\begin{equation}
	\left\{
	\begin{split}
	& \log \Lponetial_{i} = -2 \log m_i + 2 \log 2 - \lim_{w \to w_{i}}\left(\varphi(w) + \left(1-\frac{1}{m_i}\right) \log|w-w_{i}|^2\right), \hspace{.2cm} i=1,\dots,n_e,\\ \\
	&\log \Lponetial_{i} = \lim_{w \to w_{i}} \left(\log|w-w_{i}|^2 - \frac{2 e^{-\frac{\varphi(w)}{2}}}{|w-w_{i}|}\right),\hspace{2cm}  i=n_e+1,\dots,n-1,\\
	& \log \Lponetial_{n} = \lim_{w \to \infty} \left(\log|w|^2 - \frac{2 e^{-\frac{\varphi(w)}{2}}}{|w|}\right).
	\end{split}
	\right.
	\end{equation}
	with $\Lponetial_i = \left| J^{(i)}_{1} \right|^{\frac{2}{m_i}}$ for $i=1,\dots,n_e$, and $\Lponetial_{i} =  \left| J^{(i)}_{1} \right|^{2}$ for $i=n_e+1,\dots,n$.
\end{remark}
Finally, consider the points $w_1,\dots,L_k w_i, \dots, w_{n}$, corresponding to the branch points and cusps in $z_1,\dots,\beta_k z_i, \dots, z_{n}$.\footnote{One can comprehend this by recognizing that $J \circ \beta_k = L_k \circ J$.} Near the point $\beta_k z_i$, the first coefficient in the expansion \eqref{Jexpansion2} of $J(z)$ is given by $L'_{k}(w_i) J_{1}^{(i)}$. Accordingly, the positive functions $\Lponetial_i = \left| J^{(i)}_{1} \right|^{\frac{2}{m_i}}$ for $i=1,\dots,n_e$ and $\Lponetial_{i} =  \left| J^{(i)}_{1} \right|^{2}$ for $i=n_e+1,\dots,n$ are respectively replaced  with $\Lponetial_i  \left|L'_k(w_i)\right|^{\frac{2}{m_i}}$ and $\Lponetial_{i} \left|L'_k(w_{i})\right|^{2}$, when sending $w_i$ to $L_k w_i$. Moreover, let us define $\linebundle_i$ as the $i$-th relative cotangent line bundle on $\schottky_{g,n}(\boldsymbol{m})$, situated along the fibers of the projection $p_i: \schottky_{g,n}(\boldsymbol{m}) \to \schottky_{g,n-1}(\boldsymbol{m})$.\footnote{ This projection forgets the $w_i$ for $i=1,\dots,n$.} Given this understanding,
we can establish the following assertion.
\begin{lemma}\label{hihermitianmetric}
	Hermitian metrics in the holomorphic line bundles $\linebundle_i$ for $i=1,\dots,n_e$ are determined by the quantities $\Lponetial_{i}^{m_i}$.
\end{lemma}
\begin{proof}
	To prove this lemma for the branch points, we use the transformation of $\Lponetial_{i}$'s under the action of the generators $L_k$ of the Schottky group. As explained above, sending $w_1,\dots,w_i,\dots,w_{n_e}$ to $w_1,\dots,L_k w_i, \dots, w_{n_e}$ will result in $\Lponetial_{i}$ to be replaced by $\Lponetial_{i}\left|L'_k(w_{i})\right|^{\frac{2}{m_{i}}}$. Thus, we have:
	\begin{equation}
	\Delta\log \Lponetial_{i}=\frac{1}{m_{i}}\log \left|L'_k(w_{i})\right|^2.
	\end{equation} 
	This means that $\Lponetial_{i}^{m_i}$ is a Hermitian metric in the line bundle $\linebundle_i$. For the case of cusp, a similar method of proof is applicable and we conclude that $\Lponetial_{i}$s will determine Hermitian metric in the line bundle $\linebundle_i$.
\end{proof}
\section{Classical Liouville Action}\label{sec:action}
In this section, we study the classical Liouville action for hyperbolic Riemann orbisurfaces with the genus $g=0$ and $g > 1$, separately. Before proceeding, we should mention that all the proofs in sections \ref{sec:action} and \ref{sec:Potentials} are previously provided for the case of punctured Riemann surfaces in \cite{Zograf1988ONLE} and \cite{park2015potentials}. When it comes to calculations involving punctures, we can only direct the reader to those articles. However, for the reader's convenience and to facilitate a clearer understanding of distinctions in the presence of conical singularities, we find it appropriate to present all proofs side by side.
\subsection{Riemann Orbisurfaces of Genus 0}\label{subsec:liouvilleactionnogenus}
Let $O$ be a marked Riemann orbisurface with signature $(0;m_1,\dots,m_{n_e},n_p)$ and let $\boldsymbol{m}$ denotes the vector of orders  $(m_1,\dots,m_{n})$. The regularized action functional for the Liouville equation~\eqref{Liouvilleequation} in presence of conical singularities with conical angles $2\pi/m_1, \dots , 2\pi/m_{n_e}$ at $w_{1},\dots,w_{n_e}$ together with punctures at $w_{n_e+1},\dots,w_{n-2}=0, w_{n-1}=1,w_{n}=\infty$ is defined as follows (see \cite{Zograf1988ONLE, Takhtajan:1995fd, Takhtajan:2001uj}):
\begin{multline}\label{NoGenusAction}
S_{\boldsymbol{m}}[\varphi] = \\
\lim_{\epsilon \to 0^{+}} \left(\iint_{O_{\epsilon}}(|\partial_w \varphi|^2 + e^{2\varphi}) \dd[2]{w}  + \frac{\sqrt{-1}}{2} \sum_{i=1}^{n_e} \left(1-\frac{1}{m_i}\right) \oint_{C_{i}^{\epsilon}} \varphi \left(\frac{\dd{\bar{w}}}{\bar{w}-\bar{w}_{i}} - \frac{\dd{w}}{w - w_{i}}\right)\right. \\ 
\left.  - 2\pi \sum_{i=1}^{n_e} \left(1-\frac{1}{m_i}\right)^2 \log\epsilon + 2 \pi n_p \log\epsilon + 4 \pi (n_p-2) \log|\log\epsilon|\right),
\end{multline}
where\footnote{Note that  $\lim_{\epsilon\rightarrow 0}O_{\epsilon}= O^{\text{reg}}= X_{O}^{\text{reg}}$.}
\begin{equation*}
O_\epsilon = \cmpx \Big\backslash \bigcup_{i=1}^{n-1}\left\{w \, \Big| \, |w-w_i|<\epsilon\right\} \cup \left\{w \, \Big| \, |w|>\epsilon^{-1}\right\},
\end{equation*}
and the circles 
\begin{equation*}	
C_{i}^{\epsilon} = \left\{w \, \Big| \, |w-w_{i}| = \epsilon\right\},
\end{equation*}
are oriented as a component of the boundary $\partial  O_{\epsilon}$. 
\begin{remark}
	When $n_p = 0$, the appropriate classical Liouville action is given by 
	\begin{multline}\label{CompactNoGenusAction}
	S_{\boldsymbol{m}}[\varphi] = \\
	\lim_{\epsilon \to 0^{+}} \Bigg(\iint_{O_{\epsilon}}(|\partial_w \varphi|^2 + e^{2\varphi}) \dd[2]{w} + \frac{\sqrt{-1}}{2} \sum_{i=1}^{n-1} \left(1-\frac{1}{m_i}\right) \oint_{C_{i}^{\epsilon}} \varphi \left(\frac{\dd{\bar{w}}}{\bar{w}-\bar{w}_{i}} - \frac{\dd{w}}{w - w_{i}}\right) \\ 
	+ \frac{\sqrt{-1}}{2} \left(1+\frac{1}{m_n}\right) \oint_{C_{n}^{\epsilon}} \varphi \left(\frac{\dd{\bar{w}}}{\bar{w}} - \frac{\dd{w}}{w}\right) - 2\pi \sum_{i=1}^{n-1} \left(1-\frac{1}{m_i}\right)^2 \log\epsilon - 2\pi \left(1+ \frac{1}{m_n}\right)^2 \log\epsilon \Bigg).
	\end{multline}
\end{remark}
\begin{remark}
	It is worth providing some explanation regarding the action \eqref{NoGenusAction} and why the terms associated with conical points appear different from those associated with the punctures. The contour integral in the first line of \eqref{NoGenusAction} is added to ensure a well-defined variational principle around the conical singularities.\footnote{Alternatively, the line integrals are necessary in order to ensure the proper asymptotic behavior \eqref{asymptotics}. See the explanation following \cite[Eq.~(8)]{Takhtajan:1995fd} for more details.} The variation of this integral cancels the boundary term arising from the variation of the bulk term. To ensure a well-defined variational principle around the punctures, additional contour integrals must be considered, which are given by \cite[\textsection2]{Takhtajan:1995fd}
	\begin{multline}\label{BTP}
	\frac{\sqrt{-1}}{2} \sum_{j=n_e+1}^{n-1}  \oint_{C_{j}^{\epsilon}} \varphi \left(\frac{\dd{\bar{w}}}{\bar{w}-\bar{w}_{j}} - \text{c.c}\right)+\frac{\sqrt{-1}}{2} \sum_{j=n_e+1}^{n-1}  \oint_{C_{j}^{\epsilon}} \varphi \left(\frac{\dd{\bar{w}}}{(\bar{w}-\bar{w}_{j})\log|w-w_j|} -\text{c.c.}\right)\\
	-\frac{\sqrt{-1}}{2} \oint_{C_{n}^{\epsilon}} \varphi \left(\frac{\dd{\bar{w}}}{\bar{w}} - \text{c.c}\right) - \frac{\sqrt{-1}}{2} \oint_{C_{n}^{\epsilon}} \varphi \left(\frac{\dd{\bar{w}}}{\bar{w}\log|w|} -\text{c.c.}\right).
	\end{multline}
	More concretely, the classical Liouville action \eqref{NoGenusAction} in the presence of these additional contour integrals takes the form \cite[Eq.~(8)]{Takhtajan:1995fd}
	\begin{multline}\label{FullNoGenusAction}
	S_{\boldsymbol{m}}[\varphi] = \\
	\lim_{\epsilon \to 0^{+}} \left(\iint_{O_{\epsilon}}(|\partial_w \varphi|^2 + e^{2\varphi}) \dd[2]{w}  + \frac{\sqrt{-1}}{2} \sum_{i=1}^{n-1} \left(1-\frac{1}{m_i}\right) \oint_{C_{i}^{\epsilon}} \varphi \left(\frac{\dd{\bar{w}}}{\bar{w}-\bar{w}_{i}} - \frac{\dd{w}}{w - w_{i}}\right)\right. \\ 
	+\frac{\sqrt{-1}}{2} \sum_{j=n_e+1}^{n-1}  \oint_{C_{j}^{\epsilon}} \varphi \left(\frac{\dd{\bar{w}}}{(\bar{w}-\bar{w}_{j})\log|w-w_j|} - \frac{\dd{w}}{(w-w_j)\log|w-w_j|} \right) -\frac{\sqrt{-1}}{2} \oint_{C_{n}^{\epsilon}} \varphi \left(\frac{\dd{\bar{w}}}{\bar{w}} - \frac{\dd{w}}{w}\right) \\
	\left. -\frac{\sqrt{-1}}{2} \oint_{C_{n}^{\epsilon}} \varphi \left(\frac{\dd{\bar{w}}}{\bar{w}\log|w|} - \frac{\dd{w}}{w\log|w|}\right) -2\pi \sum_{i=1}^{n} \left(1-\frac{1}{m_i}\right)^2 \log\epsilon 
	\right),
	\end{multline}
	where $m_i = \infty$ for $i=n_e+1,\dots,n$. By substituting the asymptotic form \eqref{asymptotics} of the Liouville field $\varphi$ near punctures in the contour integrals of \eqref{FullNoGenusAction}, up to $ \Sorder{1}$ terms, one gets 
	\begin{equation}\label{eq:contourintpunc}
	\begin{split}
	& \frac{\sqrt{-1}}{2} \sum_{j=n_e+1}^{n-1}  \oint_{C_{j}^{\epsilon}} \varphi \left(\frac{\dd{\bar{w}}}{\bar{w}-\bar{w}_{j}} - \text{c.c}\right)+\frac{\sqrt{-1}}{2} \sum_{j=n_e+1}^{n-1}  \oint_{C_{j}^{\epsilon}} \varphi \left(\frac{\dd{\bar{w}}}{(\bar{w}-\bar{w}_{j})\log|w-w_j|} -\text{c.c.}\right)\\
	& -\frac{\sqrt{-1}}{2} \oint_{C_{n}^{\epsilon}} \varphi \left(\frac{\dd{\bar{w}}}{\bar{w}} - \text{c.c}\right) - \frac{\sqrt{-1}}{2} \oint_{C_{n}^{\epsilon}} \varphi \left(\frac{\dd{\bar{w}}}{\bar{w}\log|w|} -\text{c.c.}\right) = 4\pi n_p \log\epsilon +4 \pi (n_p-2) \log|\log\epsilon|.
	\end{split}
	\end{equation}
	Since the contour integral \eqref{BTP} evaluated on-shell is merely a divergent term, we can add it to the counterterm $-2\pi n_p \log \epsilon$ in \eqref{FullNoGenusAction} to get the Liouville action \eqref{NoGenusAction}. The resulting regulating terms will have the opposite sign (i.e., plus sign instead of minus sign) and correctly cancel the divergence coming from the bulk term $\lim_{\epsilon \to 0^{+}} \iint_{O_{\epsilon}} |\partial_w \varphi|^2 \dd[2]{w}$.
	
	Performing a similar calculation for the case of conical points, i.e., substituting the asymptotic form \eqref{asymptotics} of the Liouville field $\varphi$ near the conical points in the contour integral of \eqref{NoGenusAction}, one observes that the result is given by 
	\begin{equation}\label{eq:contourintcone}
	\frac{\sqrt{-1}}{2} \sum_{i=1}^{n-1} \left(1-\frac{1}{m_i}\right) \oint_{C_{i}^{\epsilon}} \varphi \left(\frac{\dd{\bar{w}}}{\bar{w}-\bar{w}_{i}} - \frac{\dd{w}}{w - w_{i}}\right) = 4 \pi \sum_{i=1}^{n_e} \big[(1-\tfrac{1}{m_i})^2 \log\epsilon + \tfrac{1}{2} (1-\tfrac{1}{m_i}) \log \Lponetial_{i}\big] + \Sorder{1}.
	\end{equation}
	Notice that, apart from the divergent part, the above expression contains a finite part (i.e., $2\pi\sum_{i} (1-\tfrac{1}{m_i}) \log \Lponetial_{i}$) which behaves non-trivially under quasi-conformal transformations and therefore cannot be ignored (see lemma~\ref{lemma:widerivativeloghi} for more details).\footnote{For both punctures and conical points, there are finite constant terms in the calculation of the contour integrals which can be safely ignored. These terms are not written in equations \eqref{eq:contourintpunc} and \eqref{eq:contourintcone}.} We have decided to keep the line integrals around conical points in their integral form since this form is more familiar in the literature on Liouville CFT and can also be used when the conical points are of a more general type (see, e.g. \cite{Takhtajan:2001uj}). Therefore, one arrives at the classical Liouville action \eqref{NoGenusAction}.
\end{remark}

\begin{theorem*}[Takhtajan and Zograf]
	For any fixed vector of orders $\boldsymbol{m}=(m_1,\dots,m_{n})$ such that $\sum_{j=1}^{n_e} (1-\tfrac{1}{m_j})+n_p > 2$, the function $S_{\boldsymbol{m}}: \moduli_{0,n} \to \mathbb{R}$ is differentiable and
	\begin{equation}\label{theorem1}
	c_i = - \frac{1}{2 \pi} \pdv{S_{\boldsymbol{m}}}{w_i} \qquad \text{for all} \qquad i=1,\dots,n-3,
	\end{equation}
	where $c_i$s are the accessory parameters defined by \eqref{accessory}.\footnote{See Theorem 1 in \cite{Zograf1988ONLE} for genus $g=0$ punctured Riemann surfaces.}
\end{theorem*}
\begin{proof}
	Let 
	\begin{equation}\label{def}
	\tilde{S}_{\boldsymbol{m}}^{\epsilon}(w_1,\dots,w_{n-3}) = \tilde{S}_{\boldsymbol{m}}^{(B)\epsilon}(w_1,\dots,w_{n-3}) +\tilde{S}^{\text{(ct)}\epsilon}- 2 \pi \chi(X), 
	\end{equation}
	with\footnote{In view of the Gauss-Bonnet formula for Riemann orbisurfaces \cite{Troyanov1991PrescribingCO,Li_orbifold_2018}
		\begin{equation}\nonumber\\
		\frac{\sqrt{-1}}{2}\iint_{O} e^{\varphi} \dd{w} \wedge \dd{\bar{w}} = 2 \pi \left(\sum_{j=1}^{n_e} \left(1- \frac{1}{m_j}\right)+n_p-2\right) = - 2 \pi \chi(O).
		\end{equation}}
	\begin{multline}\label{modifiedS}
	\tilde{S}_{\boldsymbol{m}}^{(B)\epsilon}(w_1,\dots,w_{n-3}) =
	\iint_{O_{\epsilon}} |\partial_{w} \varphi|^2 \dd[2]{w} + \frac{\sqrt{-1}}{2} \sum_{j=1}^{n_e} \left(1-\frac{1}{m_j}\right) \oint_{C_{i}^{\epsilon}} \varphi \left(\frac{\dd{\bar{w}}}{\bar{w}-\bar{w}_{j}} - \frac{\dd{w}}{w - w_{j}}\right)\\  
	\tilde{S}^{(\text{ct})\epsilon}=- 2\pi \sum_{j=1}^{n_e} \left(1-\frac{1}{m_j}\right)^2 \log\epsilon  + 2 \pi (n-n_e) \log\epsilon + 4 \pi (n-n_e-2) \log|\log\epsilon|.
	\end{multline}
	For any $\epsilon > 0$, the function $\tilde{S}_{\boldsymbol{m}}^{\epsilon}$ is continuously differentiable on $\moduli_{0,n}$. To prove this theorem, it suffices to show that $\Lie_{\mu_i}\tilde{S}_{\boldsymbol{m}}^{\epsilon}$ converges uniformly to $-2\pi c_i$ as $\epsilon \to 0$ in a neighborhood of any point of the moduli space $\moduli_{0,n}$. More explicitly, one needs to show that  
	\begin{equation}
	\Lie_{\mu_i}\tilde{S}_{\boldsymbol{m}}^{\epsilon} = \pdv{\tilde{S}_{\boldsymbol{m}}^{\epsilon}}{w_i} \equidef \pdv{\varepsilon}\Bigg|_{\varepsilon=0} \tilde{S}_{\boldsymbol{m}}^{\epsilon}(w_1^{\varepsilon \mu_i}, \dots, w_{n-3}^{\varepsilon \mu_i}) = -2 \pi c_i \quad \text{for} \quad i=1,\dots,n-3,
	\end{equation}
	pointwise on the moduli space $\moduli_{0,n}$.\footnote{We remind the reader that the basis $\{\mu_i\}$ for $ T_{\Phi(0)} \teich(\Gamma)$ has been defined in Eq.\eqref{Beltramibasis}.} Firstly, we can write
	\begin{equation}\label{ext01}
	\pdv{\tilde{S}_{\boldsymbol{m}}^{(B)\epsilon}}{w_i} = \frac{\sqrt{-1}}{2} \pdv{\varepsilon} \Bigg|_{\varepsilon=0} I_{\epsilon}(\varepsilon),
	\end{equation}
	where
	\begin{equation*}
	\begin{split}
	&I_{\epsilon}(\varepsilon) = \left(I_{\epsilon}[\varphi] \circ \Psi \circ \Phi \right)(\varepsilon \mu_i), \\ &\varphi^{\varepsilon \mu_i}(w) = \varphi\Big(w;(\Psi \circ \Phi)(\varepsilon \mu_i)\Big) = \varphi(w;\underbrace{F^{\varepsilon \mu_i}(w_1)}_{w_1^{\varepsilon \mu_i}},\dots, \underbrace{F^{\varepsilon \mu_i}(w_{n-3})}_{w_{n-3}^{\varepsilon \mu_i}}),
	\end{split}
	\end{equation*}
	and
	\begin{equation*}
	I_{\epsilon}[\varphi] = \iint_{O_{\epsilon}} |\partial_w \varphi|^2  \dd{w} \wedge \dd{\bar{w}} + \sum_{j=1}^{n_e} \left(1-\frac{1}{m_j}\right) \oint_{C_{i}^{\epsilon}} \varphi \left(\frac{\dd{\bar{w}}}{\bar{w}-\bar{w}_{j}} - \frac{\dd{w}}{w - w_{j}}\right).
	\end{equation*}
	Accordingly, one gets 
	\begin{equation}
	\begin{split}
	I_{\epsilon}(\varepsilon) =& \iint_{F^{\varepsilon\mu_i}(O_{\epsilon})} |\partial_w \varphi^{\varepsilon \mu_i}(w)|^2  \dd{w} \wedge \dd{\bar{w}}\\
	&+\sum_{j=1}^{n_e} \left(1-\frac{1}{m_j}\right) \oint_{F^{\varepsilon\mu_i}(C_{j}^{\epsilon})} \varphi^{\varepsilon \mu_i}(w) \left(\frac{\dd{\bar{w}}}{\bar{w}-\overline{w_{j}^{\varepsilon \mu_i}}} - \frac{\dd{w}}{w - w_{j}^{\varepsilon \mu_i}}\right),
	\end{split}
	\end{equation}
	with
	\begin{equation}
	\left\{
	\begin{split}
	& F^{\varepsilon\mu_i}(O_{\epsilon}) = \cmpx \Big\backslash \bigcup_{k=1}^{n-1}\left\{w \, \Big| \, |w-w_k^{\varepsilon\mu_i}|<\epsilon\right\} \cup \left\{w \, \Big| \, |w|>\epsilon^{-1}\right\},\\
	& F^{\varepsilon\mu_i}(C_{j}^{\epsilon}) = \left\{w \, \Big| \, |w-w_{j}^{\varepsilon\mu_i}| = \epsilon\right\}.
	\end{split}
	\right.
	\end{equation}
	The calculation of \eqref{ext01} can be done by almost verbatim repeating the corresponding computations in the proof of Theorem 1 in \cite{Zograf1988ONLE}. Accordingly, let us now use the \emph{change of variable formula} for differential forms, $\int_{F(X)} \omega = \int_{X} F^{\ast}(\omega)$, and the commutative diagram~\eqref{diagramFmu} to write
	\begin{equation*}
	\begin{split}
	I_{\epsilon}(\varepsilon) & = \iint_{O_{\epsilon}(\varepsilon)} (F^{\varepsilon\mu_i})^{\ast} \Big(|\partial_w \varphi^{\varepsilon \mu_i}|^2 \dd{w} \wedge \dd{\bar{w}}\Big)\\
	& + \sum_{j=1}^{n_e} \left(1-\frac{1}{m_j}\right) \oint_{C_{j}^{\epsilon}(\varepsilon)} (F^{\varepsilon\mu_i})^{\ast} \left(\varphi^{\varepsilon \mu_i} \left(\frac{\dd{\bar{w}}}{\bar{w}-\overline{w_{j}^{\varepsilon \mu_i}}} - \frac{\dd{w}}{w - w_{j}^{\varepsilon \mu_i}}\right)\right)\\
	& = \iint_{O_{\epsilon}(\varepsilon)} |\partial_w \varphi^{\varepsilon \mu_i} \circ F^{\varepsilon\mu_i}|^2  \dd{F^{\varepsilon\mu_i}(w)} \wedge \dd{\overline{F^{\varepsilon\mu_i}(w)}}\\
	& + \sum_{j=1}^{n_e} \left(1-\frac{1}{m_j}\right) \oint_{C_{j}^{\epsilon}(\varepsilon)}  (\varphi^{\varepsilon \mu_i} \circ F^{\varepsilon\mu_i}) \left(\frac{\dd{\overline{F^{\varepsilon\mu_i}(w)}}}{\overline{F^{\varepsilon\mu_i}(w)}-\overline{w_{j}^{\varepsilon \mu_i}}} - \frac{\dd{F^{\varepsilon\mu_i}(w)}}{F^{\varepsilon\mu_i}(w) - w_{j}^{\varepsilon \mu_i}}\right)\\
	& = \iint_{O_{\epsilon}(\varepsilon)} |\partial_w \varphi^{\varepsilon \mu_i} \circ F^{\varepsilon\mu_i}|^2 \, |\partial_w F^{\varepsilon\mu_i}|^2 (1 - |\varepsilon M_i|^2) \dd{w} \wedge \dd{\bar{w}}\\
	& + \sum_{j=1}^{n_e} \left(1-\frac{1}{m_j}\right) \oint_{C_{j}^{\epsilon}(\varepsilon)}  (\varphi^{\varepsilon \mu_i} \circ F^{\varepsilon\mu_i}) \left(\frac{\overline{\partial_w F^{\varepsilon\mu_i}}(\bar{\varepsilon} \overline{M_i} \dd{w} +\dd{\bar{w}})}{\overline{w^{\varepsilon\mu_i}}-\overline{w_{j}^{\varepsilon \mu_i}}} - \frac{\partial_w F^{\varepsilon\mu_i}(\dd{w}+ \varepsilon M_i \dd{\bar{w}})}{w^{\varepsilon\mu_i} - w_{j}^{\varepsilon \mu_i}}\right),
	\end{split}
	\end{equation*}
	where
	\begin{equation}
	\left\{
	\begin{split}
	& O_{\epsilon}(\varepsilon) = \cmpx \Big\backslash \bigcup_{k=1}^{n-1}\left\{w \, \Big| \, |w^{\varepsilon \mu_i}-w_i^{\varepsilon\mu_i}|<\epsilon\right\} \cup \left\{w \, \Big| \, |w^{\varepsilon \mu_i}|>\epsilon^{-1}\right\},\\
	&C_{j}^{\epsilon}(\varepsilon) = \left\{w \, \Big| \, |w^{\varepsilon \mu_i}-w_{j}^{\varepsilon\mu_i}| = \epsilon\right\}.
	\end{split}
	\right.
	\end{equation}
	In order to compute $\partial I_{\epsilon}(\varepsilon)/\partial\varepsilon |{\varepsilon=0}$, it is necessary to differentiate both the integrand and the  integration domains $O_{\epsilon}(\varepsilon)$ and $C_{j}^{\epsilon}(\varepsilon)$:
	\begin{equation}\label{pdvvarepsilonI}
	\begin{split}
	\pdv{\varepsilon}&\Bigg|_{\varepsilon=0} I_{\epsilon}(\varepsilon) = \iint_{O_{\epsilon}} \pdv{\varepsilon}\Bigg|_{\varepsilon=0} |\partial_w \varphi^{\varepsilon \mu_i} \circ F^{\varepsilon\mu_i}|^2 \, |\partial_w F^{\varepsilon\mu_i}|^2 (1 - |\varepsilon M_i|^2) \dd{w} \wedge \dd{\bar{w}}\\
	& + \pdv{\varepsilon}\Bigg|_{\varepsilon=0} \iint_{O_{\epsilon}(\varepsilon)} |\partial_w \varphi|^2  \dd{w} \wedge \dd{\bar{w}}\\
	& + \sum_{j=1}^{n_e} \left(1-\frac{1}{m_j}\right) \oint_{C_{j}^{\epsilon}} \pdv{\varepsilon}\Bigg|_{\varepsilon=0} (\varphi^{\varepsilon \mu_i} \circ F^{\varepsilon\mu_i}) \left(\frac{\overline{\partial_w F^{\varepsilon\mu_i}}(\bar{\varepsilon} \overline{M_i} \dd{w} +\dd{\bar{w}})}{\overline{w^{\varepsilon\mu_i}}-\overline{w_{j}^{\varepsilon \mu_i}}} - \frac{\partial_w F^{\varepsilon\mu_i}(\dd{w}+ \varepsilon M_i \dd{\bar{w}})}{w^{\varepsilon\mu_i} - w_{j}^{\varepsilon \mu_i}}\right)\\
	& + \sum_{j=1}^{n_e} \left(1-\frac{1}{m_j}\right) \pdv{\varepsilon}\Bigg|_{\varepsilon=0} \oint_{C_{j}^{\epsilon}(\varepsilon)} \varphi \left(\frac{\dd{\bar{w}}}{\bar{w}-\bar{w}_{j}} - \frac{\dd{w}}{w - w_{j}}\right).
	\end{split}
	\end{equation}
	The second and fourth terms in Eq.\eqref{pdvvarepsilonI} can be computed using the formula for differentiating a given $k$-form $\omega$ over a smooth family of variable domains $O(\varepsilon)$
	\begin{equation}\label{vardomaindiff}
	\pdv{\varepsilon}\Bigg|_{\varepsilon=0} \underbrace{\int\dotsm\int_{O_{\epsilon}(\varepsilon)}}_{k} \omega = \underbrace{\int\dotsm\int_{\textcolor{black}{\partial O_{\epsilon}}}}_{k-1} i_V(\omega), 
	\end{equation}
	where $i_V(\omega)$ denotes the interior product of the $k$-form $\omega$ with vector field $V$ which is the vector field along $\partial O_{\epsilon}$ corresponding to the family of curves $\partial O_{\epsilon}(\varepsilon)$. As a result, we get
	\begin{equation}\label{WiDerivative}
	\begin{split}
	\pdv{I_{\epsilon}}{w_i} & =  \iint_{O_{\epsilon}} \left[(\partial_{w_i}\partial_w \varphi + \partial_w^2 \varphi \dot{F}^i) \partial_{\bar{w}} \varphi + (\partial_{w_i} \partial_{\bar{w}} \varphi  + \partial_w \partial_{\bar{w}} \varphi  \dot{F}^i) \partial_w \varphi + |\partial_w\varphi|^2  \partial_w \dot{F}^i \right] \dd{w} \wedge \dd{\bar{w}}\\
	&\hspace{-.6cm} - \sum_{k=1}^{n}\int_{\partial O^{\epsilon}_k} |\partial_w \varphi|^2  \Big(\dot{F}^i(w) - \dot{F}^i(w_k)\Big) \dd{\bar{w}}\\
	&\hspace{-.6cm} + \sum_{j=1}^{n_e} \left(1-\frac{1}{m_j}\right) \oint_{C_{j}^{\epsilon}} \pdv{\varepsilon}\Bigg|_{\varepsilon=0} (\varphi^{\varepsilon \mu_i} \circ F^{\varepsilon\mu_i})\hspace{-1mm} \left(\frac{\overline{\partial_w F^{\varepsilon\mu_i}}(\bar{\varepsilon} \overline{M_i} \dd{w} +\dd{\bar{w}})}{\overline{w^{\varepsilon\mu_i}}-\overline{w_{j}^{\varepsilon \mu_i}}} - \frac{\partial_w F^{\varepsilon\mu_i}(\dd{w}+ \varepsilon M_i \dd{\bar{w}})}{w^{\varepsilon\mu_i} - w_{j}^{\varepsilon \mu_i}}\right),
	\end{split}
	\end{equation}
	where the last term in Eq.\eqref{pdvvarepsilonI} has vanished due to the fact that $\partial C_{j}^{\epsilon} = \emptyset$. By noting that 
	\begin{equation*}
	\begin{split}
	&\oint_{C_{j}^{\epsilon}} \varphi\left(\frac{d\bar{w}}{\bar{w}-\bar{w}_j}-\frac{dw}{w-w_j}\right) =\\
	&\hspace{.5cm}-2\oint_{C_{j}^{\epsilon}} \varphi\left(\frac{dw}{w-w_j}\right)-\oint_{C_{j}^{\epsilon}}\partial_{w}\varphi \log |w-w_j|^{2} dw -\oint_{C_{j}^{\epsilon}}\partial_{\bar{w}}\varphi \log |w-w_j|^{2} d\bar{w}, 
	\end{split}
	\end{equation*}
	we have
	\begin{equation*}
	\begin{split}
	&\hspace{-.8cm}\oint_{C_{j}^{\epsilon}}  (\varphi^{\varepsilon \mu_i} \circ F^{\varepsilon\mu_i}) \left(\frac{\overline{\partial_w F^{\varepsilon\mu_i}}(\bar{\varepsilon} \overline{M_i} \dd{w} +\dd{\bar{w}})}{\overline{w^{\varepsilon\mu_i}}-\overline{w_{j}^{\varepsilon \mu_i}}} - \frac{\partial_w F^{\varepsilon\mu_i}(\dd{w}+ \varepsilon M_i \dd{\bar{w}})}{w^{\varepsilon\mu_i} - w_{j}^{\varepsilon \mu_i}}\right)\\
	& = \underbrace{-2\oint_{C_{j}^{\epsilon}} (\varphi^{\varepsilon \mu_i} \circ F^{\varepsilon\mu_i})\left(\frac{\partial_{w} F^{\varepsilon\mu_i}(dw+\varepsilon M_{i} d\bar{w})}{F^{\varepsilon\mu_i}(w)-F^{\varepsilon\mu_i}(w_j)}\right)}_{B_{1}}\\
	&\underbrace{-\oint_{C_{j}^{\epsilon}}\left(\partial_{w}\varphi^{\varepsilon \mu_i} \circ F^{\varepsilon\mu_i}\right) \log|F^{\varepsilon\mu_i}(w)-F^{\varepsilon\mu_i}(w_j)|^{2} \hspace{1mm}\partial_{w}F^{\varepsilon\mu_i} \left(dw + \varepsilon M_{i}d\bar{w}\right)}_{B_2}
	\\
	&\underbrace{-\oint_{C_{j}^{\epsilon}}\left(\partial_{\bar{w}}\varphi^{\varepsilon \mu_i} \circ F^{\varepsilon\mu_i}\right)\log|F^{\varepsilon\mu_i}(w)-F^{\varepsilon\mu_i}(w_j)|^{2}\hspace{1mm}\overline{\partial_{w}F^{\varepsilon\mu_i}} \left( \bar{\varepsilon} \overline{M_{i}}dw+d\bar{w} \right)}_{B_3}.
	\end{split}
	\end{equation*}
	After simple calculations and using the Lemma \ref{varp}, one can see that
	\begin{equation*}
	\begin{split}
	&\pdv{\varepsilon}\Bigg|_{\varepsilon=0} B_{1} = 2\oint_{C_{j}^{\epsilon}}\partial_{w}\dot{F}^{i}(w)\left(\frac{dw}{w-w_j}\right)-2\oint_{C_{j}^{\epsilon}}\varphi \left(\frac{M_i\hspace{1mm}d\bar{w}}{w-w_j}\right)+\Sorder{1}\\
	&\pdv{\varepsilon}\Bigg|_{\varepsilon=0} B_{2} = \oint_{C_{j}^{\epsilon}}\partial^{2}_{w}\dot{F}^{i} \log|w-w_j|^{2} dw -\oint_{C_{j}^{\epsilon}}\partial_{w}\varphi \partial_{w}\dot{F}^{i}(w_j) dw \\
	&\hspace{2cm}-\oint_{C_{j}^{\epsilon}}\partial_{w}\varphi\hspace{1mm}M_{i} \log|w-w_j|^{2} d\bar{w}+\Sorder{1},\\
	&\pdv{\varepsilon}\Bigg|_{\varepsilon=0} B_{3} =\oint_{C_{j}^{\epsilon}}\partial_{w}\varphi\hspace{1mm}M_{i} \log|w-w_j|^{2}d\bar{w}+\oint_{C_{j}^{\epsilon}}\partial_{\bar{w}}\partial_{w}\dot{F}^{i} \log|w-w_j|^{2} d\bar{w}\\
	&\hspace{2cm}-\oint_{C_{j}^{\epsilon}}\partial_{\bar{w}}\varphi \partial_{w}\dot{F}^{i}(w_j)d\bar{w}+\Sorder{1},
	\end{split}
	\end{equation*}
	as $\epsilon \to 0$. This implies that
	\begin{equation*}
	\begin{split}
	&\pdv{\varepsilon}\Bigg|_{\varepsilon=0} \left(B_{1}+ B_{2}+B_{3}\right)\\
	&\hspace{1cm}= - \oint_{C_{j}^{\epsilon}}\partial_{w} \dot{F}^{i} \left(\frac{d\bar{w}}{\bar{w}-\bar{w}_j}-\frac{dw}{w-w_j}\right)-\oint_{C_{j}^{\epsilon}}\partial_{w}\dot{F}^{i}(w_j)d\varphi-2\oint_{C_{j}^{\epsilon}}\varphi \left(\frac{M_i\hspace{1mm} d\bar{w}}{w-w_j}\right)+\Sorder{1}\\
	& \hspace{1cm}\overset{\eqref{MAsymptotic}}{=} - \oint_{C_{j}^{\epsilon}}\partial_{w} \dot{F}^{i} \left(\frac{d\bar{w}}{\bar{w}-\bar{w}_j}-\frac{dw}{w-w_j}\right)+\Sorder{1}.
	\end{split}
	\end{equation*}
	Therefore, according to the above result, for the third term in (\ref{WiDerivative}) we get
	\begin{equation}\label{extra1}
	\begin{split}
	&\sum_{j=1}^{n_e} \left(1-\frac{1}{m_j}\right) \oint_{C_{j}^{\epsilon}} \pdv{\varepsilon}\Bigg|_{\varepsilon=0} (\varphi^{\varepsilon \mu_i} \circ F^{\varepsilon\mu_i}) \left(\frac{\overline{\partial_w F^{\varepsilon\mu_i}}(\bar{\varepsilon} \overline{M_i} \dd{w} +\dd{\bar{w}})}{\overline{w^{\varepsilon\mu_i}}-\overline{w_{j}^{\varepsilon \mu_i}}} - \frac{\partial_w F^{\varepsilon\mu_i}(\dd{w}+ \varepsilon M_i \dd{\bar{w}})}{w^{\varepsilon\mu_i} - w_{j}^{\varepsilon \mu_i}}\right) \\
	& \hspace{1cm}= -\sum_{j=1}^{n_e} \left(1-\frac{1}{m_j}\right) \oint_{C_{j}^{\epsilon}} \partial_{w}\dot{F}^{i} \left(\frac{d\bar{w}}{\bar{w}-\bar{w}_j}-\frac{dw}{w-w_j}\right)+\Sorder{1}.
	\end{split}
	\end{equation}
	Accordingly, the Eq.\eqref{WiDerivative} can be further simplified to get
	\begin{equation}\label{WiDerivativeI}
	\begin{split}
	\pdv{I_{\epsilon}}{w_i} & =  \iint_{O_{\epsilon}} \left[(\partial_{w_i}\partial_w \varphi + \partial_w^2 \varphi \dot{F}^i) \partial_{\bar{w}} \varphi + (\partial_{w_i} \partial_{\bar{w}} \varphi  + \partial_w \partial_{\bar{w}} \varphi  \dot{F}^i) \partial_w \varphi + |\partial_w\varphi|^2  \partial_w \dot{F}^i \right] \dd{w} \wedge \dd{\bar{w}}\\
	& - \sum_{k=1}^{n}\int_{\partial O^{\epsilon}_k} |\partial_w \varphi|^2  \Big(\dot{F}^i(w) - \dot{F}^i(w_k)\Big) \dd{\bar{w}}\\
	& - \sum_{j=1}^{n_e} \left(1-\frac{1}{m_j}\right) \oint_{C_{j}^{\epsilon}} \partial_{w}\dot{F}^{i} \left(\frac{\dd{\bar{w}}}{\bar{w}-\bar{w}_{j}} - \frac{\dd{w}}{w - w_{j}}\right)+\Sorder{1}.
	\end{split}
	\end{equation}
	Next, using the Lemma~\ref{varp}, we have
	\begin{equation}
	\left\{
	\begin{split}
	& \partial_{w_i}\partial_w \varphi + \partial_w^2 \varphi  \, \dot{F}^i =  - \partial_{w} \varphi \, \partial_{w} \dot{F}^i - \partial_w^2 \dot{F}^i, \\ \\
	& \partial_{w_i} \partial_{\bar{w}} \varphi  + \partial_w \partial_{\bar{w}} \varphi  \dot{F}^i = - \partial_{w} \varphi \, \partial_{\bar{w}} \dot{F}^i - \partial_{\bar{w}} \partial_{w} \dot{F}^i,
	\end{split}
	\right.
	\end{equation}
	which makes it possible to rewrite \eqref{WiDerivativeI} as
	\begin{equation}\label{WiDerivativeI2}
	\begin{split}
	\pdv{I_{\epsilon}}{w_i} & =  \iint_{O_{\epsilon}} \left[(- \partial_{w} \varphi \, \partial_{w} \dot{F}^i - \partial_w^2 \dot{F}^i) \partial_{\bar{w}} \varphi + (- \partial_{w} \varphi \, \partial_{\bar{w}} \dot{F}^i - \partial_{\bar{w}} \partial_{w} \dot{F}^i) \partial_w \varphi + |\partial_w\varphi|^2  \partial_w  \dot{F}^i \right] \dd{w} \wedge \dd{\bar{w}}\\
	& - \sum_{k=1}^{n}\int_{\partial O^{\epsilon}_k} |\partial_w \varphi|^2  \Big(\dot{F}^i(w) - \dot{F}^i(w_k)\Big) \dd{\bar{w}}\\
	& - \sum_{j=1}^{n_e} \left(1-\frac{1}{m_j}\right) \oint_{C_{j}^{\epsilon}} \partial_{w}\dot{F}^{i} \left(\frac{\dd{\bar{w}}}{\bar{w}-\bar{w}_{j}} - \frac{\dd{w}}{w - w_{j}}\right)+\Sorder{1}\\
	& = \iint_{O_{\epsilon}} \left[\Big(2 \partial_w^2 \varphi - (\partial_w \varphi)^2 \Big) \partial_{\bar{w}} \dot{F}^i - 2 \pdv{w}(\partial_w \varphi \, \partial_{\bar{w}} \dot{F}^i) + \pdv{\bar{w}}(\partial_{w} \varphi \, \partial_{w}\dot{F}^i) - \pdv{w}(\partial_{\bar{w}}\varphi \, \partial_{w} \dot{F}^i) \right] \dd{w} \wedge \dd{\bar{w}}\\
	& - \sum_{k=1}^{n}\int_{\partial O^{\epsilon}_k} |\partial_w \varphi|^2  \Big(\dot{F}^i(w) - \dot{F}^i(w_k)\Big) \dd{\bar{w}}\\
	& - \sum_{j=1}^{n_e} \left(1-\frac{1}{m_j}\right) \oint_{C_{j}^{\epsilon}} \partial_{w}\dot{F}^{i} \left(\frac{\dd{\bar{w}}}{\bar{w}-\bar{w}_{j}} - \frac{\dd{w}}{w - w_{j}}\right)+\Sorder{1}\\
	& = \underbrace{\iint_{O_{\epsilon}} \Big(2 \partial_w^2 \varphi - (\partial_w \varphi)^2 \Big) \partial_{\bar{w}} \dot{F}^i \dd{w} \wedge \dd{\bar{w}}}_{I_1} \underbrace{- 2 \int_{\partial O_{\epsilon}} \partial_w \varphi \, \partial_{\bar{w}} \dot{F}^i \dd{\bar{w}}}_{I_2} \underbrace{- \int_{\partial O_{\epsilon}} \partial_{w} \varphi \, \partial_{w}\dot{F}^i \dd{w}}_{I_3} \\
	&  \underbrace{- \int_{\partial O_{\epsilon}} \partial_{\bar{w}} \varphi \, \partial_{w} \dot{F}^i \dd{\bar{w}}}_{I_4} \underbrace{- \int_{\partial O_{\epsilon}} |\partial_w \varphi|^2  \dot{F}^i \dd{\bar{w}}}_{I_5} + \oint_{C^{\epsilon}_i} |\partial_w \varphi|^2 \dd{\bar{w}}\\
	&- \sum_{j=1}^{n_e} \left(1-\frac{1}{m_j}\right) \oint_{C_{j}^{\epsilon}} \partial_{w}\dot{F}^{i} \left(\frac{\dd{\bar{w}}}{\bar{w}-\bar{w}_{j}} - \frac{\dd{w}}{w - w_{j}}\right)+\Sorder{1}.
	\end{split}
	\end{equation}
	Let us compute each of the integrals $I_1, \dots, I_5$ separately using Lemma~\ref{lemma:asymptotics} and Corollary~\ref{corollary:FdotiAsymp} as well as equations \eqref{EnergyMomentumExpansion} and \eqref{EnergyMomentumExpansioninfty}. We begin with the integral $I_1$:
	\begin{equation*}
	\begin{split}
	I_1 &=  \iint_{O_{\epsilon}} \Big(2 \partial_w^2 \varphi - (\partial_w \varphi)^2 \Big) \partial_{\bar{w}} \dot{F}^i \dd{w} \wedge \dd{\bar{w}} = -2 \int_{\partial O_{\epsilon}} T_{\varphi} \dot{F}^i \dd{w}\\
	& = -2 \sum_{j=1}^{n_e} \oint_{C^{\epsilon}_{j}} \left(\frac{h_j}{2(w-w_{j})^2} + \frac{c_{j}}{w-w_{j}} + \dotsm \right) \left(\delta_{ij} + (w-w_{j}) \partial_{w} \dot{F}^i(w_{j}) + \dotsm \right) \dd{w} \\
	&-2 \sum_{j=n_e+1}^{n-1} \oint_{C^{\epsilon}_{j}} \left(\frac{1}{2(w-w_{j})^2} + \frac{c_{j}}{w-w_{j}} + \dotsm \right)\left(\delta_{ij} + (w-w_{j}) \partial_{w} \dot{F}^i(w_{j}) + \dotsm \right) \dd{w} \\
	& -2 \oint_{C_n^{\epsilon}} \left(\frac{1}{2w^2}+ \frac{c_n}{w^3} + \dotsm \right) \left(w \partial_{w} \dot{F}^i(\infty) + \dotsm\right) \dd{w}\\
	& = 4 \pi \sqrt{-1} c_i + 2 \pi \sqrt{-1} \sum_{k=1}^{n-1} h_k \partial_{w} \dot{F}^i(w_k)-2 \pi \sqrt{-1} \partial_{w} \dot{F}^i(\infty)  + \Sorder{1} \quad \text{as} \quad \epsilon \to 0.
	\end{split}
	\end{equation*}
	In the last line we have used the notation $h_k=1-1/m_k^2$ for $k=1,\dots,n_e$ and $h_k=1$ for $k=n_e+1,\dots,n-1$ . In addition, we have
	\begin{equation*}
	\begin{split}
	I_2 & = - 2 \int_{\partial O_{\epsilon}} \partial_w \varphi \, \partial_{\bar{w}} \dot{F}^i \dd{\bar{w}} = - 2 \int_{\partial O_{\epsilon}} \partial_w \varphi(w) \, M_i(w) \dd{\bar{w}}\\
	& \hspace{-.3cm}\overset{\eqref{MAsymptotic}}{=}-2 \sum_{j=1}^{n_e} \oint_{C^{\epsilon}_{j}} \left( - \frac{1-\frac{1}{m_j}}{w-w_{j}}+ \frac{c_{j}}{1-\frac{1}{m_j}} + \dotsm \right) \left(\frac{\bar{q}^{(j)}_1}{4 \bar{J}^{(j)}_1} \cdot |w - w_{j}|^{1-\frac{2}{m_j}} + \dotsm \right) \dd{\bar{w}} \\
	&-2 \sum_{j=n_e+1}^{n-1} \oint_{C^{\epsilon}_{j}} \left( -\frac{1}{w-w_{j}} \left(1+\frac{1}{\log\left|\frac{w-w_{j}}{J_{1}^{(j)}}\right|}\right) + c_{j} + \dotsm \right)\times \\
	&\hspace{2.7cm}\times\left(  - \frac{|\delta_{j}|^2 \bar{q}_1^{(j)}}{4 \pi^2 \bar{J}_1^{(j)}} \cdot |w-w_{j}| \log^2 |w-w_{j}|+ \dotsm \right)  \dd{\bar{w}} \\
	&-2 \oint_{C_n^{\epsilon}} \left(-\frac{1}{w} \left(1+\frac{1}{\log\left|\frac{w}{J_{-1}^{(n)}}\right|}\right) - \frac{c_n}{w^2} + \dotsm \right) \left( - \frac{|\delta_n|^2 \bar{q}_1^{(n)} \bar{J}_{-1}^{(n)}}{4 \pi^2} \cdot \frac{\log^2|w|}{|w|} + \dotsm \right)  \dd{\bar{w}}\\
	& = \Sorder{1},
	\end{split}
	\end{equation*}
	as $\epsilon \to 0$. Let us now calculate $I_3 + I_4$:
	\begin{equation*}
	\begin{split}
	I_3 + I_4 &= - \int_{\partial O_{\epsilon}} \partial_{w} \varphi \, \partial_{w}\dot{F}^i \dd{w} - \int_{\partial O_{\epsilon}} \partial_{\bar{w}} \varphi \, \partial_{w} \dot{F}^i \dd{\bar{w}}=  - \int_{\partial O_{\epsilon}} \partial_{w} \dot{F}^i \dd{\varphi} = \Sorder{1},
	\end{split}
	\end{equation*}
	as $\epsilon \to 0$. Finally, we can turn to calculating $I_5$: 
	\begin{equation*}
	\begin{split}
	I_5 & = - \int_{\partial O_{\epsilon}} |\partial_w \varphi|^2  \dot{F}^i \dd{\bar{w}}\\
	& = - \sum_{k=1}^{n-1} \oint_{C^{\epsilon}_k} |\partial_w \varphi|^2 \left(\delta_{ik} + (w-w_k) \partial_{w} \dot{F}^i(w_k) + \dotsm \right) \dd{\bar{w}}\\
	&\hspace{.4cm}- \oint_{C_n^{\epsilon}} \left(\frac{1}{|w|^2} + \dotsm \right) \left(w \partial_{w} \dot{F}^i (\infty) + \dotsm \right)\dd{\bar{w}} \\
	& =- \oint_{C^{\epsilon}_i} |\partial_w \varphi|^2 \dd{\bar{w}}   - \sum_{k=1}^{n-1} \oint_{C^{\epsilon}_k} \left(\frac{\left(1-\frac{1}{m_k}\right)^2}{|w-w_k|^2} + \dotsm \right)  \left((w-w_k) \partial_{w} \dot{F}^i(w_k) + \dotsm \right) \dd{\bar{w}} \\
	&\hspace{.4cm}- \oint_{C_n^{\epsilon}} \left(\frac{1}{|w|^2} + \dotsm \right) \left(w \partial_{w} \dot{F}^i (\infty) + \dotsm \right)\dd{\bar{w}} \\ 
	& = - \oint_{C^{\epsilon}_i} |\partial_w \varphi|^2 \dd{\bar{w}}  - 2 \pi \sqrt{-1} \sum_{k=1}^{n-1} \left(1-\frac{1}{m_k}\right)^2 \partial_{w} \dot{F}^i(w_k) + 2 \pi \sqrt{-1} \partial_{w} \dot{F}^i(\infty)+ \Sorder{1} ,
	\end{split}
	\end{equation*}
	as $\epsilon \to 0$ with $m_k = \infty$ for $k=n_e+1,\dots,n-1$. As a result, we have
	\begin{equation}\label{I1toI5}
	I_1+\dotsm+I_5 = 4 \pi \sqrt{-1} c_i  - \oint_{C^{\epsilon}_i} |\partial_w \varphi|^2 \dd{\bar{w}} + 2 \pi \sqrt{-1} \sum_{j=1}^{n_e} \left(-\frac{2}{m_j^2}+\frac{2}{m_j}\right) \partial_{w} \dot{F}^i(w_{j})  + \Sorder{1},
	\end{equation}
	as $\epsilon \to 0$. Moreover, by changing the conformal structure, the variation of counterterm action $\tilde{S}^{\epsilon}_{\text{ct}}$ gives
	\begin{equation}\label{Sct}
	\begin{split}
	&\pdv{\tilde{S}_{\text{ct}}^{\epsilon}}{w_i} = -2\pi \sum_{j=1}^{n_e} \left(1-\frac{1}{m_j}\right)^2 \partial_{w}\dot{F}^{i}(w_j) +\Sorder{1} \hspace{.5cm}\text{as}\hspace{.3cm}\epsilon\rightarrow 0.
	\end{split} 
	\end{equation}
	Notice that we have only taken into account the contributions coming from the regulating terms for conical points and not the punctures. The reason for this can be traced back to the equation \eqref{BTP}. By doing similar analysis as in equation \eqref{extra1} for the contour integrals in \eqref{BTP}, it is easy to see that the quasi-conformal transformations of these contour integrals are canceled by those of the counterterm $-2\pi n_p \log \epsilon$. In other words, it is sufficient to only consider the quasi-conformal transformations of the bulk term for the case of punctures, which is in agreement with the analysis in \cite{Zograf1988ONLE}.
	
	Now, using equations \eqref{WiDerivativeI2}, \eqref{I1toI5}, \eqref{def}, \eqref{Sct}, and putting everything together, we arrive at our desired result
	\begin{equation*}
	\begin{split}
	\pdv{\tilde{S}_{\boldsymbol{m}}^{\epsilon}}{w_i} & = \frac{\sqrt{-1}}{2} \left(4 \pi \sqrt{-1} c_i  - \oint_{C^{\epsilon}_i} |\partial_w \varphi|^2 \dd{\bar{w}} + 2 \pi \sqrt{-1} \sum_{j=1}^{n_e} \left(1-\frac{1}{m_j}\right)  \left(\frac{2}{m_j}\right) \partial_{w} \dot{F}^i(w_{j})\right.\\
	& \left. + \oint_{C^{\epsilon}_i} |\partial_w \varphi|^2 \dd{\bar{w}} -2 \pi \sqrt{-1} \sum_{j=1}^{n_e} \left(1-\frac{1}{m_j}\right)  \left(\frac{2}{m_j}\right) \partial_{w} \dot{F}^i(w_{j}) + \Sorder{1}\right)\\
	& = - 2 \pi c_i + \Sorder{1} \quad \text{as} \quad \epsilon \to 0.
	\end{split}
	\end{equation*} 
	In order to complete the proof, it remains to show that the remainder in the above formula can be estimated uniformly in a neighborhood of an arbitrary point $(w_1,\dots,w_{n-3}) \in \moduli_{0,n}$.  Let $\Gamma$ be a Fuchsian group uniformizing the orbisurface $O$. Then, the Hauptmodule $J_{\mu}$, where $\mu = t_1 \mu_1 + \dotsm + t_{n-3}\hspace{1mm}\mu_{n-3} \in \Hilbert^{-1,1}(\UHP,\Gamma)$, is continuously differentiable with respect to the Bers' coordinates $t_i$ and its coefficients in an expansion similar to Eq.\eqref{Jexpansion}  have the same property. From this, we can conclude with the help of Eq.\eqref{Liouvillefield} and Lemma~\ref{Lemma:covering} that the remainders in assertions (2) and (4) of  Lemma~\ref{lemma:asymptotics} can be estimated uniformly. Using the commutative diagram \eqref{diagramFmu}, we can conclude that an analogous assertion is valid also for the remainders in Corollary~\ref{corollary:FdotiAsymp} for $\dot{F}^i$, and this completes the proof of the theorem.
\end{proof}
\begin{remark}
	Theorem~\ref{theorem1} means that $\sum_{i=1}^{n-3} (c_i \dd{w_i} + \bar{c}_i \dd{\bar{w}_i})$ is an exact 1-form on $\moduli_{0,n}$ with anti-derivative $-S_{\boldsymbol{m}}/2\pi$.
\end{remark}

In addition to Theorem~\ref{theorem1}, we can also generalize \cite[Theorem~2]{Zograf1988ONLE} to include branch points: The Weil-Petersson metric defined on $\teich_{0,n}(\boldsymbol{m})$ in Sec.~\ref{KahlerMetrics} can be projected onto $\moduli_{0,n}$, since it is invariant under $\operatorname{Aut}(\Psi)$. We will continue to call this metric obtained on $\moduli_{0,n}$ as the Weil-Petersson metric and denote it with the same notation $\langle \cdot , \cdot \rangle_{_{\text{WP}}}$. Then, we have the following theorem:
\begin{theorem*}[Takhtajan and Zograf]
	For any fixed vector of orders $\boldsymbol{m}=(m_1,\dots,m_{n})$ such that $\sum_{j=1}^{n_e} (1-\tfrac{1}{m_j})+n_p > 2$, the function $-S_{\boldsymbol{m}}$ is a real-analytic K\"{a}hler potential for the metric $\langle \cdot , \cdot \rangle_{_{WP}}$ on $\moduli_{0,n}$:\footnote{See Theorem 2 in \cite[]{Zograf1988ONLE} for the genus $g=0$ punctured Riemann surfaces.}
	\begin{equation}\label{theorem2}
	\bar{\partial} \partial S_{\boldsymbol{m}} = - 2 \sqrt{-1} \hspace{1mm}\omega_{\text{WP}}.
	\end{equation}
\end{theorem*}
\begin{proof}
	In order to prove this theorem, we have to prove that the accessory parameters $c_1,\dots,c_{n-3}$ are continuously differentiable on $\moduli_{0,n}$, and 
	\begin{equation}\label{wjderivativeci}
	\pdv{c_i}{\bar{w}_j} = \frac{1}{2\pi} \left\langle \pdv{w_i} , \pdv{w_j} \right\rangle_{_{\text{WP}}} \quad \text{for} \quad i,j = 1, \dots,n-3.
	\end{equation}
	The proof of continuous differentiability of the functions $c_i$ on $\moduli_{0,n}$ follows readily from the definition of accessory parameters  \eqref{accessory} and continuous differentiability of the Hauptmodule $J$ with respect to the Ber's coordinates. We now turn to proving Eq.\eqref{wjderivativeci} in the same way as for Theorem 2 in \cite[]{Zograf1988ONLE} where no branch point exists. Let $(w_1,\dots,w_{n-3})$ be an arbitrary point in $\moduli_{0,n}$ and let $\Gamma$ be a Fuchsian group uniformizing the orbisurface $O$. It follows from the commutative diagram \eqref{diagramFmu} that 
	\begin{equation*}
	\text{Sch}\left(J_{\varepsilon \mu_j}^{-1} \circ F^{\varepsilon \mu_j};w\right) = \text{Sch}\left(f^{\varepsilon \mu_j} \circ J^{-1};z\right),
	\end{equation*} 
	where $\mu_j$ is an element of the basis in $\Hilbert^{-1,1}(\UHP,\Gamma)$ given by Eq.\eqref{Beltramibasis}, and $\varepsilon \in \cmpx$ is sufficiently small. By using the following well-known property of Schwarzian derivative,
	\begin{equation*}
	\text{Sch}\left(A \circ B;w\right) = \text{Sch}\left(A;w\right)\circ B \hspace{1mm}B^{\prime 2} +\text{Sch}\left(B;w\right)
	\end{equation*}
	we have
	\begin{equation*}
	\text{Sch}\left(J_{\varepsilon \mu_j}^{-1};w\right) \circ F^{\varepsilon \mu_j} \left(\partial_{w} F^{\varepsilon \mu_j}\right)^2 + \text{Sch}\left(F^{\varepsilon \mu_j};w\right)=  \text{Sch}\left(f^{\varepsilon \mu_j};z\right) \circ J^{-1} (J^{-1})'^2 + \text{Sch}\left(J^{-1};z\right). 
	\end{equation*}
	We can now differentiate both sides of the above equality with respect to $\bar{\varepsilon}$ at the point $\varepsilon =0 $, using Eq.\eqref{EnergyMomentumExpansion} and the fact that $F^{\varepsilon \mu_j}$, and as a result $w_i^{\varepsilon \mu_j}$, are holomorphic functions of $\varepsilon$ at $\varepsilon =0 $. The left-hand side gives
	\begin{equation}\label{epsilonbarderivativelhs}
	\begin{split}
	&\left.\pdv{\bar{\varepsilon}}\right|_{\varepsilon=0}  \left(\text{Sch}\left(J_{\varepsilon \mu_j}^{-1};w\right) \circ F^{\varepsilon \mu_j} \left(\partial_{w} F^{\varepsilon \mu_j}\right)^2 + \text{Sch}\left(F^{\varepsilon \mu_j};w\right)\right)\\
	&\hspace{.8cm} = \left. \pdv{T_{\varphi^{\varepsilon \mu_j}}}{\bar{\varepsilon}}\right|_{\varepsilon=0}  + \partial_w T_{\varphi} \left. \pdv{F^{\varepsilon \mu_j}}{\bar{\varepsilon}}\right|_{\varepsilon=0}+  T_{\varphi} \left.\pdv{\bar{\varepsilon}}\right|_{\varepsilon=0}  \left(\partial_{w} F^{\varepsilon \mu_j}\right)^2 + \left.\pdv{\bar{\varepsilon}}\right|_{\varepsilon=0} \text{Sch}\left(F^{\varepsilon \mu_j};w\right)\\
	& \hspace{.8cm}= \left.\pdv{\bar{\varepsilon}}\right|_{\varepsilon=0} \sum_{i=1}^{n-1}\left(\frac{h_i}{2(w-w_i^{\varepsilon \mu_j})^2} + \frac{c_i^{\varepsilon \mu_j}}{w-w_i^{\varepsilon \mu_j}}\right)=\sum_{i=1}^{n-1} \left(\left.\pdv{c_i^{\varepsilon \mu_j}}{\bar{\varepsilon}}\right|_{\varepsilon=0}\right) \frac{1}{w-w_i},
	\end{split}
	\end{equation}
	where $c_i^{\varepsilon \mu_j} = c_i(w_1^{\varepsilon \mu_j},\dots,w_{n-3}^{\varepsilon \mu_j})$. To compute the right-hand side, we will use the definition of Schwarzian derivative \eqref{Sch}, the fact that $\partial/\partial{\bar{\varepsilon}}$ of the functions $\partial_z f^{\varepsilon \mu_j}$ and $\partial_z^2 f^{\varepsilon \mu_j}$ vanishes at $\varepsilon=0$, and the well-known Ahlfors formula\footnote{The Ahlfors formula can be derived by comparing  the expression for $\dot{f}^{\mu_j}_{-}(z)$, given by Eq.\eqref{fdotmupm}, and the equality $\varLambda \varLambda^{\ast} = \operatorname{id}$ discussed in Section \ref{subsub:complexstructure}.} \cite{Ahlfors_quasiconformal_06,Ahlfors1961SomeRO}
	\begin{equation*}
	\left.\pdv{\bar{\varepsilon}} \partial_z^3 f^{\varepsilon \mu_j} \right|_{\varepsilon=0}  = - \frac{1}{2} q_j,
	\end{equation*}
	where $\mu_j$ and $q_j$ are connected by the relation \eqref{Beltramibasis}. Then, we have
	\begin{equation}\label{epsilonbarderivativerhs}
	\begin{split}
	&\left.\pdv{\bar{\varepsilon}}\right|_{\varepsilon=0} \bigg(\text{Sch}\left(f^{\varepsilon \mu_j};z\right) \circ J^{-1} (J^{-1})'^2 + \text{Sch}\left(J^{-1};z\right)\bigg)\\
	&\hspace{1.2cm}= \left.\pdv{\bar{\varepsilon}}\left[\partial_z^3 f^{\varepsilon \mu_j} \circ J^{-1} (J^{-1})'^2 \right]\right|_{\varepsilon=0}= - \frac{1}{2} q_j \circ J^{-1} (J^{-1})'^2= - \frac{1}{2} Q_j.
	\end{split}
	\end{equation}
	Equating the left-hand side, \eqref{epsilonbarderivativelhs}, and the right-hand side, \eqref{epsilonbarderivativerhs}, and using Eq.\eqref{accessoryconstraints}, we get
	\begin{equation*}
	\begin{split}
	&\sum_{i=1}^{n-1} \left(\left.\pdv{c_i^{\varepsilon \mu_j}}{\bar{\varepsilon}}\right|_{\varepsilon=0}\right) \frac{1}{w-w_i} \\
	&\hspace{.3cm} =  \sum_{i=1}^{n-3} \left(\left.\pdv{c_i^{\varepsilon \mu_j}}{\bar{\varepsilon}}\right|_{\varepsilon=0}\right) \frac{1}{w-w_i}+\left(\left.\pdv{\bar{\varepsilon}} c_{n-2}^{\varepsilon \mu_j}\right|_{\varepsilon=0}\right) \frac{1}{w} + \left(\left.\pdv{\bar{\varepsilon}} c_{n-1}^{\varepsilon \mu_j}\right|_{\varepsilon=0}\right) \frac{1}{w-1}  \\
	& \hspace{.3cm}= \sum_{i=1}^{n-3} \left(\left.\pdv{c_i^{\varepsilon \mu_j}}{\bar{\varepsilon}}\right|_{\varepsilon=0}\right) \frac{1}{w-w_i}+ \left(\left.\pdv{\bar{\varepsilon}} \left(-1+\frac{n}{2}+ \sum_{i=1}^{n-3}c_i^{\varepsilon \mu_j}(w_i^{\varepsilon \mu_j}-1)\right)\right|_{\varepsilon=0}\right) \frac{1}{w} \\
	&\hspace{1cm}+  \left(\left.\pdv{\bar{\varepsilon}} \left(1-\frac{n}{2}- \sum_{i=1}^{n-3}c_i^{\varepsilon \mu_j} w_i^{\varepsilon \mu_j}\right)\right|_{\varepsilon=0}\right) \frac{1}{w-1} \\
	& \hspace{.3cm}= \sum_{i=1}^{n-3} \left(\left.\pdv{c_i^{\varepsilon \mu_j}}{\bar{\varepsilon}}\right|_{\varepsilon=0}\right) \left(\frac{1}{w-w_i}+\frac{w_i-1}{w}-\frac{w_i}{w-1}\right)\overset{\eqref{Ridef}}{=} - \pi \sum_{i=1}^{n-3} \left(\left.\pdv{c_i^{\varepsilon \mu_j}}{\bar{\varepsilon}}\right|_{\varepsilon=0}\right) R_i = - \frac{1}{2} Q_j,
	\end{split}
	\end{equation*}
	or
	\begin{equation*}
	\sum_{i=1}^{n-3} \left(\left.\pdv{c_i^{\varepsilon \mu_j}}{\bar{\varepsilon}}\right|_{\varepsilon=0}\right) R_i = \frac{1}{2\pi} Q_j.
	\end{equation*}
	It now follows from the biorthogonality of the bases $R_i$ and $Q_j$ as well as Lemma~\ref{Lemma:covering} that
	\begin{equation*}
	\pdv{c_i}{\bar{w}_j} = \frac{1}{2\pi} \langle Q_i , Q_j \rangle = \frac{1}{2 \pi} \left\langle \pdv{w_i} , \pdv{w_j} \right\rangle_{_{\text{WP}}},
	\end{equation*}
	which is what we wanted. Combining Eq.\eqref{wjderivativeci} and Theorem~\ref{theorem1}, we can conclude that the real-valued function $-S_{\boldsymbol{m}}$ is a potential for the Weil-Petersson metric on $\moduli_{0,n}$
	\begin{equation*}
	-\pdv[2]{S_{\boldsymbol{m}}}{w_i}{\bar{w}_j} = \left\langle \pdv{w_i} , \pdv{w_j} \right\rangle_{_{\text{WP}}},
	\end{equation*}
	while the function $-S_{\boldsymbol{m}} \circ \Psi$ is a potential for the Weil-Petersson metric on the Teichm\"{u}ller space $\teich_{0,n}(\boldsymbol{m})$. In addition, the real analyticity of the Weil-Petersson metric (see \cite{Ahlfors1961SomeRO}) implies the real analyticity of the accessory parameters $c_1,\dots,c_{n-3}$ and of their generating function $-S_{\boldsymbol{m}}/2\pi$ on the space $\moduli_{0,n}$.
\end{proof}
\begin{remark}
	Equations~\eqref{theorem1} and \eqref{theorem2} also give us a proof that the Weil-Petersson metric is a K\"{a}hler metric.
\end{remark}
Finally, let $\symmoduli_{0,n}(\boldsymbol{m}) = \moduli_{0,n}/\symm{\sigtype}$ be the moduli space of orbifold Riemann surfaces with signature $(0;m_1,\dots,m_{n_e};n_p)$. One can generalize the results of \cite[Sec.1]{Zograf_1990} to prove the following Lemma:
\begin{lemma}\label{lemma:SmHermiatiannmetricoverlambdanogenus}
	A Hermitian metric in a holomorphic $\mathbb{Q}$-line bundle $\lambda_{0,\boldsymbol{m}}$ over $\symmoduli_{0,n}(\boldsymbol{m})$ is determined by $\exp[S_{\boldsymbol{m}}/\pi]$, so that 
	\begin{equation}
	\chern{\lambda_{0,\boldsymbol{m}}}{\exp[S_{\boldsymbol{m}}/\pi]} = \frac{1}{\pi^2} \omega_{WP}.
	\end{equation}
\end{lemma}
\begin{proof}
	We first need to show that $\exp[S_{\boldsymbol{m}}/\pi]$ is a hermitian metric in the holomorphic $\mathbb{Q}$-line bundle $\lambda_{0,\boldsymbol{m}}$ defined in the Section \ref{subsec:moduli}. To do so, we use the same representation of $\symm{\sigtype}$ introduced there. For our purposes, it suffices to prove this lemma for the case where the signature of orbifold Riemann surface $O$ is given by $(0;\underbrace{m,\dots,m}_{s},\underbrace{m',\dots,m'}_{s'})$ with $s \equiv s_m$ and $s' \equiv s_{m'} >3$. Namely, we have only two kinds of points, $s$ branch points of order $m$ and $s'$ branch points of order $m'$. Furthermore, we will fix the last three points with order $m' > m$ to be at $0,1,\infty$. Then, the generators and 1-cocycles of $\symm{\sigtype}=\symm{s}\times\symm{s'}$ would be the same as the ones we introduced in Section \ref{subsec:moduli}. Moreover, in this case we have
	\begin{multline}\label{actionlinebundle}
	S_{\boldsymbol{m}}[\varphi] = \\
	\hspace{.5cm}\lim_{\epsilon \to 0^{+}} \left( \iint_{O_{\epsilon}}(|\partial_w \varphi|^2 + e^{\varphi}) \dd[2]{w} + \frac{\sqrt{-1}}{2} \sum_{k=1}^{s+s'-1} \left(1-\frac{1}{m_k}\right) \oint_{C_{k}^{\epsilon}} \varphi \left(\frac{\dd{\bar{w}}}{\bar{w}-\bar{w}_{k}} - \frac{\dd{w}}{w - w_{k}}\right)\right.\\\left.
	\hspace{-1cm}+ \frac{\sqrt{-1}}{2} \left(1+\frac{1}{m'}\right) \oint_{C_{s+s'}^{\epsilon}} \varphi \left(\frac{\dd{\bar{w}}}{\bar{w}} - \frac{\dd{w}}{w}\right)\right),
	\end{multline}
	where 
	\begin{equation}
	\left\{
	\begin{split}
	&O_{\epsilon} = \cmpx \Big\backslash \bigcup_{k=1}^{s+s'-1}\left\{w \, \Big| \, |w-w_k|<\epsilon\right\} \cup \left\{w \, \Big| \, |w|>\epsilon^{-1}\right\},\\
	&	C_{k}^{\epsilon} = \left\{w \, \Big| \, |w-w_{k}| = \epsilon\right\},\hspace{.5cm}k=1,2,\dots,s+s'-1\\
	&C_{s+s'}^{\epsilon} = \left\{w \, \Big| \, |w| = 1/\epsilon\right\},
	\end{split}
	\right.
	\end{equation}
	and the circles are oriented as a component of $\partial O_{\epsilon}$. Notice that we did not consider the counterterms in (\ref{actionlinebundle}). Actually, $\symm{\sigtype}$ does not change the conformal family, and therefore the counterterms will not contribute to the variation under the action of this group, and therefore writing them would be redundant. 
	To study the variation of $S_{\boldsymbol{m}}[\varphi]$ under the action of $\symm{\sigtype}$, it suffices to study its variation under the generators $\{(\sigma_{j,j+1},\sigma'_{i,i+1})\}_{j=1,i=1}^{j=s-1,i=s'-1}$ of $\symm{\sigtype}$. Considering the background we provided in Section \ref{subsec:moduli} on the structure of the generators, variation is translated through the effect of the transformation $\gamma_{j,j+1;i,i+1}$. We have:
	\begin{equation}\label{variationSsymm}
	\begin{split}
	&\Delta S_{\boldsymbol{m}}[\varphi]=S_{\boldsymbol{m}}[\varphi]\circ (\sigma_{k,k+1},\sigma'_{i,i+1})-S_{\boldsymbol{m}}[\varphi]\\
	&\hspace{.1cm}=\lim_{\epsilon \to 0^{+}} \bigg( \frac{\sqrt{-1}}{2}\bigg\{ \Delta S_{\boldsymbol{m}}^{(1)}[\varphi]+ \big(1-\frac{1}{m}\big)\Delta  S_{\boldsymbol{m}}^{(2)}[\varphi]+ \big(1-\frac{1}{m'}\big)\Delta S_{\boldsymbol{m}}^{(3)}[\varphi]+ \big(1+\frac{1}{m'}\big)\Delta S_{\boldsymbol{m}}^{(4)}[\varphi]\bigg\}\bigg),
	\end{split}		
	\end{equation}
	where
	\begin{equation*}
	\begin{split}
	& \Delta S_{\boldsymbol{m}}^{(1)}[\varphi]=\iint_{\tilde{O}_{\epsilon}}(|\partial_{\gamma_{j,j+1;i,i+1}w} \tilde{\varphi}|^2 + e^{\tilde{\varphi}}) \dd{(\gamma_{j,j+1;i,i+1}w)}\wedge\dd{(\overline{\gamma_{j,j+1;i,i+1}w})}
	\\
	&\hspace{2cm}-\iint_{O_{\epsilon}}(|\partial_w \varphi|^2 + e^{\varphi}) \dd{w}\wedge\dd{\bar{w}}, \\
	&\Delta S_{\boldsymbol{m}}^{(2)}[\varphi] = \sum_{k=1}^{s}\bigg\{\oint_{\tilde{C}_{k}^{\epsilon}} \tilde{\varphi} \left(\frac{\dd{(\overline{\gamma_{j,j+1;i,i+1}w})}}{\overline{\gamma_{j,j+1;i,i+1}w}-\overline{\gamma_{j,j+1;i,i+1}w}_{k}} - \frac{\dd{(\gamma_{j,j+1;i,i+1}w)}}{\gamma_{j,j+1;i,i+1}w - \gamma_{j,j+1;i,i+1}w_{k}}\right)\\
	&\hspace{2cm}-\oint_{C_{k}^{\epsilon}} \varphi \left(\frac{\dd{\bar{w}}}{\bar{w}-\bar{w}_{k}}
	- \frac{\dd{w}}{w - w_{k}}\right)\bigg\},\\
	&\Delta S_{\boldsymbol{m}}^{(3)}[\varphi] =\sum_{k=s+1}^{s+s'-1}  \bigg\{\oint_{\tilde{C}_{k}^{\epsilon}} \tilde{\varphi} \left(\frac{\dd{(\overline{\gamma_{j,j+1;i,i+1}w})}}{\overline{\gamma_{j,j+1;i,i+1}w}-\overline{\gamma_{j,j+1;i,i+1}w_k}} - \frac{\dd{(\gamma_{j,j+1;i,i+1}w)}}{\gamma_{j,j+1;i,i+1}w - \gamma_{j,j+1;i,i+1}w_k}\right)\\
	&\hspace{2cm}-\oint_{C_{k}^{\epsilon}} \varphi \left(\frac{\dd{\bar{w}}}{\bar{w}-\bar{w}_{k}} - \frac{\dd{w}}{w - w_{k}}\right)\bigg\},\\
	&\Delta S_{\boldsymbol{m}}^{(4)}[\varphi] = \oint_{\tilde{C}_{s+s'}^{\epsilon}} \tilde{\varphi} \left(\frac{\dd{(\overline{\gamma_{j,j+1;i,i+1}w})}}{\overline{\gamma_{j,j+1;i,i+1}w}} - \frac{\dd{(\gamma_{j,j+1;i,i+1}w)}}{\gamma_{j,j+1;i,i+1}w}\right)-\oint_{C_{s+s'}^{\epsilon}} \varphi \left(\frac{\dd{\bar{w}}}{\bar{w}} - \frac{\dd{w}}{w }\right),
	\end{split}
	\end{equation*}
	and $\tilde{O}_{\epsilon}$,  $\tilde{C}_{k}^{\epsilon}$, $\tilde{\varphi}$ are the transformed orbifold Riemann surface, circles, and Liouville field, respectively, such that $\tilde{O}_{\epsilon}=\cmpx\backslash\cup_{k=1}^{s+s'}\operatorname{Int}\tilde{C}_{k}^\epsilon$. By looking at (\ref{transpositionsdirectprod}), we see that the index $j$ does not have any non-trivial effect and we only have to worry about different values of $i$. For $i<s'-3$, the transformation $\gamma_{j,j+1;i,i+1}$ is given by the identity of $\PSLC$, so, for this cases we have $\Delta S_{\boldsymbol{m}}[\varphi]=0$. Accordingly, the only non-trivial cases are $i=s'-3,s'-2,s'-1$. Let us look at the case with $i=s'-3$. Here we have the transformation $\gamma_{s'-3,s'-2}=(w-w_{s+s'-3})/(1-w_{s+s'-3})$ which for simplicity we call $\gamma$. Let us calculate each contribution in \eqref{variationSsymm} separately \footnote{For $\Delta S_{\boldsymbol{m}}^{(1)}[\varphi]$, the exponential terms give the same  constants, and we can safely ignore them in variation.}.
	\begin{equation*}\label{I1}
	\begin{aligned}
	\Delta S_{\boldsymbol{m}}^{(1)}[\varphi]&=
	\iint_{\tilde{O}_{\epsilon}}|\partial_{\gamma w} \tilde{\varphi}|^2 \dd{(\gamma w)}\wedge\dd{(\overline{\gamma w})}-\iint_{O_{\epsilon}}|\partial_w \varphi|^2 \dd{w}\wedge\dd{\bar{w}} \\
	&=\iint_{\gamma^{-1}(\tilde{O}_{\epsilon})}|\partial_{ w} \tilde{\varphi}\circ\gamma|^2 |\gamma'|^{-2}(|\gamma'|^2\dd{ w}\wedge\dd{\bar{w}}) -\iint_{O_{\epsilon}}|\partial_w \varphi|^2 \dd{w}\wedge\dd{\bar{w}}\\
	&=\iint_{\gamma^{-1}(\tilde{O}_{\epsilon})\backslash O_{\epsilon}} |\partial_w \varphi|^2 \dd{w}\wedge\dd{\bar{w}},
	\end{aligned}
	\end{equation*}
	where $\gamma'=\partial\gamma(w)/\partial w$ and in the last line, we used the invariance of the hyperbolic metric:
	\begin{equation}\label{hyperbolicinvariance}
	\begin{aligned}
	e^{\tilde{\varphi}\circ\gamma}\dd{(\gamma w)}\wedge\dd{(\overline{\gamma w})}=e^\varphi \dd{w}\wedge\dd{\bar{w}}\implies \tilde{\varphi}\circ\gamma+\log|\gamma'|^2=\varphi\implies
	\partial_w\tilde{\varphi}\circ\gamma=\partial_w\varphi,
	\end{aligned}
	\end{equation}
	since $\gamma'=1/(1-w_{s+s'-3})$. The region $\gamma^{-1}(\tilde{O}_{\epsilon})\backslash O_{\epsilon} $ contains a part of $\cmpx$ bounded by the circles  $\tilde{C}_{k}$ and $C_{k}$. Thus, by integration by parts and using the equations of motion together with taking into account the orientation of these circles, we have:
	\begin{equation}\label{I1-1}
	\Delta S_{\boldsymbol{m}}^{(1)}[\varphi]=\hspace{-2mm}\sum_{k=1}^{s+s'-1}\left(\oint_{\tilde{C}_{k}^{\epsilon}}\varphi\partial_{\bar{w}}\varphi\dd{\bar{w}}-\oint_{C_{k}^{\epsilon}}\varphi\partial_{\bar{w}}\varphi\dd{\bar{w}}\right)+\left(\oint_{\tilde{C}_{s+s'}^{\epsilon}}\varphi\partial_{\bar{w}}\varphi\dd{\bar{w}}-\oint_{C_{s+s'}^{\epsilon}}\varphi\partial_{\bar{w}}\varphi\dd{\bar{w}}\right).
	\end{equation}
	The the explicit form of the circles $\tilde{C}_{k}^\epsilon$ are given by
	\begin{equation}\label{Ctildes}
	\tilde{C}_{k}^\epsilon = \left\{
	\begin{split}
	&\big\{w \, \big| \, \:|\gamma w-\gamma w_k|=\epsilon \big\}= \big\{w \, \big|\, \:| w-w_k|=\epsilon|1-w_{s+s'-3}| \big\} &k\leq s+s'-1, \\\\
	& \left\{w \, \big|\, \:|\gamma w|=\tfrac{1}{\epsilon} \right\}= \big\{w \, \big|\, \:| w|=\tfrac{1}{\epsilon} |1-w_{s+s'-3}| \big\} &k= s+s'.
	\end{split}
	\right.
	\end{equation}
	Now, with the use of equations \eqref{I1-1},\eqref{Ctildes} and the asymptotic form of $\varphi$ given by \ref{lemma:asymptotics}, one can find\footnote{Note that the orientation of the contour around the point at infinity is opposite of the other points, hence the sign difference in the last term.}
	\begin{equation}\label{I1-2}
	\begin{aligned}
	&\Delta S_{\boldsymbol{m}}^{(1)}[\varphi]=\hspace{-1mm}-4\pi\sqrt{-1} \left(1-\frac{1}{m}\right)^2\hspace{-1mm}s\log\frac{1}{|1-w_{s+s'-3}|}-4\pi\sqrt{-1} \left(1-\frac{1}{m'}\right)^2\hspace{-1mm}(s'-1)\log\frac{1}{|1-w_{s+s'-3}|}\\
	&\hspace{1.75cm}+4\pi\sqrt{-1} \left(1+\frac{1}{m'}\right)^2\log\frac{1}{|1-w_{s+s'-3}|}.
	\end{aligned}
	\end{equation} 
	Next, for the $\Delta S_{\boldsymbol{m}}^{(2)}[\varphi]$ we have
	\begin{equation*}
	\begin{aligned}
	\Delta S_{\boldsymbol{m}}^{(2)}[\varphi]&=\sum_{k=1}^{s}  \left(\oint_{\tilde{C}_{k}^{\epsilon}} \tilde{\varphi} \left(\frac{\dd{(\overline{\gamma w})}}{\overline{\gamma w}-\overline{\gamma w}_{k}} - \frac{\dd{(\gamma w)}}{\gamma w - \gamma w_{k}}\right)\right. \left.-\oint_{C_{k}^{\epsilon}} \varphi \left(\frac{\dd{\bar{w}}}{\bar{w}-\bar{w}_{k}} - \frac{\dd{w}}{w - w_{k}}\right)\right)\\
	&\hspace{-1.2cm}=\sum_{k=1}^{s}  \left(\oint_{\tilde{C}_{k}^{\epsilon}}(\varphi-\log|\gamma'|^2) \left(\frac{\bar{\gamma'}\dd{\bar{ w}}}{\overline{\gamma w}-\overline{\gamma w}_{k}} - \frac{\gamma'\dd{ w}}{\gamma w - \gamma w_{k}}\right)\right. \left.-\oint_{C_{k}^{\epsilon}} \varphi \left(\frac{\dd{\bar{w}}}{\bar{w}-\bar{w}_{k}} - \frac{\dd{w}}{w - w_{k}}\right)\right)\\
	&\hspace{-1.2cm}=\sum_{k=1}^{s}  \left(\oint_{\tilde{C}_{k}^{\epsilon}}(\varphi+2\log|1-w_{s+s'-3}|) \left(\frac{\dd{\bar{w}}}{\bar{w}-\bar{w}_{k}} - \frac{\dd{w}}{w - w_{k}}\right)\right. \left.-\oint_{C_{k}^{\epsilon}} \varphi \left(\frac{\dd{\bar{w}}}{\bar{w}-\bar{w}_{k}} - \frac{\dd{w}}{w - w_{k}}\right)\right),
	\end{aligned}
	\end{equation*}
	where we have used (\ref{hyperbolicinvariance}) and  $\gamma'=1/(1-w_{s+s'-3})$. By using the asymptotics of $\varphi$ given in Lemma \ref{lemma:asymptotics} and equation (\ref{Ctildes}), the above expression can be simplified to give
	\begin{equation}\label{I2}
	\begin{aligned}
	\Delta S_{\boldsymbol{m}}^{(2)}[\varphi]
	= -8\pi\sqrt{-1}\hspace{1mm}\frac{1}{m}s\log\frac{1}{|1-w_{s+s'-3}|}.
	\end{aligned}
	\end{equation}
	By identical calculation to that of $\Delta S_{\boldsymbol{m}}^{(2)}[\varphi]$, we get:
	\begin{equation}\label{I3}
	\begin{aligned}
	\Delta S_{\boldsymbol{m}}^{(3)}[\varphi]
	=-8\pi\sqrt{-1}\hspace{1mm}\frac{1}{m'}(s'-1)\log\frac{1}{|1-w_{s+s'-3}|}.
	\end{aligned}
	\end{equation}
	Finally, 
	\begin{equation*}
	\begin{aligned}
	\Delta S_{\boldsymbol{m}}^{(4)}[\varphi]&=\oint_{\tilde{C}_{s+s'}^{\epsilon}} \tilde{\varphi} \left(\frac{\dd{(\overline{\gamma w})}}{\overline{\gamma w}} - \frac{\dd{(\gamma w)}}{\gamma w}\right)-\oint_{C_{s+s'}^{\epsilon}} \varphi \left(\frac{\dd{\bar{w}}}{\bar{w}} - \frac{\dd{w}}{w }\right)\\
	&=\oint_{\tilde{C}_{s+s'}^{\epsilon}} (\varphi-\log|\gamma'|^2) \left(\frac{\bar{ \gamma'}\dd{ \bar{ w}}}{\overline{\gamma w}} - \frac{\gamma'\dd{w}}{\gamma w}\right)-\oint_{C_{s+s'}^{\epsilon}} \varphi \left(\frac{\dd{\bar{w}}}{\bar{w}} - \frac{\dd{w}}{w }\right)\\
	&=\oint_{\tilde{C}_{s+s'}^{\epsilon}} (\varphi+2\log|1-w_{s+s'-3}|) \left(\frac{\dd{ \bar{ w}}}{\bar{ w}} - \frac{\dd{w}}{ w}\right)-\oint_{C_{s+s'}^{\epsilon}} \varphi \left(\frac{\dd{\bar{w}}}{\bar{w}} - \frac{\dd{w}}{w }\right),
	\end{aligned}
	\end{equation*}
	where in the last line we approximated $w-w_{s+s'-3}$ with $w$ in the denominator of the second term due to $w$ having a large absolute value. Again, using the asymptotics of $\varphi$ from Lemma \ref{lemma:asymptotics}, Eq.(\ref{Ctildes}) and taking into account the opposite orientation of contours around the point at infinity, we have:
	\begin{equation}
	\begin{aligned}\label{I4}
	\Delta S_{\boldsymbol{m}}^{(4)}[\varphi]
	=-8\pi\sqrt{-1}\hspace{1mm}\frac{1}{m'}\log \frac{1}{|1-w_{s+s'-3}|}.
	\end{aligned}
	\end{equation}
	Putting \eqref{variationSsymm}, \eqref{I1-2}, \eqref{I2}, \eqref{I3} and \eqref{I4} together, we get
	\begin{multline}\label{variationSfinal}
	\Delta S_{\boldsymbol{m}}[\varphi]
	=2\pi\left(\frac{1}{m^2}-1\right)s\log|1-w_{s+s'-3}|+2\pi\left(\frac{1}{m'^2}-1\right)(s'-2)\log|1-w_{s+s'-3}|.
	\end{multline}
	By comparing (\ref{variationSfinal}) with equation (\ref{cocyclescones}), we can readily see that for $i=s'-3$: 
	\begin{equation}\label{modularanomaly}
	\Delta S_{\boldsymbol{m}}[\varphi]=S_{\boldsymbol{m}}[\varphi]\circ (\sigma_{k,k+1},\sigma'_{i,i+1})-S_{\boldsymbol{m}}[\varphi]=-2\pi\log |f_{(\sigma_{k,k+1},\sigma'_{i,i+1})}|.
	\end{equation}
	Doing the analogous calculations for $i=s'-2,s'-3$ would yield similar results. One can continue this proof for the case with the direct product of the symmetric groups of more strata inductively. Furthermore, had we chosen to fix the points $0,1$ and $\infty$ in the other stratum or decided to deal with punctures, we would still get the same result with minor changes in the path of the proof. So, in principle, we proved that under the action of any element $\eta$ of $\symm{\sigtype}$, $S_{\boldsymbol{m}}$ transforms according to the rule 
	\begin{equation}
	\exp[S_{\boldsymbol{m}}\circ \eta/\pi]\hspace{1mm}|f_\eta|^2=\exp[S_{\boldsymbol{m}}/\pi].
	\end{equation}
	This means that $\exp[S_{\boldsymbol{m}}/\pi]$ is a Hermitian metric in the holomorphic $\mathbb{Q}$-line bundle $\lambda_{0,\boldsymbol{m}}$ defined in the Section \ref{subsec:moduli}. As we mentioned before, (\ref{modularanomaly}) shows that 1-cocycles can be viewed as the modular anomaly caused by the non-covariance of the action under the effect of the modular group. To complete the proof, we remind that the first Chern form $\chern{\lambda_{0,\boldsymbol{m}}}{\exp[S_{\boldsymbol{m}}/\pi]}$ of the metrized $\mathbb{Q}$-line bundle $\lambda_{0,\boldsymbol{m}}$ with a metric $\exp[S_{\boldsymbol{m}}/\pi]$ is given by  
	\begin{equation}
	\chern{\lambda_{0,\boldsymbol{m}}}{\exp[S_{\boldsymbol{m}}/\pi]}=-\frac{\sqrt{-1}}{2\pi}\hspace{1mm}\partial\bar{\partial} \left(\frac{S_{\boldsymbol{m}}}{\pi}\right).
	\end{equation}
	Thus using Theorem \ref{theorem2}, we find that
	\begin{equation}
	\chern{\lambda_{0,\boldsymbol{m}}}{\exp[S_{\boldsymbol{m}}/\pi]} = \frac{1}{\pi^2} \omega_{\text{WP}}.
	\end{equation}
	This statement completes the proof.
\end{proof}
\subsection{Riemann Orbisurfaces of Genus >1}\label{gg1}
Consider a marked normalized Schottky group of rank $g>1$ denoted by $(\Sigma; L_1, \dots, L_g)$. The Liouville action, at the classical level, is actually the on-shell value of the Liouville action functional. For the case of a closed Riemann surface with $g>1$ it was first defined by Zograf and Takhtajan \cite{1988SbMat..60..297Z} (and was later interpreted by Takhtajan and Teo \cite{Takhtajan:2002cc} in cohomological language) to be given by
\begin{equation}\label{Liouvilleactiononlygenus}
S[\varphi] = \iint_{\SchottkyFund} (|\partial_w \varphi|^2 + e^{\varphi}) \dd[2]{w} + \frac{\sqrt{-1}}{2} \sum_{k=2}^{g} \oint_{C_k} \theta_{L_k^{-1}}(\varphi),
\end{equation}
where the 1-form $\theta_{L_k^{-1}}(\varphi)$ is given by
\begin{equation}
\theta_{L_k^{-1}}(\varphi) = \left(\varphi - \frac{1}{2} \log|L_k ' |^2 - \log |l_k|^2\right)\left(\frac{L_k ''}{L_k'} \dd{w} - \frac{\overline{L_k ''}}{\overline{L_k'}} \dd{\bar{w}}\right), \qquad \forall \hspace{.5mm}L_k \in (\Sigma; L_1, \dots, L_g).
\end{equation}
In the above equations, $\SchottkyFund$ is the fundamental domain of the marked normalized Schottky group $(\Sigma; L_1, \dots, L_g)$, $\partial \SchottkyFund = \bigcup_{k=1}^{g} (C_k \cup C_k')$, and
\begin{equation}
l_k = \frac{1-\lambda_k}{\sqrt{\lambda_k}(a_k -b_k)}
\end{equation}
is the left-hand lower element in the matrix representation of the generator $L_k \in \PSLC$ for $k=2,\dots,g$. Notice that since we have chosen the marked Schottky group to be normalized, in particular $a_1=0$ and $b_1=\infty$, we have $l_1 =0$ and $\theta_{L_1^{-1}}(\varphi)= 0$. Zograf and Takhtajan \cite[Theorems~1,2]{1988SbMat..60..297Z} have proven that
\begin{equation}
\partial S = - 2 \pi \sum_{i=1}^{3g-3} c_i P_i \qquad \text{and} \qquad \bar{\partial} \partial S = - 2 \sqrt{-1} \omega_{WP},
\end{equation}
where  $\partial$ and $\bar{\partial}$ are $(1,0)$ and $(0,1)$ components of the  de~Rham differential on $\schottky_{g}$ (see Section \ref{subsec:Schottkyspace}).
This implies that $-S$ is actually a K\"{a}hler potential for the projection of Weil-Petersson metric on $\schottky_{g}$.
\begin{remark}
	The addition of the second term in the on-shell Liouville action \eqref{Liouvilleactiononlygenus} makes sure that:
	\begin{itemize}
		\item The variation $\delta S$ of the classical action $S[\varphi]$ has the form
		\begin{equation}
		\delta S[\varphi] \equidef \lim_{t \to 0} \frac{S[\varphi + t \delta \psi] - S[\varphi]}{t} = \iint_{\SchottkyFund} (-2 \partial_{\bar{w}} \partial_w  \varphi + e^{\varphi}) \delta \psi \dd[2]{w}, 
		\end{equation}
		where the variation $\delta \psi \in \mathcal{C}^{\infty}(\Omega,\mathbb{R})$ is a smooth function on $\Omega$ which is automorphic with respect to $\Sigma$.
		For the marked Schottky group $(\Sigma; L_1, \dots, L_g)$, the  $S[\varphi]$ is \emph{independent} of the specific choice of a fundamental domain $\SchottkyFund$.\footnote{See footnote~\ref{nonuniqueschottkyfund}.}
	\end{itemize}
\end{remark}
When the conical points and cusps are present, the area integral in \eqref{Liouvilleactiononlygenus} diverges in the limits  $w \to w_i$ due to asymptotics of Liouville field $\varphi$ (see Lemma \ref{lemma:asymptotics}) and the classical Liouville action needs to be regularized. Let us define $\singfund$ as the pair $(\SchottkyFund_0 , \widetilde{\brdiv}|_{_{\SchottkyFund}})$ where $\SchottkyFund_0$ is defined as $\SchottkyFund \cap \Omega_0$ and $\widetilde{\brdiv}|_{_{\SchottkyFund}}$ denotes the restriction of $\widetilde{\brdiv}$ to $\SchottkyFund$. Here, we have assumed that all singular points $w_1,\dots,w_{n}$ belong to the interior of fundamental domain $\SchottkyFund$. For sufficiently small $\epsilon >0$, define\footnote{Note that $\lim_{\epsilon\rightarrow 0}\singfund_{\epsilon}  =  \SchottkyFund_{\text{reg}}$ where $\SchottkyFund_{\text{reg}}:= \SchottkyFund_0 \backslash \operatorname{Supp}(\widetilde{\brdiv}|_{_{\SchottkyFund}})$.}
\begin{equation}
\singfund_{\epsilon} = \singfund \big\backslash \bigcup_{i=1}^{n} D_i^{\epsilon},
\end{equation}
with $D_i^{\epsilon} \equidef \left\{w \,\big| \, |w-w_i| < \epsilon \right\}$. It follows from Lemma \ref{lemma:asymptotics} that the following limit exists
\begin{multline}\label{Sb}
S_{\singfund_{\text{reg}}}[\varphi] = \\
\lim_{\epsilon \to 0^{+}} \left(\iint_{\singfund_{\epsilon}}(|\partial_w \varphi|^2 + e^{2\varphi}) \dd[2]{w}  + \frac{\sqrt{-1}}{2} \sum_{j=1}^{n_e} \left(1-\frac{1}{m_j}\right) \oint_{C_{j}^{\epsilon}} \varphi \left(\frac{\dd{\bar{w}}}{\bar{w}-\bar{w}_{j}} - \frac{\dd{w}}{w - w_{j}}\right)\right. \\ 
\left. - 2\pi \sum_{j=1}^{n_e} \left(1-\frac{1}{m_j}\right)^2 \log\epsilon  +  2 \pi n_p \big(\log\epsilon + 2 \log|\log\epsilon|\big) \right).
\end{multline}
\begin{remark}
	When $n_p = 0$, the appropriate $S_{\singfund_{\text{reg}}}[\varphi]$ is given by 
	\begin{multline}\label{CompactSb}
	\lim_{\epsilon \to 0^{+}} \Bigg(\iint_{\singfund_{\epsilon}}(|\partial_w \varphi|^2 + e^{2\varphi}) \dd[2]{w} + \frac{\sqrt{-1}}{2} \sum_{j=1}^{n} \left(1-\frac{1}{m_j}\right) \oint_{C_{j}^{\epsilon}} \varphi \left(\frac{\dd{\bar{w}}}{\bar{w}-\bar{w}_{j}} - \frac{\dd{w}}{w - w_{j}}\right) \\ 
	- 2\pi \sum_{i=1}^{n} \left(1-\frac{1}{m_j}\right)^2 \log\epsilon\Bigg).
	\end{multline}
\end{remark}
\noindent Now, we can define the regularized action as 
\begin{equation}\label{regularizeLiouvilleaction}
\begin{split}
&S_{\boldsymbol{m}}[\varphi] = S_{\boldsymbol{m}}(\SchottkyFund;w_1,\dots,w_{n}) 
= S_{\singfund_{\text{reg}}}[\varphi] + \frac{\sqrt{-1}}{2} \sum_{k=2}^{g} \oint_{C_k} \theta_{L_k^{-1}}(\varphi).
\end{split}
\end{equation}
This completes the definition of $S_{\boldsymbol{m}}$ provided that all of the fixed points $w_1,\dots,w_{n}$ lie in the interior of the fundamental domain --- i.e. $w_1,\dots,w_{n} \in \operatorname{Int} \SchottkyFund$. The $S_{\boldsymbol{m}}(\SchottkyFund;w_1,\dots,w_{n})$ depends on the choice of representatives in $\Sigma \cdot \{w_1,\dots,w_{n}\}$ and \emph{no longer determines a function on the Schottky space $\schottky_{g,n}(\boldsymbol{m})$}. Note that $w_i$ and $L_{k}(w_i)$ have the same order and are related by the action of $\symm{\sigtype}$ acting on $\singrigon$. The geometric meaning of $S_{\boldsymbol{m}}$ is given by the following Lemma (also see Lemma~\ref{hihermitianmetric}):
\begin{lemma}\label{Soverpihermitianmetric}
	The regularized Liouville action determines a Hermitian metric $\exp [S_{\boldsymbol{m}}/\pi]$ in the holomorphic $\mathbb{Q}$-line bundle $\linebundle = \bigotimes_{i=1}^{n} \linebundle_i^{h_i}$ over $\schottky_{g,n}(\boldsymbol{m})$.
\end{lemma}
\begin{proof}
	To establish this claim for $i=1,\dots, n$, it is adequate to demonstrate that we have
	\begin{equation}\label{ext02}
	S_{\boldsymbol{m}}(\tilde{\SchottkyFund};w_1,\dots,L_k w_i, \dots,w_{n}) - S_{\boldsymbol{m}}(\SchottkyFund;w_1,\dots,w_{n}) = \pi h_i\hspace{.5mm}\log|L_k'(w_i)|^2, 
	\end{equation}
	where $w_1,\dots,w_{n} \in \operatorname{Int}\SchottkyFund$ and $w_1,\dots,L_k w_i, \dots,w_{n} \in \operatorname{Int}\tilde{\SchottkyFund}$. Furthermore, it is sufficient to consider the case when
	\begin{equation*}
	\tilde{\SchottkyFund} = (\SchottkyFund \backslash \SchottkyFund_0) \cup L_k(\SchottkyFund_0),
	\end{equation*}
	and $\SchottkyFund_0 \subset \SchottkyFund$ is such that $\partial \SchottkyFund_0 \cap \partial \SchottkyFund \subset C_k$ and $w_i \in \SchottkyFund_0$, and all other $w_j \in \SchottkyFund \backslash \SchottkyFund_0$ for $j \neq i$. Note that by a finite combination of such transformations, any choice of a fundamental domain for $\Sigma$ can be obtained from $\SchottkyFund$. The computation of \eqref{ext02} closely follows the corresponding computation in the proof of Lemma 3 in \cite{park2015potentials} where no branch point were present. Let
	\begin{equation*}
	\begin{split}
	I_{\epsilon}(\SchottkyFund;w_1,\dots,w_{n}) & = \iint_{\singfund_{\epsilon}} (|\partial_w \varphi|^2 + e^{\varphi}) \dd{w} \wedge \dd{\bar{w}} + \sum_{k=2}^{g} \oint_{C_k} \theta_{L_k^{-1}}(\varphi) \\
	&+ \sum_{j=1}^{n_e} \left(1-\frac{1}{m_j}\right) \oint_{C_{j}^{\epsilon}} \varphi \left(\frac{\dd{\bar{w}}}{\bar{w}-\bar{w}_{j}} - \frac{\dd{w}}{w - w_{j}}\right),
	\end{split}
	\end{equation*}
	where we did not include the counterterms in the action due to the fact that the action of $L_k$ does not change the conformal class. Since $\tilde{C}_j = C_j$ for $j \neq k$ and $\tilde{C}_k = C_k - \partial\SchottkyFund_0$, we have
	\begin{equation*}
	\begin{split}
	\Delta I_{\epsilon} &= I_{\epsilon}(\tilde{\SchottkyFund};w_1,\dots,L_k w_i, \dots,w_{n}) -  I_{\epsilon}(\SchottkyFund;w_1,\dots,w_{n})\\
	&= \iint_{L_k(\SchottkyFund_0)\backslash \tilde{D}_i^{\epsilon}} (|\partial_w \varphi|^2 + e^{\varphi}) \dd{w} \wedge \dd{\bar{w}} - \iint_{\SchottkyFund_0\backslash D_i^{\epsilon}} (|\partial_w \varphi|^2 + e^{\varphi}) \dd{w} \wedge \dd{\bar{w}} - \oint_{\partial \SchottkyFund_0} \theta_{L_k^{-1}}(\varphi) \\
	&+\sum_{j=1}^{n_e} \delta_{i\,j} \left(1-\frac{1}{m_j}\right) \left[ \oint_{\tilde{C}_{j}^{\epsilon}} \varphi \left(\frac{\dd{\bar{w}}}{\bar{ w}-\overline{L_k w_{j}}} - \frac{\dd{w}}{ w - L_k w_{j}}\right) - \oint_{C_{j}^{\epsilon}} \varphi \left(\frac{\dd{\bar{w}}}{\bar{w}-\bar{w}_{j}} - \frac{\dd{w}}{w - w_{j}}\right) \right].
	\end{split}
	\end{equation*}
	According to Eq.\eqref{Lonphi} we have
	\begin{equation}
	\begin{split}
	L_k^{\ast}\left((|\partial_w \varphi|^2 + e^{\varphi}) \dd{w} \wedge \dd{\bar{w}}\right) & = \left((|\partial_w \varphi|^2 + e^{\varphi}) \dd{w} \wedge \dd{\bar{w}}\right)\circ L_k \, |L_k'|^2 \\
	& = (|\partial_w \varphi|^2 + e^{\varphi}) \dd{w} \wedge \dd{\bar{w}} + \dd{\theta_{L_k^{-1}}(\varphi)},
	\end{split}
	\end{equation}
	which, together with using the  Stokes theorem, we get
	\begin{equation*}
	\begin{split}
	\Delta I_{\epsilon} &=\iint_{\SchottkyFund_0\backslash L_k^{-1}(\tilde{D}_i^{\epsilon})} (|\partial_w \varphi|^2 + e^{\varphi}) \dd{w} \wedge \dd{\bar{w}} - \iint_{\SchottkyFund_0\backslash D_i^{\epsilon}} (|\partial_w \varphi|^2 + e^{\varphi}) \dd{w} \wedge \dd{\bar{w}} - \oint_{\partial L_k^{-1}(\tilde{D}_i^{\epsilon})} \theta_{L_k^{-1}}(\varphi) \\
	& + \sum_{j=1}^{n_e} \delta_{i\,j} \left(1-\frac{1}{m_j}\right) \left[ \oint_{L_k^{-1}(\tilde{C}_{j}^{\epsilon})} (\varphi-\log|L'_{k}w|^2) \left(\frac{\overline{L'_kw}\dd{\bar{w}}}{\overline{L_k w}-\overline{L_k w_{j}}} - \frac{L'_kw\dd{w}}{L_k w - L_k w_{j}}\right) \right.\\&\left.
	\hspace{2cm}- \oint_{C_{j}^{\epsilon}} \varphi \left(\frac{\dd{\bar{w}}}{\bar{w}-\bar{w}_{j}} - \frac{\dd{w}}{w - w_{j}}\right) \right].
	\end{split}
	\end{equation*}
	Since the third term's integrand does not have a pole, its contribution will be of $\Sorder{1}$ and it can be safely omitted. Furthermore, the exponential term in the Liouville term gives the Euler characteristic and  its contribution will be canceled between related terms. Now, we can rewrite the above equation as
	\begin{equation}\label{schottkyvar1}
	\begin{split}
	\Delta I_{\epsilon} &=\iint_{D_i^{\epsilon} }(|\partial_w \varphi|^2 ) \dd{w} \wedge \dd{\bar{w}} - \iint_{L_k^{-1}(\tilde{D}_i^{\epsilon})} (|\partial_w \varphi|^2) \dd{w} \wedge \dd{\bar{w}} \\
	& +\sum_{j=1}^{n_e} \delta_{i\,j} \left(1-\frac{1}{m_j}\right) \left[ \oint_{L_k^{-1}(\tilde{C}_{j}^{\epsilon})} (\varphi-\log|L'_{k}(w)|^2) \left(\partial_{\bar{w}}\log|L_kw-L_kw_{j}|^2 \right.\right.\\&\left.\left.- \partial_{w}\log|L_kw-L_kw_{j}|^2\right) - \oint_{C_{j}^{\epsilon}} \varphi \left(\frac{\dd{\bar{w}}}{\bar{w}-\bar{w}_{j}} - \frac{\dd{w}}{w - w_{j}}\right) \right].
	\end{split}
	\end{equation}
	By noting that
	\begin{equation*}
	\begin{aligned}
	\partial_{w}\log|L_kw-L_kw_{j}|^2
	=\partial_w\log|w-w_{j}|^2+\Sorder{1},
	\end{aligned}
	\end{equation*}
	we can write (\ref{schottkyvar1}) as follows
	\begin{equation}\label{schottkyvar2}
	\begin{split}
	\Delta I_{\epsilon} &=\iint_{D_i^{\epsilon} }(|\partial_w \varphi|^2 ) \dd{w} \wedge \dd{\bar{w}} - \iint_{L_k^{-1}(\tilde{D}_i^{\epsilon})} (|\partial_w \varphi|^2) \dd{w} \wedge \dd{\bar{w}} \\
	& + \sum_{j=1}^{n_e} \delta_{i\,j} \left(1-\frac{1}{m_j}\right) \left[ \oint_{L_k^{-1}(\tilde{C}_{j}^{\epsilon})} (\varphi-\log|L'_{k}(w)|^2)\left(\frac{\dd{\bar{w}}}{\bar{w}-\bar{w}_{j}} - \frac{\dd{w}}{w - w_{j}}\right)\right.\\&\left.\hspace{2cm}- \oint_{C_{j}^{\epsilon}} \varphi \left(\frac{\dd{\bar{w}}}{\bar{w}-\bar{w}_{j}} - \frac{\dd{w}}{w - w_{j}}\right) \right].
	\end{split}
	\end{equation}
	By doing the integration by parts and imposing the equations of motion together with considering the orientation of the boundary circles, the first line in \eqref{schottkyvar2} becomes
	\begin{equation}\label{schottkybulk}
	\begin{aligned}
	&\iint_{D_i^{\epsilon} }(|\partial_w \varphi|^2 ) \dd{w} \wedge \dd{\bar{w}} - \iint_{L_k^{-1}(\tilde{D}_i^{\epsilon})} (|\partial_w \varphi|^2) \dd{w} \wedge \dd{\bar{w}}=\oint_{C_{i}^{\epsilon}}\varphi\partial_{\bar{w}}\varphi\dd\bar{ w}-\oint_{\tilde{C}_{i}^{\epsilon}}\varphi\partial_{\bar{w}}\varphi\dd\bar{w},
	\end{aligned}
	\end{equation}
	where
	\begin{equation}\label{schottkycircs}
	\begin{aligned}
	\tilde{C}_{i}^{\epsilon}=\{w\big{|}\:|L_kw-L_kw_{i}|=\epsilon\}
	\equiv\{w\big{|}\:|w-w_{i}|=\frac{\epsilon}{L'_kw_{i}}\}.
	\end{aligned}
	\end{equation}
	To be more precise, the right hand side of \eqref{schottkybulk} had an extra contribution
	\begin{equation}\label{extra3}
	\frac{1}{2}\iint_{L_k^{-1}(\tilde{D}_i^{\epsilon})}\varphi e^{\varphi}\dd{w} \wedge \dd{\bar{w}} -\frac{1}{2}\iint_{D_i^{\epsilon} } \varphi e^{\varphi} \dd{w} \wedge \dd{\bar{w}},
	\end{equation}
	which by defining $y_i = |w-w_i|,\hspace{.5mm}\alpha_i = 1-1/m_i, [y_i]=: \epsilon \leq y_i \leq \epsilon\slash L'_kw_{i}$, this contribution for branch points and cusps respectively becomes
	\begin{equation*}
	\begin{split}
	&-\frac{2\pi}{2}\iint_{[y_i]}\frac{1}{y_i^{2\alpha_{i}-1}}\log{y_i^{2\alpha_{i}}}dy_i = \pi \left(\frac{\alpha_{i}}{2(1-\alpha_i)^{2}}y_{i}^{2-2\alpha_i}-\frac{\alpha_i}{(1-\alpha_{i})}y_{i}^{2-2\alpha_{i}}\log y_i\right)\bigg{|}^{\epsilon\slash L'_kw_{i}}_{\epsilon}
	\\
	&-\hspace{-.5mm}\frac{2\pi}{2}\hspace{-1mm}\iint_{[y_i]}\frac{1}{y_i\log^2 y_i}\log\left({y_i^{2}}\log^{2}y_i\right)\hspace{-.5mm} dy_i\hspace{-.5mm}=\hspace{-.5mm} \pi\hspace{-.5mm}\left(2\hspace{-.5mm}+\hspace{-.5mm}\frac{2}{\log y_i}\hspace{-.5mm}-\hspace{-.5mm}2\log|\log y_i|\hspace{-.5mm}+\hspace{-.5mm}\frac{2}{\log y_i}\log|\log y_i|\right)\hspace{-1mm}\bigg{|}^{\epsilon\slash L'_kw_{i}}_{\epsilon}
	\end{split}
	\end{equation*}
	but both of them, and accordingly \eqref{extra3}, vanish in the limit $\epsilon \rightarrow 0$. Now, by using the asymptotic form of $\varphi$ from the Lemma \ref{lemma:asymptotics} together with the relation \eqref{schottkycircs}, the Eq.\eqref{schottkybulk} is simplified to\footnote{The $D_i^{\epsilon}$'s boundaries have the opposite orientation of the boundary of the fundamental domain.}
	\begin{equation}\label{extra2}
	\begin{aligned}
	\iint_{D_i^{\epsilon} }(|\partial_w \varphi|^2 ) \dd{w} \wedge \dd{\bar{w}} - \iint_{L_k^{-1}(\tilde{D}_i^{\epsilon})} (|\partial_w \varphi|^2) \dd{w} \wedge \dd{\bar{w}}=-4\pi\sqrt{-1}\left(1-\frac{1}{m_i}\right)^2\log| L'_kw_{i}|.
	\end{aligned}
	\end{equation}
	For the rest of the integrals in (\ref{schottkyvar2}), again by using the asymptotics of $\varphi$ from Lemma \ref{lemma:asymptotics} and (\ref{schottkycircs}) we have 
	\begin{equation}
	\begin{aligned}\label{schottkyboundary}
	&\sum_{j=1}^{n_e} \delta_{ij} \left(1-\frac{1}{m_j}\right) \left[ \oint_{L_k^{-1}(\tilde{C}_{j}^{\epsilon})} (\varphi-\log|L'_{k}(w)|^2)\left(\frac{\dd{\bar{w}}}{\bar{w}-\bar{w}_{j}} - \frac{\dd{w}}{w - w_{j}}\right)\right.\\&\left.\hspace{1.3cm}-\oint_{C_{j}^{\epsilon}} \varphi \left(\frac{\dd{\bar{w}}}{\bar{w}-\bar{w}_{j}} - \frac{\dd{w}}{w - w_{j}}\right) \right]
	=8\pi\sqrt{-1}\left(\frac{1}{m_i^2}-\frac{1}{m_i}\right)\log|L'_kw_{i}|.
	\end{aligned}
	\end{equation}
	Now, combining \eqref{extra2} and \eqref{schottkyboundary} gives
	\begin{equation*}
	\Delta I_{\epsilon}=
	-4\pi\sqrt{-1}\hspace{1mm}h_i\log|L'_kw_{i}|.
	\end{equation*}
	So, in principle, we proved that
	\begin{equation*}
	\begin{aligned}
	\exp[S_{\boldsymbol{m}}(\tilde{\SchottkyFund};w_1,\dots,L_k w_i,  \dots,w_{n})/(\pi h_i)]=\exp[S_{\boldsymbol{m}}(\SchottkyFund;w_1,\dots,w_{n})/(\pi h_i)]\hspace{1mm}|L'_kw_{i}|^2,
	\end{aligned}
	\end{equation*}
	which means that $\exp[S_{\boldsymbol{m}}/\pi]$ is a Hermitian metric on $\linebundle =\bigotimes_{i=1}^{n_p+n_e} \linebundle_i^{h_i}$.
\end{proof}
\noindent Combining Lemmas \ref{Soverpihermitianmetric} and \ref{hihermitianmetric}, we can deduce the following statement:
\begin{corollary}\label{cor:invariantaction}
	Put $\mathsf{H} = \Lponetial_{1}^{m_1 h_1} \dotsm \Lponetial_{n_e}^{m_{n_e} h_{n_e}} \Lponetial_{n_e+1} \dotsm \Lponetial_{n}$. Then,
	\begin{equation}\label{curlyaction}
	\Gpotential_{\boldsymbol{m}} = S_{\boldsymbol{m}} - \pi \log \mathsf{H}
	\end{equation}
	determines a smooth real-valued function on $\schottky_{g,n}(\boldsymbol{m})$. 
\end{corollary}
\noindent The above form of $\Gpotential_{\boldsymbol{m}}$ can be easily understand by demanding the Liouville action to be independent from the choice of a fundamental domain: It is clear from commutative diagram~\eqref{globalcoords} and the definition of regularized Liouville action~\eqref{regularizeLiouvilleaction}, that the problem of defining the appropriate Liouville action on $\schottky_{g,n}(\boldsymbol{m})$ is closely related to the Fuchsian uniformization of $\singrigon$. Then, the fact that $\Omega \subset \hat{\cmpx}$ and the observation that, roughly speaking, the action of Schottky group $\Sigma$ on $\singrigon$ resembles that of $\symm{\sigtype}$ on a genus zero Riemann orbisurface,\footnote{More specifically, by acting each generator $L_k \in \Sigma$ on a singular point $w_i$ inside a particular fundamental domain $\SchottkyFund$, we will get another singular point with same order of isotropy in a different fundamental domain.} suggests that the same pattern of ``anomaly cancellation'' observed in section \ref{M0n} should also happen in this case.
\section{Potentials for Weil-Petersson and Takhtajan-Zograf Metrics}\label{sec:Potentials}
In this section, following \cite{park2015potentials,ZT_2018}, we construct K\"ahler potentials for cuspidal and elliptic TZ metrics on $\moduli_{0,n}$ (see Section \ref{KahlerMetrics}). We will also prove that the first Chern forms of the line bundles $\linebundle_i$ over the Schottky space $\schottky_{g,n}(\boldsymbol{m})$ with Hermitian metrics $\Lponetial_{i}$ are given by $\tfrac{1}{2\pi} \omega^{\text{ell}}_{\text{TZ},i}$ for $i=1,\dots,n_e$ and $\tfrac{4}{3} \omega^{\text{cusp}}_{\text{TZ},i}$ for $i=n_e+1,\dots,n$. In addition, we will show that $\tfrac{1}{\pi^2} \omega_{_{\text{WP}}}$ is the first Chern form of the $\mathbb{Q}$-line bundle $\linebundle = \bigotimes_{i=1}^{n} \linebundle_i^{h_i}$ with Hermitian metric $\exp[S_{\boldsymbol{m}}/\pi]$, where the regularized classical Liouville action $S_{\boldsymbol{m}}$ is given by Eq.\eqref{regularizeLiouvilleaction}. Then, it follows readily from these two results that the specific combination $\omega_{_{\text{WP}}}- \tfrac{4\pi^2}{3} \omega_{_{\text{TZ}}}^{\text{cusp}}-\tfrac{\pi}{2}\sum_{j=1}^{n_e} m_j h_j\hspace{.5mm} \omega^{\text{ell}}_{\text{TZ},j}$ of Weil-Petersson metric as well as cuspidal and elliptic Takhtajan-Zograf metrics has a \emph{global} K\"{a}hler potential on $\schottky_{g,n}(\boldsymbol{m})$.

\subsection{Potentials for Cuspidal and Elliptic TZ Metrics on $\mathcal{M}_{0,n}$}
As in Section \ref{subsec:moduli}, let $\Gamma$ be a  marked normalized  Fuchsian group with signature $(0;m_1,\dots,m_{n_e}$

\hspace{-.7cm}$,n_p)$ that uniformizes the orbifold Riemann surface $O$ and let $J: \UHP \to O$ be the Klien's Hauptmodule. In addition, let $\Lponetial_{i} =  \left| J^{(i)}_{1} \right|^{\frac{2}{m_i}}$ for $i=1,\dots,n_e$, $\Lponetial_{i} = \left| J^{(i)}_{1} \right|^2$ for $i=n_e+1,\dots,n-1$, and $\Lponetial_n = \left| J^{(n)}_{-1} \right|^2$ be smooth positive functions on $\moduli_{0,n}$.\footnote{See Eq.\eqref{hi}.} Now, according to the expressions for $\log \Lponetial_{i}$ in  Remark~\ref{hasymptotics}, we prove the following lemma:
\begin{lemma}\label{lemma:widerivativeloghi}
	For all $k=1,\dots,n-3$, we have
	\begin{equation}
	\begin{split}
	&\pdv{w_k} \log \Lponetial_{i}=\frac{1}{m_i} \partial_{w}\dot{F}^k(w_i), \hspace{1cm} i=1,\dots,n_e,\\ \\
	&\pdv{w_k} \log \Lponetial_{i}= \partial_{w}\dot{F}^k(w_i), \hspace{1.5cm} i=n_e+1,\dots,n.
	\end{split}
	\end{equation}
\end{lemma}
\begin{proof}
	Consider the orbifold Riemann surface $O \cong [\UHP \slash \Gamma]$. Using Lemma \ref{Lemma:covering}, it is sufficient to demonstrate that
	\begin{equation}
	\left.\left(\pdv{\log \Lponetial_{i}^{\varepsilon \mu_k}}{\varepsilon}\right)\right|_{\varepsilon=0} = \left\{
	\begin{split}
	& \frac{1}{m_i} \partial_{w} \dot{F}^k (w_{i}) &\text{for} \quad & i=1,\dots,n_e,\\ \\
	&  \partial_{w}\dot{F}^k(w_{i})\quad &\text{for} \quad &i=n_e+1,\dots,n,
	\end{split}
	\right.
	\end{equation}
	and for all $k=1,\dots,n-3$. The following proof repeats verbatim the proof of Lemma 4  in \cite{park2015potentials} for the case of punctures. Using the fact that $F^{\varepsilon \mu_k}$ is holomorphic in $\varepsilon$ at $\varepsilon=0$, Corollary~\ref{corollary:FdotiAsymp} and formulas in Remark~\ref{hasymptotics}, Eq.\eqref{expphivariations} we get:
	\begin{itemize}
		\item \emph{$i=1,\dots,n_e$ case:}
		\begin{equation*}
		\begin{split}
		\hspace{-.5cm}\left.\left(\pdv{\log \Lponetial_{i}^{\varepsilon \mu_k}}{\varepsilon}\right)\right|_{\varepsilon=0} & =-\left.\left(\pdv{\varepsilon}\right)\right|_{\varepsilon=0}\lim_{w \to w_{i}}  \ \left(\varphi^{\varepsilon \mu_k} \circ F^{\varepsilon \mu_k} + \left(1-\frac{1}{m_i}\right) \log|F^{\varepsilon \mu_k}(w-w_{i})|^2\right)
		\\
		&\hspace{-1cm}= -\lim_{w \to w_{i}} \left\{ \left.\pdv{\varepsilon}\right|_{\varepsilon=0} \left(\varphi^{\varepsilon \mu_k} \circ F^{\varepsilon \mu_k} + \left(1-\frac{1}{m_i}\right) \log|F^{\varepsilon \mu_k}(w-w_{i})|^2\right)\right\} \\
		&\hspace{-1cm} = -\lim_{w \to w_{i}} \left(\partial_{w_{k}} \varphi + \partial_{w} \varphi \dot{F}^k(w) + \left(1-\frac{1}{m_i}\right) \frac{\dot{F}^k(w) - \dot{F}^k(w_{i})}{w-w_{i}}\right)\\
		&\hspace{-1.2cm} \overset{\eqref{widerivative}}{=}  -\lim_{w \to w_{i}} \left(- \partial_{w} \varphi \dot{F}^k(w) - \partial_{w} \dot{F}^k (w) + \partial_{w} \varphi \dot{F}^k(w) + \left(1-\frac{1}{m_i}\right) \partial_{w} \dot{F}^k (w) \right)\\
		&\hspace{-1cm}= \frac{1}{m_i} \partial_{w} \dot{F}^k (w_{i}). 
		\end{split}
		\end{equation*}
		In going from the first line to the second line, we have interchanged the order of the limit $w \to w_i$ and the differentiation. This is allowed due to convergence in the above formula and the fact that the definition of $\Lponetial_{i}$ is uniform in a neighborhood of an arbitrary point $(w_1,\dots,w_{n-3}) \in \moduli_{0,n}$.
		\item \emph{$i=n_e+1,\dots,n-1$ case:}
		\begin{equation*}
		\begin{split}
		\hspace{-.4cm}\left.\left(\pdv{\log \Lponetial_{i}^{\varepsilon \mu_k}}{\varepsilon}\right)\right|_{\varepsilon=0} & = \left.\left(\pdv{\varepsilon}\right)\right|_{\varepsilon=0} \lim_{w \to w_i} \left(\log|F^{\varepsilon \mu_k}(w-w_i)|^2 - \frac{2 e^{-\frac{\varphi^{\varepsilon \mu_k} \circ F^{\varepsilon \mu_k}(w)}{2}}}{|F^{\varepsilon \mu_k}(w-w_i)|}\right) \\
		&\hspace{-1.5cm}= \lim_{w \to w_i} \left\{ \left.\pdv{\varepsilon}\right|_{\varepsilon=0} \left(\log|F^{\varepsilon \mu_k}(w-w_i)|^2 - 2 (F^{\varepsilon \mu_k})^{\ast}\left(e^{-\frac{1}{2}\varphi^{\varepsilon \mu_k}}\right)\left|\frac{\partial_{w}F^{\varepsilon \mu_k}(w)}{F^{\varepsilon \mu_k}(w-w_i)}\right|\right)\right\}\\
		&\hspace{-1.5cm} = \lim_{w \to w_i} \left\{ \left.\pdv{\varepsilon}\right|_{\varepsilon=0}\log(F^{\varepsilon \mu_k}(w-w_i)) - 2 \frac{e^{-\frac{1}{2}\varphi}}{(\bar{w} - \bar{w}_i)^{\frac{1}{2}}}\left. \pdv{\varepsilon}\right|_{\varepsilon=0}\left(\frac{\partial_{w}F^{\varepsilon \mu_k}(w)}{F^{\varepsilon \mu_k}(w-w_i)}\right)^{\frac{1}{2}}\right\}\\
		&\hspace{-1.5cm} = \lim_{w \to w_i} \left(\frac{\dot{F}^k(w) - \dot{F}^k(w_i)}{w-w_i} - \frac{e^{-\frac{1}{2}\varphi(w)} \left( (w - w_i) \partial_{w} \dot{F}^k(w)- \dot{F}^k(w) + \dot{F}^k(w_i)\right)}{(w-w_i)|w-w_i|}\right) \\
		&\hspace{-1.5cm} = \partial_{w} \dot{F}^k(w_i).
		\end{split}
		\end{equation*}
		\item \emph{$i=n$ case:}
		\begin{equation*}
		\begin{split}
		\hspace{-.3cm}\left.\left(\pdv{\log \Lponetial_{n}^{\varepsilon \mu_k}}{\varepsilon}\right)\right|_{\varepsilon=0} & = \left.\left(\pdv{\varepsilon}\right)\right|_{\varepsilon=0} \lim_{w \to \infty} \left(\log|F^{\varepsilon \mu_k}(w)|^2 - \frac{2 e^{-\frac{\varphi^{\varepsilon \mu_k} \circ F^{\varepsilon \mu_k}(w)}{2}}}{|F^{\varepsilon \mu_k}(w)|}\right) \\
		&\hspace{-1.5cm} = \lim_{w \to \infty} \left\{ \left.\pdv{\varepsilon}\right|_{\varepsilon=0} \left(\log|F^{\varepsilon \mu_k}|^2 - 2 (F^{\varepsilon \mu_k})^{\ast}\left(e^{-\frac{1}{2}\varphi^{\varepsilon \mu_k}}\right)\left|\frac{\partial_{w}F^{\varepsilon \mu_k}}{F^{\varepsilon \mu_k}}\right|\right)(w)\right\}\\
		&\hspace{-1.5cm} = \lim_{w \to \infty} \left\{ \left.\pdv{\varepsilon}\right|_{\varepsilon=0}\log(F^{\varepsilon \mu_k}(w)\overline{F^{\varepsilon \mu_k}(w)}) - 2 \frac{e^{-\frac{1}{2}\varphi}}{\bar{w}^{\frac{1}{2}}}\left. \pdv{\varepsilon}\right|_{\varepsilon=0}\left(\frac{\partial_{w}F^{\varepsilon \mu_k}(w)}{F^{\varepsilon \mu_k}(w)}\right)^{\frac{1}{2}}\right\}\\
		&\hspace{-1.5cm} = \lim_{w \to \infty} \left(\frac{\dot{F}^k(w)}{w} - \frac{e^{-\frac{1}{2}\varphi(w)} (w \partial_{w} \dot{F}^k(w)- \dot{F}^k(w))|w|}{w^2 \bar{w}}\right) \\
		&\hspace{-1.5cm} = \partial_{w} \dot{F}^k(\infty).
		\end{split}
		\end{equation*}
	\end{itemize}	
\end{proof}
As before, let $\partial$ and $\bar{\partial}$ be the $(1,0)$ and $(0,1)$ components of  de Rham differential $d=\partial + \bar{\partial}$ on $\moduli_{0,n}$. We have (see \cite{park2015potentials,ZT_2018}):
\begin{lemma*}[Takhtajan and Zograf]
	The functions $-\log \Lponetial_{i},-\log \Lponetial_{j},\log \Lponetial_{n}, : \moduli_{0,n}\to \mathbb{R}_{>0}$ for  $i=1,\dots,n_e$ and $j=n_e+1,\dots,n-1$ are K\"{a}hler potential for TZ metrics $\frac{1}{2}\langle \cdot , \cdot \rangle_{\text{TZ},i}^{\text{ell}}$ ,$\frac{4 \pi}{3} \langle \cdot , \cdot \rangle_{\text{TZ}, j}^{\text{cusp}}$, $-\frac{4 \pi}{3} \langle \cdot , \cdot \rangle_{\text{TZ}, n}^{\text{cusp}}$, respectively:\footnote{See the Proposition 1  in \cite{park2015potentials} for punctured Riemann surfaces.}
	\begin{equation}\label{lemma:localpotentialsonmoduliwnogenus}
	\bar{\partial} \partial \log \Lponetial_{i} = -\sqrt{-1}\hspace{.5mm} \omega_{\text{TZ},i}^{\text{ell}}, \qquad \bar{\partial} \partial \log \Lponetial_{j} = -\frac{8\pi \sqrt{-1}}{3}\hspace{.5mm} \omega_{\text{TZ},j}^{\text{cusp}}, \qquad \bar{\partial} \partial \log \Lponetial_{n} = \frac{8\pi \sqrt{-1}}{3}\hspace{.5mm} \omega_{\text{TZ},n}^{\text{cusp}}.
	\end{equation}
\end{lemma*}
\begin{proof}
	We need to prove that for all $j,k=1,\dots,n-3$,	
	\begin{equation*}
	-\pdv[2]{\log \Lponetial_{i}}{w_j}{\bar{w}_k} = \left\{
	\begin{split}
	&\frac{1}{2} \left\langle \pdv{w_j} , \pdv{w_k} \right\rangle_{\text{TZ},i}^{\text{ell}} & \text{for} \quad& i=1,\dots,n_e,\\ \\
	& \frac{4 \pi}{3} \left\langle \pdv{w_j} , \pdv{w_k} \right\rangle_{\text{TZ},i}^{\text{cusp}} & \text{for}\quad & i=n_e+1,\dots,n-1, \\ \\
	&- \frac{4 \pi}{3} \left\langle \pdv{w_j} , \pdv{w_k} \right\rangle_{\text{TZ},i}^{\text{cusp}} \quad & \text{for} \quad &i=n.
	\end{split}
	\right.
	\end{equation*}
	Let us consider the three cases $i=1,...,n_e$, $i=n_e+1,\dots,n-1$, and $i=n$ separately. The following proof repeats verbatim the proof of Proposition 1  in \cite{park2015potentials} for the case of punctures. According to Section \ref{subsec:teichmuller}, Lemma \ref{Lemma:covering} and Eq.\eqref{expphivariations}, for a given Riemann orbisurface $O  \cong [\UHP \slash \Gamma]$ one can write:
	\begin{itemize}
		\item \emph{$i=1,\dots,n_e$ case:}
		\begin{equation*}
		\begin{split}
		\hspace{-.5cm}-\pdv[2]{\log \Lponetial_{i}}{w_j}{\bar{w}_k} & = - \left.\left(\pdv[2]{\log \Lponetial_{i}^{\varepsilon_j \mu_j + \varepsilon_k \mu_k}}{\varepsilon_j}{\bar{\varepsilon}_k}\right)\right|_{\varepsilon_j = \varepsilon_k =0}\\
		&\hspace{-1.6cm}= \lim_{w \to w_{i}} \left\{ \left.\left(\pdv[2]{}{\varepsilon_j}{\bar{\varepsilon}_k}\right)\right|_{\varepsilon_j = \varepsilon_k =0} \left(\varphi^{(\varepsilon \mu)_{jk}} \circ F^{(\varepsilon \mu)_{jk}} + \left(1-\frac{1}{m_i}\right) \log|F^{(\varepsilon \mu)_{jk}}(w-w_{i})|^2\right)\right\} \\
		&\hspace{-1.5cm}= \lim_{w \to w_{i}} \left.\left(\pdv[2]{}{\varepsilon_j}{\bar{\varepsilon}_k}\right)\right|_{\varepsilon_j = \varepsilon_k =0} \left(F^{\varepsilon_j \mu_j + \varepsilon_k \mu_k}\right)^{\ast}(\varphi).
		\end{split}
		\end{equation*}
		In the above equation, we have used the notation $(\varepsilon \mu)_{jk} = \varepsilon_j \mu_j + \varepsilon_k \mu_k$. Using the commutative diagram \eqref{diagramFmu}, one has
		\begin{equation*}
		(F^{\varepsilon_j \mu_j + \varepsilon_k \mu_k})^{\ast} \Big(e^{\varphi^{\varepsilon_j \mu_j + \varepsilon_k \mu_k}}\Big) = (J^{-1})^{\ast}(f^{\varepsilon_j \mu_j + \varepsilon_k \mu_k})^{\ast}(\rho).
		\end{equation*}
		Taking the logarithm of the above formula, we get
		\begin{equation*}
		\varphi^{\varepsilon_j \mu_j + \varepsilon_k \mu_k} \circ F^{\varepsilon_j \mu_j + \varepsilon_k \mu_k} + \log |\partial_{w} F^{\varepsilon_j \mu_j + \varepsilon_k \mu_k}|^2 = \log((f^{\varepsilon_j \mu_j + \varepsilon_k \mu_k})^{\ast}(\rho)) \circ J^{-1} + \log|(J^{-1})'|^2.
		\end{equation*}
		Then, using the above equations together with Ahlfors formulae \eqref{firstvariation} and Wolpert's formula \eqref{WolpertSecVar}, we have
		\begin{equation*}
		\begin{split}
		- \left.\left(\pdv[2]{\log \Lponetial_{i}^{\varepsilon_j \mu_j + \varepsilon_k \mu_k}}{\varepsilon_j}{\bar{\varepsilon}_k}\right)\right|_{\varepsilon_j = \varepsilon_k =0} & = \lim_{w \to w_{i}} \left.\left(\pdv[2]{}{\varepsilon_j}{\bar{\varepsilon}_k}\right)\right|_{\varepsilon_j = \varepsilon_k =0} \log((f^{\varepsilon_j \mu_j + \varepsilon_k \mu_k})^{\ast}(\rho)) \circ J^{-1} \\
		&\overset{\eqref{WolpertSecVar}}{=}  \frac{1}{2} \lim_{w \to w_{i}} f_{\mu_j \bar{\mu}_k}\circ J^{-1}(w)= \frac{1}{2} \left\langle \pdv{w_j} , \pdv{w_k} \right\rangle_{\text{TZ},i}^{\text{ell}}.
		\end{split}
		\end{equation*}
		\item \emph{$i=n_e+1,\dots,n-1$ case:} 
		\begin{equation*}
		\begin{split}
		\hspace{-.4cm}-\pdv[2]{\log \Lponetial_{i}}{w_j}{\bar{w}_k} & = - \left.\left(\pdv[2]{\log \Lponetial_{i}^{\varepsilon_j \mu_j + \varepsilon_k \mu_k}}{\varepsilon_j}{\bar{\varepsilon}_k}\right)\right|_{\varepsilon_j = \varepsilon_k =0}\\
		&\hspace{-1.5cm}=  2 \lim_{w \to w_i} \left\{ \frac{1}{|w-w_i|}\left.\pdv[2]{}{\varepsilon_j}{\bar{\varepsilon}_k} \right|_{\varepsilon_j = \varepsilon_k =0}  \left(F^{(\varepsilon \mu)_{jk}}\right)^{\ast}\left(e^{-\frac{1}{2}\varphi^{(\varepsilon \mu)_{jk}}(w)}\right) \right.\\
		& \hspace{1cm} +\left. e^{-\frac{1}{2}\varphi(w)}  \left.\pdv{\varepsilon_j}\right|_{\varepsilon_j=0} \left(\frac{\partial_{w}F^{\varepsilon_j \mu_j}(w)}{F^{\varepsilon_j \mu_j}(w-w_i)}\right)^{\frac{1}{2}} \left.\pdv{\bar{\varepsilon}_k}\right|_{\varepsilon_k=0} \left(\frac{\overline{\partial_{w}F^{\varepsilon_k \mu_k}(w)}}{\overline{F^{\varepsilon_k \mu_k}(w-w_i)}}\right)^{\frac{1}{2}} \right\}\\
		& \hspace{-1.6cm}\overset{\eqref{expphivariations}}{=}\lim_{w \to w_i} \left\{-\frac{1}{2} \log|w-w_i| f_{\mu_j \bar{\mu}_k}\circ J^{-1}(w) + \frac{1}{2} e^{-\frac{1}{2}\varphi(w)} \, \times \right.\\
		&\left. \frac{\left((w-w_i) \partial_{w} \dot{F}^j(w)- \dot{F}^j(w) + \dot{F}^j(w_i) \right)}{|w-w_i|(\bar{w}-\bar{w}_i)^{\frac{1}{2}}} \frac{\left(\overline{(w-w_i) \partial_{w} \dot{F}^k(w)- \dot{F}^k(w) + \dot{F}^k(w_i)}\right)}{|w-w_i| (w-w_i)^{\frac{1}{2}}}\right\}\\
		&\hspace{-1.6cm} \overset{\eqref{diagJexpansion}}{=} - \frac{1}{2} \lim_{w \to \infty}  \frac{\log |w-w_i|}{\Im(\varsigma_i^{-1} \circ J^{-1}(w))} (\Im \varsigma_i z) f_{\mu_j \bar{\mu}_k}(\varsigma_i z)\\
		&\hspace{-1.5cm} = \pi \lim_{w \to \infty}  (\Im \varsigma_i z) f_{\mu_j \bar{\mu}_k}(\varsigma_i z) \overset{\eqref{constcuspmetric}}{=} \frac{4 \pi}{3} \left\langle \pdv{w_j} , \pdv{w_k} \right\rangle_{TZ,i}^{cusp},
		\end{split}
		\end{equation*}
		where in the last line, we have used the fact that $\log |w-w_i|/\Im(\varsigma_i^{-1} \circ J^{-1}(w)) = - 2 \pi$ as $ w \to w_i$.
		\item \emph{$i=n$ case:}
		\begin{equation*}
		\begin{split}
		\hspace{-.4cm}\pdv[2]{\log \Lponetial_{1}}{w_j}{\bar{w}_k} & = \left.\left(\pdv[2]{\log \Lponetial_{1}^{\varepsilon_j \mu_j + \varepsilon_k \mu_k}}{\varepsilon_j}{\bar{\varepsilon}_k}\right)\right|_{\varepsilon_j = \varepsilon_k =0}\\
		&\hspace{-1.3cm} = \lim_{w \to \infty} \left\{ \left.\pdv[2]{}{\varepsilon_j}{\bar{\varepsilon}_k} \right|_{\varepsilon_j = \varepsilon_k =0} \left(\log\left|F^{(\varepsilon \mu)_{jk}}\right|^2 - 2 \left(F^{(\varepsilon \mu)_{jk}}\right)^{\ast}\left(e^{-\frac{1}{2}\varphi^{(\varepsilon \mu)_{jk}}}\right)\left|\frac{\partial_{w}F^{(\varepsilon \mu)_{jk}}}{F^{(\varepsilon \mu)_{jk}}}\right|\right)(w)\right\}\\
		&\hspace{-1.3cm} = - 2 \lim_{w \to \infty} \left\{ \frac{1}{|w|}\left.\pdv[2]{}{\varepsilon_j}{\bar{\varepsilon}_k} \right|_{\varepsilon_j = \varepsilon_k =0}  \left(F^{(\varepsilon \mu)_{jk}}\right)^{\ast}\left(e^{-\frac{1}{2}\varphi^{(\varepsilon \mu)_{jk}}}\right) \right.\\
		& \hspace{1.5cm} +\left. e^{-\frac{1}{2}\varphi(w)}  \left.\pdv{\varepsilon_j}\right|_{\varepsilon_j=0} \left(\frac{\partial_{w}F^{\varepsilon_j \mu_j}(w)}{F^{\varepsilon_j \mu_j}(w)}\right)^{\frac{1}{2}} \left.\pdv{\bar{\varepsilon}_k}\right|_{\varepsilon_k=0} \left(\frac{\overline{\partial_{w}F^{\varepsilon_k \mu_k}(w)}}{\overline{F^{\varepsilon_k \mu_k}(w)}}\right)^{\frac{1}{2}} \right\}\\
		&\hspace{-1.4cm} \overset{\eqref{expphivariations}}{=}\lim_{w \to \infty} \left\{\frac{e^{-\frac{1}{2}\varphi(w)}}{2|w|} f_{\mu_j \bar{\mu}_k}\circ J^{-1}  - \frac{e^{-\frac{1}{2}\varphi(w)}}{2}  \frac{\left(w \partial_{w} \dot{F}^j(w)- \dot{F}^j(w)\right)}{|w|\bar{w}^{\frac{1}{2}}} \frac{\left(\overline{w \partial_{w} \dot{F}^k(w)- \dot{F}^k(w)}\right)}{|w| w^{\frac{1}{2}}}\right\}\\
		&\hspace{-1.3cm} =\lim_{w \to \infty} \left\{\frac{|w| \log |w|}{2|w|} f_{\mu_j \bar{\mu}_k}\left( J^{-1}(w)\right)\right\} \overset{\eqref{diagJexpansion}}{=} \frac{1}{2} \lim_{w \to \infty}  \frac{\log |w|}{\Im(J^{-1}(w))} (\Im z) f_{\mu_j \bar{\mu}_k}(z)\\
		&\hspace{-1.3cm} = \pi \lim_{w \to \infty}  (\Im z) f_{\mu_j \bar{\mu}_k}(z) \overset{\eqref{constcuspmetric}}{=} \frac{4 \pi}{3} \left\langle \pdv{w_j} , \pdv{w_k} \right\rangle_{\text{TZ},n}^{\text{cusp}}.
		\end{split}
		\end{equation*}
		In the above equation, we have used the fact that $\log |w|/\Im(J^{-1}(w)) = 2 \pi$ as $ w \to \infty$.
	\end{itemize}  
\end{proof}

\noindent Using Lemmas~\ref{lemma:Hermitianmetriconlambdawnogenus} and \ref{lemma:localpotentialsonmoduliwnogenus}, we can deduce the following corollary:
\begin{corollary}
	The function $-\log \mathsf{H} = -m_{1}h_{1}\log \Lponetial_{1}\dots - m_{n-1} h_{n-1}\log \Lponetial_{n-1}+\log \Lponetial_{n}$ is a K\"ahler potential for the combination $\frac{4 \pi}{3} \omega_{TZ}^{cusp} +\frac{1}{2} \sum_{j=1}^{n_e}m_j h_j\hspace{.5mm}\omega_{\text{TZ},j}^{\text{ell}}$ on $\moduli_{0,n}(\boldsymbol{m})$. The first Chern form of the Hermitian holomorphic $\mathbb{Q}$-line bundle $(\lambda_{0,\boldsymbol{m}},\mathsf{H})$ over $\symmoduli_{0,n}(\boldsymbol{m})$ is given by
	\begin{equation}
	\chern{\lambda_{0,\boldsymbol{m}}}{\mathsf{H}} = \frac{\sqrt{-1}}{2\pi}\bar{\partial}\partial\log{\mathsf{H}}= \frac{4 \pi}{3} \omega_{\text{TZ}}^{\text{cusp}} +\frac{1}{2} \sum_{j=1}^{n_e}m_j h_j \hspace{.5mm}\omega_{\text{TZ},j}^{\text{ell}}.
	\end{equation}
\end{corollary}
\begin{remark}
	In analogy with the famous accessory parameters that are generated by $S_{\boldsymbol{m}}$, the authors of \cite{Can:2017ycp} have defined the so-called ``\emph{auxiliary parameters}'' as\footnote{Our definition of auxiliary parameters differs from that of reference \cite{Can:2017ycp} by a factor of $\tfrac{1}{2}$.}
	\begin{equation}\label{AuxiliaryParameters}
		d_i := \frac{1}{\mathsf{H}} \pdv{\mathsf{H}}{w_i}.
	\end{equation} 
	As emphasized in \cite{Can:2017ycp}, the auxiliary parameters play an equally important role as the accessory parameters. In particular, it follows from the above corollary that an analog of the relation~\eqref{wjderivativeci} can be found between the auxiliary parameters and Takhtajan-Zograf metrics.
\end{remark}
Next, Let us remind that the decomposition of $r(z)$ is given by
\begin{equation}
r(z) = \sum_{i=1}^{n-3} \mathrm{a}_i \, r_i(z),
\end{equation}
where
\begin{equation}
\mathrm{a}_i = \iint_{\fund(\Gamma)} \text{Sch}\left(J;z\right) \mu_i(z) \dd[2]{z}.
\end{equation}
According to Section \ref{subsec:moduli}, for varying $\Gamma$, the $r(z)$ determines a $(1,0)$-form $\mathrm{r}$ on $\teich_{0,n}$ and the latter one is corresponding to the $(1,0)$-form $\vartheta$,\footnote{See Lemma~\ref{Lemma:covering}.}
\begin{equation*}
\vartheta = \sum_{i=1}^{n-3} \mathrm{a}_i \dd{w_i},
\end{equation*}
on $\moduli_{0,n}$. Therefore, we also have:
\begin{corollary}
	The function $\Gpotential_{\boldsymbol{m}}: \moduli_{0,n} \to \mathbb{R}$ satisfies 
	\begin{equation}\label{FD}
	\partial \Gpotential_{\boldsymbol{m}} = 2 \vartheta
	\end{equation}
	and
	\begin{equation}\label{metriccombination}
	\bar{\partial} \partial \Gpotential_{\boldsymbol{m}} = -2\sqrt{-1}\left[\omega_{\text{WP}}-\frac{4\pi^2}{3} \omega_{\text{TZ}}^{\text{cusp}}-\frac{\pi}{2} \sum_{j=1}^{n_e} m_j h_j \hspace{.5mm}\omega^{\text{ell}}_{\text{TZ},j}\right].
	\end{equation} 
\end{corollary}
To see the \eqref{FD}, put $\Gpotential_{\boldsymbol{m}} = S_{\boldsymbol{m}} - \pi \log \mathsf{H}$ and combine Lemma~\ref{lemma:widerivativeloghi} with the proof of Theorem~\ref{theorem1} which it gives
\begin{equation*}
\begin{split}
\partial\Gpotential_{\boldsymbol{m}}&=\partial S_{\boldsymbol{m}}-\pi\sum_{j=1}^{n}m_j h_j\hspace{.5mm}\partial\log\Lponetial_{j}=\sum_{i=1}^{n-3}\left(\frac{\partial S_{\boldsymbol{m}}}{\partial{w_i}}\right)dw_i-\pi\sum_{j=1}^{n}\sum_{i=1}^{n-3}m_j h_j\left(\frac{\partial}{\partial w_i}\log \Lponetial_{j}\right) dw_i\\
&=-2\pi \sum_{i=1}^{n-3} c_i dw_i-\pi\sum_{i=1}^{n-3}\sum_{j=1}^{n} h_j\dot{F}^{i}_{w}(w_j) dw_i = 2\sum_{i=1}^{n-3}\left(-\pi  c_i +\sum_{j=1}^{n}h_j(\mathscr{E}_j,M_{i})\right) dw_i \\
&\overset{
	\ref{EnergyMomentumEi}}{=} 2\sum_{i=1}^{n-3} \mathrm{a}_i dw_i = 2\vartheta
\end{split}
\end{equation*}
In the second line of the above equation, we noted that the Eq.\eqref{Fdot} implies that 
\begin{equation*}
\dot{F}^{i}_w(w) =-\frac{1}{\pi}\iint_{\cmpx} M_{i}(w^{\prime})\hspace{.5mm}\partial_{w} R(w^{\prime},w) d^{2}w^{\prime} = -\frac{1}{\pi}\iint_{\cmpx} M_{i}(w^{\prime})\left(\frac{1}{(w^{\prime}-w)^2}-\frac{1}{w^{\prime}(w^{\prime}-1)}\right) d^{2}w^{\prime},
\end{equation*}
which at the point $w=w_j$ is simplified to
\begin{equation*}
\begin{split}
\dot{F}^{i}_w(w_j)&= -\frac{1}{\pi}\iint_{\cmpx} M_{i}(w)\left(\frac{1}{(w-w_j)^2}-\frac{1}{w_j(w_j-1)}\right) d^{2}w	=-\frac{2}{\pi}\iint_{\cmpx} M_{i}(w)\mathscr{E}_j(w) d^{2}w \\
&= -\frac{2}{\pi}\iint_{\cmpx} (\mathscr{E}_j,M_i).
\end{split}
\end{equation*}
The Eq.\eqref{metriccombination} can be obtained easily by using Theorem~\ref{theorem2} as well as Lemma~\ref{lemma:localpotentialsonmoduliwnogenus}.

\begin{remark}
	According to the Lemmas~\ref{lemma:Hermitianmetriconlambdawnogenus} and \ref{lemma:SmHermiatiannmetricoverlambdanogenus}, the functions $\mathsf{H}$ and $\exp[S_{\boldsymbol{m}}/\pi]$ are Hermitian metrics in the holomorphic $\mathbb{Q}$-line bundle $\lambda_{0,\boldsymbol{m}}$ over $\symmoduli_{0,n}(\boldsymbol{m})$. This implies that $\Gpotential_{\boldsymbol{m}} = S_{\boldsymbol{m}} - \pi \log \mathsf{H}$ determines a function on $\symmoduli_{0,n}(\boldsymbol{m})$. Interestingly, the combination $\omega_{_{\text{WP}}}-\tfrac{4\pi^2}{3} \omega_{_{\text{TZ}}}^{\text{cusp}}-\tfrac{\pi}{2} \sum_{j=1}^{n_e} m_j h_j \hspace{.5mm}\omega^{\text{ell}}_{\text{TZ},j}$ in Eq.\eqref{metriccombination} with overall factor of $\tfrac{1}{12 \pi}$ appears also in the local index theorem for orbifold Riemann surfaces for $k=0,1$. This is the sign that the function $\Gpotential_{\boldsymbol{m}}$ can play the role of the \emph{Quillen's metric} in the Hodge line bundle  $\lambda_1$, defined in \cite{ZT_2018}. We will explain a little bit more about this observation in conclusion.
\end{remark}
\subsection{Chern Forms and Potentials on Schottky Space $\schottky_{g,n}(\boldsymbol{m})$}
Let $O \cong \singrigon \slash \Sigma = [\UHP \slash \Gamma]$ and $\Gamma$ be respectively an orbifold Riemann surface and Fuchsian group of signature  $(g;m_1,\dots,m_{n_e};n_p)$ and let $J : \UHP \to \singrigon$ be the corresponding orbifold covering map. As we have explained in Section \ref{subsec:Schottkyspace}, the automorphic form $\text{Sch}\left(J^{-1};w\right)$ of weight four for the Schottky group can be projected to the subspace $\Hilbert^{2,0}(\singrigon,\Sigma) \cong T_{\pi\circ\Phi(0)}^{\ast}\schottky_{g,n}(\boldsymbol{m})$,
\begin{equation}\label{Rquaddiffexpansion}
\mathsf{R}(w) = \sum_{i=1}^{3g-3+n} \mathrm{b}_i P_i(w) = \sum_{i=1}^{3g-3+n} \left( \text{Sch}\left(J^{-1}\right), M_i \right) P_i(w),
\end{equation}
where $P_i, i=1,\dots,3g-3+n$ are given by \eqref{ls} and \eqref{ls2} and by using Eq.\eqref{SchottkyEnergyMomentum} we get
\begin{equation}
\mathsf{R}(w)  =\sum_{i=1}^{3g-3+n}\left(-\pi  c_i + \sum_{j=1}^{n} h_j\left( \mathscr{E}_j, M_i \right)\right)P_{i}(w)\equiv\pi \mathsf{R}_{0}(w)+\sum_{j=1}^{n} h_j \mathsf{R}_{j}(w).
\end{equation}
The definition of $\mathscr{E}_j$ is provided by \eqref{scrEdef}. The $\mathsf{R}(w)$ coincides with a $(1,0)$-form $\mathscr{Q}$ on the Schottky pace $\schottky_{g,n}(\boldsymbol{m})$\hspace{1mm}
\begin{equation}\label{D}
\begin{split}
\mathscr{Q}= \sum_{i=1}^{3g-3+n}\mathrm{b}_i \textcolor{black}{dw_i} & = \mathrm{b}_1d\lambda_1+\dots +\mathrm{b}_g d\lambda_g+\mathrm{b}_{g+1}da_3+\dots+\mathrm{b}_{2g-3}da_g\\
&+\mathrm{b}_{2g-2}db_2+\dots+\mathrm{b}_{3g-3}db_g+\mathrm{b}_{3g-2}dw_1+\dots+\mathrm{b}_{3g-3+n}dw_n.
\end{split}
\end{equation}
In the following two theorems, which can be regarded as generalizations of Theorems 1 and 2 of \cite{park2015potentials}, we will explicitly describe canonical connections and curvature forms of the Hermitian holomorphic ($\mathbb{Q}$-)line bundles $\linebundle_i$ and $\linebundle=\bigotimes_{i=1}^{n} \linebundle_i^{\Lponetial_i}$:
\begin{maintheorem}\label{mainthrm1}
	Let $\partial$ and $\bar{\partial}$ be $(1,0)$ and $(0,1)$ components of the de Rham differential on Schottky space $\schottky_{g,n}(\boldsymbol{m})$. The following statements are true.
	\begin{enumerate}[(i)]
		\item  
		On the Hermitian holomorphic line bundle $(\linebundle_i,\Lponetial_{i}^{m_i})$, the canonical connection is given by\footnote{When $m_i = \infty$, we will simply ignore $m_i$ in the following formula.} 
		\begin{equation*}
		\partial \log \Lponetial_{i}^{m_i} = - \frac{2}{\pi} \mathsf{R}_i.
		\end{equation*}
		\item 
		On the Hermitian holomorphic $\mathbb{Q}$-line bundle $(\linebundle,e^{\frac{S_{\boldsymbol{m}}}{\pi}})$, the canonical connection is given by
		\begin{equation*}
		\frac{1}{\pi} \partial S_{\boldsymbol{m}} = 2 \mathsf{R}_0.
		\end{equation*}
		\item The function $\Gpotential_{\boldsymbol{m}}: \schottky_{g,n}(\boldsymbol{m}) \to \mathbb{R}$ given by Eq.\eqref{curlyaction} satisfies
		\begin{equation*}
		\partial \Gpotential_{\boldsymbol{m}} = 2 \mathscr{Q}.
		\end{equation*}
	\end{enumerate}
\end{maintheorem}
\begin{proof} We will prove each statement separately:
	\begin{itemize}
		\item
		In order to prove part $(i)$, it is sufficient to show that
		\begin{equation*}
		\left.\left( \pdv{\log \Lponetial_{i}^{\varepsilon \mu_j}}{\varepsilon} \right)\right|_{\varepsilon=0} =\left\{ 
		\begin{split}
		& -\frac{2}{\pi m_i} \left(\mathscr{E}_i, M_j\right) \quad & \text{for} \quad & i=1,\dots,n_e,\\
		& -\frac{2}{\pi} \left(\mathscr{E}_i, M_j \right) &\text{for} \quad & i=n_e+1,\dots,n.
		\end{split}
		\right.
		\end{equation*}
		Using Lemma~\ref{lemma:widerivativeloghi}, we have
		\begin{equation}\label{ex}
		\left.\left( \pdv{\log \Lponetial_{i}^{\varepsilon \mu_j}}{\varepsilon} \right)\right|_{\varepsilon=0} =	\left\{ 
		\begin{split}
		& \frac{1}{m_i}\partial_{w} \dot{F}^{j}(w_i) \quad & \text{for} \quad & i=1,\dots,n_e,\\
		& \partial_{w}\dot{F}^{j}(w_i)  &\text{for} \quad & i=n_e+1,\dots,n.
		\end{split}
		\right.
		\end{equation}
		In another side, according to \eqref{Fdot} and \eqref{scrEdef}, one see
		\begin{equation*}\label{EMinnerprod}
		\pi \partial_{w}\dot{F}^{j}(w_i) = - \iint_{\singfund} M_j(w) \left(\frac{1}{(w-w_i)^2}-\frac{1}{w(w-1)}\right) \dd[2]{w} = -2 \left(\mathscr{E}_i, M_j\right),
		\end{equation*}
		which by substituting in \eqref{ex} gives  the desired result.
		\item To prove part $(ii)$, we need to show that
		\begin{equation*}
		\left.\pdv{\varepsilon}\right|_{\varepsilon=0} S_{\boldsymbol{m}}([L_1^{\varepsilon \mu_i},\dots,L_g^{\varepsilon \mu_i}]; w_1^{\varepsilon \mu_i},\dots, w_{n}^{\varepsilon \mu_i}) = -2 \pi c_i \quad \text{for} \quad i=1,\dots,3g-3+n.
		\end{equation*}
		We have
		\begin{equation}\label{LiederivativeSdef}
		\begin{split}
		&\Lie_{\mu_i}S_{\boldsymbol{m}} = \left.\pdv{\varepsilon}\right|_{\varepsilon=0} S_{\boldsymbol{m}}\left([L_1^{\varepsilon \mu_i},\dots,L_g^{\varepsilon \mu_i}]; w_1^{\varepsilon \mu_i},\dots, w_{n}^{\varepsilon \mu_i}\right) \\
		&\hspace{1.1cm} = \frac{\sqrt{-1}}{2} \lim_{\epsilon \to 0} \left.\pdv{\varepsilon}\right|_{\varepsilon=0} I_{\epsilon}(\varepsilon)+\lim_{\epsilon \to 0} \frac{\partial \tilde{S}^{\text{(ct)}\epsilon}}{\partial w_i},
		\end{split}
		\end{equation}
		with $\tilde{S}^{\text{(ct)}\epsilon}$ is given in equation \eqref{modifiedS} and
		\begin{equation*}
		\begin{split}
		&\hspace{-1cm}I_{\epsilon}(\varepsilon) = \iint_{F^{\varepsilon \mu_i}(\singfund_{\epsilon})} |\partial_w \varphi^{\varepsilon \mu_i}|^2  \dd{w} \wedge \dd{\bar{w}} + \sum_{k=2}^{g} \oint_{F^{\varepsilon \mu_i}(C_k)} \theta_{(L_k^{\varepsilon \mu_i})^{-1}}(\varphi^{\varepsilon \mu_i}) \\
		&\hspace{1cm}+\sum_{j=1}^{n_e} \left(1-\frac{1}{m_j}\right) \oint_{F^{\varepsilon \mu_i}(C_{j}^{\epsilon})} \varphi^{\varepsilon \mu_i} \left(\frac{\dd{\bar{w}}}{\bar{w}-\overline{w_{j}^{\varepsilon \mu_i}}} - \frac{\dd{w}}{w - w_{j}^{\varepsilon \mu_i}}\right),
		\end{split}
		\end{equation*}
		where once again, we have used the Gauss-Bonnet formula for Riemann orbisurfaces \cite{Troyanov1991PrescribingCO,Li_orbifold_2018}
		\begin{equation*}
		\frac{\sqrt{-1}}{2}\iint_{F^{\varepsilon \mu_i}(\singfund)} e^{\varphi^{\varepsilon \mu_i}} \dd{w} \wedge \dd{\bar{w}} = 2 \pi \left(2g+\sum_{j=1}^{n_e} \left(1- \frac{1}{m_j}\right)+n_p-2\right) = - 2 \pi \chi(X),
		\end{equation*}
		to conclude that 
		\begin{equation*}
		\frac{\sqrt{-1}}{2} \Lie_{\mu_i} \iint_{\singfund} e^{\varphi} \dd{w} \wedge \dd{\bar{w}} = 0.
		\end{equation*}
		The calculation of $\Lie_{\mu_i}S_{\boldsymbol{m}}$ closely follows the corresponding computation in the proof of Theorem 1 in \cite{1988SbMat..60..297Z}, where regularization at the punctures can be found in the proof of Theorem 1 in \cite{Zograf1988ONLE} and for branch points in the proof of Theorem~\ref{theorem1}. More explicitly, by applying the change of variable formula $\int_{F(\singfund)} \omega = \int_{\singfund} F^{\ast}(\omega)$ and noting to the commutative diagram~\ref{diagramFmuschottky}, one finds
		\begin{equation*}
		\begin{split}
		I_{\epsilon}(\varepsilon) & = \iint_{\singfund_{\epsilon}(\varepsilon)} (F^{\varepsilon\mu_i})^{\ast} \Big(|\partial_w \varphi^{\varepsilon \mu_i}|^2 \dd{w} \wedge \dd{\bar{w}}\Big) + \sum_{k=2}^{g} \oint_{C_k} (F^{\varepsilon\mu_i})^{\ast} \left(\theta_{(L_k^{\varepsilon \mu_i})^{-1}}(\varphi^{\varepsilon \mu_i})\right)\\
		& + \sum_{j=1}^{n_e} \left(1-\frac{1}{m_j}\right) \oint_{C_{j}^{\epsilon}(\varepsilon)} (F^{\varepsilon\mu_i})^{\ast} \left(\varphi^{\varepsilon \mu_i} \left(\frac{\dd{\bar{w}}}{\bar{w}-\overline{w_{j}^{\varepsilon \mu_i}}} - \frac{\dd{w}}{w - w_{j}^{\varepsilon \mu_i}}\right)\right)\\ \\
		& = \iint_{\singfund_{\epsilon}(\varepsilon)} |\partial_w \varphi^{\varepsilon \mu_i} \circ F^{\varepsilon\mu_i}|^2  \dd{F^{\varepsilon\mu_i}(w)} \wedge \dd{\overline{F^{\varepsilon\mu_i}(w)}}\\ 
		&+ \sum_{k=2}^{g} \oint_{C_k} \left(\varphi^{\varepsilon \mu_i} \circ F^{\varepsilon\mu_i} - \frac{1}{2} \log|(L_k^{\varepsilon\mu_i})' \circ F^{\varepsilon\mu_i}|^2 - \log |l_k^{\varepsilon\mu_i}|^2 \right) \\
		& \hspace{2.8cm} \times \left(\frac{(L_k^{\varepsilon\mu_i})'' \circ F^{\varepsilon\mu_i}}{(L_k^{\varepsilon\mu_i})' \circ F^{\varepsilon\mu_i}} \dd{F^{\varepsilon\mu_i}(w)} - \frac{\overline{(L_k^{\varepsilon\mu_i})'' \circ F^{\varepsilon\mu_i}}}{\overline{(L_k^{\varepsilon\mu_i})' \circ F^{\varepsilon\mu_i}}} \dd{\overline{F^{\varepsilon\mu_i}(w)}}\right)\\
		& + \sum_{j=1}^{n_e} \left(1-\frac{1}{m_j}\right) \oint_{C_{j}^{\epsilon}(\varepsilon)}  (\varphi^{\varepsilon \mu_i} \circ F^{\varepsilon\mu_i}) \left(\frac{\dd{\overline{F^{\varepsilon\mu_i}(w)}}}{\overline{F^{\varepsilon\mu_i}(w)}-\overline{w_{j}^{\varepsilon \mu_i}}} - \frac{\dd{F^{\varepsilon\mu_i}(w)}}{F^{\varepsilon\mu_i}(w) - w_{j}^{\varepsilon \mu_i}}\right).
		\end{split}
		\end{equation*}
		which by noting that $\dd{F^{\varepsilon\mu_i}(w)}=\partial_{w}F^{\varepsilon\mu_i}(dw + \varepsilon M_i \hspace{.5mm}d\bar{w})$ and $\dd{\overline{F^{\varepsilon\mu_i}(w)}}=\overline{\partial_{w}F^{\varepsilon\mu_i}}(\overline{\varepsilon}\overline{M_i} dw + d\bar{w})$, it turns to
		\begin{equation*}
		\begin{split}
		I_{\epsilon}(\varepsilon)& = \iint_{\singfund_{\epsilon}(\varepsilon)} |\partial_w \varphi^{\varepsilon \mu_i} \circ F^{\varepsilon\mu_i}|^2 \, |\partial_w F^{\varepsilon\mu_i}|^2 (1 - |\varepsilon M_i|^2) \dd{w} \wedge \dd{\bar{w}}\\
		&\hspace{-0.7cm}-2  \sum_{k=2}^{g} \oint_{C_k} \varphi^{\varepsilon \mu_i} \circ F^{\varepsilon\mu_i} \, \frac{\overline{(L_k^{\varepsilon\mu_i})'' \circ F^{\varepsilon\mu_i}}}{\overline{(L_k^{\varepsilon\mu_i})' \circ F^{\varepsilon\mu_i}}} \, \overline{\partial_w F^{\varepsilon\mu_i}}(\bar{\varepsilon} \overline{M_i} \dd{w}+ \dd{\bar{w}}) \\
		&\hspace{-0.7cm}- \sum_{k=2}^{g} \oint_{C_k} \partial_{w}\varphi^{\varepsilon \mu_i} \circ F^{\varepsilon\mu_i} \, \log|(L_k^{\varepsilon\mu_i})' \circ F^{\varepsilon\mu_i}|^2 \, \partial_w F^{\varepsilon\mu_i} (\dd{w}+ \varepsilon M_i  \dd{\bar{w}})\\
		&\hspace{-0.7cm}-\sum_{k=2}^{g} \oint_{C_k} \partial_{\bar{w}}\varphi^{\varepsilon \mu_i} \circ F^{\varepsilon\mu_i} \, \log|(L_k^{\varepsilon\mu_i})' \circ F^{\varepsilon\mu_i}|^2 \, \overline{\partial_w F^{\varepsilon\mu_i}}(\bar{\varepsilon} \overline{M_i} \dd{w}+ \dd{\bar{w}})\\
		&\hspace{-0.7cm}+  \sum_{k=2}^{g} \oint_{C_k} \log|(L_k^{\varepsilon\mu_i})' \circ F^{\varepsilon\mu_i}|^2 \, \frac{\overline{(L_k^{\varepsilon\mu_i})'' \circ F^{\varepsilon\mu_i}}}{\overline{(L_k^{\varepsilon\mu_i})' \circ F^{\varepsilon\mu_i}}} \, \overline{\partial_w F^{\varepsilon\mu_i}}(\bar{\varepsilon} \overline{M_i} \dd{w}+ \dd{\bar{w}})\\
		&\hspace{-0.7cm}+ 8 \pi \sqrt{-1} \sum_{k=2}^{g} \log |l_k^{\varepsilon\mu_i}|^2\\
		&\hspace{-0.7cm}+\sum_{j=1}^{n_e} \hspace{-1mm}\left(1-\frac{1}{m_j}\right)\hspace{-1mm} \oint_{C_{j}^{\epsilon}(\varepsilon)}  \hspace{-1mm}(\varphi^{\varepsilon \mu_i} \circ F^{\varepsilon\mu_i}) \hspace{-1mm}\left(\frac{\overline{\partial_w F^{\varepsilon\mu_i}}(\bar{\varepsilon} \overline{M_i} \dd{w} +\dd{\bar{w}})}{\overline{w^{\varepsilon\mu_i}}-\overline{w_{j}^{\varepsilon \mu_i}}} - \frac{\partial_w F^{\varepsilon\mu_i}(\dd{w}+ \varepsilon M_i \dd{\bar{w}})}{w^{\varepsilon\mu_i} - w_{j}^{\varepsilon \mu_i}}\right)
		\end{split}
		\end{equation*}
		where $C_k$s are the usual components of $\partial\singfund$ and
		\begin{equation*}
		\left\{
		\begin{split}
		&\singfund_{\epsilon}(\varepsilon) = \singfund \Big\backslash \bigcup_{i=1}^{n} D_i^{\epsilon}(\varepsilon),\\
		& D_i^{\epsilon}(\varepsilon) = \left\{w \in \SchottkyFund \, \Big| \, |w^{\varepsilon \mu_i}-w_i^{\varepsilon\mu_i}|<\epsilon\right\},\\
		&C_{j}^{\epsilon}(\varepsilon) =\partial D_{j}^{\epsilon}(\varepsilon) = \left\{w \in \SchottkyFund\, \Big| \, |w^{\varepsilon \mu_i}-w_{j}^{\varepsilon\mu_i}| = \epsilon\right\}.
		\end{split}
		\right.
		\end{equation*}
		Just as in the proof of Theorem~\ref{theorem1}, to calculate $\dot{I}^i_{\epsilon}$, we must differentiate both the integrand and the integration domain $\singfund_{\epsilon}(\varepsilon)$.
		Using Eq.\eqref{vardomaindiff}, the second contribution gives
		\begin{equation}
		\left.\pdv{\varepsilon}\right|_{\varepsilon=0} \iint_{\singfund_{\epsilon}(\varepsilon)}|\partial_w \varphi|^2 \dd{w} \wedge \dd{\bar{w}} = - \sum_{j=1}^{n}  \oint_{\partial D_{j}^{\epsilon}(\varepsilon)} |\partial_w \varphi|^2 \left(\dot{F}^i(w)-\dot{F}^i(w_j)\right) \dd{\bar{w}},
		\end{equation}
		where the boundaries $\partial D_{j}^{\epsilon}(\varepsilon)$ are oriented in a manner that is counter to the orientation of $\partial\singfund_{\epsilon}$. These boundaries also oriented as a boundary of $D_{j}^{\epsilon}(\varepsilon)$. Moreover, the contribution from the variation of integral domain $C_{j}^{\varepsilon}$ vanishes since $\partial C_{j}^{\epsilon} = \emptyset$. The differentiation under the integral sign repeats the calculations done in the proof of \cite[Theorem~1]{1988SbMat..60..297Z} almost word-for-word. The only change is that the integration domain has now changed to $\SchottkyFund_{\epsilon}$.  Accordingly,
		\begin{equation*}\label{variationISchottky}
		\begin{split}
		\dot{I}^i_{\epsilon} &= \iint_{\singfund_{\epsilon}} \left[(\partial_w \dot{\varphi}^i + \partial_w^2 \varphi \dot{F}^i) \partial_{\bar{w}} \varphi + (\partial_{\bar{w}} \dot{\varphi}^i  + \partial_w \partial_{\bar{w}} \varphi  \dot{F}^i) \partial_w \varphi + |\partial_w\varphi|^2  \partial_w \dot{F}^i \right] \dd{w} \wedge \dd{\bar{w}}\\
		& - \sum_{j=1}^{n}  \oint_{\partial D_{j}^{\epsilon}} |\partial_w \varphi|^2 \left(\dot{F}^i(w)-\dot{F}^i(w_j)\right) \dd{\bar{w}} - 2\sum_{k=2}^g \oint_{C_k}( \dot{\varphi}^i + \partial_{w} \varphi \dot{F}^i)\frac{\overline{L_k''}}{\overline{L_k'}} \dd{\bar{w}}\\
		&-\sum_{k=2}^g \oint_{C_k} \left[(\partial_{w}\dot{\varphi}^i+\partial_{w}^2 \varphi \dot{F}^i) \log|L_k'|^2 \dd{w} + \partial_{w}\varphi \frac{(\dot{L}_k^i)' + L_k''\dot{F}^i}{L'_k} \dd{w}\right.\\
		&\hspace{3cm} + \partial_{w}\varphi \log|L'_k|^2 \partial_{w} \dot{F}^i \dd{w} + \partial_{w}\varphi \log|L'_k|^2 M_i \dd{\bar{w}}\Bigg]\\
		& -\sum_{k=2}^g \oint_{C_k} \left[(\partial_{\bar{w}}\dot{\varphi}^i +\partial_{w} \partial_{\bar{w}} \varphi \dot{F}^i) \log|L_k'|^2 + \partial_{w}\varphi \frac{(\dot{L}_k^i)'  + L_k''\dot{F}^i}{L'_k} \right]\dd{\bar{w}}\\
		& + \sum_{k=2}^g \oint_{C_k} \frac{(\dot{L}_k^i)'  + L_k''\dot{F}^i}{L'_k} \frac{\overline{L_k''}}{\overline{L_k'}} \dd{\bar{w}} + 8 \pi \sqrt{-1} \sum_{k=2}^g  \frac{\dot{l}_k^i}{l_k}\\
		&  -\sum_{j=1}^{n_e} \left(1-\frac{1}{m_j}\right) \oint_{C_{j}^{\epsilon}} \partial_{w}\dot{F}^{i} \left(\frac{d\bar{w}}{\bar{w}-\bar{w}_j}-\frac{dw}{w-w_j}\right)+\Sorder{1},
		\end{split}
		\end{equation*}
		where we have used Eq.\eqref{extra1} and the fact that $F^{\varepsilon \mu_i}$ and $L_k^{\varepsilon \mu_i}$ are holomorphic in the variable $\varepsilon \in \cmpx$. As in the proof of Theorem~\ref{theorem1}, we can use Lemma~\ref{varp} and the equality $\partial_{\bar{w}} \dot{F}^i = M_i$ to rewrite the above equation as
		\begin{equation}\label{variationISchottky2}
		\begin{split}
		\dot{I}^i_{\epsilon} &= \iint_{\singfund_{\epsilon}} \left[(-\partial_w \varphi \partial_w \dot{F}^i - \partial_w^2 \dot{F}^i) \partial_{\bar{w}} \varphi + (-\partial_{w}\varphi \partial_{\bar{w}} \dot{F}^i  - \partial_w \partial_{\bar{w}} \dot{F}^i) \partial_w \varphi + |\partial_w\varphi|^2  \partial_w \dot{F}^i \right] \dd{w} \wedge \dd{\bar{w}}\\
		& - \sum_{j=1}^{n}  \oint_{\partial D_{j}^{\epsilon}} |\partial_w \varphi|^2 \left(\dot{F}^i(w)-\dot{F}^i(w_j)\right) \dd{\bar{w}} + 2\sum_{k=2}^g \oint_{C_k} \partial_{w} \dot{F}^i \frac{\overline{L_k''}}{\overline{L_k'}} \dd{\bar{w}}\\
		&-\sum_{k=2}^g \oint_{C_k} \left[(-\partial_w \varphi \partial_w \dot{F}^i - \partial_w^2 \dot{F}^i) \log|L_k'|^2 \dd{w} + \partial_{w}\varphi \frac{(\dot{L}_k^i)' + L_k''\dot{F}^i}{L'_k} \dd{w}\right.\\
		&\hspace{3.7cm} + \partial_{w}\varphi \log|L'_k|^2 \partial_{w} \dot{F}^i \dd{w} + \partial_{w} \partial_{\bar{w}} \dot{F}^i \log|L'_k|^2 \dd{\bar{w}}\Bigg]\\
		& - \sum_{k=2}^g \oint_{C_k} \left(\partial_{w}\varphi-\frac{\overline{L_k''}}{\overline{L_k'}} \right)\frac{(\dot{L}_k^i)'  + L_k''\dot{F}^i}{L'_k} \dd{\bar{w}} + 8 \pi \sqrt{-1} \sum_{k=2}^g  \frac{\dot{l}_k^i}{l_k}\\
		& -\sum_{j=1}^{n_e} \left(1-\frac{1}{m_j}\right) \oint_{C_{j}^{\epsilon}} \partial_{w}\dot{F}^{i} \left(\frac{d\bar{w}}{\bar{w}-\bar{w}_j}-\frac{dw}{w-w_j}\right)+\Sorder{1}.
		\end{split}
		\end{equation}
		Let us first compute the integral over $\singfund_{\epsilon}$:
		\begin{equation*}
		\begin{split}
		&\iint_{\singfund_{\epsilon}} \left[(-\partial_w \varphi \partial_w \dot{F}^i - \partial_w^2 \dot{F}^i) \partial_{\bar{w}} \varphi + (-\partial_{w}\varphi \partial_{\bar{w}} \dot{F}^i  - \partial_w \partial_{\bar{w}} \dot{F}^i) \partial_w \varphi + |\partial_w\varphi|^2  \partial_w \dot{F}^i \right] \dd{w} \wedge \dd{\bar{w}}\\
		&=\iint_{\singfund_{\epsilon}} \left[\big(2\partial_{w}^2 \varphi - (\partial_{w}\varphi)^2\big)\partial_{\bar{w}}\dot{F}^i - 2\pdv{w}(\partial_{w}\varphi \partial_{\bar{w}}\dot{F}^i)\right.\\ & \hspace{4.35cm} \left. + \pdv{\bar{w}}(\partial_{w}\varphi \partial_{w}\dot{F}^i) - \pdv{w}(\partial_{\bar{w}}\varphi \partial_{w}\dot{F}^i)\right] \dd{w} \wedge \dd{\bar{w}}\\
		& = 2 \iint_{\singfund_{\epsilon}} T_{\varphi} M_i \dd{w} \wedge \dd{\bar{w}} \underbrace{-2 \int_{\partial\singfund_{\epsilon}}\partial_{w}\varphi \partial_{\bar{w}}\dot{F}^i \dd{\bar{w}}}_{I_1} \underbrace{-\int_{\partial\singfund_{\epsilon}}\partial_{w}\varphi \partial_{w}\dot{F}^i \dd{w}}_{I_2}\underbrace{- \int_{\partial\singfund_{\epsilon}}\partial_{\bar{w}}\varphi \partial_{w} \dot{F}^i \dd{\bar{w}}}_{I_3},
		\end{split}
		\end{equation*}
		where in the last line, we have used the definition of the Energy-Momentum $T_{\varphi}  = \text{Sch}\left(J^{-1};w\right) = \partial_w^2 \varphi - \frac{1}{2} (\partial_w \varphi)^2$. In order to compute the integrals $I_1$, $I_2$, and $I_3$, we need some useful identities. From the equality
		$F^{\varepsilon \mu_i} \circ L_k = L_k^{\varepsilon \mu_i} \circ F^{\varepsilon \mu_i} $ for $k=1,\dots,g$
		one can see that
		\begin{equation*}
		\left\{
		\begin{split}
		&\dot{F}^i \circ L_k = \dot{F}^i L'_k + \dot{L}_k^i,\\ \\
		&\partial_{w}\dot{F}^i \circ L_k \, L'_k = \partial_{w}\dot{F}^i L'_k  + \dot{F}^i L''_k + \dot{L}_k^i,\\ \\
		& \partial_{\bar{w}}\dot{F}^i \circ L_k \, \overline{L'_k} = \partial_{\bar{w}}\dot{F}^i L'_k ,
		\end{split}
		\right.
		\end{equation*}
		Moreover, from Eq.\eqref{Lonphi} one has
		\begin{equation*}
		\partial_{w} \varphi \circ L_k \, L'_k + \frac{L''_k}{L'_k} = \partial_{w}\varphi, \qquad \partial_{\bar{w}} \varphi \circ L_k \, \overline{L'_k} + \frac{\overline{L''_k}}{\overline{L'_k}} = \partial_{\bar{w}}\varphi.
		\end{equation*}
		Now, recalling that $\partial\singfund_{\epsilon} = \left[\bigcup\limits_{k=1}^g (C_k \cup C'_k)\right] \cup \left[\bigcup\limits_{i=1}^{n+l} C_i^{\epsilon}\right]$ where $C'_k = - L_k(C_k)$, together with implementing some of the above identities we get
		\begin{equation*}
		\begin{split}
		I_1 &= -2 \int_{\partial\singfund_{\epsilon}}\partial_{w}\varphi \partial_{\bar{w}}\dot{F}^i \dd{\bar{w}}\\
		&= -2 \sum_{k=2}^g \oint_{C_k}\partial_{w}\varphi \partial_{\bar{w}}\dot{F}^i \dd{\bar{w}} -2 \sum_{j=1}^{n_e} \oint_{C^{\epsilon}_j}\partial_{w}\varphi \partial_{\bar{w}}\dot{F}^i \dd{\bar{w}} -2 \sum_{j=1}^{n_p} \oint_{C^{\epsilon}_{n_e+j}}\partial_{w}\varphi \partial_{\bar{w}}\dot{F}^i \dd{\bar{w}}\\
		&= -2 \sum_{k=2}^g \oint_{C_k}\partial_{\bar{w}}\dot{F}^i \frac{L''_k}{L'_k} \dd{\bar{w}} -2 \sum_{j=1}^{n_e} \oint_{C^{\epsilon}_j}\partial_{w}\varphi M_i \dd{\bar{w}} -2 \sum_{j=1}^{n_p} \oint_{C^{\epsilon}_{n_e+j}}\partial_{w}\varphi M_i \dd{\bar{w}},
		\end{split}
		\end{equation*} 
		which it can be simplified more by using the asymptotic expansions \eqref{MAsymptotic}
		\begin{equation*}
		\begin{split}
		&I_1 = 2 \sum_{k=2}^g \oint_{C_k}\left[\left(\frac{L''_k}{L'_k} \right)' \dot{F}^i + \frac{L''_k}{L'_k} \partial_{w}\dot{F}^i \right]\dd{w} \\
		&-2 \sum_{j=1}^{n_e} \oint_{C^{\epsilon}_{j}} \left( - \frac{1-\frac{1}{m_j}}{w-w_{j}}+ \frac{c_{j}}{1-\frac{1}{m_j}} + \dotsm \right) \left(\frac{\bar{q}^{(j)}_1}{4 \bar{J}^{(j)}_1} \cdot |w - w_{j}|^{1-\frac{2}{m_j}} + \dotsm \right) \dd{\bar{w}} \\
		&-2 \sum_{j=1}^{n_p-1} \oint_{C^{\epsilon}_{n_e+j}} \left( -\frac{1}{w-w_{n_e+j}} \left(1+\frac{1}{\log\left|\frac{w-w_{n_e+j}}{J_{1}^{(n_e+j)}}\right|}\right) + \dotsm \right)\times \\
		&\hspace{2cm}\times \left(  - \frac{|\delta_{n_e+j}|^2 \bar{q}_1^{(n_e+j)}}{4 \pi^2 \bar{J}_1^{(n_e+j)}} \cdot |w-w_{n_e+j}| \log^2 |w-w_{n_e+j}|+ \dotsm \right)  \dd{\bar{w}} \\
		&-2 \oint_{C_n^{\epsilon}} \left(-\frac{1}{w} \left(1+\frac{1}{\log\left|\frac{w}{J_{-1}^{(n)}}\right|}\right) - \frac{c_n}{w^2} + \dotsm \right) \left( - \frac{|\delta_n|^2 \bar{q}_1^{(n)} \bar{J}_{-1}^{(n)}}{4 \pi^2} \cdot \frac{\log^2|w|}{|w|} + \dotsm \right)  \dd{\bar{w}}\\
		& = 2 \sum_{k=2}^g \oint_{C_k}\left[\left(\frac{L''_k}{L'_k} \right)' \dot{F}^i + \frac{L''_k}{L'_k} \partial_{w}\dot{F}^i \right]\dd{w}+ \Sorder{1};
		\end{split}
		\end{equation*}
		Moreover,
		\begin{equation*}
		\begin{split}
		I_2 &=-\int_{\partial\singfund_{\epsilon}}\partial_{w}\varphi \partial_{w}\dot{F}^i \dd{w} = - \sum_{k=2}^g \oint_{C_k} \partial_{w}\varphi \partial_{w}\dot{F}^i \dd{w} - \sum_{j=1}^{n} \oint_{C_j^{\epsilon}} \partial_{w}\varphi \partial_{w}\dot{F}^i \dd{w}\\
		& = - \sum_{k=2}^g \oint_{C_k} \left[\frac{L''_k}{L'_k} \partial_{w}\dot{F}^i - \left(\partial_{w}\varphi - \frac{L''_k}{L'_k}\right) \frac{(\dot{L}_k^i)'  + L_k''\dot{F}^i}{L'_k} \right]\dd{w} - \sum_{j=1}^{n} \oint_{C_j^{\epsilon}} \partial_{w}\varphi \partial_{w}\dot{F}^i \dd{w};
		\end{split}
		\end{equation*}
		Furthermore,
		\begin{equation*}
		\begin{split}
		I_3 &= - \int_{\partial\singfund_{\epsilon}}\partial_{\bar{w}}\varphi \partial_{w} \dot{F}^i \dd{\bar{w}} = - \sum_{k=2}^g \oint_{C_k} \partial_{\bar{w}}\varphi \partial_{w} \dot{F}^i \dd{\bar{w}} - \sum_{j=1}^{n} \oint_{C_j^{\epsilon}} \partial_{\bar{w}}\varphi \partial_{w} \dot{F}^i \dd{\bar{w}}\\
		& = - \sum_{k=2}^g \oint_{C_k} \left[\frac{\overline{L''_k}}{\overline{L'_k}} \partial_{w}\dot{F}^i - \left(\partial_{\bar{w}}\varphi - \frac{\overline{L''_k}}{\overline{L'_k}}\right) \frac{(\dot{L}_k^i)'  + L_k''\dot{F}^i}{L'_k} \right]\dd{\bar{w}}  - \sum_{j=1}^{n} \oint_{C_j^{\epsilon}} \partial_{\bar{w}}\varphi \partial_{w} \dot{F}^i \dd{\bar{w}};
		\end{split}
		\end{equation*}
		Using the fact that 
		\begin{equation*}
		- \sum_{j=1}^{n} \oint_{C_j^{\epsilon}} \partial_{w}\varphi \partial_{w} \dot{F}^i \dd{w}- \sum_{j=1}^{n} \oint_{C_j^{\epsilon}} \partial_{\bar{w}}\varphi \partial_{w} \dot{F}^i \dd{\bar{w}} = 	- \sum_{j=1}^{n} \oint_{C_j^{\epsilon}} \partial_{w} \dot{F}^i \dd{\varphi} = \Sorder{1} \quad \text{as} \quad \epsilon \to 0,
		\end{equation*}
		we get
		\begin{equation*}
		\begin{split}
		I_1 + I_2 + I_3 & = 2 \sum_{k=2}^g \oint_{C_k}\left[\left(\frac{L''_k}{L'_k} \right)' \dot{F}^i + \frac{L''_k}{L'_k} \partial_{w}\dot{F}^i \right]\dd{w}\\
		& - \sum_{k=2}^g \oint_{C_k} \left[\frac{L''_k}{L'_k} \partial_{w}\dot{F}^i - \left(\partial_{w}\varphi - \frac{L''_k}{L'_k}\right) \frac{(\dot{L}_k^i)'  + L_k''\dot{F}^i}{L'_k} \right]\dd{w}\\
		& - \sum_{k=2}^g \oint_{C_k} \left[\frac{\overline{L''_k}}{\overline{L'_k}} \partial_{w}\dot{F}^i - \left(\partial_{\bar{w}}\varphi - \frac{\overline{L''_k}}{\overline{L'_k}}\right) \frac{(\dot{L}_k^i)'  + L_k''\dot{F}^i}{L'_k} \right]\dd{\bar{w}}  + \Sorder{1}.
		\end{split}
		\end{equation*}
		Substituting the above results in Eq.\eqref{variationISchottky2}, we get
		\begin{equation}\label{variationISchottky4}
		\begin{split}
		\dot{I}^i_{\epsilon} &= 2 \iint_{\singfund_{\epsilon}} T_{\varphi} M_i \dd{w} \wedge \dd{\bar{w}} - \sum_{j=1}^{n}  \oint_{C_{j}^{\epsilon}} |\partial_w \varphi|^2 \left(\dot{F}^i(w)-\dot{F}^i(w_j)\right) \dd{\bar{w}}\\
		& +\sum_{k=2}^g \oint_{C_k} \left[\partial_{w} \dot{F}^i \left(\frac{L''_k}{L'_k} \dd{w} + \frac{\overline{L''_k}}{\overline{L'_k}} \dd{\bar{w}}\right) + \log|L'_k|^2 \left(\partial_{w}^2 \dot{F}^i \dd{w}+ \partial_{\bar{w}}\partial_{w} \dot{F}^i \dd{\bar{w}}\right)\right] \\
		& +\sum_{k=2}^g \oint_{C_k} \dot{F}^i \left[2\left(\frac{L''_k}{L'_k}\right)' - \left(\frac{L''_k}{L'_k}\right)^2\right] \dd{w} \\
		& +\sum_{k=2}^g \oint_{C_k} \frac{(\dot{L}_k^i)'  + L_k''}{(L'_k)^2} \dd{w} + 8 \pi \sqrt{-1} \sum_{k=2}^g  \frac{\dot{l}_k^i}{l_k}\\
		& -\sum_{j=1}^{n_e} \left(1-\frac{1}{m_j}\right) \oint_{C_{j}^{\epsilon}} \partial_{w}\dot{F}^{i} \left(\frac{d\bar{w}}{\bar{w}-\bar{w}_j}-\frac{dw}{w-w_j}\right)+\Sorder{1}.
		\end{split}
		\end{equation}
		Using the fact that the second line in \eqref{variationISchottky4} can be written as
		\begin{equation*}
		\begin{split}
		&\sum_{k=2}^g \oint_{C_k} \left[\partial_{w} \dot{F}^i \left(\frac{L''_k}{L'_k} \dd{w} + \frac{\overline{L''_k}}{\overline{L'_k}} \dd{\bar{w}}\right) + \log|L'_k|^2 \left(\partial_{w}^2 \dot{F}^i \dd{w}+ \partial_{\bar{w}}\partial_{w} \dot{F}^i \dd{\bar{w}}\right)\right] \\
		& = \sum_{k=2}^g \oint_{C_k} \dd{\left(\partial_{w} \dot{F}^i \log|L'_k|^2\right)},
		\end{split}
		\end{equation*}
		one can see that this sum vanishes; the third line in \eqref{variationISchottky4} is also equal to $0$, since
		\begin{equation*}
		\sum_{k=2}^g \oint_{C_k} \dot{F}^i \left[2\left(\frac{L''_k}{L'_k}\right)' - \left(\frac{L''_k}{L'_k}\right)^2\right] \dd{w} = 2 \sum_{k=2}^g \oint_{C_k} \dot{F}^i \hspace{1mm}\text{Sch}\left(L_k;w\right)\dd{w} = 0.
		\end{equation*}
		As for the fourth line of Eq.\eqref{variationISchottky4}, since the contour $C_k$ encircles a zero of the function $L'_k$, i.e. the point $L_k^{-1}(\infty)$, we get from the Cauchy formula that 
		\begin{equation*}
		\oint_{C_k} \frac{(\dot{L}_k^i)'  + L_k''}{(L'_k)^2} \dd{w} = - 8 \pi \sqrt{-1} \frac{\dot{l}_k^i}{l_k} \quad \text{for} \quad k=2,\dots,g;
		\end{equation*}
		therefore, the fourth line of Eq.\eqref{variationISchottky4} vanishes as well. Putting everything together, one has
		\begin{equation}\label{variationISchottky5}
		\begin{split}
		\dot{I}^i_{\epsilon} &= 2 \iint_{\singfund_{\epsilon}} T_{\varphi} M_i \dd{w} \wedge \dd{\bar{w}} - \sum_{j=1}^{n}  \oint_{C_{j}^{\epsilon}} |\partial_w \varphi|^2 \left(\dot{F}^i(w)-\dot{F}^i(w_j)\right) \dd{\bar{w}}\\
		&\hspace{1cm} -\sum_{j=1}^{n_e} \left(1-\frac{1}{m_j}\right) \oint_{C_{j}^{\epsilon}} \partial_{w}\dot{F}^{i} \left(\frac{d\bar{w}}{\bar{w}-\bar{w}_j}-\frac{dw}{w-w_j}\right)+\Sorder{1}.
		\end{split}
		\end{equation}
		Using Eq.\eqref{LiederivativeSdef}, the limit
		\begin{equation*}
		\lim_{\epsilon \to 0} \iint_{\singfund_{\epsilon}} T_{\varphi} M_i \dd{w} \wedge \dd{\bar{w}} \overset{\eqref{Rquaddiffexpansion}}{=} -2 \sqrt{-1} \mathrm{b}_i,
		\end{equation*}
		and
		\begin{equation*}\label{variationISchottky6}
		\begin{split}
		& \lim_{\epsilon \to 0}  \sum_{j=1}^{n}  \oint_{C_{j}^{\epsilon}} |\partial_w \varphi|^2 \left(\dot{F}^i(w)-\dot{F}^i(w_j)\right) \dd{\bar{w}}\\
		& \overset{\ref{lemma:asymptotics}}{=}\lim_{\epsilon \to 0}  \sum_{j=1}^{n_e} \left(1-\frac{1}{m_j}\right)^2 \oint_{C_{j}^{\epsilon}} \frac{\dot{F}^i(w)-\dot{F}^i(w_{j})}{|w-w_{j}|^2} \dd{\bar{w}}\\
		& + \lim_{\epsilon \to 0}  \sum_{j=1}^{n_p-1} \oint_{C_{j}^{\epsilon}} \left(\frac{\dot{F}^i(w)-\dot{F}^i(w_j)}{|w-w_j|^2} + \frac{2\big(\dot{F}^i(w)-\dot{F}^i(w_j)\big)}{|w-w_j|^2 \log|w-w_j|}\right) \dd{\bar{w}}\\
		&+ \lim_{\epsilon \to 0}   \oint_{C_{n}^{\epsilon}} \left(\frac{\dot{F}^i(w)-\dot{F}^i(\infty)}{|w|^2} + \frac{2\big(\dot{F}^i(w)-\dot{F}^i(\infty)\big)}{|w|^2 \log|w|}\right) \dd{\bar{w}}\\
		&= 2\pi \sqrt{-1} \sum_{j=1}^{n} \left(1-\frac{1}{m_j}\right)^2 \partial_{w}\dot{F}^i(w_j) \quad \text{with} \quad m_j=\infty \quad \text{for} \quad j=n_e+1,\dots,n,
		\end{split}
		\end{equation*}
		Accordingly, the Eq.\eqref{variationISchottky5} is simplified to
		\begin{equation}\label{extra7}
		\begin{split}
		\dot{I}^i_{\epsilon} = -4 \sqrt{-1} \mathrm{b}_i -2\pi \sqrt{-1} \sum_{j=1}^{n} \left(1-\frac{1}{m_j}\right)^2 \partial_{w}\dot{F}^i(w_j)-4\pi \sqrt{-1}\sum_{j=1}^{n_e} \left(1-\frac{1}{m_j}\right)  \partial_{w}\dot{F}^{i}(w_j)
		\end{split}
		\end{equation}
		Finally, by putting everything together from \eqref{LiederivativeSdef}, \eqref{extra7}  and \eqref{Sct}, we get
		\begin{equation*}
		\begin{split}
		&\Lie_{\mu_i}S_{\boldsymbol{m}} = 2 \mathrm{b}_i + \pi \sum_{j=1}^{n} h_j \partial_{w}\dot{F}^i(w_j)  = -2 \pi c_i + 2 \sum_{j=1}^{n} h_j \left( \mathscr{E}_j , M_i\right) + \pi \sum_{j=1}^{n} h_j \partial_{w} \dot{F}^i(w_{j}) \\
		&\overset{\eqref{EMinnerprod}}{=}  -2 \pi c_i .
		\end{split}
		\end{equation*}
		\item
		The proof of part $(iii)$ follows readily from those of parts $(i)$ and $(ii)$.
	\end{itemize}
\end{proof}

\begin{maintheorem}\label{mainthrm2}
	The following statements are also true:
	\begin{enumerate}[(i)]
		\item The first Chern form of the Hermitian holomorphic line bundle $(\linebundle_{i},\Lponetial_{i}^{m_i})$  is given by 
		\begin{equation*}
		\chern{\linebundle_{i}}{\Lponetial_{i}^{m_i}} = \frac{m_i}{2\pi} \omega^{\text{ell}}_{\text{TZ},i}, \qquad i=1,\dots,n_e.
		\end{equation*}
		\item The first Chern form of the Hermitian holomorphic $\mathbb{Q}$-line bundle $(\linebundle,\exp[S_{\boldsymbol{m}}/\pi])$ is given by
		\begin{equation*}
		\chern{\linebundle}{exp[S_{\boldsymbol{m}}/\pi]} = \frac{1}{\pi^2} \omega_{\text{WP}}.
		\end{equation*}
		\item The function $\Gpotential_{\boldsymbol{m}} = S_{\boldsymbol{m}}-\pi \log \mathsf{H}$ satisfies 
		\begin{equation*}
		-\hspace{.5mm}\bar{\partial} \partial \Gpotential_{\boldsymbol{m}} = 2\sqrt{-1}\left(\omega_{\text{WP}}-\frac{4\pi^2}{3} \omega^{\text{\text{cusp}}}_{\text{TZ}}- \frac{\pi}{2} \sum_{j=1}^{n_e} m_j h_j \hspace{.5mm}\omega^{\text{\text{ell}}}_{\text{TZ},j}\right)
		\end{equation*}
	\end{enumerate}
\end{maintheorem}
\begin{proof}
	Since
	\begin{equation*}
	\chern{\linebundle_i}{\Lponetial_i^{m_i}} =\frac{\sqrt{-1}}{2\pi} \bar{\partial}\partial \log \Lponetial_{i}^{m_i},
	\end{equation*}	
	the proof of part \emph{(i)} is exactly the same as that of Lemma \ref{lemma:localpotentialsonmoduliwnogenus}. The proof of part \emph{(ii)} can be obtained by Theorem \ref{mainthrm1} and accordingly following similar analysis with the Theorem \ref{theorem2}. The last part immediately follows from \emph{(i)-(ii)}.
\end{proof}
\section{Discussion and some Future Directions}\label{Discussion}
This paper explores the semi-classical limit of Liouville field theory on Riemann orbisurfaces of finite conformal type $(g>1,n)$. This is accomplished through the use of Schottky global coordinates. This study can be seen as an extension of a previous work by Park, Takhtajan, and Teo \cite{park2015potentials} in which they examined the classical Liouville action on the Schottky space of compact Riemann surfaces with only punctures. In this paper, we include the contributions of orbifold points to the Liouville action. 

In Section~\ref{gg1}, we noted that the Liouville action is independent from a specific choice of fundamental domain. Corollary~\ref{cor:invariantaction} already demonstrates that $\Gpotential_{\boldsymbol{m}} = S_{\boldsymbol{m}} - \pi \log \mathsf{H}$ is not reliant on the choice of a fundamental domain for the Schottky group  $\Sigma$; however, this can also be verified using holography. The holographic principle for compact Riemann surfaces $X$ posits that the classical Liouville action on such surfaces is equal to the renormalized volume of a hyperbolic 3-manifold $M$, where that surface is its conformal boundary.  \cite{witten1998anti,Krasnov:2000zq,Takhtajan:2002cc}, extended this principle to punctured (for both Schottky and quasi-Fuchsian cases), and \cite{Teo_2018,chandra2022semiclassical} to orbifold Riemann surfaces (for only the quasi-Fuchsian case). However, this principle is not yet proven for 3-dimensional handlebody orbifolds with Riemann orbisurfaces as their conformal boundary. Thus, in an upcoming paper \cite{ANT23}, we aim to demonstrate that the Liouville action $\Gpotential_{\boldsymbol{m}}$  is linked to the renormalized volume of the corresponding hyperbolic 3-dimensional manifold $M$ up to some constants that do not rely on moduli parameters. This would also imply that the Liouville action is independent of the fundamental domain chosen. 

Although the presented proofs only apply to orbifold Riemann surfaces (i.e. conical points of angle $2 \pi/m_i$), we believe that most of our findings can be extended to (hyperbolic) conical Riemann surfaces with genus $g>1$. Specifically, for weighted punctured Riemann surfaces,  the modified classical Liouville action (with Schottky global coordinates) is expected to be accurately given by $\Gpotential_{\boldsymbol{\alpha}}$ with $\boldsymbol{\alpha}= (\alpha_1,\dots, \alpha_n)$ and the cone angles $2\pi(1-\alpha_i)$ where each factor of $(1-\tfrac{1}{m_i})$ in $\Gpotential_{\boldsymbol{m}}$ is replaced by $\alpha_i$. Let's talk about it in more detail here.

The weighted punctured Riemann surface $(X,\brdiv)$ is a compact Riemann surface $X$ together with an $\mathbb{R}$-divisor $\brdiv= \sum_{i=1}^{n} \alpha_i \,x_i$ such that the weights $0< \alpha_i \leq 1$ are associated to each puncture (or marked point) $x_i$. Hyperbolic metrics on weighted punctured Riemann surfaces have, by definition, conical singularities at the punctures --- for this reason, the pair $(X,\brdiv)$ is also called a Riemann cone-surface (or conical Riemann surface). The existence
and uniqueness of a conical hyperbolic metric with prescribed singularities at
a finite number of points on a Riemann surface is a classical problem that is
closely related (and in special cases, is equivalent) to the famous Uniformization
Problem of Klein and Poincar\'{e}. Let's remember that such metrics have been considered beginning with the work of Picard \cite{picard1893equation}. Starting from the classical results by Kazdan-Warner \cite{kazdan1972curvature,kazdan1974curvature,kazdan1975direct}, the existence and uniqueness of conical hyperbolic metrics in every conformal class on a $(X,\brdiv)$ were proved by McOwen \cite{McOwen-1988} and Troyanov \cite{Troyanov1991PrescribingCO}. The necessary and sufficient condition for the existence of a hyperbolic conical metric according to \cite{McOwen-1988,Troyanov1991PrescribingCO}, is that the statement of the Gauss-Bonnet theorem holds --- in other words, the degree of the log-canonical divisor $\mathcal{K}_X +\brdiv$ should be positive, where $\mathcal{K}_X$ denotes the canonical divisor of $X$. The positivity of the log-canonical divisor implies that $\sum_{i=1}^{n}\alpha_i >2$. In this case, the unique hyperbolic metric $ds^2$
on $X$ is said to be compatible with the divisor $\brdiv$ if $ds^2$ is a hermitian metric
of class $\mathcal{C}^{\infty}$ on the punctured Riemann surface $X_\text{reg} : X \backslash \text{supp}(\brdiv)$ such that if
$u_i$ is a holomorphic local coordinate in a neighborhood $U_i$ of $x_i$, then there
exists a real-analytic function $\varphi_i(u_i,\bar{u}_i)$, which is smooth on $U_i \backslash  \{x_i\}$, and
such that in $U_i$ the metric $ds^2$ is of the form
\begin{equation}\nonumber
ds^2 = \frac{e^{\varphi_i} |d u_i|^2}{|u_i-x_i|^{2\alpha_i}},
\end{equation}
for $0<\alpha_i <1$ and
\begin{equation}\nonumber
ds^2 = \frac{e^{\varphi_i} |d u_i|^2}{-2|u_i-x_i|^2 \log^2 |u_i-x_i|},
\end{equation}
for $\alpha_i =1$. Moreover, Area$(X,ds^2) = \pi \hspace{.5mm}\text{deg} (\mathcal{K}_X + \brdiv) = -\pi \chi(X,\brdiv)$, where $\chi(X,\brdiv) =\chi(X)-\text{deg}(\brdiv)$ is by definition
the Euler-Poincar\'{e} characteristic of the Riemann cone surface $(X,\brdiv)$.
The dependence of these hyperbolic cone metrics on the vector of weights
$\boldsymbol{\alpha}= (\alpha_1, \dots, \alpha_n)$ is characterized by \cite[Proposition 2.2]{Schumacher2008WeilPeterssonGF}. The (unique) conical metrics of constant negative curvature (with fixed weights) induce
new K\"{a}hler structures on the Teichm\"{u}ller spaces of punctured Riemann
surfaces that depend on the cone angles: In \cite{Takhtajan:2001uj}, Takhtajan and Zograf introduced a generalized Weil-Petersson metric, parameterized by the vector of weights
$\boldsymbol{\alpha}= (\alpha_1,\dots,\alpha_n)$, 
on the moduli spaces of n-pointed rational
curves in the context of Liouville actions.
Later, a variation of hyperbolic
conical metrics in holomorphic families was studied by Schumacher and Trapani
in \cite{Schumacher2005VariationOC,Schumacher2008WeilPeterssonGF} (see also \cite{mazzeo2017teichmuller,do2009weil,anagnostou2022volumes}). They showed that it is still possible to introduce "harmonic
Beltrami differentials" with respect to a hyperbolic conical metric, together
with a Kodaira-Spencer map needed for the notion of a Weil-Petersson metric. To be more precise, the
generalized Kodaira-Spencer map derived in \cite{Schumacher2005VariationOC,Schumacher2008WeilPeterssonGF} identifies the tangent
and cotangent spaces to Teichm\"{u}ller space $\teich_{g,n}(\boldsymbol{\alpha})$\footnote{The $\teich_{g,n}(\boldsymbol{\alpha})$
	is complex-analytically isomorphic to the Teichm\"{u}ller space of n-pointed
	stable curves $\teich_{g,n}$ (see e.g. \cite{Schumacher2005VariationOC,Schumacher2008WeilPeterssonGF}).} of Riemann cone-surfaces $(X,\brdiv)$ with appropriately defined
"harmonic Beltrami differentials" and "holomorphic quadratic differentials"
on the Riemann cone-surface. It turned out that for 
$0 < \alpha_i < 1$, 
the generalized Weil-Petersson metric depends
in a smooth monotone way on the weights $\alpha_i$. 
In particular, for $\alpha_i \rightarrow 0$
one recovers the Weil-Petersson metric for non-punctured surfaces, while for $\alpha_i \rightarrow 1$, 
one gets the Weil-Petersson metric for Riemann surfaces with cusps. Note that in the case of general conical singularities, the monodromy group $\Gamma$ of the Fuchsian differential equation is
no longer discrete in $\PSLR$; hence, $\teich_{g,n}(\boldsymbol{\alpha}) \neq \teich(\Gamma)$ (see \cite{Takhtajan:2001uj,takhtajan2014real} and \cite[\textsection{2.2}]{nag1988complex}). 
Moreover, one can  consider the variation of hyperbolic conical metrics
in holomorphic families, studied in \cite{Schumacher2005VariationOC,Schumacher2008WeilPeterssonGF}, and introduce a generalized
Takhtajan-Zograf (TZ) metric. 
In particular, one should study the behavior of the integral kernel of the Resolvant $(\Delta_0+1/2)^{-1}$ (see, e.g. \cite[\textsection{8}]{garbin2019spectral}). Although a conical singularity is simple, its presence has a profound impact on the Laplace
operator. Unlike surfaces with smooth boundaries, the Laplace operator no
longer has a canonical self-adjoint extension. Instead, it has many self-adjoint
extensions, and these need not be equivalent \cite{loya2007zeta}. Probably the most notable
self-adjoint extension of the Laplace-Beltrami operator on a Riemann surface
with conical singularities is given by the so-called Friedrichs extension (see,
e.g. \cite{kalvin2022determinant,kalvin2023determinants}). Using the same method as in \cite{Schumacher2005VariationOC,Schumacher2008WeilPeterssonGF}, one can
show that this generalized TZ metric depends in a smooth monotone way on the weights
$0 <\alpha_i < 1$ and that, for $\alpha_i = 1-1/m_i$ 
with integers $2 \leq m_i <\infty $,
it corresponds to the elliptic TZ metric introduced in \cite{ZT_2018}. One can also construct
a K\"{a}hler potential ($\sim\log \mathsf{H}_{\boldsymbol{\alpha}}$) for these generalized TZ metrics in terms of the solutions of Fuchsian differential equations. To be more precise, K\"{a}hler potentials for generalized TZ metric
should be defined in terms of the coefficients of the expansion \cite[Eq.(9)]{Takhtajan:2001uj}
(alternatively, in terms of the expansion of $\varphi$ as in \cite[Eq. (1.2) and (1.5)]{kalvin2023determinants}). This generalizes the results of
\cite{park2015potentials} to the case of Riemann cone surfaces. Lastly, regularized Liouville action $S_{\boldsymbol{\alpha}}$ should be defined similar to \cite{Takhtajan:2001uj} and the combination $\Gpotential_{\boldsymbol{\boldsymbol{\alpha}}} = S_{\boldsymbol{\alpha}} - \pi \log \mathsf{H}_{\boldsymbol{\alpha}}$ should define a function on   $\schottky_{g,n}(\boldsymbol{\alpha})$ (and $\symmoduli_{g,n}(\boldsymbol{\boldsymbol{\alpha}})$) (see also  \cite[Theorem 1.1]{kalvin2023determinants}). Moreover, the different asymptotic behavior should be derived using the Fuchsian differential equation in analogy with \cite[Lemma 4]{Takhtajan:2001uj}. This asymptotics can be used to derive the first and second variations of $\Gpotential_{\boldsymbol{\boldsymbol{\alpha}}} = S_{\boldsymbol{\alpha}} - \pi \log \mathsf{H}_{\boldsymbol{\alpha}}$. It is worth remembering that, when $g > 1$, the regularized Liouville action $S_{\boldsymbol{\alpha}}$ has to be defined on the Schottky fundamental domain with the singularities removed; in this case, the fibration $\jmath: \schottky_{g,n}(\boldsymbol{\alpha}) \to \schottky_g$ allows us to write $T_{\pi\circ\Phi(0)}^{\ast}\schottky_{g,n}(\boldsymbol{\alpha})$ as $\jmath^{\ast}(T_{\pi\circ\Phi(0)}^{\ast}\schottky_{g}) \oplus T_{\pi\circ\Phi(0)}^{\ast}\confspace{n}{\SchottkyFund}$. This decomposition makes it possible to use the methods of \cite{1988SbMat..60..297Z} on $\jmath^{\ast}(T_{\pi\circ\Phi(0)}^{\ast}\schottky_{g})$ while variations w.r.t. cone points should be carried out using a method similar
to \cite{Takhtajan:2001uj}. Accordingly, most of our findings can be extended to conical Riemann surfaces with genus $g>1$.

Finally, we mention some interpretations related to our results and also future directions:
\begin{itemize}
	\item 
	For the marked Riemann surface with $g>1$, one can define a normalized basis $v_1,\dots,v_g$ of the space of holomorphic 1-forms - abelian differentials of the first kind, and the period matrix $\boldsymbol{\tau}$,\footnote{The $\boldsymbol{\tau}$ is called period matrix.}
	\begin{equation*}
	\int_{\alpha_{k}} v_{k^\prime} =\delta_{k k ^\prime},\hspace{1cm}\boldsymbol{\tau}_{k k^\prime}=\int_{\beta_{k}}v_{k^\prime},
	\end{equation*}
	with
	\begin{equation*}
	\Im \boldsymbol{\tau}_{kk'}=\langle v_{k},v_{k'}\rangle=\frac{i}{2} \int_{\SchottkyFund}v_{k}\hspace{.5mm}\bar{v}_{k'}\hspace{1mm} e^{\phi(w)} dw d\bar{w}.
	\end{equation*}
	By defining also $\det' \Delta_0$ (zeta function regularized determinant of the Laplace operator in the hyperbolic metric $\exp(\varphi)|dw|^2$ acting on functions) and $S_{g} = S_{g,\boldsymbol{m}=0}$ as  function on the Schottky space $\mathfrak{S}_{g}$, Zograf \cite{Zograf_1990} shown that there exists a holomorphic function $\mathfrak{F}_{g}: \mathfrak{S}_{g}\rightarrow \cmpx$,  such that  
	\begin{equation}\label{Zograf}
	\frac{\det^{\prime}\Delta_0}{\det \Im\boldsymbol{\tau}} = c_{g} \, e^{-\frac{S_{g}}{12\pi}} \, |\mathfrak{F}_{g}|^{2},
	\end{equation}	
	with $c_g$ as a constant that depends just on the genus $g$ and 
	\begin{equation}\label{Fg}
	\mathfrak{F_{g}} = \prod_{\{\gamma\}}\prod_{k=0}^{\infty}\left(1- \hspace{1mm}q_{\gamma}^{1+k}\right),
	\end{equation} 
	where $\{\gamma\}$ runs over all the distinct primitive conjugacy classes in $\Gamma$ (i.e., $\gamma \in \Gamma$, that cannot be written as the power of any other element of $\Gamma$)\footnote{Or equivalently, the set of simple closed
		geodesics on the surface, and $\log \{\gamma\}$ is the corresponding geodesic length.}, excluding the identity, and $q_{\gamma}$ is the multiplier of $\gamma \in \Gamma$ ( $q_{\gamma}+1/q_{\gamma}=|\text{Tr}\gamma|$). Now, by comparing our result in part $(iv)$ of Theorem~\ref{mainthrm2},
	\begin{equation*}
	\bar{\partial} \partial \Gpotential_{\boldsymbol{m}} = -2\sqrt{-1}\left(\omega_{\text{WP}}-\frac{4\pi^2}{3} \omega^{\text{\text{cusp}}}_{\text{TZ}}- \frac{\pi}{2} \sum_{j=1}^{n_e} m_j h_j \hspace{.5mm}\omega^{\text{\text{ell}}}_{\text{TZ},j}\right),
	\end{equation*}	
	with \cite{ZT_2018}\footnote{In this case, one needs to extend the conditions $\int_{\alpha_{k}} v_{k^\prime} =\delta_{k k ^\prime}$ to $\int_{\{\alpha_{k},\kappa_k,\tau_k\}} v_{k^\prime} =\delta_{k k ^\prime}$.}
	\begin{equation*}
	\bar{\partial}\partial\hspace{1mm} \log \frac{\det \Im\boldsymbol{\tau}}{\det ^{\prime}\Delta_0}=-\frac{\sqrt{-1}}{6\pi}\left(\omega_{\text{WP}}-\frac{4\pi^2}{3} \omega^{\text{\text{cusp}}}_{\text{TZ}}- \frac{\pi}{2} \sum_{j=1}^{n_e} m_j h_j \hspace{.5mm}\omega^{\text{\text{ell}}}_{\text{TZ},j}\right),
	\end{equation*}
	we can generalize Zograf's formula \eqref{Zograf} using a function $\mathfrak{F}_{g,n}(\boldsymbol{m})$ on $\schottky_{g,n}(\boldsymbol{m})$, i.e. $\mathfrak{F}_{g,n}(\boldsymbol{m}): \schottky_{g,n}(\boldsymbol{m})\rightarrow \mathbb{C}$, such that
	\begin{equation}\label{gen}
	\frac{\det^{\prime} \Delta_0}{\det\Im\boldsymbol{\tau} } =c_{g,n}(\boldsymbol{m})\hspace{1mm} e^{-\frac{\Gpotential_{\boldsymbol{m}}}{12\pi}} \hspace{1mm}\big|\mathfrak{F}_{g,n}(\boldsymbol{m})\big|^2,
	\end{equation} 
	where $\Gpotential_{\boldsymbol{m}} = S_{\boldsymbol{m}} - \pi \log \mathsf{H}$ is the classical generalized Liouville action and $c_{g,n}(\boldsymbol{m})$ is a constant depending only on $g,n$ and $\boldsymbol{m}$. It would be interesting to find the exact form of the function $\mathfrak{F}_{g,n}(\boldsymbol{m})$ which its importance, particularly from the perspective of physics, will be clear in the following items.\footnote{The exact form of this holomorphic anomaly formula for the determinant of Laplacian on sphere with just three conical singularities and on singular genus zero surfaces that can be glued from copies of (hyperbolic, spherical, or flat) double triangles  have been found in \cite{kalvin2023determinants} and \cite{kalvin2023triangulations}, respectively. As a by-product, for these cases, the accessory parameters, the Liouville action, and $\log{\mathsf{H}}$ are also explicitly evaluated.  We would like to thank Victor Kalvin for informing us about \cite{kalvin2023triangulations}.}
	\item 
	Zograf's formula has an interesting geometric description in the context of the Quillen metric and local index theorem. The Quillen metric on $\lambda_1$ (the determinant line bundle associated with the Cauchy-Riemann operator $\bar{\partial}_1$, usually called the \emph{Hodge line bundle}) is defined by
	\begin{equation*}
	\left(\|v\|_{1}^{\text{Quill}}\right)^{2}
	= \frac{\det \Im \boldsymbol{\tau}}{\det^{\prime}\Delta_0},
	\end{equation*}
	and accordingly, the first Chern form of the Hermition line bundle $(\lambda_{1}, \|v\|_{1}^{\text{Quill}})
	$ over $\mathfrak{S}_{g}$ is given by
	\begin{equation*}
	c_{1}(\lambda_{1}, \|v\|_{1}^{\text{Quill}}) =\frac{1}{12\pi^2} \omega_{\text{WP}}.
	\end{equation*}
	This observation provides the existence of an isometry between the line bundle over $\symmoduli_{g}$ determined by carrying the Hermitian metric $\exp[S_{g}/12\pi]$ and the line bundle $\lambda_1$ with the Quillen metric. More generally, Takhtajan and Zograf have studied the local index theorem for families of $\bar{\partial}$-operators in the orbifold setting \cite{ZT_2018}. The main result of this paper (see  \cite[Theorem~2]{ZT_2018}) is the following formula on the moduli space $\symmoduli_{g,n}(\boldsymbol{m})$ of punctured orbifold Riemann surfaces $O=[\UHP\slash\Gamma]$:
	\begin{equation}\label{Quillen}
	\chern{\lambda_k}{\|\cdot\|_{k}^{Quill}} \hspace{-1mm}=\hspace{-1mm} \frac{6k^2 - 6k +1}{12 \pi^2} \omega_{WP} - \frac{1}{9} \omega_{cusp} -\frac{1}{4\pi}\hspace{-1mm}\sum_{j=1}^{n_e} m_j\hspace{-1mm} \left[B_2\hspace{-1mm}\left(\left\{\frac{k-1}{m_j}\right\}\right)\hspace{-1mm} - \frac{1}{6 m_j^2}\right]\hspace{-1mm} \omega^{ell}_j,
	\end{equation}
	for $k \geq 1$. Here $\lambda_k$ is the determinant line bundle associated with the Cauchy-Riemann operator $\bar{\partial}_k$ and is a holomorphic $\modular(\Gamma)$-invariant line bundle on $\teich(\Gamma)$ whose fibers are given by the determinant lines $\bigwedge^{\max} \ker \bar{\partial}_k \otimes \left(\bigwedge^{\max} \operatorname{coker} \bar{\partial}_k\right)^{-1}$ while $\|\cdot\|_{k}^{\text{Quill}}=\|\cdot\|_{k}/\sqrt{\det \Delta_k}$ denotes the \emph{Quillen norm} in $\lambda_k$;\footnote{Hermition line bundles $\left(\lambda_k,\|\cdot\|_{k}^{\text{Quill}}\right)$ and $\left(\lambda_{1-k},\|\cdot\|_{1-k}^{\text{Quill}}\right)$ are isometrically isomorphic.} since the determinant line bundle $\lambda_k$ is $\modular$-invariant, one can alternatively think of $\lambda_k$ as a holomorphic $\mathbb{Q}$-line bundle on the moduli space $\symmoduli_{g,n}(\boldsymbol{m}) = \teich(\Gamma)/\modular(\Gamma)$. Finally, $B_2(x) = x^2 - x + 1/6$ is the second Bernoulli polynomial, and $\{x\}$ denotes the fractional part of $x \in \mathbb{Q}$. It is clear from Eq.\eqref{Quillen} that for $k=1$, the first Chern form of  the determinant line bundle $\lambda_1$ with Quillen norm is given by  $\tfrac{1}{12\pi^2}(\omega_{\text{WP}}-\frac{4\pi^2}{3} \omega^{\text{\text{cusp}}}_{\text{TZ}}- \frac{\pi}{2} \sum_{j=1}^{n_e} m_j h_j \omega^{\text{\text{ell}}}_{\text{TZ},j})$. Comparing this observation with our result in part $(iv)$ of Theorem \ref{mainthrm2} and noting to \cite[Theorem~3.1]{Zograf_1990}, suggests that the following remark is correct:
	\begin{remark}\label{mainthrm3}
		The function $\Gpotential_{\boldsymbol{m}}$ on the Schottky space $\schottky_{g,n}(\boldsymbol{m})$ determines a holomorphic $\mathbb{Q}$-line bundle $\lambda_{_{Sch}}$ on the moduli space $\symmoduli_{g,n}(\boldsymbol{m})$ with Hermitian metric $\langle \cdot , \cdot \rangle_{_{Sch}}$, where $\langle 1 , 1 \rangle_{_{Sch}} = \exp(\Gpotential_{\boldsymbol{m}}/12 \pi)$  (here, $1$ is understood as the corresponding section of the trivial bundle $\schottky_{g,n}(\boldsymbol{m}) \times \cmpx \to \schottky_{g,n}(\boldsymbol{m})$). The Hermitian $\mathbb{Q}$-line bundle $(\lambda_{_{Sch}} ; \langle \cdot , \cdot \rangle_{_{Sch}})$ is isometrically isomorphic to the Hodge line bundle $(\lambda_{_{Hod}} ; \langle \cdot , \cdot \rangle_{_{Quil}})$ over $\symmoduli_{g,n}(\boldsymbol{m})$ --- i.e. there exists an isomorphism $\imath: \lambda_{_{Sch}} \to  \lambda_{_{Hod}}$ such that $\langle s , s \rangle_{_{Sch}} = \langle \imath \circ s , \imath \circ s \rangle_{_{Quil}}$ for every local section $s$ of the bundle $\lambda_{_{Sch}}$.
	\end{remark}

	While we do not attempt proving the above claim rigorously, we would like to comment that it might be possible to do so in analogy with the proof of \cite[Theorem~3.1]{Zograf_1990}. In particular, a very interesting question is to investigate whether it would be possible to determine the constants $c_{\gamma}$ and $c$ in the following naive generalization
	\begin{equation}\label{ConjecturedCocyle}
		f_{\gamma}(t) = \exp(c \int_{t_{\ast}}^{t} \partial(\tilde{\Gpotential}_{\boldsymbol{m}}-\tilde{\Gpotential}_{\boldsymbol{m}} \circ \gamma) + \frac{c}{2} \left(\tilde{\Gpotential}_{\boldsymbol{m}}(t_{\ast}) - \tilde{\Gpotential}_{\boldsymbol{m}}(\gamma t_{\ast}) + 2\pi \sqrt{-1} c_{\gamma}\right) ),
	\end{equation} 
	of $f_{\gamma}$s constructed in \cite{Zograf_1990}, such that Eq.~\eqref{ConjecturedCocyle} still defines a 1-cocycle of the action of Teichm\"uller modular group $\modular(\Gamma)$ (see \cite{Zograf_1990} for more details). In this formula, $\gamma \in \modular(\Gamma)$ denotes mapping classes, $\tilde{\Gpotential}_{\boldsymbol{m}}$ is defied as $\Gpotential_{\boldsymbol{m}} \circ \pi$ where $\pi: \teich_{g,n}(\boldsymbol{m}) \to \schottky_{g,n}(\boldsymbol{m})$ is a natural holomorphic cover, $t_{\ast}$ denotes a marked point in the Teichm\"uller space $\teich_{g,n}(\boldsymbol{m})$ while $t$ is any other point in this space, and $\partial$ denotes the (1,0)-component of the exterior differentiation operator on $\teich_{g,n}(\boldsymbol{m})$.
	
	It is worth mentioning that the Determinant Line Bundles and Quillen Metrics in the conical case have been
	constructed in \cite[\textsection{6}]{Schumacher2008WeilPeterssonGF}. The curvature tensor of the Weil-Petersson metric for Teichm\"{u}ller spaces of compact (or punctured) Riemann surfaces was computed explicitly by Tromba \cite{tromba1986natural} and Wolpert \cite{wolpert1986chern}. In \cite[\textsection{7}]{Schumacher2008WeilPeterssonGF}, the authors show the analogous result for the weighted punctured case.
	
	\item
	The function $\mathfrak{F}_g$ has an important rule to find the physical wave-function of 3-dimensional pure AdS quantum gravity and define the last ingredient of canonical quantization, i.e., the scalar product between the physical wave-functions: the physical norm should be invariant under the mapping class group transformations. As we will see, it is achievable via the Quillen norm as follows. Let us recall that according to Eq.\eqref{Zograf},
	\begin{equation*}
	\frac{\det^{\prime}\Delta_0}{\det \Im\tau} = c_{g} \, e^{-\frac{S_{g}}{12\pi}} \, |\mathfrak{F}_{g}|^{2}.
	\end{equation*} 
	The determinant of the operator $\Delta_0$ is invariant under large diffeomorphisms, and let's assume that under $\modular$ transformation, we have
	\begin{equation}\label{e1}
	\det\Im\tau\rightarrow \det\Im \tau^{\prime} = |\mathfrak{A}|^2\hspace{1mm}
	\det\Im\tau,\hspace{1cm} S_{g} \rightarrow S'_{g}=  S_{g}+\mathfrak{B}+\bar{\mathfrak{B}},
	\end{equation}
	which imply that the function $\mathfrak{F}_g$ should transforms as follows 
	\begin{equation*}
	\mathfrak{F}_g \rightarrow \mathfrak{F}_g^{'}= \frac{e^{\mathfrak{B}}}{\mathfrak{A}}\hspace{1mm}\mathfrak{F}_g.
	\end{equation*}
	Now, if the matrix $U_{i}^{j}$, under large diffeomorphism,
	relates wave functions defined in two coordinate systems, i.e. $\Psi^{\prime ^i} = U^{i}_{j}\hspace{.5mm}e^{-\mathfrak{B}\mathsf{D}}\hspace{.5mm}\Psi^{j}$, then the functions
	\begin{equation}\label{e2}
	\tilde{\Psi}^{i} = \left(\mathfrak{F}_g\right)^{\mathsf{D}}\hspace{.1mm}\Psi^{i},
	\end{equation}
	will transform as 
	\begin{equation}\label{e3}
	\tilde{\Psi}^{i} \rightarrow \tilde{\Psi}^{\prime i} =  \mathfrak{A}^{-\mathsf{D}}\hspace{.5mm}U^{i}_{j}\hspace{.5mm}\tilde{\Psi}^{j}.
	\end{equation}
	By noting the equations \eqref{e1}, \eqref{e2} and \eqref{e3}, the MCG-invariant scalar product of physical wave functions $\tilde{\Psi}$ can be defined based on Quillen norm (see \cite{kim2015canonical,maloney2015geometric} as well as \cite{eberhardt2022off})
	\begin{equation*}
	\langle \tilde{\Psi}_1,\tilde{\Psi}_2\rangle = \int_{\teich(\Sigma)} \left(\frac{\det^{\prime}\Delta_0}{\det\Im\tau}\right)^{-\mathsf{D}}\hspace{1mm} \overline{\tilde{\Psi}}_1\wedge\ast\tilde{\Psi}_2.
	\end{equation*}
	Accordingly, the function  $\mathfrak{F}_{g}$ determines the physical wave functions of 3-dimensional pure quantum gravity and their MCG-invariant scalar product.
	
	In light of the aforementioned observation and taking into account the remarks in the Introduction regarding the inclusion of massive particles in the path integral of 3-dimensional gravity (or considering 3-dimensional Seifert manifolds whose KK-reductions are related to conical Riemann surfaces), the function $\mathfrak{F}_{g,n}(\boldsymbol{m})$ plays crucial rule in defining the physical wave-functions (and their scalar product) of 3-dimensional quantum gravity in presence of those massive particles (or in presence of  Seifert manifold's contribution). This itself presents an intriguing problem warranting further exploration. Apart from that, $\mathfrak{F}_{g}$ plays another significant role in the connection between the physical wave function of 3-dimensional pure gravity and dual two-dimensional Virasoro conformal blocks. According to \cite{teschner2003liouville}, the 3-dimensional physical
	wave-functions, obeying the Gauss law constraints, are 
	conformal blocks of quantum Liouville theory. Exploring this proposal in the presence of conical singularities, based on our results, is also an interesting problem, i.e., finding the relation between the wave-function of 3-dimensional AdS gravity in the presence of particles with special masses and quantum Liouville theory with some special vertex operators and even more to determine whether the partition function of gravity can indeed be decomposed into Virasoro characters. 
	\item 	
	Our results are applicable to two-dimensional theories of gravity such as deformed Jackiw-Teitelboim gravity \cite{witten2020matrix,maxfield2021path,alishahiha2021free,turiaci20212d,eberhardt20232d,artemev2023p,lin2023revisiting}, where the gravitational path integral can be written as an integral over the moduli space of orbifold (conical) Riemann surfaces.
	\item
	Our results, via an analytic continuation of the classical Matschull process, have potential application in the study of black hole production with a non-trivial topology inside the horizon from the collision of massive point particles with a certain mass   \cite{Krasnov:2001ui,Krasnov:2002rn, lindgren2016black}.\footnote{Also see \cite{Chen:2016dfb,Chen:2016kyz,Cho:2017fzo,Cardy:2017qhl,Can:2017ycp,Collier:2019weq,Belin:2021ibv} and \cite[\textsection4]{chang2016bootstrap}.} 
	
	\item 
	There is a rich mathematical
	theory- that of the Selberg trace formula and its generalizations - where the sum over
	elements $\gamma \in \Gamma$ is used to compute the spectrum of differential operators on
	$\UHP/\Gamma$. The Selberg zeta function $Z_{\text{Sz}}(s,\Gamma,\mathfrak{U})$ associated with the Riemann surface is
	\begin{equation}\label{Z}
	Z_{\text{Sz}}(s,\Gamma,\mathfrak{U}) =\prod_{\{\gamma\}}\prod_{k=0}^{\infty}\left(1-\mathfrak{U}(\gamma) \hspace{1mm}q_{\gamma}^{s+k}\right),
	\end{equation}
	where $\mathfrak{U}:\Gamma \rightarrow U(1)$ be a unitary character. For $\Re s>1$, the above product admits a meromorphic continuation to the complex $s$-plane. According to the Selberg trace formula, for the case where $\Gamma$ has just hyperbolic and parabolic elements, it is shown \cite{d1986determinants} that 
	\begin{equation*}
	\det\hspace{.01mm}^{\prime}\hspace{.5mm} \Delta_0 \simeq Z_{\text{Sz}}^{\prime}(1),
	\end{equation*}
	where $Z_{\text{Sz}}(s) \equiv Z_{\text{Sz}}(s,\Gamma,\mathfrak{U}=1)$. Hence, Zograf's formula \eqref{Zograf} also gives a factorization of $Z_{\text{Sz}}^{\prime}(1)$ as a function on $\mathfrak{S}_{g}$. Accordingly, it would be intriguing to find a function akin to the Selberg zeta function $Z_{\text{Sz}}(s)$ acting on $\schottky_{g,n}(\boldsymbol{m})$, utilizing the generalized holomorphic factorization theorem \eqref{gen}. That is important, at least, for finding the contribution of 1-loop effects to the partition function of 3-dimensional AdS gravity in the presence of massive particles. Let us remind that the one-loop partition function of pure three-dimensional gravity is given by
	\begin{equation*}
	Z^{1-\text{loop}}_{\text{gravity}}=\frac{\det\Box^{(1)}}{\sqrt{\det\Box^{(2)}\hspace{1mm}\det\Box^{(0)}}},
	\end{equation*}
	where $\Box^{(2)}$ is the kinetic operator for linearized graviton (symmetric traceless) fluctuations around the chosen background and $\Box^{(1)},\Box^{(0)}$ are related to kinetic operators of vector ghost and weyl mode, respectively, since to compute the partition function we must also include the Fadeev-Popov modes arising due to gauge fixing. By computing each one concretely, it is shown that \cite{giombi2008one} 
	\begin{equation}\label{1loop}
	Z^{1-\text{loop}}_{\text{gravity}}= \prod_{\{\gamma\}}\prod_{k=0}^{\infty}\frac{1}{\big{|}1-q_{\gamma}^{2+k}\big{|}}.
	\end{equation}
	When $M= \UHP_{3}/\Gamma$ is a solid torus, $\Gamma =\mathbb{Z}$, and
	$\{\gamma\}$ consists of the generator of $\Gamma$ and its inverse, the above formula reduces to
	\begin{equation*}
	Z^{1-\text{loop}}_{\text{gravity}}(\tau,\bar{\tau}) = \prod_{k=2}^{\infty}\frac{1}{|1-q^k|^2}, 
	\end{equation*}
	with $q=e^{2\pi i \tau}$.\footnote{
		Unlike the $\UHP_{3}/\mathbb{Z}$,
		the higher-loop contributions to the partition function may not vanish for the general case $M = \UHP_3/\Gamma$. At least for the handlebody geometries, there exists a proposal to calculate all-loop expressions \cite{yin2007partition}.} Comparing \eqref{1loop} with \eqref{Z} implies that
	\begin{equation*}
	Z^{1-\text{loop}}_{\text{gravity}}= Z^{-1}_{\text{Sz}}(2).
	\end{equation*} 
	Let us also remind that, according to \eqref{Fg}, $\mathfrak{F_{g}}=Z_{\text{Sz}}(1)$. Hence, with knowledge of the function $\mathfrak{F}_{g,n}(\boldsymbol{m})$, we can investigate the appropriate function that assumes the role of Selberg zeta function $Z_{\text{Sz}}(s)$ in the presence of conical singularities. Consequently, this exploration can also provide insights into the contribution of conical singularities to the 1-loop gravity partition function. For the special case $M = \UHP_{3}/(\mathbb{Z}\times\mathbb{Z}_{m})$, it is shown  \cite{benjamin2020pure} that in the presence of conical singularities, the 1-loop partition function again is given by $Z_{\text{Sz}}(s)$ but evaluated at a different point, i.e. $Z^{1-\text{loop}}_{\text{gravity}}\sim Z^{-1}_{\text{Sz}}(1)$.	
	\item 
	In this paper, we explored further the relation between the modified classical Liouville action (i.e., the one which its Euler-Lagrange equation admits the hyperbolic Riemann surfaces with conical singularities as a solution), uniformization of orbifold Riemann surfaces and complex geometry of moduli spaces. One knows that the quantized Liouville theory can potentially describe the quantum corrections to those hyperbolic geometries (outside the singularities). Accordingly, one approach to understanding two-dimensional quantum gravity through quantum Liouville theory is to demand that this theory takes advantage of conformal symmetry akin to its classical counterpart. If the conformal symmetry is also symmetry of quantum Liouville theory, it should show itself in the conformal Ward identities (CWIs) for correlation functions of components of stress-energy tensor with another operators. An important point to consider is the space on which the Liouville action functional is expressed. This choice significantly impacts the form of the conformal Ward identities (CWIs) and their implications.
	
	Let's consider the simplest case, i.e., the Liouville action on the moduli space $\moduli_{0,n}$
	and identifying $X_{\boldsymbol{m}}$ with  $V_{m_{1}}(x_1) \dotsm V_{m_{n}}(x_n)$ and restoring the $\hbar$. Since the $S_{\boldsymbol{m}}$ in \eqref{NoGenusAction} is a well-defined (single-valued) function on $\moduli_{0,n}$, the quantum Liouville theory is defined by 
	\begin{equation}\label{PI1}
	\expval{ X_{\boldsymbol{m}}} = \int\limits_{\mathscr{CM}_{\boldsymbol{m}}(X)} \hspace*{-15pt}\mathcal{D}\psi \,\, e^{-\frac{1}{2 \pi\hbar} S_{\boldsymbol{m}}[\psi]},
	\end{equation}
	where the partition function $\expval{ X_{\boldsymbol{m}}}$ is a real-valued function on the $\moduli_{0,n}$. Moreover, we have
	\begin{equation}\label{TX1}
	\expval{\boldsymbol{T}(w) X_{\boldsymbol{m}}} = \int\limits_{\mathscr{CM}_{\boldsymbol{m}}(X)} \hspace*{-15pt}\mathcal{D}\psi \,\, \boldsymbol{T}(\psi)(w)\hspace{1mm}e^{-\frac{1}{2 \pi\hbar} S_{\boldsymbol{m}}[\psi]}.
	\end{equation} 
	To have conformal symmetry, one is required to prove
	\begin{equation}\label{CWI}
	\frac{1}{\hbar}\expval{\boldsymbol{T}(w) X_{\boldsymbol{m}}} = \sum_{i=1}^{n} \left(\frac{\boldsymbol{h}_{m_i}(\hbar) }{(w-w_i)^2}+\frac{1}{(w-w_i)}\partial_{w_i}\right)\expval{X_{\boldsymbol{m}}},
	\end{equation}
	where $\boldsymbol{h}_{m_i}(\hbar)$ are considered as conformal dimensions of vertex operators 
	$V_{m_i}= e^{\frac{\alpha(m_i)}{\hbar} \psi}$. Note that at the tree level when $\hbar \rightarrow 0$,
	\begin{equation*}
	\begin{split}
	& \expval{ X_{\boldsymbol{m}}} \sim e^{-\frac{1}{2\pi\hbar}S_{\boldsymbol{m}}[\varphi]},\\
	&\expval{\boldsymbol{T}(w) X_{\boldsymbol{m}}}\sim T_{\varphi}(w) \hspace{1mm}e^{-\frac{1}{2\pi\hbar}S_{\boldsymbol{m}}[\varphi]},\\
	&\boldsymbol{h}_{m_i}(\hbar)\sim \frac{h_{\text{cl}}(m_i)}{2\hbar}.
	\end{split}
	\end{equation*}	
	By substituting the above relations in \eqref{CWI}, one gets \cite{Takhtajan:1994vt,Takhtajan:1994zi}
	\begin{equation*}
	\frac{1}{\hbar}\hspace{1mm}T_{\varphi}(w) \hspace{1mm}e^{-\frac{1}{2\pi\hbar}S_{\boldsymbol{\boldsymbol{m}}}[\varphi]} = \sum_{i=1}^{n}\left(\frac{h_{\text{cl}}(m_i)}{2\hbar(w-w_i)^2}+\frac{1}{(w-w_i)}\partial_{w_i}\right)e^{-\frac{1}{2\pi\hbar}S_{\boldsymbol{\boldsymbol{m}}}[\varphi]},
	\end{equation*}
	which implies that
	\begin{equation}\label{T1}
	T_{\varphi}(w)\hspace{1mm}= \sum_{i=1}^{n}\left(\frac{h_{\text{cl}}(m_i)}{2(w-w_i)^2}-\frac{1}{2\pi}\frac{1}{(w-w_i)}\partial_{w_i}S_{\boldsymbol{\boldsymbol{m}}}[\varphi]\right).
	\end{equation}
	Comparing the above result with $T_{\varphi}(w) = \text{Sch}(J^{-1};w)$, where $\text{Sch}(J^{-1};w)$ is given by \eqref{EnergyMomentumExpansion}, implies that
	\begin{equation}\label{res1}
	h_{\text{cl}}(m_i)= h_i = 1-\frac{1}{m_i^2},\hspace{1cm}\partial_{w_i} S_{\boldsymbol{m}} = -2\pi c_i,
	\end{equation}
	which are in agreement with Theorem~\ref{theorem1}. 
	
	Moreover, the conformal symmetry at the quantum level implies that the vertex operators and components of the energy-momentum tensor satisfy the operator product expansion (OPE) of Belavin-Polyakov-Zamolodchikov, for example,
	\begin{equation*}
	\begin{split}
	&\frac{1}{\hbar^2}\boldsymbol{T}(w)\boldsymbol{T}(w^\prime)=\frac{\mathsf{c}(\hbar)/2}{(w-w^{\prime})^4}+\frac{1}{\hbar}\frac{2\boldsymbol{T}(w^{\prime})}{(w-w^{\prime})^2}+\frac{1}{\hbar}\frac{1}{(w-w^{\prime})}\partial_{w^{\prime}}\boldsymbol{T}(w^{\prime})+\text{regular terms},\\
	&\frac{1}{\hbar^2}\boldsymbol{T}(w)\overline{\boldsymbol{T}}(\bar{w}^\prime)=\text{regular terms}
	\end{split}
	\end{equation*}
	which yields the following CWIs,
	\begin{equation}\label{CWI2}
	\begin{split}
	&\frac{1}{\hbar^2}\expval{\boldsymbol{T}(w)\boldsymbol{T}(w^{\prime}) X_{\boldsymbol{m}}} =\frac{\mathsf{c}(\hbar)/2}{(w-w^{\prime})^4}\expval{X_{\boldsymbol{m}}} +\frac{1}{\hbar}\left(\frac{2}{(w-w^{\prime})^2}+\frac{1}{(w-w^{\prime})}\partial_{w^{\prime}}\right)\expval{\boldsymbol{T}(w^{\prime})X_{\boldsymbol{m}}}\\
	&\hspace{2cm}+\frac{1}{\hbar} \sum_{i=1}^{n} \left(\frac{\boldsymbol{h}_{m_i}(\hbar) }{(w-w_i)^2}+\frac{1}{(w-w_i)}\partial_{w_i}\right)\expval{\boldsymbol{T}(w^{\prime})X_{\boldsymbol{m}}},\\
	&\frac{1}{\hbar^2}\expval{\boldsymbol{T}(w)\overline{\boldsymbol{T}}(\bar{w}^{\prime}) X_{\boldsymbol{m}}} =\frac{1}{\hbar} \sum_{i=1}^{n} \left(\frac{\boldsymbol{h}_{m_i}(\hbar) }{(w-w_i)^2}+\frac{1}{(w-w_i)}\partial_{w_i}\right)\expval{\overline{\boldsymbol{T}}(\bar{w}^{\prime}) X_{\boldsymbol{m}}}.
	\end{split}
	\end{equation}
	Let us show that our results also agree with the second CWI in \eqref{CWI2} at the tree-level. To see that, let us define the normalized connected two-point function
	\begin{equation}\label{nc}
	\expval{\boldsymbol{T}(w)\overline{\boldsymbol{T}}(\bar{w}^{\prime}) X_{\boldsymbol{m}}}_{\text{nc}} = \frac{\expval{\boldsymbol{T}(w)\overline{\boldsymbol{T}}(\bar{w}^{\prime}) X_{\boldsymbol{m}}}}{\expval{X_{\boldsymbol{m}}}}-\frac{\expval{\boldsymbol{T}(w) X_{\boldsymbol{m}}}}{\expval{X_{\boldsymbol{m}}}}\frac{\expval{\overline{\boldsymbol{T}}(\bar{w}^{\prime}) X_{\boldsymbol{m}}}}{\expval{X_{\boldsymbol{m}}}}.
	\end{equation}
	By using \eqref{CWI} and its similar expression for $\expval{\overline{\boldsymbol{T}}(\bar{w}^{\prime}) X_{\boldsymbol{m}}}$,
	\begin{equation*}
	\frac{1}{\hbar}\expval{\overline{\boldsymbol{T}}(\bar{w}^{\prime}) X_{\boldsymbol{m}}} = \sum_{j=1}^{n} \left(\frac{\bar{\boldsymbol{h}}_{m_j}(\hbar) }{(\bar{w}^{\prime}-\bar{w}^{\prime}_j)^2}+\frac{1}{(\bar{w}^{\prime}-\bar{w}^{\prime}_j)}\partial_{\bar{w}^{\prime}_j}\right)\expval{X_{\boldsymbol{m}}},
	\end{equation*}
	the normalized connected correlation function \eqref{nc} is simplified to
	\begin{equation}\label{CWI4}
	\expval{\boldsymbol{T}(w)\overline{\boldsymbol{T}}(\bar{w}^{\prime}) X_{\boldsymbol{m}}}_{\text{nc}}=\sum_{i=1}^{n}\sum_{j=1}^{n}\frac{1}{(w-w_i)}\frac{1}{(\bar{w}^{\prime}-\bar{w}_j^{\prime})}\partial_{w_i}\partial_{\bar{w}^{\prime}_j}\log \expval{X_{\boldsymbol{m}}}.
	\end{equation}
	At the tree-level, the right-hand side of the above equation has an order of $\mathcal{O}(\hbar^{-1})$. To find the same order of $\hbar$ on the left-hand side, one can use the path integral representation of the two-point function,
	\begin{equation}\label{TP}
	\expval{\boldsymbol{T}(w)\overline{\boldsymbol{T}}(\bar{w}^{\prime}) X_{\boldsymbol{m}}} = \int\limits_{\mathscr{CM}_{\boldsymbol{m}}(X)} \hspace*{-15pt}\mathcal{D}\psi \,\, \boldsymbol{T}(\psi)(w)\overline{\boldsymbol{T}}(\psi)(\bar{w}^{\prime})\hspace{1mm}e^{-\frac{1}{2 \pi\hbar} S_{\boldsymbol{m}}[\psi]}
	\end{equation}
	Let us expand the field $\psi$ around its classical value $\varphi$,
	\begin{equation*}
	\psi(w) = \varphi(w)+\sqrt{\pi\hbar}\hspace{.5mm}\delta\psi(w)+...
	\end{equation*}
	which implies that
	\begin{equation*}
	\begin{split}
	&S_{\boldsymbol{m}}[\psi]= S_{\boldsymbol{m}}[\varphi]+\pi\hbar \int \delta\psi \left(\Delta_0+\frac{1}{2}\right)\delta\psi\hspace{1mm}e^{\varphi}d^2w+...\\
	&\frac{1}{\hbar}\boldsymbol{T}(w) =\frac{1}{\hbar}T_{\varphi}(w)+\sqrt{\frac{\pi}{\hbar}}\hspace{1mm}\mathsf{D}_w\delta\psi(w)+... 
	\end{split}
	\end{equation*}
	with $\mathsf{D}_w=\left(\partial_{w}^2-(\partial_w \varphi)\partial_w\right)$. Now, substituting the above expansion in the right-hand side of \eqref{TP}, the CWI \eqref{CWI4} at the tree-level becomes
	\begin{equation}\label{CWI5}
	\frac{2\pi}{\hbar}\mathsf{D}_{w}\mathsf{D}_{\bar{w}^{\prime}} \hspace{.5mm}G(w,w^{\prime})=-\frac{1}{2\pi\hbar}\sum_{i=1}^{n}\sum_{j=1}^{n}\frac{1}{(w-w_i)}\frac{1}{(\bar{w}^{\prime}-\bar{w}_j^{\prime})}\hspace{.5mm}\partial_{w_i}\partial_{\bar{w}^{\prime}_j}S_{\boldsymbol{m}}[\varphi],
	\end{equation}
	where $G = 1/2\hspace{1mm}(\Delta_0+1/2 )^{-1}$ is the propagator of quantum Liouville theorem, $\Delta_0= -\exp{-\varphi}\hspace{1mm}\partial^2/\partial_w \partial_{\bar{w}}$. Moreover, by using the overall $SL(2,\mathbb{C})$-symmetry to normalize $w_{n-2}=0,w_{n-1}=1,w_n=\infty$, the right hand side can be written as
	\begin{equation}\label{RHS}
	\begin{split}
	\text{RHS} &= -\frac{1}{2\pi\hbar}\sum_{i=1}^{n-3}\sum_{j=1}^{n-3}R(w,w_i)R(\bar{w}^{\prime},\bar{w}^{\prime}_j)\hspace{.5mm}\partial_{w_i}\partial_{\bar{w}^{\prime}_j}S_{\boldsymbol{m}}[\varphi],\\
	&=-\frac{\pi}{2\hbar}\sum_{i=1}^{n-3}\sum_{j=1}^{n-3}R_i(w)R_j(\bar{w}^{\prime})\hspace{.5mm}\partial_{w_i}\partial_{\bar{w}^{\prime}_j}S_{\boldsymbol{m}}[\varphi],
	\end{split}
	\end{equation}
	where $R(w,w_i)$ and $R_i$ are defined in \eqref{Rab} and \eqref{Ridef}, respectively. Therefore, according to Theorem \ref{mainthrm2}, the RHS of \eqref{CWI5} is (1,1)-component of the curvature form for the connection in the holomorphic line bundle over the moduli.
	Through rigorous mathematical computations \cite{Takhtajan:1994vt,Takhtajan:1994zi}, it can be demonstrated that the left-hand side of equation \eqref{CWI5} encompasses precisely the identical information. However, a streamlined approach can be employed to substantiate this claim by using the principles of Friedan-Shenker modular geometry \cite{Friedan:1986ua}. Friedan and Shenker interpret the expectation value $\expval{X}$ as a Hermitian metric in a certain holomorphic line bundle over moduli space and the quadratic differential
	$\expval{T(w)X}_{nc} dw^2$  as a (1,0)-component of the canonical metric connection. This interpretation aligns with the tree-level analysis of the CWI \eqref{CWI} above. Moreover, they interpret $\expval{\boldsymbol{T}(w)\overline{\boldsymbol{T}}(\bar{w}^{\prime}) X_{\boldsymbol{m}}}_{\text{nc}}dw d\bar{w}$ as (1,1)-component of the curvature form of that connection;  as illustrated in Eq.~\eqref{RHS}, this is also in agreement with the tree level analysis of \eqref{CWI4} for the moduli space $\moduli_{0,n}$. 
	
	Now, One can utilize the aforementioned terminology to find new constraints on K\"{a}hler geometry of moduli space $\symmoduli_{0,n}(\boldsymbol{m})$ by demanding that the CWIs hold true when they are projected to this space. Conversely, one can also  use the known information about K\"{a}hler geometry to find the ``dynamical proof'' of the Virasoro  symmetry of the Liouville theory. Let's find one of those constraints by studying the quantum Liouville theory and its associated CWIs on the moduli space $\symmoduli_{0,n}(\boldsymbol{m})$. In comparison to the previous case, two points should be highlighted. Firstly, to define the path integral representation of quantum Liouville theory, one needs a Liouville functional that is a well-defined (single-valued) function on the moduli space $\symmoduli_{0,n}(\boldsymbol{m})$. It is shown that, unlike the function $S_{\boldsymbol{m}}$ that is not invariant under the action of $\symm{\boldsymbol{m}}$, at least at the semi-classical limit, the function $\Gpotential_{\boldsymbol{m}}$  has this property. Accordingly, for the moduli space $\symmoduli_{0,n}(\boldsymbol{m})$, instead of \eqref{PI1}, we have 
	\begin{equation}\label{PI2}
	\expval{ X_{\boldsymbol{m}}} = \int\limits_{\mathscr{CM}_{\boldsymbol{m}}(X)} \hspace*{-15pt}\mathcal{D}\psi \,\, e^{-\frac{1}{2 \pi\hbar} \Gpotential_{\boldsymbol{m}}[\psi]}.
	\end{equation}
	Secondly, to check the counterpart of \eqref{CWI} for the moduli space $\symmoduli_{0,n}(\boldsymbol{m})$, the automorphic form $\text{Sch}\left(J^{-1};w\right)$ should be projected to the subspace $ T_{[0]}^{\ast}\symmoduli_{0,n}(\boldsymbol{m})$.\footnote{If we define the mapping $\tilde{\Psi}:\moduli_{0,n} \rightarrow \symmoduli_{0,n}$, then $[0] \equiv \tilde{\Psi}\circ \Psi\circ\Phi(0)$.} In order to identify the appropriate bases for $ T_{[0]}^{\ast}\symmoduli_{0,n}(\boldsymbol{m})$, one can get help from their counterpart bases for $ T_{\Psi\circ\Phi(0)}^{\ast}\moduli_{0,n}$. In the semi-classical limit $\hbar \rightarrow 0$, the projection of $\text{Sch}\left(J^{-1};w\right)$ on $\moduli_{0,n}$, according to \eqref{EMT2}, is given by 
	\begin{equation}\label{EMT3}
	\begin{split}
	\sum_{i=1}^{n-3}(\text{Sch}\left(J^{-1};w\right),M_i)dw_i &=\sum_{i=1}^{n-3}\bigg(- \pi \sum_{j=1}^{n-3} c_j R_j(w)+ \sum_{j=1}^{n} h_j \varE_j(w),M_i\bigg)dw_i \\
	&= \sum_{i=1}^{n-3}\bigg(-\pi  c_i +\sum_{j=1}^{n} h_j \big(\varE_j(w),M_i\big)\bigg)dw_i
	\end{split}
	\end{equation}
	where $R_j(w)$ and $\varE_j(w)$ are defined by \eqref{Ridef} and \eqref{varE}, respectively, and none of them are invariant under the action of $\symm{\boldsymbol{m}}$. To find the projection of $\text{Sch}\left(J^{-1};w\right)$ on the space $ T_{[0]}^{\ast}\symmoduli_{0,n}(\boldsymbol{m})$, one needs the $\symm{\boldsymbol{m}}$-invariant version of $R_i(w)$ and $\varE_i(w)$ which are given by
	\begin{equation}\label{ls3}
	\begin{split}
	&\tilde{R}_{i}(w)=-\frac{1}{\pi} \sum_{\gamma_{j,j+1}\in \symm{\boldsymbol{m}}} R\left(\gamma_{j,j+1}(w), w_i\right) \gamma_{j,j+1}'(w)^2,\\
	&\tilde{\mathscr{E}}_k(w)  = \frac{1}{2} \sum_{\gamma_{j,j+1} \in \symm{\boldsymbol{m}}} \left(\frac{1}{(\gamma_{j,j+1}(w) - w_k)^2} - \frac{1}{\gamma_{j,j+1}(w) (\gamma_{j,j+1}(w) -1)}\right) \gamma_{j,j+1}'(w)^2,\\
	&\tilde{\varE}_n(w) =\frac{1}{2} \sum_{\gamma_{j,j+1} \in \symm{\boldsymbol{m}}} \frac{1}{\gamma_{j,j+1}(w)(\gamma_{j,j+1}(w)-1)}\hspace{1mm}\gamma_{j,j+1}'(w)^2
	\end{split}
	\end{equation} 
	with $k=1,\dots,n-1$ and $\gamma_{j,j+1}(w)$ are defined in Section \ref{M0n}. In the above, just for simplicity, we assumed that the signature of $O$ is $(O;m,...,m)$, but it can be easily extended to the general case. Accordingly, the projection of $\text{Sch}\left(J^{-1};w\right)$ on moduli space $\symmoduli_{0,n}(\boldsymbol{m})$ is given by
	\begin{equation}\label{EMT4}
	\begin{split}
	\sum_{i=1}^{n-3}\left(\text{Sch}\left(J^{-1};w\right),M_i\right)d\tilde{w}_i &=\sum_{i=1}^{n-3}\bigg(-\pi \sum_{j=1}^{n-3} c_j \tilde{R}_j(w)+ \sum_{j=1}^{n} h_j \tilde{\varE}_j(w),M_i\bigg)d\tilde{w}_i\\
	&=\sum_{i=1}^{n-3}\bigg(-\pi c_i+\big(T_{s}(w),M_i\big)\bigg)d\tilde{w}_i,
	\end{split}
	\end{equation}
	where
	\begin{equation}\label{EMT5}
	\begin{split}
	T_s(w) &= \frac{1}{2}\sum_{i=1}^{n-1}  \sum_{\gamma_{j,j+1} \in \symm{\boldsymbol{m}}} \frac{h_i}{(\gamma_{j,j+1}(w) - w_i)^2} \gamma_{j,j+1}'(w)^2\\
	&-\frac{1}{2} \sum_{\gamma_{j,j+1} \in \symm{\boldsymbol{m}}} \frac{\sum_{i=1}^{n-1} h_i-h_n}{\gamma_{j,j+1}(w) (\gamma_{j,j+1}(w) -1)}\gamma_{j,j+1}'(w)^2.
	\end{split}
	\end{equation}
	In the above, the $d\tilde{w}_i$s are $(1,0)$ forms on cotangent space $ T_{[0]}^{\ast}\symmoduli_{0,n}(\boldsymbol{m})$. The counterpart of \eqref{CWI} for this case is 
	\begin{equation}\label{T2}
	T^{\symmoduli_{0,n}(\boldsymbol{m})}_{\varphi}(w)\hspace{-1mm}=\hspace{-.5mm}\sum_{i=1}^{n} \sum_{\gamma_{j,j+1} \in \symm{\boldsymbol{m}}}\hspace{-2mm}\left(\frac{h_{\text{cl}}(m_i)}{2(\gamma_{j,j+1}(w)-w_i)^2}-\frac{1}{2\pi}\frac{1}{(\gamma_{j,j+1}(w)-w_i)}\partial_{\tilde{w}_i}\Gpotential_{\boldsymbol{m}}[\varphi]\right),
	\end{equation}
	which together with the relations \eqref{EMT4} and \eqref{EMT5} gives information about the moduli space of Riemann orbisurface $O$, i.e, a relation between accessory parameters $c_i$ in \eqref{EMT4} for the Fuchsian uniformization and the first variation $\partial_{\tilde{w}_i}\Gpotential_{\boldsymbol{m}}[\varphi]$. It also provides the opportunity to check the observed relation between the universal CWIs and Friedan-Shenker modular geometry. It is noteworthy that using the same methodology, one can probe the K\"{a}hler geometry of moduli space $\symmoduli_{g,n}(\boldsymbol{m})$ based on using the equations \eqref{ls},\eqref{ls2},\eqref{scrEdef} and \eqref{ls3}.
	
	Last but not least, it is shown \cite{ZT_localindextheorem_1991} that the validity of CWI \eqref{CWI} at 1-loop approximation for the genus $g=0$ case with just $n_p$ punctures yields a formula for the first derivative of $Z_{\text{Sz}}(2)$. Moreover, it is shown \cite{ZT_localindextheorem_1991} that when $g\neq 0$ with just $n_p$ punctures, the CWI \eqref{CWI4} at the 1-loop level is equivalent to the local index theorem for the families of $\bar{\partial}$-operators on the punctured Riemann surfaces. Based on the above-examined cases, it would be interesting to explore more the relation between CWIs for the correlation functions
	\begin{equation*}
	\frac{1}{\hbar^2}\expval{\prod_{k=0}\boldsymbol{T}(w_k)\prod_{l=0}\overline{\boldsymbol{T}}(\bar{w}_l^{\prime}) \hspace{1mm}X_{\boldsymbol{m}}},
	\end{equation*}
	and uniformization
	of the Riemann surfaces and complex geometry of Teichm\"{u}ller space $\teich_{g,n} (\boldsymbol{m})$, Schottky space $\schottky_{g,n}(\boldsymbol{m})$ and Moduli spaces $\moduli_{g,n},\symmoduli_{g,n}(\boldsymbol{m})$ and more importantly 2-dimensional quantum gravity. Moreover, the multi-point correlation functions can provide further evidence in favor of the profound role of Ward identities in Friedan-Shenker modular geometry. 
	\item
	Also, very recently, it has been shown \cite{Benjamin:2023uib} that asymptotic series (in powers of the central charge) for expansion of 2-dimensional conformal blocks, involving the exchange of identity operator, necessitates the existence of non-perturbative effects via resurgence analysis. In the dual three-dimensional theory, this implies that the graviton loop expansion is also an asymptotic series, and to cure it, one needs to consider new saddle points, which are particle-like states with large negative mass (non-manifold saddles with conical excesses). In 
	this paper, we focus on geometries with  conical defects (orbifold geometries), but studying the relation between our work and \cite{Benjamin:2023uib} would be interesting.
\end{itemize}	

\acknowledgments
The authors would like to thank Lorenz Eberhardt, Ruben Hidalgo, Kirill Krasnov, Alex Maloney, Leon Takhtajan, Lee-Peng Teo, Tina Torkaman, and Peter Zograf for inspiring discussions. We are especially grateful to Peter Zograf for providing us with an English translation of his paper \cite{Zograf_1990} and to Lorenz Eberhardt and Leon Takhtajan for very helpful feedback on a draft version of this work. We are also grateful to Hossein Mohammadi for pointing out a large number of typos in an earlier version of this manuscript. The research of B.T. and A.N. is supported by the Iranian National Science Foundation (INSF) under Grant No.~4001859.
\appendix
\section{Introduction to Orbifold Riemann Surfaces}\label{Apx:orbifoldbackground}
Orbifolds lie at the intersection of many different areas of mathematics and Physics, including algebraic and differential geometry, as well as conformal field theory and string theory. Orbifolds were first introduced into topology and differential geometry by Satake \cite{satake1956generalization}, who called them \emph{V-manifolds}. Satake described them as a generalization of smooth manifolds that are locally modeled on a quotient of $\mathbb{R}^\mathfrak{n}$ by the action of a finite group and generalized concepts such as de Rham cohomology and the Gauss-Bonnet theorem to orbifolds. Shortly after the original paper of Satake, Baily introduced complex V-manifolds and generalized both the Hodge decomposition theorem \cite{baily1956decomposition} and Kodaira's projective embedding theorem \cite{baily1957imbedding} to V-manifolds. The concept of V-manifolds was later re-invented by Thurston \cite{thurston80} under the name of ``orbifolds'' and the notion of fundamental groups was generalized for these objects. Even though orbifolds were already very important objects in mathematics, the work of Dixon, Harvey, Vafa, and Witten \cite{Dixon:1985jw,Dixon:1986jc} as well as the subsequent work of Dixon, Friedan, Martinec, and Shenker \cite{Dixon:1986qv} lead to a dramatic increase of interest in orbifolds among Physicists.\footnote{Professor H. Arfaei was among the first string theorists who quickly realized the significance of orbifolds in string theory and worked on various aspects of their role within this framework \cite{ardalan1987quotient,Ardalan:1989xq,Ardalan:1990ef}.} The main objective of this appendix is to compile some basic facts about orbifolds and fix some notations used throughout this paper. Although we start with a general setting, the main focus of this appendix is on complex 1-dimensional orbifolds, called orbifold Riemann surfaces, as they are the objects of study in the main body of this paper. 

To motivate our interest in orbifold Riemann surfaces, let us recall that ordinary Riemann surfaces are complex 1-dimensional algebro-geometric objects with a lot of good properties: Geometrical facts about Riemann surfaces are as ``nice'' as possible, and they often provide the intuition and motivation for generalizations to more complicated manifolds or varieties. The name ``surface'' comes from the fact that every Riemann surface is a two-dimensional real analytic manifold (i.e., a surface), but it contains more structures: In fact, a Riemann surface is the simplest example of a K\"{a}hler manifold which means that it admits three mutually compatible structures --- a complex structure, a Riemannian structure, and a symplectic structure. In addition, the existence of non-constant meromorphic functions on these surfaces can be used to show that any compact Riemann surface is a \emph{projective algebraic curve} and, therefore can be given by polynomial equations inside a projective space. \emph{Orbifold Riemann surfaces} are the natural generalization of Riemann surfaces in the orbifold world. Just like their manifold counterparts, orbifold Riemann surfaces can be viewed both as a complex orbifold of dimension 1 (\emph{complex analysts viewpoint}) or as a smooth, proper Deligne--Mumford stack\footnote{\emph{Deligne--Mumford stacks} (or \emph{DM stacks}) are an algebraic generalization of orbifolds, and can be roughly thought of as an ``orbivariety'' or ``orbischeme''. Just as orbifolds are locally the quotient of a manifold by a finite group, a DM stack can be characterized as being ``locally'' the quotient of a scheme by a finite group action.} (over $\cmpx$) of dimension 1 (\emph{algebraic geometers viewpoint}). These orbisurfaces also admit Riemannian metrics and can be regarded as the simplest examples of K\"{a}hler orbifolds. When the emphasis is on the algebro-geometric viewpoint, orbifold Riemann surfaces are usually called \emph{orbifold curves} or \emph{orbicurves}.

The viewpoint that one takes on the singular points of an orbifold depends a lot on what type of ``space'' one is working with: When working in the topological realm, one usually treats the orbifold singularities as an \emph{additional structure} -- an \emph{orbifold structure} -- on an \emph{underlying topological space} in the same way that one thinks of a smooth structure as an additional structure on a topological manifold (see \cite{thurston80,steer1999grothendieck,Scott1983TheGO,Cooper_2000,caramello2019introduction}). In particular, a topological space is allowed to have several different orbifold structures. On the other hand, from an algebro-geometric viewpoint \cite{behrend2006uniformization}, it is more convenient to consider (analytic or algebraic) \emph{stacks} as the proper notion of space. Such a stack is then called an orbifold (be it an analytic or an algebraic one) if it admits a \emph{covering} by open substacks of the form $[\tilde{U}/\Gamma]$,\footnote{If $M$ is a complex manifold of dimension $\mathfrak{n}$ and $\Gamma \subset \operatorname{Aut}(M) \subset \operatorname{GL}(\mathfrak{n}, \cmpx)$ is a finite subgroup of holomorphic automorphisms of $M$ which does \emph{not} act \emph{freely} (i.e. has fixed points on $M$), the quotient space $M/\Gamma$ will have the structure of an analytic stack (or, equivalently, of a complex orbifold) and we will use the notation $[M/\Gamma]$ to mean $M/\Gamma$ as an analytic orbifold/stack. The notation $M/\Gamma$ will be reserved for the \emph{coarse moduli space} or the \emph{underlying analytic space} of $[M/\Gamma]$ which will be a \emph{variety with quotient singularities}.} parameterising families of $\Gamma$-orbits in $\tilde{U}$, where $\tilde{U}$ is the local model for representable stacks (i.e. manifolds) and $\Gamma$ is a finite subgroup of the automorphism group of $\tilde{U}$. This second point of view treats an orbifold singularity as an intrinsic structure of the space. 

For us, the appropriate notion of space will be that of \emph{analytic spaces} \cite{Gunning:1969rm}, by which we mean a generalization of complex manifolds that allow the presence of singularities and are locally isomorphic to the common zero locus of a finite collection of holomorphic functions. Our point of view, which will be reflected in our introduction to orbifolds, lies somewhere in between the two extremes mentioned above: We will treat the subclass of codimension $\geq 2$ orbifold singularities as the intrinsic structure of an \emph{underlying complex analytic space}\footnote{This underlying analytic space is actually the same as the coarse moduli space of the orbifold regarded as an analytic stack.} while the orbifold singularities of codimension-1 still need to be treated as additional structures on this analytic space (see \cite{Boyer_2008} for more details).

Throughout this paper, we will need to work with different characterizations of orbifold Riemann surfaces that have appeared in the literature: In order to introduce orbifold geometric structures in Appendix~\ref{Apx:geometricorbifolds}, we will need to work with a definition of complex orbifolds based on \emph{orbifold charts} \cite{kawasaki1979riemann} while a characterization of Riemann orbisurfaces as \emph{log pairs} \cite{Campana2004} will be more suitable for studying orbifold metrics \cite{McOwen-1988}. A third way of characterizing Riemann orbisurfaces will be as \emph{Riemann surfaces with signature} and this viewpoint -- which is closely related to the notion of Riemann orbisurfaces as log pairs -- is the one that we have adapted in the main body of this paper. 

Our presentation in this appendix is closer to that of references \cite[\textsection4]{Boyer_2008}, \cite[Appx.~E]{milnor1990dynamics}, and \cite{uludaug2007orbifolds}. For more details about other approaches, the reader is encouraged to consult with \cite{thurston80,Scott1983TheGO,Cooper_2000, caramello2019introduction,adem2007orbifolds} among others.
\subsection{Analytic Geometry}\label{orbifolddef}
We will start our introduction to complex orbifolds with a quick review of some background information about complex analytic spaces and analytic mappings --- particularly, ramified covering maps of analytic spaces. The reader is advised to consult with references \cite{gunning1970lectures, Gunning:1969rm, Grauert-Remmert, Riemenschneide_2011, namba1991finite, namba1992finite, uludaug2007orbifolds} for more details.
\subsubsection{Complex Analytic Spaces and Analytic Mappings}
Let us start our review of complex analytic spaces with defining a complex analytic subvariety:
\begin{definition}[Analytic subvariety]\label{def:subvariety}
	Let $U$ be an open subset of $\cmpx^{\mathfrak{n}}$ (or of any complex analytic manifold $M$) and let $X$ be a subset of $U$. We say that $X$ is an \emph{analytic subvariety} in $U$ if, for any point $x$ in $U$, there exist a neighborhood $V$ of $x$ and a finite number of holomorphic functions $f_1, \dots, f_k$ on $V$ such that
	\begin{equation*}
	X \cap V = \left\{z \in V \, \Big| \, f_1(z) = \dotsm = f_k(z) = 0\right\}.
	\end{equation*}
	In other words, in some open neighborhood of each point of $U$, an analytic subvariety $X \subset U$ is the set of common zeros of a finite number of complex analytic functions. We call $f_1,\dots,f_k$ a system of local defining functions for $X$ and a non-empty analytic subvariety of $U$ which is locally defined by a \emph{single} (not identically zero) holomorphic function will be called an \emph{analytic hypersurface} in $U$. Finally, a subvariety $X$ of $U$ is called \emph{irreducible} if it cannot be written as a union $X = X_1 \cup X_2$ where $X_i$ are analytic subvarieties of $U$ properly contained in $X$.
\end{definition}

As suggested before, complex analytic subvarieties can be viewed as a generalization of complex (sub)manifolds which allow for the presence of singularities:
\begin{definition}[Smooth and singular points of subvarieties]
	A point $x$ on an analytic subvariety $X$ is said to be \emph{regular} or \emph{smooth} if it is possible to choose coordinates $(z_1,\dots,z_{\mathfrak{n}})$ in an open neighborhood $V \subset U \subset \cmpx^{\mathfrak{n}}$ of point $x$ such that locally $X$ is a linear subspace $\{z \in V \, | \,  z_{k+1} = \dotsm = z_{\mathfrak{n}} = 0 \}$ --- i.e. if $X \cap V$ is a $k$-dimensional submanifold of $\cmpx^{\mathfrak{n}}$. The set of all regular points of $X$ is an open dense subset of $X$ and will be denoted by $\operatorname{Reg}(X)$. The points of a subvariety that are not regular points are called the \emph{singular points}. The set $X\backslash \operatorname{Reg}(X)$ of all singular points of $X$ will be denoted by $\sing(X)$ and is called the \emph{singular locus} of $X$.
\end{definition}
\noindent An analytic subvariety $X$ will be called a \emph{smooth analytic subvariety} if $X= \operatorname{Reg}(X)$; evidently, a smooth analytic variety is just a complex analytic manifold itself. When $X$ is an irreducible analytic subvariety, the \emph{complex dimension of $X$} is defined as the dimension of its smooth part $\operatorname{Reg}(X)$ regarded as a complex manifold. More generally, if $X$ is reducible, the dimension of $X$ is defined as the maximum of the dimensions of its irreducible components. A reducible analytic subvariety $X$ is called \emph{pure dimensional} if every irreducible component of $X$ has the same dimension.

It is clear that analytic subvariety $X \subset U$ can be endowed with the relative topology coming from $U$; however, the main point in the study of analytic subvarieties is that one should take into account consideration not only about the topology of these analytic subvarieties but also about their \emph{function-theoretic properties}: For simplicity, let us take the open set $U \subset \cmpx^{\mathfrak{n}}$ to be a sufficiently small open polydisc $\varDelta(\boldsymbol{\epsilon})$ such that an analytic subvariety $X$ of $\varDelta(\boldsymbol{\epsilon})$ can be determined as the set of common zeros of a finite number of functions that are analytic throughout $\varDelta(\boldsymbol{\epsilon})$. Under the natural addition and multiplication of complex-valued functions, the set of all holomorphic functions on $\varDelta(\boldsymbol{\epsilon})$ forms a ring $\mathfrak{O}_{\varDelta(\boldsymbol{\epsilon})}$ containing the constants $c \in \cmpx^{\mathfrak{n}}$ --- hence in fact a $\cmpx$-algebra. The set of all analytic functions in $\varDelta(\boldsymbol{\epsilon})$ which vanish on $X$ form an \emph{ideal} $\mathfrak{I}(X)$ in the ring $\mathfrak{O}_{\varDelta(\boldsymbol{\epsilon})}$, called the \emph{ideal of $X$}. Then, the \emph{ring of holomorphic functions on $X$} is given by the quotient ring $\mathfrak{O}_{X}:=\mathfrak{O}_{\varDelta(\boldsymbol{\epsilon})}/\mathfrak{I}(X)$. It is easy to see that a subvariety $X$ of $\varDelta(\boldsymbol{\epsilon})$ is irreducible precisely when the ideal $\mathfrak{I}(X)$ is a \emph{prime} ideal in $\mathfrak{O}_{\varDelta(\boldsymbol{\epsilon})}$, by which we mean $\mathfrak{I}(X) \neq \mathfrak{O}_{\varDelta(\boldsymbol{\epsilon})}$ and for any two holomorphic functions $f, f' \in \mathfrak{O}_{\varDelta(\boldsymbol{\epsilon})}$ the statement $f  f' \in \mathfrak{I}(X)$ implies $f \in \mathfrak{I}(X)$ or $f' \in \mathfrak{I}(X)$ (or both).

\begin{remark}\label{rmk:radical}
	Let $f_1,\dots,f_k \in \mathfrak{O}_{\varDelta(\boldsymbol{\epsilon})}$ be the system of defining functions for general (i.e., not necessarily irreducible) subvariety $X$. Then, the defining functions $f_1,\dots,f_k$ generate an \emph{ideal} $\mathfrak{I}$ in the $\mathfrak{O}_{\varDelta(\boldsymbol{\epsilon})}$ which is sometimes called the \emph{defining ideal} for $X$. While all holomorphic functions $f \in \mathfrak{I}$ vanish on $X$, i.e. $\mathfrak{I} \subseteq \mathfrak{I}(X)$, the inverse statement is \emph{not} necessarily true --- i.e. in general $\mathfrak{I}(X) \nsubseteq \mathfrak{I}$. However, provided that the polydisc $\varDelta(\boldsymbol{\epsilon})$ is small enough, an important fact is that \emph{all} holomorphic functions $f \in \mathfrak{I}(X)$ have a power which is contained in $\mathfrak{I}$. This motivates us to define the \emph{radical} of $\mathfrak{I}$ as the set
	\begin{equation*}
	\sqrt{\mathfrak{I}} := \left\{f \in \mathfrak{O}_{\varDelta(\boldsymbol{\epsilon})} \, \Big| \, f^{k'} \in \mathfrak{I} \text{ for some positive integer } k'\right\}.
	\end{equation*}
	This is again an ideal in $\mathfrak{O}_{\varDelta(\boldsymbol{\epsilon})}$ and we have $\sqrt{\mathfrak{I}} = \mathfrak{I}(X)$. When the defining ideal $\mathfrak{I}$ is prime in $\mathfrak{O}_{\varDelta(\boldsymbol{\epsilon})}$, we get $\sqrt{\mathfrak{I}} = \mathfrak{I}$; hence, $\mathfrak{I} = \mathfrak{I}(X)$ exactly when $X$ is irreducible.
\end{remark}
More generally, when the open set $U \subset \cmpx^{\mathfrak{n}}$ is not restricted to be small (or to be an open polydisc), the above statements hold true in small neighborhoods of each point $x \in X$. For such local considerations, it is often convenient to introduce the notion of \emph{germs}: To make this precise, consider two pairs $(X, U)$ and $(X', U')$ where $U, U'$ are open neighborhoods of the origin in $\cmpx^{\mathfrak{n}}$ and $X,X'$ are analytic subvarieties of $U, U'$ respectively. The two pairs $(X, U)$ and $(X', U')$ are said to define the same \emph{germ of analytic subvarieties} at the origin in $\cmpx^{\mathfrak{n}}$ if there exists a neighborhood $W \subseteq U \cap U'$ of 0 such that $X \cap W = X' \cap W$. We will denote the germ of analytic subvariety $X$ at 0 in $\cmpx^{\mathfrak{n}}$ by $[X]_0$. Now, let $\mathfrak{O}_{U}$ be the ring of holomorphic functions in some open subset $U \subset \cmpx^{\mathfrak{n}}$ containing the origin. Analogously, we can define an equivalence relation $\sim_{0}$ between two holomorphic functions $f,f' \in \mathfrak{O}_{U}$ where $f \sim_{0} f'$ if there exists a neighborhood $W$ of 0 such that the restrictions of $f$ and $f'$ to $W$ are identical --- i.e. $f\big|_{_{W}} = f'\big|_{_{W}}$. The equivalence class of a function $f$ is called the \emph{germ of holomorphic function $f$ at the origin} and will be denoted by $[f]_0$. In addition, the quotient ring $\mathfrak{O}_{U,0} := \mathfrak{O}_{U}/\sim_{0}$ will be the \emph{ring of germs of holomorphic functions at the origin}.\footnote{If we denote by $\cmpx\{ z_1, \dots,z_{\mathfrak{n}}\}$ the set of power series which converge absolutely in some neighborhood of 0, this set also has the structure of a ring. Since, as in the one variable case, $f \sim_{0} f'$ if and only if $f$ and $f'$ have the same power series expansion, we may identify $\mathfrak{O}_{U,0}$ with $\cmpx\{ z_1, \dots,z_{\mathfrak{n}}\}$.}

Similar to the case of analytic subvariety of a sufficiently small open polydisc, to each germ $[X]_0$ of an analytic subvariety at the origin in $\cmpx^{\mathfrak{n}}$ there is canonically associated an ideal $\mathfrak{I}([X]_0)$ in the local ring\footnote{A commutative ring is called \emph{local} if it has only one maximal ideal. The unique maximal ideal of $\mathfrak{O}_{\cmpx^{\mathfrak{n}},0}$ (equivalently, $\cmpx\{ z_1, \dots,z_{\mathfrak{n}}\}$) corresponds to the set of all germs of holomorphic functions that vanish at the origin (equivalently, the set of all power series whose constant term vanishes).} $\mathfrak{O}_{\cmpx^{\mathfrak{n}},0}$ which is defined as the ideal of germs of all analytic functions vanishing on the subvariety $X$ representing the germ $[X]_0$. In the other direction, to each ideal $\mathfrak{I} \subseteq \mathfrak{O}_{\cmpx^{\mathfrak{n}},0}$ there is canonically associated a germ of an analytic subvariety at the origin in $\cmpx^{\mathfrak{n}}$, called the \emph{locus of the ideal $\mathfrak{I}$} and denoted by $[X(\mathfrak{I})]_0$. The germ $[X(\mathfrak{I})]_0$ is defined as the germ represented by the analytic subvariety $X = \{z \in U \, | \, f_1(z) = \dotsm = f_k(z) = 0\}$ of the open set $U \subset \cmpx^{\mathfrak{n}}$, where $f_i \in  \mathfrak{O}_{U}$ are analytic functions in $U$ whose germs in $\mathfrak{O}_{U,0}$ generate the ideal $\mathfrak{I}$. Note that, similar to what we saw in Remark~\ref{rmk:radical}, $\mathfrak{I}\big([X(\mathfrak{I})]_0\big) = \sqrt{\mathfrak{I}}$ and $\sqrt{\mathfrak{I}} = \mathfrak{I}$ iff the ideal $\mathfrak{I} \subseteq \mathfrak{O}_{\cmpx^{\mathfrak{n}},0}$ is prime (equivalently, if the germ $[X(\mathfrak{I})]_0$ is irreducible). Finally, the residue class ring $\mathfrak{O}_{\cmpx^{\mathfrak{n}},0}/\mathfrak{I}([X]_0)$ will now be denoted by $\mathfrak{O}_{X,0}$ and will be called the \emph{ring of germs of holomorphic functions on the subvariety $X$} at the origin in $\cmpx^{\mathfrak{n}}$.

Having learned how to think about the function-theoretic properties of an analytic subvariety locally, we are now ready to study those properties from a global perspective; it is then convenient to think about \emph{sheaves} of rings/ideals/modules as we shall now explain: Let us start with the local rings $\mathfrak{O}_{\cmpx^{\mathfrak{n}},z}$ of germs of holomorphic functions at any point $z \in \cmpx^{\mathfrak{n}}$. The set of rings $\mathfrak{O}_{\cmpx^{\mathfrak{n}}, z}$ for all points $z \in \cmpx^{\mathfrak{n}}$ can be taken to form the \emph{sheaf of germs of holomorphic functions of $\mathfrak{n}$ complex variables} $\mathscr{O}_{\cmpx^{\mathfrak{n}}}$; the restriction $\mathscr{O}_{\cmpx^{\mathfrak{n}}}\big|_{_{U}}$ of the sheaf of rings $\mathscr{O}_{\cmpx^{\mathfrak{n}}}$ to any open set $U \subset \cmpx^{\mathfrak{n}}$ will be simply denoted by $\mathscr{O}_{U}$. Similarly, consider an analytic subvariety $X$ of an open subset $U \subset \cmpx^{\mathfrak{n}}$ and to each point $z \in U$ associate the ideal $\mathfrak{I}([X]_z) \subseteq \mathfrak{O}_{\cmpx^{\mathfrak{n}},z}$ of the germ of the subvariety $X$ at that point (if $z \notin X$ the ideal $\mathfrak{I}([X]_z)$ is of course the trivial ideal $\mathfrak{O}_{\cmpx^{\mathfrak{n}},z}$). The set of all ideals $\mathfrak{I}([X]_z)$ at any point $z \in U$ form an analytic subsheaf \footnote{An \emph{analytic sheaf} over an open set $U \subset \cmpx^{\mathfrak{n}}$ is a sheaf of modules over $\mathscr{O}_{U}$.} of the sheaf $\mathscr{O}_{U}$ over the set $U$, which will be dented by $\mathscr{I}(X)$ and called the \emph{sheaf of ideals of the analytic subvariety $X$}. Finally, the restriction to the subvariety $X$ of the analytic sheaf $\mathscr{O}_{U}/\mathscr{I}(X)$ will be called the \emph{the sheaf of germs of holomorphic functions on the subvariety $X$} and is dented by $\mathscr{O}_{X}$; the local rings $\mathfrak{O}_{X,x} = \mathfrak{O}_{U,x}/\mathfrak{I}([X]_x)$ at any point $x \in X$ can then be viewed as the \emph{stalks} of $\mathscr{O}_{X}$.

Locally, a germ of an analytic subvariety determines a germ of a topological space and this space further possess a distinguished subring of the ring of germs of continuous complex-valued functions, namely the ring of germs of holomorphic functions on the subvariety. This observation suggests that the correct way of characterizing an analytic subvariety $X$ is as a \emph{$\cmpx$-ringed space}:
\begin{definition}
	A \emph{ringed space} $X$ is a pair $(|X|,\mathscr{O}_X)$ consisting of a Hausdorff topological space $|X|$ and a sheaf of rings $\mathscr{O}_{X}$ on $|X|$, called the \emph{structure sheaf} of $X$. It is called a \emph{locally ringed space} when, for every $x \in |X|$, the stalk $\mathscr{O}_{X,x}$ is a local ring. Its maximal ideal is denoted by $\mathfrak{m}_{X,x}$. A locally ringed space is called a \emph{$\cmpx$-ringed space} when furthermore $\mathscr{O}_{X}$ is a sheaf of $\cmpx$-algebras and, for every $x \in |X|$, there is an isomorphism $\mathscr{O}_{X,x}/\mathfrak{m}_{X,x} \cong \cmpx$ of $\cmpx$-algebras.
\end{definition}
\noindent It is clear that for analytic subvarieties, the role of structure sheaf is played by the \emph{sheaf of germs of holomorphic functions} on the subvariety.

The notion of an analytic subvariety as defined in \ref{def:subvariety} depends quite essentially on a particular embedding in the ambient space $\cmpx^{\mathfrak{n}}$. For example, the germ of an analytic subvariety at the origin in $\cmpx^{\mathfrak{n}}$ can also be viewed as the germ of an analytic subvariety at the origin in $\cmpx^{\mathfrak{n}+1}$ through the canonical embedding $\cmpx^{\mathfrak{n}} \hookrightarrow \cmpx^{\mathfrak{n}+1}$ but these will be inequivalent germs of analytic subvarieties. It is thus evident that there is a point to introducing an equivalence relation among analytic subvarieties in order to investigate those properties, which are to some extent independent of the embeddings of these subvarieties in their ambient complex number spaces. 

Once again, for the sake of simplicity, let us consider sufficiently small open polydiscs $\varDelta^{\mathfrak{n}_1}(\boldsymbol{\epsilon}_1) \subset \cmpx^{\mathfrak{n}_1}$ and  $\varDelta^{\mathfrak{n}_2}(\boldsymbol{\epsilon}_2) \subset \cmpx^{\mathfrak{n}_2}$ such that analytic subvarieties $X_1 \subset \varDelta^{\mathfrak{n}_1}(\boldsymbol{\epsilon}_1)$ and $X_2 \subset \varDelta^{\mathfrak{n}_2}(\boldsymbol{\epsilon}_2)$ can be determined as the set of common zeros of a finite number of functions that are analytic throughout $\varDelta^{\mathfrak{n}_1}(\boldsymbol{\epsilon}_1)$ and $\varDelta^{\mathfrak{n}_2}(\boldsymbol{\epsilon}_2)$, respectively. A continuous mapping between two such analytic subvarieties $X_1 \subset \varDelta^{\mathfrak{n}_1}(\boldsymbol{\epsilon}_1)$ and $X_2 \subset \varDelta^{\mathfrak{n}_2}(\boldsymbol{\epsilon}_2)$ is said to be a \emph{complex analytic mapping} $f: X_1 \to X_2$ between these two subvarieties if there is a holomorphic mapping $F: \varDelta^{\mathfrak{n}_1}(\boldsymbol{\epsilon}_1) \to  \varDelta^{\mathfrak{n}_2}(\boldsymbol{\epsilon}_2)$ such that the restriction of $F$ to the subvariety $X_1 \subset \varDelta^{\mathfrak{n}_1}(\boldsymbol{\epsilon}_1)$ is just $f$ --- i.e. $F \big|_{X_1} = f$. Additionally, two analytic subvarieties $X_1 \subset \varDelta^{\mathfrak{n}_1}(\boldsymbol{\epsilon}_1)$ and $X_2 \subset \varDelta^{\mathfrak{n}_2}(\boldsymbol{\epsilon}_2)$ are said to be \emph{analytically equivalent} if there are complex analytic mappings $f: X_1 \to X_2$ and $g: X_2 \to X_1$ such that the compositions $f \circ g$ and $g \circ f$ are the appropriate identity mappings. This notion of equivalence thus allows one to speak of analytic subvarieties without reference to the spaces in which they are embedded; an equivalence class is called an \emph{analytic variety}, and a space which has locally the
structure of an analytic variety is called an \emph{analytic space}.

For global considerations, it is more convenient to think about an analytic subvariety $X$ as a  $\cmpx$-ringed space $(|X|,\mathscr{O}_X)$. Complex analytic mappings $f: X_1 \to X_2$ between two subvarieties $X_1$ and $X_2$ are then viewed as \emph{morphisms of $\cmpx$-ringed spaces} between $(|X_1|,\mathscr{O}_{X_1})$ and $(|X_2|,\mathscr{O}_{X_2})$: 

\begin{definition}[Morphism of $\cmpx$-ringed spaces]
	A \emph{morphism} $f: X_1 \to X_2$ of ringed spaces $(|X_1|,\mathscr{O}_{X_1})$ and $(|X_2|,\mathscr{O}_{X_2})$ is a pair $f=(|f|,f^{\ast})$ consisting of a continuous map
	\begin{equation*}
	|f|: |X_1| \to |X_2|
	\end{equation*}
	and a homomorphism
	\begin{equation*}
	f^{\ast}: \mathscr{O}_{X_2} \to |f|_{\ast}(\mathscr{O}_{X_1})
	\end{equation*}
	of sheaves of rings on $X_2$. For any point $x \in X_1$, we think of $f^{\ast}_{x}$ as the ring homomorphism
	\begin{equation*}
	f_{x}^{\ast}: \mathscr{O}_{X_2,f(x)} \to \mathscr{O}_{X_1,x}
	\end{equation*}
	defined as the composition of the canonical homomorphisms
	\begin{equation*}
	\mathscr{O}_{X_2,f(x)} \to \Big(|f|_{\ast}(\mathscr{O}_{X_1})\Big)_{f(x)} \to \mathscr{O}_{X_1,x}.
	\end{equation*}
	In case $X_1$ and $X_2$ are locally ringed spaces, a morphism by definition has to be \emph{local}, that is, satisfy
	\begin{equation*}
	f_{x}^{\ast}\big(\mathfrak{m}_{X_2,f(x)}\big) \subset \mathfrak{m}_{X_1,x}
	\end{equation*}
	for every $x \in X_1$. A \emph{morphism of $\cmpx$-ringed spaces} $X_1$ and $X_2$ is a morphism of ringed spaces, where $f^{\ast}$ is furthermore a homomorphism of sheaves of $\cmpx$-algebras. In this case, $f_{x}^{\ast}$ is automatically local for every $x \in X_1$.
\end{definition}
\noindent It is an immediate consequence of this definition that two analytic subvarieties $X_1$ and $X_2$ determine \emph{equivalent varieties} if and only if there is a \emph{topological homeomorphism} $|f|: |X_1| \to |X_2|$ inducing an isomorphism $f^{\ast}: \mathscr{O}_{X_2} \xrightarrow{\cong} |f|_{\ast}(\mathscr{O}_{X_1})$ between the sheaves of $\cmpx$-algebras. Thus, the coherent analytic sheaf $\mathscr{O}_X$ on an analytic subvariety $X$ is the complete invariant determining equivalence as varieties --- hence the name \emph{structure sheaf}.

To build up complex analytic spaces, we construct local models as follows: Let $U \subset \cmpx^{\mathfrak{n}}$ be an open subset, and assume that $\mathscr{I}$ is a coherent sheaf of $\mathscr{O}_U$-ideals. Then
\begin{equation*}
\operatorname{Supp}\big(\mathscr{O}_U/\mathscr{I}\big) = \Big\{x \in U \, \Big| \, \big(\mathscr{O}_U/\mathscr{I}\big)_{x} \neq 0 \Big\}
\end{equation*}
is an analytic subset of $U$ which we will denote by $\tilde{A}$. The pair $\big(\tilde{A}, (\mathscr{O}_U/\mathscr{I})\big|_{\tilde{A}}\big)$ is a $\cmpx$-ringed space, which is called a \emph{local model of an analytic space}.
\begin{definition}[Complex analytic spaces and analytic mappings]
	A \emph{complex analytic space}, or an \emph{analytic space} for short, is a $\cmpx$-ringed space $(|X|, \mathscr{O}_X)$ satisfying the following conditions:
	\begin{enumerate}[(i)]
		\item $|X|$ is Hausdorff,
		\item For every $x \in X$, there is an open neighbourhood $V_x$ of $x$ such that $(V_x, \mathscr{O}_X\big|_{V_{x}})$ is isomorphic (as $\cmpx$-ringed space) to some local model.
	\end{enumerate}
	If $X= (|X|, \mathscr{O}_X) $ and $Y=(|Y|, \mathscr{O}_Y)$ are complex analytic spaces, then any morphism
	\begin{equation*}
	(|f|,f^{\ast}) : (|X|, \mathscr{O}_X) \to (|Y|, \mathscr{O}_Y),
	\end{equation*}
	of $\cmpx$-ringed spaces is called an \emph{analytic map} (or \emph{holomorphic map}).
\end{definition}
\begin{remark}
	Note that any complex manifold $M$ can be considered as an analytic space $(|M|, \mathscr{O}_M)$ and holomorphic maps $f: M \to N$ between complex manifolds can be extended to analytic mappings $(|f|,f^{\ast}): (|M|, \mathscr{O}_M) \to (|N|, \mathscr{O}_N)$.
\end{remark}

Let $|X|$ be a topological space. Any analytic variety $(|U|, \mathscr{O}_{U})$, where $|U|$ is an open set in $|X|$, is called an \emph{analytic chart} on $|X|$. A family $\big\{(|U_a|, \mathscr{O}_a), g^{\ast}_{ab}\big\}_{a,b \in A}$ consisting of complex charts on $|X|$ and of $\cmpx$-algebra isomorphisms 
\begin{equation*}
g^{\ast}_{ab}: \mathscr{O}_b\big|_{_{|U_a| \cap |U_b|}} \to  \mathscr{O}_a\big|_{_{|U_a| \cap |U_b|}},
\end{equation*}
is called an \emph{analytic atlas} on $X$, if $\big\{|U_a|\big\}_{a \in A}$ is an open covering of $|X|$, and if furthermore
\begin{equation*}
g^{\ast}_{ab} \circ  g^{\ast}_{bc} =  g^{\ast}_{ac} \quad \text{for all} \quad a,b,c \in A \qquad \text{(coycle condition)};
\end{equation*}
the maps $g^{\ast}_{ab}$ are the \emph{gluing isomorphisms} of the atlas. Furthermore, we have the following lemma (see e.g. \cite[\textsection1.7]{Grauert-Remmert} for the proof):
\begin{lemma}[Gluing Lemma]\label{lemma:gluing}
	Let $\big\{(|V_i|, \mathscr{O}_i), g^{\ast}_{ij}\big\}_{i,j \in I}$ be an analytic atlas on a Hausdorff topological space $|X|$. Then, there exists a unique (up to isomorphism) complex analytic space $(|X|, \mathscr{O}_X)$ and $\cmpx$-algebra isomorphisms $f^{\ast}_i: \mathscr{O}_{X}\big|_{_{|V_i|}} \to \mathscr{O}_i$ for all $i 
	\in I$, such that
	\begin{equation*}
	g^{\ast}_{ij} = f^{\ast}_i \circ (f^{\ast}_j)^{-1}
	\end{equation*}
	on every intersection $|V_i| \cap |V_j|$.
\end{lemma}

A complex analytic space $Y$ is called an \emph{open complex analytic subspace} of $X$, if $|Y|$ is an open subset of $|X|$, and $\mathscr{O}_Y = \mathscr{O}_X\big|_{_Y}$. In addition, $Y$ is called a \emph{closed complex analytic subspace} of $X$, if there is a coherent ideal $\mathscr{J} \subset  \mathscr{O}_X$ such that $|Y|= \operatorname{Supp}(\mathscr{O}_X/\mathscr{J})$ and $\mathscr{O}_Y = (\mathscr{O}_X/\mathscr{J})\big|_{Y}$.
In this case, there is a canonical analytic map determined by the injection, which we denote by $Y \hookrightarrow X$. A subset $A$ of a complex analytic space $X$ is called \emph{analytic} when there is a coherent ideal $\mathscr{J}\subset \mathscr{O}_X$ such that $A = \operatorname{Supp}(\mathscr{O}_X/\mathscr{J})$. 

We now discuss several possibilities of how good (respectively, how bad) a given analytic space $X= (|X|, \mathscr{O}_X) $ may behave at a point $x \in X$. The situation is optimal if $x$ is a smooth point of $X$:
\begin{definition}[Smooth and singular points of analytic spaces]
	A point $x$ in an analytic space $X=(|X|, \mathscr{O}_{X})$ is called \emph{smooth} or \emph{regular}, if there exists a neighborhood $V_x$ of $x$ in $|X|$ and an open set $U$ in some complex number space $\cmpx^\mathfrak{n}$ such that $(V_x, \mathscr{O}_X\big|_{_{V_x}})$ and $(U, \mathscr{O}_{U})$ are analytically isomorphic --- i.e. if there exist analytic mappings 
	\begin{equation*}
	(|f|,f^{\ast}): (V_x, \mathscr{O}_X\big|_{_{V_x}}) \to (U, \mathscr{O}_{U}) \quad \text{and} \quad (|g|,g^{\ast}): (U, \mathscr{O}_{U}) \to (V_x, \mathscr{O}_X\big|_{_{V_x}}),
	\end{equation*}
	such that the compositions $|f| \circ |g|$, $|g| \circ |f|$, $f^{\ast} \circ g^{\ast}$, and $g^{\ast} \circ f^{\ast}$ are the appropriate identity mappings. In other words, a point $x \in X$ is smooth if and only if $\mathscr{O}_{X,x} \cong \mathscr{O}_{\cmpx^{\mathfrak{n}},x}$. Of course, the singular locus is defined to be the set of all non-smooth points: $\sing(X) = X\backslash \operatorname{Reg}(X)$.
\end{definition}
\noindent 
When $X$ is not necessarily smooth (i.e., has a non-empty singular locus), we define the following notions regarding the behavior of an analytic space $X$ at a point $x \in X$: 
\begin{itemize}
	\item  The analytic space $X$ is called \emph{irreducible at point $x$} if the stalk $\mathscr{O}_{X,x}$ is an integral domain, otherwise $X$ is called reducible at $x$. All smooth points are irreducible points, since for such a point $x$, the stalk $\mathscr{O}_{X,x}$ is isomorphic to the ring of convergent power series $\cmpx\{z_1,\dots,z_{\mathfrak{n}}\}$. The analytic space $X$ will be called \emph{locally irreducible} if all points of $X$ are irreducible; in particular, complex manifolds are locally irreducible. 
	\item The complex analytic space $X$ is called \emph{reduced at $x$} if the stalk $\mathscr{O}_{X,x}$ is a reduced ring ---  i.e. does not contain nilpotent elements. All irreducible points are reduced points of $X$. We call $X$ a \emph{reduced analytic space} if $X$ is reduced at all its points; this happens when every local model for the space is defined by a radical sheaf of ideals. An analytic space $X$, which isn't reduced has a reduction $X_{\text{red}}$, which is a reduced analytic space with the same underlying topological space. There exists a canonical embedding $\iota: X_{\text{red}} \hookrightarrow X$ and every morphism from $X$ to a reduced analytic space factors through $\iota$. In particular, every analytic mapping $f: Y \to X$ of complex analytic spaces induces a canonical analytic morphism
	$f_{\text{red}} : Y_{\text{red}} \to X_{\text{red}}$ of their reductions such that $f \circ \iota_Y = \iota_X \circ f_{\text{red}}$ where $\iota_X: X_{\text{red}} \hookrightarrow X$ and $\iota_Y: Y_{red} \hookrightarrow Y$ are the canonical embeddings. When $Y=Y_{\text{red}}$, the sheaf homomorphism component, $f^{\ast}$, of an analytic map $(|f|,f^{\ast}): (|Y|,\mathscr{O}_Y) \to (|X|,\mathscr{O}_X)$ is \emph{uniquely} determined by its continuous mapping component $|f|$. 
	\item A reduced point $x \in X$ will be called a \emph{normal point} of $X$, if the stalk $\mathscr{O}_{X,x}$ is integrally closed in its quotient ring. Smooth points are normal, and $X$ is irreducible at every normal point. An analytic space $X$ will be called \emph{normal} if every point $x \in X$ is a normal point; in a normal analytic space, the \emph{singular locus has codimension at least two}. Once again, all non-normal analytic spaces can be smoothed out into normal spaces in a canonical way; this construction is called the \emph{normalization}.
\end{itemize}
In the rest of this appendix, we will mainly focus on \emph{normal analytic spaces}.
\subsubsection{An Intermezzo on Line Bundles and Divisors}
In this subsection, we will introduce the basic definitions concerning line bundles and divisors. Although we are mainly interested in complex analytic spaces, we will mostly focus our attention to complex manifolds (i.e., smooth analytic spaces) to avoid complication; we will comment on some of the subtleties of generalizing the introduced notion to singular analytic spaces. We start with reviewing the notions of connection, curvature, and Chern classes for complex vector bundles (see \cite{Boyer_2008,Grauert-Remmert,suwa2007introduction,akhiezer2012several} for more details).
\subsubsection*{Complex Vector Bundles}
Let $M$ be a complex manifold of dimension $\mathfrak{n}$. A \emph{complex vector bundle} of \emph{rank} $r$ over $M$ is a smooth manifold $E$ together with a continuous map $\pi: E \to M$ such that there exists an open covering $\mathcal{V} = \{V_{\alpha}\}_{\alpha \in A}$ of $M$ with the following properties:
\begin{enumerate}[(i)]
	\item for each $\alpha \in A$, there is a homeomorphism
	\begin{equation*}
	\psi_{\alpha}: \pi^{-1}(V_{\alpha}) \xrightarrow{\sim} V_{\alpha} \times \cmpx^r
	\end{equation*}
	with $\operatorname{pr} \circ \psi_{\alpha} = \pi$ where $\operatorname{pr}$ denotes the projection $V_{\alpha} \times \cmpx^r \to V_{\alpha}$;
	\item for each pair $(\alpha, \beta) \in A \times A$, there is a $\mathcal{C}^{\infty}$ map
	\begin{equation*}
	g^{\alpha \beta}: V_{\alpha} \cap V_{\beta} \to \operatorname{GL}(r,\cmpx)
	\end{equation*}
	with
	\begin{equation*}
	\psi_{\alpha} \circ \psi_{\beta}^{-1} (p,\zeta) = \big( p, g^{\alpha \beta}(p) \zeta \big) \qquad \text{for} \qquad (p,\zeta) \in V_{\alpha} \cap V_{\beta} \times \cmpx^r.
	\end{equation*}
\end{enumerate}
We call $\psi_{\alpha}$ a \emph{trivialization} of $E$ on $V_{\alpha}$. We also call $g^{\alpha \beta}$ the \emph{transition matrix} of $E$ on $V_{\alpha} \cap V_{\beta}$ and the collection $\{g^{\alpha \beta}\}_{(\alpha, \beta) \in A \times A}$ the \emph{system of transition matrices} of $E$. For each point $p$ in $V_{\alpha} \cap V_{\beta} \cap V_{\gamma}$ we have the identity
\begin{equation}\label{cocyclecondition}
g^{\alpha \beta}(p) \, g^{\beta \gamma}(p) = g^{\alpha \gamma}(p) \qquad \qquad \text{(cocycle condition)}.
\end{equation}
Thus, in particular, $g^{\alpha \alpha}(p) = \mathbbm{1}$ (the identity matrix) and $g^{\beta \alpha }(p) =\big(g^{\alpha \beta}(p) \big)^{-1}$. We may think of the system $\big\{(V_{\alpha}, \psi_{\alpha}, g^{\alpha \beta})\big\}_{(\alpha, \beta) \in A \times A}$ as defining a \emph{vector bundle structure} on $E$.

Conversely, if we are given an open covering $\mathcal{V} = \{V_{\alpha}\}_{\alpha \in A}$ of $M$ and a collection $\{g^{\alpha \beta}\}_{(\alpha, \beta) \in A \times A}$ of $\mathcal{C}^{\infty}$ maps
\begin{equation*}
g^{\alpha \beta}: V_{\alpha} \cap V_{\beta} \to \operatorname{GL}(r,\cmpx),
\end{equation*}
satisfying the cocycle condition \eqref{cocyclecondition} for $p \in V_{\alpha} \cap V_{\beta} \cap V_{\gamma}$, we may construct a vector bundle as follows:  For $(p_{\alpha},\zeta_{\alpha}) \in V_{\alpha} \times \cmpx^r$ and $(p_{\beta},\zeta_{\beta}) \in V_{\beta} \times \cmpx^r$, we define
$(p_{\alpha},\zeta_{\alpha}) \sim  (p_{\beta},\zeta_{\beta})$ if and only if
\begin{equation*}
\left\{
\begin{split}
& p_{\alpha} = p_{\beta} (= p) \\
& \zeta_{\alpha} = g^{\alpha \beta}(p) \, \zeta_{\beta}.
\end{split}
\right.
\end{equation*}
Then, it is easy to see that this is an equivalence relation in the disjoint union $\bigsqcup_{\alpha} (V_{\alpha} \times \cmpx^r)$. Now, let $E$ be the quotient space $(\bigsqcup_{\alpha} (V_{\alpha} \times \cmpx^r))/\sim$. Then, since
\begin{equation*}
(V_{\alpha} \times \cmpx^r)/\sim \, = \, V_{\alpha} \times \cmpx^r,
\end{equation*}
$E$ has a vector bundle structure with $\{g^{\alpha \beta}\}_{(\alpha, \beta) \in A \times A}$ as a system of transition matrices.

A complex vector bundle over a complex manifold $M$ is said to be \emph{holomorphic} if $E$ admits a system of transition matrices $\{g^{\alpha \beta}\}_{(\alpha, \beta) \in A \times A}$ such that each $g^{\alpha \beta}$ is holomorphic. Note that in this case, $E$ has the structure of a complex manifold so that the projection $\pi: E \to M$ is a holomorphic submersion. We will come back to this point later when we study holomorphic vector bundles on complex analytic spaces.

Let $\pi: E \to M$ be a complex vector bundle of rank $r$ and $V$ an open set in $M$. A smooth complex section of $E$ on $V$ is a $\mathcal{C}^{\infty}$-map $s: V \to E\big|_{V} := \pi^{-1}(V) $ such that $\pi \circ s = \operatorname{id}_{V}$, the identity map of $V$. A vector bundle, $E$, always admits the \emph{zero section} --- i.e., the map $M \to E$ which assigns to each point $p \in M$ the zero of the vector space $E_{p}$. The set of $\mathcal{C}$ complex sections of $E$ on $V$ is denoted by $\mathcal{C}^{\infty}(V,E)$. This has a natural structure of vector space by the operations defined by $(s_1+s_2)(p)= s_1(p)+s_2(p)$ and $(c  s)(p)= c \, s(p)$ for $s_1, s_2$ and $s$ in $\mathcal{C}^{\infty}(V,E)$, $c \in \cmpx$ and $p \in V$. More precisely, a section $s$ on $V$ can be described as follows: We fix a system of transition matrices $\{g^{\alpha \beta}\}_{(\alpha, \beta) \in A \times A}$ of $E$ on an open covering $\mathcal{V} = \{V_{\alpha}\}_{\alpha \in A}$. Using the $\mathcal{C}^{\infty}$-diffeomorphism $\psi_{\alpha}: E\big|_{V_{\alpha}} \xrightarrow{\sim} V_{\alpha} \times \cmpx^r$, we may write
\begin{equation*}
\psi_{\alpha}\big(s(p)\big) = \big(p, s^{\alpha}(p) \big) \quad \text{for} \quad p \in V \cap V_{\alpha},
\end{equation*}
where $s^{\alpha}$ is a $\mathcal{C}^{\infty}$-map from $V \cap V_{\alpha}$ into $\cmpx^r$. For each point $p \in V \cap V_{\alpha} \cap V_{\beta}$, we have
\begin{equation}\label{sectiontrans}
s^{\alpha}(p) = g^{\alpha \beta}(p)  \, s^{\beta}(p).
\end{equation}
Conversely, suppose we have a system $\{s^{\alpha}\}_{\alpha \in A}$ of  $\mathcal{C}^{\infty}$-maps satisfying \eqref{sectiontrans}. Then, by setting $s(p) = \psi_{\alpha}^{-1}\big(p, s^{\alpha}(p) \big)$ for $p$ in $V \cap V_{\alpha}$, we have a section $s$ over $V$.

For $k = 1,\dots,r$, a \emph{$k$-frame} of $E$ on an open set $V \subset M$ is a collection $\boldsymbol{s} = (s_1, \dots, s_k)$ of $k$ sections $s_i$ of $E$ on $V$ linearly independent at each point in $V$. An $r$-frame is simply called a \emph{frame}. Note that a frame of $E$ on $V$ determines a trivialization of $E$ over $V$.

Let us denote by $\mathscr{E}_{M}$ the sheaf of germs of $\mathcal{C}^{\infty}$ complex functions on $M$. If $\pi: E \to M$ is a $\mathcal{C}^{\infty}$ complex vector bundle of rank $r$ over $M$, we denote by $\mathscr{E}(E)$ the sheaf of germs of $\mathcal{C}^{\infty}$-sections of $E$ --- i.e. the sheaf whose space of sections on an open subset $V \subset M$ is $\mathscr{E}(E)\big|_{V} = \mathcal{C}^{\infty}(V,E)$. It is clear that $\mathscr{E}(E)$ is a $\mathscr{E}_{M}$-module. Furthermore, the sheaf $\mathscr{E}(E)$ is a \emph{locally free $\mathscr{E}_{M}$-module of rank $r$}: There exists a covering $\mathcal{V} = \{V_{\alpha}\}_{\alpha \in A}$ of $M$ and a sheaf isomorphism
\begin{equation*}
\psi_{\alpha}^{\ast}: \mathscr{E}(E)\big|_{V_{\alpha}} \to \mathscr{E}_{V_{\alpha}}^r\qquad \mathscr{E}_{V_{\alpha}}^r := \underbrace{\mathscr{E}_{V_{\alpha}} \oplus \dotsm \oplus \mathscr{E}_{V_{\alpha}}}_{r}.
\end{equation*}
Then, we have transition isomorphisms $\psi_{\alpha}^{\ast} \circ (\psi_{\beta}^{\ast})^{-1} : \mathscr{E}_{M}^{r} \to \mathscr{E}_{M}^{r}$ defined on $V_{\alpha} \cap V_{\beta}$, and such isomorphism is the multiplication by an invertible matrix with $\mathcal{C}^{\infty}$ coefficients on $V_{\alpha} \cap V_{\beta}$. The concepts of complex vector bundles and locally free  $\mathscr{E}_{M}$-modules are thus completely equivalent.

If we are given some vector bundles, we may construct new ones by algebraic operations. Thus, we let $E_1$ and $E_2$ be complex vector bundles of rank $r_1$ and $r_2$ on $M$. We may construct the direct sum $E_1 \oplus E_2$, the tensor product $E_1 \otimes E_2$, and the homomorphism $\operatorname{Hom}(E_1,E_2)$. Note that there is a natural isomorphism $\operatorname{Hom}(E_1,E_2) = E_1^{\ast} \otimes E_2$. We may also construct the complex conjugate $\bar{E_1}$ and the k-th exterior power $\bigwedge^k E_1$. The bundles $E_1$ and $E_2$ can be trivialized over the same covering $\mathcal{V} = \{V_{\alpha}\}_{\alpha \in A}$ of $M$ (otherwise take a common refinement). If $\{g_1^{\alpha \beta}\}_{\alpha,\beta \in A}$ and $\{g_2^{\alpha \beta}\}_{\alpha,\beta \in A}$ are the corresponding transition matrices of $E_1$ and $E_2$, then for example $E_1 \otimes E_2$, $\bigwedge^k E_1$, $E^{\ast}_1$ are the bundles defined by the transition matrices $g_1^{\alpha \beta} \otimes g_2^{\alpha \beta}$, $\bigwedge^k g_1^{\alpha \beta}$, $\big((g_1^{\alpha \beta})^T\big)^{-1}$ where ``$\cdot^T$'' denotes transposition.
\subsubsection*{Connections and Curvature}
Now, for a $\mathcal{C}^{\infty}$ complex vector bundle $E$ of rank $r$ on $M$, we let $\mathcal{A}^{k}(V,E)$ be the vector space of $\mathcal{C}^{\infty}$ sections of $\big( \bigwedge^k \, T^{\ast} M \big) \otimes E$ on $V \subset M$, which are called \emph{differential forms on $V$ with values in $E$}. Thus $\mathcal{A}^{0}(V,E) := \mathcal{C}^{\infty}(V,E)$ is the $\mathcal{A}^{0}(V):= \mathscr{E}_{M}\big|_{_V}$-module of $\mathcal{C}^{\infty}$ sections of $E$.
\begin{definition}[Connection]
	A \emph{connection} for $E$ is a $\cmpx$-linear map
	\begin{equation*}
	\nabla: \mathcal{A}^{0}(M,E) \to \mathcal{A}^{1}(M,E)
	\end{equation*}
	satisfying
	\begin{equation*}
	\nabla(f s) = df \otimes s + f \nabla(s) \qquad \text{for} \qquad f \in \mathcal{A}^{0}(M) \quad\text{and}\quad s \in \mathcal{A}^{0}(M,E).
	\end{equation*}
\end{definition}

A connection $\nabla$ is a local operator --- i.e., if a section $s$ is identically 0 on an open set $V \subset M$, so is $\nabla(s)$. Thus, the restriction of $\nabla$ to an open set $V$ makes sense, and it is a connection for $E\big|_{V}$. In addition, from the above definition, we conclude that if $\nabla_1,\dots,\nabla_k$ are connections for $E$ and $f_1, \dots, f_k$ are $\mathcal{C}^{\infty}$ functions on $M$ with $\sum_{i =1}^{k} f_i = 1$, then $\sum_{i =1}^{k} f_i \nabla_i$ will also be a connection on $E$.

If $\nabla$ is a connection for $E$, it induces a $\cmpx$-linear map
\begin{equation*}
\nabla: \mathcal{A}^{1}(M,E) \to \mathcal{A}^{2}(M,E)
\end{equation*}
and satisfying
\begin{equation*}
\nabla(\omega \otimes s) = d\omega \otimes s - \omega \wedge \nabla(s) \qquad \text{for} \qquad \omega \in \mathcal{A}^{1}(M) \quad\text{and}\quad s \in \mathcal{A}^{0}(M,E).
\end{equation*}
The composition
\begin{equation*}
\Theta = \nabla \circ \nabla : \mathcal{A}^{0}(M,E) \to \mathcal{A}^{2}(M,E)
\end{equation*}
is called the \emph{curvature} of $\nabla$. It is not difficult to see that
\begin{equation*}
\Theta (f s) = f \Theta(s) \qquad \text{for} \qquad f \in \mathcal{A}^{0}(M) \quad\text{and}\quad s \in \mathcal{A}^{0}(M,E).
\end{equation*}
The fact that a connection is a local operator allows us to get local representations of it and its curvature by matrices whose entries are differential forms. Thus, suppose that $\nabla$ is a connection for a complex vector bundle $E$ of rank $r$ and that $E$ is trivial on $V$ --- i.e. $E\big|_{V} \cong V \times \cmpx^r$. If $\boldsymbol{s} = (s_1,\dots,s_r)$ is a frame of $E$ on $V$, then we may write 
\begin{equation*}
\nabla(s_i) = \sum_{j=1}^{r} A_{ij} \otimes s_j \quad \text{with} \quad A_{ij} \in \mathcal{A}^{1}(V) \quad \text{and for} \quad i=1,\dots,r.
\end{equation*}
We call $A = (A_{ij})$ the \emph{connection matrix with respect to $s$}. For an arbitrary section $s$ on $V$, we may write $s = \sum_{i=1}^{r} f_i s_i$ with $f_i$ being $\mathcal{C}^{\infty}$-functions on $V$ and we compute
\begin{equation*}
\nabla(s) = \sum_{i=1}^{r} \left(df_i + \sum_{j=1}^r f_j A_{ji}\right) \otimes s_i.
\end{equation*}
The connection $\nabla$ is \emph{trivial} with respect to $s$, if and only if $A = 0$. Thus in this case we have $\nabla(s) = \sum_{i=1}^{r} df_i\otimes s_i$. Also, from the definition, we compute to get
\begin{equation*}
\Theta(s_i) = \sum_{j=1}^{r} \theta_{ij} \otimes s_j \quad \text{with} \quad \theta_{ij}= dA_{ij} - \sum_{k=1}^{r} A_{ik} \wedge A_{kj}.
\end{equation*}
We call $\theta = (\theta_{ij})$ the \emph{curvature matrix with respect to $s$}. If $\boldsymbol{s}'=(s'_1,\dots,s'_r)$ is another frame of $E$ on $V'$, we have $s'_i = \sum_{j=1}^{r} \mathsf{g}_{ij} s_j$ for some $\mathcal{C}^{\infty}$ functions $\mathsf{g}_{ij}$ on $V \cap V'$. The matrix $\mathsf{g} = (\mathsf{g}_{ij})$ is non-singular at each point of $V \cap V'$. If we denote by $A'$ and $\theta'$ the connection and curvature matrices of $\nabla$ with respect to $\boldsymbol{s}'$, we obtain the \emph{gauge transformation law}
\begin{equation}\label{gaugetrans}
A' =  \mathsf{g} A \mathsf{g}^{-1} + \mathsf{g}^{-1} d\mathsf{g} \qquad \text{and} \qquad \theta' = \mathsf{g} \theta \mathsf{g}^{-1} \quad \text{on} \quad V \cap V'.
\end{equation}

Now, suppose that $(M,\mathcal{J})$\footnote{Here $\mathcal{J}$ denotes an (integrable) almost complex structure.} is a complex manifold, and $E$ is a complex vector bundle on $M$. We can consider the tensor product bundle $\bigwedge^{k,l} T^{\ast}M \otimes E$, and we let $\mathscr{E}^{(k,l)}_{M}(E)$ denote the sheaf of germs of smooth sections of $\bigwedge^{k,l} T^{\ast}M \otimes E$. Smooth sections of this sheaf are $(k,l)$-forms with values in $E$, the set of which we denote by $\mathcal{A}^{(k,l)}(E)$. The connection $\nabla$ in $E$ induces a connection, also written as $\nabla$, in $\mathcal{A}^{(k,l)}(E)$. This connection splits as $\nabla = \nabla^{(1,0)} \oplus \nabla^{(0,1)}$ giving maps
\begin{equation*}
\nabla^{(1,0)}:  \mathcal{A}^{(k,l)}(E) \to  \mathcal{A}^{(k+1,l)}(E) \qquad \text{and} \qquad \nabla^{(0,1)}:  \mathcal{A}^{(k,l)}(E) \to  \mathcal{A}^{(k,l+1)}(E). 
\end{equation*}
\begin{theorem}
	A smooth, complex vector bundle $E$ over a complex manifold $M$ admits a holomorphic structure if and only if there exists a connection $\nabla$ in $E$ such that $\nabla^{0,1} = \bar{\partial}$.
\end{theorem}
\noindent The holomorphic structure in $E$ is \emph{uniquely determined} by the condition $\nabla^{0,1} = \bar{\partial}$, and this condition says that the $(0, 2)$ component $\nabla^{0,1} \circ \nabla^{0,1}$ of the curvature of $\nabla$ vanishes.\\

Using partitions of unity, one easily sees that Hermitian metrics exist on every complex vector bundle:
\begin{definition}[Hermitian metric]
	A \emph{Hermitian metric} $h$ on a complex vector bundle $E$ is an assignment of a Hermitian inner product to each fibre $E_p$ of $E$ that varies smoothly with $p$. A connection $\nabla$ in $E$ is called a Hermitian connection if $\nabla h = 0$ for some Hermitian metric $h$. A vector bundle equipped with a Hermitian metric is often called an \emph{Hermitian vector bundle}.
\end{definition}
\noindent We then have:
\begin{proposition}
	Let $E$ be a holomorphic vector bundle with a Hermitian metric $h$. Then there exists a {\bfseries unique} Hermitian connection $\nabla$ such that $\nabla^{0,1} = \bar{\partial}$. This unique connection is called {\bfseries the} ``Hermitian connection''.
\end{proposition}
\subsubsection*{Chern Forms}
Consider the space $\mathsf{M}_{r\times r}(\cmpx)$ of complex $r \times r$ matrices. For any $\mathtt{M} \in \mathsf{M}_{r\times r}(\cmpx)$, we define
\begin{equation*}
\det(\mathtt{M}  + \lambda \mathbbm{1}) = \sigma_r(\mathtt{M}) + \lambda \sigma_{r-1}(\mathtt{M}) + \dotsm + \lambda^{r-1} \sigma_1(\mathtt{M}) + \lambda^r. 
\end{equation*}
Clearly, for any $i=1,\dots,r$ the function $\sigma_i: \mathsf{M}_{r\times r}(\cmpx) \to \cmpx$ is a $\operatorname{GL}(r,\cmpx)$-invariant, complex homogeneous polynomial of $\deg(\sigma_i) = i$. Note that $\sigma_i$ is the i-th elementary symmetric function of the eigenvalues of $\mathtt{M}$. In particular, $\sigma_r(\mathtt{M}) = \det(\mathtt{M})$ and $\sigma_1(\mathtt{M}) = \tr(\mathtt{M})$. Since differential forms of even degrees commute with one another with respect to the exterior product, we may treat the curvature 2-form $\theta$ as an ordinary matrix whose entries are numbers. Thus, we define
\begin{definition}[Chern form]
	Let $E \to M$ be a rank $r$ complex vector bundle over $M$, and let $\nabla$ be a complex connection on $E$ with curvature 2-form $\theta$. For each $i=1,\dots,r$ we define the $2i$-form
	\begin{equation*}
	c_i(E,\nabla) := \sigma_i\left(\frac{\sqrt{-1}}{2\pi} \theta \right)
	\end{equation*}
	and call it the \emph{i-th Chern form} of $E$.
\end{definition}
\noindent Additionally, we have the following definition
\begin{definition}[Chern classes]
	Given $(E,\nabla)$ and any $1 \leq i \leq r$, the i-th Chern form $c_i(E,\nabla)$ is closed. Furthermore, if $\nabla'$ is another complex connection on $E$ the difference $c_i(E,\nabla) - c_i(E,\nabla')$ is exact --- i.e. the cohomology class $[c_i(E,\nabla)] \in H^{2i}_{dR}(M) \simeq H^{2i}(M,\cmpx)$ is independent of $\nabla$. The resulting cohomology class is called the \emph{i-th Chern class} of $E$ and is denoted by $c_i(E)$.
\end{definition}
\begin{remark}
	It is known that $c_i(E)$ is in the image of the canonical homomorphism
	\begin{equation*}
	H^{2i}(M,\mathbb{Z}) \to H^{2i}(M,\cmpx).
	\end{equation*}
	In fact, it is possible to define $c_i(E)$ in $H^{2i}(M,\mathbb{Z})$ using the obstruction theory; it is the primary obstruction to constructing $r-i+1$ sections linearly independent everywhere \cite[\textsection I]{akhiezer2012several}.
\end{remark}
\subsubsection*{Complex Line Bundles}
Assume now that $L \to M$ is a line bundle ($r = 1$). Then, every collection of transition functions $\{g^{\alpha \beta}\}_{(\alpha, \beta) \in A \times A}$ defines a Cech 1-cocycle with values in the multiplicative sheaf $\mathscr{E}_{M}^{\ast}$ of invertible $\mathcal{C}^{\infty}$ complex functions on $M$. In fact, the definition of the Cech differential (see e.g. \cite[\textsection1.3]{brylinski2007loop}) gives $(\delta g)^{\alpha \beta \gamma} = g^{\beta \gamma} (g^{\alpha \gamma})^{-1} g^{\alpha \beta}$, and we have $\delta g = 1$ in view of \eqref{cocyclecondition}. Let $\psi'_{\alpha}$ be another family of trivializations and $\{g'^{\alpha \beta}\}_{(\alpha, \beta) \in A \times A}$ the associated cocycle (it is no loss of generality to assume that both are defined on the same covering since we may otherwise take a refinement). Then we have
\begin{equation*}
\psi'_{\alpha} \circ \psi_{\alpha}^{-1}: V_{\alpha} \times \cmpx \to V_{\alpha} \times \cmpx, \qquad (p,\zeta) \mapsto (p, u_{\alpha}(p) \zeta), \qquad u_{\alpha} \in \mathscr{E}_{M}^{\ast}(V_{\alpha}).
\end{equation*}
It follows that $g^{\alpha \beta} = g'^{\alpha \beta} u_{\alpha}^{-1}   u_{\beta}$ --- i.e. the Cech 1-cocycles $g^{\alpha \beta}$ and $g'^{\alpha \beta}$ differ only by the Cech 1-coboundary $\delta u$. Therefore, there exists a well-defined map which associates to every complex line bundle $L$ over $M$ the Cech cohomology class $\{g^{\alpha \beta}\}_{(\alpha, \beta) \in A \times A} \in H^1(M, \mathscr{E}_{M}^{\ast})$ of its cocycle of transition functions. It is easy to verify that the cohomology classes associated to two complex line bundles $L$ and $L'$ are equal if and only if these bundles are isomorphic. It is also clear that the multiplicative group structure on $H^1(M, \mathscr{E}_{M}^{\ast})$ corresponds to the tensor product of line bundles (the inverse of a line bundle being its dual). We may summarize this discussion by the following:
\begin{proposition}\label{prop:CechCohomology}
	The group of isomorphism classes of complex $\mathcal{C}^{\infty}$ line bundles is in one-to-one correspondence with the Cech cohomology group $\check{H}^1(M, \mathscr{E}_{M}^{\ast})$.
\end{proposition}

Now, let $\nabla$ be a connection on $L$ with curvature 2-form $\theta$. The de Rham class $[\theta] \in H^2_{dR}(M,\cmpx)$ does not depend on the particular choice of $\nabla$. If $\nabla$ is chosen to be hermitian with respect to a given hermitian metric on $L$ (such a connection can always be constructed by means of a partition of unity) then $\sqrt{-1} \theta$ is a real 2-form, thus $[\sqrt{-1} \theta] \in H^2_{dR}(M,\mathbb{R})$. 

Consider now the one-to-one correspondence given by Proposition~\ref{prop:CechCohomology} and the exponential short exact sequence of sheaves on $M$
\begin{equation*}
0 \to \mathbb{Z} \xrightarrow{\imath} \mathscr{E}_{M} \xrightarrow{\exp} \mathscr{E}_{M}^{\ast} \to 0,
\end{equation*}
where the map $\imath$ is $\imath(k) = 2 \pi \sqrt{-1} k$ and the exponential map sends the germ $f$ of any complex $\mathcal{C}^{\infty}$ function to $\exp(f)$. Since $\mathscr{E}_{M}$ is a fine sheaf,\footnote{A fine sheaf over a over \emph{paracompact} Hausdorff space $M$ is one with ``partitions of unity.'' More precisely, for any open cover of the space $M$, we can find a family of homomorphisms from the sheaf to itself with sum $1$ such that each homomorphism is $0$ outside some element of the open cover (see e.g., \cite[Def.~4.35]{Voisin_2003} for more details).} we have $H^q(M,\mathscr{E}_{M}) = 0$ for all $q > 0$; in particular $H^1(M,\mathscr{E}_{M}) = H^2(M,\mathscr{E}_{M}) = 0$. So, the induced long exact cohomology sequence 
\begin{equation*}
\dotsm \to H^1(M,\mathscr{E}_{M}) \to H^1(M,\mathscr{E}_{M}^{\ast}) \to H^2(M,\mathbb{Z}) \to H^2(M,\mathscr{E}_{M}) \to \dotsm
\end{equation*}
gives an isomorphism $H^1(M,\mathscr{E}_{M}^{\ast}) \simeq H^2(M,\mathbb{Z})$ which says that the \emph{topological invariant $H^2(M,\mathbb{Z})$} can be thought of as the \emph{group of complex line bundles on $M$}. This isomorphism is realized by associating to a complex line bundle $L$ its first Chern class $c_1(L)$.

The natural morphism
\begin{equation*}
H^2(M,\mathbb{Z}) \to H^2(M,\mathbb{R}) \simeq H^2_{dR}(M,\mathbb{R})  
\end{equation*}
results in the following theorem
\begin{theorem}
	The image of $c_1(L)$ in $H^2_{dR}(M, \mathbb{R})$ coincides with the de Rham cohomology class $[\sqrt{-1} \theta]$ associated to any (hermitian) connection $\nabla$ on $L$.
\end{theorem}
\subsubsection*{Holomorphic Vector Bundles on Analytic Spaces}
Let us now generalize our discussion of holomorphic vector bundles to those defined over complex analytic spaces. We have:
\begin{definition}[Holomorphic vector bundle]
	Let $\pi: E \to X$ be an analytic map between reduced analytic spaces such that every fiber $E_x := \pi^{-1}(x)$ over a point $x \in X$, is equipped with the structure of a $r$-dimensional complex vector space. Then, $\pi: E \to X$ will be called a \emph{holomorphic vector bundle of rank $r$} on $X$ if every point $x \in X$ has an open neighborhood $V$ in $X$ such that the restricted map $\pi\big|_{E|_{V} := \pi^{-1}(V)}: E\big|_{V} \to V$ is (analytically) \emph{trivial} --- i.e. there exists a biholomorphic map $E\big|_{V} \xrightarrow{\cong} V \times \cmpx^{r}$, called the \emph{holomorphic trivialization} of $E$ over $V$, which maps every fiber $E_x$ for $x \in V$ onto $x \times \cmpx^{r}$ as an isomorphism of complex vector spaces. A holomorphic vector bundle of rank $r=1$ on $X$ is called a \emph{holomorphic line bundle} on $X$. 
\end{definition}
\noindent An analytic map $E \to E'$ between holomorphic vector bundles is called a \emph{bundle map} if it is fiber preserving and if all induced maps $E_x \to E'_x$ are linear; clearly, the holomorphic vector bundles with bundle maps as morphisms form a category.\\

Assume now that $E$ is a holomorphic vector bundle on $X$ of rank $r$. A \emph{holomorphic section} of $E$ over $V \subset X$ is a holomorphic map $s: V \to \pi^{-1}(V) \subset E$ such that $\pi \circ s = \operatorname{id}_{V}$. If $\pi^{-1}(V) \cong V \times \cmpx^{r}$ then $s$ is simply given by a $r$-tuple of holomorphic functions on $V$. Hence, local holomorphic sections in $E$ determine, in a natural way, a canonical presheaf on $X$ which gives rise to the analytic sheaf $\mathscr{O}_X(E)$ on $X$ of germs of holomorphic sections in $E$. Such a sheaf is always \emph{locally free of the same rank as the rank of the vector bundle}: If $E\big|_{V}$ is trivial, we have an isomorphism $\mathscr{O}(E)\big|_{V} = \mathscr{O}_{V}^r$. It follows that if $E$ is the trivial line bundle $X \times \cmpx$, the sheaf $\mathscr{O}(E = X \times \cmpx)$ coincides with the structure sheaf $\mathscr{O}_{X}$ of the complex analytic space $X$. Moreover, the cohomology of a holomorphic vector bundle $E$ over $X$ is defined to be the sheaf cohomology of $\mathscr{O}(E)$. In particular, we have
\begin{equation*}
H^{0}\big(X, \mathscr{O}(E)\big) = \Gamma\big(X,E\big)
\end{equation*}
the \emph{space of global holomorphic sections of $E$}.

To study the holomorphic line bundles on the analytic space $X$, we consider the exact sequence
\begin{equation*}
0 \to \mathbb{Z} \xrightarrow{\imath} \mathscr{O}_{X} \xrightarrow{\exp} \mathscr{O}_{X}^{\ast} \to 0,
\end{equation*}
where the maps $\imath$ and $\exp$ are defined as before and $\mathscr{O}_{X}^{\ast}$ denotes the sheaf of invertible elements in $ \mathscr{O}_{X}$ --- in other words, the sheaf of nowhere-vanishing holomorphic function on $X$. This induces a long exact sequence in cohomology,
\begin{equation*}
\dotsm \to H^1(X,\mathscr{O}_{X}) \to H^1(X,\mathscr{O}_{X}^{\ast}) \xrightarrow{\delta} H^2(X,\mathbb{Z}) \to H^2(X,\mathscr{O}_{X}) \to \dotsm.
\end{equation*}
The group $H^1(X,\mathscr{O}_{X}^{\ast})$ represents the group of holomorphic line bundles on the analytic space $X$ with group multiplication being the tensor product, and the inverse bundle being the dual bundle. This group is called the \emph{Picard group} of $X$ and often denoted by $\operatorname{Pic}(X)$. As seen above, the connecting homomorphism $\delta$ takes a holomorphic line bundle $\linebundle$ to its first Chern class $c_1(\linebundle)$, and the group $H^2(X, \mathbb{Z})$ is isomorphic to the group of topological complex line bundles on $X$. So if $H^2(X,\mathscr{O}_{X}) \neq 0$, we see that not every complex line bundle gives rise to a holomorphic line bundle. Similarly, if $H^1(X,\mathscr{O}_{X}) \neq 0$, there can be inequivalent holomorphic bundles associated to the same complex line bundle. The kernel of the map $\delta$ is denoted by $\operatorname{Pic}^{0}(X)$ and represents the subgroup of holomorphic line bundles that are trivial topologically.

There is a holomorphic line bundle canonically associated with every analytic space:
\begin{definition}[Canonical line bundle]
	Let $X$ be a reduced analytic space of dimension $\mathfrak{n}$. The $\mathfrak{n}$-th exterior power $\bigwedge^{\mathfrak{n}} T^{\ast}_{(1,0)}X$ is a holomorphic line bundle, called the \emph{canonical line bundle} and denoted by $\mathcal{K}_{X}$. The dual or inverse line bundle $\mathcal{K}^{-1}_{X}$ is called the \emph{anticanonical line bundle}.
\end{definition}
\noindent When the underlying analytic space $X$ is understood we often write just $\mathcal{K}$ for $\mathcal{K}_{X}$. It is easy to see that 
\begin{proposition}
	The first Chern class of $\mathcal{K}_{X}$ satisfies $c_1(\mathcal{K}_{X}) = - c_1(X)$.
\end{proposition}
\subsubsection*{Meromorphic Functions and Divisors}
There are two equivalent ways to describe divisors on smooth complex manifolds. However, they are not equivalent for singular analytic spaces. We discuss both of these notions here.

\begin{definition}[Weil divisor]\label{def:Weil}
	A \emph{Weil divisor} $\mathcal{D}$ on an analytic space $X$ is a locally finite formal linear combination of irreducible analytic hypersurfaces $H_i$
	\begin{equation*}
	\mathcal{D} = \sum_i a_i H_i \quad \text{with} \quad a_i \in \mathbb{Z},
	\end{equation*}
	where locally finite means that every point $x \in X$ has a neighborhood intersecting only finitely many of the $H_i$'s. $\mathcal{D}$ is said to be \emph{effective} if $a_i \geq 0$ for all $i$ (not all $a_i$'s equal to zero). For a Weil divisor $\mathcal{D} = \sum_i a_i H_i$, we set $\operatorname{Supp}(\mathcal{D}) := \bigcup_{i} H_i$ and call it the \emph{support} of $\mathcal{D}$. Additionally, the coefficient $a_i$ is called the \emph{multiplicity of $\mathcal{D}$ along $H_i$} and will be denoted by $\operatorname{mult}_{H_i}(\mathcal{D})$; we set $\operatorname{mult}_{H}(\mathcal{D}) = 0$ for every other irreducible divisor $H \neq H_i \,\, \forall i$. Finally, the degree of $\mathcal{D}$ is denoted by $\deg(\mathcal{D})$ and is defined as the sum of coefficients $a_i$ --- i.e. $\deg(\mathcal{D}) := \sum_i\operatorname{mult}_{H_i}(\mathcal{D}) = \sum_i a_i$.
	
	Under the formal sum operation, Weil divisors form a group called the divisor group and denoted by $\operatorname{Div}(X)$. It then follows, from definition~\ref{def:subvariety} of hypersurfaces, that a Weil divisor is described locally by the zero set of holomorphic functions.
\end{definition}

Now, let us recall that a \emph{meromorphic function} on an open set $V \subset X$ is a ratio $f/g$ of relatively prime holomorphic functions $f$ and $g$ on $V$. We will denote by $\mathscr{M}_X$ the sheaf of meromorphic functions on $X$ and by $\mathscr{M}^{\ast}_X$ the subsheaf of \emph{not-identically-zero meromorphic functions}. We also denoted by $\mathscr{O}^{\ast}_X$ the subsheaf of invertible elements in $\mathscr{O}_X$ and called it the sheaf of \emph{nowhere-vanishing holomorphic functions}. We have the following short, exact sequence
\begin{equation}\label{divisorsheaf}
0 \to \mathscr{O}^{\ast}_X \to \mathscr{M}^{\ast}_X \to \mathscr{M}^{\ast}_X / \mathscr{O}^{\ast}_X  \to 0 .
\end{equation}
Then, we can define:
\begin{definition}[Cartier divisor]\label{def:Cartier}
	A \emph{Cartier divisor} on $X$ is a global section of the sheaf $\mathscr{M}^{\ast}_X / \mathscr{O}^{\ast}_X$ --- i.e. an element of the group $H^0(X, \mathscr{M}^{\ast}_X / \mathscr{O}^{\ast}_X)$. Any Cartier divisor can be represented by giving an open covering $\mathcal{V} := \{V_{\alpha}\}_{\alpha \in A}$ of $X$ and, for all $\alpha \in A$, an element $\phi^{\alpha} := f^{\alpha}/g^{\alpha} \in \mathscr{M}^{\ast}_X(V_{\alpha})$ such that $\phi^{\alpha} = u^{\alpha \beta} \phi^{\beta}$ on any intersection $V_{\alpha} \cap V_{\beta}$ with $u^{\alpha \beta} \in \mathscr{O}^{\ast}_X(V_{\alpha} \cap V_{\beta})$. A Cartier divisor $\mathcal{D}$ on $X$ is called \emph{effective} if it can be represented by the system $\{(V_{\alpha},f^{\alpha})\}_{\alpha \in A}$ with all local equations $f^{\alpha} \in \Gamma(V_{\alpha},\mathscr{O}_X)$. Additionally, two systems $\{(V_{\alpha}, \phi^{\alpha})\}_{\alpha \in A}$ and  $\{(V'_{\beta}, \phi'^{\beta})\}_{\beta \in B}$ represent the same Cartier divisor if and only if on $V_{\alpha} \cap V'_{\beta}$, $\phi^{\alpha}$ and $\phi'^{\beta}$ differ by a multiplicative factor in $\mathscr{O}^{\ast}_X(V_{\alpha} \cap V'_{\beta})$. The abelian group of Cartier divisors on $X$ will denoted by $H^0(X, \mathscr{M}^{\ast}_X / \mathscr{O}^{\ast}_X)$. If $\mathcal{D}_1 :=  \{(V^1_{\alpha}, \phi_{1}^{\alpha})\}_{\alpha \in A}$ and $\mathcal{D}_2 :=  \{(V^2_{\beta}, \phi_2^{\beta})\}_{\beta \in B}$ then $\mathcal{D}_1+ \mathcal{D}_2 = \{(V^1_{\alpha} \cap V^2_{\beta}, \phi_{1}^{\alpha}\phi_2^{\beta})\}_{\alpha \in A, \beta \in B}$.
\end{definition}

Since on a smooth analytic space $X$ (i.e. a complex manifold) the local rings $\mathscr{O}_{X,x}$ are unique factorization domains (UFD), Weil divisors and Cartier divisors coincide: If we cover $X$ by open sets $\{V_{\alpha}\}_{\alpha \in A}$ so that $H_i$ is defined by $f_i^{\alpha}$ on $V_{\alpha}$, we have the meromorphic function $\phi^{\alpha} = \prod_{i} (f_i^{\alpha})^{a_i}$ which is determined by the expression of the Weil divisor $\mathcal{D} = \sum_i a_i H_i$. The systems $\{(V_{\alpha}, \phi^{\alpha})\}_{\alpha \in A}$ would then correspond to a Cartier divisor.
\begin{theorem}
	Let $M$ be a smooth complex manifold. Then there is an isomorphism
	\begin{equation*}
	\operatorname{Div}(M) \simeq H^0(M,\mathscr{M}^{\ast}_M / \mathscr{O}^{\ast}_M).
	\end{equation*}
\end{theorem}
\noindent On smooth complex manifolds $M$, such as the regular locus of an analytic space, we shall often identify Weil divisors and Cartier divisors by just referring to a \emph{divisor}, $\text{Div}(M)$. This isomorphism does \emph{not} hold on singular analytic spaces: Let $X$ be a normal analytic space and let $\mathcal{D}^{\text{reg}} = \sum_i a_i H_i^{\text{reg}}$ be a Weil divisor defined on the regular locus $\operatorname{Reg}(X)$. Since the singular set of $X$ has codimension at least 2, the Remmert--Stein extension theorem (see e.g. \cite[p.~181]{Grauert-Remmert}) ensures that $\mathcal{D}^{\text{reg}}$ admits a unique extension to a Weil divisor $\mathcal{D}$ on $X$. However, not every Cartier divisor on $\operatorname{Reg}(X)$ extends to a Cartier divisor on $X$. Thus, the group of Cartier divisors $H^0(X, \mathscr{M}^{\ast}_X / \mathscr{O}^{\ast}_X)$ on a singular analytic space $X$ is identified with a subgroup of $\operatorname{Div}(X)$. 

Any global section $\phi \in \Gamma(X, \mathscr{M}_X^{\ast})$ determines a \emph{principal Cartier divisor} $(\phi):= \{(X,\phi)\}$ by taking all local equations equal to $\phi$. Equivalently, if $\phi$ is a global meromorphic function on $X$ which can be written locally as $\phi= f/g$, we may consider a Weil divisor $(\phi) = \ord(f) Z_{f} - \ord(g) Z_g$ where $Z_f$ denotes the zero set of the holomorphic function $f$ and $\ord(f)$ denotes its order of vanishing. Then, two divisors $\mathcal{D}$ and $\mathcal{D}'$ on $X$ are said to be \emph{linearly equivalent}, written $\mathcal{D} \sim \mathcal{D}'$, if $\mathcal{D}' = \mathcal{D} + (\phi)$, where $(\phi)$ denotes the principal divisor defined by the global meromorphic function $\phi$. We denote by $[\mathcal{D}]$ the set of all divisors on $X$ that are linearly equivalent to $\mathcal{D}$. It is called the \emph{linear system} of divisors defined by $\mathcal{D}$. The common intersection $\bigcap_{\mathcal{D}' \in [D]} \mathcal{D}'$  is called the \emph{base locus} of linear system $[\mathcal{D}]$. We will also denote by $\operatorname{Cl}(X)$ the \emph{divisor class group} of Weil divisors modulo linear equivalence, and by $\operatorname{CaCl}(X)$ the \emph{group of Cartier divisor classes} (Cartier divisors modulo principal divisors). On a singular analytic space, the group $\operatorname{CaCl}(X)$ is generally a subgroup of the divisor class group $\operatorname{Cl}(X)$.

We now describe the relationship between line bundles and divisors: From the short exact sequence \eqref{divisorsheaf} one has
\begin{equation*}
0 \to H^0(M,\mathscr{M}^{\ast}_X / \mathscr{O}^{\ast}_X)/H^0(M,\mathscr{M}_{X}) \to H^1(M,\mathscr{O}_{X}^{\ast}) \to H^1(M,\mathscr{M}_{X}^{\ast}).
\end{equation*}
This says that every divisor $\mathcal{D}$ on $M$ determines a holomorphic line bundle $\linebundle(\mathcal{D})$, and the line bundle $\linebundle(\mathcal{D})$ is holomorphically trivial if and only if $\mathcal{D}$ is a \emph{principal divisor} --- i.e. the divisor $(\phi)$ of a global meromorphic function. The holomorphic line bundle $\linebundle(\mathcal{D})$ has as the system of transition functions, the collection $\{u^{\alpha \beta}\}$ of nowhere vanishing holomorphic functions $u^{\alpha \beta} \in \mathscr{O}^{\ast}_X(V_{\alpha} \cap V_{\beta})$ defined uniquely in terms of Cartier divisors in \ref{def:Cartier}. If $\sum_i a_i H_i$ is the Weil divisor corresponding to the Cartier divisors $\mathcal{D}$, we may write $\linebundle(\mathcal{D}) = \bigotimes_{i} \linebundle_i^{a_i}$ where $\linebundle_i := \linebundle(H_i)$ and $\linebundle_i^{a_i}$ denotes the tensor product of $a_i$ copies of $\linebundle_i$, for $a_i > 0$, and the tensor product of $-a_i$ copies of $\linebundle_i^{\ast}$, for $a_i < 0$. The quotient $H^0(M,\mathscr{M}^{\ast}_X / \mathscr{O}^{\ast}_X)/H^0(M,\mathscr{M}_{X})$ is precisely the Cartier divisor class group $\operatorname{CaCl}(M)$ defined before. Furthermore, $H^1(M,\mathscr{M}_{X}^{\ast})=0$ if and only if every holomorphic line bundle on $M$ has a global meromorphic section. In this case, we get an isomorphism between the Cartier divisor class group and the Picard group. This happens, for example, for smooth projective algebraic varieties (such as Riemann surfaces):
\begin{proposition}
	Let $X$ be a smooth projective algebraic variety. Then, we have $\operatorname{Cl}(X) \simeq \operatorname{CaCl}(X) \simeq \operatorname{Pic}(X )$.
\end{proposition}
Finally, let us finish this subsection by pointing out that if each \emph{prime divisor} (i.e. irreducible hypersurface) $H_i$ in a Weil divisor $\mathcal{D} = \sum_{i} a_i H_i$ is compact, up to linear equivalence, $\mathcal{D}$ defines a homology class $[\mathcal{D}] = \sum_i a_i [H_i]$ in $H_{2\mathfrak{n}-2}(M, \mathbb{Z})$. Furthermore, it is known that if $M$ is compact, the homology class $[\mathcal{D}] \in H_{2\mathfrak{n}-2}(M, \mathbb{Z})$ is the Poincar\'{e} dual of the cohomology class $c_1(\linebundle(\mathcal{D})) \in H^2(M,\mathbb{Z})$.
\subsubsection{Finite Group Actions, Quotient Singularities, and Galois Coverings}\label{covering}
Let $Y = (|Y|, \mathscr{O}_Y)$ be a reduced normal\footnote{Let remind that an analytic space is \emph{normal} if every stalk of the structure sheaf is a normal ring (meaning an integrally closed integral domain). In a normal analytic space, the singular locus has codimension of at least two.} complex analytic space and let $\Gamma$ be a finite subgroup of the group $\operatorname{Aut}(Y)$ of analytic automorphisms of $Y$. Our main goal in this subsection is to study the analytic quotient of $Y$ by the group $\Gamma$. More concretely, we want to construct a reduced normal complex analytic space $X = (|X|, \mathscr{O}_X)$ together with a surjective analytic map $\varpi: Y \to X$ which is invariant under $\Gamma$ --- i.e. $\varpi \circ \gamma = \varpi$ for all $\gamma \in \Gamma$. 
\subsubsection*{Finite Group Actions and Analytic Quotients}
Since any analytic space $Y$ has an underlying (Hausdorff) topological space $|Y|$, we start our study of analytic quotients with a few definitions regarding the action of a topological group on a topological space:
\begin{definition}[Group action]
	A topological group $\Gamma$ induces a \emph{(left) group action} on a topological space $|Y|$ if there is a map $\Gamma \times |Y|  \to |Y|$ such that:
	\begin{enumerate}
		\item For any $y \in |Y|$, $(\mathbbm{1} , y) \mapsto y$ where $\mathbbm{1} \in \Gamma$ is the identity element,
		\item For any $y \in |Y|$ and any two group elements $\gamma_1,\gamma_2 \in \Gamma$, $(\gamma_1 \gamma_2 , y) = \big(\gamma_1 , (\gamma_2 , y)\big)$.
	\end{enumerate}
	We will usually write $\gamma \cdot y$ or even $\gamma(y)$ instead of $(\gamma,y)$. There is also a notion of \emph{right group action}, which we will avoid introducing.
\end{definition}
\begin{remark}
	Throughout this paper, all groups are assumed to act \emph{effectively}, which means for every two distinct elements in the group, there is some point in the space at which they differ. 
\end{remark}
Given a group action $\Gamma \times |Y|  \to |Y|$, we can associate to every element $\gamma \in \Gamma$ a homeomorphism $\imath_{\gamma}: |Y|  \to |Y|$ which is defined as $\imath_{\gamma}(y) = \gamma \cdot y$ for all $y \in |Y|$. The map $\gamma \mapsto \imath_{\gamma}$ induces a group homomorphism $\imath: \Gamma \to \operatorname{Aut}\big(|Y|\big)$ where $\operatorname{Aut}\big(|Y|\big)$ is the group of automorphisms of topological space $|Y|$. Conversely, it is easy to see that any group homomorphism $\imath: \Gamma \to \operatorname{Aut}\big(|Y|\big)$ yields a group action $\Gamma \times |Y| \to |Y|$, by setting $\gamma \cdot y = \imath_{\gamma}(y)$. Observe that a group action is \emph{effective} if and only if $\imath: \Gamma \to \operatorname{Aut}\big(|Y|\big)$ is a \emph{monomorphism}.

Let us now run through some terms that are associated with this group action: The \emph{isotropy subgroup}, also called the \emph{stabilizer subgroup}, of any point $y \in |Y|$ is defined as the set $\Gamma_{y} := \{\gamma \in \Gamma \, | \, \gamma \cdot y = y \}$ and is a closed subgroup of $\Gamma$. The action of $\Gamma$ on $|Y|$ is said to be \emph{free} if $\Gamma_{y} = \{\mathbbm{1}\}$, for all $y \in |Y|$. The set $\Gamma(y) := \big\{\gamma \cdot y \in |Y| \, \big| \, \gamma \in \Gamma \big\}$ denotes the \emph{orbit} of point $y$. The action of $\Gamma$ on $|Y|$ is called \emph{transitive} if $\Gamma(y) = |Y|$ for any point $y \in |Y|$ --- i.e. if for any two points $y_1, y_2 \in |Y|$ there is a element $\gamma \in \Gamma$ such that $\gamma \cdot y_1 = y_2$. Moreover, we call an action \emph{regular} if it is both transitive and free.
We always denote by $|Y|/\Gamma$ the set of all $\Gamma$-orbits in $|Y|$.\footnote{More concretely, we should denote the set of all \emph{left} $\Gamma$-orbits in $|Y|$ by ``$\Gamma\backslash|Y|$'' instead of ``$|Y|/\Gamma$'' (the later should be reserved for set of all \emph{right} $\Gamma$-orbits). However, we will continue to use the notation $|Y|/\Gamma$ with the understanding that it represents the quotient of $|Y|$ by the left action of $\Gamma$.} The orbit space $|Y|/\Gamma$ will be called the \emph{topological quotient} of $|Y|$ by (the left action of) $\Gamma$ and the natural map $|\varpi|: |Y| \to |Y|/\Gamma$, sending $y$ to its left orbit $\Gamma(y)$, is called the corresponding \emph{quotient map}. To make the quotient map $|\varpi|: |Y| \to |Y|/\Gamma := |X|$ continuous for an arbitrary topological $\Gamma$-action on $|Y|$, we have to endow $|X|$ with the \emph{quotient topology}: $W \subset |X|$ is open, if and only if $|\varpi|^{-1}(W)$ is open in $|Y|$.
\begin{remark}
	For an open set $V \subset |Y|$, the image $\gamma(V)$ is also open for all $\gamma \in \Gamma$. Hence,
	\begin{equation*}
	|\varpi|^{-1}\big(|\varpi|(V)\big) = \bigcup_{\gamma \in \Gamma} \gamma(V)
	\end{equation*}
	is an open set --- i.e. $|\varpi|(V)$ is open in $|X|$. In other words, the quotient map $|\varpi|: |Y| \to |X|$ is an \emph{open map}. In particular, if $|\varpi|$ is (locally) a bijective, it is (locally) a homeomorphism.
\end{remark}

Since analytic subvarieties (representing an analytic variety) inherit a locally compact Hausdorff structure with a countable basis from their ambient complex number spaces, we shall assume that all topological spaces in the present text have these properties, at least locally. However, while patching local models together, we also want to avoid the creation of new pathologies. Therefore, we always assume $|Y|$ to be globally Hausdorff and to have a countable basis; in particular, all topological spaces in this paper are paracompact. Then, we have to put strong conditions on the action of group $\Gamma$ in order to make sure that the topological quotient $|Y|/\Gamma$ preserves these properties (e.g. being Hausdorff). For this reason, we will always assume that $\Gamma$ acts \emph{properly discontinuously} on $|Y|$ by which we mean that for all compact sets $V \subset |Y|$, the set 
\begin{equation*}
\left\{\gamma \in \Gamma \, \Big| \, \gamma(V) \cap V \neq \emptyset \right\}
\end{equation*}
is \emph{finite}; this ensures that the topological quotient $|Y|/\Gamma$ is indeed a Hausdorff space. Note that \emph{finite groups} always have this property for trivial reasons.
\begin{remark}\label{rmk:properlydiscontinous}
	Since the finite group $\Gamma$ acts properly discontinuously on $|Y|$, there exist for all $y \in |Y|$ (arbitrarily small) neighborhoods $V_y$ of $y$ such that
	\begin{equation*}
	\left\{
	\begin{split}
	& \gamma \big(V_y\big)  = V_y  \quad \text{for all} \quad \gamma \in \Gamma_y,\\
	& \gamma \big(V_y\big) \cap V_y = \emptyset  \quad \text{for all} \quad \gamma \in \Gamma \backslash \Gamma_y.\\
	\end{split}
	\right.
	\end{equation*}
	It is then sufficient to construct the analytic quotient of $\big(V_y, \mathscr{O}_{Y} \big|_{_{V_y}}\big)$ by $\Gamma_y$ for all $y \in Y$ (which is obviously identical with the quotient of $\bigcup_{\gamma \in \Gamma} \gamma(V_y)$ by $\Gamma$).\footnote{In fact, these spaces may be glued together (in a uniquely determined manner) to a space which possesses the desired property, if and only if the underlying topological space is Hausdorff.} 
\end{remark}
When $Y$ carries more structure, we are often compelled to equip the quotient $Y/\Gamma$ with a comparable structure: Let $Y = (|Y|, \mathscr{O}_Y)$ be a reduced complex analytic space and let $\Gamma$ be a finite subgroup of $\operatorname{Aut}(Y)$, the group of complex analytic automorphisms of $Y$.\footnote{As a general rule, $\operatorname{Aut}(Y)$ refers to the automorphisms of an object $Y$ in a category which will sometimes not be mentioned explicitly, if in the given context there is no ambiguity.} The orbit space $X:= Y/\Gamma$ is then called the \emph{analytic quotient} of $Y$ by $\Gamma$ and is constructed as follows: Define topologically $|X| := |Y|/\Gamma$ and denote by $|\varpi|: |Y| \to |X|$ the corresponding (topological) quotient map; according to remark~\ref{rmk:properlydiscontinous}, it is sufficient to consider the case in which $|Y|$ is small with respect to $y \in |Y|$ and that $\Gamma=\Gamma_{y}$ is finite. Then, for $W$ open in $|X|$, the set $|\varpi|^{-1}(W)$ is open and $\Gamma$-invariant in $|Y|$ such that we can form the invariant algebra
\begin{equation}
\mathscr{O}_X(W) := \left(\mathscr{O}_Y\left(|\varpi|^{-1}(W)\right)\right)^{\Gamma}
\end{equation}
where $\Gamma$ acts in the obvious manner on the algebra $\mathscr{O}_Y\left(|\varpi|^{-1}\big(|U|\big)\right)$ of holomorphic functions. Hence, we have furnished the topological space $|X|$ with a ringed structure $(|X|,\mathscr{O}_X)$. In order to see that $(|X|,\mathscr{O}_X)$ is indeed a (reduced) complex analytic space, we need the following local identity:
\begin{equation}
\mathscr{O}_{X,|\varpi|(y)} = \mathscr{O}_{Y,y}^{\Gamma}.
\end{equation}
In fact, since $|\varpi|$ is finite and $|\varpi|^{-1}\big(|\varpi|(y)\big) = \{y\}$, we have
\begin{equation}
\begin{split}
\mathscr{O}_{X,|\varpi|(y)}  &= \lim_{\overset{\longrightarrow}{W \ni |\varpi|(y)}} H^{0}(W,\mathscr{O}_{Y}) = \lim_{\overset{\longrightarrow}{W \ni |\varpi|(y)}} H^{0}(|\varpi|^{-1}(W),\mathscr{O}_{X})^{\Gamma}\\
&= \big[\lim_{\overset{\longrightarrow}{W \ni |\varpi|(y)}} H^{0}(|\varpi|^{-1}(W),\mathscr{O}_{X})\big]^{\Gamma} = \mathscr{O}_{Y,y}^{\Gamma}.
\end{split}
\end{equation}
More concretely, we have the following theorem \cite[Theorem~8.1]{Riemenschneide_2011}:
\begin{theorem}[Analytic quotient]\label{thrm:analyticquotient}
	If $\Gamma$ acts properly discontinuously on the reduced complex analytic space $Y$, then there exists the analytic quotient $X = Y/\Gamma$. The analytic quotient map $\varpi: Y \to X$ is locally finite and surjective and (near any point $y$) isomorphic to the quotient map $Y \to Y/\Gamma_y$, where $\Gamma_y$ denotes the finite stabilizer subgroup of $\Gamma$ at $y$. In particular, the analytic algebra $\mathscr{O}_{X,x}$ can be identified with the invariant algebra $\mathscr{O}_{Y,y}^{\Gamma_y}$ for an arbitrary point $y \in |\varpi|^{-1}(x)$, and $\varpi_y^{\ast}: \mathscr{O}_{X,x} \to \mathscr{O}_{Y,y}$ is just the finite
	inclusion $\mathscr{O}_{Y,y}^{\Gamma_y} \hookrightarrow \mathscr{O}_{Y,y}$.
\end{theorem}
\begin{remark}
	If $\mathscr{O}_{Y,y}$ is reduced or an integral domain, $\mathscr{O}_{Y,y}^{\Gamma_y}$ is obviously reduced or an integral domain for arbitrary automorphism groups $\Gamma_y$. Since the inclusion of $\cmpx$-algebras $\mathscr{O}_{Y,y}^{\Gamma_y} \hookrightarrow \mathscr{O}_{Y,y}$ is a finite homomorphism for finite groups $\Gamma_y$, both algebras have the same dimension. Hence, under our standard assumptions, the analytic quotient $X=Y/\Gamma$ has in $x = \varpi(y)$ the same dimension as $Y$ in $y$. We finally note that also \emph{normality} will be inherited from $Y$.
\end{remark}
\subsubsection*{Quotient Singularities}
Anticipating the introduction of complex orbifolds as objects which locally look like the quotient of a complex manifold with a finite group action, we turn to studying the singularities of an analytic quotient of a complex manifold by a finite subgroup of its holomorphic automorphisms. Such singularities are called quotient singularities:
\begin{definition}[Quotient singularity]\label{def:quotientsingularity}
	By a \emph{quotient singularity}, we understand a singular point $x$ of an analytic quotient $X = Y/\Gamma$, where $Y$ is a smooth analytic space and $\Gamma$ is a finite group action on $Y$ by analytic automorphisms. 
\end{definition}
\noindent Since a smooth analytic space is, in fact, a complex analytic manifold, we will use $M$ (instead of $Y$) to denote such spaces. Following Cartan, one can then show that quotient singularities are locally analytically isomorphic to quotients of affine spaces by linear actions:\footnote{This result can be applied to show that quotient singularities are in fact \emph{algebraic} in all dimensions.}
\begin{theorem}[Cartan]
	Each quotient singularity is isomorphic to a quotient $\cmpx^{\mathfrak{n}}/\Gamma$, where $\Gamma$ is a finite subgroup of $\operatorname{GL}({\mathfrak{n}},\cmpx)$.
\end{theorem}

Let $M$ be a complex analytic manifold and let $\Gamma$ be a finite subgroup of its holomorphic automorphisms $\operatorname{Aut}(M)$. We will denote the quotient singularity $M/\Gamma$ by $X$ and the corresponding analytic quotient map by $\varpi: M \to X$. Since $M$ is a complex manifold, it can be embedded in an ambient complex number space $\cmpx^{\mathfrak{n}}$ and the structure sheaf $\mathscr{O}_{M}$ is given by the restriction $\mathscr{O}_{\cmpx^{\mathfrak{n}}}\big|_{_{M}}$; in particular, the stalks $\mathscr{O}_{M,p}$ for all points $p \in M$ are isomorphic with the $\cmpx$-algebra $\mathscr{O}_{\cmpx^{\mathfrak{n}},p} = \cmpx\{z_1, \dots, z_{\mathfrak{n}}\}$ of convergent power series at that point. Then, it follows from theorem~\ref{thrm:analyticquotient} and the above definition~\ref{def:quotientsingularity} that quotient singularities are completely determined by the normal invariant algebra
\begin{equation*}
\mathscr{O}_{X,x} = \mathscr{O}_{\cmpx^{\mathfrak{n}},p}^{\Gamma_p}
\end{equation*}
for any point $p \in |\varpi|^{-1}(x)$. As the ring $\mathscr{O}_{X,x}$ depends only on the conjugacy class of $\Gamma_p$, the quotients for conjugate groups are isomorphic. Hence, we consider only a representative for each conjugacy class.
\begin{definition}[Reflection groups and small groups]
	An element $\gamma \in \operatorname{Aut}(M)$, $M$ a connected complex manifold, is called a \emph{reflection} (or, perhaps more precisely, a \emph{pseudoreflection}) if it is of finite order and if the (analytic) fixpoint set
	\begin{equation*}
	\operatorname{Fix}(\gamma) := \big\{ p \in M \, \big| \, \gamma \cdot p = p \big\},
	\end{equation*}
	is of pure codimension-1 in $M$. A finite group $\Gamma \subset \operatorname{Aut}(M)$ is called a \emph{reflection group} if it is generated by pseudoreflections in $\operatorname{Aut}(M)$. Of course, an element $\gamma$ of finite order in $\operatorname{GL}(\mathfrak{n},\cmpx)$ is a reflection if and only if it leaves a hyperplane in $\cmpx^{\mathfrak{n}}$ pointwise fixed; this is equivalent to $\gamma$ having the eigenvalues 1 (of multiplicity $\mathfrak{n}-1$) and $e^{\frac{2 \pi \sqrt{-1}}{m}}$ --- an $m$-th root of unity with $m \geq 2$. On the other hand, a finite subgroup $\Gamma \subset \operatorname{GL}(\mathfrak{n},\cmpx)$ will be called \emph{small} if it contains no pseudoreflections.
\end{definition}
\noindent Then, we have the following well-known result due to Prill \cite{Prill1967ClassificationOfQuotients}:
\begin{theorem}[Classification of quotients of complex manifolds]\label{thrm:singularityclassification}
	Let $\Gamma \subset \operatorname{GL}(\mathfrak{n},\cmpx)$ be a finite subgroup. The following two statements are true:
	\begin{enumerate}[(i)]
		\item The analytic quotient $\cmpx^{\mathfrak{n}}/\Gamma$ is smooth if and only if $\Gamma$ is a reflection group;
		\item There exists a small group $\digamma$ such that $\cmpx^{\mathfrak{n}}/\Gamma$ and $\cmpx^{\mathfrak{n}}/\digamma$ are analytically isomorphic.
	\end{enumerate}
\end{theorem}
This, of course, is equivalent to the claim that the invariant algebra $\mathscr{O}_{\cmpx^{\mathfrak{n}},p}^{\Gamma_p}$ is isomorphic to the convergent power series ring $\cmpx\{z_1, \dots, z_{\mathfrak{n}}\}$ if and only if $\Gamma \subset \operatorname{GL}(\mathfrak{n},\cmpx)$ is a finite reflection group (see e.g. \cite[\textsection8.8]{Riemenschneide_2011} for more details). 
\begin{proposition}[Analytic spaces with at most quotient singularities]
	An analytic space $X = (|X|, \mathscr{O}_X)$ admitting only quotient singularities has the following properties:
	\begin{enumerate}[(i)]
		\item $X = (|X|, \mathscr{O}_X)$ is always a reduced normal analytic space;
		\item  The singular locus $\sing(X)$ is a closed reduced analytic subspace of $X$ and has complex codimension at least two in $X$;
		\item The smooth locus $\operatorname{Reg}(X)$ is a complex manifold and a dense open subset of $X$.		
	\end{enumerate}
\end{proposition}
\noindent 
Finally, since all finite subgroups of $\operatorname{GL}(1,\cmpx) \cong \cmpx^{\ast}$ are reflection groups, we have the following important corollary to the above theorem:
\begin{corollary}\label{corll:1dAnalytic}
	If $M$ is a 1-dimensional complex analytic manifold (i.e., a Riemann surface) and $\Gamma \subset \operatorname{Aut}(M)$ is a finite subgroup of its holomorphic automorphisms, the analytic quotient $M/\Gamma$ will always be a smooth, complex analytic space --- i.e., another Riemann surface.
\end{corollary}
\subsubsection*{Ramified Analytic Coverings}
As we saw in theorem~\ref{thrm:analyticquotient}, an analytic quotient map $\varpi: Y \to X := Y/\Gamma$ between two complex analytic spaces $Y$ and $X$ is locally finite and surjective. This motivates us to study such analytic mappings in more detail (see e.g., \cite[\textsection7.2]{Grauert-Remmert}):
\begin{definition}[Analytic covering map]
	A finite surjective analytic map $\varpi: Y \to X$ between irreducible analytic spaces is called an \emph{analytic covering}. This means that there exists a thin subset\footnote{A subset $T \subset X$ is called \emph{thin} if it has the property that each point has a neighborhood on which some non-zero holomorphic function vanishes. Since the set on which a holomorphic function vanishes is closed and has an empty interior, a thin set is \emph{nowhere dense}, and the closure of a thin set is also thin.} $T \subset X$,  called the \emph{critical locus} of the covering, such that 
	\begin{enumerate}[(i)]
		\item $\varpi^{-1}(T)$ is thin in $Y$, and
		\item the restriction $\varpi\big|_{_{Y \backslash \varpi^{-1}(T)}} : Y \backslash \varpi^{-1}(T) \to X \backslash T$ is locally an analytic isomorphism.
	\end{enumerate}
	The second condition means that for a sufficiently small open neighborhood $W_{x} \subset X \backslash T$ of any point $x \in X \backslash T$, the inverse image $\varpi^{-1}(W_{x})$ consists of a finite number of components, called \emph{sheets} of $\varpi$, such that the restriction of $\varpi$ to each component is a complex analytic isomorphism between that component and $W_{x}$. 
\end{definition}
\begin{remark}
	We will always assume that our analytic spaces are irreducible so that ``analytic covering'' and ``finite analytic surjection'' can be regarded as synonyms.
\end{remark}

If $\varpi: Y \to X$ is an analytic covering, the restriction of $\varpi$ to the complement of critical locus is necessarily a finite-sheeted covering map; the number of sheets of this covering map will be called the (total) \emph{degree} of the analytic covering $\varpi$ and will be denoted by $\deg(\varpi)$. Additionally, for any point $y \in Y$, there are arbitrarily small open neighborhoods $V_y \subset Y$ of $y$ such that the restriction $\varpi\big|_{_{V_y}}$ is also an analytic covering (see e.g. \cite[\textsection5]{gunning1970lectures}). Since the degrees of these local analytic coverings can only decrease as the neighborhoods $V_y$ shrink to the point $y$, it is evident that the degree is the same for all sufficiently small such neighborhoods; this common degree will be called the \emph{ramification index} (or the \emph{local degree} or \emph{multiplicity}) of the mapping $\varpi$ at the point $y$, and will be denoted by $\deg_{\varpi}(y)$. 
\begin{remark}
	Note that if $\varpi: Y \xrightarrow{d:1} X$ is an analytic covering of total degree $d$, then selecting any point $x \in X$ and letting $\{y_i\}_{i \in I}:=\varpi^{-1}(x) \subset Y$ be the collection of distinct points in the pre-image of $x$, it follows that $\sum_{i \in I} \deg_{\varpi}(y_i) = d$.
\end{remark}

For any point $y \notin \varpi^{-1}(T)$, it is clear that $\deg_{\varpi}(y)=1$. However, in general, there may very well be points $y' \in \varpi^{-1}(T)$ for which $\deg_{\varpi}(y')=1$. This is because not all the points of $\varpi^{-1}(\varpi(y'))$ need necessarily have the same ramification index, even when the critical locus is chosen to be as small as possible. Hence, one usually introduces the subset
\begin{equation*}
R_{\varpi} := \left\{y \in Y \, \Big| \, \deg_{\varpi}(y)>1 \right\}
\end{equation*}
of $\varpi^{-1}(T)$ which will be called the \emph{ramification locus} of the analytic covering $\varpi: Y \to X$ and is a closed analytic subspace of $Y$. The set $B_{\varpi} := \varpi(R_{\varpi})$ is called the \emph{branching locus} of $\varpi$ and is a closed analytic subspace of $X$; note that $B_{\varpi}$ is clearly a subset of the critical locus $T$. An analytic covering $\varpi: Y \to X$ is said to \emph{branched at most at $T$} if the branch locus $B_{\varpi}$ is contained in $T$. In addition, the analytic covering $\varpi$ will be called \emph{unbranched}, if $B_{\varpi}$ is empty. Observe that when $X$ is singular, $R_{\varpi}$ and $B_{\varpi}$ can be of codimension $>1$, even when $\varpi$ is a non-trivial branched covering. However, when $X$ is a smooth, complex analytic space, $R_{\varpi}$ will be a hypersurface in $Y$, and $B_{\varpi}$ will be a hypersurface in $X$.

Now, let $\varpi: Y \to X$ be an analytic covering of normal complex spaces. We define an automorphism of this analytic covering as a complex analytic automorphism $f: Y \to Y$ with the property that the diagram
\begin{equation}
\begin{tikzcd}
Y \ar[rr,"f"] \ar[dr,"\varpi" swap] &  &Y \ar[ld,"\varpi"] \\
&X&
\end{tikzcd}
\end{equation}
commutes --- i.e. $\varpi \circ f = \varpi$. The group 
\begin{equation}
\operatorname{Aut}(\varpi) := \left\{f \in \operatorname{Aut}(Y) \, \Big| \, \varpi \circ f = \varpi \right\}
\end{equation}
of all such automorphisms of $\varpi: Y \to X$ is called the \emph{group of covering transformations} or \emph{deck transformations}. An analytic covering $\varpi: Y \to X$ will be called a \emph{Galois covering} (or, in topologists language, \emph{regular covering}) if $\operatorname{Aut}(\varpi)$ acts \emph{transitively} on every fiber of $\varpi$; the group $\operatorname{Aut}(\varpi)$ itself will be called the \emph{Galois group} of $\varpi$ and will be denoted by $\gal(Y/X)$ or $\gal(\varpi)$. In this case, the analytic quotient $Y/\gal(\varpi)$ is complex analytically equivalent to $X$ and we get the following corollary of theorem~\ref{thrm:analyticquotient}:
\begin{corollary}
	Every analytic quotient map $\varpi: Y \to X := Y/\Gamma$ is locally a Galois covering map.
\end{corollary}
When the branched analytic covering $\varpi: Y \to X$ is Galois, the total degree of $\varpi$ is given by the order of its Galois group --- i.e. $\deg(\varpi) = \# \gal(\varpi)$. Consequently, using the fact that the restriction of  $\varpi$ to any (sufficiently small) neighborhood $V_y$ of a point $y$ is again a branched Galois covering, the ramification index of each point $y \in Y$ will be equal to the order of its stabilizer subgroup $\gal(\varpi)_{y} := \{\gamma \in \gal(\varpi) \, | \, \gamma \cdot y = y \}$. Additionally, since for any point $y \in Y$ the stabilizer subgroup of different points in its $\gal(\varpi)$-orbit $\varpi^{-1}(\varpi(y))$ vary only up to the conjugation by elements of $\gal(\varpi)$, one can always choose the critical locus $T$ such that 
\begin{equation*}
R_{\varpi} := \left\{ y \in Y \, \Big| \, \# \gal(\varpi)_{y} > 1 \right\} = \varpi^{-1}(T) \quad \text{and} \quad B_{\varpi} := \varpi(R_{\varpi}) = T
\end{equation*} 
for any branched Galois covering $\varpi: Y \to X$. In this case, we say that Galois covering $\varpi: Y \to X$ is \emph{branched along} $B_{\varpi}$. Moreover, since for an arbitrary lift $y$ of any point $x \in X$ the ramification index $\deg_{\varpi}(y) = \# \gal(\varpi)_y$ is independent of our choice of $y \in \varpi^{-1}(x)$, we can define the \emph{branching index} of $\varpi$ at any point $x \in X$ to be given by this common value. More generally, one defines a \emph{branching function} for $\varpi: Y \to X$ on $X$ as
\begin{equation*}
\nu_{\varpi}: X \ni x \to \# \gal(\varpi)_{y} \in \mathbb{N}
\end{equation*}
where $y$ is any lift of $x$ --- i.e. $y \in \varpi^{-1}(x)$.

We are finally ready to introduce notions of \emph{ramification divisor} and \emph{branch divisor} for a ramified Galois covering $\varpi: Y \to X$ between connected normal analytic spaces: Let us start by defining
\begin{equation*}
X' := \left\{x \in \operatorname{Reg}(X) \, \Big| \, \varpi^{-1}(x)  \subset \operatorname{Reg}(Y)\right\} \quad \text{and} \quad Y' := \varpi^{-1}(X')
\end{equation*}
such that $Y^{\prime}$ and $X^{\prime}$ are open analytic subsets of $Y$ and $X$ respectively and their complements have codimension at least 2. Then, the restriction $\varpi\big|_{Y^{\prime}}: Y^{\prime} \to X^{\prime}$ is a Galois covering map between complex analytic manifolds. Next, let us pick local coordinates $z_1, \dots,z_{\mathfrak{n}}$ on a neighborhood $V_y$ of a point $y \in Y^{\prime}$ and let $w_1, \dots, w_{\mathfrak{n}}$ be the coordinates around its image $\varpi(y) \in X^{\prime}$. Then, $w_i = \varpi_i(z_1,\dots,z_{\mathfrak{n}})$ gives the local expression of $\varpi$ near the point $y$ and the set
\begin{equation*}
\mathcal{R}' := \left\{y \in Y' \, \Bigg| \,  \det(\pdv{\varpi_i}{z_j} (y)) = 0 \right\}
\end{equation*}
can be viewed as the ramification locus of the restriction $\varpi\big|_{Y^{\prime}}$ --- i.e. set of all points $y \in Y^{\prime}$ around which $\varpi\big|_{V_y}$ is not a biholomorphism. Notice that since both $Y'$ and $X^{\prime}$ are by definition smooth, $\mathcal{R}'$ is necessarily a hypersurface in $Y'$. Since the complement of $Y^{\prime}$ has codimension at least 2, the Remmert--Stein extension theorem (see e.g., \cite[p.~181]{Grauert-Remmert}) ensures that the topological closure of this set will be a hypersurface $\mathcal{R}_{\varpi}$ in $Y$. Additionally, since the Galois covering $\varpi$ is finite, the set $\mathcal{B}_{\varpi} := \varpi(\mathcal{R}_{\varpi})$ will be a hypersurface in $X$. We will denote the irreducible components of the hypersurface $\mathcal{R}_{\varpi} \subset Y$ by $\mathcal{R}_i$; for each irreducible hypersurface $\mathcal{R}_i \subset Y$, the image $\mathcal{B}_i := \varpi(\mathcal{R}_i)$ is also an irreducible hypersurface in $X$ such that $\mathcal{B} = \bigcup_i \mathcal{B}_i$.

\begin{remark}
	Observe that if Galois covering $\varpi: Y \to X$ is ramified only along a \emph{singular part} of $Y$, the sets $\mathcal{R}_{\varpi}$ and $\mathcal{B}_{\varpi}$ as defined above will be empty. Since we are assuming that analytic spaces $Y$ and $X$ are normal (i.e. their singular locus has codimension $\geq 2$), one concludes that $\mathcal{R}_{\varpi}, \mathcal{B}_{\varpi} \neq \emptyset$ if and only if $\gal(\varpi)$ contains at least one pseudoreflection.
\end{remark}
Now, let us consider the sets 
\begin{equation*}
Y'' := Y' \big\backslash \Big(\sing(\mathcal{R}) \cup \varpi^{-1}\big(\sing(\mathcal{B})\big)\Big) \quad \text{and} \quad X'' := \varpi(Y'').
\end{equation*}
Both subsets are open and have complements of codimension at least 2 in $Y$ and $X$, respectively. Note that if $y \in Y''$ either $y \notin \mathcal{R}$ or $y$ belongs to one and only one irreducible component $\mathcal{R}_i$. In the first case, we say that $\varpi$ is \emph{unramified} at $y$; then $\varpi$ is a local biholomorphism at $y$. In the latter case (i.e., when $y$ has ramification index $>1$), let $\mathcal{R}_i$ be the unique irreducible component of $\mathcal{R}$ passing through $y$. Then, there are local coordinates $z_1,\dots, z_{\mathfrak{n}}$ on $V \subset Y''$ and $w_1,\dots,w_{\mathfrak{n}}$ on $W \subset X''$ centered at $y$ and $x = \varpi(y)$ respectively, such that locally $\mathcal{R}_i \cap V = \{z_1 = 0 \}$, $\mathcal{B}_i \cap W = \{w_1 = 0 \}$ and
\begin{equation}
\varpi\big|_{V}: (z_1,\dots,z_{\mathfrak{n}}) \mapsto (w_1= z_1^m, w_2 = z_2,\dots,w_{\mathfrak{n}} = z_{\mathfrak{n}}),
\end{equation}
where $m \in \mathbb{N}^{>1} := \mathbb{N}\backslash\{1\}$ denotes the ramification index of $\varpi$ at point $y$ --- i.e. $m = \deg_{\varpi}(y)$. For any irreducible component $\mathcal{R}_i$, the ramification index $\deg_{\varpi}(y) = \# \gal(\varpi)_y$ will be the same for all points $y \in \mathcal{R}_i \cap Y''$; this common value is denoted by $\deg_{\varpi}(\mathcal{R}_i )$ and will be called the \emph{ramification index of $\varpi$ along $\mathcal{R}_i$}. This enables us to define the \emph{ramification divisor} of a branched Galois covering $\varpi$ as the formal linear combination
\begin{equation}
\mathscr{R}_{\varpi} := \sum_{i} (m_i -1) \mathcal{R}_i,
\end{equation}
where $m_i = \deg_{\varpi}(\mathcal{R}_i )$ are the ramification indices (or multiplicities) of $\varpi$ along irreducible hypersurfaces (or prime divisors) $\mathcal{R}_i$.
\begin{remark}
	The ramification divisor $\mathscr{R}$ defined in this way is an \emph{effective Weil divisor} on $Y$. As we discussed earlier, the notions of Weil and Cartier divisors coincide only on smooth analytic spaces (e.g., on $Y'$). However, when $Y$ only contains quotient singularities, every Weil divisor is \emph{$\mathbb{Q}$-Cartier} by which we mean some multiple of it is a Cartier divisor; such normal reduced analytic spaces are said to be \emph{$\mathbb{Q}$-factorial} (see \cite{artal2014cartier} for more details).
\end{remark}
When the ramified analytic covering $\varpi: Y \to X$ is Galois, it is possible to define the branch divisor on $X$ as follows: As mentioned before, for each prime divisor $\mathcal{R}_i$ on $Y$, the image $\mathcal{B}_i = \varpi(\mathcal{R}_i)$ will be a prime divisor on $X$. If $\nu_{\varpi}: X \to \mathbb{N}$ denotes the branching function associated to the Galois covering $\varpi$, the restriction $\nu_{\varpi}\big|_{_{\mathcal{B}_i \cap X''}}: \mathcal{B}_i \cap X'' \to \mathbb{N}$ will be a constant function for each prime divisor $\mathcal{B}_i$; this constant value will be denoted by $\nu_{\varpi}(\mathcal{B}_i)$ and is called the \emph{branching index of $\varpi$ along $\mathcal{B}_i$}. Then, we can define the \emph{branch divisor} of Galois covering $\varpi$ as the effective $\mathbb{Q}$-divisor 
\begin{equation*}
\mathscr{B}_{\varpi} := \sum_{i} \left(1-\frac{1}{\nu_{\varpi}(\mathcal{B}_i)}\right) \mathcal{B}_i.
\end{equation*}
Here, by a $\mathbb{Q}$-divisor, we simply mean a formal finite sum of irreducible hypersurfaces with coefficients in $\mathbb{Q}$ (i.e., a Weil divisor with rational coefficients). Note that with this convention, $\mathscr{R}_{\varpi} = \varpi^{\ast}(\mathscr{B}_{\varpi})$, that is, the ramification divisor is the pull-back of the branch divisor.

Let us finish our discussion of analytic coverings by focusing on the very important example of branched Galois coverings $\varpi: Y \to X$ of 1-dimensional complex analytic spaces --- i.e., when both analytic spaces $X$ and $Y$ are Riemann surfaces. In that case, any non-constant map between compact Riemann surfaces is a finite branched covering. Let $f: Y \to X$ be such a non-constant, holomorphic map between compact Riemann surfaces $Y$ and $X$. For every $y \in Y$ there exist charts for $f(y)$ such that the local expression of the branched covering $f$ is of the form $z \mapsto z^m$ where $m = \deg_{f}(y)$ is the ramification index of $f$ at point $y$, $z$ is the local coordinate on the covering Riemann surface $Y$, and $w=z^m$ is the local coordinate on the base Riemann surface $X$ (see Fig.~\ref{fig:branchedcovering}). 
\begin{figure}
	\centering
	\includegraphics[width=0.8\linewidth]{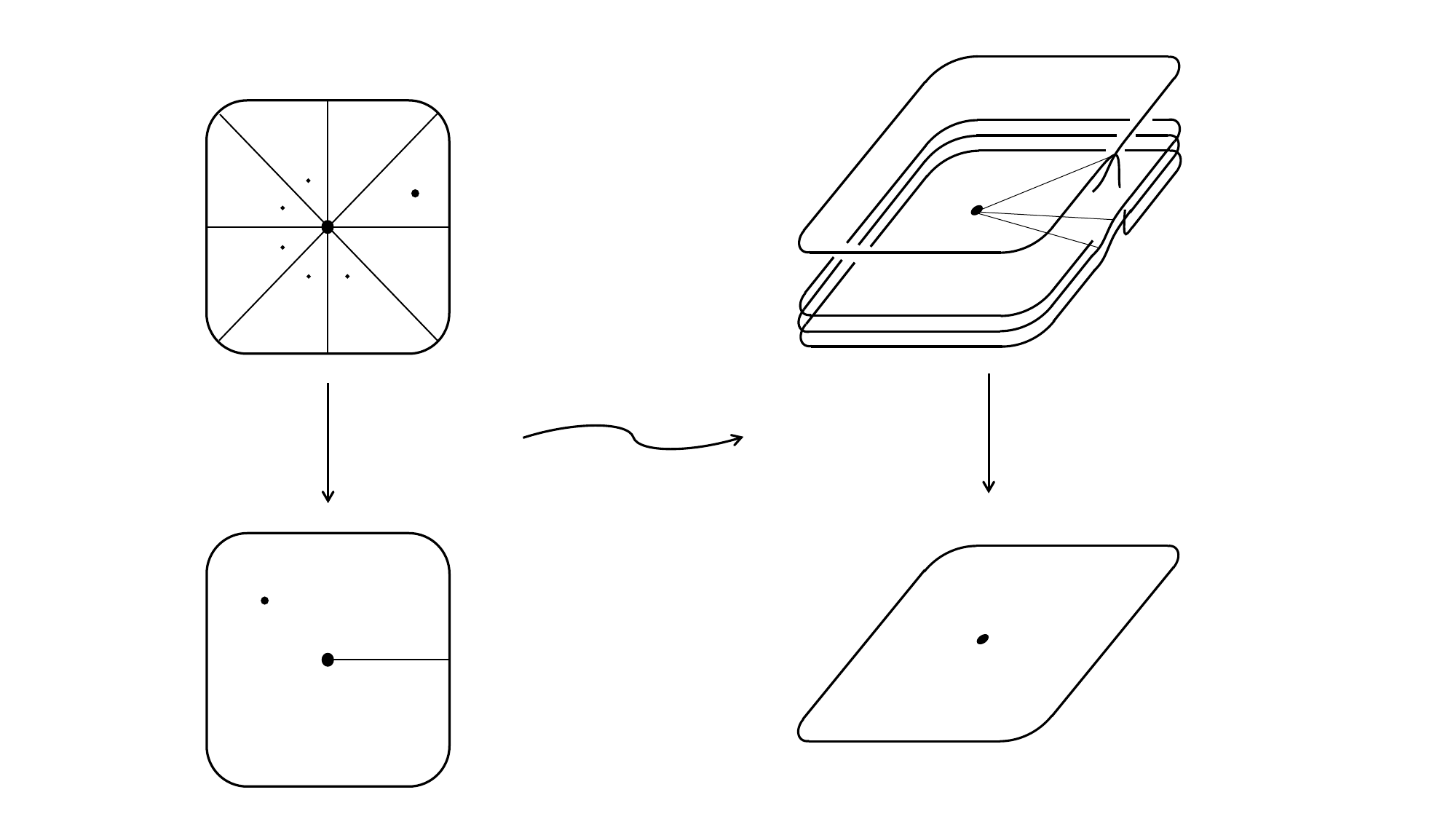}
	\put(-238,162){\rotatebox{0}{\fontsize{10}{10}\selectfont{} $\cmpx$}}
	\put(-238,59){\rotatebox{0}{\fontsize{10}{10}\selectfont{} $\cmpx$}}
	\put(-65,59){\rotatebox{0}{\fontsize{10}{10}\selectfont{} $\cmpx$}}
	\put(-65,170){\rotatebox{0}{\fontsize{10}{10}\selectfont{} $\cmpx$}}
	\put(-265,103){\rotatebox{0}{\fontsize{7}{7}\selectfont{} \rotatebox{270}{$z \mapsto z^{m}$}}}
	\put(-285,53){\rotatebox{0}{\fontsize{9}{9}\selectfont{} $w=z^{m}$}}
	\put(-267.5,148.5){\rotatebox{0}{\fontsize{7}{7}\selectfont{} $2$}}
	\put(-260,142.5){\rotatebox{0}{\fontsize{7}{7}\selectfont{} $1$}}
	\put(-262,133.5){\rotatebox{0}{\fontsize{7}{7}\selectfont{} $m$}}
	\put(-248,147){\rotatebox{0}{\fontsize{9}{9}\selectfont{} $z$}}
	\caption{\emph{A model branched covering}.} \label{fig:branchedcovering}
\end{figure}
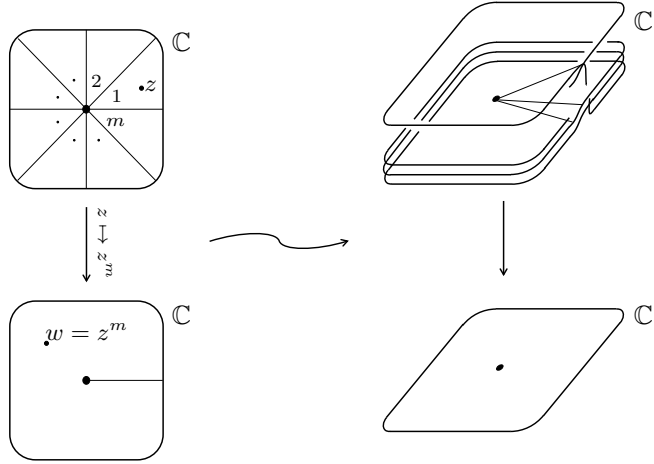	
\subsubsection{Complex Analytic Orbifolds}\label{sect:orbifolddef}
Let $X_O$ be an analytic space of dimension $\mathfrak{n}$ admitting only quotient singularities. An \emph{orbifold chart} or \emph{local uniformizing system} on an open subset $U \subset X_O$ is a tuple $(U,\tilde{U},\Gamma,\ff)$ where connected and open set $\tilde{U} \subset \cmpx^{\mathfrak{n}}$ is biholomorphic to the open unit ball $B^{\mathfrak{n}}$, $\Gamma \subset \operatorname{GL}(\mathfrak{n},\cmpx)$ is a finite group acting effectively\footnote{The condition that local uniformizing groups act effectively is not always imposed in the literature, and there are occasions when this requirement is too restrictive. However, since we are exclusively concerned with \emph{effective orbifolds or \emph{reduced orbifolds}} in this paper, it is convenient to incorporate this condition as part of our definition.} on $\tilde{U}$ as holomorphic automorphisms, and the ramified covering $\ff: \tilde{U} \to U$, called a \emph{folding map}, is a $\Gamma$-invariant map which induces a biholomorphism $\tilde{U}/\Gamma \xrightarrow{\cong} U$. The pair $(\tilde{U},\Gamma)$ is called a \emph{local model} for $U$.
\begin{definition}[Analytic orbifold atlas]\label{def:orbifold}
	An \emph{analytic orbifold atlas} $\mathcal{U}$ on an analytic space $X$ (admitting only quotient singularities) is a collection  $\{(U_a,\tilde{U}_a,\Gamma_a,\mathrm{f}_a)\}_{a\in A}$ charts on this analytic space such that the following conditions are satisfied:
	\begin{enumerate}[(i)]
		\item $\{U_a\}_{a \in A}$ is an open cover of the underlying complex space $X_O$;
		\item If $(U_a,\tilde{U}_a,\Gamma_a,\ff_a)$ and $(U_b,\tilde{U}_b,\Gamma_b,\ff_b)$ for $a,b \in A$ are two orbifold charts with $U_a \cap U_b \neq \emptyset$, then for each $x \in U_a \cap U_b$ there exists an orbifold chart $(U_c,\tilde{U}_c,\Gamma_c,\ff_c) \in \mathcal{U}$ that contains $x$ -- i.e. $x \in U_c \subseteq U_a \cap U_b$;
		\item If $(U_a,\tilde{U}_a,\Gamma_a,\ff_a)$ and $(U_b,\tilde{U}_b,\Gamma_b,\ff_b)$ for $a,b \in A$ are two orbifold charts with $U_a \subset U_b$, then there exists a holomorphic embedding,  
		\begin{equation*}
		\eta_{ab}: \tilde{U}_a \hookrightarrow \tilde{U}_b
		\end{equation*} 
		called the \emph{change of charts} (or \emph{embedding} or \emph{gluing map}), such that the folding maps satisfy $\ff_a = \ff_b \circ \eta_{ab}$. Moreover, for any three orbifold charts labeled by $a,b,c \in A$ and having the property $U_a \subset U_b \subset U_c$, the corresponding embeddings should satisfy $\eta_{ac} = \eta_{ab} \circ \eta_{bc}$. 
	\end{enumerate}
\end{definition}
\begin{remark}\label{rmk:embeddings}
	The choice of embedding $\eta_{ab}: \tilde{U}_a \hookrightarrow \tilde{U}_b$ is unique only up to the action of $\Gamma_b$: Let $(U_a,\tilde{U}_a,\Gamma_a,\ff_a)$ and $(U_b,\tilde{U}_b,\Gamma_b,\ff_b)$ be two orbifold charts on $X$ with $U_a \subset U_b$. If $\eta_{ab}, \eta'_{ab}: \tilde{U}_a \hookrightarrow \tilde{U}_b$ are two embeddings, then there exists a \emph{unique} $\gamma \in \Gamma_b$ such that $\eta'_{ab} = \gamma \circ \eta_{ab}$. As a result, an embedding $\eta_{ab}: \tilde{U}_a \hookrightarrow \tilde{U}_b$ induces a \emph{monomorphism}\footnote{We use the term ``monomorphism'' to mean an injective group homomorphism.} $\Upsilon_{ab}: \Gamma_a \to \Gamma_b$ which is given by
	\begin{equation}
	\eta_{ab} \circ \gamma = \Upsilon_{ab}(\gamma) \circ \eta_{ab},
	\end{equation}
	that is $\eta_{ab}(\gamma \tilde{x}) = \Upsilon_{ab}(\gamma) \eta_{ab}(\tilde{x})$ for all $\gamma \in \Gamma_a$ and $\tilde{x} \in \tilde{U_a}$. 
\end{remark}

An analytic orbifold atlas $\mathcal{U}$ is said to be a refinement of another analytic orbifold atlas $\mathcal{V}$ if there exists an embedding of every chart in $\mathcal{U}$ into some chart of $\mathcal{V}$. This enables us to define an \emph{equivalence relation} between orbifold atlases where two orbifold atlases are said to be equivalent if they have a common refinement. Then,
\begin{definition}[Complex analytic orbifolds]
	A \emph{complex analytic orbifold} $O$ of dimension $\mathfrak{n}$ is a pair $(X_O,[\mathcal{U}])$ where $X_O$ is the \emph{underlying (complex) analytic space} with at most finite quotient singularities and $[\mathcal{U}]$ is an equivalence class of analytic orbifold atlases on $X_O$. A 1-dimensional complex analytic orbifold will be called an \emph{orbifold Riemann surface} or a \emph{Riemann orbisurface}. 
\end{definition}
\begin{remark}
	As in the manifold case, an orbifold atlas is always contained in a \emph{unique maximal} one and two orbifold atlases are equivalent if, and only if, they are contained in the same maximal atlas. Therefore, we can equivalently define an \emph{analytic orbifold structure} on a complex analytic space $X_O$ (with at most finite quotient singularities) as the \emph{datum of a maximal (analytic) orbifold atlas} on this space.
\end{remark}
\begin{remark}\label{rmk:topologistsdef}
	Let $O = (X_O , [\mathcal{U}])$ be a complex analytic orbifold where $X_O$ is the underlying complex analytic space. Remember that $X_O$ can be characterized as a $\cmpx$-ringed space $(|X_O|,\mathscr{O}_{X})$ where $|X_O|$ is a Hausdorff paracompact topological space; we will denote the topological space $|X_O|$ simply by $|O|$. Then, one can define $O$ alternatively as a pair $(|O|, [\mathscr{U}])$ where the charts of $\mathscr{U}$ are now given by a quadruple $(|U|,\tilde{U},\Gamma,|\ff|)$. Here,  $|U| \subset |O|$ denotes the underlying topological space of each open analytic subset $U \subset X_O$ and $|\ff|$ denotes the continuous map associated with the $\Gamma$-invariant analytic map $\ff =(|\ff|,\ff^{\ast})$ that induces a homeomorphism between $|U|$ and $\tilde{U}/\Gamma$ as topological spaces; as evident from our notation, the local models $(\tilde{U},\Gamma)$ are defined as before. This alternative definition of complex orbifolds is more common among topologists and closely resembles the definition of a \emph{smooth orbifold} (where the only differences are in local models and embeddings). 
\end{remark}
\subsubsection*{Local Groups and Canonical Stratification}
Consider a point $x \in X_O$ and a local chart $(U_a,\tilde{U}_a,\Gamma_a,\ff_a) \in \mathcal{U}$ containing $x$. In addition, let us choose $\tilde{x} \in \tilde{U}_a$ to be a particular pre-image of $x$ and denote by $\Gamma_{\tilde{x}}$ the subgroup of $\Gamma_a$ that fixes $\tilde{x}$. As our choice of $\tilde{x} \in \ff_a^{-1}(x)$ varies, the stabilizer subgroup of $\tilde{x}$ varies only up to conjugation by the elements of $\Gamma_a$. Similarly, as our choice of $U_i$ containing $x$ varies, the stabilizer varies only up to conjugation by a transition map. Therefore, we can define the \emph{isotropy group} or the \emph{local group} $\Gamma_x$ to be the conjugacy class of the stabilizer subgroup $\Gamma_{\tilde{x}} \subset \Gamma_a$ for some $\tilde{x} \in \ff_a^{-1}(x)$. It is then clear that local group $\Gamma_x$ for a point $x \in X_O$ is \emph{independent} of both the \emph{chart}, $U_a$, and the \emph{lift} $\tilde{x} \in \tilde{U}_a$.
\begin{remark}\label{rmk:pseudoreflections}
	The main observation we would like to make about this definition is that an orbifold chart $(U, \tilde{U}, \Gamma, \ff)$ contains more data than simply the analytic quotient $\tilde{U}/\Gamma$. In particular, an orbifold chart ``remembers'' the pseudoreflections contained in a local uniformizing group $\Gamma_i$ and their corresponding fixed-point sets. 
\end{remark}
This allows us to define the \emph{singular points} of $O$ as points whose local isotropy group $\Gamma_{x} \neq \{\mathbbm{1}\}$; those points with $\Gamma_{x} = \{\mathbbm{1}\}$ are called \emph{regular points} of $O$. The set $\big\{ x \in X_O \, \big| \, \Gamma_{x} \neq \{\mathbbm{1}\}\big\}$ of singular points of $O$ is called the \emph{singular locus} or the \emph{singular set} of $O$ and will be denoted by $\sing(O)$. It follows from theorem~\ref{thrm:singularityclassification} (or remark~\ref{rmk:pseudoreflections}) that if any local uniformizing group contains a pseudoreflection, the orbifold singular set $\sing(O)$ will be bigger than the singular set of the underlying analytic space $\sing(X_O)$ and that $\sing(O)=\sing(X_O)$ if and only if none of the local uniformizing groups contain a pseudoreflection. 
The subset of all orbifold regular points, denoted by $X_O^{\text{reg}}$, is an open dense subset of $X_O$; in particular, since $\sing(X_O) \subseteq \sing(O)$, the orbifold regular locus $X_O^{\text{reg}} := X_O \backslash \sing (O)$ will always be a complex sub-manifold of $\operatorname{\text{Reg}}(X_O) = X_O \backslash \sing (X_O)$. The local isotropy groups give a \emph{canonical stratification} of $X_O$ by stating that two points lie in the same \emph{stratum} $\sing_j(O)$ if their local groups are conjugate. Thus, we get a decomposition of $X_O$ as
\begin{equation}
X_O = X_O^{\text{reg}}\bigsqcup_{j} \sing_j(O), \qquad  \sing(O) = \bigsqcup_{j} \sing_j(O),
\end{equation}  
where the union is taken over all conjugacy classes. The dense open subset of regular points, $X_O^{\text{reg}}$, is sometimes called the \emph{principal stratum} and corresponds to the trivial conjugacy class. 
\subsubsection*{Analytic Orbifold Maps and Orbifold Covering}
The standard notion of structure preserving maps between complex analytic orbifolds can be given in an analogous manner to complex manifolds (see \cite{baily1956decomposition,baily1957imbedding}). However, it wasn't realized until recently (see, e.g., \cite{Chen:2000cy,moerdijk1997orbifolds}) that certain problems arise with the usual definition of an orbifold map; namely, as we will see later, this definition does not, in general, induce morphisms of sheaves or V-bundles. This led to the introduction of the notion of ``good maps'' in \cite{Chen:2000cy}:
\begin{definition}[Analytic orbifold Maps]\label{def:orbmap}
	Let $O=(X_O,\mathcal{U})$ and $O'=(X'_O,\mathcal{U}')$ be two complex analytic orbifolds (not necessarily of the same dimension). A map $f: O \to O'$ is said to be an \emph{analytic orbifold map} (or a \emph{holomorphic orbifold map}) if $f$ gives an analytic mapping between the underlying complex analytic spaces, denoted by $(|f|, f^{\ast}): (|X_O|, \mathscr{O}_{X}) \to (|X'_O|, \mathscr{O}_{X'})$, which admits a \emph{local lift} at each point $x \in X_O$: For every pair of orbifold charts $(U_a,\tilde{U}_a,\Gamma_a,\ff_a) \in \mathcal{U}$ and $(U'_a,\tilde{U}'_a,\Gamma'_a,\ff'_a) \in \mathcal{U}'$ containing an arbitrary point $x \in X$ and its image $f(x) \in X'$ respectively, with $f(U_a) \subset U'_a$, the orbifold map $f$ induces a group homomorphism between local isotropy groups $\bar{f}_{x}: \Gamma_{x} \to \Gamma'_{f(x)}$ and a holomorphic $\bar{f}_{x}$-equivariant map $\tilde{f}_x: \tilde{U}_a \to \tilde{U}'_a$ such that the diagram
	\begin{equation}\label{diagramorbifoldmap}
	\begin{CD}
	\tilde{U}_a	@> \tilde{f}_{x}>> \tilde{U}'_a\\
	@V \ff_a VV		@VV \ff '_a V\\
	U_a 		@>> (|f|, f^{\ast}) > U'_a
	\end{CD}
	\end{equation}
	commutes (see Figure \ref{fig:orbifoldmap}). Moreover, a holomorphic orbifold map $f_{\text{orb}}$ is said to be \emph{good} if it is compatible with the embeddings: If $\eta_{ab}: \tilde{U}_a \hookrightarrow \tilde{U}_b$ is an embedding on $O$, then there is an embedding $\hat{f}(\eta_{ab}) : \tilde{U}'_a \hookrightarrow \tilde{U}'_b$ on $O'$, such that
	\begin{enumerate}[(i)]
		\item $\tilde{f}_{\tilde{U}_b} \circ \eta_{ab} = \hat{f}(\eta_{ab}) \circ \tilde{f}_{\tilde{U}_a}$, and
		\item $\hat{f}(\eta_{bc} \circ \eta_{ab}) = \hat{f}(\eta_{bc}) \circ \hat{f}(\eta_{ab})$.
	\end{enumerate}
	Note that conditions (i) and (ii) imply that the composition of good orbifold maps is again a good orbifold map. Finally, a holomorphic orbifold map $f_{\text{orb}} : O \to O'$ is called an \emph{analytic orbifold automorphism} (or an \emph{orbifold biholomorphism}) if it admits an analytic inverse. In this case, we clearly have $\Gamma_x \cong \Gamma'_{f(x)}$ for all $x \in X \cong X'$ --- i.e., biholomorphisms must preserve the orbifold stratification.
\end{definition}
\begin{figure}
	\centering
	\includegraphics[width=0.9\textwidth]{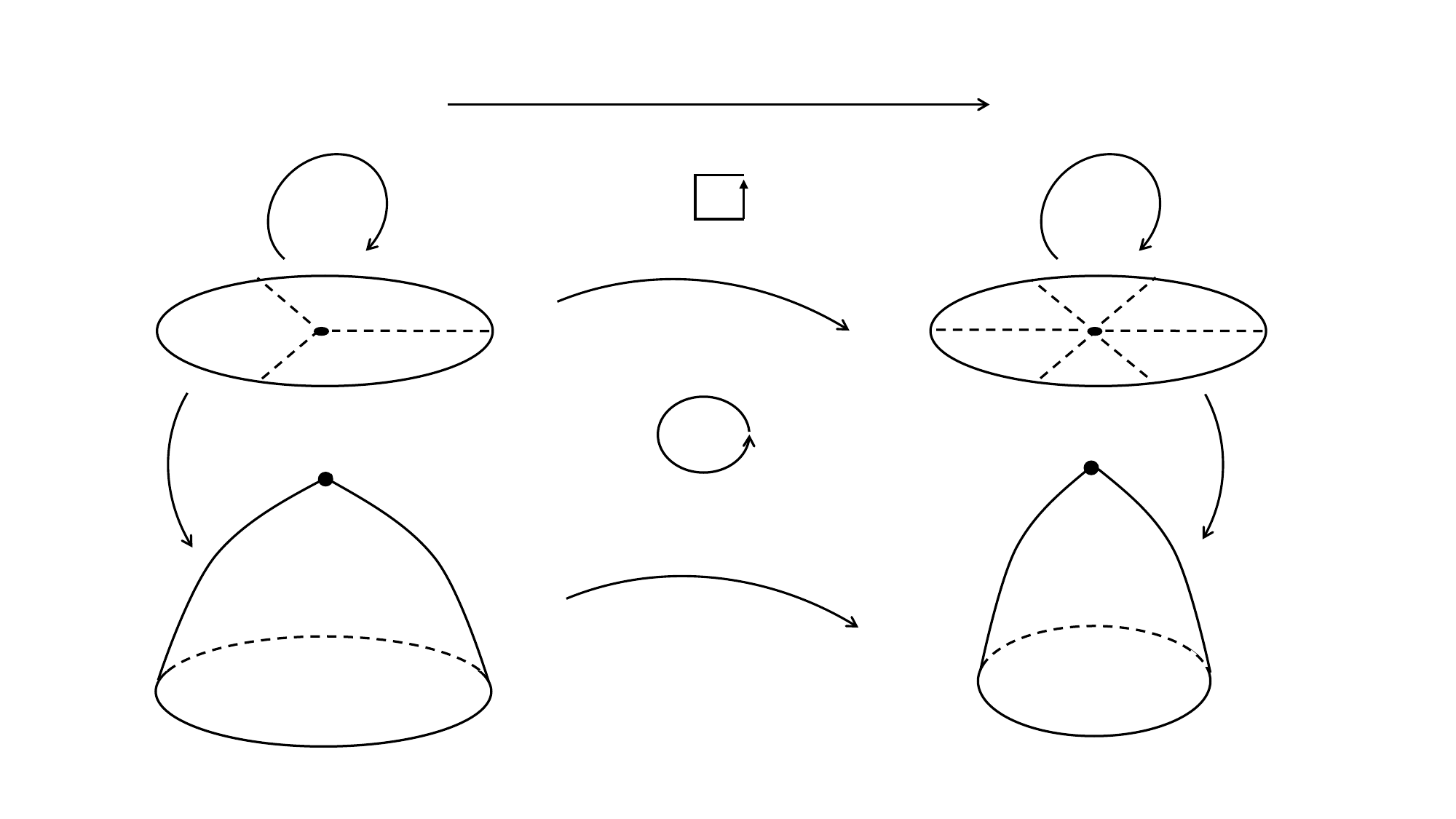}
	\put(-227,70){\rotatebox{0}{\fontsize{11}{11}\selectfont{} $f=(|f|, f^{\ast})$}}
	\put(-217,130){\rotatebox{0}{\fontsize{11}{11}\selectfont{} $\tilde{f}_{\tilde{U}}$}}
	\put(-217,195){\rotatebox{0}{\fontsize{11}{11}\selectfont{} $\bar{f}_x$}}
	\put(-307,187){\rotatebox{0}{\fontsize{14}{14}\selectfont{} $\Gamma_x$}}
	\put(-110,187){\rotatebox{0}{\fontsize{14}{14}\selectfont{} $\Gamma_{f(x)}$}}
	\put(-70,142){\rotatebox{0}{\fontsize{14}{14}\selectfont{} $\tilde{U}'$}}
	\put(-65,89){\rotatebox{0}{\fontsize{14}{14}\selectfont{} $\ff'$}}
	\put(-70,53){\rotatebox{0}{\fontsize{14}{14}\selectfont{} $U'$}}
	\put(-347,141){\rotatebox{0}{\fontsize{14}{14}\selectfont{} $\tilde{U}$}}
	\put(-358,89){\rotatebox{0}{\fontsize{14}{14}\selectfont{} $\ff$}}
	\put(-355,53){\rotatebox{0}{\fontsize{14}{14}\selectfont{} $U$}}
	\put(-308,93){\rotatebox{0}{\fontsize{12}{12}\selectfont{} $x$}}
	\put(-108,98){\rotatebox{0}{\fontsize{11}{11}\selectfont{} $f(x)$}}
	\caption{\emph{A local lift corresponding to an orbifold map.} Every holomorphic orbifold map between two complex orbifolds $f_{\text{orb}}: O \to O'$ consists of an analytic map $f=(|f|,f^{\ast})$ between the underlying analytic spaces together with a holomorphic local lift which is given by a group homomorphism $\bar{f}_x: \Gamma_{x} \to \Gamma_{f(x)}$ between local isotropy groups at each point $x \in O$ and its corresponding image $f(x) \in O'$ and a holomorphic $\bar{f}_{x}$-equivariant map $\tilde{f}_{\tilde{U}}: \tilde{U} \to \tilde{U}'$ between local uniformizing neighborhoods.} \label{fig:orbifoldmap}
\end{figure}
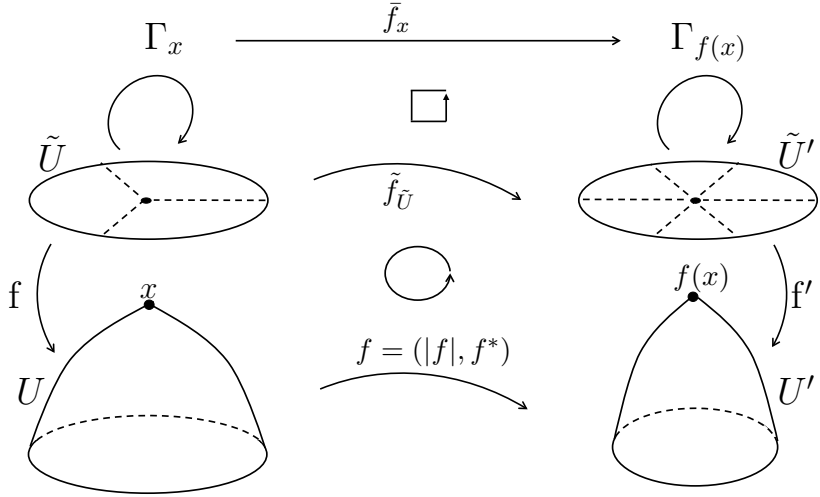	

\begin{remark}\label{rmk:orbifunctions}
	Considering $\cmpx$ or $\hat{\cmpx}$ as orbifolds with an empty singular set, one can define \emph{holomorphic/meromorphic orbifold functions} on a complex orbifold $O$ as holomorphic orbifold maps $f_{\text{orb}}: O \to \cmpx \text{ or } \hat{\cmpx}$. (see remark A.24!)
\end{remark}
Now, we can easily define an orbifold Galois covering using the above definitions: 
\begin{definition}[Orbifold Galois covering]
	An \emph{orbifold Galois covering} $\varpi_{\text{orb}}: \tilde{O}=(\tilde{X},\mathcal{V}) \to O = (X,\mathcal{U})$ is an analytic orbifold map such that $\varpi : \tilde{X} \to X$ is a Galois analytic covering and $\gal(\varpi)  \subset \operatorname{Aut}(O)$.\footnote{Remember that for any orbifold $O = (X_O, \mathcal{U})$, we have $\operatorname{Aut}(O) \subset \operatorname{Aut}(X_O)$.}
\end{definition}
\noindent In other words, the holomorphic orbifold map $\varpi_{\text{orb}}: \tilde{O} \to O$ is a projection map such that each point $x \in X_O$ has a  neighborhood $U \cong \tilde{U}/\Gamma$ for which connected components $V_i$ of $\varpi^{-1}(U)$ are analytically isomorphic to $\tilde{U}/\Gamma_i$ where $\Gamma_i \subset \Gamma$. Therefore, the restriction of projection map $\varpi$ to each sheet $V_i$, i.e. $\varpi\big|_{V_i} : V_i \to U$, corresponds to the natural projection $\tilde{U}/\Gamma_i \to \tilde{U}/\Gamma$ (see Figure \ref{fig:orbifoldcovering}). By the \emph{Galois group} and \emph{degree} of an orbifold Galois covering, we mean its Galois group and degree as an analytic cover. Finally, in cases that we need to keep track of base points on covering and base spaces, we will use the notation $(\tilde{O},\tilde{x}_{0}) \xrightarrow{\varpi_{\text{orb}}} (O,x_0)$ to refer to an orbifold covering for which $\varpi(\tilde{x}_{0}) = x_0$.
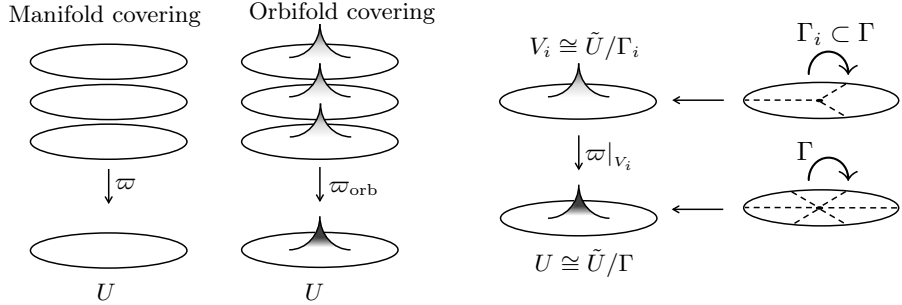
\begin{figure}
	\centering
	\includegraphics[width=0.8\linewidth]{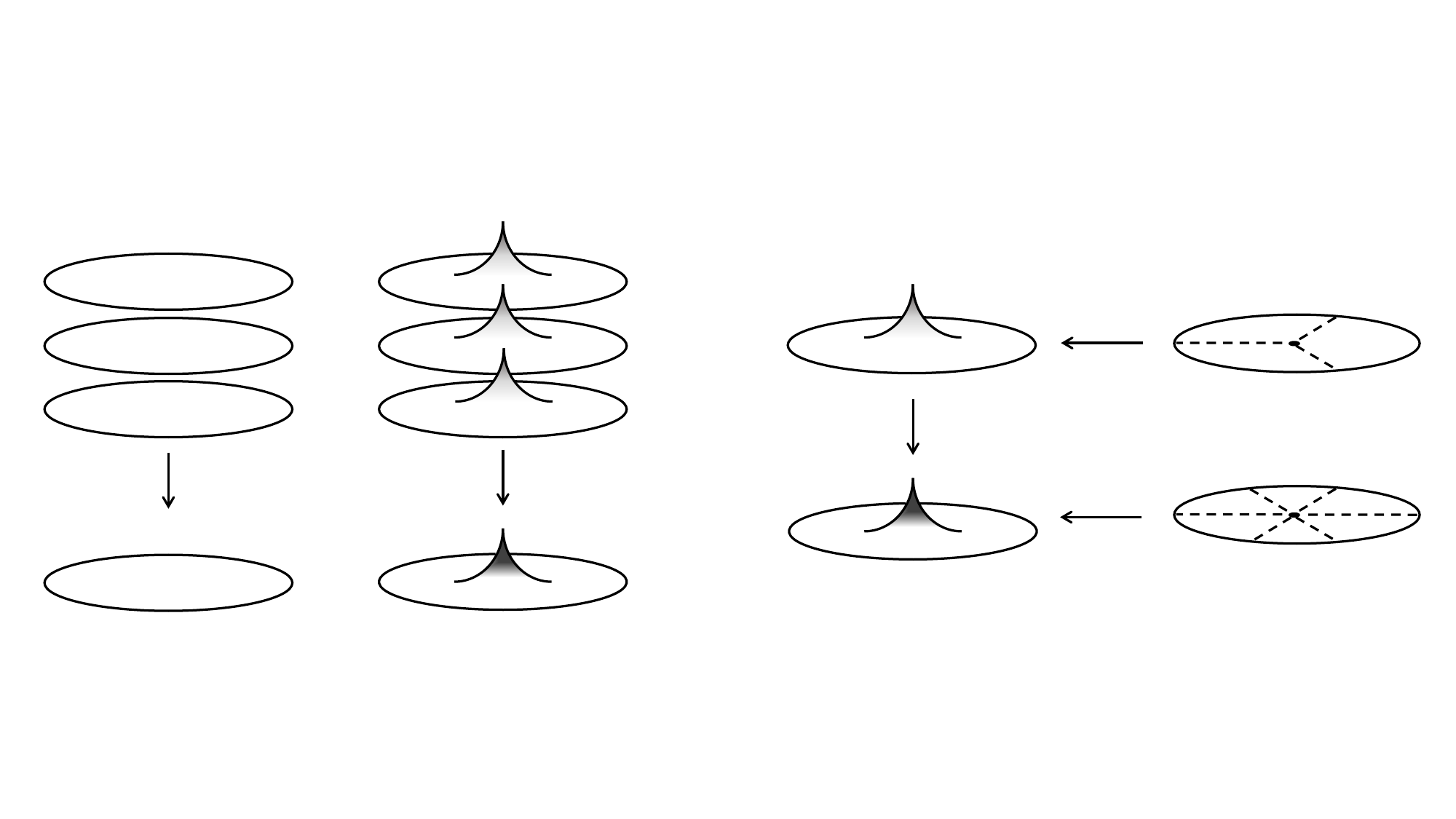}
	\put(-345,142){\rotatebox{0}{\fontsize{9}{9}\selectfont{} $\text{Manifold covering}$}}
	\put(-255,143){\rotatebox{0}{\fontsize{9}{9}\selectfont{} $\text{Orbifold covering}$}}
	\put(-305,79){\rotatebox{0}{\fontsize{9}{9}\selectfont{} $\varpi$}}
	\put(-225,78){\rotatebox{0}{\fontsize{9}{9}\selectfont{} $\varpi_{\text{orb}}$}}
	\put(-234,37){\rotatebox{0}{\fontsize{9}{9}\selectfont{} $U$}}
	\put(-312,37){\rotatebox{0}{\fontsize{9}{9}\selectfont{} $U$}}
	\put(-129,90){\rotatebox{0}{\fontsize{9}{9}\selectfont{} $\varpi|_{_{V_i}}$}}
	\put(-150,130){\rotatebox{0}{\fontsize{9}{9}\selectfont{} $V_i \cong \tilde{U}/\Gamma_i$}}
	\put(-50,122){\rotatebox{0}{\fontsize{20}{20}\selectfont{} $\curvearrowright$}}
	\put(-50,135){\rotatebox{0}{\fontsize{10}{10}\selectfont{} $\Gamma_i \subset \Gamma$}}
	\put(-148,48){\rotatebox{0}{\fontsize{9}{9}\selectfont{} $U \cong \tilde{U}/\Gamma$}}
	\put(-50,82){\rotatebox{0}{\fontsize{20}{20}\selectfont{} $\curvearrowright$}}
	\put(-50,90){\rotatebox{0}{\fontsize{10}{10}\selectfont{} $\Gamma$}}
	\vspace{-1cm}
	
	\caption{\emph{Orbifold covering}. The figure on the left compares the covering maps for manifolds and orbifolds. The figure on the right shows how the local groups of an orbifold and its covering orbifold are related.} \label{fig:orbifoldcovering}
\end{figure}
\subsubsection*{Global Quotient Orbifolds}
As pointed out in remark~\ref{rmk:embeddings}, the choice of embedding $\eta_{ab}: \tilde{U}_a \hookrightarrow \tilde{U}_b$ is in general \emph{not unique}; so, the charts do not have to satisfy a cocycle condition upstairs, though of course they do downstairs where the open sets $U_a$ glue to give the analytic space $X_O$ (see the discussion around Lemma \ref{lemma:gluing}). That is, the orbifold charts need not glue since an orbifold need not be a global quotient by a finite group. 

However, the most natural examples of complex orbifolds appear precisely when we take the quotient space $M/\Gamma$ of a complex manifold $M$ by a finite subgroup $\Gamma \subset \operatorname{Aut}(M)$ of its analytic automorphisms: 
\begin{proposition}\label{prop:quotientorb}
	If $M$ is a complex manifold and $\Gamma$ is a group acting holomorphically, effectively, and properly discontinuously on $M$, the quotient $M/\Gamma$ has the structure of a complex orbifold. Such orbifolds are called ``(effective) global quotient orbifolds'' and we will denote them by $[M/\Gamma]$ in order to emphasize their orbifold structure.\footnote{More concretely, one must use the geometric invariant theory (GIT) quotient to define the global quotient orbifolds. This is due to the fact that, as we have already seen, the analytic quotient defined in theorem~\ref{thrm:analyticquotient} is unable to produce the codimension-1 orbifold singularities created by the action of pseudoreflections contained in $\Gamma$. Therefore, with this notation, the analytic quotient $M/\Gamma$ plays the role of the underlying analytic space of the complex orbifold $[M/\Gamma]$.} 
\end{proposition}
\begin{proof}
	For any $x \in M/\Gamma$, choose a point $\tilde{x} \in M$ that projects onto $x$ and let $\Gamma_{\tilde{x}} :=\left\{\gamma \in \Gamma \,\big|\, \gamma \tilde{x} = \tilde{x}\right\}$ denote the stabilizer subgroup (or isotropy subgroup) of $\tilde{x}$. Since the action of $\Gamma$ on $M$ is properly discontinuous, according to Remark~\ref{rmk:properlydiscontinous}, there exists a (small enough) neighborhood $\tilde{U}_{\tilde{x}} \subset M$ containing $\tilde{x}$ such that $\gamma(\tilde{U}_{\tilde{x}}) = \tilde{U}_{\tilde{x}}$ for all $\gamma \in \Gamma_{\tilde{x}}$ and $\tilde{U}_{\tilde{x}} \cap \gamma(\tilde{U}_{\tilde{x}}) = \emptyset$ for all elements of $\Gamma$ not in $\Gamma_{\tilde{x}}$. Then, the analytic quotient map $\varpi: M \to M/\Gamma$ will be a (branched) Galois covering which induces a biholomorphism between an open analytic subset $U \subset M/\Gamma$ containing $x$ and $\tilde{U}_{\tilde{x}}/\Gamma_{\tilde{x}}$. Then, the quadruple $(U_x, \tilde{U}_{\tilde{x}}, \Gamma_{\tilde{x}}, \varpi|_{\tilde{U}_{\tilde{x}}})$ will be an orbifold chart on $M/\Gamma$ containing the point $x$. By augmenting some cover $\{U_{x}\}$ of $M/\Gamma$ via adjoining finite intersections,\footnote{see \cite[Chap.~14]{thurston80} for more details.} one obtains a natural orbifold structure on $M/\Gamma$ induced by the atlas $\mathcal{U} = \{(U_a,\tilde{U}_a,\Gamma_a,\varpi|_{\tilde{U}_a})\}_{a\in A}$. The corresponding embeddings $\eta_{ab}: \tilde{U}_a \hookrightarrow \tilde{U}_b$ are induced by transition functions on the complex manifold $M$ and are thus guaranteed to be holomorphic; this ensures that the quotient orbifold $[M/\Gamma] = (M/\Gamma,\mathcal{U})$ inherits a complex structure that is uniquely determined by the complex structure of $M$ and the group $\Gamma$.
\end{proof}
Since any complex manifold can also be regarded as a complex orbifold for which all of the local uniformizing groups $\Gamma_i$ are given by the trivial group $\{\mathbbm{1}\}$, the analytic Galois covering map $\varpi: M \to X_O := M/\Gamma$ is naturally promoted to an orbifold Galois covering map $\varpi_{\text{orb}}: M \to O := [M/\Gamma]$ for which $\Gamma$ is the group of covering transformations -- i.e. $\gal(\varpi_{\text{orb}}) \cong \Gamma$. More generally, if $\Gamma$ is a discrete subgroup of holomorphic automorphisms of a complex manifold $M$ that acts effectively and properly discontinuously on it and if $\Gamma'$ is a subgroup of $\Gamma$, i.e. $\Gamma' \subset \Gamma \subset \operatorname{Aut}(M)$, the holomorphic orbifold map $[M/ \Gamma'] \to [M/ \Gamma]$ will be an orbifold Galois covering. 
\begin{remark}\label{rmk:universalcov}
	When the manifold $M$ in the above construction of $[M/\Gamma]$ is \emph{simply connected}, $M$ plays the role of \emph{universal covering space} and $\Gamma$ plays the role of \emph{orbifold fundamental group}.
\end{remark}

Due to many nice features that global quotient orbifolds share with ordinary manifolds, they were given the name ``good orbifolds'' by Thurston \cite[Chap.~14]{thurston80}:
\begin{definition}[Good orbifolds]
	A complex orbifold $O$ is called \emph{good} or \emph{developable} if it is analytically isomorphic to a global quotient orbifold $[M/\Gamma]$ in an orbifold sense ($\Gamma$ is discrete but not necessarily finite). When the group $\Gamma$ is also finite, the orbifold $O \cong [M/\Gamma]$ will be called \emph{very good}. In other words, a (complex) orbifold $O$ is good (respectively, very good) \emph{if and only if} $O$ has a covering (respectively, finite covering) that is a (complex) manifold. Otherwise, we have a \emph{bad} orbifold. 
\end{definition}
\subsubsection*{Complex Analytic Orbifolds as Log Pairs}
When local isotropy groups contain pseudoreflections, each reflection fixes a hyperplane in $\tilde{U}_i$, and the folding map $\ff_i: \tilde{U}_i \to U_i$ will have a \emph{ramification divisor} $\mathscr{R}_{i}$ on $\tilde{U}_i$ and a \emph{branch divisor} $\mathscr{B}_{i}$ on $U_i$. Let $\mathcal{B}_{ij}$ denote the irreducible components of the branch divisor $\mathscr{B}_{i}$ and let $m_{ij}$ be the \emph{branching index} (or the \emph{multiplicity}) of $\ff_i$ along each prime divisor $\mathcal{B}_{ij}$ such that $\mathscr{B}_{i} = \sum_{j} \big(1-\tfrac{1}{m_{ij}}\big) \mathcal{B}_{ij}$. Then, the compatibility condition between orbifold charts means that there are \emph{global prime divisors} $\mathcal{B}_j \subset X_O$ and \emph{ramification indices} $m_j$ such that $\mathcal{B}_{ij} = U_i \cap \mathcal{B}_{j}$  and $m_{ij} = m_j$ (after suitable re-indexing). Therefore, it will be convenient to codify the above data by a single effective $\mathbb{Q}$-divisor
\begin{equation}
\brdiv := \sum_{\mathcal{B}_j \subset \sing(O)} \left(1-\frac{1}{m_j}\right) \mathcal{B}_j,
\end{equation}
called the \emph{branch divisor} of $O$. It turns out that a complex orbifold $O=(X_O, \mathcal{U})$ can be uniquely determined by the pair $(X_O,\brdiv)$, called a \emph{log pair}: An orbifold atlas $\mathcal{U}= \{(U_a, \tilde{U}_a, \Gamma_a, \ff_a)\}_{a \in A}$ on a normal analytic space $X_O$ is said to be \emph{compatible} with $\brdiv$ if every branch divisor $\mathscr{B}_{i}$ associated with the Galois coverings $\ff_i$ coincide with $\brdiv \cap U_i$. Therefore, we can alternatively characterize (although slightly inaccurately) a complex orbifold $O$ as being defined by the pair $\big(X_O, \brdiv\big)$.
This point of view was taken in \cite{Campana2004}.\footnote{Traditional approaches studied either the singularities of a normal analytic space $X$, or the singularities of a divisor $\brdiv$ on a smooth analytic space, but did not concentrate on problems that occur when both $X$ and $\brdiv$ are singular.} 

As in definition~\ref{def:Weil}, let us define the multiplicity $\operatorname{mult}_{\brdiv}(H)$ of $\brdiv$ along any irreducible divisor $H \subset X_O$ as the rational number $1-\tfrac{1}{m_j}$, if $H = \mathcal{B}_j$ for any $j$, and as $0$ if $H \neq \mathcal{B}_j$ for all $j$. Then, we put
\begin{equation}
\deg_{\brdiv}(H) := \frac{1}{1-\operatorname{mult}_{\brdiv}(H)},
\end{equation}
and call it the \emph{branching index of $\brdiv$ along $H$}. Observe that
\begin{equation*}
\deg_{\brdiv \cap U_i}(\mathcal{B}_{ij}) = \deg_{\ff_i}(\mathcal{B}_{ij}) = m_{ij},
\end{equation*}
and the branching index of $\brdiv$ along any $H \neq \mathcal{B}_j \, \, \forall j$ is equal to 1. Then, the definition~\ref{def:orbmap} of analytic orbifold maps is equivalent to the following one: 
\begin{definition}[Analytic orbifold maps between log pairs]
	Let $(Y, \brdiv_{Y})$ and $(X,\brdiv_{X})$ be two complex orbifolds. A finite analytic map $f: Y \to X$ between the underlying analytic spaces is called an \emph{orbifold analytic map} $f_{\text{orb}}: (Y, \brdiv_{Y}) \to (X,\brdiv_X)$ if 
	\begin{equation*}
	\deg_{\brdiv_X}\big(f(H)\big) \Big| \deg_{\brdiv_Y}(H) \cdot \deg_{f}(H),
	\end{equation*}
	for every irreducible hypersurface $H \subset Y$.
\end{definition}

The notions of orbifold biholomorphism and orbifold Galois covering can similarly be defined in this language. In particular, a branched Galois covering $\varpi: Y \to X$ will be called an orbifold Galois covering $\varpi_{\text{orb}}: (Y, \brdiv_{Y}) \to (X,\brdiv_X)$ if 
\begin{equation*}
\deg_{\brdiv_X}\big(\varpi(H)\big) = \deg_{\brdiv_Y}(H) \cdot \deg_{\varpi}(H),
\end{equation*}
for every irreducible hypersurface $H \subset Y$.
\subsection{Orbifold Riemann Surfaces}
Now that we have seen the basic definitions for the case of an $\mathfrak{n}$-dimensional complex orbifold let us specialize to the case of $\mathfrak{n}=1$ where everything greatly simplifies (see e.g., \cite[Appx.~E]{milnor1990dynamics} and \cite{nasatyr1995orbifold} for more details). The most significant simplification comes from the fact that, by corollary~\ref{corll:1dAnalytic}, every 1-dimensional complex analytic space is smooth. In other words, the underlying analytic space of an orbifold Riemann surface is an ordinary (i.e., smooth) Riemann surface. Moreover, the orbifold charts on a Riemann surface can always be chosen to have the form $(\mathrm{D},\mathrm{D},\mathbb{Z}_{m_i},\ff_i)$ where $\mathrm{D} \subset \cmpx$ is the unit disk, $\mathbb{Z}_{m_i}$ denotes the cyclic group of order $m_i \geq 2$ which acts on $\mathrm{D}$ in the standard way as the $m_i$-th roots of unity, and the folding map $\ff_i: \mathrm{D} \to \mathrm{D}/\mathbb{Z}_{m_i} \cong \mathrm{D}$ is a branched Galois covering of the form $z \mapsto z^{m_i}$ (see Figure~\ref{fig:branchedcovering}).

Additionally, observe that ($\mathbb{Q}$-)divisors on Riemann surfaces are nothing but a formal linear combination of points with coefficients in $\mathbb{Q}$ or $\mathbb{Z}$.\footnote{Since Riemann surfaces are smooth, we don't have to differentiate between Weil and Cartier divisors.} Therefore, we have 
\begin{definition}[Orbifold Riemann surface]
	A closed \emph{orbifold Riemann surface} $O$ is a pair $(X_O,\brdiv)$ consisting of a closed Riemann surface $X_O$, called the \emph{underlying Riemann surface structure} of $O$, together with a \emph{branch divisor} $\brdiv = \sum_{x_i \in \sing(O)} \left(1-\frac{1}{m_i}\right) x_i$, where $x_i \in \sing(O)$ are pairwise distinct marked points on $X_O$, called the \emph{conical points}, and each integer $m_i \geq 2$ is the corresponding \emph{branching index} (or the \emph{order of isotropy}) of the conical point $x_i$.
\end{definition}
\begin{remark}
	The above definition of a closed Riemann orbisurface can be generalized to include punctured orbifold Riemann surfaces as well. Formally, cusps can be viewed as a limit $m \to \infty$ (or equivalently cone angle $2\pi/m \to 0$) of a conical singularity on a Riemann orbisurface (see Figure \ref{fig:cusplimit}). However, this limit is \emph{singular} and cannot be blindly taken from the formulas for conical singularities.
\end{remark}
\begin{figure}
	\centering
	\includegraphics[width=1.1\linewidth]{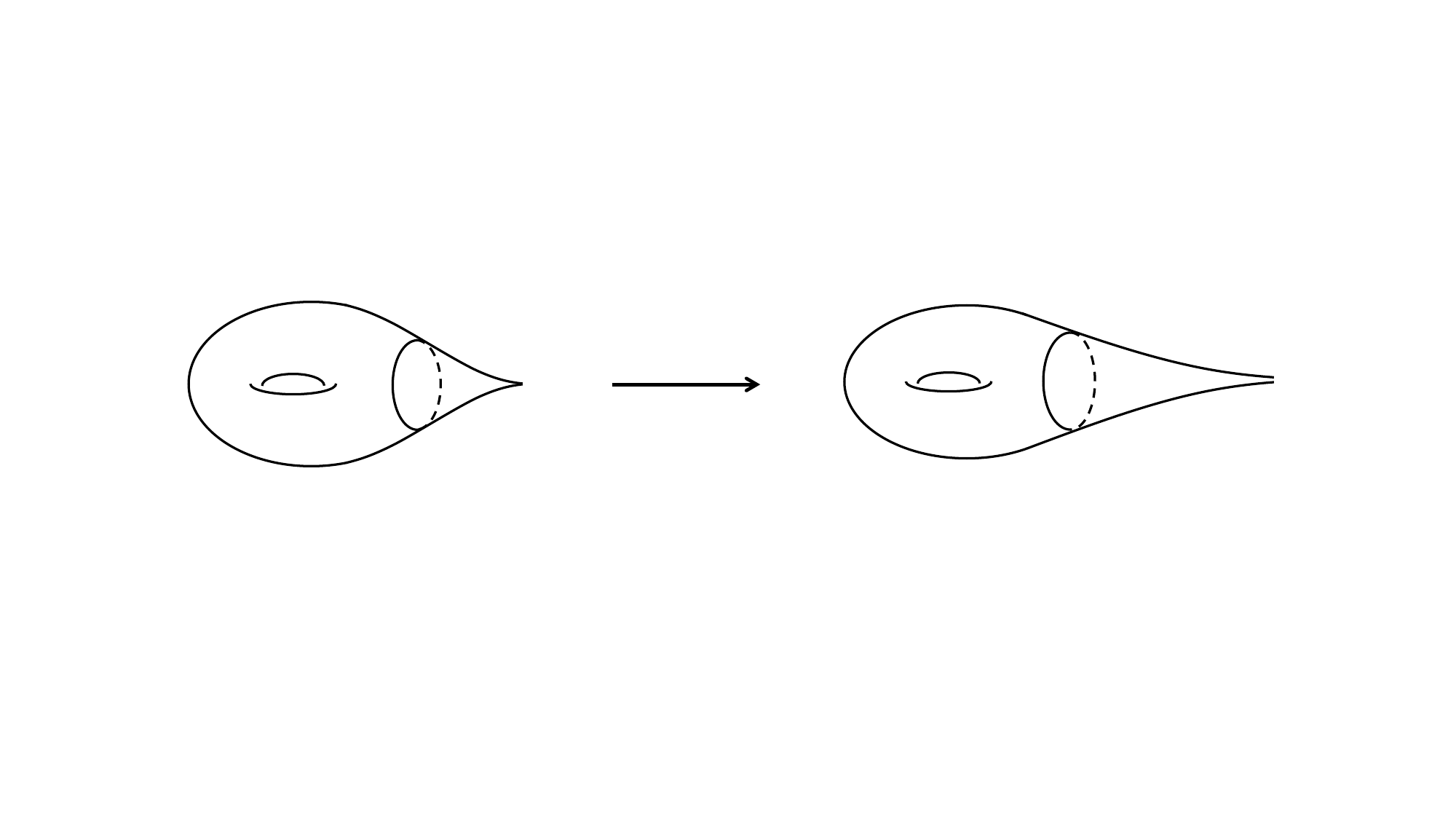}
	\put(-365,172){\rotatebox{0}{\fontsize{10}{10}\selectfont{} $\text{Conical point}$}}
	\put(-345,112){\rotatebox{0}{\fontsize{10}{10}\selectfont{} $\theta_m=\frac{2\pi}{m}$}}
	\put(-270,145){\rotatebox{0}{\fontsize{10}{10}\selectfont{} $m\rightarrow \infty$}}
	\put(-270,130){\rotatebox{0}{\fontsize{10}{10}\selectfont{} $\theta_m\rightarrow 0$}}
	\put(-110,169){\rotatebox{0}{\fontsize{10}{10}\selectfont{} $\text{Cusp}$}}
	\vspace{-3.5cm}
	
	\caption{\emph{Cusps as limits of cone points}. Formally, cusps can be viewed as a limit $m \to \infty$ (or cone angle $\theta_m =2\pi/m \to 0$) of a conical singularity on a Riemann orbisurface.} \label{fig:cusplimit}
\end{figure}
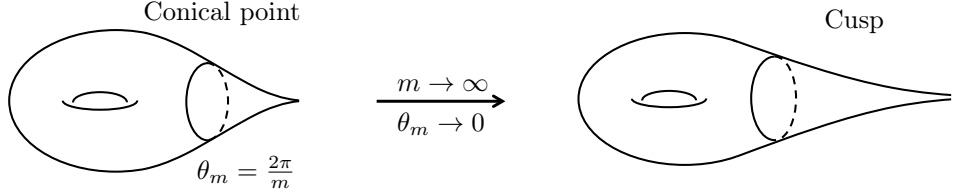

The above definition makes it clear that the whole theory of orbifold Riemann surfaces can be phrased in terms of Riemann surfaces with signature, where the signature is the map $X_O \to \mathbb{N} \cup \{\infty\}$ taking a point to its order:
\begin{definition}[Riemann surfaces with signature]\label{def:surfacewithsignature}
	By a \emph{Riemann surface with signature} we mean a Riemann surface $X$ of finite type\footnote{A Riemann surface $X$ is said to be of \emph{finite type} if $X$ is a (stable) Riemann surface (with or without nodes) such that either $n=0$ and $X$ is \emph{compact}, or $n>1$ and $X$ is compact except for $n$ punctures.} $(g,n)$ together with an assignment of a \emph{branching index} $m_i$ to each marked point; these branching indices $m_i$ are either integers $\geq 2$ or the symbol $\infty$. The \emph{signature} of $O$ is the tuple $(g;m_1,\dots,m_n)$ where $g$ is the genus of $X_O$ and branching indices $m_i$ are ordered such that $2 \leq m_1 \leq m_2 \leq \dotsm \leq m_n$. 
\end{definition}
\begin{remark}
	More precisely, by assigning the branching index 1 to every other point (except the marked points), we can define a Riemann surface with signature as a pair $(X,\nu)$ where $\nu : X \to \mathbb{N} \cup \{\infty\}$ is a branching function. Note that if $\mathcal{U} = \{(U_i,\mathrm{D},\mathbb{Z}_{m_i},\ff_i)\}_{i \in I}$ is an orbifold atlas on $X$ which is compatible with $\nu$, the restriction $\nu|_{_{U_i}}$ is precisely the same as the branching function associated to the branched Galois covering $\ff_i$ --- i.e. $\nu|_{_{U_i}}  \equiv \nu_{\ff_i}$.
\end{remark}

Let $(\tilde{X}, \tilde{\nu})$ and $(X, \nu)$ be two orbifold Riemann surfaces. A Galois branched covering map $\varpi: \tilde{X} \to X$ is said to yield an orbifold Galois covering $\varpi_{\text{orb}}: (\tilde{X}, \tilde{\nu}) \to (X, \nu)$ if
\begin{equation}
\nu\big(\varpi(\tilde{x})\big) = \tilde{\nu}(\tilde{x}) \deg_{\varpi}(\tilde{x}) \qquad \text{for all} \qquad \tilde{x} \in \tilde{X}.
\end{equation}
Note that a conical point $x_j \in \sing(X, \nu)$ may be obtained in two different ways: 
\begin{enumerate}
	\item If $\tilde{x} \notin \sing(\tilde{X},\tilde{\nu})$ but it is a ramification point of $\varpi$ with ramification index $\deg_{\varpi}(\tilde{x})>1$, then $\varpi(\tilde{x})$ is a conical point of order $\deg_{\varpi}(\tilde{x})$.
	\item If $\tilde{x} \in \sing(\tilde{X},\tilde{\nu})$ and the local degree of $\varpi$ at $\tilde{x}_i$ is given by $\deg_{\varpi}(\tilde{x}_i) \geq 1$, then $\varpi(\tilde{x}_i)$ is also a conical point of order $\tilde{\nu}(\tilde{x}_i) \deg_{\varpi}(\tilde{x}_i)$. As a result, we will always have $|\sing(\tilde{X},\tilde{\nu})| \leq |\sing(X,\nu)|$.
\end{enumerate}
When the branching function $\tilde{\nu}(x)$ is the trivial branching function $\tilde{\nu} \equiv 1$, we conclude:
\begin{proposition}\label{prop:branchedorbifoldcorrepondence}
	Let $\varpi: \tilde{X} \to X$ be a branched Galois covering between two Riemann surfaces. The base Riemann surface $X$ can be naturally given the structure of a (developable) Riemann orbisurface $O=(X,\nu_{\varpi})$ where $\nu_{\varpi}: X \to \mathbb{N}$ is the branching function associated to the branched Galois covering $\varpi$. Note that $O \cong [\tilde{X}/\gal(\varpi)]$.
\end{proposition}
\subsection{Universal Orbifold Covering and Orbifold Fundamental Group}

In Remark~\ref{rmk:universalcov}, we have briefly mentioned universal covers of good orbifolds. However, even when an orbifold $O$ is not necessarily assumed to be developable, it is rather easy to define a \emph{universal covering orbifold} of $O$ in much the same way as one defines the universal covering of manifolds and the same uniqueness property holds in the orbifold case as well:
\begin{definition}[Universal covering orbifold]
	An orbifold Galois covering
	\begin{equation}
	\varpi_{\text{orb}}: \tilde{O} \to O
	\end{equation}
	is called a \emph{universal orbifold covering} of $O$ if for any other orbifold covering
	\begin{equation}
	\varpi'_{\text{orb}}: \breve{O} \to O
	\end{equation}
	there exists a lifting of $\varpi_{\text{orb}}$ to an orbifold covering map $\varpi''_{\text{orb}}:\tilde{O} \to \breve{O}$ such that the diagram
	\begin{equation}
	\begin{tikzcd}
	\tilde{O} \ar[dd,"\varpi_{\text{orb}}"] \ar[dr, "\varpi''_{\text{orb}}"] &  \\
	& \breve{O} \ar[ld,"\varpi'_{\text{orb}}"] \\
	O
	\end{tikzcd}
	\end{equation}
	commutes.
\end{definition}
In the case of developable orbifolds $O \cong [M / \Gamma]$, any (unramified) manifold covering $\tilde{M} \to M$ gives an orbifold covering by composition with the quotient map $M \to [M / \Gamma]$. In particular, the universal covering of $M$ gives rise to a universal orbifold covering of $O$, and the orbifold fundamental group is given by the short exact sequence
\begin{equation}
\mathbbm{1} \to \pi_1(M) \to \pi_1(O) \to \Gamma \to \mathbbm{1}.
\end{equation}

Once again, the situation is particularly nice for the case of orbifold Riemann surfaces:
\begin{theorem}[\cite{milnor1990dynamics} Theorem~E.1]\label{thm:universalcovering}
	With the following two exceptions, every orbifold Riemann surface of finite type $O=(X_O,\brdiv)$ admits as the universal cover either the Riemann sphere $\hat{\cmpx}$, the complex plane $\cmpx$ or the hyperbolic plane $\UHP$ which is necessarily a \emph{finite} Galois branched covering and is unique up to conformal isomorphism over $X_O$; this is a consequence of the classical uniformization theorem. The only exceptions, called ``bad'' orbisurfaces in Thurston's terminology, are given by:
	\begin{enumerate}[(i)]
		\item  Teardrop orbisurface: Riemann sphere $\hat{\cmpx}$ with just one ramified point (see Figure \ref{fig:badorbifolds}), or 
		\item Spindle orbisurface: Riemann sphere $\hat{\cmpx}$ with two ramified points for which the ramification indices are different (see Figure \ref{fig:badorbifolds}).
	\end{enumerate}
\end{theorem}
\begin{figure}
	\centering
	\includegraphics[width=.95\linewidth]{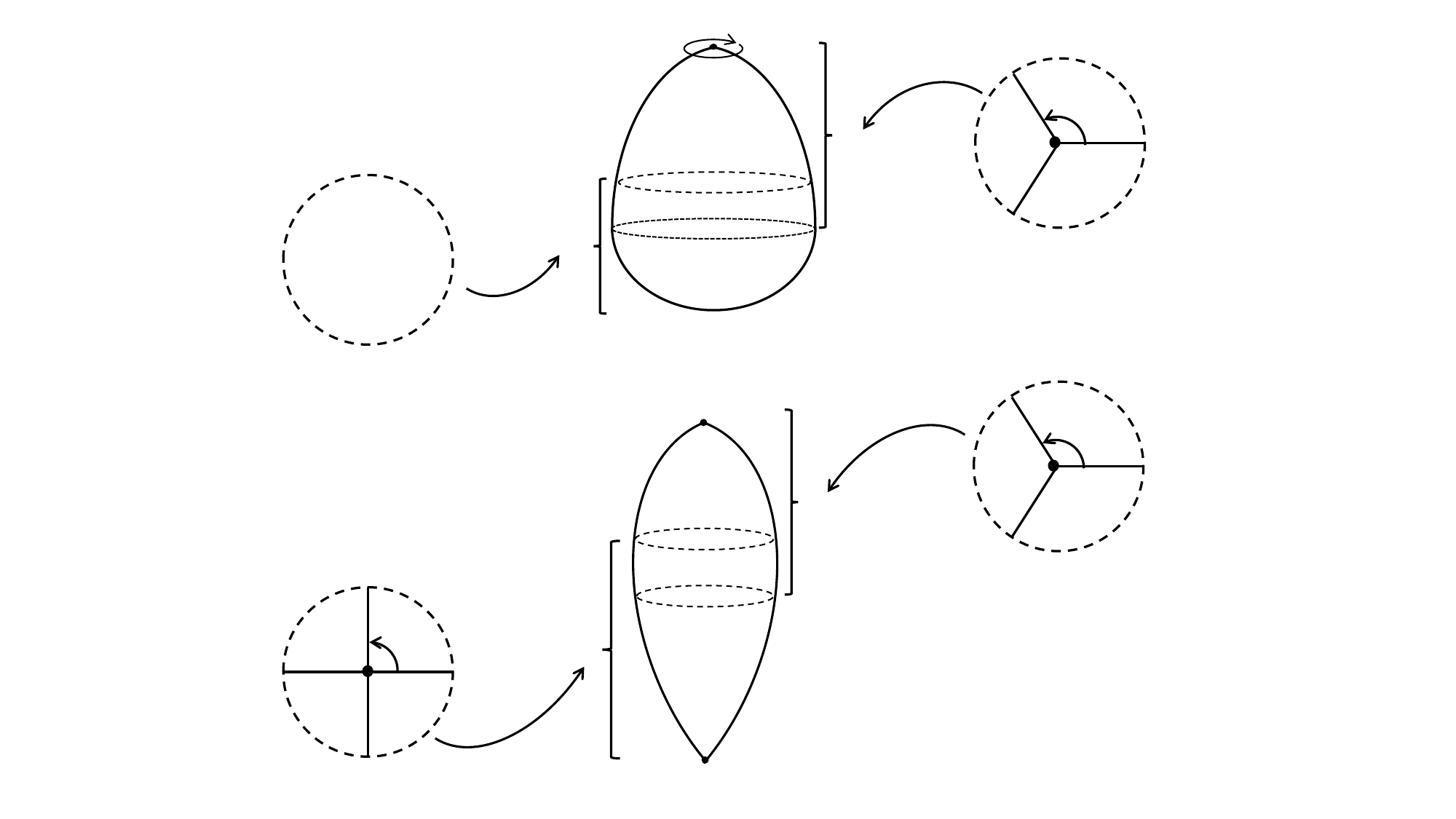}
	\put(-332,182){\rotatebox{0}{\fontsize{10}{10}\selectfont{} $\mathrm{D}$}}
	\put(-332,66){\rotatebox{0}{\fontsize{10}{10}\selectfont{} $\mathrm{D}$}}
	\put(-102,213){\rotatebox{0}{\fontsize{10}{10}\selectfont{} $\mathrm{D}$}}
	\put(-102,121){\rotatebox{0}{\fontsize{10}{10}\selectfont{} $\mathrm{D}$}}
	\put(-272,138){\rotatebox{0}{\fontsize{10}{10}\selectfont{} $\ff_2$}}
	\put(-258,160){\rotatebox{0}{\fontsize{10}{10}\selectfont{} $U_2$}}
	\put(-255,46){\rotatebox{0}{\fontsize{10}{10}\selectfont{} $U_2$}}
	\put(-272,13){\rotatebox{0}{\fontsize{10}{10}\selectfont{} $\ff_2$}}
	\put(-187,84){\rotatebox{0}{\fontsize{10}{10}\selectfont{} $U_1$}}
	\put(-160,114){\rotatebox{0}{\fontsize{10}{10}\selectfont{} $\ff_1$}}
	\put(-155,210){\rotatebox{0}{\fontsize{10}{10}\selectfont{} $\ff_1$}}
	\put(-177,186){\rotatebox{0}{\fontsize{10}{10}\selectfont{} $U_1$}}
	\put(-216,225){\rotatebox{0}{\fontsize{10}{10}\selectfont{} $N$}}
	\put(-218,203){\rotatebox{0}{\fontsize{10}{10}\selectfont{} $\frac{2\pi}{3}$}}
	\put(-115,200){\rotatebox{0}{\fontsize{10}{10}\selectfont{} $\frac{2\pi}{3}$}}
	\put(-115,109){\rotatebox{0}{\fontsize{10}{10}\selectfont{} $\frac{2\pi}{3}$}}
	\put(-303,50){\rotatebox{0}{\fontsize{10}{10}\selectfont{} $\frac{\pi}{2}$}}
	\put(-219,114){\rotatebox{0}{\fontsize{10}{10}\selectfont{} $N$}}
	\put(-219,7){\rotatebox{0}{\fontsize{10}{10}\selectfont{} $S$}}
	\caption{\emph{Bad orbisurfaces}. The top figure shows the (3)-teardrop orbisurface, which consists of the tuple $\big(\hat{\cmpx}, \brdiv = 3 N \big)$ where the Riemann sphere $\hat{\cmpx}$ is the underlying Riemann surface and the north pole $N$ is the only marked point with multiplicity 3. The bottom figure shows the (3,4)-spindle orbisurface that is given by the tuple $\big(\hat{\cmpx}, \brdiv = 3 N + 4 S\big)$ where now both the north pole $N$ and the south pole $S$ are marked points with inequivalent cone orders. Note that the $(m_N)$-teardrop can be viewed as the $(m_N,1)$-spindle. Both of these orbisurfaces are ``bad'' orbisurfaces in the sense that they admit no Riemann surface as their universal covering.} \label{fig:badorbifolds}
\end{figure}
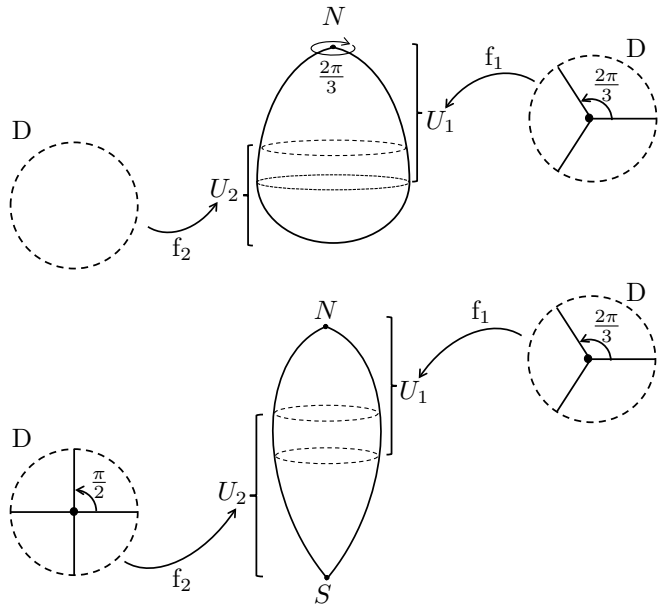	
\begin{remark}\label{rmk:finitecovering}
	The statement that \emph{every} developable Riemann orbisurface of finite type is \emph{finitely} covered by $\hat{\cmpx}$, $\cmpx$, or $\UHP$ is equivalent to asserting that \emph{any} finitely generated, discrete subgroup $\Gamma$ of (orientation-preserving) isometries $\isom^{+}(\hat{\cmpx})\cong \PSLC$, $\isom^{+}(\cmpx) $, or $\isom^{+}(\UHP) \cong \PSLR$ with quotient space of finite type has a \emph{torsion-free subgroup of finite index}. The conjugacy classes of torsion in $\Gamma$ correspond to the conical points of $[\mathbb{X}/\Gamma]$, for $\mathbb{X} \in \{\hat{\cmpx},\cmpx,\UHP\}$, and one obtains a torsion-free subgroup of finite index of $\Gamma$ by avoiding these (finite number of) conjugacy classes. The resulting torsion-free subgroup is precisely the fundamental group of the underlying Riemann surface, $\pi_1(X_O)$.
\end{remark}

More generally, Thurston proved that every orbifold $O$ has a universal cover regardless of being developable or not (see \cite[Prop.~13.2.4]{thurston80}) and also defined the orbifold fundamental group $\pi_1(O)$ as the group of covering transformations of its universal covering:
\begin{definition}[Orbifold fundamental group]\label{def:orbifoldfundgroup}
	The \emph{orbifold fundamental group} $\pi_1(O)$ of an orbifold $O$ is the Galois group of its universal covering. 
\end{definition}

We need an interpretation of $\pi_1(O)$ in terms of homotopy classes of loops in $O$. However, defining the correct notion of homotopy would be an issue: For instance, a disk $\mathrm{D}_m \cong [\mathrm{D}/\mathbb{Z}_{m}]$ with one cone point of order $m$ has as its universal covering a non-singular disk, with covering transformation group $\mathbb{Z}_{m}$ acting by rotations. Thus, intuitively, a loop in $\mathrm{D}_m$ winding $m$ times around the cone point should be null-homotopic (see Figure \ref{fig:orbifoldfundgroup}).
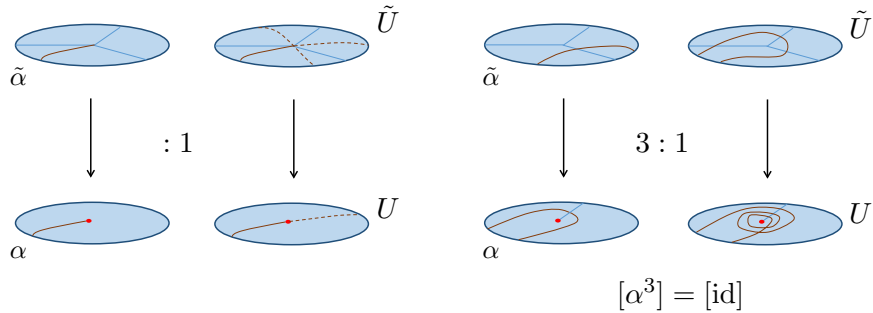
\begin{figure}
	\centering
	\includegraphics[width=0.8\linewidth]{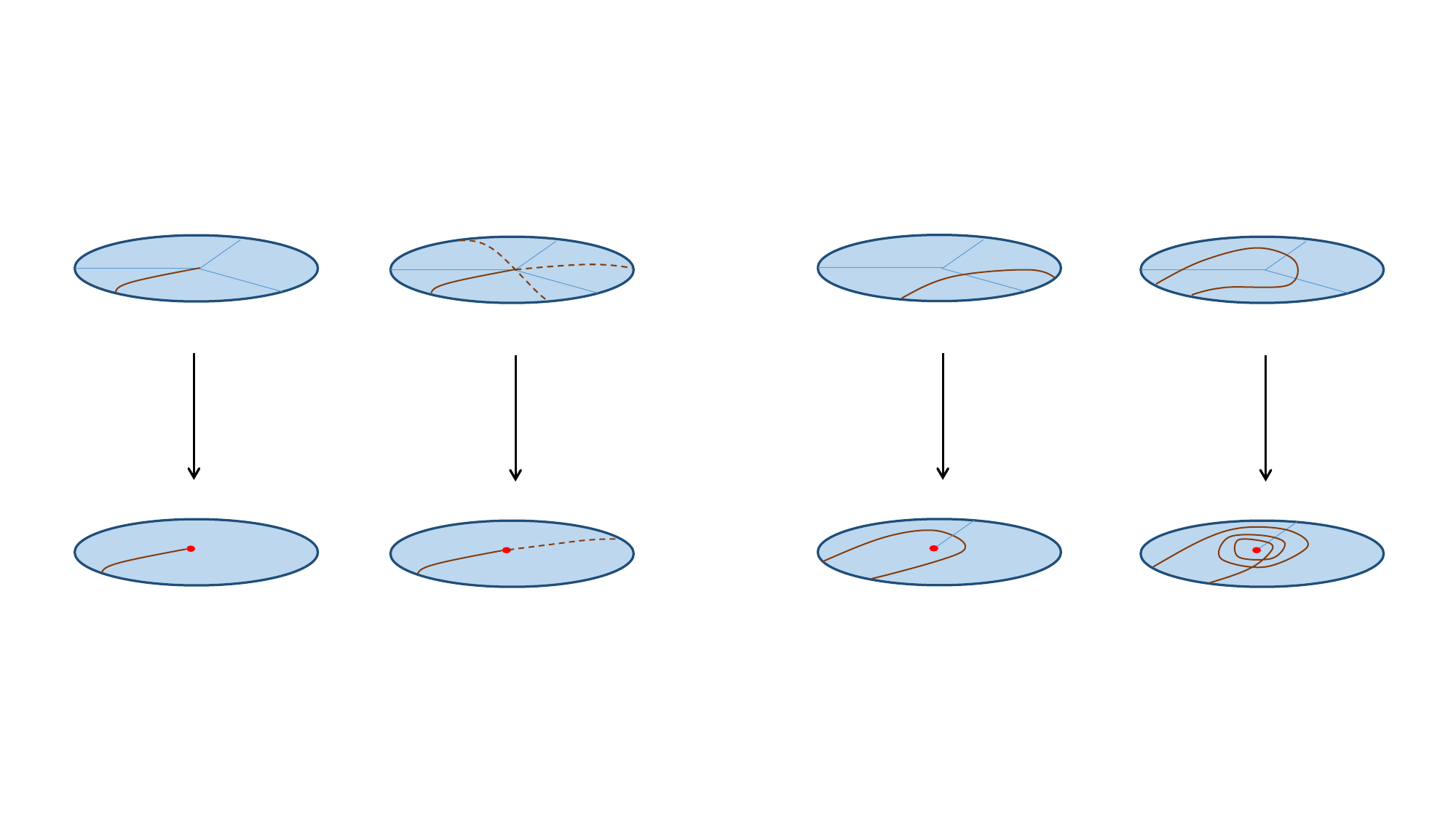}
	\put(-332,115){\rotatebox{0}{\fontsize{10}{10}\selectfont{} $\tilde{\alpha}$}}
	\put(-155,115){\rotatebox{0}{\fontsize{10}{10}\selectfont{} $\tilde{\alpha}$}}
	\put(-332,50){\rotatebox{0}{\fontsize{10}{10}\selectfont{} $\alpha$}}
	\put(-155,50){\rotatebox{0}{\fontsize{10}{10}\selectfont{} $\alpha$}}
	\put(-275,90){\rotatebox{0}{\fontsize{11}{11}\selectfont{} $:1$}}
	\put(-98,90){\rotatebox{0}{\fontsize{11}{11}\selectfont{} $3:1$}}
	\put(-195,135){\rotatebox{0}{\fontsize{11}{11}\selectfont{} $\tilde{U}$}}
	\put(-195,65){\rotatebox{0}{\fontsize{11}{11}\selectfont{} $U$}}
	\put(-18,133){\rotatebox{0}{\fontsize{11}{11}\selectfont{} $\tilde{U}$}}
	\put(-18,63){\rotatebox{0}{\fontsize{11}{11}\selectfont{} $U$}}
	\put(-105,33){\rotatebox{0}{\fontsize{11}{11}\selectfont{} $[\alpha^{3}]= [\operatorname{id}]$}}
	\vspace{-1cm}
	
	\caption{\emph{Orbifold covering and homotopy classes of loops}. This figure compares the loops around a ramified point of index three on an orbifold Riemann surface with those around the pre-image of this cone point on the covering space. We can see that circling once around the pre-image of the ramification point on the covering (orbi)surface is equivalent to circling three times around the ramified point on the base orbisurface.} \label{fig:orbifoldfundgroup}
\end{figure}

In order to introduce the correct notion of orbifold homotopy, it is more convenient to view the complex $\mathfrak{n}$-dimensional orbifold $O$ as a smooth (oriented) real orbifold of dimension $2\mathfrak{n}$ and work with topologists definition of an orbifold (see remark~\ref{rmk:topologistsdef}): Let $O=(|O|,\mathscr{U})$ be a smooth orbifold and let $(|U|, \tilde{U} \subset \mathbb{R}^{2\mathfrak{n}}, \Gamma \subset \operatorname{SO}(2\mathfrak{n}), |\ff|) \in \mathscr{U}$ be an orbifold chart on the underlying topological space $|O|$. Then, $([0,1] \times |U|, [0,1] \times \tilde{U}, \Gamma, \operatorname{id} \times |\ff|)$ where $\gamma \in \Gamma$ acts on $[0,1] \times \tilde{U}$ via $\gamma(t,x) = (t,\gamma x)$ is an orbifold chart on $[0,1] \times O$. The collection of all such orbifold charts forms an orbifold atlas, giving $[0,1] \times O$ the structure of a smooth orbifold. It thus makes sense to say that two orbifold maps
\begin{equation*}
f_{\text{orb}}, f'_{\text{orb}}: O \to O'
\end{equation*}
are \emph{homotopic in the category of orbifolds} if there is an orbifold map
\begin{equation*}
F: [0,1] \times O \to O', \qquad F(t,x) = F_{t}(x)
\end{equation*}
with $F_0 = f_{\text{orb}}$ and $F_1 = f'_{\text{orb}}$.

Armed with the notion of homotopy of orbifold maps, one can define the orbifold fundamental group $\pi_1(O)$ exactly as one does for the usual fundamental group, just replacing homotopies by orbifold homotopies. One should note that any two orbifold maps which are homotopic as orbifold maps are also homotopic as maps between topological spaces, but that the converse does not need to be true. In fact, there are plenty of orbifolds which are simply connected as topological spaces but whose orbifold fundamental group is non-trivial --- i.e., there are orbifold maps $\ell : S^{1} \to O$ which, as orbifold maps, are not homotopic to the constant map. See e.g. \cite[\textsection2.3]{erlandsson2020counting} and \cite[\textsection2.2]{boileau2003three} for more details.

At this stage, the reader might wonder how one can \emph{compute} fundamental groups of orbifolds. This can be done by studying the fundamental group of the regular locus $O\backslash \sing(O)$. Indeed, if singular locus $\sing(O)$ has \emph{real} codimension at least two (which is always the case if $O$ is complex), then $O\backslash \sing(O)$ is connected, and we have a surjective homomorphism $\pi_1\big(O \backslash \sing(O)\big) \to \pi_1(O)$ induced by inclusion $O \backslash \sing(O) \hookrightarrow O$. The surjectivity of this homomorphism comes from the fact that any loop on $O$ can be perturbed to avoid the singular locus. To compute $\pi_1(O)$, we only need to find the kernel of $\pi_1\big(O \backslash \sing(O)\big) \to \pi_1(O)$ --- i.e. which elements of $\pi_1\big(O \backslash \sing(O)\big)$ get killed. For the special case of orbifold Riemann surfaces, we can use the \emph{orbifold Seifert--van Kampen theorem} to arrive at the following proposition (see e.g. \cite{thurston80,Scott1983TheGO} for more details):\footnote{While for the proof of the following proposition requires more complicated machinery for orbifolds of general dimension, the conclusion of proposition~\ref{prop:fundgroup} is valid in all dimensions.} 
\begin{proposition}\label{prop:fundgroup}
	The group $\pi_1(O)$ is the quotient of $\pi_1\big(O \backslash \sing(O)\big)$ by the group normally generated by the elements $\mu_{x}^{m_x}$, for all $x \in \sing(O)$, where $m_x := \#\Gamma_{x}$ is the order of the local isotropy group of $x$ and $\mu_x$ is a meridian around $x$.
\end{proposition}

Thus, if $O$ is an orbifold Riemann surface with signature $(g; m_1,\dots,m_{n_e};n_p)$, $\pi_1(O)$ can be presented as (see Figure \ref{fig:fundamentalgroup} and \ref{fig:fundgroup})
\begin{footnotesize}
	\begin{equation}	
	\pi_1(O,x_{\ast}) = \left\langle \mathsf{A}_1, \mathsf{B}_1,\dots, \mathsf{A}_g, \mathsf{B}_g, \mathsf{C}_1, \dots, \mathsf{C}_{n_e}, \mathsf{P}_1, \dots, \mathsf{P}_{n_p} \,\bigg|\, \mathsf{C}_1^{m_1} = \dotsm = \mathsf{C}_{n_e}^{m_{n_e}} = \prod_{i=1}^{g}[\mathsf{A}_i,\mathsf{B}_i] \prod_{j=1}^{n_e} \mathsf{C}_j \prod_{k=1}^{n_p} \mathsf{P}_k = \operatorname{id} \right\rangle,
	\end{equation}
\end{footnotesize}
where $\mathsf{A}_i$s and $\mathsf{B}_i$s are homotopy classes of loops (based at $x_{\ast}$) that span $H_1(X_O,\mathbb{Z})$, $\mathsf{C}_j$s and $\mathsf{P}_k$s are meridians around conical points and cusps respectively, the commutator $[\mathsf{A}_i,\mathsf{,B}_i]$ is defined as $\mathsf{A}_i \mathsf{B}_i \mathsf{A}_i^{-1} \mathsf{B}_i^{-1}$, and the relation $\prod_{i=1}^{g}[\mathsf{A}_i,\mathsf{B}_i] \prod_{j=1}^{n_e} \mathsf{C}_j \prod_{k=1}^{n_p} \mathsf{P}_k = \operatorname{id}$ comes from cutting open $O$ along the chosen basis for $\pi_1(O,x_{\ast})$.

\begin{figure}
	\centering
	\includegraphics[width=0.9\linewidth]{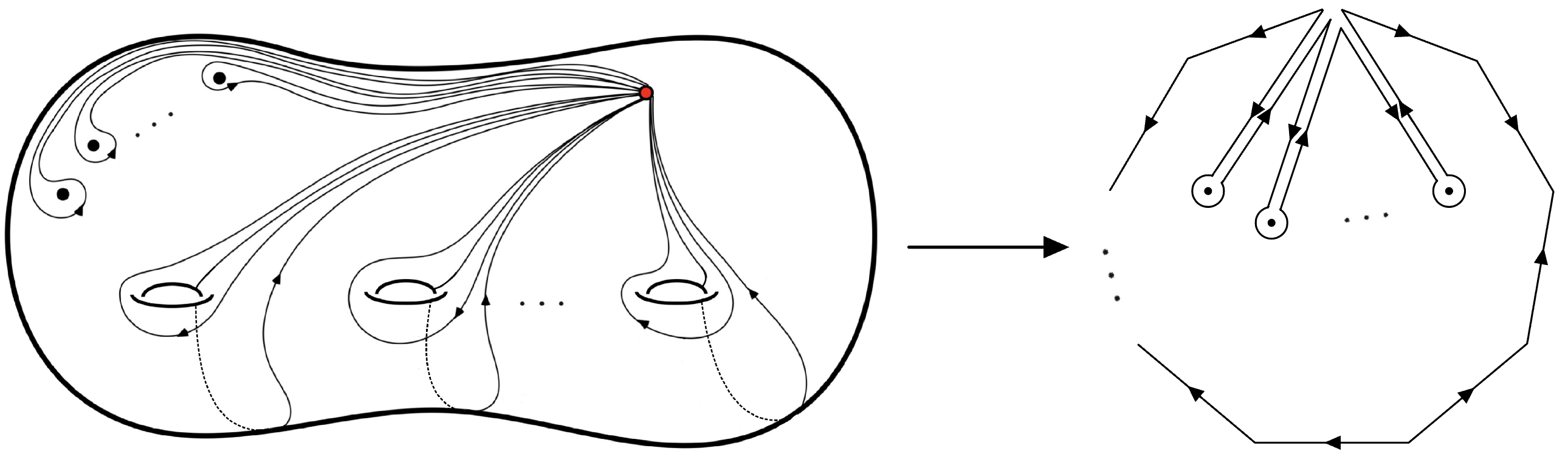} 
	\put(-368,63){\rotatebox{0}{\fontsize{9}{9}\selectfont{} $C_{1}$}}
	\put(-360,74){\rotatebox{0}{\fontsize{9}{9}\selectfont{} $C_{2}$}}
	\put(-340,84){\rotatebox{0}{\fontsize{9}{9}\selectfont{} $P_{n_p}$}}
	\put(-368,30){\rotatebox{0}{\fontsize{9}{9}\selectfont{} $B_{1}$}}
	\put(-313,30){\rotatebox{0}{\fontsize{9}{9}\selectfont{} $B_{2}$}}
	\put(-243,30){\rotatebox{0}{\fontsize{9}{9}\selectfont{} $B_{g}$}}
	\put(-193,29){\rotatebox{0}{\fontsize{9}{9}\selectfont{} $A_{g}$}}
	\put(-266,20){\rotatebox{0}{\fontsize{9}{9}\selectfont{} $A_{2}$}}
	\put(-334,19){\rotatebox{0}{\fontsize{9}{9}\selectfont{} $A_{1}$}}
	\put(-228,90){\rotatebox{0}{\fontsize{9}{9}\selectfont{} $x_{\ast}$}}
	\put(-167,56){\rotatebox{0}{\fontsize{7}{7}\selectfont{} $\text{Cut open}$}}
	\put(-168,45){\rotatebox{0}{\fontsize{7}{7}\selectfont{} $\text{along loops}$}}
	\put(-64,114){\rotatebox{0}{\fontsize{9}{9}\selectfont{} $x_{\ast}$}}
	\put(-121,84){\rotatebox{0}{\fontsize{9}{9}\selectfont{} $A_{g}^{-1}$}}
	\put(-90,110){\rotatebox{0}{\fontsize{9}{9}\selectfont{} $B_{g}^{-1}$}}
	\put(-44,108){\rotatebox{0}{\fontsize{9}{9}\selectfont{} $A_{1}$}}
	\put(-13,83){\rotatebox{0}{\fontsize{9}{9}\selectfont{} $B_{1}$}}
	\put(-8,45){\rotatebox{0}{\fontsize{9}{9}\selectfont{} $A_{1}^{-1}$}}
	\put(-27,9){\rotatebox{0}{\fontsize{9}{9}\selectfont{} $B_{1}^{-1}$}}
	\put(-65,-4){\rotatebox{0}{\fontsize{9}{9}\selectfont{} $A_{2}$}}
	\put(-105,10){\rotatebox{0}{\fontsize{9}{9}\selectfont{} $B_{2}$}}
	\put(-96,55){\rotatebox{0}{\fontsize{9}{9}\selectfont{} $C_{1}$}}
	\put(-80,47){\rotatebox{0}{\fontsize{9}{9}\selectfont{} $C_{2}$}}
	\put(-35,55){\rotatebox{0}{\fontsize{9}{9}\selectfont{} $P_{n_p}$}} 
	
	\caption{\emph{Orbifold fundamental group}. A choice of generators for the homotopy group
		of a surface with marked points.} 
	\label{fig:fundgroup}
\end{figure}
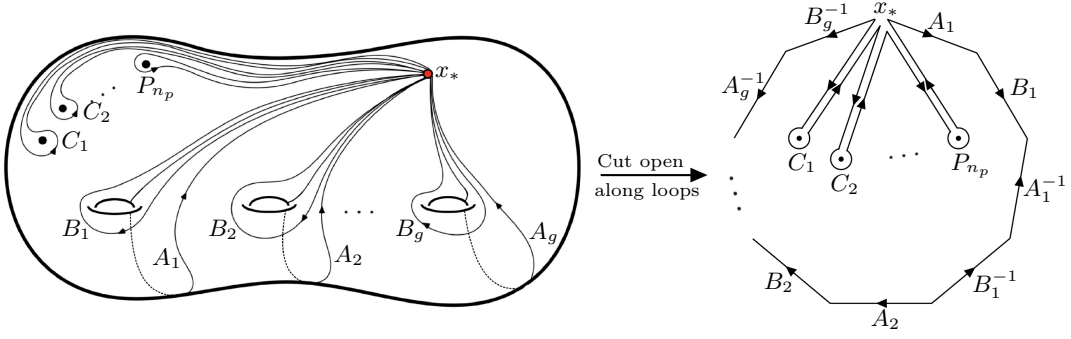
\subsection{Orbifold Euler Characteristic and Riemann-Hurwitz Formula}\label{sect:Eulerchracter}
The \emph{orbifold Euler characteristic} is a generalization of the notion of Euler characteristic for manifolds that includes contributions coming from nontrivial automorphisms. In particular, while every manifold has an integer Euler characteristic, the orbifold Euler characteristic is, in general a rational number. In this subsection, we will only focus on studying the orbifold Euler characteristic for the case of developable Riemann orbisurfaces.\footnote{The interested reader can consult the reference \cite{hirzebruch1990euler} for more general cases.} As we will see, the Euler characteristic has an important connection to branched Galois coverings, and this allows us to calculate the orbifold Euler characteristic for developable orbifold Riemann surfaces using the so called Riemann-Hurwitz formula.

For simplicity, let us start by only considering closed Riemann surfaces. We remember from elementary topology that a closed Riemann surface $Y$ has only one topological invariant, which we may take to be its genus $g$. In this case, the Euler Characteristic of $Y$, denoted by $\chi(Y)$, is found by triangulating $Y$ and using the formula $\chi(Y) = F_{aces} - E_{dges} + V_{ertices}$. As expected, the result only depends on the genus $g$ and is given by $\chi(Y) = 2-2g$. Now, let $\varpi: Y \to X$ be a Galois covering map between closed Riemann surfaces. There is a formula relating the various invariants involved: the genus of $Y$, the genus of $X$, the degree of $\varpi$, and the amount of ramification:
\begin{theorem}[Riemann-Hurwitz relation]
	Let $\varpi: Y \xrightarrow{d:1} X$ be a branched Galois covering map of degree $d$ --- i.e. $\deg(\varpi) = \#\gal(Y/X) = d$. We have the relation
	\begin{equation*}
	\chi(Y) = d \, \big[ \chi(X)  - \deg(\mathscr{B}_{\varpi}) \big]
	\end{equation*}
	where $\mathscr{B}_{\varpi} = \sum_{x_j \in B_{\varpi}} (1-\tfrac{1}{\nu_{\varpi}(x_j)}) x_j$ is the branch divisor associated to the branched Galois covering $\varpi$ and $\deg(\mathscr{B}_{\varpi}) = \sum_{x_j \in B_{\varpi}} (1-\tfrac{1}{\nu_{\varpi}(x_j)})$. 
\end{theorem}
\begin{proof}
	Here, we will provide a simple topological proof of the Riemann-Hurwitz (RH) formula:\footnote{One can also prove the RH relation by using a mixture of the Gauss-Bonnet formula and topological considerations \cite[\textsection2.1]{griffiths2014principles} or by using the relation $\mathcal{K}_{Y} = \varpi^{\ast} \mathcal{K}_{X} + \mathscr{R}_{\varpi}$ between the corresponding canonical divisors. We will come back to this point later in this appendix.} Choose a sufficiently small triangulation of $X$ so that each triangle is contained in an evenly covered neighborhood and such that every branch point of $\varpi$ is a vertex in this triangulation. Then, as mentioned above, $\chi(X) = F - E + V$ where $F$, $E$, and $V$ are the number of faces, edges, and vertices (respectively) of the chosen triangulation. Since $\varpi: Y \to X$ is surjective, the pullback of this triangulation is clearly a triangulation of $Y$. Thus, we just need to count the number of faces, edges, and vertices of this pulled-back triangulation to calculate the Euler Characteristic of $Y$. We denote these numbers by $\tilde{F}$, $\tilde{E}$, and $\tilde{V}$ respectively.
	
	The pull-back of each evenly covered neighborhood will contain $d$ copies of the triangle contained within it. Thus we have $d$ faces and $d$ edges --- i.e. $\tilde{F} = d \, F$ and $\tilde{E} = d \, E$. Naively, one would expect there to be $d$ vertices as well; however, since a branch point $x_j \in B_{\varpi}$ has
	\begin{equation}\label{eq:numlifts}
	\big| \varpi^{-1}(x_j)\big| = d - \sum_{y_i \in \varpi^{-1}(x_j)} \big(\deg_{\varpi}(y_i) -1\big) = d/\nu_{\varpi}(x_j)
	\end{equation}
	distinct pre-images, we have\footnote{Note that $\sum\limits_{y_i \in R_{\varpi}}=\sum\limits_{x_j\in B_{\varpi}}\sum\limits_{y_i \in \varpi^{-1}(x_j)}$.}
	\begin{equation}\label{eq:numverticies}
	\tilde{V} = d \, V - \sum_{y_i \in R_{\varpi}} \big(\deg_{\varpi}(y_i)-1\big) = d \, V - \sum_{x_j \in B_{\varpi}} \frac{d}{\nu_{\varpi}(x_j)} \big(\nu_{\varpi}(x_j) -1\big).
	\end{equation}
	Therefore, 
	\begin{equation}
	\chi(Y) = \tilde{F} - \tilde{E} + \tilde{V} =  d \Big[ (F - E + V ) - \sum_{x_j \in B_{\varpi}} (1-\tfrac{1}{\nu_{\varpi}(x_j)}) \Big] = d \, \big[\chi(X) - \deg(\mathscr{B}_{\varpi}) \big].
	\end{equation}
\end{proof}
\begin{remark}
	The RH relation proved above assumes that the branched covering $\varpi$ is Galois. However, it is possible to write RH relation in a way which is true for all branched coverings regardless of whether they are Galois or not. We have to pay attention to two main differences in this case: (i) For a general branched covering the degree $d$ of the covering is not necessarily equivalent to the order of the covering transformation group; in general, $d \leq \# \operatorname{Aut}(\varpi)$ with equality happening only when $\varpi$ is Galois. (ii) When $\varpi$ is not required to be Galois, the ramification indices of ramified points $y_1 , y_2 \in  \varpi^{-1}(x)$ need not be the same. Therefore, the last equalities in both \eqref{eq:numlifts} and \eqref{eq:numverticies} are not true in this general setting; in fact, when $\varpi$ is not necessarily Galois, we cannot even define branching indices and branch divisors. However, we can still write $\tilde{V} = d \, V - \sum_{y_i \in R_{\varpi}} \big(\deg_{\varpi}(y_i)-1\big)$ which results in the following form for RH relation:
	\begin{equation*}
	\chi(Y) = d \, \chi(X)  - \deg(\mathscr{R}_{\varpi}). 
	\end{equation*}
	This form of the Riemann--Hurwitz formula holds true for a general branched covering and is obviously equivalent to the previous form when $\varpi$ is Galois.
\end{remark}

Now, let the orbifold Riemann surface $O = (X_O , \brdiv)$ be a developable Riemann orbisurface with signature $(g;m_1,\dots,m_{n_e};n_p)$. Any such orbifold Riemann surface is finitely covered by a Riemann surface $Y$ such that $\brdiv$ is the branch divisor of the branched Galois covering $\varpi: Y \to X_O$. Then, it immediately follows from the proposition~\ref{prop:branchedorbifoldcorrepondence} that there exists a corresponding orbifold Galois covering $\varpi_{\text{orb}}: Y \xrightarrow{d:1} O = (X_O , \brdiv)$ and it is natural to define the orbifold Euler characteristic of $O$ by using the equation $\chi(Y) = d \, \chi(O)$. Hence, we get
\begin{equation}\label{orbEulerCh}
\chi(O) = \frac{1}{d} \, \chi(Y) = \chi(X_O)  - \deg(\brdiv) = 2-2g - n_p - \sum_{i=1}^{n_e} \left(1-\frac{1}{m_i}\right). 
\end{equation}
In the following subsection, we will derive this relation in an equivalent way using the notion of \emph{orbifold canonical divisor}.

Let us end this subsection by making the theorem~\ref{thm:universalcovering} a little sharper. It immediately follows from equation $\chi(Y) = d \, \chi(O)$ that the Euler characteristic of a developable Riemann orbisurface $O$ should have the same sign as the Euler characteristic of its universal covering. Hence, we have the following corollary:
\begin{corollary}
	Let $O$ be a closed (or possibly punctured) orbifold Riemann surface, which is not a teardrop or a spindle. Then, $O$ admits $\UHP$, $\cmpx$, or $\hat{\cmpx}$ as its universal covering if and only if $\chi(O)<0$, $\chi(O)=0$, or $\chi(O)>0$ respectively.
\end{corollary}
\noindent The orbifold Riemann surfaces with $\chi(O)<0$ are called \emph{hyperbolic} and are the main focus of our study in the main body of this paper. Notice that all hyperbolic Riemann surfaces are, by definition, developable.
\subsection{Orbisheaves, Orbibundles, and Orbidivisors}
The notions of bundle theory, and more generally, sheaf theory, are fundamental to doing geometry on any object. Fortunately, these notions and many other usual differential geometric concepts can be generalized to the orbifold case with the help of orbifold maps (see, e.g., \cite[\textsection4.2]{Boyer_2008}).
\subsubsection*{V-bundles}
In this section, we define holomorphic vector V-bundles (or orbibundles) as a reasonable generalization of holomorphic vector bundles over complex manifolds. Remember that we defined a rank $r$ holomorphic vector bundle $E$ over an analytic space $X$ as an analytic map $\pi: E \to X$ such that $\pi$ is locally a projection $V \times \cmpx^r \to V$. Similarly, a holomorphic vector V-bundle of rank $r$ over an orbifold $O=(X,\mathcal{U})$ should be thought of as a pair $\big(\mathcal{E}=(E,\mathcal{V}) , \, \pi_{\text{orb}}: \mathcal{E} \to O \big)$ where $\mathcal{E}$ is a complex orbifold and $\pi_{\text{orb}}$ is an analytic orbifold map. Thus, starting from an analytic map $\pi: E \to X$ between the underlying analytic spaces, our task reduces to the \emph{construction of the appropriate local lifts of  $\pi$} (as in definition~\ref{def:orbmap}):
\begin{definition}[Holomorphic vector V-bundle]\label{def:Vbundle}
	Let $O = (X, \mathcal{U})$ be a complex orbifold. A \emph{holomorphic vector V-bundle} (or a \emph{holomorphic vector orbibundle}) of rank $r$ over $O$ is a collection of holomorphic vector bundles $\tilde{\pi}_a: \tilde{E}_a \to \tilde{U}_a$ with fiber $\cmpx^r$ for each orbifold chart $(U_a, \tilde{U}_a, \Gamma_a, \ff_a)$ of $O$, together with a collection of group homomorphisms $\bar{\pi}_a: \Gamma_a \to \bar{\pi}_a(\Gamma_a)$ defining an action of $\Gamma_a$ on $\tilde{E}_a$ by (ordinary) holomorphic bundle maps, such that:
	\begin{enumerate}[(i)]
		\item Each $\tilde{\pi}_a$ is $\Gamma_a$-equivariant, so that the following diagram is commutative for any $\gamma \in \Gamma_a$:
		\begin{equation*}
		\begin{tikzcd}
		\tilde{E}_a \ar[d,"\tilde{\pi}_a"] \ar[r, "\bar{\pi}_a(\gamma)"] &  \tilde{E}_a  \ar[d,"\tilde{\pi}_a"] \\
		\tilde{U}_a \ar[r,"\gamma"] & \tilde{U}_a 
		\end{tikzcd}
		\end{equation*}
		\item For any holomorphic embedding $\eta_{ab}: \tilde{U}_a \hookrightarrow \tilde{U}_b$ of charts on $O$, there exists a holomorphic bundle isomorphism $\widehat{\eta}_{ab}: \tilde{E}_a \to \tilde{E}_b\big|_{\eta_{ab}(\tilde{U}_a )} := \tilde{\pi}_a^{-1}\big(\eta_{ab}(\tilde{U}_a )\big)$, such that $\widehat{\eta}_{ab}$ is $\bar{\pi}_a$-equivariant.
		\item For two embeddings $\eta_{ab}: \tilde{U}_a \hookrightarrow \tilde{U}_b$ and $\eta_{bc}: \tilde{U}_b \hookrightarrow \tilde{U}_c$, we have $\widehat{\eta_{ab} \circ \eta_{bc}} = \widehat{\eta}_{ab} \circ \widehat{\eta}_{bc}$.
	\end{enumerate}
\end{definition}
\begin{remark}
	The total (underlying) analytic space $E$ of an orbibundle is obtained from the local bundles $\tilde{E}_a$ in the following way: Choosing small enough orbifold charts on $O$, there always exists a local trivialization $\tilde{E}_a \cong \tilde{U}_a \times \cmpx^r$ such that $\tilde{\pi}_a: \tilde{U}_a \times \cmpx^r \to \tilde{U}_a$ is a holomorphic projection on the first factor and the action of $\bar{\pi}_a(\Gamma_a)$ on $\tilde{U}_a \times \cmpx^r$ is diagonalized --- i.e. for any pair $(\tilde{x}, v) \in \tilde{U}_a \times \cmpx^r$ and any $\gamma \in \Gamma_i$, we have $\bar{\pi}_a(\gamma) \cdot (\tilde{x}, v) := (\gamma \cdot \tilde{x}, \widehat{\Upsilon}_a(\gamma) \cdot v)$ where $\widehat{\Upsilon}_a: \Gamma_a \to \operatorname{GL}(r, \cmpx)$ is a monomorphism. Then, we have a branched Galois covering $\ff^{b}_a: \tilde{E}_a \to E_a := \tilde{E}_a/\bar{\pi}_a(\Gamma_a)$. As a result, since $\tilde{\pi}_a$ is $\Gamma_a$-equivariant, we get a unique \emph{analytic projection map} $\pi_a: E_a \to U_a$ such that the following diagram commutes:
	\begin{equation*}
	\begin{tikzcd}
	\tilde{E}_a \ar[d,"\tilde{\pi}_a"] \ar[r, "\ff^{b}_a"] &  E_a  \ar[d,"\pi_a"] \\
	\tilde{U}_a \ar[r,"\ff_a"] & U_a 
	\end{tikzcd}
	\end{equation*}
	Now, we can glue the analytic varieties $E_a$ in the following way, stemming from the gluing condition on $X$: Let $(U_a, \tilde{U}_a, \Gamma_a ,\ff_a)$ and $(U_b, \tilde{U}_b, \Gamma_b ,\ff_b)$ be any two orbifold charts in $\mathcal{U}$ with $U_a \cap U_b \neq \emptyset$ and let $x \in U_a \cap U_b$ be a point. Then, according to definition~\ref{def:orbifold}, there always exists another orbifold chart $(U_c \subset U_a \cap U_b, \tilde{U}_c, \Gamma_c ,\ff_c) \in \mathcal{U}$ containing $x$ such that embeddings $\eta_{ca}: \tilde{U}_c \hookrightarrow \tilde{U}_a$ and  $\eta_{cb}: \tilde{U}_c \hookrightarrow \tilde{U}_b$ induce bundle biholomorphisms $\widehat{\eta}_{ca}: \tilde{E}_c \to \tilde{E}_a\big|_{\eta_{ca}(\tilde{U}_c )}$ and $\widehat{\eta}_{cb}: \tilde{E}_c \to \tilde{E}_b\big|_{\eta_{cb}(\tilde{U}_b )}$. Gluing $E_a$ and $E_b$ according to this data results in a complex orbifold $\mathcal{E}$ with an underlying  analytic space $E$ and an analytic  
	orbifold map $\pi_{\text{orb}}: \mathcal{E} \to O$, which is determined by the analytic map $\pi: E \to X$ (obtained by gluing analytic projection maps $\pi_a = \pi|_{E_a}$) and local lifts $\tilde{\pi}_a: \tilde{E}_a \to \tilde{U_a}$.
\end{remark}
The next concept needed to defined is holomorphic sections of holomorphic vector V-bundles. Defining this is easy globally : a section of $\mathcal{E}$ in orbibudle $\pi_{\text{orb}}: \mathcal{E} \to O$ is a holomorphic orbifold map $s: O \to \mathcal{E}$ satisfying  $\pi_{\text{orb}} \circ s = \operatorname{id}_{O}$. Locally, this concept can be defined as follows:
\begin{definition}[Sections of V-bundles]\label{def:secVbundle}
	If we consider the holomorphic vector V-bundle $\pi_{\text{orb}}: \mathcal{E} \to O$, a \emph{holomorphic section} of $\mathcal{E}$ can be defined in either of the following two equivalent ways:
	\begin{enumerate}
		\item $s: O \to \mathcal{E}$ is a holomorphic orbifold map satisfying $\pi_{\text{orb}} \circ s = \operatorname{id}_{O}$.
		\item A collection of \emph{$\Gamma_a$-equivariant} holomorphic sections $s_a: \tilde{U}_a \to \tilde{E}_a$ such that for any embedding $\eta_{ab}: \tilde{U}_a \hookrightarrow \tilde{U}_b$ the following diagram commutes:
		\begin{equation*}
		\begin{tikzcd}
		\tilde{E}_a \ar[r, "\widehat{\eta}_{ab}"] &  \tilde{E}_b\big|_{\eta_{ab}(\tilde{U}_a)} \\
		\tilde{U}_a \ar[u,"s_a"] \ar[r, "\eta_{ab}"] & \eta_{ab}(\tilde{U}_a) \ar[u,"s_a|_{\eta_{ab}(\tilde{U}_a)}"]
		\end{tikzcd}
		\end{equation*}
	\end{enumerate}
\end{definition}
	To glue the local sections $s_i$ to the global section $s: O \to \mathcal{E}$ one should demand the Equivariance of the local sections. We will call a local holomorphic section $s_a: \tilde{U}_a \to \tilde{E}_a$ \emph{$\Gamma_a$-invariant} (as opposed to \emph{$\Gamma_a$-equivariant}) if $\gamma \circ s_a = s_a$. Given the local holomorphic sections $s_a: \tilde{U}_a \to \tilde{E}_a$ of a holomorphic vector V-bundle $\mathcal{E}$, we can always construct $\Gamma_a$-invariant local sections $s_a^{\Gamma_a}$ by ``averaging over the group'' --- i.e. we define an invariant local section by
	\begin{equation}\label{invariantsection}
	s_a^{\Gamma_a} = \frac{1}{\#\Gamma_a} \sum_{\gamma \in \Gamma_a} s_a \circ \gamma.
	\end{equation}
	Notice that this determines a well-defined map from the underlying analytic space $X_O$, namely $s_{_{U_a}} := (\ff_a)_{\ast} \big( s_a^{\Gamma_a}\big): U_a \to \tilde{E}_a$. Gluing these invariant local sections over each orbifold chart, we obtain \emph{global invariant sections} and view them \emph{interchangeably} as invariant objects on $\tilde{U}_a$s or as objects on $U$s. However, note that smoothness in the orbifold sense is somewhat different from ordinary smoothness, and holomorphic invariant sections can have singular behavior (although usually in a controlled way) when viewed as objects on the open analytic subsets $U_a$. 
\begin{remark}\label{rmk:trivialVbunle}
	Consider the easiest (but still important) example of a \emph{trivial holomorphic line V-bundle}: This line V-bundle is given by trivial holomorphic line bundles $\tilde{E}_a \cong \tilde{U}_a \times \cmpx$ on each local uniformizing neighborhood $\tilde{U}_a$ together with a trivial action of $\Gamma_a$ on the second factor --- i.e. $\widehat{\Upsilon}_a(\Gamma_a) = 1\in \operatorname{GL}(1,\cmpx)$. Then, clearly $E_a \cong U_a \times \cmpx$ and the total space $\mathcal{E}$ is just $O \times \cmpx$. Holomorphic sections of this bundle clearly are in a one-to-one-correspondence with analytic orbifold maps from $O$ to $\cmpx$ endowed with the trivial orbifold structure. So, according to remark~\ref{rmk:orbifunctions}, they seem to be a good candidate for a \emph{structure orbisheaf} on $O$; however, in order to get coherent sheaves on the underlying space $X_O$, nonetheless, we have to deal with \emph{invariant sections} of the trivial holomorphic line bundles $\tilde{U}_a \times \cmpx \to \tilde{U}_a$ or sheaves $\tilde{\mathcal{F}}_a$ on the local uniformizing neighborhoods $\tilde{U}_a$ (see remark~\ref{rmk:invarinatsheaf}).
\end{remark}
All of the standard notions of the tangent bundle, cotangent bundle, and the different associated tensor bundles have V-bundle analogues: On every local uniformizing neighborhood $\tilde{U}_a$, take the holomorphic tangent bundle $T^{1,0}\tilde{U} \cong \tilde{U} \times \cmpx^\mathfrak{n}$ and for any change of charts $\eta_{ab}: \tilde{U}_a \hookrightarrow \tilde{U}_b$ on $O$, construct the corresponding bundle biholomorphism $\widehat{\eta}_{ab}: T^{1,0}\tilde{U}_a \to T^{1,0}\tilde{U}_b\big|_{\eta_{ab}(\tilde{U}_a)}$ by defining it to be given by $\eta_{ab}$ on the first factor and the Jacobian $\operatorname{Jac}[\eta_{ab}]$ on the second one. If we denote by $(\partial_1^{(a)}, \dots, \partial_\mathfrak{n}^{(a)})$ the local coordinate basis on each $T^{1,0}\tilde{U}_a$, the Jacobian matrix $\operatorname{Jac}[\eta_{ab}] \in \operatorname{GL}(\mathfrak{n},\cmpx)$ is defined as
\begin{equation*}
\Big(\operatorname{Jac}[\eta_{ab}]\Big)_{k,l} = \frac{\partial_k^{(b)} \circ \eta_{ab}}{\partial_l^{(a)}}.
\end{equation*} 
\begin{remark}
	Locally, around any point $x \in X_O$, the fiber $\big(\pi_{\text{orb}}^{1,0}\big)^{-1}(x) \subset T^{1,0}O$ is not biholomorphic to $\cmpx^\mathfrak{n}$, but is biholomorphic to a small neighbourhood of $x \in X_O$ --- i.e. in general $\big(\pi_{\text{orb}}^{1,0}\big)^{-1}(x)  \cong \cmpx^\mathfrak{n}/\Gamma_{x}$. This is because, in a local chart, the actions of $\gamma \in \Gamma_a$ on $\tilde{U}_a$ and of $\operatorname{Jac}(\gamma)$ on $T_{\ff_a^{-1}(x)}^{1,0}\tilde{U}_a$ are essentially the same. On the other hand, the underlying analytic space of $T^{1,0}O$ is not necessarily the ordinary tangent space $T^{1,0}X_O$.
\end{remark}
The above construction obviously generalizes to the anti-holomorphic tangent V-bundle $T^{0,1}O$, holomorphic and anti-holomorphic cotangent V-bundles, symmetric and antisymmetric tensor V-bundles of type $(k,l)$, etc. Particularly, if $O$ has dimension $\mathfrak{n}$, we denote the highest exterior power of the holomorphic cotangent V-bundle, $\bigwedge^{\mathfrak{n}} T^{\ast}_{1,0}O$, by $K_O$ and call it the \emph{orbifold canonical bundle}. Additionally, one can easily generalize notions such as Riemannian and Hermitian metrics, orbifold $(p,q)$-differential forms, Hermitian and Chern connections, Chern forms, etc., to the orbifold setting. We will come back to these notions in the next subsections.
\begin{remark}
	Notice that all of the above definitions simplify for the case of developable orbifolds $O \cong [M/\Gamma]$. In this case, we can always view objects defined on $O$ --- such as tensors, differential forms, connections, etc. --- as globally defined ordinary objects defined on $M$ that are invariant under the action of $\Gamma$. We will come back to this point in the later subsections when we study differential forms and metrics on hyperbolic Riemann orbisurfaces in greater detail.
\end{remark}
Now, consider a Weil divisor $\mathcal{D}$ on the underlying analytic space $X_O$. We can lift its restriction $\mathcal{D} \cap U_a$ to a divisor $\tilde{D}_{\tilde{U}_a}$ on the local uniformizing neighborhood $\tilde{U}_a$ by $\tilde{D}_{\tilde{U}_a} := \ff_a^{-1}(\mathcal{D} \cap U_a)$. The collection of all such divisors $\tilde{D}_{\tilde{U}_a}$ on each $\tilde{U}_a$ defines an \emph{orbidivisor} (or a \emph{Baily divisor}) on the orbifold $O$ (see e.g. \cite[Def.~4.4.11]{Boyer_2008} for more details). In fact, we have the following proposition:
\begin{proposition}
	The branch divisor $\brdiv$ or more generally any $\mathbb{Q}$-divisor on $X_O$ of the form $\sum_i \frac{b_i}{m_i} \mathcal{D}_i$, where $m_i$ is a ramification index and $b_i \in \mathbb{Z}$, lifts to an orbidivisor on $O = (X_O, \mathcal{U})$.
\end{proposition}
An divisor obtained as the lift of a branch divisor is called a \emph{ramification divisor}. The following is straightforward
\begin{proposition}
	To each Baily divisor $\mathcal{D}$ on the orbifold $O$, there corresponds a complex line V-bundle $\linebundle(\mathcal{D})$.
\end{proposition}
The most important Baily divisor on a complex orbifold $O$ is the \emph{orbifold canonical divisor} $\mathcal{K}_O$ which is any Baily divisor associated to the canonical orbibundle $K_O$. In the presence of a branch divisor $\brdiv$, an orbifold canonical divisor
$\mathcal{K}_O$ is not the same (meaning not linearly equivalent) as the canonical divisor $\mathcal{K}_X$ of the underlying analytic space $X_O$. In fact we have (for proof see \cite[Prop.~4.4.15]{Boyer_2008})
\begin{proposition}
	The orbifold canonical divisor $\mathcal{K}_O$ and canonical divisor $\mathcal{K}_X$  of its underlying analytic space are related by
	\begin{equation*}
	\mathcal{K}_O \cap U_i \equiv \ff_i^{\ast}(\mathcal{K}_X \cap U_i) + \sum_j (1-\tfrac{1}{m_j}) \ff_i^{\ast}(\mathcal{D}_j \cap U_i).
	\end{equation*}
	In terms of the orbifold rational Chern class, the above equation implies
	\begin{equation*}
	c_1(O) = - c_1(\mathcal{K}_O) = \underbrace{- c_1(\mathcal{K}_X)}_{c_1(X)} - \sum_j (1-\tfrac{1}{m_j}) c_1\big(\linebundle(\mathcal{D}_j)\big) \in H^2(X,\mathbb{Q}).
	\end{equation*}
\end{proposition}
Let $O = (X, \brdiv)$ be a good orbifold Riemann surface. An orbifold canonical divisor is given by
\begin{equation}
\mathcal{K}_O  = \ff^{\ast} \mathcal{K}_X + \brdiv,
\end{equation}
where $\mathcal{K}_X$ is an ordinary canonical divisor on the underlying Riemann surface $X$. Thus, if $g$ denotes the genus of $X$, the orbifold Chern number (obtained by integrating the first Chern character over $X$) is
\begin{equation}
c_1(O) = - \deg(\mathcal{K}_O)  = -\deg(\mathcal{K}_X) - \deg(\brdiv) = 2-2g-n_p - \sum_{i=1}^{n_e} (1-\tfrac{1}{m_i}),
\end{equation}
which equals the orbifold Euler characteristic $\chi(O)$ defined before (this follows from the equivalence of the top Chern class with the Euler class).

\subsubsection*{Orbisheaves}
We first introduce the notion of an orbisheaf following \cite[Def.~4.2.1]{Boyer_2008}. Similar to V-bundles, \emph{Orbifold sheaves} or \emph{orbisheaves} consist of a sequence of sheaves defined on the disjoint union $\bigsqcup_{a} \tilde{U}_a$ of the local uniformizing neighborhoods that satisfy certain compatibility conditions with respect to the local uniformizing groups and embeddings:
\begin{definition}[Orbisheaf]
	Let $O=(X_O,\mathcal{U})$ be a complex orbifold. An \emph{orbisheaf} $\mathcal{F}_O$ on $O$ consists of a collection of sheaves $\{\tilde{\mathcal{F}}_a\}_{a \in A}$ defined over each local uniformizing neighborhood $\tilde{U}_a$ of $O$, such that for each embedding $\eta_{ab}: \tilde{U}_a \hookrightarrow \tilde{U}_b$ there exists an isomorphism of sheaves $\eta^{\ast}_{ab}: \tilde{\mathcal{F}}_a \to (\eta_{ab})^{\ast}\big(\tilde{\mathcal{F}}_b\big)$, which is functorial.
\end{definition}
Let $\mathcal{F}_O$ be an orbisheaf on $O$, and $(U_a, \tilde{U}_a, \Gamma_a, \ff_a)$ an orbifold chart. Then, one can define an action of $\Gamma_a$ on the sheaf $\tilde{\mathcal{F}}_a$, which says that $\tilde{\mathcal{F}}_a$ is a \emph{$\Gamma_a$-equivariant} sheaf on $\tilde{U}_a$. So, every orbisheaf $\mathcal{F}_O$ is equivariant under the local uniformizing groups $\Gamma_a$. We now have
\begin{definition}[Structure orbisheaf]\label{def:structureorbisheaf}
	The \emph{structure orbisheaf} $\mathscr{O}_O$ of an orbifold $O$ is the orbisheaf defined by the collection of structure sheaves $\mathscr{O}_{\tilde{U}_a}$ defined on each local uniformizing neighborhood $\tilde{U}_a$. The structure orbisheaf $\mathscr{O}_O$ is well-defined since each embedding $\eta_{ab}: \tilde{U}_a \hookrightarrow \tilde{U}_b$ induces an isomorphism $\mathscr{O}_{\tilde{U}_a} \approx (\eta_{ab})^{\ast}(\mathscr{O}_{\tilde{U}_b})$ by sending $f \in \mathscr{O}_{\tilde{U}_a, \tilde{x}}$ to $f \circ \eta_{ab}^{-1} \in (\eta_{ab})^{\ast}(\mathscr{O}_{\tilde{U}_b})$.
\end{definition}
	This definition evidently does not align with the holomorphic sections of the trivial line V-bundle nor it yields a sheaf on the underlying space $X_O$. Therefore, we need to utilize local $\Gamma_a$-invariant sections (in contrast to $\Gamma_a$-equivariant) of such sheaves and then assemble them across $X_O$ \cite[Lemma~4.2.4]{Boyer_2008}. Accordingly, $\mathcal{F}_X$ are defined as sheaves on $X_O$ which are invariant local sections of orbisheaves $\mathcal{F}_O$. In this regard, $H^{0}(\tilde{U}_a, \mathscr{O}_O)^{\Gamma_a} \simeq H^0(U_a, \mathscr{O}_X)$ holds for the structure sheaves. For a coherent orbisheaf $\mathcal{F}_O$ of $\mathscr{O}_O$-modules, the $\Gamma_a$-invariant sections are coherent sheafs of $\mathscr{O}_X$-modules. Interestingly, this helps one to construct the orbisheaf cohomology. But one should note that this cohomology only probes the topology of the underlying analytic space $X_O$. Hence, a more complicated notion of cohomology, the so called \emph{Chen--Ruan cohomology}, of an orbifold is needed to probe the full topological features of an orbifold. See e.g., \cite{Chen:2000cy} and \cite{adem2007orbifolds} for more details.

There are several important orbisheaves on complex orbifolds that we shall work with: First, there is the structure orbisheaf $\mathscr{O}_O$ defined in \ref{def:structureorbisheaf}, where each $\mathscr{O}_{\tilde{U}_a}$ is the sheaf of holomorphic functions on $\tilde{U}_a$. Similarly, there is the meromorphic orbisheaf $\mathscr{M}_{O}$ consisting of meromorphic functions on each local uniformizing neighborhood $\tilde{U}_a$. Finally, there is the \emph{canonical orbisheaf} of a complex orbifold: On a complex orbifold $O$ of complex dimension $\mathfrak{n}$, we denote by $\Omega^{k}_O$ the orbisheaf of \emph{holomorphic differential $k$-forms} on $O$. This is the orbisheaf constructed from the collection of ordinary canonical sheaves $\Omega^k_{\tilde{U}_a}$ on each orbifold uniformizing neighborhood $\tilde{U}_a$. $\Omega^{k}_O$  is a locally free orbisheaf of rank ${\mathfrak{n} \choose k}$. Then,
\begin{definition}[Canonical orbisheaf]
	The \emph{canonical orbisheaf} of a complex orbifold $O$ of complex dimension $\mathfrak{n}$ is the orbisheaf $\Omega^{\mathfrak{n}}_O$.
\end{definition}
\subsection{Orbifold Metrics}\label{sect:orbifoldmetric}
In this section, we delve into the examination of metrics on orbifolds. It is clear that, for each $\tilde{U}_i$ on $(X_O, \mathcal{U})$, this metric should be defined as a $\Gamma_i$-invariant metric:
\begin{definition}[
	Hermitian orbifold metrics] 
	A Hermitian metric, $\mathsf{h}$, on a complex orbifold $O=(X_O,\mathcal{U})$ can be characterized as a family of $\Gamma_a$-invariant (local) Hermitian metrics $\tilde{\mathsf{h}}_a^{\Gamma_a}$ defined on each neighborhood $\tilde{U}_a$ such that the change of charts are Hermitian isometries. A complex orbifold with a Hermitian metric is called a \emph{Hermitian orbifold}.
\end{definition}
\begin{remark}
	A slight modification of the usual partition of unity arguments assures us that 
	\emph{every complex orbifold} admits a \emph{Hermitian metric} (see \cite{moerdijk_mrcun_2003} for more details).
\end{remark}
\begin{remark}
	There is a beautiful connection between the preceding discussion and the geometry of the situation which is provided by the \emph{Gauss-Bonnet theorem}: Consider a good compact orbifold Riemann surface $O$ which is expressed as $[\tilde{X}/\Gamma]$ where $\tilde{X} \in \{\hat{\cmpx},\cmpx,\UHP\}$ and $\Gamma < \isom^{+}(\tilde{X})$ is discrete group. There exists a canonical constant curvature Hermitian metric on $\tilde{X}$ which induces a Hermitian metric of constant curvature on $O$. $O$ has a well-defined area $A(O)$ which has the same naturality property under finite coverings as the Euler number, i.e. if $\tilde{O}$ is an orbifold covering of $O$ of degree $d$, then $A(\tilde{O}) = d \cdot A(O)$. Hence, we can the use the fact that $O$ is finitely covered by some Riemann surface and apply the usual Gauss-Bonnet Theorem to this Riemann surface. In particular, if $O$ is $[\hat{\cmpx}/\Gamma]$, we deduce that $A(O) = 2 \pi \chi(O)$ and if $O$ is $[\UHP/\Gamma]$, we deduce that $A(O) = - 2 \pi \chi(O)$.
\end{remark}
More generally, one can define Hermitian metric on every holomorphic V-bundle:
\begin{definition}[Hermitian metrics on holomorphic V-bundles]
	Let $O= (X_O,\mathcal{U})$ be a complex orbifold and $\pi_{\text{orb}}: \mathcal{E} \to O$ a holomorphic vector V-bundle. A \emph{Hermitian orbifold metric} $\mathsf{h}$ on $\mathcal{E}$ is a collection of local $\Gamma_a$-invariant Hermitian metrics $\tilde{h}_a^{\Gamma_a}$ on each local holomorphic vector bundle $\tilde{E}_a \to \tilde{U}_a$, such that all embeddings are \emph{Hermitian isometries}.
\end{definition}

Finally, let $\mathcal{E} \to O$ be a holomorphic vector V-bundle endowed with a Hermitian metric $\mathsf{h}$. A \emph{Hermitian connection} $\nabla$ on $\mathcal{E}$ is defined to be a collection $\{\nabla_a\}$ of $\Gamma_a$-equivariant Hermitian connections supported on each local uniformizing neighborhood $\tilde{U}_a$ such that $\nabla_a$s are compatible with changes of charts. Then, the first Chern class or degree of a V-bundle can be defined using Chern-Weil theory; notice that the degree of a V-bundle is a rational number. Sobolev spaces and Hodge theory for V-bundles follow in the same way.


%
Let $O=(X_O,\brdiv)$ be an orbifold Riemann surface. We say that a hermitian metric of class $\mathcal{\infty}$ on the underlying Riemann surface $X_O$ is compatible with the branch divisor $\brdiv = \sum (1-1/m_i) x_i$ if in a holomorphic local coordinate system centered at $x_i$ the metric is of the form $(\rho(u)/|u|^{2 - 2/m_i}) |du|^2$ for $m_i \neq \infty$, whereas it is of the form $\rho(u)/|u|^2 \log^2(|u|^{-2}) ) |du|^2$ if $m_i=\infty$. Here, $\rho$ is continuous at the marked points and positive. The cone angle is $2\pi/m_i$, including the complete case with angle zero. Let $K_X$ be the canonical divisor of $X_O$; the orbifold Riemann surface is called stable, if the degree of the divisor $K_X + \brdiv$ is positive. In this case, by a result of McOwen \cite{McOwen-1988} and Troyanov \cite{Troyanov_1986,Troyanov1991PrescribingCO}, there exists a unique conical metric $ds^2_{\text{hyp}}(\brdiv)$ on $X_O$ in the given conformal class, which has constant curvature $-1$ and prescribed cone angles. Moreover $\operatorname{Vol}(X_O, ds^2_{\text{hyp}}(\brdiv))/\pi = \deg(K_X + \brdiv) = -\chi(O)$. Where by definition $\chi(O) = \chi(X_O) - \deg(\brdiv)$ is the Euler characteristic of the Riemann orbisurface $O=(X_O,\brdiv)$.
\subsubsection{Hyperbolic metric on Riemann Orbisurfaces}
The Poincar\'{e} metric on $\UHP$
\begin{equation}
\dd{s_{\text{hyp}}^2} = \frac{|\dd{z}|^2}{(\Im z)^2},
\end{equation}
is the \emph{unique} (up to multiplicative constant) Riemannian metric that is invariant under $\PSLR$, and descends to a Riemannian metric on $O =  [\UHP\slash \Gamma]$. As a metric on $O$, it has singularities at the elliptic and parabolic fixed points. One can describe the local geometry of a hyperbolic cusp and a hyperbolic cone using a distinguished holomorphic coordinate $w$ (called \emph{rotationally symmetric} ($rs$) by Wolpert \cite{Wolpert1990TheHM,Wolpert2005CuspsAT}) that is unique up to a constant of modulus 1:
\begin{itemize}
	\item \emph{The model cusp:} Let $\mathcal{C}_{\infty}$ denote an infinite (non-compact) cusp.\footnote{See Lemma 2.1.1 of \cite{Obitsu-2008}.} A fundamental domain for $\mathcal{C}_{\infty}$ in the $\UHP$ is given by the set $\left\{z \in \UHP \, \Big| \, 0 \leq \Re z \leq 1 \right\}$ and by identifying the boundary points $\Re z =0$ with $\Re z =1$ -- i.e. $\mathcal{C}_{\infty} \cong \mathbb{S}^1 \times \mathbb{R}^{+}$. The isotropy group that corresponds to the above fundamental domain consists of $\mathbb{Z}$ acting by addition. Let $\mathcal{C}_{\infty,\epsilon}$, the hyperbolic cusp with apex at infinity and horocycle at height $\epsilon$, denote the submanifold of $\mathcal{C}_{\infty}$ obtained by restricting the previous fundamental domain to $\Im z > \epsilon$ -- i.e. $\mathcal{C}_{\infty, \epsilon} \cong \mathbb{S}^1 \times [\epsilon, \infty)$. This fundamental domain can be endowed with the Poincar\'{e} metric 
	\begin{equation}
	\dd{s_{\text{cusp}}^2} = \frac{|\dd{z}|^2}{(\Im z)^2}.
	\end{equation}
	Observe that this is a \emph{complete} metric of Gaussian curvature $-1$ and finite volume, $\operatorname{Vol}(\mathcal{C}_{\infty, \epsilon}) = \epsilon$.
	
	The hyperbolic cusp $\mathcal{C}_{\infty, \epsilon}$ can equivalently be presented as a Riemann surface with boundary, parameterized by the complex coordinate $w \equiv e^{i 2 \pi z}$, valued in the punctured disk $\UD^{\ast}(0,e^{-2 \pi \epsilon})$. The hyperbolic metric can then be written as
	\begin{equation}\label{cuspcoordinate}
	\dd{s_{\text{cusp}}^2} = \frac{|\dd{w}|^2}{\left(|w| \log|w|\right)^2}.
	\end{equation}
	The coordinate $w$ is uniquely determined by this condition, up to a factor of modulus 1. Following \cite{Wolpert1990TheHM,Wolpert2005CuspsAT}, we will call $w$ an $rs$-coordinate in a neighborhood of a parabolic fixed-point (see Figure \ref{fig:localcoordinate}).  
	
	\item \emph{The model cone:} 
	For a given positive integer $m$, let $\mathcal{C}_{m}$ denote the infinite hyperbolic cone of angle $2\pi/m$.\footnote{See section 2 of \cite{Judge_1995}.} One can realize $\mathcal{C}_{m}$ as a half-infinite cylinder $\mathbb{S}^1 \times \mathbb{R}^{+}$, equipped with the constant curvature $-1$ metric
	\begin{equation}
	\dd{s_{\text{cone}}^2} = \left(\frac{2 \pi}{m}\right)^2\frac{|\dd{z}|^2}{\sinh[2](\frac{2 \pi}{m} \Im z)}.
	\end{equation}
	In contrast to the cusp case, this metric is \emph{not} complete. A suitable change of variables provides a parameterization of the hyperbolic cone by $\mathcal{C}_{m} \cong (0,\infty) \times (0,2\pi]$ with coordinates $(\rho,\theta)$. The metric in this coordinate becomes 
	\begin{equation}
	\dd{s_{\text{cone}}^2} = \dd{\rho^2} + m^{-2} \sinh[2](\rho) \dd{\theta^2},
	\end{equation}
	having volume form\footnote{The $\star$ in Eq.~\eqref{volumeform} is the Hodge star, and the notation $\star 1$ emphasizes that the volume form is the Hodge dual of the constant map on the manifold.}
	\begin{equation}\label{volumeform}
	\star \hspace*{-1pt}1 = m^{-1} \sinh(\rho) \dd{\rho} \wedge \dd{\theta}.
	\end{equation}
	A fundamental domain for $\mathcal{C}_{m}$ in the hyperbolic unit disk $\UD$ is provided by a sector with a vertex at the origin and with angle $2 \pi/m$ -- i.e. $\left\{\mathsf{u}_{\UD} \in \UD \, \Big| \, 0 \leq \arg(\mathsf{u}_{\UD}) < \frac{2 \pi}{m}\right\}$. The hyperbolic metric on $\mathcal{C}_{m}$ is the metric induced onto the fundamental domain (viewed as a subset of the $\UD$ endowed with its complete hyperbolic metric). The isotropy group which corresponds to this fundamental domain is the group $\mathbb{Z}_{m}$ consisting of the numbers $\exp(\sqrt{-1}\, 2 \pi j/m)$ for $j=1,2,\dots,m$ acting by multiplication. As before, let the hyperbolic cone of angle $2\pi/m$ and a boundary at height $\epsilon$, $\mathcal{C}_{m,\epsilon} \cong \mathbb{S}^1 \times [\epsilon, \infty)$, be the submanifold of $\mathcal{C}_{m}$ obtained by restricting the $(\rho,\theta)$-coordinate to $0 \leq \rho < \cosh[-1](1+\epsilon m/ 2\pi)$. A fundamental domain for $\mathcal{C}_{m,\epsilon}$ in the unit disk model is obtained by adding the restriction that $|u| < \sqrt{\epsilon m / (4 \pi + \epsilon m)}$. An elementary calculation shows that the volume of this manifold is finite and is given by $\operatorname{vol}(\mathcal{C}_{m,\epsilon}) = \epsilon$. 
	
	Finally, the hyperbolic cone can also be seen as a Riemann surface with boundary, parameterized by the complex coordinate $\tilde{w} \in \UD^{\ast}(0,R)$, such that
	\begin{equation}\label{conersmetric}
	\dd{s_{\text{cone}}^2} = \frac{4 |\dd{\tilde{w}}|^2}{m^2 |\tilde{w}|^{2-2/m} \left(1-|\tilde{w}|^{2/m}\right)^2}.
	\end{equation} 
	As for the case of the cusp, a coordinate $\tilde{w}$ with this property is unique up to a factor of modulus 1 and was also called \cite{Montplet2016RiemannRochII} an $rs$-coordinate after Wolpert \cite{Wolpert1990TheHM,Wolpert2005CuspsAT} (see Figure \ref{fig:localcoordinate}). The parameter $R$ can be easily obtained by computing and comparing volumes in different coordinates. In particular, as $\epsilon m \to 0$, we have $R \sim (\epsilon m / 4\pi)^{m/2}$.
\end{itemize}
\begin{figure}
	\centering
	\includegraphics[width=0.9\linewidth]{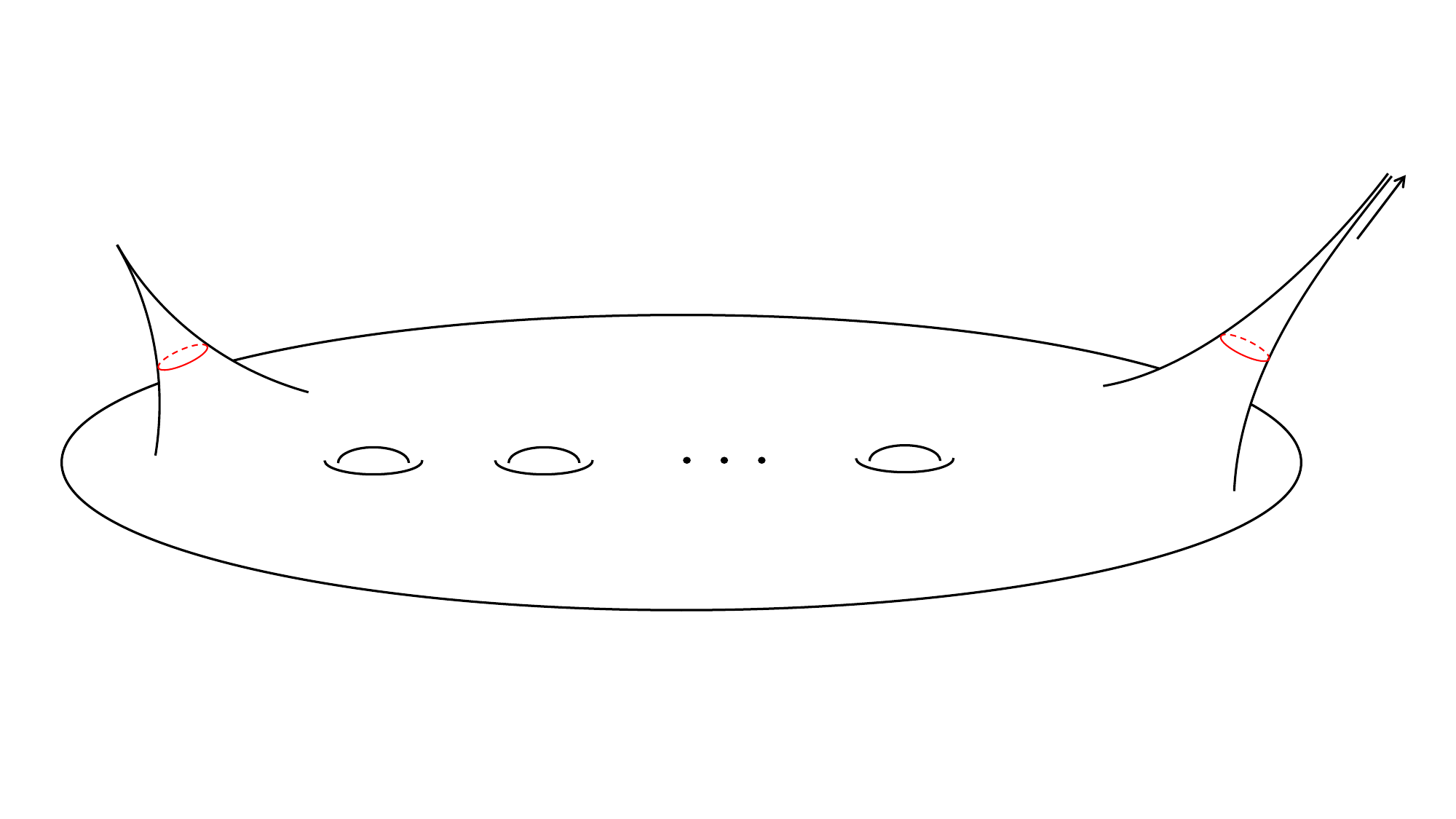}
	\put(-360,160){\rotatebox{0}{\fontsize{12}{12}\selectfont{} $\dd{s_{cone}^2} = \frac{4 |\dd{\tilde{w}}|^2}{m^2 |\tilde{w}|^{2-2/m} \left(1-|\tilde{w}|^{2/m}\right)^2}$}}
	\put(-140,157){\rotatebox{0}{\fontsize{12}{12}\selectfont{} $\dd{s_{cusp}^2} = \frac{ |d w|^{2}}{\left(|w| \log{|w|}\right)^{2}}$}}
	\put(-19,172){\rotatebox{0}{\fontsize{13}{13}\selectfont{} $\infty$}}
	\put(-50,115){\rotatebox{0}{\fontsize{12}{12}\selectfont{} $\mathcal{C}_{\infty, \epsilon}$}}
	\put(-372,120){\rotatebox{0}{\fontsize{12}{12}\selectfont{} $\mathcal{C}_{m, \epsilon}$}}
	\vspace{-1cm}
	
	\caption{\emph{Rotation symmetric coordinates}. Model cusp, $\mathcal{C}_{\infty, \epsilon}$, and model cone, $\mathcal{C}_{m, \epsilon}$, are shown in an $\epsilon$-neighborhoods of the parabolic and elliptic fixed points (i.e. at height $\epsilon$ in $rs$-coordinates $w,\tilde{w}$). The hyperbolic metrics, $\dd{s_{\text{cusp}}^2}$ and $\dd{s_{\text{cone}}^2}$, in rotation symmetric coordinates, $w$ and $\tilde{w}$, are also shown in neighborhoods of cusps and cones respectively.} \label{fig:localcoordinate}
\end{figure}
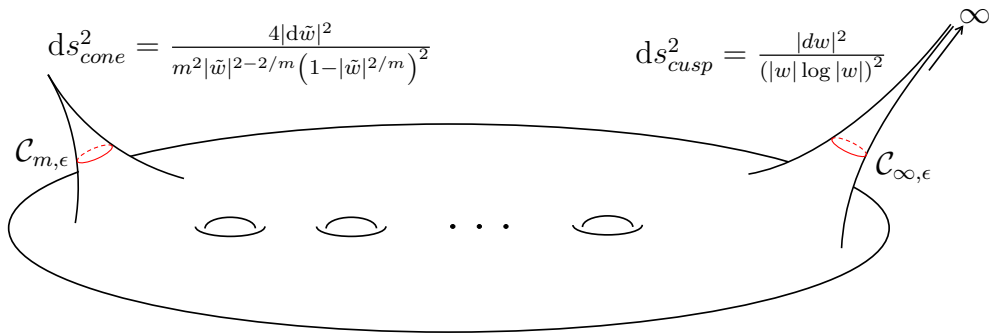
\subsection{Orbifold Differential Forms and Automorphic Forms}\label{sec:differentials}
In this section, we will study differential forms on hyperbolic Riemann orbisurfaces in more details (see, e.g. \cite{ZT_localindextheorem_1991,McIntyre:2004xs,ZT_2018} for more details). We start by defining orbifold differential forms on a general complex orbifold as a collection of invariant differential forms on each local uniformizing neighborhood:
\begin{definition}[Orbifold differential forms]\label{def:orbidif}
	If $O = (X_O, \mathcal{U})$ is a complex orbifold with an atlas of analytic orbifold charts $\mathcal{U} =\{(U_a,\tilde{U}_a,\Gamma_a,\ff_a)\}_{a \in A}$, we can define a complex \emph{orbifold $k$-form} $\phi$ on $O$ as a collection of local $\Gamma_a$-invariant complex $k$-forms $\{\tilde{\phi}_a^{\Gamma_a}\}_{a \in A}$ defined on each local uniformizing neighborhood $\tilde{U}_i$ such that every $\tilde{\phi}_i^{\Gamma_i}$ is preserved by all the change of charts. We say that the complex orbifold $k$-form $\phi$ is \emph{bigraded} of type $(p,q)$, with $k=p+q$, if $\phi$ is an invariant section of the V-bundle $\bigwedge^{p,q} T^{\ast}O:= \left(\bigwedge^{p} T^{\ast}_{1,0}O\right)\wedge \left(\bigwedge^{q} T^{\ast}_{0,1}O\right)$. We will denote the vector space of all such orbifold $(p,q)$-forms on $O$ by $\mathcal{E}^{p,q}(O)$.
\end{definition}
\begin{remark}\label{rmk:Stokes}
	Integration theory also goes through: Let $(U_a,\tilde{U}_a,\Gamma_a,\ff_a) \in \mathcal{U}$ be an orbifold chart and let $\phi$ be an orbifold differential from compactly supported on $V \subset X_O$. The characterization of $\phi$ as a collection of  local $\Gamma_a$-invariant differential forms $\{\tilde{\phi}_a^{\Gamma_a}\}$ that are supported on each $\tilde{U}_a$, enables us to define the integration of $\phi$ over $V$ as
	\begin{equation}
	\int_{V} \phi = \sum_{a \in A} \frac{1}{\#\Gamma_a}\int_{\ff^{-1}_a(U_a \cap V)} \tilde{\phi}_a^{\Gamma_a},
	\end{equation}
	where we have used partitions of unity to write the integral over $V$ as a sum of integrals over $V \cap U_a$. So, all of the standard integration techniques, such as \emph{Stokes' theorem}, are equally valid on orbifolds. 
\end{remark}

Now, consider a hyperbolic orbifold Riemann surface $O$ and let $K_O$ ($\neq K_X$) denote its holomorphic cotangent V-bundle (or its orbifold canonical bundle). For any $k,l \in \mathbb{Z}$, an orbifold $(k,l)$-differential on $O$ is defined as an element of $\diffspace^{k,l}(O) := \Gamma(O, K_O^{k} \otimes \overline{K}_O^{l})$ --- the vector space of smooth global sections of line V-bundle $K_O^{k} \otimes \overline{K}_O^{l}$.\footnote{Whenever $k$ or $l$ are negative, we understand $K_O^{k} := (TO)^{-k}$ and $\overline{K}_O^{l} := (\bar{T}O)^{-l}$.} For any pair of non-negative integers $p$ and $q$, there exists an isomorphism between the space $\mathcal{E}^{p,q}(O, K_O^{k} \otimes \overline{K}_O^{l})$ of orbifold differential forms of type $(p,q)$ with coefficients in the line V-bundle $K_O^{k} \otimes \overline{K}_O^{l}$ and the complex vector space $\diffspace^{0,0}(O,K_O^{k+p} \otimes \overline{K}_O^{l+q})$.

Every hyperbolic Riemann orbisurface $O$ can be realized as an orbifold quotient of the upper half-plane $\UHP = \{z \in \cmpx \, | \, \Im z > 0\}$ by a finite discrete subgroup group $\Gamma$ of its (orientation preserving) isometries $\operatorname{Isom}^{+}(\UHP) \cong \PSLR$, called a \emph{Fuchsian group}. Then, using the realization of $O$ as $[\UHP/\Gamma]$, one can identify every orbifold $(k,l)$-differential with a $\Gamma$-automorphic form of weight $(2k,2l)$ on $\UHP$: An \emph{automorphic form} of weight $(2k,2l)$ for $\Gamma$ is a $\Gamma$-invariant global section of the line bundle $K_{\UHP}^{k} \otimes \overline{K_{\UHP}}^{l}$. We will denote the space of $\Gamma$-automorphic forms of weight $(2k,2l)$ by $\diffspace^{k,l}(\UHP,\Gamma)$; an arbitrary element $\phi$ of $\diffspace^{k,l}(\UHP,\Gamma)$ has the form $\phi = \phi(z) \dd{z^k} \dd{\bar{z}^l}$ where $\phi(z)$ transforms according to the rule $\phi(\gamma z) \gamma'(z)^{k} \overline{\gamma'(z)}^{l} = \phi(z)$ for all $\gamma \in \Gamma$ and $z\in \UHP$.
\section{Geometric structures on orbifolds}\label{Apx:geometricorbifolds}
\subsection{Basic definitions and some theorems}\label{B1}
In \cite{ehresmann1936espaces} Ehresmann studied what he called \emph{locally homogeneous spaces}. More precisely, a locally homogeneous geometry on a manifold $M$, gives a $(G, \mathbb{X})$-structure on $M$ in the following sense described by Ehresmann. Let $\mathbb{X}$ be a (homogeneous) complex manifold -- called the \emph{model space} -- and let $G$ be a group acting effectively, transitively, and holomorphically on $\mathbb{X}$. A \emph{holomorphic $(G,\mathbb{X})$-structure} on a complex manifold $M$ is given by an open cover $\{U_a\}_{a \in A}$ of $M$ with holomorphic charts  $\mathfrak{f}_a: U_a \to \mathbb{X}$ such that the transition maps $\mathfrak{f}_a \circ \mathfrak{f}_b^{-1}: \mathfrak{f}_b(U_a \cap U_b) \to \mathfrak{f}_a(U_a \cap U_b)$ are given (on each connected component) by the restriction of an element $\textcolor{black}{\mathrm{g}_{a,b}} \in G$ (see Figure \ref{fig:geometricstructure}). Note that any geometric feature of $\mathbb{X}$ which is invariant by the symmetry group $G$ has an intrinsic meaning on the manifold $M$ equipped with a $(G,\mathbb{X})$-structure.
\begin{figure}
	\centering
	\includegraphics[width=0.7\linewidth]{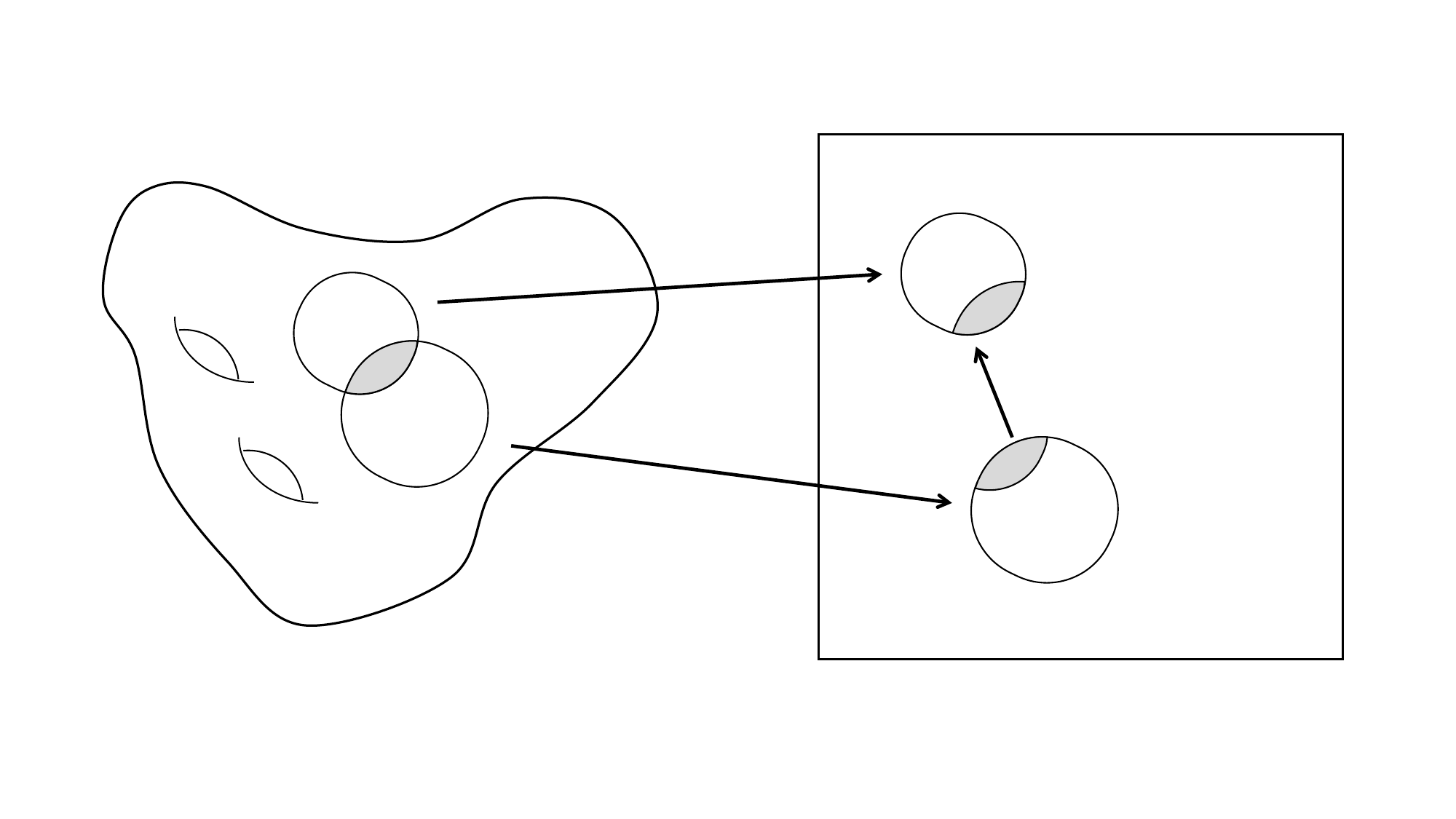}
	\put(-295,120){\rotatebox{0}{\fontsize{12}{12}\selectfont{} $M$}}
	\put(-23,115){\rotatebox{0}{\fontsize{12}{12}\selectfont{} $X$}}
	\put(-96,87){\rotatebox{0}{\fontsize{8}{8}\selectfont{} $\text{Element of G}$}} 
	\caption{\emph{Geometric structure on manifolds}. A geometric $(G,\mathbb{X})$-structure on a complex manifold $M$ is given by an atlas of holomorphic charts such that open neighborhoods are biholomorphic to open subsets of $\mathbb{X}$ and transition maps are given by restrictions of elements of $G$.} 
	\label{fig:geometricstructure}
\end{figure}

There exists a useful \emph{globalization of the coordinate charts} of a geometric structure in terms of the universal covering space and the fundamental group. The $(G,\mathbb{X})$-coordinate atlas $\{(U_a,\mathfrak{f}_a)\}_{a \in A}$ is replaced by a universal covering space $\tilde{M}$ with its group of deck transformations $\pi_1(M)$: The coordinate charts $\mathfrak{f}_a: U_a \to \mathbb{X}$ are replaced by a globally defined map $\dev: \tilde{M} \to \mathbb{X}$ called a \emph{developing map} (see Figure \ref{fig:devpair}). In addition, the developing map is equivariant with respect to the actions of $\pi_1(M)$:
\begin{equation}
\dev \circ \gamma = \hol(\gamma) \circ \dev,
\end{equation}
where $\gamma \in \pi_1(M)$ is a deck transformation and $\hol: \pi_1(M) \to G$ is called a \emph{holonomy representation} -- i.e., the coordinate changes are replaced by the holonomy homomorphism. The resulting \emph{developing pair} $(\dev,\hol)$ is unique up to composition/conjugation by elements in $G$, i.e. up to $\big(\dev,\hol(\cdot)\big) \to \big(\mathrm{g} \circ \dev,\mathrm{g} \hol(\cdot) \mathrm{g}^{-1}\big)$ transformations. This determines the structure.
\begin{figure}
	\centering
	\includegraphics[width=0.9\linewidth]{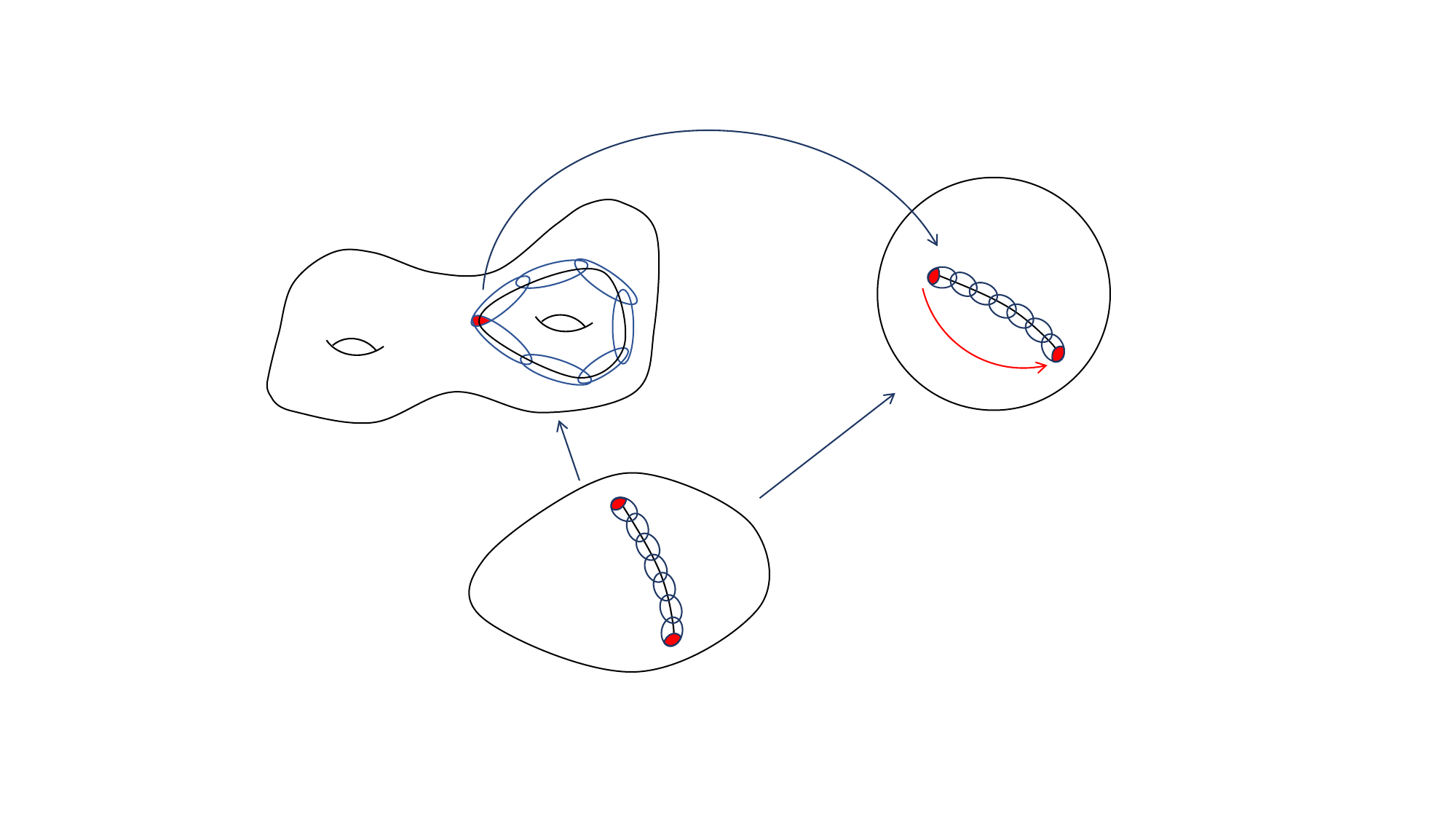}
	\put(-310,155){\rotatebox{0}{\fontsize{12}{12}\selectfont{} $M$}}
	\put(-210,173){\rotatebox{0}{\fontsize{12}{12}\selectfont{} $\mathfrak{f}$}}
	\put(-97,155){\rotatebox{0}{\fontsize{12}{12}\selectfont{} $\mathbb{X}$}}
	\put(-274.5,137){\rotatebox{0}{\fontsize{12}{12}\selectfont{} $\textcolor{blue}{\mathscr{U}}$}}
	\put(-275,130){\rotatebox{0}{\fontsize{10}{10}\selectfont{} $x_{\ast}$}}
	\put(-230,148){\rotatebox{0}{\fontsize{10}{10}\selectfont{} $l$}}
	\put(-142,114){\rotatebox{0}{\fontsize{6.5}{6.5}\selectfont{} $\textcolor{red}{\hol(l)\hspace{-.7mm}\in\hspace{-.7mm}G}$}}
	\put(-270,38){\rotatebox{0}{\fontsize{12}{12}\selectfont{} $\tilde{M}$}}
	\put(-225,60){\rotatebox{0}{\fontsize{10}{10}\selectfont{} $\tilde{l}$}}
	\put(-182,103){\rotatebox{0}{\fontsize{10}{10}\selectfont{} $\dev$}}
	\vspace{-1cm}
	
	\caption{\emph{Development pair}.} 
	\label{fig:devpair}
\end{figure}
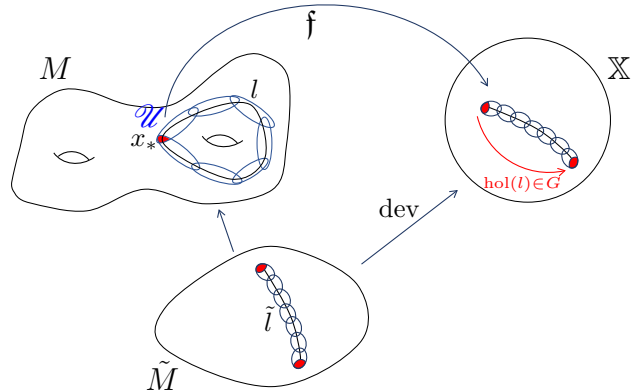
In this section, we will introduce $(G,\mathbb{X})$-structures on orbifolds. Simply put, a \emph{$(G,\mathbb{X})$-orbifold} is locally modeled on $\mathbb{X}$ modulo finite subgroups of $G$. We will start by giving a definition of geometric structures on orbifolds based on atlases of charts, as well as using developing maps from the universal orbifold covering. Then, we will introduce and study the deformation spaces of these orbifold $(G,\mathbb{X})$-structures in analogy with Goldman's work \cite{goldman2010locally,goldman1987geometric,goldman2008higgs} for the manifold case. Some of the most important examples of these geometric structures on orbifolds are provided by projective structures as well as hyperbolic, Euclidean, or spherical structures; we will end this section by studying these specific examples and the relation between them. See \cite{choi2004geometric,Choi_2012} for more on orbifold geometric structures.

In order to give a precise definition of an orbifold geometric structure modeled by the pair $(G,\mathbb{X})$, we need to introduce the notion of (holomorphic) orbifold $(G,\mathbb{X})$-charts: Let $O=(X_O,\mathcal{U})$ be a complex orbifold and let $U \subset X_O$ be an open subset of $X_O$ with a local model $(\tilde{U},\Gamma)$. Then, a \emph{holomorphic $(G,\mathbb{X})$-chart} on $U$ is defined to be given by a pair $(\tilde{\mathfrak{f}},\mathfrak{C})$ where $\tilde{\mathfrak{f}}: \tilde{U} \hookrightarrow \mathbb{X}$ gives a holomorphic embedding of local uniformizing neighborhood $\tilde{U}$ into the model space $\mathbb{X}$ (equivalently, $\tilde{\mathfrak{f}}: \tilde{U} \to \tilde{V} \subset \mathbb{X}$ is a holomorphic isomorphism onto an open subset of $\mathbb{X}$) and is considered as a \emph{local $\mathbb{X}$-coordinate} on \textcolor{black}{$\tilde{V}$} while $\mathfrak{C}: \Gamma \hookrightarrow G$ is an injective group homomorphism (equivalently, a group isomorphism between the local uniformizing group $\Gamma$ and a finite subgroup of $G$).
\begin{definition}[Orbifold $(G,\mathbb{X})$-structures]
	Let $\mathcal{U}:=\big\{(U_a, \tilde{U}_a, \Gamma_a, \ff_a)\big\}_{a \in A}$ denote an orbifold atlas that induces an orbifold structure on $X_O$. A \emph{holomorphic orbifold $(G,\mathbb{X})$-atlas} on \textcolor{black}{$X_O$} that is \emph{compatible with orbifold atlas $\mathcal{U}$} is given by a collection of holomorphic $(G,\mathbb{X})$-charts $\big\{(\tilde{\mathfrak{f}}_a: \tilde{U}_a \hookrightarrow \mathbb{X},\, \mathfrak{C}_a: \Gamma_a \hookrightarrow G)\big\}_{_{a \in A}}$ such that holomorphic embeddings $\eta_{ab}: \tilde{U}_a \to \tilde{U}_b$ are realized as elements $\textcolor{black}{\mathrm{g}_{a,b}} \in G$ and monomorphisms $\Upsilon_{ab}: \Gamma_a \hookrightarrow \Gamma_b$ are given by conjugations $\gamma \mapsto \mathrm{g}_{ab}\circ\gamma\circ \mathrm{g}_{ab}^{-1}$ for all $\gamma \in \Gamma_a$. The datum of a holomorphic orbifold $(G,\mathbb{X})$-atlas compatible with the \emph{maximal} complex orbifold atlas, $\mathcal{U}_{\max}$, defines a \emph{holomorphic $(G,\mathbb{X})$-structure} on complex orbifold \textcolor{black}{$O=(X_O,\mathcal{U}_{\max})$}. If a complex orbifold $O$ admits a holomorphic $(G,\mathbb{X})$-structure, one can always choose a local model $(\tilde{U}_a,\Gamma_a)$ for each open set $U_a \subset X_O$ where $\tilde{U}_a \subset \mathbb{X}$ and $\Gamma_a < G$; notice that if we require the collection $\{(\tilde{U}_a \subset \mathbb{X}, \Gamma_a<G)\}$ to be a set of local models for a complex orbifold, the space $\mathbb{X}$ should itself admit a complex structure and the action of group $G$ on $\mathbb{X}$ should be given by holomorphic automorphisms. When there is no risk of confusion, we will say that a maximal orbifold atlas $\mathcal{U}_{\max}:=\{(U_a,\tilde{U}_a \subset \mathbb{X},\Gamma_a < G,\textcolor{black}{\ff_a})\}_{_{a \in A}}$ induces a complex orbifold $(G,\mathbb{X})$-structure on \textcolor{black}{$X_O$}.
\end{definition} 

Once again, the definition of complex orbifold $(G,\mathbb{X})$-structures simplifies considerably for the case of Riemann orbisurfaces (orbifold Riemann surfaces) due to the restricted nature of singular points in one complex dimension:
\begin{definition}[Complex $(G,\mathbb{X})$-structures on Riemann orbisurfaces]
	Let $O=(X_O,\brdiv)$ be a Riemann orbisurface. A complex $(G,\mathbb{X})$-structure on the Riemann orbisurface $O$ is given by a $(G,\mathbb{X})$-structure on its underlying Riemann surface $X_O$ such that the complex structure that is induced on $X_O$ by the $(G,\mathbb{X})$-structure coincides with the already existent complex structure on $X_O$.
\end{definition} 
A \emph{holomorphic orbifold $(G,\mathbb{X})$-map} $f: O \to O'$ is a holomorphic orbifold map $O \overset{f}{\to} O'$ such that its holomorphic local lift at each point $x \in \textcolor{black}{X_O}$ is given by a pair of maps between local models $(\tilde{U}_a,\Gamma_{x}) \overset{(\tilde{f}_{x},\bar{f}_{x})}{\longrightarrow} (\tilde{U}'_a,\Gamma'_{f_{_{O}}(x)})$ where the group homomorphism $\bar{f}_{x}: \Gamma_{x} \to \Gamma'_{f_{_{O}}(x)}$ is induced by conjugation $\gamma_{a} \mapsto \mathrm{g}_{ab} \circ\gamma_{a}\circ \mathrm{g}_{ab}^{-1}$ for all $\gamma_{a} \in  \Gamma_{x}$ and holomorphic $\tilde{f}_{x}$-equivariant map $\tilde{f}_{x}: \tilde{U}_a \to \tilde{U}'_a$ is given by a restriction of $\mathrm{g}\in G$. Note that if $O$ is a complex orbifold and $f: O \to O'$ is a holomorphic orbifold map to another complex orbifold $O'$ equipped with a (holomorphic) $(G,\mathbb{X})$-structure \textcolor{black}{$\mathfrak{G}'$}, we can pull-back the $(G,\mathbb{X})$-structure $\mathfrak{G}'$ on $O'$ to another $(G,\mathbb{X})$-structure $f^{\ast}(\mathfrak{G}')$ on $O$ such that $f$ becomes a $(G,\mathbb{X})$-map. In particular, a complex $(G,\mathbb{X})$-structure on an orbifold $O$ induces a complex $(G,\mathbb{X})$-structure on its covering orbifolds through a pull-back by the covering map.

\begin{theorem}[Thurston]\label{theorem:geometricorbifolds}
	When $G$ is a group of biholomorphisms of a complex manifold $\mathbb{X}$, then every complex $(G,\mathbb{X})$-orbifold is \emph{good}.
\end{theorem}
\begin{remark}
If $G$ is a subgroup of a linear group, then $O$ is \emph{very good} by Selberg's lemma provided that $O$ has a finitely generated fundamental group. In particular, all geometric orbifold Riemann surfaces are very good -- i.e. finitely covered by a manifold.
\end{remark}
Next, we note that the idea of developing map extends to orbifolds with complex geometric structures:
\begin{theorem}
Let $O$ be a $(G,\mathbb{X})$-orbifold, where $(G,\mathbb{X})$ is a complex analytic geometry. Then, there exists a developing map
\begin{equation*}
\dev: \tilde{O} \to \mathbb{X}
\end{equation*}
defined on the universal covering $\tilde{O}$ and a holonomy representation $\hol: \pi_1(O) \to G$ such that
\begin{equation*}
\dev \circ \gamma = \hol(\gamma) \circ \dev
\end{equation*}
for any deck transformation $\gamma \in \pi_1(O)$.
\end{theorem}
\subsection{Hierarchy of Geometric Structures}
Often one geometric structure \emph{contains} or \emph{refines} another geometry as follows. Suppose that $G$ and $G'$ act transitively on $\mathbb{X}$ and $\mathbb{X}'$ respectively, and $\mathbb{X} \xrightarrow{f} \mathbb{X}'$ is a local biholomorphism which is equivariant with respect to a homomorphism $G \xrightarrow{F} G'$ --- i.e. for each $\mathrm{g} \in G$, the following diagram
\begin{equation}
\begin{CD}
\mathbb{X}	@> f >> \mathbb{X}' \\
@V \mathrm{g} VV		@VV F(\mathrm{g}) V\\
\mathbb{X}	@> f >> \mathbb{X}'
\end{CD}
\end{equation}
commutes. Then (by composition with $f$ and $F$) every $(G,\mathbb{X})$-structure determines a $(G',\mathbb{X}')$-structure. There are many important examples of this correspondence, most of which occur when $f$ is an embedding. For example, when $f$ is the identity map and $G \subset G'$ is a subgroup preserving some extra structure on $\mathbb{X}=\mathbb{X}'$, then every $(G,\mathbb{X})$-structure is a fortiori an $(G',\mathbb{X}')$-structure. A more important example for us is the relation between projective and hyperbolic structures: In this case the map $f: \UHP \hookrightarrow \CPl$ will be an embedding and $\PSLR$ can be viewed as the subgroup of $\PSLC$ which leaves the subspace $\UHP \subset \CPl$ invariant. Therefore, every hyperbolic structure determines a projective structure.
\subsection{Orbifold $\CPl$-Structure and Projective Connections}\label{sec:projconn}
Let $O=(X_O,\brdiv)$ be a hyperbolic orbifold Riemann surface with signature $(g;m_1,\dots,m_{n_e};n_p)$. A \emph{$\CPl$-structure} or \emph{complex projective structure} on the Riemann orbisurface $O$ is an orbifold $(G,\mathbb{X})$-structure on its underlying Riemann surface $X_O$ with $\mathbb{X}=\CPl$ and $G=\PSLC$ such that the complex structure on $X_O$ induced by the $\CPl$-structure \textcolor{black}{coincides} with the given complex structure on this Riemann surface (see Appx.\ref{B1} for more details). Similar to other orbifold $(G,\mathbb{X})$-structures, a $\CPl$-structure on $O$ can be equivalently described as a \emph{developing pair} $(\dev,\hol)$ where $\dev: \tilde{O} \to \CPl$ is the \emph{developing map} defined on the universal cover $\tilde{O} \cong \UHP$ and $\hol: \pi_1(O) \to \PSLC$ is the \emph{holonomy} or \emph{monodromy representation} such that developing map is a $\hol$-equivariant immersion.

Since any hyperbolic Riemann orbisurface is developable, there exists a finite Galois covering
\begin{equation}
\varpi: Y \to X_O^{\text{punc}}
\end{equation}
such that $\varpi$ is unramified over $X_O^{\text{reg}} = X_O^{\text{punc}} \backslash \sing_{\curlywedge}(O)$, and for each $x_i \in \sing_{\curlywedge}(O)$, the order of ramification at every point of $\varpi^{-1}(x_i)$ is $\textcolor{black}{\nu}(x_i) = m_i$ --- i.e. $\textcolor{black}{\varpi}$ induces a holomorphic orbifold covering map $Y \to O$. Let us, for the sake of simplicity, take $O$ to be a compact Riemann orbisurface so that $Y$ is a closed Riemann surface of genus $\tilde{g}$; the case with punctures follows from the formal limit $m \to \infty$. Let \textcolor{black}{$H$} be the group of deck transformations or the \emph{Galois group} for $\varpi$ such that $O \cong [Y/H]$; we have $H \subseteq \operatorname{Aut}(Y)$ where $\operatorname{Aut}(Y)$ denotes the group of holomorphic automorphisms of $Y$. Now, let us fix a projective structure $P$ on the closed Riemann surface $Y$ and consider the convex combination 
\begin{equation}
P^{{}_H} := \frac{1}{\# H} \sum_{\textcolor{black}{\mathrm{h}} \in H} \mathrm{h}^{\ast} P,
\end{equation}
where $\mathrm{h}^{\ast} P$ denotes the pullback of $\CPl$-structure $P$ by holomorphic automorphism $\mathrm{h}$, $\#  H$ is the order of $H$, and the average is defined using the convex structure of the space $\mathcal{P}(Y)$ of all projective structures on $Y$ compatible with its complex structure. Note that this projective structure $P^{{}_H}$ on $Y$ is clearly left invariant by the action of $H$ on $Y$. 

We will now construct an orbifold projective structure on $X_O$ using $P^{{}_H}$ (see \cite{biswas2006orbifold,ares2013symplectic}): Let $\big\{(\tilde{U}_a, \tilde{\mathfrak{f}}_a)\big\}_{a \in A}$ be a maximal $\CPl$-atlas on $Y$ where open subsets $\tilde{U}_a$ of $Y$ left invariant by the action of $H$ on $Y$ and $\CPl$-coordinate functions $\tilde{\mathfrak{f}}_a : \tilde{U}_a \to \tilde{V}_a \subset \CPl$ are holomorphic isomorphisms compatible with the projective structure $P^{{}_H}$. Consider the ramified coverings 
\begin{equation}
\varpi \circ \tilde{\mathfrak{f}}_a^{-1}: \tilde{V_a} \to U_a \cong \tilde{U}_a/H_a \subset X_O,
\end{equation}
where $H_a \subset H$ is the Galois group of the restriction $\varpi|_{_{\tilde{U}_a}}$. Then, the collection of ramified coverings $\big\{\varpi \circ \tilde{\mathfrak{f}}_a^{-1}\big\}_{a \in A}$ combine together to define an orbifold projective structure on $X_O$. Indeed, that they define a $\CPl$-structure on $O$ is an immediate consequence of the facts that $P^{{}_H}$ is left invariant by the action on $Y$ of the Galois group $H$ and $\varpi$ is ramified exactly over $\sing(O)$ with $m_i = \nu (x_i)$ as the order of ramification over each $x_i \in \sing(O)$. Note that one consequence of this construction is that the space of all $\CPl$-structures on $O$ compatible with its complex structure, $\mathcal{P}(O)$, is the fixed point locus for the action of $H$ on $\mathcal{P}(Y)$ --- i.e. $\mathcal{P}(O) = \mathcal{P}^{{}_H}(Y) \subset \mathcal{P}(Y)$.

It is well known that if $\big\{(\tilde{U}'_b, \tilde{\mathfrak{f}'}_b)\big\}_{b \in B}$ is a $\CPl$-atlas on a compact Riemann surface $Y$ that defines a complex projective structure $P \in \mathcal{P}(Y)$ and $\big\{(\tilde{U}_a, \tilde{u}_a:\tilde{U}_a \to \cmpx)\big\}_{a \in A}$ is any atlas of holomorphic charts on this Riemann surface, the collection of holomorphic functions   $\left\{\text{Sch}\left(\tilde{\mathfrak{f}'}_b;\tilde{u}_a\right)\right\}$, defined on overlaps $\tilde{U}'_b \cap \tilde{U}_a$, give what is known as a ``(holomorphic) projective connection'' on $Y$.
Let $\big\{(\tilde{U}_a, \tilde{u}_a)\big\}_{a \in A}$ be any \textcolor{black}{complex-analytic atlas} on $Y$ with local coordinates $\tilde{u}_a:\tilde{U}_a \to \cmpx$ and transition functions $\tilde{u}_a = \tilde{g}_{ab} \circ \tilde{u}_b$ on $\tilde{U}_a \cap \tilde{U}_b$. A \emph{holomorphic projective connection} $\tilde{R}$ on $Y$ is in general defined as a collection $\{\tilde{r}_a\}_{a \in A}$ of holomorphic functions $\tilde{r}_a$ supported on $\tilde{U}_a$ that satisfy
\begin{equation}\label{proj}
\tilde{r}_b  = \tilde{r}_a \circ \tilde{g}_{ab} (\tilde{g'}_{ab})^2 + \text{Sch}\left(\tilde{g}_{ab};\tilde{u}_b\right),
\end{equation}
on every intersection $\tilde{U}_a \cap \tilde{U}_b$. One can conversely show that any projective connection $\tilde{R}$ on $Y$ defines a $\CPl$-structure $P \in \mathcal{P}(Y)$ as follows (see e.g. \cite[Proposition~3.3]{shengyuannote}): Let $\tilde{R} = \{\tilde{r}_a\}_{a \in A}$ be a holomorphic projective connection with respect to complex-analytic atlas $\big\{(\tilde{U}_a, \tilde{u}_a)\big\}_{a \in A}$. On each open subset $\tilde{U}_a$, let $\zeta_a$ be any solution of the equation 
\begin{equation}
\text{Sch}\left(\zeta_a;\tilde{u}_a\right) = - \tilde{r}_a.
\end{equation}
The holomorphic functions $\zeta_a$ have nowhere vanishing derivatives, so that we can assume they are injective up to shrinking $\tilde{U}_a$s. Then, the new coordinates $\big\{\zeta_a \circ \tilde{u}_a\big\}_{a \in A}$ define the same complex structure on $Y$. In addition, the Schwarzian derivatives $\text{Sch}\left(\zeta_a \circ \tilde{u}_a;\zeta_b \circ \tilde{u}_b\right)$ can easily be shown to vanish by using $\text{Sch}\left(f \circ g;z\right) = \text{Sch}\left(f ;g(z)\right) (g')^2 + \text{Sch}\left(g;z\right)$. This implies that $\big\{(\tilde{U}_a, \zeta_a \circ \tilde{u}_a)\big\}_{a \in A}$ is an atlas of complex projective structure. A different choice of atlas or a different collection of solutions $\zeta_a$ would define the same complex projective structure. Therefore, one concludes that the set of holomorphic projective connections $\tilde{R}$ on a Riemann surface $Y$ is in bijection with the set of $\CPl$-structures on $Y$.

Similar to the way that we have constructed orbifold $\CPl$-structures on $X_O$ using $H$-invariant projective structures $P^{{}_H}$ on the covering Riemann surface $Y$, we can try to construct projective connections on $X_O$ by starting from $H$-invariant projective connections on $Y$. More concretely, let $\big\{(\tilde{U}_a, \tilde{u}_a)\big\}_{a \in A}$ be a complex-analytic atlas on $Y$ and let $\tilde{R}^{{}_H}$ be a holomorphic projective connection corresponding to the $H$-invariant $\CPl$-structure $P^{{}_H} \in \mathcal{P}^{{}_H}(Y)$. The projective connection $\tilde{R}^{{}_H}$ is given by a collection $\{\tilde{r}^{{}_H}_a\}_{a \in A}$ of holomorphic functions $\tilde{r}^{{}_H}_a$ supported on each open subset $\tilde{U}_a$ that are invariant under the action of Galois groups $H_a \subset H$ of restrictions $\varpi|_{_{\tilde{U}_a}}: \tilde{U}_a \to U_a \cong \tilde{U}_a/H_a \subset X_O$ and satisfy \eqref{proj} on every overlap $\tilde{U}_a \cap \tilde{U}_b$. Let $\tilde{U}_a$ be an open subset containing only one ramification point of the covering $\varpi: Y \to X_O$ with multiplicity $m_a$ such that $H_a := \operatorname{Gal}(\varpi|_{_{\tilde{U}_a}}) \cong \mathbb{Z}_{m_a}$. If $\tilde{r}^{{}_H}_a$ is one of the holomorphic functions defining $\tilde{R}^{{}_H}$ that is supported on $\tilde{U}_a$ and is left invariant by the Galois group $\mathbb{Z}_{m_a}$ for $\varpi|_{_{\tilde{U}_a}}: \tilde{U}_a \to U_a \cong \tilde{U}_a/\mathbb{Z}_{m_a}\subset X_O$, then $\tilde{r}^{{}_H}_a$ descends, by the map $\varpi|_{_{\tilde{U}_a}}$, to a meromorphic function with at most a pole of order 2 at singular point $x_a \in U_a$. In other words, $\tilde{r}^{{}_H}_a = (\varpi|_{_{\tilde{U}_a}})^{\ast} r_a$, where $r_a$ is a meromorphic function on $U_a$ with pole only at $x_a$ of order at most two. In fact, if $x_a \in U_a$ is a singular point of order $m_a$ and $u_a(x_a) = 0$, the behavior of $r_a$ near this singularity is given by
\begin{equation}
r_a(u_a) = \frac{1-1/m_a^2}{u_a^2} + \order{|u_a|^{-1}} \quad \text{as} \quad u_a \to 0.
\end{equation}
If $m_a \neq \infty$, the monodromy in $\PSLC$ around this \emph{regular singularity} will be given by multiplication with $e^{\frac{2\pi i}{m_a}}$ while if $m_a = \infty$, the monodromy around this singularity will be given by a non-trivial parabolic element. The collection $\{r_a\}_{a \in A}$ of such meromorphic functions will define a \emph{quasi-bounded projective connection} $R$ on the underlying Riemann surface $X_O$ which is in bijection with orbifold $\CPl$-structure on this surface (see e.g. \cite{faltings1983real,luo1993monodromy,gupta2021monodromy} for more details). A quasi-bounded projective connection $R$ naturally determines a second-order linear differential equation on the Riemann surface $X_O$, the Fuchsian differential equation
\begin{equation}\label{FuchsianODE}
\dv[2]{\mathsf{y}_a}{u_a} + \frac{1}{2} r_a \mathsf{y}_a = 0,
\end{equation}
where $\{\mathsf{y}_a\}_{a \in A}$ is understood and as a multi-valued meromorphic differential of order $-1/2$ on $X_O$.

Last but not least, the difference between two projective connections is a meromorphic quadratic differential with only simple poles --- i.e., a collection $\{q_a\}_{a \in A}$ of meromorphic functions on each open subset $U_a \in X_O$ with the transformation law
\begin{equation}
q_b  = q_a \circ g_{ab} (g'_{ab})^2,
\end{equation}
and the additional condition that $q_a(u_a) = \order{|u_a|^{-1}}$ as $u_a \to 0$, if $x_a \in U_a$ is a singular point and $u_a(x_a)=0$. Conversely, we can add a meromorphic quadratic differential to a given quasi-bounded projective connection $R$ to obtain a new quasi-bounded projective connection. Since we know that each Riemann orbisurface has at least one $\CPl$-structure, the one given by Poincar\'{e}-Koebe uniformization, we will have the Proposition \ref{Biswas} (see \cite{biswas2006orbifold}).
\section{Asymptotics Near Elliptic and Parabolic Fixed Points}\label{Apx:Asymptotics}
In this appendix, we will sketch the derivation of the asymptotics of the Liouville field $\varphi$ in a neighborhood of each of the parabolic and elliptic points (see section 2 of \cite{Montplet2016RiemannRochII} as well as the proof of lemma 2 in \cite{Zograf1988ONLE} and lemma 4 in \cite{Takhtajan:2001uj}). 
One of the remarkable corollaries of the uniformization theorem is that the orbifold Riemann surface $O^{\text{reg}}$ has a unique metric of constant curvature $-1$ compatible with the complex structure. It is the projection on $O$ of the Poincar\'{e} metric $(\Im z)^{-2} |\dd{z}|^2$ on $\UHP$ and has the form $\dd{s^2} = e^{\varphi(w)} |\dd{w}|^2$. The condition that the curvature is constant and equal to $-1$ means that the function $\varphi$ satisfies the Liouville's equation on $O^{\text{reg}}$,
\begin{equation}\label{extra8}
\partial_w \partial_{\bar{w}} \varphi = \frac{1}{2} e^{\varphi}.
\end{equation}
The two asymptotic forms of metric \eqref{cuspcoordinate} and \eqref{conersmetric} in an $\epsilon$-neighborhood of cusps and cones enable us to find the asymptotic form of the Liouville field $\varphi$ near these  parabolic and elliptic fixed-points:
\begin{itemize}
	\item \emph{Asymptotic form of $\varphi(w)$ near cusps}:
	\begin{equation}
	\varphi(w) = \left\{
	\begin{split}
	& -2 \log |w-w_i| - 2 \log \left| \log |w-w_i|\right| + \order{1} \quad(w_i \neq \infty),\qquad & w \to w_i,\\
	& -2 \log |w| - 2 \log \log |w| + \order{1}, & w \to \infty;
	\end{split}
	\right.
	\end{equation}
	\item \emph{Asymptotic form of $\varphi(w)$ near cones}:
	\begin{equation}
	\varphi(w) = \left\{
	\begin{split}
	& -2 (1-\frac{1}{m_i})\log |w-w_i| + \order{1} \quad(w_i \neq \infty),\qquad & w \to w_i,\\
	\end{split}
	\right.
	\end{equation}
\end{itemize}
One can also derive these asymptotics by studying the mapping $J$, called \emph{Klein's Hauptmodule}, as a meromorphic function on $\UHP$ which is automorphic with respect to the Fuchsian group $\Gamma$-- i.e.
\begin{equation}\label{J} 
J(\gamma z) = J(z) \hspace{1cm}\text{for}\hspace{.2cm} z \in \UHP \hspace{.2cm}\text{and}\hspace{.2cm} \forall \gamma \in \Gamma
\end{equation}
For the sake of simplicity, let us assume that the Fuchsian group $\Gamma$ has genus $0$; we can choose a standard system of $n_e$ elliptic generators $\tau_1,\dots,\tau_{n_e}$ of orders $m_1,\dots,m_{n_e}$ and $n_p$ parabolic generators $\kappa_{1}, \dots, \kappa_{n_p}$ ($n_e+n_p =n$) satisfying the single relation $\tau_1 \dotsm \tau_{n_e}\kappa_{1} \dotsm \kappa_{n_p} = 1$.  Let $z_{1}, \dots, z_{n_e} \in \hat{\cmpx}$ and $z_{n_e+1}, \dots, z_n \in \mathbb{R} \cup \{\infty\}$ be the fixed points of elliptic and parabolic generators respectively which project into $w_1, \dots, w_{n_e}, w_{n_e+1}, \dots, w_{n} \in \hat{\cmpx}$. We will assume that $z_{n-2}= 0, z_{n-1}=1,$ and $z_{n}=\infty$, which can always be achieved by conjugation with $\PSLR$. The elements   $\tau_1,\dots,\tau_{n_e}$ and $\kappa_{1}, \dots, \kappa_{n_p}$ are then represented in their \emph{matrix form} as
\begin{equation}\label{generators}
\tau_i =  \mqty(z_{i} - \lambda_i \bar{z}_{i} & (\lambda_i -1) z_{i} \bar{z}_{i} \\ 1-\lambda_i & \lambda_i z_{i} - \bar{z}_{i}), \quad
\kappa_j = \mqty(1+\delta_{n_e+j}\hspace{.3mm} z_{n_e+j} & - \delta_{n_e+j}\hspace{.3mm} z_{n_e+j}^2 \\ \delta_{n_e+j}  & 1-\delta_{n_e+j}\hspace{.3mm} z_{n_e+j}), \quad \kappa_{n_p} = \mqty(1 &  \delta_n \\ 0 & 1),
\end{equation}
with $i=1,\dots,n_e$ and $j=1,\dots,n_p-1$. In the above equation, $\lambda_1,\dots, \lambda_{n_e} = e^{\frac{2 \pi i}{m_1}}, \dots, e^{\frac{2\pi i}{m_{n_e}}}$ are called the \emph{multipliers} of $\tau_1, \dots, \tau_{n_e}$ and $\delta_{n_e+1}, \dots, \delta_{n} \in \mathbb{R}$ are called \emph{translation lengths} of $\kappa_{1},\dots,\kappa_{n_p}$. Therefore, according to the Eq.\eqref{J} and using \eqref{generators}, it is easy to see
that in a neighborhood of each elliptic point $z_{i}$, $i=1, \dots, n_e$ with ramification index $m_i$, the Hauptmodule can be expanded as
\begin{equation}\label{Hauptmodulecone}
J(z) = w_{i} + \sum_{k=1}^{\infty} J^{(i)}_k  \left(\frac{z-z_{i}}{z-\bar{z}_{i}}\right)^{k m_i},
\end{equation}
and in a neighborhood of each of the parabolic points $z_{i}$, $i=n_e+1,\dots,n-1$:
\begin{equation}\label{Hauptmodulecusp}
J(z) = w_{i} + \sum_{k=1}^{\infty} J^{(i)}_k \exp(-\frac{2 \pi \sqrt{-1}k}{|\delta_{i}|(z-z_{i})}).
\end{equation}
Similarly, in a neighborhood of the parabolic point $z_n = \infty$, the function $J(z)$ can be expanded as 
\begin{equation}\label{Hauptmodulecuspinfty}
J(z) = \sum_{k=-1}^{\infty} J^{(n)}_k \exp(\frac{2 \pi \sqrt{-1}k z}{|\delta_n|}).
\end{equation}
In equations~\eqref{Hauptmodulecone}, \eqref{Hauptmodulecusp} and \eqref{Hauptmodulecuspinfty}, the coefficients $J^{(i)} _{1}, J^{(n)}_{-1}\neq 0$ because the mapping $J$ is univalent in any fundamental domain for $\Gamma$. If we choose elements $\varsigma_{n_e+j}, \dots, \varsigma_n \in \PSLR$ such that $\varsigma_{n_e+j} (\infty) = z_{n_e+j}$ and
\begin{equation*}
\varsigma_{n_e+j}^{-1}\hspace{1mm} \kappa_{j}\hspace{1mm} \varsigma_{n_e+j} = \mqty(1 & \pm 1 \\ 0 & 1),
\end{equation*}
for $j=1,\dots,n_p$, we can also rewrite \eqref{Hauptmodulecusp} and \eqref{Hauptmodulecuspinfty} in the form
\begin{equation}\label{diagJexpansion}
J(\varsigma_{i} z) = \left\{
\begin{split}
& w_{i} + \sum_{k=1}^{\infty} J^{(i)}_k \exp(2 \pi \sqrt{-1}k z), &\hspace{.5cm} i=n_e+1\dots n-1,\\
& \sum_{k=-1}^{\infty} J^{(i)}_k \exp(2 \pi \sqrt{-1} k z), \qquad & i=n.
\end{split}
\right.
\end{equation}
The hyperbolic metric $e^{\varphi(w)} |\dd{w}|^2$ on $O$ is given by
\begin{equation}\label{extra9}
e^{\varphi(w)} = \frac{\left|J^{-1}(w)'\right|^2}{\left(\Im J^{-1}(w)\right)^2},
\end{equation}
and satisfies the Liouville's equation \eqref{Liouvilleequation}. In order to find the asymptotic behavior of the Liouville field $\varphi(w)$ near branch points and cusps, we need to study the \emph{multivalued} analytic function $J^{-1}: O \to \UHP$ which is a locally univalent linearly polymorphic function\footnote{Linearly polymorphic means that the branches of this function are connected by linear fractional transformations in $\Gamma$.} on $O$. Let us first calculate the expansion of $J^{-1}(w)$ near elliptic fixed points. To do so, we will rewrite the expansion \eqref{Hauptmodulecone} as
\begin{equation}
J = w_{i} + \sum_{k=1}^{\infty} J^{(i)}_k  \left(\mathsf{u}_{\UD}^{m_i}\right)^k,
\end{equation}
where $i=1,\dots,n_e$ and $\mathsf{u}_{\UD} \equiv (z - z_{i})/(z-\bar{z}_{i})$ is the coordinate on $\UD$ instead of $\UHP$. Formally, by inverting the above power series, we get
\begin{multline}\label{Jtildeinv}
\mathsf{u}_{\UD}=\tilde{J}^{-1}(w)= \left(\frac{1}{J^{(i)}_1 }\right)^{\frac{1}{m_i}} (w-w_{i})^{\frac{1}{m_i}} - \frac{J^{(i)}_2}{m_i \left(J^{(i)}_1\right)^{2+\frac{1}{m_i}}} (w-w_{i})^{1+\frac{1}{m_i}} \\ + \frac{-2 m_i J^{(i)}_1 J^{(i)}_3 + \left(J^{(i)}_2\right)^2 (1 + 3 m_i)}{2 m_i^2 \left(J^{(i)}_1\right)^{4 + \frac{1}{m_i}}} (w-w_{i})^{2+\frac{1}{m_i}} + \cdots 
\end{multline}
Then, by using the definitions $\mathsf{u}_{\UD}$, we have
\begin{equation}\label{JInversecone}
J^{-1}(w) = z_{i}  + 2 \sqrt{-1} \, \Im z_{i} \left(\frac{w-w_{i}}{J^{(i)}_1}\right)^{\frac{1}{m_i}} \left(1+ \frac{J^{(i)}_2}{m_i \left(J^{(i)}_1\right)^2} (w-w_{i}) + \cdots \right).
\end{equation}
We can use the same method to find the expansion of $J^{-1}$ near cusps as well. Near parabolic points by rewriting the equations  \eqref{Hauptmodulecusp} and \eqref{Hauptmodulecuspinfty} as a power series in $\iota_{j \neq n} \equiv \exp(-\frac{2 \pi \sqrt{-1}}{|\delta_{j}|(z-z_{j})})$ and $\iota_n \equiv \exp(\frac{2 \pi \sqrt{-1} z}{|\delta_n|})$, we get 
\begin{equation*}
J(z) = \left\{
\begin{split}
& w_{j} + \sum_{k=1}^{\infty} J^{(j)}_k  \iota_{j}^k, &\hspace{.5cm} j=n_e+1,\dots,n-1 ,\\
& \sum_{k=-1}^{\infty} J^{(j)}_k  \iota_{j}^k , \qquad & j=n.
\end{split}
\right.
\end{equation*}
As before, we can formally invert the above power series  ($j=n_e+1,\dots,n-1$)
\begin{equation*}
\left\{
\begin{split}
&\iota_{j} = \frac{1}{J_{1}^{(j)}} (w-w_{j}) - \frac{J_{2}^{(j)}}{\left(J_{1}^{(j)}\right)^3} (w-w_{j})^2\\
&\hspace{1.5cm} +  \frac{2 \left(J_{2}^{(j)}\right)^2 - J_{1}^{(j)} J_{3}^{(j)}}{\left(J_{1}^{(j)}\right)^5} (w-w_{j})^3 + \cdots,\hspace{2.7cm} w \to w_{j}, \\
& \iota_n = \frac{J_{-1}^{(n)}}{w} + \frac{J_{-1}^{(n)} J_{0}^{(n)}}{w^2} + \frac{ J_{-1}^{(n)} \left( \left( J_{0}^{(n)}\right)^2 - J_{-1}^{(n)} J_{1}^{(n)}\right)}{w^3} + \cdots,   \hspace{1.2cm} w \to \infty
\end{split}
\right.
\end{equation*}
and again using the definitions $\iota_{j \neq n}$ and $\iota_n$ , to see ($j=n_e+1,\dots,n-1$)
\begin{equation*}\label{JInversecusp}
J^{-1}(w) = \left\{
\begin{split}
& z_{j} - \frac{2 \pi \sqrt{-1}}{|\delta_{j}|} \bigg(\log(\frac{w-w_{j}}{J_1^{(j)}}) - \frac{J_2^{(j)}}{\left(J_1^{(j)}\right)^2} (w-w_{j})\hspace{1mm}+ \\
&\hspace{1cm}+ \frac{\frac{3}{2} \left(J_{2}^{(j)}\right)^2 - J_{1}^{(j)} J_{3}^{(j)}}{\left(J_{1}^{(j)}\right)^4} (w-w_{j})^2 + \cdots\bigg)^{-1}, \hspace{2.2cm}w \to w_{j},\\
&\frac{|\delta_n|}{2 \pi \sqrt{-1}} \left(\log(\frac{J_{-1}^{(n)}}{w}) + \frac{J_{0}^{(n)}}{w} + \frac{\frac{1}{2}\left( J_{0}^{(n)}\right)^2 - J_{-1}^{(n)} J_{1}^{(n)}}{w^2} + \cdots \right), \hspace{.2cm} w \to \infty.
\end{split}
\right.
\end{equation*}
Let us identify the accessory parameters as (see \emph{lemma 1} of  \cite{Zograf1988ONLE} and \emph{lemma 3} of  \cite{Takhtajan:2001uj})
\begin{equation*}
\left\{
\begin{split}
&c_{i} \equiv -\frac{h_i  J^{(i)}_2}{\left(J^{(i)}_1\right)^2}, \hspace{2.2cm} i=1,\dots,n_e,\\
& c_{i} \equiv - \frac{J^{(i)}_2}{\left(J^{(i)}_1\right)^2},  \hspace{2,1cm}i=n_e+1,\dots,n-1,\\
& c_n \equiv J_0^{(i)},  \hspace{3.1cm} i=n,
\end{split}
\right.
\end{equation*}
where $h_i/2 \equiv (1 - 1/m_i^2)/2$ is the \emph{conformal weight} of the \emph{twist operators} corresponding to branch points \cite{Knizhnik:1987xp}. Accordingly, we can summarize the asymptotic behavior of the $J^{-1}(w)$ near conical singularities and cusps ( $i=1,\dots,n_e$ ,  $j=n_e+1,\dots,n-1$)
\begin{equation}\label{JInverse}
J^{-1}(w) = \left\{
\begin{split}
& z_{i}  + 2 \sqrt{-1} \, \Im z_{i} \left(\frac{w-w_{i}}{J^{(i)}_1}\right)^{\frac{1}{m_i}} \left(1-\frac{c_{i}}{m_i h_i} (w-w_{i}) + \cdots \right), \hspace{1.2cm} w \to w_{i}, \\
& z_{j} - \frac{2 \pi \sqrt{-1}}{|\delta_{j}|} \left(\log(\frac{w-w_{j}}{J_1^{(j)}}) - c_{j} (w-w_{j}) + \cdots\right)^{-1}\hspace{-.5cm},\hspace{2.5cm} w \to w_{j},\\
&\frac{|\delta_n|}{2 \pi \sqrt{-1}} \left(\log(\frac{J_{-1}^{(n)}}{w}) + \frac{c_n}{w} + \cdots \right), \hspace{5.2cm} w \to \infty .
\end{split}
\right.
\end{equation}
Finally, we are ready to derive the asymptotic behavior of the Liouville filed $\varphi(w)$ and its derivatives  using Eq.\eqref{extra9} and the above expansion near the cones and cusps:
\begin{lemma}\label{lemma:asymptotics}
	The function $\varphi(w)$ has the following properties:
	\begin{enumerate}
		\item  $\partial_w^2 \varphi - \frac{1}{2} (\partial_w \varphi)^2 = \text{Sch}\left(J^{-1};w\right)$;
		\item
		For $i=1,\dots,n_e$ and $j=n_e+1,\dots,n-1$,
		\begin{equation*}
		\varphi(w) = \left\{
		\begin{split}
		& -2 (1-\frac{1}{m_i}) \log|w-w_{i}| + \log \frac{4|J^{(i)}_1|^{-\frac{2}{m_i}}}{m_i^2} + 	\Sorder{1} \hspace{1.4cm}w \to w_{i}\\
		& -2 \log |w-w_{j}| - 2 \log\left|\log\left|\frac{w-w_{j}}{J^{(j)}_{1}}\right|\right| + \Sorder{1} \hspace{1.9cm} w \to w_{j},\\
		&-2 \log|w| - 2 \log\log \left|\frac{w}{J^{(n)}_{-1}}\right| + \order{|w|^{-1}}, \hspace{2.8cm} w\to\infty
		\end{split}
		\right.
		\end{equation*}
		\item 
		For $i=1,\dots,n_e$ and $j=n_e+1,\dots,n-1$,
		\begin{equation*}
		\left\{
		\begin{split}
		&\frac{4|J^{(i)}_1|^{-\frac{2}{m_i}}}{m_i^2} = \lim_{w \to w_{i}} e^{\varphi(w)} |w-w_{i}|^{2-\frac{2}{m_i}},\\
		& \left|J^{(j)}_{1}\right|^2 = \lim_{w \to w_{j}} \exp(\log|w-w_{j}|^2 - \frac{2 e^{-\frac{\varphi(w)}{2}}}{|w-w_{j}|}),\\
		& \left|J^{(n)}_{-1}\right|^2 = \lim_{w \to \infty} \exp(\log|w|^2 - \frac{2 e^{-\frac{\varphi(w)}{2}}}{|w|}),
		\end{split}
		\right.
		\end{equation*}
		\item
		For $i=1,\dots,n_e$ and $j=n_e+1,\dots,n-1$,
		\begin{equation*}
		\partial_w \varphi(w) = \left\{
		\begin{split}
		&- \frac{1-\frac{1}{m_i}}{w-w_{i}}+ \frac{c_{i}}{1-\frac{1}{m_i}} + \order{1} \hspace{4.8cm}w \to w_{i},\\
		&-\frac{1}{w-w_{j}} \left(1+\left(\log\left|\frac{w-w_{j}}{J_{1}^{(j)}}\right|\right)^{-1}\right) + c_{j}  + \order{1} \hspace{1.2cm}w \to w_{j},\\
		& -\frac{1}{w} \left(1+\left(\log\left|\frac{w}{J_{-1}^{(n)}}\right|\right)^{-1}\right) - \frac{c_n}{w^2} + \order{\frac{1}{|w|^2}}, \hspace{1.5cm} w \to \infty
		\end{split}
		\right.
		\end{equation*}
		\item
		For $i=1,\dots,n_e$ and $j=n_e+1,\dots,n-1$,
		\begin{equation*}
		\partial_w^2 \varphi(w) = \left\{
		\begin{split}
		& \frac{1-\frac{1}{m_i}}{(w-w_{i})^2}+ \dotsm \hspace{8cm}w \to w_{i},\\
		& \frac{1}{(w-w_{j})^2} + \frac{1}{(w-w_{j})^2 \log\left|\frac{w-w_{j}}{J_{1}^{(j)}}\right|}+ \frac{1}{(w-w_{j})^2 \log^2\left|\frac{w-w_{j}}{J_{1}^{(j)}}\right|} + \dotsm\\
		&\hspace{10.6cm}w \to w_{j},\\
		&\frac{2 c_n}{w^3} + \frac{1}{w^2} + \frac{1}{w^2 \log\left|\frac{w}{J_{-1}^{(n)}}\right|}+ \frac{1}{w^2 \log^2\left|\frac{w}{J_{-1}^{(n)}}\right|} + \dotsm, \hspace{2.5cm}w \to \infty
		\end{split}
		\right.
		\end{equation*}
	\end{enumerate}
\end{lemma}
\begin{proof}
	Parts 2-5 follow from \eqref{extra9} and \eqref{JInverse}. Also, the part 1 follows from parts 2-5 and the definition of \emph{Schwarzian derivative} in \eqref{Sch}.
\end{proof}
\begin{remark}\label{extraa}
	In this paper, we sometimes need to study the behavior of objects in the case when the point at the infinity is a conical point. So it is good to clear things out in that case as well. In this remark, we assume that the point $w_n$ is a conical point at infinity. Instead of \eqref{Hauptmodulecone} we have:
	\begin{equation*} 
	J(z) = \sum_{k=-1}^{\infty} J^{(n)}_k  \left(\mathsf{u}_{\UD}^{m_n}\right)^k.
	\end{equation*}
	By inverting this expansion, we arrive at
	\begin{equation*}
	\begin{aligned}
	\mathsf{u}_{\UD}&=\left(J_{-1}^{(n)}\right)^{\frac{1}{m_n}}\left(\frac{1}{w}\right)^{\frac{1}{m_n}}+\frac{\left(J_{-1}^{(n)}\right)^{\frac{1}{m_n}}J_0^{(n)}}{m_n}\left(\frac{1}{w}\right)^{1+\frac{1}{m_n}}+\cdots,\\
	\end{aligned}
	\end{equation*}
	which, according to the  \emph{lemma 3} of \cite{Takhtajan:2001uj}, it implies that in the case of the conical point at the infinity, the accessory parameter becomes
	\begin{equation}
	c_n \equiv h_n J_0^{(n)},
	\end{equation}
	with $h_n\equiv 1-1/m_n^2$.
	Next point of interest is the part 3 of lemma \ref{lemma:asymptotics} which implies that for conical points $w_i$ that are not at infinity, we have
	\begin{equation*}
	\Lponetial_{i}=|J^{(i)}_1|^{\frac{2}{m_i}}=\frac{4}{m^2_i}\lim_{w\rightarrow w_i}e^{-\varphi(w)} |w-w_{i}|^{-2+\frac{2}{m_i}}.
	\end{equation*} 
	This expression is also subject to a slight change for $\Lponetial_{i}$ when $w_i\to\infty$. Considering $\mathsf{u}_{\UD}=(z-z_i)/(z-\bar{z}_i)$ and $J^{-1}(w)=z$ to write $\mathsf{u}_{\UD}=(J^{-1}(w)-z_i)/(z-\bar{z}_i)$ implies that
	\begin{equation*}
	J^{-1}(w)=z_i+2\sqrt{-1}\Im z_i\left(\left(J_{-1}^{(n)}\right)^{\frac{1}{m_n}}\left(\frac{1}{w}\right)^{\frac{1}{m_n}}+\frac{\left(J_{-1}^{(n)}\right)^{\frac{1}{m_n}}J_0^{(n)}}{m_n}\left(\frac{1}{w}\right)^{1+\frac{1}{m_n}}+\cdots\right).
	\end{equation*}
	Now from \eqref{extra9} and the equation above we have
	\begin{equation*}
	e^{\varphi(w)}=\frac{4}{m_n^2}\left(J_{-1}^{(n)}\right)^{\frac{2}{m_n}}\left(\frac{1}{w}\right)^{2\left(1+\frac{1}{m_n}\right)}+\cdots.
	\end{equation*}
	Accordingly, 
	\begin{equation}
	\Lponetial_{n}=\left|J_{-1}^{(n)}\right|^{\frac{2}{m_n}}=\frac{m^2_n}{4}\lim_{w\rightarrow \infty}e^{\varphi(w)} |w|^{2(1+\frac{1}{m_n})}.
	\end{equation}
\end{remark}
\subsection{Quadratic and Beltrami Differentials}\label{Apx:QuadraticBeltrami}
In this appendix, we will study the behavior of quadratic differentials, i.e., cusp forms of weight 4, and harmonic Beltrami differentials, automorphic forms of weight (-2,2), near elliptic and parabolic fixed points and connect them with functions on the Riemann orbisurface $O$ (see \cite{Kra_1972,Lehner_1964} for more details). Let us begin by studying the behavior of quadratic differential $q(z) \in \Hilbert^{2,0}(\UHP,\Gamma)$ as $z \to \infty$: Parabolic generator $\kappa_{n_p} \in \Gamma$ with the fixed-point $z_n = \infty$ has the following normal form (see \eqref{generators}):
\begin{equation*}
\kappa_{n_p}(z) = z + \delta_n,
\end{equation*}
i.e., it is a translation by $\delta_n$. Therefore, the equation $q(\kappa_{n_p} z) \kappa'_{n_p}(z)^2 = q(z)$ satisfied by the quadratic differential $q(z)$ can be translated as
\begin{equation}\label{quaddiffinfty}
q(z+\delta_n) = q(z).
\end{equation}
It is then easy to verify that the following Fourier series expansion
\begin{equation}
q(z) = \sum_{k=1}^{\infty} q_k^{(n)} \exp(\frac{2 \pi \sqrt{-1} k z}{|\delta_n|}),
\end{equation}
satisfies the equation \eqref{quaddiffinfty} and therefore represents the behavior of $q(z)$ near $z_n = \infty$. In order to study the behavior of $q(z)$ near the fixed points of other parabolic generators $\kappa_{1},\dots,\kappa_{n_p-1}$, we note that the M\"{o}bius transformation
\begin{equation*}
\PSLR \ni \varsigma_{n_e+j} (z) = \frac{1}{z-z_{n_e+j}}, \qquad j=1,\dots,n_p-1,
\end{equation*} 
sends that fixed-point to $\infty$ and we have $\varsigma_{n_e+j}^{-1}\hspace{1mm} \kappa_{j}\hspace{1mm} \varsigma_{n_e+j} (z) = z + \delta_{n_e+j}$. As a result, the normal form of the parabolic generators $\kappa_{j}$ for $j=1,\dots, n_p-1$ is given by
\begin{equation*}
\frac{1}{\kappa_{j}(z)-z_{n_e+j}} = \frac{1}{z-z_{n_e+j}} + \delta_{n_e+j},
\end{equation*} 
and we have that $q(z)$ satisfies
\begin{equation*}
q\left(z_{n_e+j} + \frac{(z-z_{n_e+j})}{1+ \delta_{n_e+j} (z-z_{n_e+j})}\right) \frac{1}{\left(1+\delta_{n_e+j}(z-z_{n_e+j})\right)^4} = q(z), 
\end{equation*}
near parabolic points $z_{n_e+1},\dots,z_{n-1}$. Therefore, as $z \to z_{i}$ the quadratic differential has the following Fourier series expansion
\begin{equation}
q(z) = \frac{1}{(z-z_{i})^4} \sum_{k=1}^{\infty} q_k^{(i)} \exp(-\frac{2 \pi \sqrt{-1} k}{|\delta_{i}|(z-z_{i})}), \qquad  i=n_e+1,\dots,n-1
\end{equation}
To study the behavior of differentials near the elliptic fixed points $z_{1},\dots,z_{n_e}$, it is easier to first use the unit disk model of the hyperbolic plane. Let $\tau_i \in \Gamma$ be an elliptic generator of order $m_i$ with fixed points $z_{i}\in \UHP$ and $\bar{z}_{i} \in \bar{\UHP}$. Then, the M\"{o}bius transformation
\begin{equation}\label{UHPtoUD}
\PSLR \ni \varsigma_{i} (z) = \frac{z-z_{i}}{z-\bar{z}_{i}}, \qquad i=1,\dots,n_e,
\end{equation} 
sends the points $(z_{i},\bar{z}_{i})$ to $(0,\infty)$ and therefor $\varsigma_{i}^{-1} \hspace{.5mm}\tau_i \hspace{.5mm}\varsigma_{i}(z)$ fixes both $0$ and $\infty$ and should be given by $\varsigma_{i}^{-1}\hspace{.5mm} \tau_i \hspace{.5mm}\varsigma_{i}(z) = \lambda_i z$. For elliptic generators, the \emph{multiplier} $\lambda_i$ can be determined to be the $m_i$-th primitive root of unity, i.e. $\lambda_i = \exp(2 \pi \sqrt{-1}/m_i)$,
via the use of condition $\tau_i^{m_i} = \mathbbm{1}$. Therefore, the normal form of the elliptic generators $\tau_i$ of order $m_i$ is given by
\begin{equation*}
\frac{\tau_i(z)-z_{i}}{\tau_i(z)-\bar{z}_{i}} = e^{\frac{2\pi \sqrt{-1}}{m_i}} \frac{z-z_{i}}{z-\bar{z}_{i}}, \qquad i= 1, \dots, n_e.
\end{equation*}
The M\"{o}bius transformations \eqref{UHPtoUD} give the standard isomorphism $\varsigma_{i}:\UHP \to \UD$; let us denote the coordinate on $\UD$ by $\mathsf{u}_{\UD} = \varsigma_{i}(z)$ and the push-forward of the density of Poincar\'{e} metric, $\rho(z) = (\Im z)^{-2}$, by $\rho(\mathsf{u}_{\UD}) = 4 (1-|\mathsf{u}_{\UD}|^2)^{-2}$. Similarly, if we denote the push-forward of $q(z) \in \Hilbert^{2,0}(\UHP,\Gamma)$ with $q(\mathsf{u}_{\UD}) \in \Hilbert^{2,0}(\UD,\Gamma)$, it satisfies 
\begin{equation*}
q(\lambda_i \mathsf{u}_{\UD})
\hspace{.5mm}\lambda_i^2 = q(\mathsf{u}_{\UD}).
\end{equation*}
It is easy to see that the solution to the above equation has the following power series form in $\mathsf{u}_{\UD}$:
\begin{equation}\label{quaddifdisk}
q(\mathsf{u}_{\UD}) = \sum_{k=1}^{\infty} q^{(i)}_k \mathsf{u}_{\UD}^{k m_i - 2} \qquad \text{as} \quad \mathsf{u}_{\UD} \to 0,
\end{equation}
since $\lambda_i^{m_i}=1$. We can also find the pullback of $q(\mathsf{u}_{\UD}) \in \Hilbert^{2,0}(\UD,\Gamma)$ to $\Hilbert^{2,0}(\UHP,\Gamma)$ by using $q(\mathsf{u}_{\UD}) \dd{\mathsf{u}_{\UD}}^2 = q(z) \dd{z}^2$ -- i.e. $q(z) = q(\mathsf{u}_{\UD}) \left(\dv{\mathsf{u}_{\UD}}{z} \right)^2$,
\begin{equation}
q(z) = - \frac{4 (\Im z_{i})^2}{(z - \bar{z}_{i})^4} \sum_{k=1}^{\infty} q^{(n+j)}_k \left(\frac{z-z_{i}}{z - \bar{z}_{i}}\right)^{k m_i -2} \qquad i=1,\dots,n_e,
\end{equation}
which satisfies the equation
\begin{equation*}
q\left(\frac{- z  \, z_{i} + \bar{z}_{i} \left(\lambda_i z + z_{i} +\lambda_i z_{i}\right)}{z(\lambda_i-1) - \lambda_i z_{i} + \bar{z}_{i}}\right) \hspace{1mm} \lambda_i^2 \frac{\left(\frac{- z  \, z_{i} + \bar{z}_{i} \left(\lambda_i z + z_{i} +\lambda_i z_{i}\right)}{z(\lambda_i-1) - \lambda_i z_{i} + \bar{z}_{i}} - \bar{z}_{i} \right)^4}{(z-\bar{z}_{i})^4} = q(z).
\end{equation*}
We can summarize the behavior of $q(z) \in \Hilbert^{2,0}(\UHP,\Gamma)$ near branch points and cusps as
\begin{equation}\label{q}
q(z) = \left\{
\begin{split}
& - \frac{4 (\Im z_{i})^2}{(z - \bar{z}_{i})^4} \, \sum_{k=1}^{\infty} q^{(i)}_k \left(\frac{z-z_{i}}{z - \bar{z}_{i}}\right)^{k m_i -2} \hspace{.7cm}(i=1,\dots,n_e),\hspace{.5cm}z \to z_{i},\\
& \frac{1}{(z-z_{i})^4} \, \sum_{k=1}^{\infty} q_k^{(i)} \exp(-\frac{2 \pi \sqrt{-1} k}{|\delta_i|(z-z_{i})}) \hspace{.5cm} (i=n_e+1,\dots,n-1),\hspace{.1cm}z \to z_{i},\\
&\sum_{k=1}^{\infty} q_k^{(n)} \exp(\frac{2 \pi \sqrt{-1} k z}{|\delta_n|}),\hspace{4cm} z \to \infty.
\end{split}
\right.
\end{equation}
Using the complex anti-linear isomorphism $\Hilbert^{2,0}(\UHP,\Gamma) \cong \Hilbert^{-1,1}(\UHP,\Gamma)$ given by $q(z) \mapsto \mu(z) = \rho(z)^{-1} \, \overline{q(z)}$,\footnote{Equivalently, $q(\mathsf{u}_{\UD}) \mapsto \mu(\mathsf{u}_{\UD}) = \rho(\mathsf{u}_{\UD})^{-1} \, \overline{q(\mathsf{u}_{\UD})}$ induces the isomorphism $\Hilbert^{2,0}(\UD,\Gamma) \cong \Hilbert^{-1,1}(\UD,\Gamma)$.}  we have 
\begin{equation}\label{Beltramiasymptotics}
\mu(z) = \left\{
\begin{split}
& -\frac{4 (\Im z)^2 (\Im z_{i})^2}{(\bar{z} - z_{i})^4} \, \sum_{k=1}^{\infty} \bar{q}_k^{(i)} \left(\frac{\bar{z}-\bar{z}_{i}}{\bar{z} - z_{i}}\right)^{k m_i -2} \hspace{.5cm}(i=1,\dots,n_e),\hspace{.5cm}z \to z_{i},\\
& \frac{(\Im z)^2}{(\bar{z}-z_{i})^4} \, \sum_{k=1}^{\infty} \bar{q}_k^{(i)}\exp(\frac{2 \pi \sqrt{-1} k}{|\delta_{i}|(\bar{z}-z_{i})}) \hspace{.8cm}(i=n_e+1,\dots,n-1),\hspace{.1cm}z \to z_{i},\\
& (\Im z)^2 \, \sum_{k=1}^{\infty} \bar{q}_k^{(n)} \exp(-\frac{2 \pi \sqrt{-1} k \bar{z}}{|\delta_n|}), \hspace{3.5cm} z \to \infty,\\
\end{split}
\right.
\end{equation}
where we have used the fact that $\bar{z}_{n_e+j} = z_{n_e+j} \in \mathbb{R}$ for $j=1,\dots,n_p-1$. The mapping
\begin{equation}
q(z) \mapsto Q(w) = (q \circ J^{-1})(w) \,  \left(J^{-1}(w)'\right)^{2},
\end{equation} 
determines a linear isomorphism of spaces $\Hilbert^{2,0}(\UHP,\Gamma)$ and $\Hilbert^{2,0}(O)$ and the inverse mapping is given by 
\begin{equation}
Q(w) \mapsto q(z) = (Q \circ J)(z) \, J'(z)^2.
\end{equation}
The Petersson inner product in $\Hilbert^{2,0}(\UHP,\Gamma)$ can be carried over to $\Hilbert^{2,0}(O)$ by setting $\langle Q_1, Q_2 \rangle \equidef  \langle Q_1\circ J J'^2, Q_2 \circ J J'^2 \rangle$ for all $Q_1, Q_2 \in \Hilbert^{2,0}(O)$, so that
\begin{equation}\label{Qinnerprod}
\langle Q_1, Q_2 \rangle = \iint_{\cmpx} e^{-\varphi(w)} Q_1(w)\, \overline{Q_2(w)} \dd[2]{w}.
\end{equation}
It follows from Eq.\eqref{q} and Eq.\eqref{JInverse} that 
\begin{equation*}
\begin{split}
Q(w) & = \sum_{k=1}^{\infty} q_k^{(n)} \exp(\frac{2 \pi \sqrt{-1} k}{|\delta_n|} \times \frac{|\delta_n|}{2 \pi \sqrt{-1}} \left(\log(\frac{J_{-1}^{(n)}}{w}) + \frac{c_n}{w} + \cdots \right))\times\\
& \hspace{4cm} \times \left(\frac{|\delta_n|}{2 \pi \sqrt{-1}} \left(-\frac{1}{w} - \frac{c_n}{w^2} + \cdots \right)\right)^2 \\
& \simeq \left(-\frac{|\delta_n|^2}{4 \pi^2} \left(\frac{1}{w^2} + \frac{2 c_n}{w^3} + \frac{c_n^2}{w^4} + \cdots \right)\right) \sum_{k=1}^{\infty} q_k^{(n)} \left(\frac{J_{-1}^{(n)}}{w}\right)^k\\
& = - \frac{|\delta_n|^2 q_1^{(n)} J_{-1}^{(n)}}{4\pi^2} \cdot \frac{1}{w^3} - \frac{|\delta_n|^2 J_{-1}^{(n)} \big(2 c_n q_1^{(n)} + q_2^{(n)} J_{-1}^{(n)}\big)}{4 \pi^2} \cdot \frac{1}{w^4} + \order{|w|^{-5}} \hspace{.2cm}\text{as} \quad w \to \infty,
\end{split}
\end{equation*}
and
\begin{equation*}
\begin{split}
Q(w) &= \left(- \frac{2 \pi \sqrt{-1}}{|\delta_{i}|} \frac{1}{\log(\frac{w-w_{i}}{J_1^{(i)}}) - c_{i} (w-w_{i}) + \cdots} \right)^{-4} \times \\
&\hspace{-.7cm} \times\sum_{k=1}^{\infty} q_k^{(i)} \exp(-\frac{2\pi \sqrt{-1} k}{|\delta_{i}|}\times - \frac{|\delta_{i}|}{2 \pi \sqrt{-1}} \left(\log(\frac{w-w_{i}}{J_1^{(i)}}) - c_{i} (w-w_{i}) + \cdots \right))\times \\
&\hspace{-.7cm} \times \left(- \frac{2 \pi \sqrt{-1}}{|\delta_{i}|} \frac{-1}{(w-w_{i})\log^2(\frac{w-w_{i}}{J_1^{(i)}})} + \cdots \right)^2\\
&\hspace{-.7cm}=- \frac{|\delta_{i}|^2 q_1^{(i)}}{4 \pi^2 J_1^{(i)}} \cdot \frac{1}{w-w_{i}} + \order{|w-w_{i}|^{-2}}, \hspace{.5cm}(i=n_e+1,\dots,n-1) \hspace{.1cm}\text{as}\hspace{.1cm} w \to w_{i}.
\end{split}
\end{equation*}
Similarly, near branch points $w_{1}, \dots, w_{n_e}$ we have
\begin{equation}
Q(w) = (q \circ \tilde{J}^{-1})(w) \left(\tilde{J}^{-1}(w)'\right)^2,
\end{equation}
where $\tilde{J}^{-1}: O \to \UD$ is the inverse of Klien's Hauptmodule in the unit disk model of the hyperbolic plane. Using equations \eqref{Jtildeinv} and \eqref{quaddifdisk}, we get
\begin{equation*}
\begin{split}
Q(w) &= \sum_{k=1}^{\infty} q^{(i)}_k \left( \left(\frac{1}{J^{(i)}_1 }\right)^{\frac{1}{m_i}} (w-w_{i})^{\frac{1}{m_i}} - \frac{J^{(i)}_2}{m_i \left(J^{(i)}_1\right)^{2+\frac{1}{m_i}}} (w-w_{i})^{1+\frac{1}{m_i}} + \cdots \right)^{k m_i - 2}\times\\
&\hspace{-.5cm}\times \left( m_i^{-1} \left(\frac{1}{J^{(i)}_1 }\right)^{\frac{1}{m_i}} (w-w_{i})^{-1+\frac{1}{m_i}} - (1+\frac{1}{m_i}) \frac{J^{(i)}_2}{m_i \left(J^{(i)}_1\right)^{2+\frac{1}{m_i}}} (w-w_{i})^{\frac{1}{m_i}} + \cdots \right)^{2} \\
& = \frac{q^{(i)}_1}{m_i^2 J^{(i)}_1} \cdot \frac{1}{w-w_{i}} + \order{1} \hspace{1cm} (i=1,\dots,n_e), \hspace{.1cm}\text{as}\hspace{.2cm} w \to w_{n+j}.
\end{split}
\end{equation*}
We can summarize the behavior of $Q(w) \in \Hilbert^{2,0}(O)$ near conical singularities and  punctures as follows
\begin{equation}
Q(w) = \left\{
\begin{split}
&\frac{q^{(i)}_1}{m_i^2 J^{(i)}_1} \cdot \frac{1}{w-w_{i}} + \order{1} \hspace{3cm} (i=1,\dots,n_e),\hspace{.5cm} w \to w_{i} ,\\
& - \frac{|\delta_{i}|^2 q_1^{(i)}}{4 \pi^2 J_1^{(i)}} \cdot \frac{1}{w-w_{i}} + \order{|w-w_{i}|^{-2} }, \hspace{.2cm} (i=n_e+1,\dots,n-1), \hspace{.4cm}w \to w_{i},\\
&- \frac{|\delta_n|^2 q_1^{(n)} J_{-1}^{(n)}}{4 \pi^2} \cdot \frac{1}{w^3} + \order{|w|^{-4}} \hspace{4.5cm}w \to \infty .
\end{split}
\right.
\end{equation}
Finally, using the complex anti-linear isomorphism $\Hilbert^{2,0}(O) \cong \Hilbert^{-1,1}(O)$ given by $Q(w) \mapsto M(w) = e^{-\varphi(w)} \, \overline{Q(w)}$, we have 
\begin{equation}\label{MAsymptotic}
M(w) = \left\{
\begin{split}
& \frac{\bar{q}^{(i)}_1}{4 \bar{J}^{(i)}_1} \cdot |w - w_{i}|^{1-\frac{2}{m_i}}  + \order{|w - w_{i}|} \hspace{1.5cm} (i=1,\dots,n_e),\hspace{.2cm} w \to w_{i},\\
& - \frac{|\delta_{i}|^2 \bar{q}_1^{(i)}}{4\pi^2 \bar{J}_1^{(i)}} \cdot |w-w_{i}| \log^2 |w-w_{i}|+ \order{\log^2 |w-w_{i}|} \\
&\hspace{6.9cm} (i=n_e+1,\dots,n-1), \hspace{.2cm}w \to w_{i},\\
& - \frac{|\delta_n|^2 \bar{q}_1^{(n)} \bar{J}_{-1}^{(n)}}{4 \pi^2} \cdot \frac{\log^2|w|}{|w|} + \order{|w|^{-2} \log^2|w|} \hspace{1cm} w \to \infty.
\end{split}
\right.
\end{equation}
Using the definition \eqref{Mdef} as well as equations \eqref{JInverse} and \eqref{Beltramiasymptotics}, one can independently check the above asymptotic behaviors for $M(w)$.
\vspace{3cm}

\section{List of symbols in the main text}\label{symbol}
\begin{longtable}{| p{.12\textwidth} | p{.88\textwidth} |}
	\hline 
	$g$ & Genus of a corresponding surface, or rank of a Schottky group\\
	$n_e$ & Number of elliptic fixed points or branch points\\
	$n_p$ & Number of parabolic fixed points or punctures\\
	$n$ & Number of all marked points, namely $n_e+n_p$\\
	$m_i$ & Orders of the marked points\\
	$\boldsymbol{m}$ & $n$-tuple of all branching indices for an orbifold, i.e. $(m_1,\dots,m_n)$\\
	$h_i/2$ & conformal weights corresponding to the marked points, i.e. $1-\frac{1}{m_i^2}$\\
	$\orderrange$ & Natural numbers with the inclusion of $\infty$ and omission of $1$\\
	$\mathrm{D}$ & Unit disk in $\cmpx$\\
	$X$ & Riemann surface without singularities\\
	$O$ & Orbifold Riemann surface\\
	$O^\mu$ & Orbifold Riemann surface deformed by $\mu$\\
	$X_O$ & Underlying  Riemann surface of $O$\\
	$X_O^{\text{reg}}$ & Regular locus of $O$\\
	$\hat{X}_O$ & Compactified underlying Riemann surface\\  
	$K$ & Kleinian group\\
	$\Gamma$ & Fuchsian group\\
	$\Gamma^\mu$ & Fuchsian group deformed by $\mu$\\
	$\Sigma$ & Schottky group\\
	$N$ & Smallest normal subgroup of $\Gamma$ containing $\{\alpha_1, \dots, \alpha_g ,\tau_1, \dots, \tau_{n_e},\kappa_1, \dots, \kappa_{n_p}\}$\\
	$\pi_1(O,x_{*})$ & Fundamental group of an orbifold based at $x_{*}$\\
	$\mathbb{Z}_{m}$ & Cyclic group of order $m$\\
	$\Automorphism_{\ast}(\Gamma)$ & Group of proper automorphisms of $\Gamma$\\ 
	$\Inn(\Gamma)$ & Group of inner automorphisms\\
	$\operatorname{MCG}(O)$ & Mapping class group, namely $\operatorname{Homeo}^{+}(O)/\operatorname{Homeo}_{_{\operatorname{id}}}^{+} (O)$\\
	$\modular(\Gamma)$ & Teichm\"{u}ller modular group of $\Gamma$, i.e. $\Automorphism_{\ast}(\Gamma)/\Inn(\Gamma)=\operatorname{Out}^{+}(\Gamma)$\\
	$\operatorname{Homeo}^{+}(O)$ & Group of orientation preserving homeomorphisms of $O$ in the category of orbifolds, which has $\operatorname{Homeo}_{_{\operatorname{id}}}^{+} (O)$ as its identity component\\ 
	$\operatorname{Out}^{+}(\Gamma)$ & Group of outer automorphisms of $\Gamma \simeq \pi_1(O)$\\ 
	$\operatorname{MCG}_{0}(O)$ & Group of pure mapping classes of $O$\\
	$\puremod(\Gamma)$ & Group of pure mapping classes of $\Gamma$\\  
	$\symm{\s_i}$ & Symmetric group associated with the stratum of order $m_i$\\
	$\symm{\sigtype}$ & Product symmetric group $\symm{\s_2} \times \symm{\s_3} \times \dotsm \times \symm{\s_{\infty}}$\\
	$\operatorname{deck}$ & Group of covering or deck transformations\\
	$H_1$ & First homology group\\
	$\Omega$ & Region of discontinuity of Schottky group $\Sigma$\\
	$\Omega_{0}$& Region of discontinuity with pre-images of cusps subtracted\\
	$\singrigon$ &  $(\Omega_{0},\widetilde{\brdiv})$\\
	$\Omega^{\text{reg}}$ & $\Omega_0 \backslash \text{Supp}(\widetilde{\brdiv})$\\
	$\Lambda$ & Limit set of a Kleinian group, i.e. $\hat{\cmpx}\backslash\Omega$\\
	\hline
	$\fund$ & Fundamental domain of a Fuchsian group\\
	$\SchottkyFund$ & Fundamental domain of a Schottky group\\
	$\SchottkyFund_0$ & $\SchottkyFund \cap \Omega_0$\\
	$\singfund$ &  $(\mathcal{D}_0,\widetilde{\brdiv}\big{|}_{_{\mathcal{D}}})$\\
	$\singfund_{\epsilon}$ & Regularized singular fundamental domain of a Schottky group defined by $\singfund \big\backslash \bigcup_{i=1}^{n} D_i^{\epsilon}$\\
	$\gamma$ & An element of a Fuchsian group\\
	$\sigma$ & An element of a Schottky group\\
	$\eta$ & An element of a symmetric group\\
	$\mathsf{A}_i$ & Homotopy class of loops around a handle\\
	$\mathsf{B}_i$ & Homotopy class of loops around a hole\\
	$\mathsf{C}_i$ & Homotopy class of loops around a branch point\\
	$\mathsf{P}_i$ & Homotopy class of loops around a puncture\\
	$\alpha_i,\beta_i$ & Hyperbolic generators of a Fuchsian group\\
	$\tau_i$ & Elliptic generator of a Fuchsian group\\
	$\kappa_i$ & Parabolic generator of a Fuchsian group\\
	$L_i$ & Generator of a Schottky group\\
	$\tilde{a}_i, \tilde{b}_i$ & attracting and repelling fixed points of the loxodromic element $L_i$\\
	$\lambda_i$ & multiplier of the loxodromic element $L_i$\\
	$\mathsf{a}_i$ & A retrosection of a Riemann surface\\
	$f_{\eta}$ & 1-cocycle associated with an element $\eta\in\symm{\sigtype}$\\
	$\sing(O)$ & Set of all singular points for an orbifold $O$\\
	$\sing_m(O)$ & Set of singular points of order $m$ for an orbifold $O$, i.e. $\nu^{-1}(m)$\\
	$\sing_{\curlywedge}(O)$ & Set of singular points of finite order for an orbifold $O$, i.e. $\bigsqcup_{m \neq \infty} \sing_m(O)$\\
	$\s_m$ & Cardinality of the stratum of singular points of order $m \in \orderrange$\\
	$\sigtype$ & Signature type of an orbifold $O$, which is the unordered set of cardinalities of all strata, i.e. $\{\s_m\}_{m \in \orderrange}$\\
	$\nu$ & Branching function which assigns to each singular point its corresponding branching order\\
	$x_i$ & Marked points on the Riemannian orbifold\\
	$z^{{}_e}_{i},z^{{}_p}_{i}$ & Elliptic and Parabolic fixed points of a Fuchsian group, respectively\\
	$x^{{}_c}_{i},z^{{}_p}_{i}$ & Images of the elliptic and Parabolic fixed points of a Fuchsian group, under the projection  $\UHP \to O = [\UHP/\Gamma]$, respectively\\
	$z$ & Coordinates on the upper half-plane\\
	$w$ & Global coordinates on $\Omega$\\
	$w_i$ & Singular points on an orbifold, namely branch points for $i=1,\cdots,n_e$, and cusps for $i=n_e+1,\cdots,n$\\
	$t$ & Bers coordinates\\
	$\varepsilon$ & Small complex parameter for variation\\
	$\epsilon$ & Small real parameter for regularization\\
	\hline
	$\mathscr{CM}_{\boldsymbol{\alpha}}(X)$ & Space of all smooth conformal metrics $e^{\psi(u,\bar{u})} |du|^2$ on $X\backslash\{x_1,\dots,x_n\}$ which have conical singularities of angles $2\pi(1- \alpha_i)$ at the insertion points\\
	$\mathscr{CM}(O)$ & Space of singular conformal metrics on $X_O$ representing $\brdiv$\\
	$\brdiv$ & Branch divisor\\
	$\widetilde{\brdiv}$ & Branch divisor corresponding to the lift of the original $\brdiv$ under the Schottky group\\
	$U_a$ & Open subset in $X_O$\\
	$u_a$ & Coordinate function on an orbifold chart\\
	$u_{\UD}$& Coordinate on unit disk\\
	$g_{ab}$ & Transition function between two orbifold coordinate charts, $U_a$ and $U_b$\\
	$\mathcal{C}^{\infty}$ & Smooth  functions\\
	$\varphi,\psi$ & (Classical) Liouville field\\
	$\Gpotential_{\boldsymbol{m}}$ & Liouville action functional for non-zero genus\\
	$S_{\boldsymbol{m}}$ & Liouville action functional for zero genus\\
	$\mathbb{F}_{O}$ & Free energy defined through the action functional on an orbifold\\
	$\Delta_0$ & Laplace operator\\
	$\pi_\Gamma$ & Orbifold covering map between $\UHP$ and an orbifold $O$ that provides it with the Fuchsian global coordinates\\
	$\pi_\Sigma$ & Orbifold covering map between $\Omega$ or $\singrigon$ and an orbifold $O$ which restricted to $\Omega^{\text{reg}}$ gives the Schottky global coordinates\\
	$J$ & Depending on the context, the orbifold covering map between $\UHP$ and $\Omega$ or $\singrigon$, namely Klein's Hauptmodul, or orbifold covering map between $\UHP$ and $O$\\
	$\mathfrak{j}$ & The epimorphism $\Gamma \to \Sigma$ which maps $\beta_i$s to $L_i$s and all $\alpha_i$s, $\kappa_j$s, and $\tau_k$s to $\mathbbm{1}$\\
	$J_\mu$ & Orbifold covering map between $\UHP$ and an underlying Riemann surface $O^\mu$\\
	$J_k^{(i)}$ & The $k^{th}$ order coefficient of $J$'s expansion around the $i^{th}$ marked point\\
	$\Lponetial_{i}$ & Smooth functions on $\moduli_{0,n}$ defined by  $\left| J^{(i)}_{1} \right|^{\frac{2}{m_i}}$ for $i=1,\dots,n_e$, $\left| J^{(i)}_{1} \right|^2$ for $i=n_e+1,\dots,n-1$, and $ \left| J^{(n)}_{-1} \right|^2$ for $i=n$\\
	$\varLambda$ & Complex anti-linear mapping between $\diffspace^{-1,1}(\UHP,\Gamma)$ and $\Hilbert^{2,0}(\UHP,\Gamma)$ defined through the Bergman integral\\
	$\schottky_{g}$ & Schottky space of genus $g$, i.e. the set of equivalence classes of representations $\left[\varrho_{\mu}\right] : \Sigma \to \PSLC$, $\mu \in \deformspace(\Sigma)$.\\
	$\schottky_{g,n}(\boldsymbol{m})$ & Generalized Schottky space of genus $g$ and $n$ marked points, corresponding to $\boldsymbol{m}=(m_1,\dots,m_n)$\\
	$\teich_{g,n}(\boldsymbol{m})$ & Teichm\"{u}ller space of marked Riemann orbisurfaces of genus $g > 1$, corresponding to $\boldsymbol{m}=(m_1,\dots,m_n)$\\
	$\teich(\Gamma)$ & Teichm\"{u}ller space defined as the set of equivalence classes of representations $\left[\varrho_{\mu}\right] : \Gamma \to \PSLR$, $\mu \in \deformspace(\Gamma)$\\
	$\mathscr{P}_{g,n}(\boldsymbol{m})$ & Affine bundle $\mathscr{P}_{g,n}(\boldsymbol{m}) \to \teich_{g,n}(\boldsymbol{m})$, that the Fuchsian projective connection $\operatorname{Sch}(\pi_{\Gamma}^{-1})$ gives a canonical section for it or $\mathscr{P}_{g,n}(\boldsymbol{m}) \to \schottky_{g,n}(\boldsymbol{m})$ which has the Schottky projective connection $\operatorname{Sch}(\pi_{\Sigma}^{-1})$ as a canonical section\\
	\hline
	$\symmoduli_{g,n}(\boldsymbol{m})$ & Moduli space of Riemannian orbifold with signature $(g;m_1,\dots,m_{n_e};n_p)$ that is isomorphic to $\teich(\Gamma) / \modular(\Gamma)$\\
	$\moduli_{g,n}$ & Moduli space of smooth complex algebraic curves of genus $g$ with $n$ labeled points, i.e. $\teich(\Gamma)/\puremod(\Gamma)$\\
	$R$ & Projective connection\\
	$r_a$ & Elements of a projective connection\\
	$\text{Sch}\left(f;z\right)$ & Schwartzian derivative of $f$ with respect to $z$\\
	$\Lie_{\mu},\Lie_{\bar{\mu}}$ & Lie derivative in holomorphic and anti-holomorphic tangential directions, $\mu$ and $\bar{\mu}$\\
	$Q$ & Quadratic differential\\
	$q_a$ & Holomorphic functions constructing a quadratic differential\\
	$\mu,\mu'$ & Beltrami differentials for a Fuchsian group\\
	$M$ & Functions defined by $ \left(\mu \circ J^{-1}\right) \frac{\overline{(J^{-1})'}}{(J^{-1})'}$, which construct the analogue of the Beltrami equation for $F^\mu$\\
	$\mu_i$ & Basis element for harmonic differentials in $\Hilbert^{-1,1}(\UHP,\Gamma)$\\
	$M_i$ & Basis element for harmonic differentials in $\Hilbert^{-1,1}(O)$\\
	$q_i$ & Basis element for quadratic differentials in $\Hilbert^{2,0}(\UHP,\Gamma)$\\
	$R_i$ & Linearly independent elements that generate the space $\Hilbert^{2,0}(O)$\\
	$Q_i$ & Basis element for quadratic differentials in $\Hilbert^{2,0}(O)$ biorthogonal to $R_i$\\
	$P_i$ & Basis element for quadratic differentials in $\Hilbert^{2,0}(\singrigon,\Sigma)$\\
	$\varE_i$ & Terms multiplying to conformal weights in the energy-momentum tensor expansion for zero genus\\
	$\mathscr{E}_i$ & Terms multiplying to conformal weights in the energy-momentum tensor expansion for non-zero genus\\
	$v_i$ & Normalized basis of the space of holomorphic 1-forms - abelian differentials of the first kind\\
	$\phi$ & An element of $\diffspace^{k,\ell}(\UHP,\Gamma)$\\
	$\mathsf{R}(w)$ & Projection of the automorphic form $\text{Sch}\left(J^{-1};w\right)$ of weight four for the Schottky group to the subspace $\Hilbert^{2,0}(\singrigon,\Sigma) \cong T_{\pi\circ\Phi(0)}^{\ast}\schottky_{g,n}(\boldsymbol{m})$\\
	$\mathrm{b}_i$ & Coefficients of the projection $\mathsf{R}(w)$\\
	$\operatorname{Proj_{\Hilbert^{k,\ell}}}$ & Projection operator onto $\Hilbert^{k,\ell}(\UHP,\Gamma)$\\
	$\diffspace^{-1,1}(\UHP,\Gamma)$ & Complex Banach space of the Beltrami differentials for $\Gamma$\\
	$\diffspace^{-1,1}(\Omega,\Sigma)$ & Complex Banach space of the Beltrami differentials for $\Sigma$\\
	$\deformspace(\Gamma)$ & Open ball of radius $1$ in $\diffspace^{-1,1}(\UHP,\Gamma)$\\
	$\deformspace(\Sigma)$ & Open ball of radius $1$ in $\diffspace^{-1,1}(\Omega,\Sigma)$\\
	$\Hilbert^{2,0}(\UHP,\Gamma)$ & Space of cusp forms of weight $4$ for the group $\Gamma$--equivalently, meromorphic (2,0)-tensors/quadratic differentials on the Riemann orbisurface $O$\\
	$\Hilbert^{2,0}(\singrigon,\hspace{-.7mm}\Sigma\hspace{-.5mm})$ & Meromorphic quadratic differential for Schottky group $\Sigma$\\
	$\Hilbert^{k,0}(\UHP,\Gamma)$ & cusp forms of weight $2k$ for $\Gamma$\\
	$\Hilbert^{-1,1}(\UHP,\Gamma)$ & Harmonic Beltrami differentials, that are a subspace $ \Lambda^{\ast}\left(\Hilbert^{2,0}(\UHP,\Gamma)\right)$ of $\diffspace^{-1,1}(\UHP,\Gamma)$\\
	\hline
	$\Hilbert^{-1,1}(\singrigon,\hspace{-.7mm}\Sigma\hspace{-.5mm})$ &Space of harmonic Beltrami differentials with respect to the hyperbolic metric on $\singrigon$\\
	$\Hilbert^{2,0}(O)$ & Space of quadratic differentials on an orbifold\\
	$\Hilbert^{-1,1}(O)$ & Space of harmonic differentials on an orbifold\\
	$\diffspace^{k,\ell}(\UHP,\Gamma)$ & Smooth family of automorphic forms of weight $(2k,2\ell)$, or $(k,\ell)$-tensors on an orbifold $O$\\
	$\Phi$ & Mapping from $\deformspace(\Gamma)$ to $\teich_{g,n}(\boldsymbol{m})$\\
	$\Psi$ & Mapping from $\teich_{0,n}(\boldsymbol{m})$ to  $\cmpx^{n-3}$, defined by $(\Psi \circ \Phi)(\mu) = (w_1^{\mu},\dots,w_{n-3}^{\mu}) \in \cmpx^{n-3}$ \\
	$\pi$ & Mapping from $\teich_{g,n}(\boldsymbol{m})$ to $\schottky_{g,n}(\boldsymbol{m})$\\
	$\kerphi(\UHP,\Gamma)$ & Kernel of the differential $\dd{\Phi}$ at the point $0 \in \deformspace(\Gamma)$\\
	$\mathcal{P}(O)$ & Space of all $\CPl$-structures on a complex orbifold $O$\\
	$T_{\varphi},\bar{T}_{\varphi}$ & $(2,0)$ and $(0,2)$ components of the classical energy-momentum tensor on an orbifold $O$\\
	$T_{a},\bar{T}_{a}$ & Components of the classical energy-momentum tensor on each chart $U_a$\\
	$\mathscr{Q}$ & Difference between the Fuchsian and Schottky projective connections, i.e. $\operatorname{Sch}(\pi_{\Gamma}^{-1}) - \operatorname{Sch}(\pi_{\Sigma}^{-1})$\\
	$\omega_{_{\text{WP}}}$ & Symplectic form of the Weil-Petersson metric\\
	$\omega_{\text{TZ},i}^{\text{cusp}}$ & Symplectic form of $i^{th}$-cuspidal Takhtajan-Zograf metric \\
	$\omega_{\text{TZ},i}^{\text{ell}}$ & Symplectic form of $i^{th}$-elliptic Takhtajan-Zograf metric \\
	$\langle\cdot,\cdot\rangle_{\text{WP}}$ & Petersson inner product on $T_{[\mu]} \teich_{g,n}(\boldsymbol{m}) \cong \Hilbert^{-1,1}(\UHP,\Gamma^{\mu})$ \\
	$\langle\cdot ,\cdot \rangle^{\text{cusp}}_i$ & $i^{th}$-cuspidal Takhtajan-Zograf (TZ) inner product\\
	$\langle \cdot, \cdot\rangle^{\text{cusp}}$ & The invariant inner product under $\modular(\Gamma)$, i.e. $\langle \cdot , \cdot \rangle^{\text{cusp}} = \langle \cdot , \cdot \rangle^{\text{cusp}}_1 + \cdots + \langle \cdot , \cdot \rangle^{\text{cusp}}_{n_p}$\\
	$\langle \cdot,\cdot \rangle^{\text{ell}}_{i}$ & $i^{th}$-elliptic Takhtajan-Zograf (TZ) inner product\\
	$\left(||\cdot||_{1}^{\text{Quill}}\right)^
	{2}$ & Quillen metric on $\lambda_1$\\
	$||\cdot||_{k}^{\text{Quill}}$ & Quillen norm on $\lambda_k$\\
	$\langle \cdot , \cdot \rangle_{_{Sch}}$ & Hermitian metric on $\lambda_{Sch}$\\
	
	$\lambda_H$ & Hodge line bundle over $\symmoduli_{g,n}$\\
	$\lambda_{0,n}(\sigtype)$ & Holomorphic line bundle over $\symmoduli_{0,n}$\\
	$\lambda_{Sch}$ & Holomorphic $\mathbb{Q}$-line bundle over $\symmoduli_{g,n}$\\
	$\lambda_1$ & Hodge line bundle, i.e. the determinant line bundle associated with the Cauchy-Riemann operator $\bar{\partial}_1$\\
	$\lambda_k$ & Determinant line bundle associated with the Cauchy-Riemann operator $\bar{\partial}_k$\\
	$\lambda_{_{Hod}}$ & Hodge line bundle $(\lambda_{_{Hod}} ; \langle \cdot , \cdot \rangle_{_{Quil}})$\\
	
	$f^\mu$ & Solution of Beltrami equation corresponding to a Beltrami differential $\mu$\\
	$F^\mu$ & Quasi-conformal mapping that satisfies $F^{\mu} \circ J = J_{\mu} \circ f^{\mu}$\\
	$\dot{f}_{+}^{\mu},\dot{f}_{-}^{\mu}$ & The derivatives $\left.\left(\pdv{\varepsilon} f^{\varepsilon \mu}\right) \right|_{\varepsilon=0}$, $ \left.\left(\pdv{\bar{\varepsilon}} f^{\varepsilon \mu}\right) \right|_{\varepsilon=0}$, respectively\\
	$\operatorname{hol}$ & Holonomy representation\\
	$\varrho_{\mu_i}$ & Representations of Fuchsian or Schottky group corresponding to a Beltrami differential $\mu_i$\\
	$\rho(z)$ & Density of a hyperbolic metric on $\UHP$\\
	$E_i(z,s)$ & Eisenstein-Mass series associated with the cusp $z_{n_e+i}$\\
	\hline
	$\mathsf{H}$ & Hermitian metric defined on the holomorphic line bundle $\lambda_{0,n}(\sigtype)$ defined by $\Lponetial_{1}^{m_1h_1}\dotsm\Lponetial_{n_e}^{m_{n_e} h_{n_e}}\Lponetial_{n_e + l}\dotsm \Lponetial_{n-1} \Lponetial_{n}^{-1}$ for zero genus and, $\Lponetial_{1}^{m_1 h_1} \dotsm \Lponetial_{n_e}^{m_{n_e} h_{n_e}} \Lponetial_{n_e+1} \dotsm \Lponetial_{n}$ for non-zero genus\\
	$c_i$ & Accessory parameters\\
	$C_{i}^{\epsilon}$ & Regularization circles around branch points and punctures\\
	$C_i, C'_i$ & Jordan curves constructing the boundary of the Schottky fundamental domain\\
	$\mathrm{D}_i^\epsilon$ & Disks of radius $\epsilon$ around branch points and punctures\\
	$O_\epsilon$ & The region outside the regularizing disks, i.e. $\cmpx \Big\backslash \bigcup_{i=1}^{n-1}\left\{w \, \Big| \, |w-w_i|<\epsilon\right\} \cup \left\{w \, \Big| \, |w|>\epsilon^{-1}\right\}$\\ 
	$\jmath$ & Holomorphic fibration between $\schottky_{g,n}(\boldsymbol{m})$ and $\schottky_{g}$ whose fibers are configuration spaces of $n$ labeled points\\
	$\confspace{n}{\hat{\cmpx}}$ & Configuration space of complex $n$-tuples in $\hat{\cmpx}$\\
	$\linebundle_i$ & $i^{th}$ relative dualizing sheaf on $\schottky_{g,n}(\boldsymbol{m})$ or the $i^{th}$ tautological line bundle\\
	$\linebundle$ & $\mathbb{Q}$-line bundle defined by $\bigotimes_{i=1}^{n} \linebundle_i^{\Lponetial_{i}}$\\
	$\mathsf{c}_1$ & First Chern class\\
	$c_1(E,\nabla)$ & First Chern form of a vector bundle $E$\\
	$\theta_{L_k^{-1}}$ & Boundary 1-form of Takhtajan-Zograf action\\
	$l_k$ & Left-hand lower element in the matrix representation of the generator $L_k \in \PSLC$ for $k=2,\dots,g$\\
	$\deg(\mathcal{D})$ & Degree of a divisor\\
	$[\cdot/\cdot]$ & Quotient as an analytic orbifold/stack\\
	$\boldsymbol{\tau}$ & Period matrix\\
	$\det^{\prime}\Delta_0$ & Zeta-function regularized determinant of the Laplace operator in the hyperbolic metric $\exp(\varphi)|dw|^2$ acting on functions\\
	$\mathfrak{F}_{g}$ & The function from $\mathfrak{S}_{g}$ to $\cmpx$ given by $\mathfrak{F_{g}} = \prod_{\{\gamma\}}\prod_{k=0}^{\infty}\left(1- \hspace{1mm}q_{\gamma}^{1+k}\right)$\\
	$q_\gamma$ & Multiplier of $\gamma\in\Gamma$\\
	$\mathfrak{F}_{g,n}(\boldsymbol{m})$ & Generalization of $\mathfrak{F}_{g}$ to $\schottky_{g,n}(\boldsymbol{m})$\\
	$B_2(x)$ & Second Bernoulli polynomial\\
	$Z_{\text{Sz}}(s,\Gamma,\mathfrak{U})$ & Selberg zeta function\\
	$Z^{1-\text{loop}}_{\text{gravity}}$ & One-loop partition function of 3-dimensional gravity\\
	$V_{\alpha_{i}}, V_{m_i}$ & Liouville vertex operators with charges $\alpha_i$\\
	$\boldsymbol{T}(w)$ & Conformal energy momentum tensor\\
	$\boldsymbol{h}_{m_i}$ & Conformal dimensions of vertex operators 
	$V_{m_i}$ \\
	$h_{\text{cl}}(m_i)$ & Classical limit of $\boldsymbol{h}_{m_i}$\\
	\hline
\end{longtable}
\vspace{2cm}

\section{List of symbols in the appendices}\label{symbolapp}
\begin{longtable}{| p{.12\textwidth} | p{.88\textwidth} |} 
	\hline
	$\mathfrak{n}$ & Dimension or rank of different objects\\
	$g$ & genus of an orbifold Riemann surface\\
	$n_e$ & Number of elliptic fixed points or branch points\\
	$n_p$ & Number of parabolic fixed points or punctures\\
	$n$ & Number of all marked points, namely $n_e+n_p$\\
	$m_i$ & Orders of the marked points\\
	$[\cdot/\cdot]$ & Quotient as an analytic orbifold/stack\\
	$X$ & Analytic subvariety\\
	$(|X|,\mathscr{O}_X)$ & Ringed space or complex analytic space\\
	$|X|$ & Underlying topological space of an analytic space $X$\\
	$X_{\text{red}}$ & Reduction of an analytic space\\
	$E$ & Vector bundles of different kind\\
	$\bar{E_1}$ & Complex conjugate of a vector bundle\\
	$\mathsf{M}_{r\times r}(\cmpx)$ & Space of complex $r \times r$ matrices\\
	$L$ & Complex line bundle\\
	$\linebundle$ & Holomorphic line bundle\\
	$\mathcal{K}_{X}$ & Canonical line bundle, i.e. $\mathfrak{n}$-th exterior power $\bigwedge^{\mathfrak{n}} T^{\ast}_{(1,0)}X$\\
	$\mathcal{K}_{O}$ & Canonical orbidivisor\\
	$\mathcal{K}^{-1}_{X}$ & Anticanonical line bundle, i.e. The dual or inverse line bundle of $\mathcal{K}_{X}$\\
	$\linebundle(\mathcal{D})$ & Holomorphic line bundle corresponding to a divisor\\
	$O$ & Complex analytic orbifold\\
	$X_O$ & Underlying complex analytic space for an orbifold\\
	$X_O^{\text{reg}}$ & Orbifold regular locus or the principal stratum\\
	$\mathbb{X}$ & A complex model space\\
	
	$f_1,\dots,f_k$ & A system of local defining functions for an analytic subvariety\\
	$\operatorname{Reg}(X)$ & Set of all regular point of a subvariety $X$\\
	$\sing(X)$ & Singular locus or the set of singular points of of a subvariety $X$, i.e. $X\backslash \operatorname{Reg}(X)$\\
	$\varDelta(\boldsymbol{\epsilon})$ & A small polydisc\\
	$\mathfrak{O}_{\varDelta(\boldsymbol{\epsilon})}$ & Ring of all holomorphic functions on $\varDelta(\boldsymbol{\epsilon})$\\
	$\mathfrak{O}_{X}$ & Ring of holomorphic functions on a subvariety $X$, i.e. $\mathfrak{O}_{\varDelta(\boldsymbol{\epsilon})}/\mathfrak{I}(X)$\\
	$\mathfrak{O}_{U}$ & Ring of holomorphic functions in some open subset $U \subset \cmpx^{\mathfrak{n}}$ containing the origin \\
	$\mathfrak{O}_{U,0}$ & Ring of germs of holomorphic functions at the origin\\
	$\mathfrak{O}_{X,0}$ & Ring of germs of holomorphic functions on the subvariety $X$ defined by $\mathfrak{O}_{\cmpx^{\mathfrak{n}},0}/\mathfrak{I}([X]_0)$\\
	$\mathfrak{I}(X)$ & Ideal of the subvariety $X$, i.e. ideal of all vanishing functions in $X$ in the ring $\mathfrak{O}_{\varDelta(\boldsymbol{\epsilon})}$\\
	$\mathfrak{I}$ & Defining ideal of a subvariety $X$, i.e ideal formed by the set of defining functions of $X$\\
	$\sqrt{\mathfrak{I}}$ & Radical ideal of a subvariety $X$,\hspace{-.5mm} i.e.\hspace{-.5mm}$\left\{\hspace{-.5mm}f\hspace{-.3mm} \in \mathfrak{O}_{\varDelta(\boldsymbol{\epsilon})} \,\hspace{-.5mm} \Big|\hspace{-.5mm} \, f^{k'}\hspace{-1.5mm} \in\hspace{-.5mm} \mathfrak{I} \text{ for some positive integer } k'\right\}$\\
	\hline
	$\mathfrak{I}([X]_0)$ & Ideal canonically associated to a germ $[X]_0$ of an analytic subvariety at the origin defined as the ideal of germs of all analytic functions vanishing on the subvariety $X$ representing the germ $[X]_0$\\
	$\mathfrak{m}_{X,x}$ & Maximal ideal of a ringed space $X$\\
	$\mathscr{J}$ & Coherent ideal used in the definition of the closed complex analytic subspace\\
	$[X]_0$ & Germ of analytic subvariety $X$ at 0  in $\cmpx^{\mathfrak{n}}$\\
	$[f]_0$ & Germ of holomorphic function $f$ at the origin\\
	$[X(\mathfrak{I})]_0$ & Locus of the ideal $\mathfrak{I}$, i.e. a germ of an analytic subvariety at the origin in $\cmpx^{\mathfrak{n}}$ canonically associated to an ideal $\mathfrak{I} \subseteq \mathfrak{O}_{\cmpx^{\mathfrak{n}},0}$\\ 
	$\mathscr{O}_{\cmpx^{\mathfrak{n}}},\mathscr{O}_{U}$ & Sheaf of germs of holomorphic functions of $\mathfrak{n}$ complex variables and its restriction to $U \subset \cmpx^{\mathfrak{n}}$, repectively\\
	$\mathscr{I}(X)$ & Sheaf of ideals of the analytic subvariety $X$\\
	$\mathscr{O}_{X}$ & Sheaf of germs of holomorphic functions on the subvariety $X$, i.e. $\mathscr{O}_{U}/\mathscr{I}(X)$\\
	$\mathscr{I}$ & Coherent sheaf of $\mathscr{O}_U$-ideals\\
	$\mathscr{E}_{M}$ & Sheaf of germs of $\mathcal{C}^{\infty}$ complex functions on a complex manifold $M$\\
	$\mathscr{E}_{M}^r$ & Direct sum sheaf defined by $ \underbrace{\mathscr{E}_{M} \oplus \dotsm \oplus \mathscr{E}_{M}}_{r}$\\
	$\mathscr{E}(E)$ & Sheaf of germs of $\mathcal{C}^{\infty}$-sections of a vector bundle $E$\\
	$\mathscr{E}^{(k,l)}_{M}(E)$ & Sheaf of germs of smooth sections of $\bigwedge^{k,l} T^{\ast}M \otimes E$\\ 
	$\mathscr{E}_{M}^{\ast}$ & Multiplicative sheaf of invertible $\mathcal{C}^{\infty}$ complex functions on a complex manifold $M$\\
	$\mathscr{O}_X(E)$ & Analytic sheaf on $X$ of germs of holomorphic sections in $E$\\
	$\mathscr{O}_{X}^{\ast}$ & Sheaf of nowhere-vanishing holomorphic functions, i.e. sheaf of invertible elements in $ \mathscr{O}_{X}$\\
	$\mathscr{M}_X$ & Sheaf of meromorphic functions on $X$\\
	$\mathscr{M}^{\ast}_X$ & subsheaf of not-identically-zero meromorphic functions on $X$\\
	$\mathcal{F}_O$ & Orbisheaf\\
	$\tilde{\mathcal{F}}_a$ & Sheaves that construct an orbisheaf $\mathcal{F}_O$\\
	$\mathscr{O}_O$ & Structure orbisheaf of an orbifold $O$\\ 
	$\mathcal{F}_X$ & Sheaves on $X_O$ coming from invariant local sections of orbisheaves $\mathcal{F}_O$ \\
	$\Omega^{\mathfrak{n}}_O$ & canonical orbisheaf of a complex orbifold $O$ of complex dimension $\mathfrak{n}$\\
	$(|f|,f^{\ast})$ & Morphism of $\cmpx$-ringed spaces\\
	$g^{\ast}_{ab}$ & Gluing isomorphisms in a analytic atlas\\
	$g^{\alpha \beta}$ & Transition matrix of a vector bundle $E$\\
	$\psi_{\alpha}$ & Trivialization of a vector bundle $E$ on an open set $V_{\alpha}$\\
	$\pi$ & The projection from $E$ to the base space in the definition of a vector bundle\\
	$\operatorname{pr}$ & The projection $V_{\alpha} \times \cmpx^r \to V_{\alpha}$ in the definition of a Vector bundle\\
	$\delta$ & Connecting homomorphism for a holomorphic line bundle $\linebundle$\\
	$\{u^{\alpha \beta}\}$ & System of transition functions for the holomorphic line bundle $\linebundle(\mathcal{D})$\\
	$\varpi$ & Surjective analytic map invariant under $\Gamma$\\
	$|\varpi|$ & Quotient map sending $y$ to its left orbit $\Gamma(y)$\\
	$\eta_{ab}$ & Change of charts for an orbifold\\
	\hline
	$\Upsilon_{ab}$ & Monomorphism between $\Gamma_a$ and $\Gamma_b$ in orbifolds\\
	$f_{\text{orb}}$ & Analytic orbifold automorphism\\
	$\varpi_{\text{orb}}$ & Orbifold Galois covering\\
	$\pi_{\text{orb}}$ & Orbifold projection map in the definition of orbibundles\\
	$\widehat{\eta}_{ab}$ & Holomorphic bundle isomorphism in the definition of orbibundles\\
	$\widehat{\Upsilon}_a$ & Monomorphism in the definition of orbibundles\\
	$\mathfrak{f}_a$ & Transition maps for a $(G,\mathbb{X})$-structure\\
	$\dev$ & Developing map\\
	$\hol$ & Holonomy representation\\
	$\mathfrak{C}$ & Injective group homomorphism between $\Gamma$ and $G$ in a $(G,\mathbb{X})$-structure\\
	$J$ & Klein's Hauptmodule, i.e. a meromorphic function on $\UHP$ which is automorphic with respect to the Fuchsian group $\Gamma$\\
	
	$s$ & Smooth complex section of vector bundle $E$ or Holomorphic section of a holomorphic vector bundle $E$\\
	$\boldsymbol{s}$ & $k$-frame of vector bundle $E$, i.e. a collection $(s_1, \dots, s_k)$ of $k$ sections $s_i$ of vector bundle $E$ on $V$ linearly independent at each point in $V$\\
	\small{$\operatorname{Hom}(E_1,E_2)$} & Homomorphism of vector bundles $E_1$ and $E_2$\\
	$\bigwedge^k E$ & k-th exterior power of a vector bundle\\
	$\mathcal{A}^{k}(V,E)$ & Vector space of $\mathcal{C}^{\infty}$ sections of $\big( \bigwedge^k \, T^{\ast} M \big) \otimes E$ on $V \subset M$, which are called differential forms on $V$ with values in the vector bundle $E$\\
	$\mathcal{A}^{(k,l)}(E)$ & Vector space of smooth sections of this sheaf are $(k,l)$-forms with values in $E$\\
	$\nabla$ & Connection of a vector bundle $E$ or $\mathcal{A}^{(k,l)}(E)$\\
	$\{\nabla_a\}$ & $\Gamma_a$-equivariant Hermitian connections supported on each local uniformizing neighborhood $\tilde{U}_a$ such that $\nabla_a$s are compatible with changes of charts\\
	$\Theta$ & Curvature of a connection $\nabla$, i.e. $\nabla \circ \nabla$\\
	$A = (A_{ij})$ & Connection matrix with respect to a frame\\
	$\theta = (\theta_{ij})$ & Curvature matrix with respect to a frame\\
	$\mathsf{g} = (\mathsf{g}_{ij})$ & Gauge transformation matrix\\
	$h$ & Hermitian metric on a complex vector bundle $E$\\
	$\mathsf{h}$ & Hermitian metric on a complex orbifold\\
	$\tilde{\mathsf{h}}_a^{\Gamma_a}$ & $\Gamma_a$-invariant (local) Hermitian metrics used in the definition of $\mathsf{h}$\\
	$c_i(E,\nabla)$ & i-th Chern form of a vector bundle $E$\\
	$c_i(E)$ & i-th Chern class of a vector bundle $E$\\
	$\operatorname{Pic}(X)$ & Picard group of $X$, i.e. the group $H^1(X,\mathscr{O}_{X}^{\ast})$ which is the group of holomorphic line bundles on the analytic space $X$ with group multiplication being the tensor product, and the inverse bundle being the dual bundle\\
	$\operatorname{Pic}^{0}(X)$ & Kernel of the connecting homomorphism $\delta$ for a holomorphic line bundle $\linebundle$\\
	$\operatorname{Div}(X)$ & Divisor group, i.e. the group constructed via the formal sum of Weil divisors\\
    {\fontsize{6.5}{150}\selectfont $H^0(X, \mathscr{M}^{\ast}_X / \mathscr{O}^{\ast}_X)$}     & Abelian group of Cartier divisors on $X$\\
	$\operatorname{Cl}(X)$ & Divisor class group of Weil divisors modulo linear equivalence\\
	$\operatorname{CaCl}(X)$ & Group of Cartier divisor classes, i.e. Cartier divisors modulo principal divisors\\
	$\Gamma$ & A finite subgroup of the group $\operatorname{Aut}(X)$ of analytic automorphisms of $X$\\
	$\Gamma_{y}$ & Isotropy subgroup or stabilizer subgroup of $y$\\
	$\Gamma(y)$ & Orbit of a point $y$\\
	$\operatorname{Fix}(\gamma)$ & Fixpoint set\\
	$T$ & Thin subset\\
	$R_{\varpi}$ & Ramification locus of an analytic covering\\
	$B_{\varpi}$ & Branching locus of an analytic covering, i.e. $\varpi(R_{\varpi})$\\
	$\operatorname{Aut}(\varpi)$ & Group of covering transformations or deck transformations, i.e. group of all automorphisms of an analytic covering\\
	{\fontsize{6}{150}\selectfont$\gal(Y/X), \gal(\varpi)$} & Galois group of an analytic covering, i.e. $\operatorname{Aut}(\varpi)$\\
	$\pi_1$ & Fundamental group\\
	$\isom^{+}$ & Group of orientation preserving isometries\\
	$\mathcal{D}$ & Divisor\\
	$\text{Div}(M)$ & Weil or Cartier divisors on a smooth complex manifold $M$\\
	$\mathcal{D}^{\text{reg}}$ & Weil divisor defined on $\operatorname{Reg}(X)$\\
	$\phi$ & Principal Cartier divisor or orbifold $k$-form\\
	$\mathcal{R}_i$ & Prime divisors or irreducible hypersurfaces\\
	$\mathscr{R}_{\varpi}$ & Ramification divisor of an analytic covering\\
	$\mathscr{B}_{\varpi}$ & Branching divisor of an analytic covering\\
	$\brdiv$ & Branch divisor of an orbifold\\
	$H_i$ & Irreducible analytic hypersurface\\
	$\operatorname{Supp}(\mathcal{D})$ & Support of a Weil divisor, i.e. $\bigcup_{i} H_i$\\
	$\operatorname{mult}_{H_i}(\mathcal{D})$ & Multiplicity of a Weil divisor $\mathcal{D}$, i.e. the coefficient $a_i$ in its definition\\
	$\deg(\mathcal{D})$ & Degree of a Weil divisor, i.e. $ \sum_i\operatorname{mult}_{H_i}(\mathcal{D}) = \sum_i a_i$\\
	$Z_f$ & Zero set of a holomorphic function $f$\\
	$\ord(f)$ & Order of vanishing of a holomorphic function $f$\\
	$[\mathcal{D}]$ & Linear system of divisors defined by $\mathcal{D}$, i.e. the set of all divisors on $X$ that are linearly equivalent to $\mathcal{D}$\\
	$\deg(\varpi)$ & degree of an analytic covering\\
	$\#$ & Order of a group\\
	$\nu_{\varpi}$ & Branching function of a covering\\
	$m_i$ & Ramification indices of a covering $\varpi$ along irreducible hypersurfaces $\mathcal{R}_i$\\
	$(X_O,\brdiv)$ & Log pair\\
	{\fontsize{9.5}{150}\selectfont$(\mathrm{D},\mathrm{D},\mathbb{Z}_{m_i},\ff_i)$} & Charts on an orbifold Riemann surface\\
	$m_i$ & Branching index of an orbifold Riemann surface\\
	$U_a$ & Open subset in $X_O$ in the definition of orbifold\\
	$\tilde{U}_a$ & Open subset in $\cmpx^{\mathfrak{n}}$ in the definition of orbifold\\
	$\Gamma_a$ & Subgroup of $\operatorname{GL}(\mathfrak{n},\cmpx)$ in the definition of orbifold\\
	$\mathrm{f}_a$ & Folding map in the definition of orbifold\\
	$\mathcal{U}$ & Orbifold atlas\\
	$\mathcal{U}_{\max}$ & Maximal orbifold atlas\\
	$[\mathcal{U}]$ & Equivalence class of analytic orbifold atlases on $X_O$\\
	$\Gamma_x$ & isotropy group or the local group\\
	\hline
	$m_x$ & Order of the local isotropy group\\
	$\mathsf{A}_i$ & Homotopy class of loops around a handle\\
	$\mathsf{B}_i$ & Homotopy class of loops around a hole\\
	$\mathsf{C}_i$ & Homotopy class of loops around a branch point\\
	$\mathsf{P}_i$ & Homotopy class of loops around a puncture\\
	$F$ & Number of faces\\
	$E$ & Number of edges\\
	$V$ & Number of vertices\\
	$\tilde{\phi}_a^{\Gamma_a}$ & $\Gamma_a$-invariant complex $k$-forms used in the definition of the orbifold $k$-form\\
	$\mathcal{E}^{p,q}(O)$ & Vector space of all orbifold $(p,q)$-froms on an orbifold $O$\\
	{\fontsize{6.5}{150}\selectfont$\diffspace^{p,q}\bigg(O, K_O^{k} \hspace{-.5mm}\otimes \overline{K}_O^{l}\bigg)$} & Vector space of smooth differential forms of type $(p,q)$ on $O$ with values in $K_O^{k} \otimes \overline{K}_O^{l}$\\
	$\mathcal{E}^{k,l}(\UHP,\Gamma)$ & Hilbert space of automorphic forms of weight $(2k,2l)$ with the natural scalar product $\langle \phi_1,\phi_2 \rangle = \int_{X} \phi_1 \overline{\phi_2} \rho^{-k-\ell+1}$\\
	$\Hilbert^{k,\ell}(X)$ & space of harmonic $(k,\ell)$-differentials that are square-integrable  with respect to the hyperbolic metric on $X =\UHP \slash \Gamma$\\
	$\mathcal{P}$ & Complex structure\\
	$P$ & Projective structure\\
	$P^{{}_H}$ & Projective structure fixed by the convex combination $\frac{1}{\# H} \sum_{\textcolor{black}{\mathrm{h}} \in H} \mathrm{h}^{\ast} P$\\ 
	$\text{Sch}$ & Schwarzian derivative\\
	$q_a$ & Elements of a quadratic differential\\
	$\phi$ & Liouville field\\
	$\alpha_i,\beta_i$ & Hyperbolic generators of a Fuchsian group\\
	$\tau_i$ & Elliptic generator of a Fuchsian group\\
	$\kappa_i$ & Parabolic generator of a Fuchsian group\\
	$\lambda_i$ & Multipliers of elliptic generators of a Fuchsian group\\
	$\delta_i$ & Translation length of parabolic generators of a Fuchsian group\\
	$J_k^{(i)}$ & The $k^{th}$ order coefficient of $J$'s expansion around the $i^{th}$ marked point\\
	$c_i$ & Accessory parameters\\
	$h_i$ & Conformal weight corresponding to the order of a marked point $m_i$\\
	$\Hilbert^{2,0}(\UHP,\Gamma)$ & Space of cusp forms of weight $4$ for the group $\Gamma$--equivalently, meromorphic (2,0)-tensors/quadratic differentials on the Riemann orbisurface $O$\\
	$\Hilbert^{-1,1}(\UHP,\Gamma)$ & Harmonic Beltrami differentials, that are a subspace $ \Lambda^{\ast}\left(\Hilbert^{2,0}(\UHP,\Gamma)\right)$ of $\diffspace^{-1,1}(\UHP,\Gamma)$\\
	$\Hilbert^{2,0}(O)$ & Space of quadratic differentials on an orbifold\\
	$\Hilbert^{-1,1}(O)$ & Space of harmonic differentials on an orbifold\\
	$q(z)$ & An element of $\Hilbert^{2,0}(\UHP,\Gamma)$\\
	$\mu(z)$ & An element of $\Hilbert^{-1,1}(\UHP,\Gamma)$\\
	$Q(w)$ & An element of $\Hilbert^{2,0}(O)$\\
	$M(w)$ & An element of $\Hilbert^{-1,1}(O)$\\
	\hline
\end{longtable}

\bibliographystyle{JHEP}
\bibliography{BIB}

\end{document}